\PassOptionsToClass {plain,biber}{nowfnt}
\documentclass[CIT,biber]{nowfnt}% creates the journal version, needs biber version
\usepackage[utf8]{inputenc}
\usepackage{mathtools}
\mathtoolsset{showonlyrefs}
 
\newcommand{\ucalR}{\underline{\calR}}
\newcommand{\ocalR}{\overline{\calR}}

\usepackage[mathscr]{eucal}
\usepackage{epsfig,epsf,psfrag}
\usepackage{amssymb,amsmath,amsthm,amsfonts,latexsym}
\usepackage{amsmath,graphicx,bm,xcolor,url}
\usepackage[caption=false]{subfig} 
\usepackage{fixltx2e}%ordering of single and double column floats
\usepackage{array}%array and tabular environments
\usepackage{verbatim}
\usepackage{bm}
\usepackage{algorithmic}
\usepackage{algorithm}
\usepackage{verbatim}
\usepackage{textcomp}
\usepackage{mathrsfs, bbm,overpic}
\usepackage{epstopdf}

\usepackage{stackengine}

\usepackage{tikz}
\usepackage{float}
\usepackage{tabularx}
\usepackage{multirow}
\usetikzlibrary{patterns}
\newcommand{\grad}{\nabla}
\newcommand{\Unif}{\mathrm{Unif}}
\newcommand{\sign}{\mathrm{sign}}

\catcode`~=11 \def\UrlSpecials{\do\~{\kern -.15em\lower .7ex\hbox{~}\kern .04em}} \catcode`~=13 

\allowdisplaybreaks[1]
 
\newcommand{\nn}{\nonumber}

\newcommand{\calA}{\mathcal{A}}
\newcommand{\calB}{\mathcal{B}}
\newcommand{\calC}{\mathcal{C}}

\newcommand{\calF}{\mathcal{F}}

\newcommand{\calI}{\mathcal{I}}

\newcommand{\calK}{\mathcal{K}}

\newcommand{\calM}{\mathcal{M}}
\newcommand{\calN}{\mathcal{N}}

\newcommand{\calP}{\mathcal{P}}

\newcommand{\calR}{\mathcal{R}}
\newcommand{\calS}{\mathcal{S}}
\newcommand{\calT}{\mathcal{T}}
\newcommand{\calU}{\mathcal{U}}
\newcommand{\calV}{\mathcal{V}}
\newcommand{\calW}{\mathcal{W}}
\newcommand{\calX}{\mathcal{X}}
\newcommand{\calY}{\mathcal{Y}}
\newcommand{\calZ}{\mathcal{Z}}

\newcommand{\ba}{\mathbf{a}}
\newcommand{\bA}{\mathbf{A}}

\newcommand{\bB}{\mathbf{B}}

\newcommand{\bC}{\mathbf{C}}

\newcommand{\bD}{\mathbf{D}}

\newcommand{\bg}{\mathbf{g}}

\newcommand{\bI}{\mathbf{I}}

\newcommand{\bK}{\mathbf{K}}

\newcommand{\bL}{\mathbf{L}}

\newcommand{\bM}{\mathbf{M}}

\newcommand{\bS}{\mathbf{S}}

\newcommand{\bu}{\mathbf{u}}
\newcommand{\bU}{\mathbf{U}}
\newcommand{\bv}{\mathbf{v}}
\newcommand{\bV}{\mathbf{V}}

\newcommand{\bW}{\mathbf{W}}
\newcommand{\bx}{\mathbf{x}}

\newcommand{\bz}{\mathbf{z}}

\newcommand{\rmc}{\mathrm{c}}

\newcommand{\rmd}{\mathrm{d}}

\newcommand{\rme}{\mathrm{e}}

\newcommand{\rmm}{\mathrm{m}}

\newcommand{\bbB}{\mathbb{B}}
\newcommand{\bbC}{\mathbb{C}}

\newcommand{\bbE}{\mathbb{E}}

\newcommand{\bbL}{\mathbb{L}}

\newcommand{\bbN}{\mathbb{N}}

\newcommand{\bbR}{\mathbb{R}}
\newcommand{\bbS}{\mathbb{S}}

\newcommand{\bbU}{\mathbb{U}}

\newcommand{\scC}{\mathscr{C}}

\DeclareMathAlphabet{\mathbsf}{OT1}{cmss}{bx}{n}
\DeclareMathAlphabet{\mathssf}{OT1}{cmss}{m}{sl}% slanted sans serif

\newcommand{\rvD}{\mathsf{D}}

\newcommand{\rvH}{\mathsf{H}}

\newcommand{\rvX}{\mathsf{X}}

\DeclareSymbolFont{bsfletters}{OT1}{cmss}{bx}{n}  
\DeclareSymbolFont{ssfletters}{OT1}{cmss}{m}{n}
\DeclareMathSymbol{\bsfGamma}{0}{bsfletters}{'000}
\DeclareMathSymbol{\ssfGamma}{0}{ssfletters}{'000}
\DeclareMathSymbol{\bsfDelta}{0}{bsfletters}{'001}
\DeclareMathSymbol{\ssfDelta}{0}{ssfletters}{'001}
\DeclareMathSymbol{\bsfTheta}{0}{bsfletters}{'002}
\DeclareMathSymbol{\ssfTheta}{0}{ssfletters}{'002}
\DeclareMathSymbol{\bsfLambda}{0}{bsfletters}{'003}
\DeclareMathSymbol{\ssfLambda}{0}{ssfletters}{'003}
\DeclareMathSymbol{\bsfXi}{0}{bsfletters}{'004}
\DeclareMathSymbol{\ssfXi}{0}{ssfletters}{'004}
\DeclareMathSymbol{\bsfPi}{0}{bsfletters}{'005}
\DeclareMathSymbol{\ssfPi}{0}{ssfletters}{'005}
\DeclareMathSymbol{\bsfSigma}{0}{bsfletters}{'006}
\DeclareMathSymbol{\ssfSigma}{0}{ssfletters}{'006}
\DeclareMathSymbol{\bsfUpsilon}{0}{bsfletters}{'007}
\DeclareMathSymbol{\ssfUpsilon}{0}{ssfletters}{'007}
\DeclareMathSymbol{\bsfPhi}{0}{bsfletters}{'010}
\DeclareMathSymbol{\ssfPhi}{0}{ssfletters}{'010}
\DeclareMathSymbol{\bsfPsi}{0}{bsfletters}{'011}
\DeclareMathSymbol{\ssfPsi}{0}{ssfletters}{'011}
\DeclareMathSymbol{\bsfOmega}{0}{bsfletters}{'012}
\DeclareMathSymbol{\ssfOmega}{0}{ssfletters}{'012}

\newcommand{\hatC}{\hat{C}}

\newcommand{\tilC}{\tilde{C}}

\newcommand{\tilf}{\tilde{f}}

\newcommand{\hatM}{\hat{M}}

\newcommand{\tilP}{\tilde{P}}

\newcommand{\tilQ}{\tilde{Q}}

\newcommand{\tilS}{\tilde{S}}

\newcommand{\tilT}{\tilde{T}}

\newcommand{\hatx}{\hat{x}}
\newcommand{\hatX}{\hat{X}}

\newcommand{\tilX}{\tilde{X}}

\newcommand{\haty}{\hat{y}}
\newcommand{\hatY}{\hat{Y}}

\newcommand{\tilY}{\tilde{Y}}

\newcommand{\bara}{\bar{a}}
\newcommand{\barb}{\bar{b}}

\newcommand{\barx}{\bar{x}}

\newcommand{\eps}{\varepsilon}
\newcommand{\Bern}{\mathrm{Bern}}

\DeclareMathOperator{\supp}{supp}

\DeclareMathOperator{\var}{\mathrm{Var}}

\DeclareMathOperator{\cov}{\mathrm{Cov}}

\DeclareMathOperator{\rank}{rank}

\newcommand{\bzero}{\mathbf{0}}
\newcommand{\bone}{\mathbbm{1}}
\newcommand{\Wyner}{\mathrm{W}}
\newcommand{\GKW}{\mathrm{GKW}}

\newtheorem{problem}{Problem}[chapter] 
\newtheorem{conjecture}{Conjecture}[chapter] 
\newtheorem{convention}{Convention}[chapter] 
\newtheorem{assumption}{Assumption}[chapter] 

\newcommand{\qednew}{\nobreak \ifvmode \relax \else
      \ifdim\lastskip<1.5em \hskip-\lastskip
      \hskip1.5em plus0em minus0.5em \fi \nobreak
      \vrule height0.75em width0.5em depth0.25em\fi}

\usepackage{diagbox}
\title{Common Information, Noise Stability, and Their Extensions\footnote{Suggested Citation: Lei Yu and Vincent Y. F. Tan (2022), ``Common Information,
Noise Stability, and Their Extensions'', Foundations and Trends® in Communications and
Information Theory: Vol. 19, No. 3, pp 264--546. DOI: 10.1561/0100000122.}}

\maintitleauthorlist{
Lei Yu \\
School of Statistics and Data Science,\\
LPMC, KLMDASR, and LEBPS\\
Nankai University \\
China \\
leiyu@nankai.edu.cn
\and
Vincent Y.~F.~Tan\\
Department of Mathematics, \\
Department of ECE, and IORA \\
National University of Singapore \\
Singapore \\
vtan@nus.edu.sg
}

\author[1]{Yu, Lei}
\author[2]{Tan, Vincent Y. F.}

\affil[1]{School of Statistics and Data Science, LPMC, KLMDASR, and LEBPS, Nankai University, China; leiyu@nankai.edu.cn}
\affil[2]{Department of Mathematics, Department of Electrical and Computer Engineering, Institute of Operations Research and Analytics, National University of Singapore, Singapore; vtan@nus.edu.sg}

\begin{document}

\makeabstracttitle

\begin{abstract}
Common information is ubiquitous in information theory and related areas such as theoretical computer science and discrete probability. However, because there are multiple notions of common information, a unified understanding of the deep interconnections between them is lacking. This monograph seeks to fill this gap by leveraging a small set of mathematical techniques that are applicable across seemingly disparate problems. 

%However, the  understanding of common information is far from deep, due to the lack in unveiling the essential connections among  various kinds of common informations proposed in the literature. This monograph provides an attempt to this end by presenting a unified treatment of various common informations and related topics within or beyond information theory.  

In Part~\ref{part:one}, we review the operational tasks and properties associated with Wyner's and G\'acs--K\"orner--Witsenhausen's (GKW's) common information. %We present several mathematical tools that are  convenient for studying these quantities in greater generality. 
In Part~\ref{part:two},   we discuss  extensions of the former from the perspective of distributed source simulation.  This includes the R\'enyi common information which forms a bridge between Wyner's common information and  the exact common information. Via a surprising equivalence between the R\'enyi common information  of order~$\infty$ and the exact common information, we demonstrate  the existence of a  joint source in which the exact common information strictly exceeds Wyner's common information. Other closely related topics discussed in Part~\ref{part:two} include  the channel synthesis problem and the connection of Wyner's and exact common information  to the nonnegative rank of matrices.

In Part~\ref{part:three}, recognizing that GKW's common information is zero for most non-degenerate sources, we examine it with a more refined lens via the Non-Interactive Correlation Distillation (NICD) problem in which we quantify the agreement probability of  extracted bits from a bivariate source. We extend this to the noise stability problem which includes as special cases  the $k$-user NICD and $q$-stability problems. This allows us to seamlessly transition to discussing their connections to various conjectures in information theory and discrete probability, such as the Courtade--Kumar, Li--M\'edard and Mossell--O'Donnell conjectures. Finally, we consider functional inequalities  (e.g., the hypercontractivity and Brascamp--Lieb  inequalities), which constitute a further generalization of the noise stability problem in which the Boolean functions therein are replaced by nonnnegative functions. We demonstrate that the key ideas behind the proofs in Part~\ref{part:three} can be presented in a pedagogically coherent  manner and unified via information-theoretic and Fourier-analytic methods. %\cite{Cov06} %\cite{elgamal}
%In Part \ref{part:three}, we consider several topics related to  G\'acs--K\"orner--Witsenhausen's  common information. These  include functional inequalities (e.g., Brascamp--Lieb and hypercontractivity inequalities and their strong versions), isoperimetric inequalities (e.g., small-set expansion theorems and their strong versions), and the analysis of Boolean functions. Our underlying goal is to  demonstrate that the proofs of most of the results in the monograph can be presented in a pedagogically coherent  manner and unified via  information-theoretic  and  Fourier-analytic methods. 
%\cite{elgamal}  
\end{abstract}

%\clearpage
\tableofcontents
%\newpage
\clearpage
 
\chapter{Introduction} \label{ch:intro}

\section{Motivation}\label{sec:motivation}
Let $X$ be the statistical description of a set of images whose foregrounds and backgrounds are those of an airplane and the  blue sky respectively. Let~$Y$, which is  correlated to~$X$, be the statistical description of another set of images whose foregrounds are those of a unicorn  and the blue sky respectively. It seems natural and intuitive that the common information in~$X$ and~$Y$ should be the number of bits needed to describe the blue sky, which is the common part of $X$ and~$Y$.  Can we make this observation precise and quantitative   for {\em arbitrary} $(X,Y)$ pairs? This monograph is centered on this fundamental question in information and probability theory. In other words, we would like to quantify, via an assortment of well-motivated measures, the {\em intrinsic similarity} or {\em common information} between two correlated random variables $X$ and $Y$. Regardless of what applications there may be, the pursuit of operationally meaningful measures that quantify the common information between two random variables seems to be an extremely worthy academic endeavor. This is especially so for researchers in information and coding theory, theoretical computer science, and cryptography who are seeking to understand the inherent difficulties in generating correlated bits   from a single joint source, or simulating a joint source using a single source of randomness in a distributed manner. 

In probability, statistics, and data analysis, there are numerous popular functionals of joint distributions that quantify the amount of correlation or dependence between two random variables $X$ and $Y$. If these random variables have joint distribution $\pi_{XY}$ and  means $\mu_X$ and $\mu_Y$ respectively, such paradigmatic examples include the {\em Pearson correlation coefficient} 
\begin{equation}
\rho(X;Y) := \frac{\cov(X,Y)}{\sqrt{\var(X)\var(Y)}}  = \frac{\bbE [(X-\mu_X)(Y-\mu_Y) ]}{ \sqrt{ \bbE[(X-\mu_X)^2] \bbE[(Y-\mu_Y)^2] } } \label{eqn:pcc}
\end{equation}
and the {\em Hirschfeld--Gebelein--R\'enyi (HGR)  maximal  correlation}
\begin{align}
\rho_{\mathrm{m}}(X;Y):=\sup_{f,g}\rho\big(f(X);g(Y) \big),\label{eq:mc}
\end{align}
where  the supremum is taken over all real-valued functions $f$ and $g$
such that $0<\var(f(X)),\var(g(Y))<\infty$. In addition, an 
information-theoretic quantity known as the {\em mutual information} 
\begin{equation}
I(X;Y) = \left\{ \begin{array}{cl}
\displaystyle  \int_{\mathcal{X} \times \mathcal{Y}}\log\bigg(\frac{\rmd \pi_{XY}}{\rmd (\pi_{X}\pi_Y)}\bigg)\, \rmd \pi_{XY}  &\mbox{if}\;\, \pi_{XY}\ll \pi_X\pi_Y \vspace{.03in}\\
+\infty & \mbox{otherwise}
\end{array} \right. ,\label{eq:mi}
%\bbE_{\pi_{XY}}\left[ \log\frac{\rmd \pi_{XY}}{\rmd (\pi_{X}\pi_Y)}(X,Y)\right]%\, \rmd \pi_{XY}% \bbE\left[ \log\frac{\pi_{XY}(X,Y)}{\pi_{X}(X)\pi_Y(Y)} \right]. 
\end{equation}
%if $\pi_{XY}$ is absolutely continuous with respect to $\pi_{X}\pi_Y$, and $I(X;Y) =\infty$ otherwise,
%}
also serves to quantify the dependence between two random variables.  These measures have the property that they are zero if the two random variables are independent, fulfilling a basic requirement of any measure that quantifies the dependence between two random variables. These measures can be regarded as common information quantities between~$X$ and~$Y$, jointly distributed as $\pi_{XY}$. Indeed, the mutual information $I(X;Y)$ captures the amount of information
about~$X$ provided by observing~$Y$, as can observed in the celebrated distributed lossless compression theorem of  Slepian and Wolf~\cite{sw73,cover75}. Are there any other {\em operationally-motivated} measures that allow us to gain deeper insights on the common information between $X$ and $Y$ given their numerical values?% of the above dependency measures. 

% However, one major disadvantage of these measures is that they are not {\em operationally-motivated}. This means that they are not the consequence of a particular coding problem that naturally informs us about the common information embedded within $X$ and $Y$. As a result, insights of the dependencies between the random variables cannot be readily obtained given  numerical values of the above dependency measures. 

In information and coding theory, there are two canonical  examples of operationally-motivated common information measures that have been widely accepted since their inceptions in the 1970s. The first,  which was introduced in 1973,  is {\em G\'acs--K\"orner--Witsenhausen's (GKW's) common information}  \cite{gacs1973common, witsenhausen1975sequences}, defined as
\begin{equation}
C_{\GKW} (\pi_{XY})  := \sup_{f,g} H\big( f(X) \big),\label{eqn:gkw_ci_intro} % \max_{P_{UXY}:H(U|X)=H(U|Y)=0, P_{XY}=\pi_{XY}}H (U),  \label{eqn:gkw_ci_intro}%\max_{P_{UXYV} :U= f(X), V=g(Y), U=V}H(U) 
\end{equation}
where the supremum is taken over all pairs of
deterministic functions $(f,g)$  defined respectively on $\mathcal{X}$
and $\mathcal{Y}$ such that $f(X)=g(Y)$ with $\pi_{XY}$-probability one.
%where the maximum extends over  all triples of random variables $(U,X,Y)$ such that $P_{XY}=\pi_{XY}$ %and  quadruples of random variables $(U,X,Y,V)\sim P_{UXYV}$ such that $U$ is a function of $X$, $V$ is a function of $Y$ and $U$ and $V$ are equal (almost surely). %such that 
% and the conditional entropies of $U$ given $X$ and $U$ given $Y$ are both equal to $0$. The latter means that $U$ is a function of $X$ and  a function of $Y$.  
 The second, which was introduced in 1975, is {\em Wyner's common information} \cite{WynerCI}, defined as 
\begin{equation}
C_{\Wyner} (\pi_{XY})  := \inf_{P_W P_{X|W} P_{Y|W} : P_{XY}=\pi_{XY} } I_P(W;XY), \label{eqn:wyner_ci_intro}
\end{equation}
where the infimum extends over   triples of random variables $(W,X,Y)\sim P_{WXY}$ such that $X-W-Y$ forms a Markov chain and $P_{XY}=\pi_{XY}$. 

%The second is {\em G\'acs--K\"orner--Witsenhausen's (GKW's) common information}  \cite{gacs1973common, witsenhausen1975sequences}, defined as
%\begin{equation}
%C_{\GKW} (\pi_{XY})  :=  \max_{P_{UXY}:H(U|X)=H(U|Y)=0, P_{XY}=\pi_{XY}}H (U),  \label{eqn:gkw_ci_intro}%\max_{P_{UXYV} :U= f(X), V=g(Y), U=V}H(U) 
%\end{equation}
%where the maximum extends over  all triples of random variables $(U,X,Y)$ such that $P_{XY}=\pi_{XY}$ %and  quadruples of random variables $(U,X,Y,V)\sim P_{UXYV}$ such that $U$ is a function of $X$, $V$ is a function of $Y$ and $U$ and $V$ are equal (almost surely). %such that 
% and the conditional entropies of $U$ given $X$ and $U$ given $Y$ are both equal to $0$. The latter means that $U$ is a function of $X$ and  a function of $Y$. 

\section{Overview of the  Monograph}
Our twin objectives in this monograph are as follows. Firstly, we seek to provide a concise review of these classical notions of common information. Secondly, we endeavor to connect these quantities to new notions of common information in the literature that have gained traction recently. A flowchart of the sections in this monograph is provided in Fig. \ref{fig:flowchart}.

\subsection{Part I: Classic Common Information Quantities}
We commence in Part~\ref{part:one} by reviewing  the operational tasks associated with the classical common information  quantities in~\eqref{eqn:gkw_ci_intro} and~\eqref{eqn:wyner_ci_intro} and describing their salient properties. This part consists of Sections~\ref{ch:wynerCI} and~\ref{ch:gkw} on Wyner's  and GKW's common information respectively.

%The core objective of the monograph, however, is to make connections of these classical quantities to new notions of common information in the literature that have gained traction  recently. 

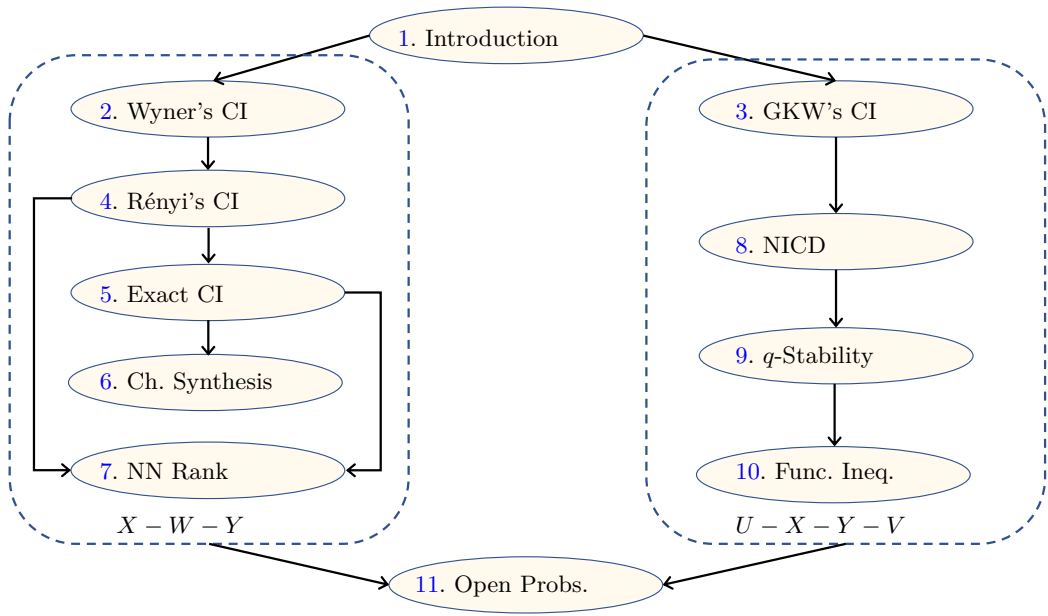
\begin{figure}
\centering
\begin{overpic}[width=.99\columnwidth]{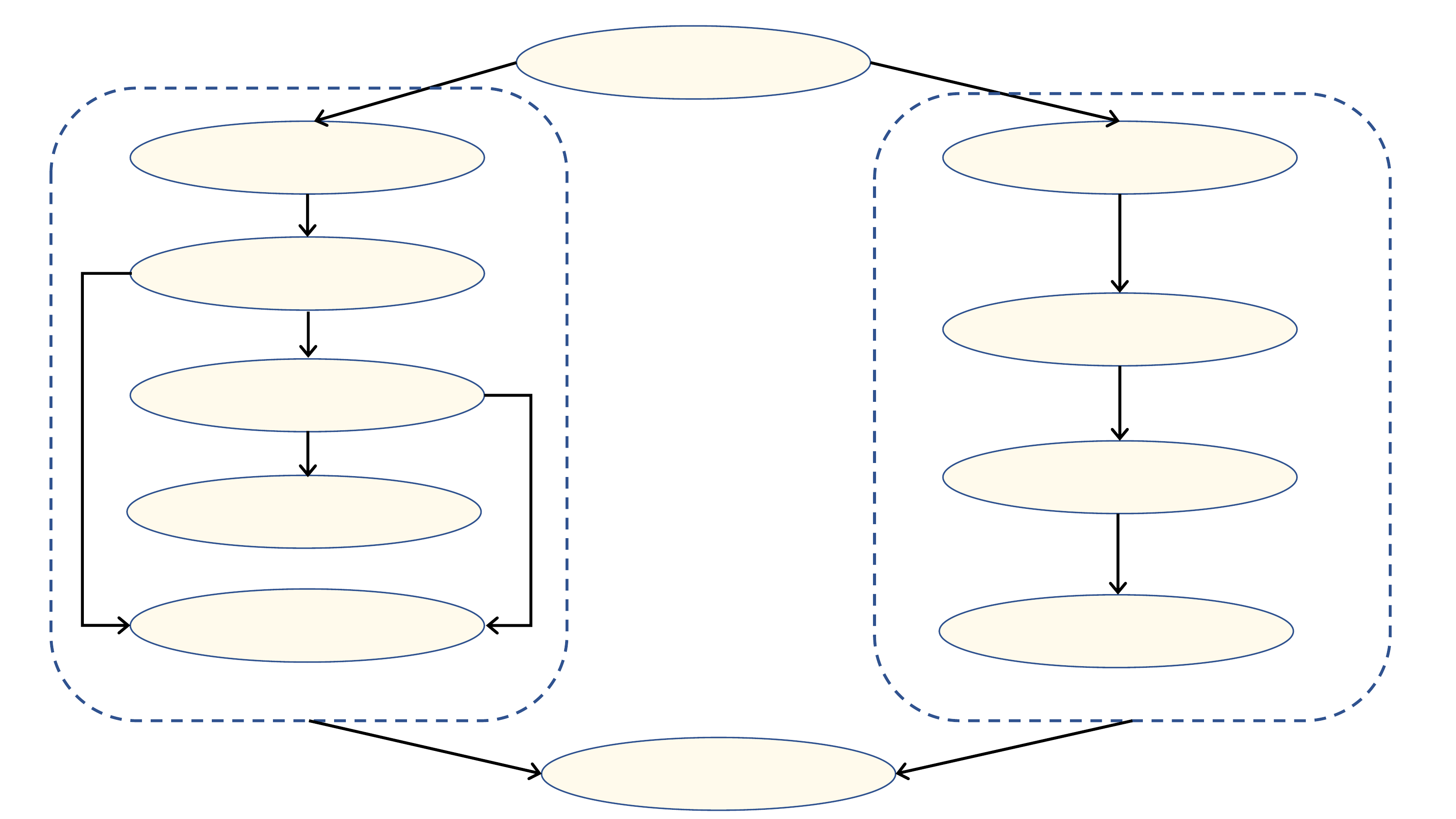}
{\footnotesize
\put(38,51){\mbox{\ref{ch:intro}. Introduction }}
\put(11.5,44.5){\mbox{\ref{ch:wynerCI}. Wyner's CI }}
\put(11.5,28.3){\mbox{\ref{ch:exact}. Exact CI }}
\put(11.5,36.5){\mbox{\ref{ch:renyi}. R\'enyi's CI }}
\put(11.5,20.5){\mbox{\ref{ch:ecs}. Ch.\ Synthesis}}
\put(11.5,12.5){\mbox{\ref{ch:nr}. NN Rank}}

\put(68,44.5){\mbox{\ref{ch:gkw}. GKW's CI }}
\put(68,32.6){\mbox{\ref{ch:NICD}. NICD}}
\put(68,22.8){\mbox{\ref{ch:Stability}. $q$-Stability}}
\put(68,12.5){\mbox{\ref{ch:funineq}. Func.\ Ineq.}}
%\put(11,12.5){\mbox{\ref{ch:nr}. NN Rank}}

\put(13,7.7){\mbox{$X-W-Y$}}
\put(68,7.7){\mbox{$U-X-Y-V$}}

\put(39.5,2.5){\mbox{\ref{ch:summary}. Open Probs.}}
}
\end{overpic}
\caption{Flowchart of the sections in this monograph (CI, NN, and NICD stand  for Common Information, Nonnegative, and Non-Interactive Correlation Distillation  respectively)}
\label{fig:flowchart}
\end{figure}

%This generalizes the normalized relative entropy criterion used to measure the discrepancy between two distributions in Wyner's original work~\cite{WynerCI}. Indeed

\subsection{Part II: Extensions of Wyner's Common Information}
We then extend and generalize Wyner's common information in Part~\ref{part:two} of this monograph, which consists of four sections. In Section~\ref{ch:renyi}, we review  the  {\em R\'enyi common information},  originally studied by the present authors~\cite{YuTan2018, yu2020corrections}. In his seminal paper~\cite{WynerCI}, Wyner used the normalized relative entropy
\begin{equation}
\frac{1}{n} D\left(P_{X^nY^n} \middle\| \pi_{XY}^n\right)=\frac{1}{n}\sum_{x^n,y^n}P_{X^nY^n}(x^n,y^n)\log\frac{P_{X^nY^n}(x^n,y^n)}{\pi_{XY}^n(x^n,y^n)}. 
\end{equation}
to quantify the discrepancy between the synthesized distribution $P_{X^nY^n}$ and the target distribution $\pi_{XY}^n $ and sought the minimum rate for distributed source synthesis for which this quantity vanishes as the blocklength $n$ grows. The R\'enyi common information \cite{YuTan2018,yu2020corrections} generalizes this to the case in which the discrepancy measures used belong to the families of normalized and unnormalized R\'enyi divergences. For R\'enyi order $1+s \in (0,\infty)\setminus\{1\}$,  the unnormalized form can be expressed as 
\begin{equation}
 D_{1+s}(P_{X^n Y^n} \| \pi_{XY}^n) = \frac{1}{s }\log\sum_{ x^n, y^n  } P_{X^nY^n}(x^n,y^n)\left( \frac{P_{X^nY^n}(x^n,y^n)}{\pi_{XY}^n(x^n,y^n)}\right)^s. 
\end{equation}

\enlargethispage{-\baselineskip}
We use this family of measures to build a bridge to the topic of discussion in Section~\ref{ch:exact}, namely, the {\em exact common information}, a quantity first defined and studied by \citet{KLE2014}; see Definition~\ref{def:ECI} for its precise definition. In contrast to the R\'enyi common information, the exact version requires that synthesized distribution
be   {\em exactly} equal to the target distribution for some blocklength $n$; however, {\em variable-length} codes are permitted.   Using an unexpected equivalence between the unnormalized R\'enyi common information of order~$\infty$  (the limit of $D_{1+s}$ as $s\to\infty$) 
\begin{equation}
 D_{\infty}(P_{X^n Y^n} \| \pi_{XY}^n) = \log\max_{(x^n,y^n): P_{X^nY^n}(x^n,y^n)>0} \frac{P_{X^nY^n}(x^n,y^n)}{\pi_{XY}^n(x^n,y^n)},
\end{equation}
 and the exact common information, we argue that the latter can be strictly larger than Wyner's common information for some sources, specifically the doubly symmetric binary source (DSBS). 

In Section~\ref{ch:ecs}, we use the preceding notions to describe the problem of {\em channel synthesis}. We review this problem in both the approximate and exact settings and show that it produces a continuum  of common information measures that interpolate from the mutual information to Wyner's or exact common information. 

In Section~\ref{ch:nr}, we describe a seemingly tangential topic in numerical linear algebra, namely  the {\em nonnegative rank} of nonnegative  matrices~\cite{Vavasis,gillisBook}. It turns out that this area of research has intimate connections to the preceding notions of common information, leading to some  interesting open problems.

%\begin{figure}
%\centering
%\begin{overpic}[width=.9\textwidth]{figs/flowgraph4.pdf}
%\put(7,34){\mbox{$\mathbf{\ref{ch:intro}}$}}
%\put(20,44){\mbox{$\mathbf{\ref{ch:wynerCI}}$}}
%\put(33,44){\mbox{$\mathbf{\ref{ch:renyi}}$}}
%\put(45,44){\mbox{$\mathbf{\ref{ch:exact}}$}}
%\put(57,44){\mbox{$\mathbf{\ref{ch:ecs}}$}}
%\put(69,44){\mbox{$\mathbf{\ref{ch:nr}}$}}
%\put(30,52){\mbox{$X-W-Y$}}
%
%\put(20,22){\mbox{$\mathbf{\ref{ch:gkw}}$}}
%\put(33,22){\mbox{$\mathbf{\ref{ch:NICD}}$}}
%\put(45,22){\mbox{$\mathbf{\ref{ch:Stability}}$}}
%\put(56,22){\mbox{$\mathbf{\ref{ch:funineq}}$}}
%\put(30,14){\mbox{$U-X-Y-V$}}
%
%{\footnotesize
%\put(68,37){\mbox{\ref{ch:intro}. Introduction }}
%\put(68,33){\mbox{\ref{ch:wynerCI}. Wyner's CI }}
%\put(68,29){\mbox{\ref{ch:gkw}. GKW's CI }}
%\put(68,25){\mbox{\ref{ch:renyi}. R\'enyi CI }}
%\put(68,21){\mbox{\ref{ch:exact}. Exact CI }}
%\put(68,17){\mbox{\ref{ch:ecs}. Channel Synthesis }}
%\put(68,13){\mbox{\ref{ch:nr}. Nonnegative Rank}}
%\put(68,9){\mbox{\ref{ch:NICD}. NICD}}
%\put(68,5){\mbox{\ref{ch:Stability}. Noise Stability}}
%\put(68,1){\mbox{\ref{ch:funineq}. Functional Inequalities}}
%}
%\end{overpic} 
%\caption{Flowchart of chapters in this monograph (CI and NICD stand  for Common Information and Non-Interactive Correlation Distillation  respectively)}
%\label{fig:flowchart}
%\end{figure} 

\subsection{Part III: Extensions of G\'acs--K\"orner--Witsenhausen's Common Information}
It is known that GKW's common information is zero for most non-pathological sources such as the doubly symmetric binary source and the bivariate Gaussian source. Consequently,  in itself,   GKW's common information does not provide any tangible quantification of how ``similar'' two sources are. The goal of Part~\ref{part:three}, which consists of three sections, is thus to consider  several  {\em refinements} of GKW's common information in which new  insights can be  readily gleaned. 

We start  in Section~\ref{ch:NICD} by  providing an extensive discussion of the  {\em  $2$-user   Non-Interactive Correlation Distillation} (NICD) problem \cite{kamath2016non,mossel2006non}. Given a pair of random vectors  $(X^n,Y^n) \sim \pi_{XY}^n$ in which each $(X_i,Y_i)$ is drawn independently from a DSBS, this problem concerns the  agreement probability   of the random bits that can be extracted from $X^n$ and $Y^n$ individually. In other words, we wish to quantify
\begin{equation}
\max_{f,g}\; \Pr \big(U=V\big)\quad \mbox{and}\quad \min_{f,g}\; \Pr \big(U=V \big),
\end{equation}
where $U=f(X^n)$ and $V=g(Y^n)$ and $f $ and $g$ are $\{0,1\}$-valued  (i.e., Boolean) functions such that the marginals $\Pr(U=1)$ and $\Pr(V=1)$ are appropriately constrained. For example, for the maximization version of the NICD problem, we place upper bounds on $\Pr(U=1)$ and $\Pr(V=1)$.   
 We quantify these agreement probabilities by studying various geometric structures such as Hamming subcubes and Hamming balls. We discuss their optimality in several asymptotic regimes (such as the central limit or large deviations regimes) using results from concentration of measure  and Boolean Fourier analysis, among other techniques.
 
 In Section~\ref{ch:Stability}, we  extend the NICD problem to the {\em multi-user} version. For the $k$-user case, there are $k$ correlated sources $X_1^n, X_2^n,\ldots, X_k^n$ that are generated independently conditioned on another source $Y^n$ such that the joint distribution of $X_i^n$ and $Y^n$ is $\pi_{XY}^n$. 
 % we are given realizations of a source $(X_1^n,\ldots, X_k^n)$ in which $(X_i^n,Y^n)$ is generated independently from a joint distribution $P_{XY}$ and $X_1^n, \ldots, X_k^n$ are conditionally independent given $Y^n$. 
 We are interested in quantifying 
 \begin{equation}
 \max_{f_1,f_2,\ldots, f_k} \; \Pr\big(U_1 = U_2 = \ldots = U_k\big),
 \end{equation}
 where $U_i = f_i(X_i^n), i = 1,2,\ldots, k$ and the maximum extends over all $k$-tuples of Boolean  functions $f_i$'s whose marginals are also constrained by placing upper bounds on  $\Pr(U_i=1)$.
  We also discuss the connection of the $k$-user NICD problem to {\em $q$-stability} \cite{eldan2015two,li2021boolean} in which the number of users $k$ is replaced by an arbitrary real number $q$.  This allows us to seamlessly segue into a review of recent  advances in  contemporary conjectures in information theory and discrete probability. These include  the Courtade--Kumar conjecture \cite{courtade2014boolean}, the Mossel--O'Donnell conjecture \cite{mossel2005coin}, and  the Li--M\'edard conjecture \cite{li2021boolean}. Mathematical tools used here include the analysis of Boolean functions \cite{ODonnell14analysisof} and, in particular, edge-isoperimetric inequalities and the study of the maximal degree-$1$ Fourier weight.
  
   In Section~\ref{ch:funineq}, we connect these notions and results to {\em functional inequalities} including  the  hypercontractivity, the logarithmic Sobolev, the Brascamp--Lieb inequalities, as well as their strengthened counterparts. This section generalizes the preceding two sections in that the Boolean functions $f_i$ are replaced by arbitrary nonnnegative functions. 
   
   The monograph is concluded in Section~\ref{ch:summary} in which we summarize   open problems in this fascinating area of study. 

The common theme in  Part~\ref{part:two}  is the Markov chain $X-W-Y$; this corresponds to the constraint that defines Wyner's common information in~\eqref{eqn:wyner_ci_intro}. In contrast, in Part~\ref{part:three}, we focus   on the Markov chain  $U-X-Y-V$; this  corresponds to the Markov chain in the NICD problem in which $U=f(X^n) $ and $V=g(Y^n)$ for some Boolean functions $f$ and~$g$. It is also present in GKW's  common information. At first glance, this appears to be different from the constraint in \eqref{eqn:gkw_ci_intro}. However, this constraint   is merely a special case of   $U-X-Y-V$ by taking~$U$ and~$V$ to be deterministic functions of $X$ and $Y$ respectively such that they are also constrained to be equal almost surely.  %We will see from Part~\ref{part:three} that the generality that having two auxiliaries $U$ and $V$ will be beneficial in aiding our understanding of refinements of GKW's common information.

\section{Notation}
To appreciate the material in this monograph, the reader is expected to have
some   background in information theory at the level of \citet{Cov06}. We will also make  frequent use
of the method of types, for which an  excellent exposition   can be found
in \citet{Csi97}. 
%The measure-theoretic foundations of probability will 
%generally not be needed.

In this monograph, we generally follow the notation in \citet{Cov06}, \citet{elgamal}, and \citet{Csi97}.

\subsection{Random Variables and Probability Distributions} \label{sec:rvs}
Random variables and their realizations are denoted by upper case letters (such as $X$ and~$Y$) and  lower case letters (such as $x$ and $y$) respectively. The sets of values that the realizations take on, also called {\em alphabets}, are denoted by calligraphic letters such as $\calX$ and $\calY$. We use $P_X, \tilP_X, Q_X$, and  $\pi_X$ to
denote various probability distributions on alphabet  $\calX$. If a random variable $X$ is distributed according to $P_X$,   we write $X\sim P_X$.
%We use $P_X$ to denote the probability distribution of a random variable $X$; this is also written as $X\sim P_X$. 
As we work with both discrete and continuous random variables in this monograph, we will often have to distinguish between probability mass functions (PMFs) for discrete random variables and probability density functions (PDFs) for continuous random variables. %For ease of notation, however, for any random variable, we use $x\in\calX
If $X$ is discrete, we use $x\in \calX \mapsto P_X(x)$ to denote its PMF. The PDF of a  (real-valued) continuous random variable is denoted as 
%the Radon--Nikodym derivative  
$f_X: x\in\bbR\mapsto  (\rmd P_X/ \rmd \mu )(x)$, where $\mu$ is the Lebesgue measure on $\bbR$.  These will also be denoted as $P$ or $f$
%$\rmd P/ \rmd \mu$ %is will also be denoted as $\{P(x):x\in\calX\}$, or simply $P$, 
when the random variable $X$ is clear from the context. %We also use $\pi_X, \tilP_X, \hatP_X$, and $Q_X$ to denote various probability distributions on alphabet $\calX$. 
Throughout the monograph, the notations $\pi_X$ and $\pi_{XY}$ are   reserved for {\em target} and {\em source distributions}. % in Part~\ref{part:two}. 

The set of PMFs on $\calX$ is denoted as $\calP(\calX )$ and the set of conditional PMFs  on $\calY$ given a variable taking values in $\calX$ is denoted as $\calP(\calY|\calX) = \{P_{Y|X} : P_{Y|X}(\cdot |x) \in\calP(\calY),x\in\calX\}$. The joint distribution induced by  $P_X\in\calP(\calX)$ and $P_{Y|X}\in\calP(\calY|\calX)$ is denoted as $P_XP_{Y|X} \in\calP(\calX\times\calY)$. The {\em support} of a discrete distribution is denoted as $\supp(P) := \{x\in\calX : P(x)>0\}$. Given an input distribution $P_X\in\calP(\calX)$ and a conditional distribution $P_{Y|X}\in\calP(\calY|\calX)$, if the induced  output distribution is $P_Y(y) = \sum_x P_X(x)P_{Y|X}(y|x)$ (for the discrete case), we write this as $P_X\rightarrow P_{Y|X}\rightarrow P_Y$.  For two distributions $P$ and $Q$ (defined on the same measurable space), we use $P \ll Q$ to denote that $P$ is {\em absolutely continuous}  with respect to $Q$. %If $P\ll Q$, we use ${\rmd P}/{\rmd Q}$ to denote the Radon-Nikodym derivative of $P$ with respect to $Q$.   
In the finite alphabet case, $P\ll Q$ means that for every $x\in\calX$ such that $Q(x)=0$, it   holds that $P(x)=0$. 

We say that three  random variables $X, Y$, and $Z$ {\em form a Markov chain in this order} if $X$ and $Z$ are conditionally independent given $Y$.  In this case, we write $X - Y- Z$.  For discrete random variables, $X - Y- Z$ if and only if  $P_{XYZ}(x,y,z)= P_{Y}(y)P_{X|Y}(x|y) P_{Z|Y}(z|y)$ for all $(x,y,z)\in\calX\times\calY\times\calZ$.  As is customary in information theory, for two integers $m$ and $n$, we write $X_m^n$ to mean the random vector $(X_m, X_{m+1}, \ldots, X_n)$; when $m=1$, this is abbreviated to $X^n$. A particular realization of $X^n$, a deterministic vector, is denoted as $x^n = (x_1,x_2,\ldots, x_n)$.  We denote the {\em $n$-fold product distribution} of $P$ as $P^n$, which is defined by the formula $P^n(x^n)=\prod_{i=1}^n P(x_i)$ for all $x^n\in\calX^n$.

\enlargethispage{-\baselineskip}
A \emph{stationary memoryless source}, denoted by $X \sim P_X\in\calP(\calX)$, is a discrete-time stochastic process $\{X_i \}_{i\in\bbN}$ such that $X_i$'s are independent   copies of $X$. We also denote a source $X$ by its distribution $P_X$. We use $X^n$ to denote the first $n$ random variables in the stochastic process $\{X_i \}_{i\in\bbN}$. With a slight abuse of terminology, $X^n$ is also  called a \emph{source sequence} of the source $X$. 
A \emph{stationary memoryless channel}, denoted by $ P_{Y|X}\in\calP(\calY|\calX)$, is a random transformation  that outputs a length-$n$ random vector $Y^n \sim  P_{Y|X}^n(\cdot|x^n)$ if the input is the length-$n$   vector $x^n \in \calX^n$. Since we deal almost exclusively with stationary memoryless sources and channels in this monograph, we will  omit the term ``stationary memoryless'' when we mention sources and channels.

We will work mainly with three types of random variables in this monograph. A discrete {\em uniform} random variable $X$ takes equal probabilities on its support $\calX$ and its probability distribution is denoted as $\mathrm{Unif}(\calX)$.   A {\em Bernoulli} random variable $X$ is one with support $\{0,1\}$. Its probability distribution is abbreviated as $\mathrm{Bern}(a)$ if $\Pr(X=1)=a$. A ($d$-dimensional) {\em normal} or {\em Gaussian} random variable or vector $X$ has a PDF that is denoted by 
\begin{equation}
\bx\in\bbR^d\mapsto\calN(\bx;\bm{\mu},\bm{\Sigma}) = \frac{1}{\sqrt{(2\pi)^d \ \mathrm{det}(\bm{\Sigma})} }\exp\Big(-\frac{1}{2}(\bx- \bm{\mu})^\top\bm{\Sigma}^{-1} (\bx- \bm{\mu})\Big),
\end{equation}
 (or simply $\calN(\bm{\mu},\bm{\Sigma})$) where $\bm{\mu}$ and  $\bm{\Sigma}$  are the mean vector and the covariance matrix respectively.

\subsection{Types or Empirical Distributions}
We will often use the method of types~\cite{Csi97}  in our calculations, especially for finite alphabets. Given a sequence $x^n \in\calX^n$, we use 
\begin{equation}
T_{x^n} (a) :=\frac{1}{n}\sum_{i=1}^n \mathbbm{1}\{x_i=a\}\quad  \mbox{for all} \;\, a\in\calX
\end{equation}
to denote its {\em type} or {\em empirical distribution}. The type of a length-$n$ sequence will be denoted by~$T$ or $T_X^{(n)}$ depending on the context. The set of all sequences with type $T$ is denoted as $\calT_T \subset\calX^n$. This is known as the {\em type class} of $T$. The set of all types that can be formed from sequences of length $n$ taking values in alphabet $\calX^n$ is denoted as $\calP_n(\calX)$, which is a subset of the probability simplex $\calP(\calX)$. %We say that $y^n\in\calY^n$ has {\em conditional type} $V$ given $x^n\in\calX^n$ if $T_{x^n,y^n}(x,y) = T_{x^n}(x)V(y|x)$ for all $(x,y)\in\calX\times\calY$. For a given $x^n\in\calX^n$ and a stochastic matrix $V:\calX\to\calY$, the set of sequences $y^n$ having conditional type $V$ given $x^n$ is called the {\em $V$-shell} of $x^n$ and is denoted as $\calT_V(x^n)$.  The $V$-shell is nonempty only if $T_{x^n} V \in\calP(\calX\times \calY)$ is a type with denominator $n$. The set of stochastic matrices for which $\calT_V(x^n)$ is nonempty for some $x^n\in\calT_T$ is denoted as $\calV_n(T)$.  For brevity,  we   use $T (x, y)$ to denote the joint distribution  $T (x)V (y|x)$ or, equivalently,  $T (y)V (x|y)$. 

\subsection{Information Measures}
We now recap the necessary information measures used in this monograph. For $X\sim P_X$, we denote its {\em Shannon entropy} as % or simply the {\em entropy}  of $X\sim P_X$ as 
\begin{equation}
H(X)=H_P(X) = H(P_X) :=-\sum_{x\in \mathrm{supp}(P_X)} P_X(x)\log P_X(x). \label{eqn:shan_ent}
\end{equation}
All logarithms are to the base $2$ unless otherwise specified. For $(X,Y)\sim P_{XY}$, we denote the {\em conditional entropy} of $X$ given $Y$ as 
\begin{align}
H(X|Y) &= H_P(X|Y) = H(P_{X|Y}|P_Y)  \nn\\
&:=-\sum_{y\in\calY} P_Y(x)\sum_{x\in \mathrm{supp}(P_{X|Y} (\cdot|y)) } P_{X|Y}(x|y)\log P_{X|Y}(x|y).  \label{eqn:cond_ent}
\end{align}
The {\em mutual information} between $X$ and $Y$ where $(X,Y)\sim P_{XY}$ is denoted as 
\begin{equation}
I_P(X;Y) = I(P_X, P_{Y|X}) := H_P(X) - H_P(X|Y). \label{eqn:mi}
\end{equation}
The subscripts in $H_P$ and $I_P$ are used to emphasize the distribution of $(X,Y)$ under which these information measures are computed. When the distribution is clear from the context, the subscripts will be omitted.  The {\em relative entropy} or {\em Kullback--Leibler divergence} between two distributions $P_X$ and $Q_X$ defined on the same (countable) alphabet is\footnote{This definition is only applicable when the alphabets are countable, and the convention $x/0=\infty $ for $x>0$ is adopted.  For $P_X$ and $Q_X$ defined on a general probability  space,  the ratio ${P_X}/{Q_X}$ should be replaced with the 
Radon--Nikodym derivative ${\mathrm{d}P_{X}}/{\mathrm{d}Q_{X}}$ (if $P_{X}\ll Q_{X}$), and the expectation with respect to $P_{X}$  should be  written as a  Lebesgue integral over $\calX$.  If $P_{X}$ is not absolutely continuous with respect to $Q_{X}$, $D(P_X\|Q_X)$ is defined to be~$+\infty$. In the following, for simplicity, we only provide definitions of information-theoretic quantities for countable alphabets.}
\begin{equation}
D(P_X\| Q_X) :=\sum_{x\in\mathrm{supp}(P_X)} P_X(x)\log\frac{P_X(x)}{Q_X(x)}.
\end{equation}
The {\em conditional relative entropy} of two conditional distributions $P_{Y|X}$ and $Q_{Y|X}$, given a distribution $P_X$, is 
\begin{equation}
D(P_{Y|X} \| Q_{Y|X}|P_X) := D(P_X P_{Y|X} \| P_X Q_{Y|X}) . \label{eqn:cond_RE}
\end{equation}

In addition to the Shannon information measures above, we need to recap the family of R\'enyi information measures \cite{renyi1959measures,Erven}  as this is central to the majority of our discussion in this monograph. For two distributions $P_X, Q_X\in\calP(\calX)$ on a countable set $\calX$, the {\em R\'enyi divergence} of order $1+s\in (0,1)\cup(1,\infty)$ is 
\begin{equation}
D_{1+s}(P_X\|Q_X) :=\frac{1}{s}\log\sum_{x\in\mathrm{supp}(P_X)}P_X(x) \left(\frac{P_X(x)}{Q_X(x)} \right)^{s}. \label{eqn:renyi_div}
\end{equation} 
The R\'enyi divergence is monotonically nondecreasing in its order. 
Sibson's~\cite{sibson1969information} version of the  {\em conditional R\'enyi divergence} between two conditional distributions $P_{Y|X}$ and $Q_{Y|X}$ given a distribution $P_X$ is 
\begin{equation}
D_{1+s}(P_{Y|X}\| Q_{Y|X} | P_X) := D_{1+s}(P_X P_{Y|X} \| P_X Q_{Y|X}) .\label{eqn:cond_renyi_div}
\end{equation}
We note that while the conditional relative entropy in \eqref{eqn:cond_RE} is the expectation of\\
$D(P_{Y|X}(\cdot |X) \| Q_{Y|X}(\cdot |X) )$ over $X\sim P_X$, the conditional R\'enyi divergence  in~\eqref{eqn:cond_renyi_div} depends on $D_{1+s}(P_{Y|X}(\cdot |X) \| Q_{Y|X}(\cdot |X) )$ in a more involved way; indeed, it is a generalized mean of the random variable  $D_{1+s} ( P_{Y |X} (\cdot |X ) \| Q_{Y |X }( \cdot |X ))$ evaluated at $s$. For a more detailed discussion on this point, the reader is referred to \citet{caiVerdu}. We also note that there are other definitions of the conditional R\'enyi divergence but we will use the definition in \eqref{eqn:cond_renyi_div} in this monograph; see \cite{sibson1969information, csiszar1995generalized, bleuler20conditional}. The R\'enyi divergence and its conditional version in \eqref{eqn:cond_renyi_div} can be extended to all orders $1+s \in \{0,1,\infty\}$ by taking the appropriate limits. In particular, when $s\to 0$, we recover the usual relative entropy. An order of the R\'enyi divergence that will be of particular interest to us in this monograph is the {\em R\'enyi divergence of order $\infty$}. This is the divergence we obtain when we let $s\to \infty$, i.e., 
\begin{equation}
D_\infty (P_X\| Q_X) := \log \sup_{x\in\mathrm{supp}(P_X)}\frac{P_X(x)}{Q_X(x)}. \label{eqn:max_div}
\end{equation}

The {\em R\'enyi entropy} of order $1+s \in (0,1) \cup (1,\infty)$ of a probability mass function $P_X \in \calP(\calX)$ is defined as 
\begin{equation}
H_{1+s}(P_X) = -\frac{1}{s}\log \sum_{x\in\mathrm{supp}(P_X)} \big(P_X(x)\big)^{1+s}.\label{eqn:renyi_ent2}
\end{equation}
It is easy to check that 
\begin{align}
H_{1+s}(P_X) := \log |\calX| - D_{1+s}(P_X \| \mathrm{Unif}(\calX) ). \label{eqn:renyi_ent} 
\end{align}
%where $U_\calX = $. 
Similarly to the R\'enyi divergence, we define $H_0(P_X)$ and $H_\infty(P_X)$ as the limits of $H_{1+s}(P_X)$ as $s\downarrow -1$ and $s\to\infty$ respectively. These are known as the {\em max-entropy} and {\em min-entropy} respectively. Of special importance is the case when $s\to 0$, in which case $H_{1+s}(P_X)$ reduces to the Shannon entropy defined in \eqref{eqn:shan_ent}. Since the relation in \eqref{eqn:renyi_ent}  holds and the R\'enyi divergence is nondecreasing in its order, the R\'enyi entropy is nonincreasing in its order. 
% The Shannon entropy $H(P_X) $ is recovered when $s\to 0$. In addition, 
%\begin{align}
%H_0(P_X) := \lim_{s\downarrow -1 }H_{1+s}(P_X)= \log|\calX|
%\end{align}

We need one additional measure of the discrepancy between two distributions. The {\em total variation distance} or simply the {\em TV distance} is defined for two distributions $P$ and $Q$ on a common (countable) alphabet $\calX$ as 
\begin{equation}
|P-Q|:=\frac{1}{2}\sum_{x\in\calX}|P(x)-Q(x)|.
\end{equation}
More generally, $|P-Q|= \sup_{\calA\subset \calX} |P(\calA)-Q(\calA)|$, 
where $\calA$ runs over all (measurable) subsets of $\calX$. Pinsker's inequality yields the following bound on the TV distance in terms of the relative entropy 
\begin{equation}
|P-Q|^2\le \frac{\ln 2}{2}\cdot D(P\|Q).\label{eqn:pinsker}
\end{equation}

\subsection{Typical Sets}
In our achievability proofs, we will often need to use the notion of {\em typical sets}~\cite{Cov06,elgamal,OrlitskyRoche}.  The {\em $\epsilon$-strongly typical set} with respect to a distribution $P_X\in\calP(\calX)$ is defined as 
\begin{equation}
\calT_\epsilon^{(n)}(P_X):= \Big\{ x^n\in\calX^n: \big|T_{x^n}(x) -  P_X(x)\big|\le \epsilon P_X(x),\forall\, x\in \calX \Big\}.
\end{equation}
This notion of typicality, proposed by \citet{OrlitskyRoche}, is also commonly known as {\em robust typicality} and is convenient for coding problems with cost constraints or rate-distortion problems. However, it suffers from the deficiency that it is amenable  only to {\em finite} alphabets. This is mitigated by the availability of the {\em $\epsilon$-weakly typical set} with respect to a distribution $P_X\in\calP(\calX)$, which is defined as 
\begin{equation}
\calA_\epsilon^{(n)}(P_X) := \bigg\{ x^n\in\calX^n: \Big|\frac{1}{n}\log\frac{1}{P_X^n(x^n)}-H(P_X) \Big|<\epsilon \bigg\}.
\end{equation}
When $X$ is a continuous random variable, $H(P_X)$ is to be replaced by the {\em differential entropy} of $X$ \cite{Cov06}. The conditional versions of these sets can be defined in a natural manner, e.g., the {\em conditionally $\epsilon$-strongly typical set} of $Y$ given a sequence $x^n\in\calX^n$ is 
\begin{equation}
\calT_\epsilon^{(n)}(P_{XY}|x^n):=\left\{y^n\in\calY^n: (x^n,y^n)\in\calT_\epsilon^{(n)}(P_{XY}) \right\}.
\end{equation}

\subsection{Asymptotic Notations}
Asymptotic notation is used in the usual way~\cite{Cor03}. Given two real-valued sequences $\{a_n\}_{n=1}^\infty \subset\bbR$ and $\{b_n\}_{n=1}^\infty\subset\bbR$, we say that  $a_n=O(b_n)$ if $\limsup_{n\to\infty} |a_n/b_n|<\infty$,  $a_n=\Omega(b_n)$ if $\liminf_{n\to\infty} |a_n/b_n|>0$, and $a_n=\Theta(b_n)$ if $a_n=O(b_n)$ and $a_n=\Omega(b_n)$. Similarly, $a_n=o(b_n)$ if $\lim_{n\to\infty} |a_n/b_n|=0$. Finally, if $\{a_n\}_{n=1}^\infty $ and $\{b_n\}_{n=1}^\infty$ are positive sequences,   we write $a_n\doteq b_n$ if these sequences are {\em equal to first-order in the exponent}~\cite{Cov06}, i.e., $\lim_{n\to\infty} n^{-1}\log (a_n/b_n)=0$. % Finally, $a_n\sim b_n$ if $\lim_{n\to\infty}|a_n/b_n|=1$ and $a_n\doteq b_n$ if $\lim_{n\to\infty} n^{-1}\log(a_n/b_n)=0$.  

\subsection{Miscellaneous}
%The sets of real and natural numbers is denoted by $\bbR$ and $\bbN$ respectively. 
For two integers $m$ and $n$, we write $[m:n] = \{m,m+1,\ldots, n\}$ to denote the discrete interval. When $m=1$, this is abbreviated  as $[n]$. Often, for an $R>0$,  we write  $[2^{nR}]$ to refer to the set $\{1,2,\ldots, \lfloor 2^{ nR}\rfloor \}$. Given a number $a\in [0,1]$, we write $\overline{a}:=1-a$.  Given two numbers $a,b\in [0,1]$, we write $a\ast b:=a\barb+b\bara$ to denote their binary convolution. We write $[a]^+$ to mean $\max\{a,0\}$ for $a\in\bbR$.  For two bits $a,b\in\{0,1\}$, $a\oplus b$ denotes the binary addition (modulo-$2$ sum) operation, i.e., $a\oplus b=0$ if $a=b$ and~$1$ otherwise.  Logarithms are always to the base~$2$ unless otherwise specified. When we write $\ln$, we are referring to the natural logarithm (to base $\rme$). 

Vectors are interchangeably denoted by boldface lower case font  (e.g.,~$\bu$) or, as mentioned   in Section~\ref{sec:rvs}, with a lower case letter and with a superscript indicating its length (e.g., $u^n = (u_1,u_2,\ldots, u_n)$). Matrices (e.g., $\bM$) are denoted in boldface  upper case font. The $i^{\mathrm{th}}$ element of a vector $\bu$ is denoted interchangeably as $u_i$ or $[\bu]_i$. Similarly, the $(i,j)^{\mathrm{th}}$ element of a matrix $\bM$ is denoted interchangeably  as $M_{i,j}$ or $[\bM]_{i,j}$. 

\section{Mathematical Tools}
% \chapter{Mathematical Tools}
%We summarize several useful mathematical results from the method of types and couplings. %These will turn out to be particularly useful in our discussions and analyses.

% \textcolor{red}{Lei: Should we move this whole chapter to the end of Chapter 1 (or a new Chapter 2) so that we consolidate all notation and mathematical tools up front even before we talk about the Wyner CI and GKW CI? We should also have a section on notation before we introduce techniques. }
\subsection{The Method of Types}
We summarize a few key property of types which will turn out to be useful in proving both achievability and converse parts of various common information problems, particularly those with finite alphabets. For an extensive discussion, the reader is referred to the book by \citet{Csi97}. %We  collate some properties that we use extensively here. 

First, the number of types $|\calP_n(\calX)| \le (n+1)^{|\calX|}$ is polynomial in~$n$. Second, for a given type $T\in\calP_n(\calX)$, the size of the type class $(n+1)^{-|\calX|}2^{n H(T)}\le|\calT_{T}|\le 2^{n H(T)}$ is related to the entropy of the type $H(T)$. Third, the $Q^n$-probability of a sequence $x^n \in\calT_T$ is $Q^n(x^n) = 2^{-n(D(T \| Q) + H(T))}$. Consequently, the $Q^n$-probability of the type class $\calT_T$ is bounded as $(n+1)^{-|\calX|}2^{-nD(T\|Q)}\le Q^n(\calT_T)\le 2^{-nD(T\|Q)}$. 

%The  properties stated in the previous paragraph extend naturally to their conditional counterparts. For example $|\calV_n(T)|\le (n+1)^{ |\calX| |\calY|}$ and for $x^n\in\calT_T$ and two stochastic matrices $W$ and $V$, the $W^n$-probability of the $V$-shell $\calT_V(x^n)$ given $x^n \in\calT_T$ is $W^n(\calT_V(x^n) |x^n)\le 2^{-n D(V\|W|T)}$. 

A particularly useful result that we use repeatedly in  Part~\ref{part:three} of the monograph is {\em Sanov's theorem}~\cite{Sanov61,Dembo, Cov06}, so we review it here. 
\begin{theorem}[Sanov's theorem] \label{thm:sanov}
Let the components of  the random vector $X^n = (X_1, X_2, \ldots, X_n)$ be generated in an independently and identically distributed (i.i.d.) manner from a PMF $Q \in \calP(\calX)$. For any $n\in\bbN$ and any set of distributions $\calS \subset \calP(\calX)$, 
\begin{equation}
Q^n\big(\{x^n: T_{x^n}\in  \calS \}\big) \le (n+1)^{|\calX|}2^{-nD(P^*\|Q)},
\end{equation}
where  the {\em information projection} of $Q$ onto $\calS$ is any distribution $P^*$ that satisfies
\begin{equation}
D(P^*\|Q)=\inf_{P\in\calS} D(P\|Q).
\end{equation}
If additionally, $\calS$ is equal to the   closure of its interior (under the relative topology),\footnote{This regularity condition will always be satisfied in the sections to follow.}
\begin{equation}
\liminf_{n\to\infty} -\frac{1}{n}\log Q^n\big(\{x^n: T_{x^n}\in  \calS \}\big)\ge D(P^*\|Q),
\end{equation}
%Consequently, if $\calS$ is equal to the closure of its interior,
and hence,
\begin{equation}
Q^n\big(\{x^n: T_{x^n}\in  \calS \}\big)\doteq  2^{-nD(P^*\|Q)}.
\end{equation}
\end{theorem}
Sanov's theorem basically  says that the exponent of the probability that the type $T_{X^n}$ of a random sequence $X^n\sim Q^n$  belongs to a set $\calS$ is dominated by the relative entropy between the information projection  of $Q$ onto $\calS$ and~$Q$.

\subsection{Couplings} \label{sec:coupling}
In this monograph, we will often encounter the optimization problems over joint distributions for which their marginals are fixed. Such a joint distribution is known as a coupling. More precisely, a {\em coupling} $P_{XY}$ of two distributions $Q_X \in \calP(\calY)$ and $Q_Y \in \calP(\calY)$ is a joint distribution on $\calX\times \calY$ whose $X$- and $Y$-marginals are respectively $Q_X$ and $Q_Y$. The set of all couplings with marginals $Q_X$ and $Q_Y$ is denoted as 
\begin{equation}
\calC(Q_X, Q_Y) := \big\{P_{XY} \in \calP(\calX\times\calY): P_X = Q_X , P_Y=Q_Y\big\}. 
\end{equation}
%A coupling $Q_{XY}$ of two distributions $Q_{X}$ and $Q_{X}$ is
%a joint distribution on $\mathcal{X}\times\mathcal{Y}$ whose marginals
%are respectively equal to $Q_{X}$ and $Q_{X}$. 
Similarly, a conditional coupling $P_{XY|W}$ is a joint conditional distribution whose  $X$- and $Y$-marginals agree with  given marginals $Q_{X|W}$ and $Q_{Y|W}$ respectively. The set of all  conditional couplings with marginals $Q_{X|W}$ and $Q_{Y|W}$ is
\begin{align}
&\calC(Q_{X|W}, Q_{Y|W}) \nn\\*
&\quad:= \big\{P_{XY|W} \in \calP(\calX\times\calY|\calW): P_{X|W} = Q_{X|W} , P_{Y|W}=Q_{Y|W}\big\}.
\end{align}
%coupling $Q_{XY|}$ of two (regular) conditional distributions
%$Q_{X|UW}$ and $Q_{Y|VW}$ is a (regular) conditional distribution
%such that for each $(u,v,w)$, $Q_{XY|UVW}(\cdot|u,v,w)$ is a coupling
%of $Q_{X|UW}(\cdot|u,w)$ and $Q_{Y|VW}(\cdot|v,w)$. We use $\calC(Q_{X},Q_{Y})$
%to denote the set of couplings of $Q_{X},Q_{Y}$, and $\calC(Q_{X|UW},Q_{Y|VW})$
%to denote the set of conditional couplings $Q_{XY|UVW}$ with conditional
%marginals $Q_{X|UW},Q_{Y|VW}$. Note that for any $Q_{XY|UVW}\in\calC(Q_{X|UW},Q_{Y|VW}),$
%its marginals satisfy $Q_{X|UVW}=Q_{X|UW},Q_{Y|UVW}=Q_{Y|VW}$, i.e.,
%under the conditional distribution $Q_{XY|UVW}$, $X\to(U,W)\to V$
%and $Y\to(V,W)\to U$ hold.
%
%\textcolor{red}{change coupling sets to $\calC$; change Markov chain
%to $X-Y-Z$ instead of the arrow version for consistency?}

Couplings have many beautiful properties, but we will not elaborate on them   in this monograph; see \citet{thorisson2000coupling} or \citet{YuTan2019} for example. One property that is quite remarkable is the {\em maximal coupling equality}   which says that given two distributions $Q_X$ and $Q_Y$, the total variation distance between them is equal to the probability that $X$ is not equal to $Y$ minimized over all couplings induced by $Q_X$ and $Q_Y$, i.e., 
\begin{equation}
\min_{P_{XY} \in \calC(Q_X,Q_Y)}\Pr(X\ne Y) = |Q_X-Q_Y|.
\end{equation}
A generalization of the maximal coupling equality that turns out to be useful in the GKW common information problem (Section~\ref{ch:gkw}) is stated as follows. This lemma is due to  the present authors \cite{YuTan2019}.

\begin{lemma}[Maximal guessing coupling equality]  
\label{thm:max_guess_coup} Given two distributions $Q_{X}$ and $Q_{Y}$,
we have 
\begin{equation}
\min_{P_{XY} \in \calC(Q_X,Q_Y)}\min_{f:{\cal X}\to{\cal Y}}\Pr \big( Y\ne f(X) \big)=\min_{f : \calX\to\calY} \big|Q_{Y}-Q_{f(X)}\big|.\label{eq:max_guess_coup}
\end{equation}
\end{lemma}
The minimization problem on the left-hand side  of \eqref{eq:max_guess_coup}  is termed the {\em maximal guessing coupling
problem} (because we would like to maximize the probability that $Y$ is guessed correctly by $f$ acting on $X$). The minimization problem on the right-hand side  is a classical problem
in information theory which is termed the {\em distribution approximation} or {\em random number generation}
problem \cite[Chapter~2]{Han10}. Lemma~\ref{thm:max_guess_coup}  implies that the maximal guessing coupling
problem is equivalent to the distribution approximation problem.

%Generalizations and extensions of this relation, e.g., to the asymptotic case, can be found in \citet{YuTan2019}.

The concept of coupling is naturally involved when we study a problem
involving Markov chains, e.g., Wyner's common information and its extensions. One
key step to analyze such problems is to simplify  multi-letter expressions that involve optimizations over couplings to single-letter ones.
This is conveniently facilitated by the {\em chain rule}   on couplings. Before stating 
this, we first define the \emph{product coupling set}
%note from the definition of the conditional coupling set that  %the {\em Cartesian product of conditional coupling sets} can be written as
\begin{align}
 & \prod_{i=1}^{n}\calC(Q_{X_{i}|X^{i-1}W},P_{Y_{i}|Y^{i-1}W}):=\Bigg\{\prod_{i=1}^{n}P_{X_{i}Y_{i}|X^{i-1}Y^{i-1}W}:\nonumber \\*
 & \qquad P_{X_{i}Y_{i}|X^{i-1}Y^{i-1}W}\in\calC(Q_{X_{i}|X^{i-1}W},Q_{Y_{i}|Y^{i-1}W}),i\in[n]\Bigg\}.
\end{align}

\begin{lemma}[Chain Rule for Coupling Sets] \label{lem:coupling}
For any pair of conditional distributions $(Q_{X^{n}|W},Q_{Y^{n}|W})$,
we have 
\begin{equation}
\prod_{i=1}^{n}\calC(Q_{X_{i}|X^{i-1}W},Q_{Y_{i}|Y^{i-1}W})\subset\calC(Q_{X^{n}|W},Q_{Y^{n}|W}).
\end{equation}
\end{lemma} This lemma can be interpreted as follows. By the usual chain rule for joint distributions,
the conditional distributions $Q_{X^{n}|W}$ and $Q_{Y^{n}|W}$ can
be factorized as $\prod_{i=1}^{n}Q_{X_{i}|X^{i-1}W}$ and $\prod_{i=1}^{n}Q_{Y_{i}|Y^{i-1}W}$ respectively.
Let $P_{X_{i}Y_{i}|X^{i-1}Y^{i-1}W}$ be a coupling of each pair of component conditional distributions $(Q_{X_{i}|X^{i-1}W},Q_{Y_{i}|Y^{i-1}W})$.
%let $P_{X_{i}Y_{i}|X^{i-1}Y^{i-1}W}$ be a coupling of them. 
Then,
this lemma says that the product of $P_{X_{i}Y_{i}|X^{i-1}Y^{i-1}W}$
forms a coupling of the product of $Q_{X_{i}|X^{i-1}W}$ and the product
of $Q_{Y_{i}|Y^{i-1}W}$.

The proof  of this lemma can be found in \citet{YuTan2020_exact}.

\part{Classic Common Information Quantities} \label{part:one}
\chapter{Wyner's Common Information}  \label{ch:wynerCI}

What constitutes a meaningful notion of the {\em common information} between  two random variables $X$ and $Y$? As mentioned in the Introduction, there are at least two such notions that have gained traction in the information theory community as well as adjacent communities such as theoretical computer science and cryptography. In this section, we   focus on {\em Wyner's common information}~\cite{WynerCI}. To motivate this fundamental quantity,  let us consider the special case in which $X$ and $Y$ can be written as $X=(\tilX,V)$ and $Y=(\tilY,V)$ where $\tilX$, $\tilY$ and $V$ are independent. It seems natural to define the amount of common information between $X$ and $Y$ as the entropy $H(V)$ of the common part they share, namely~$V$. Taking this idea (much) further is the subject of the current and later sections (in Part~\ref{part:two}).

We review the notion of  Wyner's common information from two seemingly disparate information processing tasks. We show that these perspectives are, somewhat surprisingly, equivalent. In Section~\ref{sec:dist_sim}, we consider the scenario in which one would like to {\em simulate} a joint distribution $\pi_{XY}$ given a single source of common randomness. The minimum amount of common randomness to obtain an asymptotically exact reconstruction of $\pi_{XY}$ constitutes   Wyner's common information between~$X$ and $Y$. The perspective concerning simulation of random variables is the common thread throughout the monograph. Nevertheless, we find it useful to  provide a complementary  perspective of Wyner's common information by revisiting the {\em Gray--Wyner   source coding} problem in Section~\ref{sec:wyner-GW}. In this problem, Wyner's common information is the minimum common rate $R_0$ such that the sum of the two private rates $R_1$ and $R_2$ and the common rate $R_0$ is constrained to be almost equal to the joint entropy of the source $H(XY)$. We evaluate   Wyner's common information for the doubly symmetric binary source (DSBS) and the symmetric binary erasure source (SBES) in Sections~\ref{sec:dsbs} and~\ref{sec:sbes} respectively. 

Moving on to more contemporary topics, in Section~\ref{sec:cts},  we discuss the subtleties and techniques to extend  Wyner's common information to continuous sources, allowing us to evaluate it for jointly Gaussian random variables. Finally,  in Section~\ref{sec:gen_wyner}, we discuss   several recent extensions and applications of Wyner's common information. % and its application to caching. 

\section{Distributed Simulation of a Target Joint Distribution} \label{sec:dist_sim}
How much common randomness is needed to simulate 
a joint source 
%a pair of discrete correlated  sources 
 $(X,Y)\sim \pi_{XY}$ in a distributed fashion? This problem, as depicted in Fig.~\ref{fig:source_sim} and termed {\em distributed source simulation},   was first  studied by \citet{WynerCI} in his celebrated paper on common information. In this problem, there is a {\em target distribution} $\pi_{XY}$ and we would like to use a uniform random variable and two distributed {\em processors} to approximate the product distribution $\pi_{XY}^n$ to an arbitrary precision as the number of copies $n$ of the target distribution   tends to infinity. What is the minimum cardinality (or rate) of the support of the uniform random variable such that this is achievable? 

Before turning to formal definitions and results, let us revisit the simple example in which $X=(\tilX,V)$ and $Y=(\tilY,V)$ for some tuple of independent random variables $\tilX$, $\tilY$ and $V$. Clearly, one can use a lossless source code to encode $V^n$ by a binary string  of length approximately $nH(V)$. This binary string is then sent through the processors. Shannon's lossless source coding theorem tells us that we can reconstruct~$V^n$ almost losslessly as long as $n$ is sufficiently large. Additionally, the processors can themselves generate $\tilX^n$ and $\tilY^n$ independently. Thus, it is clear that an achievable rate of common randomness is $H(V)$. It is also  plausible that any rate strictly below $H(V)$ is not achievable as the common part of $X$ and $Y$ cannot be reliably reconstructed. % in view of the converse part of Shannon's lossless source coding theorem. 

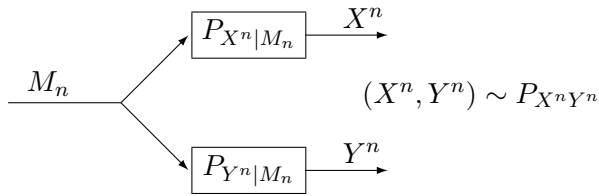
\begin{figure}
\centering \setlength{\unitlength}{0.05cm} %\scalebox{2}
{ \begin{picture}(100,60) %\linethickness{1pt}
\put(-5,30){\line(1,0){30}} 
\put(25,30){\vector(1,1){18}}
\put(25,30){\vector(1,-1){18}} 
\put(44,42){\framebox(30,12){$P_{X^{n}|M_{n}}$}}
\put(44,6){\framebox(30,12){$P_{Y^{n}|M_{n}}$}} 
\put(74,48){\vector(1,0){22}}
\put(74,12){\vector(1,0){22}} \put(0,33){%
\mbox{%
$M_{n}$%
}} \put(84,50){%
\mbox{%
$X^{n}$%
}} \put(84,14){%
\mbox{%
$Y^{n}$%
}
\put(-5, 16){$(X^n,Y^n)\sim P_{X^nY^n}$}
} \end{picture}}
\caption{The distributed source simulation problem}
\label{fig:source_sim}
\end{figure}

We now turn to formal definitions of the problem. Consider the distributed source simulation setup depicted
in Fig.~\ref{fig:source_sim}. Each of the two terminals has access to a uniformly distributed random variable~$M_n$, also known as the {\em common} or {\em shared randomness}. Given a target distribution
$\pi_{XY}$, one of terminals uses $M_n \in \calM_n$ and its own local randomness
to generate a random vector $X^n$ and the other one uses~$M_n$ and its own
local randomness to generate another random vector~$Y^n$. The terminals' goal is to ensure that  the {\em synthesized distribution}
\begin{equation}
P_{X^nY^n}(x^n,y^n)=\frac{1}{|\calM_n|} \sum_{m\in\calM_n} P_{X^n|M_n}(x^n|m)P_{Y^n|M_n}(y^n|m) \label{eqn:syn_dist}
\end{equation}
is ``close to'' the $n$-fold product of the target distribution $\pi_{XY}^n$. We wish to quantify the minimum  amount of common randomness---that is the cardinality $|\calM_n|$ or its normalized logarithm $\frac{1}{n}\log|\calM_n|$---satisfying this requirement. Of course, we have to quantify what we mean by ``close to''. In Wyner's original paper, this discrepancy between $P_{X^nY^n}$ and $\pi_{XY}^n$ was quantified via the  {\em normalized relative entropy}
\begin{equation}
\frac{1}{n} D\left(P_{X^nY^n} \middle\| \pi_{XY}^n\right)=\frac{1}{n}\sum_{x^n,y^n}P_{X^nY^n}(x^n,y^n)\log\frac{P_{X^nY^n}(x^n,y^n)}{\pi_{XY}^n(x^n,y^n)}.  \label{eqn:norm_re}
\end{equation}

\begin{definition} \label{def:wyner_code}
An {\em $(n,R)$-fixed-length distributed source simulation code} consists of a pair of random mappings  called {\em processors}  $P_{X^n | M_n} \in\calP( \calX^n |\calM_n)$ and  $P_{Y^n | M_n} \in\calP( \calY^n |\calM_n)$ such that $\log |\calM_n|\le nR$. 
\end{definition}
In the above definition, $n$ and $R$ are known respectively as the {\em blocklength} and the {\em rate} of the code $(P_{X^n | M_n},P_{Y^n | M_n})$. We are now ready to define Wyner's common information from the distributed source simulation perspective. 
\begin{definition} \label{def:wyner_CI}
The  {\em  minimal distributed simulation rate} $T(\pi_{XY})$ between a pair of random variables $(X,Y)\sim \pi_{XY}$ is the infimum of all rates $R$ such that there exists a sequence of $(n,R)$-fixed-length  distributed source simulation  codes $\{ (P_{X^n | M_n},P_{Y^n | M_n})\}_{n=1}^\infty$ satisfying 
\begin{equation}
\lim_{n\to\infty} \frac{1}{n} D\left(P_{X^nY^n} \middle\| \pi_{XY}^n\right)=0, \label{eqn:norm_re2}
\end{equation}
where $P_{X^nY^n}$ denotes the synthesized distribution in \eqref{eqn:syn_dist}. 
\end{definition}
At this point, the  reader may wonder whether the minimal distributed simulation rate  $T(\pi_{XY})$ as defined in Definition~\ref{def:wyner_CI} is ``sensitive'' to the choice of the discrepancy measure---namely, that it is the normalized relative entropy in \eqref{eqn:norm_re2}. We reassure the reader that this will be discussed extensively in the sequel---as a matter of fact, this is a central  theme in Part~\ref{part:two} of the monograph. Just to provide a sneak peek at the results in the subsequent sections, we mention the $T(\pi_{XY})$ remains unchanged if we choose not to normalize by $n$ in~\eqref{eqn:norm_re2}; this results in a {\em more stringent} criterion. Furthermore, $T(\pi_{XY})$ also remains the same if the normalized relative entropy is replaced by the TV distance $|P_{X^nY^n} - \pi_{XY}^n|$. More importantly, we discuss the ramifications of changing the discrepancy measure to various members of the family of normalized and unnormalized R\'enyi divergences; these have implications for other notions of common information such as the exact common information. 

One of Wyner's key contributions in his seminal paper on common information~\cite{WynerCI} is the following. 

\begin{theorem}\label{thm:wynerCI_sim}
The  minimal distributed simulation rate is given by 
\begin{equation}
T(\pi_{XY}) = C_{\Wyner}(\pi_{XY})=\min_{P_{W}P_{X|W}P_{Y|W}: P_{XY}=\pi_{XY}} I_P(XY;W).  \label{eqn:wyner_sim}
\end{equation}
%where $C_{\Wyner}(\pi_{XY})$ is defined in~\eqref{eqn:wyner_ci_intro}.
\end{theorem}
Thus, the minimal distributed simulation rate is exactly Wyner's common information as defined in \eqref{eqn:wyner_ci_intro}.  We reiterate that the minimization in~\eqref{eqn:wyner_sim} is performed over all triples of random variables $(W,X,Y)$ such that $X-W-Y$ forms a Markov chain in this order and the marginal distribution of $(X,Y)\sim P_{XY}$ is exactly the target distribution $\pi_{XY}$. The use of the $\min$ (in place of an $\inf$ as in~\eqref{eqn:wyner_ci_intro}) in~\eqref{eqn:wyner_sim} is justified by the fact that the cardinality of~$W$ can be restricted to be no more than $|\calX||\calY|$. This is a consequence of an application  of the convex cover method; see \citet[Appendix~C]{elgamal} for a detailed discussion. Having established the  equivalence between the minimal distributed simulation rate and Wyner's common information,  in the following, we will no longer distinguish between these two notions.

\begin{example} \label{ex:XYV_Wyner} Let us do a sanity check of the expression in~\eqref{eqn:wyner_sim}  based on our running example in which  $X=(\tilX,V)$ and $Y=(\tilY,V)$ and $\tilX$, $\tilY$ and~$V$ are mutually independent. By taking $W=V$, we see that $T(\pi_{XY})\le H(V)$. On the other hand, we have the Markov chain $V-X-W-Y-V$, so $V$ is a deterministic function of $W$. As a result, 
\begin{equation}
I(XY;W) = I(\tilX\tilY V; W) = I(\tilX\tilY V; WV)\ge H(V) .
\end{equation}
Since this holds true for all $X-W-Y$,  minimizing the left-hand side over all such joint distributions yields $T(\pi_{XY})\ge H(V)$ as desired.
%\begin{align}
%I(XY;W) &= H(XY)-H(XY|W)  \nn\\
%&= H(V)+ H(\tilX) + H(\tilY )  -H(X|W) - H(Y|W) .
%\end{align}
%At the same time, $H(X|W) = H(\tilX  |VW)$ because $H(V|W) =0$. Indeed, due to the Markov chain $X-W-Y$, we have $0 = I(X;Y|W) \ge I(V;V|W) = H(V|W) \ge 0$. As a result,
%\begin{align}
%I(XY;W)&= H(V) + H(\tilX) + H(\tilY )  - H(\tilX |VW)-H(\tilY|VW) \nn\\
%&=H(V)+I(\tilX;VW) + I(\tilY;VW)\ge  H(V).
%\end{align}
%This inequality holds for all $X-W-Y$ so minimizing the left-hand-side yields $T(\pi_{XY})\ge H(V)$ as desired.  
So indeed, the formula in \eqref{eqn:wyner_sim} coincides with the intuitive expression for the common information of $X=(\tilX,V)$ and $Y=(\tilY,V)$, namely $H(V)$.
\end{example} 

Although we will not provide detailed proofs in this monograph, we briefly mention the main idea to prove the direct (or achievability) part of Theorem~\ref{thm:wynerCI_sim} as it is a prevailing theme in Part~\ref{part:two}. This is based on the following lemma, which, in today's information theory parlance, is known as   {\em approximation of output statistics}~\citep{HV93}, {\em channel resolvability}~\cite{Hayashi06, Hayashi11}, or {\em soft-covering}~\citep{cuff13}. We term  any subset  $\calC_n$ of $\calW^n$ as a {\em codebook}. Any codebook $\calC_n$ takes the form   $\{w^n(m):m\in \calM_n\}$. The elements of $\calC_n$, namely $w^n(m)$, are  called   {\em codewords}.

\begin{lemma}[Soft-Covering]\label{lem:soft-covering}
Let $(U,W)\sim P_{UW} \in\calP(\calU\times\calW)$ be a given pair of random variables with mutual information $I(U;W)$. For any $R> I(U;W)$, there exists a sequence of codebooks  $\{ \calC_n \}_{n=1}^\infty$ with 
\begin{equation}
\limsup_{n\to\infty}\frac{1}{n}\log|\calM_n|\le R
\end{equation}
such that the  corresponding sequence of synthesized distributions
\begin{equation}
 P_{U^n}(u^n) := \frac{1}{|\calM_n|} \sum_{m \in \calM_n}P_{U|W}^n   (u^n | w^n(m) ),\quad n\in\bbN  \label{eqn:syn_dis_sc}
\end{equation}
is arbitrarily close in the normalized relative entropy to the product distribution $P_U^n$, i.e., 
\begin{equation}
\lim_{n\to\infty}\frac{1}{n}D\big( P_{U^n} \big\| P_U^n \big)=0. \label{eqn:soft-covering-approx}
\end{equation}
In addition, the TV distance between $P_{U^n}$ and $P_U^n$   vanishes, i.e., 
\begin{equation}
\lim_{n\to\infty}\big| P_{U^n} -P_U^n\big| =0.\label{eqn:soft-covering-tv}
\end{equation}
\end{lemma}
%We can think of each codebook as being in one-to-one correspondence with a {\em synthesizer} 
%\begin{equation}
%P_{U^n|M_n} (u^n|m) =\frac{1}{|\calM_n|}\sum_{m\in\calM_n} P_{U^n| W^n 
%\end{equation}
We can interpret the soft-covering lemma by considering drawing a codeword $w^n$ from the codebook $\calC_n$ uniformly at random. This codeword  is then 
%a uniform random variable $M_n$ with support $\{1, \ldots , |\calC_n|\}$. One draws a sample $m$ from $M_n$. Then the $m^{\mathrm{th}}$ codeword $w^n(m)$ is chosen and 
sent through $n$ uses of the {\em test channel} $P_{U|W}$. 
Lemma \ref{lem:soft-covering} says that as long as the cardinality of $\calC_n$ is large enough in the sense that its rate exceeds $I(U;W)$, the synthesized distribution $P_{U^n}$ can be made arbitrarily close to $P_U^n$ in the sense of~\eqref{eqn:soft-covering-approx} or \eqref{eqn:soft-covering-tv}; see Fig.~\ref{fig:soft_cover}. The statement in~\eqref{eqn:soft-covering-approx} is due to \citet{WynerCI} while that in \eqref{eqn:soft-covering-tv} is due to \citet{HV93}, \citet{Hayashi06} and \citet{cuff13}. The soft-covering lemma has found numerous applications in information-theoretic security. 

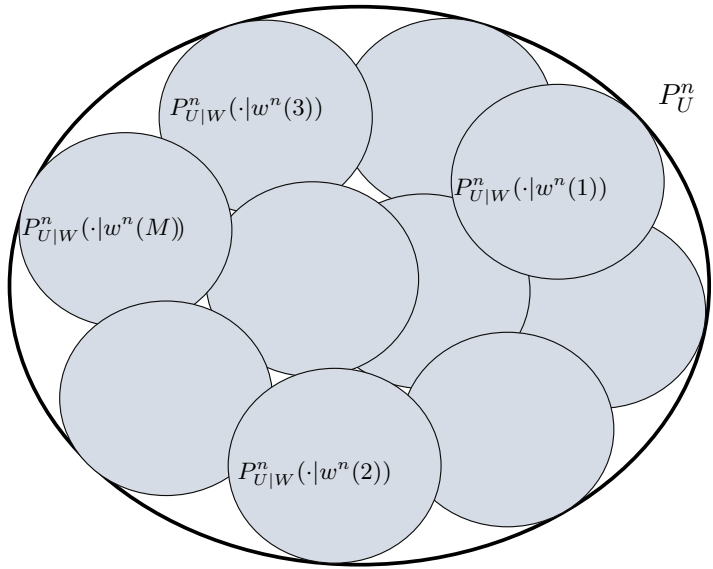
\begin{figure}
\centering
\begin{overpic}[width=1\columnwidth]{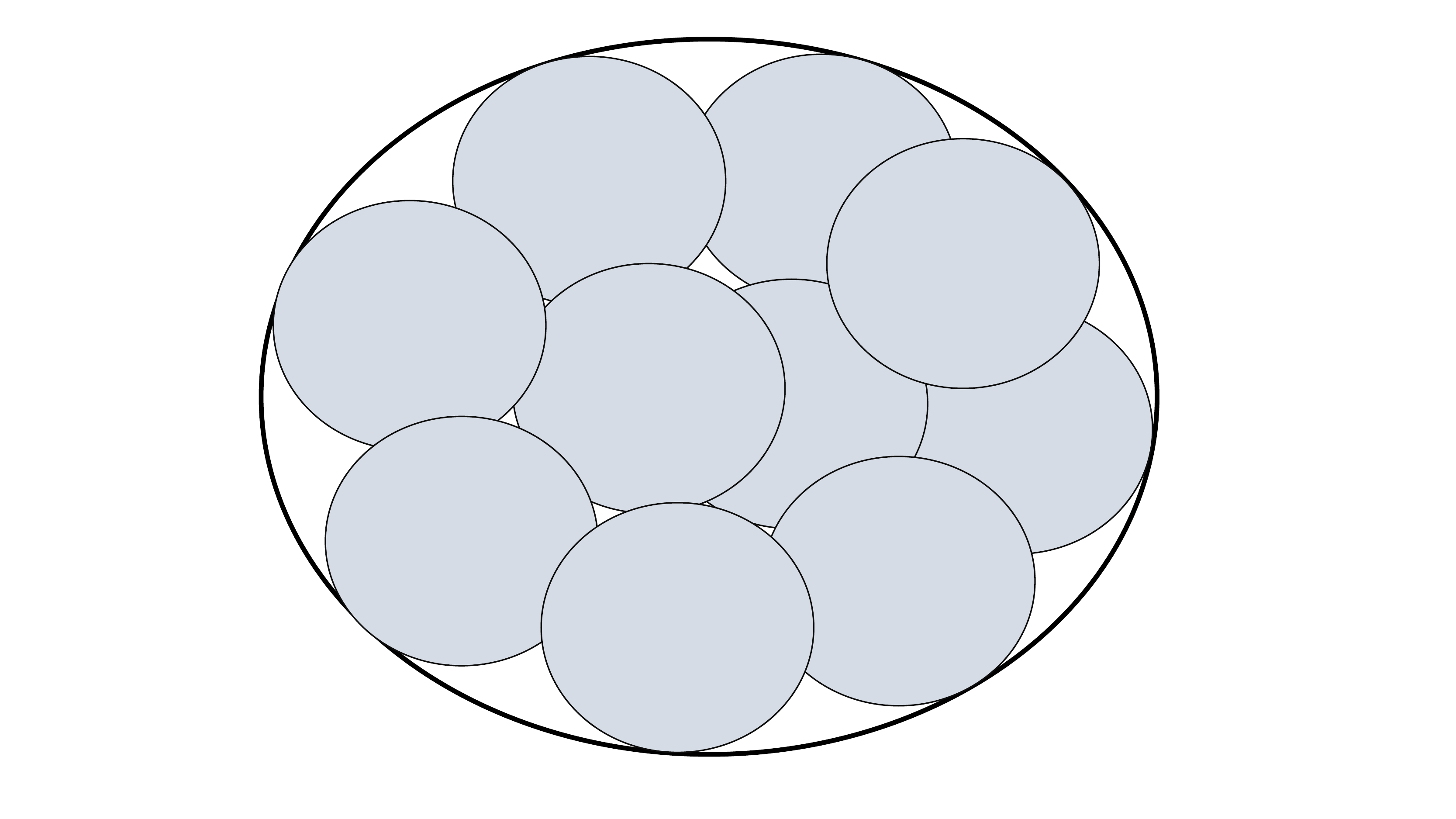}
\put(75,45){\mbox{$P_{U}^n$}}
{\footnotesize 
\put(57,37){\mbox{$P_{U|W}^n(\cdot|w^n(1))$}}
\put(38,12){\mbox{$P_{U|W}^n(\cdot|w^n(2))$}}
\put(32,44){\mbox{$P_{U|W}^n(\cdot|w^n(3))$}}
\put(19,33.5){\mbox{$P_{U|W}^n(\cdot|w^n(M)\!)$}}
}
\end{overpic}
\vspace{-.38in}
\caption{Illustration of the soft-covering lemma. If $M$ is large enough (i.e., its exponential rate exceeds $I(U;W)$), the uniform mixture of the conditional distributions $P_{U|W}^n( \cdot | w^n(m))$ approximates the product distribution $P_U^n$  well in the sense of the normalized relative entropy and the total variation distance.}
%. If $M$ is large enough (i.e., its rate exceeds $I(U;W)$), the small balls (conditional typical sets $\calT_\epsilon^{(n)}(P_{U W}|w^n(m))$) can cover almost the entire large ball (the typical set $\calT_\epsilon^{(n)}(P_U)$). }
\label{fig:soft_cover}
\end{figure}

The application of the soft-covering lemma to prove the achievability part of Wyner's common information is now apparent. Particularize  $U\sim P_U$ in Lemma~\ref{lem:soft-covering} to be $(X,Y)\sim \pi_{XY}$ and since $X-W-Y$ forms a Markov chain, the synthesized distribution in~\eqref{eqn:syn_dis_sc}  reduces to that in~\eqref{eqn:syn_dist} by setting for each $u^n=(x^n,y^n)$ and $m\in\calM_n$,
\begin{equation}
P_{X^n|M_n}(x^n|m)P_{Y^n|M_n}(y^n|m) =P_{U^n|W^n} (u^n|w^n(m)).
\end{equation}

The converse is proved via single-letterization steps that are  commonplace in network information theory.  We omit them here as we will, in Section~\ref{sec:dist_ss_cts} and subsequent sections, sketch proofs that yield stronger and more general results, thus recovering the converse of Theorem \ref{thm:wynerCI_sim} ``for free''. See \cite[Section~5]{WynerCI} for the original converse proof. 

%\iffalse
\begin{remark} \label{rmk:wyner_lp}
Wyner's common information  can be alternatively written as
\begin{align}
C_{\Wyner}(\pi_{XY})&=H_{\pi}(X,Y) + \min \bigg\{ \sum_{x\in\calX} \bbE_W \big[ P_{X|W}(x|W)\log P_{X|W}(x|W) \big]  \nn\\*
& \qquad + \sum_{y\in\calY} \bbE_W\big[ P_{Y|W}(y|W)\log P_{Y|W}(y|W) \big]    \bigg\}, \label{eqn:bilinear}
\end{align}
where the minimum extends over all pairs of  collections of  random variables $ \{ P_{X|W} (x|W )\}_{x\in\calX}$ and $ \{ P_{Y|W} (y|W) \}_{y\in\calY}$ satisfying
\begin{alignat}{3}
P_{X|W}(x|W)  , P_{Y|W}(y|W) &\ge 0 &&\quad\forall\, (x,y )\in\calX\times\calY\\
\sum_{x \in \calX} P_{X|W}(x|W) = \sum_{y\in\calY} P_{Y|W} (y|W) &= 1&&\quad\mbox{and} \label{eqn:linear_constr} \\
\bbE_W \big[ P_{X|W}(x|W)   P_{Y|W}(y|W) \big] &= \pi_{XY}(x,y)&&\quad \forall\, (x,y )\in\calX\times\calY.
\end{alignat}
%Thus, $C_{\Wyner}(\pi_{XY})$ can be represented as  the minimum of a function in  $ \{ P_{X|W} (x|W )  P_{Y|W} (y|W) \}_{x,y }$  subject to the linear constraints above. This allows one to alternatively express $C_{\Wyner}(\pi_{XY})$  as a ; see 
In~\cite[Eqn.~(1.16)]{WynerCI}, Wyner  claims that $C_{\Wyner}(\pi_{XY})$ can be expressed as a max-min  of a difference of relative entropies. However, the authors have disproved this claim numerically. The problem with Wyner's argument is that one cannot swap the $\min$ and $\max$ operations because the Lagrangian corresponding to the minimization in~\eqref{eqn:bilinear} and constraints in~\eqref{eqn:linear_constr} is {\em bilinear}  in  $ \{ P_{X|W} (x|W ) \}_x$ and $\{  P_{Y|W} (y|W) \}_{y }$ and not  (jointly) linear in them. 
\end{remark}
%\fi
\section{The Gray--Wyner System} \label{sec:wyner-GW}
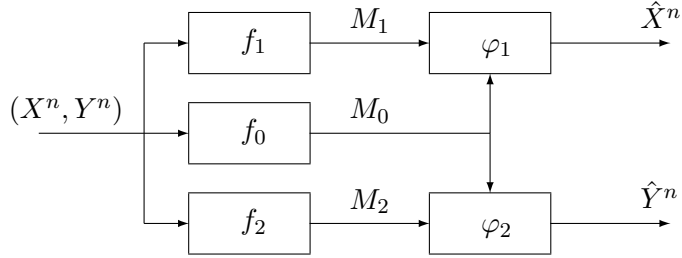
\begin{figure}[t]
\centering
\setlength{\unitlength}{0.4cm}
\scalebox{1}{
\begin{picture}(26,9)
%\linethickness{1pt}
\put(0,5.5){\makebox{$(X^n,Y^n)$}}
%put three encoders
\put(6,1){\framebox(4,2)}
\put(6,4){\framebox(4,2)}
\put(6,7){\framebox(4,2)}
%put encoder names
\put(7.7,1.8){\makebox{$f_2$}}
\put(7.7,4.8){\makebox{$f_0$}}
\put(7.7,7.8){\makebox{$f_1$}}
%line the source to encoders
\put(1,5){\vector(1,0){5}}
\put(4.5,5){\line(0,1){3}}
\put(4.5,8){\vector(1,0){1.5}}
\put(4.5,5){\line(0,-1){3}}
\put(4.5,2){\vector(1,0){1.5}}
%put two decoders
\put(14,1){\framebox(4,2)}
\put(14,7){\framebox(4,2)}
%put decoder names
\put(15.7,7.7){\makebox{$\varphi_1$}}
\put(15.7,1.7){\makebox{$\varphi_2$}}
%put arrows to decoders
\put(10,2){\vector(1,0){4}}
\put(12,2.75){\makebox(0,0){$M_2$}}
\put(10,5){\line(1,0){6}}
\put(12,5.75){\makebox(0,0){$M_0$}}
\put(16,5){\vector(0,1){2}}
\put(16,5){\vector(0,-1){2}}
\put(10,8){\vector(1,0){4}}
\put(12,8.75){\makebox(0,0){$M_1$}}
%put decoded symbols
\put(18,2){\vector(1,0){4}}
\put(21,2.5){\makebox{$\hat{Y}^n$}}
\put(18,8){\vector(1,0){4}}
\put(21,8.5){\makebox{$\hat{X}^n$}}
\end{picture}}
\caption{The Gray--Wyner source coding problem~\citep{GrayWyner}.}
\label{fig:gw}
\end{figure}

In addition to the Wyner's common information being interpreted as the minimum rate required to simulate a joint source $\pi_{XY}$ in a distributed manner, there is another natural interpretation in terms of a distributed lossless source coding system---the Gray--Wyner system \cite{GrayWyner} as depicted in Fig.~\ref{fig:gw}. In this problem, there is a joint source $(X,Y)\sim \pi_{XY}$ that is to be reconstructed almost losslessly. This joint source is encoded into three bit strings $(M_0,M_1,M_2)$   of rates $(R_0,R_1,R_2)$ via three encoders that observe $n$ independent copies of $(X,Y)$. There are two decoders. Bit strings $M_0$ and $M_1$ are sent to the first decoder, while bit strings $M_0$ and $M_2$ are sent to the second decoder.  The two decoders generate estimates $\hatX^n$ and $\hatY^n$ of $X^n$ and~$Y^n$ respectively.

In distributed lossless source coding problems, one is concerned with the tradeoff among the rates; in this case, $(R_0, R_1, R_2)$. If the three encoders are combined into a single entity---equivalently, the common rate is allowed to be arbitrarily large---by Shannon's lossless source coding theorem, we can describe the joint source  using roughly $nH(XY)$ bits, or at a rate of $H(XY)$. Clearly, we can do more to reduce the common rate. Using our running example in which $X=(\tilX,V)$ and $Y=(\tilY, V)$ with $\tilX, \tilY$, and $V$ being independent, any coding scheme involves compressing the common part of $X$ and $Y$ using encoder $f_0$. For lossless reconstruction, this requires a rate of roughly $H(V)$. The other encoders $f_1$ and $f_2$ are tasked with compressing the private parts of the sources, namely $\tilX$ and $\tilY$ respectively. These require rates of roughly $H(\tilX)$ and $H(\tilY)$. Reconstruction of the sources by the decoders is clearly possible. For example, $\varphi_1$ takes the descriptions $(M_0,M_1)$ and reconstructs $V^n$ and $\tilX^n$, which when concatenated, is approximately $X^n$.  Thus, the required sum rate is $H(V)+ H(\tilX)+ H(\tilY ) = H(XY)$. Motivated by this special case, it seems natural to alternatively define the common information of the any source $(X,Y)\sim\pi_{XY}$ as the minimum common rate $R_0$ such that the sum rate $R_0+R_1+R_2$ is no larger than  the joint entropy $H(XY)$. The set of all $(R_0, R_1, R_2)$ such that $R_0+R_1+R_2=H(XY)$ is   known as the {\em Pangloss plane} of the source. The term ``Pangloss plane'' was  coined by \citet{GrayWyner}.

\begin{definition}
An {\em $(n,R_0, R_1, R_2)$-Gray--Wyner code} consists of 
\begin{itemize}
\item Three encoders $f_i:\calX^n\times\calY^n\to [2^{nR_i}]$ where $i = 0, 1,2 $; 
\item Two decoders $\varphi_1 \!:\! [2^{nR_0}]\!\times\! [2^{nR_1}]\!\to\! \calX^n$ and $\varphi_2 \!:\! [2^{nR_0}]\!\times\! [2^{nR_2}]\!\to \!\calY^n$. 
\end{itemize}
The {\em probability of error} of the code is 
\begin{equation}
\Pr\big( \big(\varphi_1 (M_0,  M_1) , \varphi_2 ( M_0 ,  M_2) \big) \ne (X^n,Y^n) \big),\label{eqn:prob_error}
\end{equation}
where $M_i = f_i(X^n, Y^n)$ for $i = 0,1,2$. 
\end{definition}

\begin{definition} \label{def:gw_R0}
The {\em Pangloss-common information based on the Gray--Wyner system} $T_{\mathrm{GW}}(\pi_{XY})$ between two random variables $(X,Y)\sim\pi_{XY}$ is the infimum of all  $R_0$ such that for all $\epsilon>0$, there exists a sequence of  $(n,R_0, R_1, R_2)$-Gray--Wyner codes $\{ (f_{0,n} , f_{1,n} , f_{2,n}, \varphi_{1,n}, \varphi_{2, n} )\}_{n=1}^\infty$  such that $R_0+R_1+R_2\le H(XY)+\epsilon$ for all $n$ sufficiently large and the probability of error in \eqref{eqn:prob_error} vanishes as the length of the code $n$ tends to infinity. 
\end{definition}
The term ``Pangloss''  is used in the above definition to emphasize that sum rate should be close $H(XY)$; this is to distinguish this definition from an analogous one for the GKW common information (Definition~\ref{def:gw_R0_max2}).  We also adopt the somewhat verbose qualifier ``based on the Gray--Wyner system'' and the subscript $\mathrm{GW}$ in $T_{\mathrm{GW}}(\pi_{XY})$ because {\em a priori},  there is little evidence to suggest that $T_{\mathrm{GW}}(\pi_{XY})$  equals to the quantity in Definition~\ref{def:wyner_CI}. The qualifier can, however, be jettisoned in view of the following theorem also due to \citet{WynerCI}.
\begin{theorem} \label{thm:gw_wyner}
The Pangloss-common information based on the Gray--Wyner system 
\begin{equation}
T_{\mathrm{GW}}(\pi_{XY})=C_{\Wyner}(\pi_{XY})  . \label{eqn:wyner_gw2}
\end{equation}
\end{theorem}
Thus, both definitions of the common information (in Definitions~\ref{def:wyner_CI} and~\ref{def:gw_R0}) coincide and we can use a single symbol $C_{\Wyner}(\pi_{XY})$ to name the quantity  on the right-hand side of \eqref{eqn:wyner_gw2}. The subscript $\Wyner$ refers to \underline{W}yner.
%, i.e., 
%\begin{equation}
%C_{\Wyner}(\pi_{XY}) :=\min_{P_W P_{X|W} P_{Y|W}: P_{XY}=\pi_{XY}}I_P(XY;W), \label{eqn:c_wyner_def}
%\end{equation}
%where the subscript $\Wyner$ refers to \underline{W}yner. 
The quantity $C_{\Wyner}(\pi_{XY}) $ thus  has two operational interpretations; one as the minimum rate required to simulate a joint source in a distributed manner and another as the minimum common rate of the Gray--Wyner system keeping the sum rate at the joint entropy of $\pi_{XY}$.

%\eqref{eqn:wyner_gw}

%Discuss the Gray--Wyner distributed lossless source coding system and relate the minimal $R_0$ on the Pangloss plane to Wyner's CI. Briefly discuss the lossy generalization.

\section{Doubly Symmetric Binary Sources} \label{sec:dsbs}
Due to the optimization over the Markov chain $X-W-Y$, Wyner's common information is difficult to evaluate for most pairs of sources $\pi_{XY}$. Two notable exceptions are  the doubly symmetric binary source (DSBS) and the symmetric binary erasure source (SBES). We describe the former in this section and the latter in the next. %WE, which we revisit repeatedly in this monograph, so we find it instructive to describe it in detail here.  

%\begin{example} \label{ex:dsbs}
Consider a DSBS $(X,Y) \in\{0,1\}^2$ which is defined by  the joint distribution
\begin{align}
\pi_{XY}=\begin{bmatrix}
\alpha &\beta \\ \beta & \alpha
\end{bmatrix}\label{eqn:dsbs} ,
\end{align}
where $\alpha=(1-p)/{2}$, $\beta= {p}/{2}$ and $p\in (0,1/2)$. This is equivalent to $X\sim\Bern(1/2)$ and $Y=X\oplus E$ with $E\sim\Bern(p)$ and independent of~$X$. Here, $p$ represents the crossover probability of a binary symmetric channel (BSC) with $X$ being the input and $Y$ the output. Intuitively, if $p \downarrow 0$, $X$ and $Y$ become highly correlated, and the common information increases. On the other hand, if $p \uparrow 1/2$, $X$ and $Y$ become close to independent and the common information decreases to $0$.

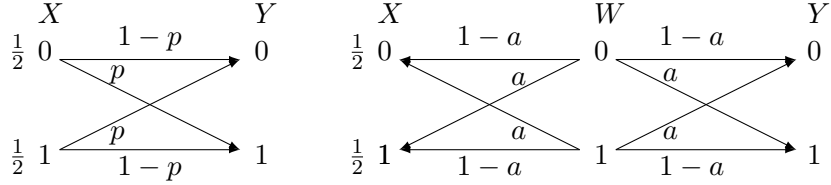
\begin{figure}
\centering \setlength{\unitlength}{0.048cm} %\scalebox{2}
\begin{tabular}{cc}
{ 
\begin{picture}(64,49) 
\put(10,35){\vector(1,0){50}} 
\put(10,10){\vector(1,0){50}} 
\put(10,10){\vector(2,1){50}} 
\put(10,35){\vector(2,-1){50}} 
\put(26,3){\mbox{$1-p$}} 
\put(26,39){\mbox{$1-p$}} 
\put(24,30){\mbox{$p$}} 
\put(24,13){\mbox{$p$}} 
\put(4,35){\mbox{$0$}} 
\put(-4,35){\mbox{$\frac{1}{2}$}} 
\put(4,6){\mbox{$1$}} 
\put(-4,6){\mbox{$\frac{1}{2}$}} 
\put(64,35){\mbox{$0$}} 
\put(64,6){\mbox{$1$}} 
\put(4,45){\mbox{$X$}} 
\put(64,45){\mbox{$Y$}}  
\end{picture}
} \hspace{.25in} &
{
\begin{picture}(130,49) 
\put(-4,35){\mbox{$\frac{1}{2}$}} 
\put(4,6){\mbox{$1$}} 
\put(-4,6){\mbox{$\frac{1}{2}$}} 
\put(60,35){\vector(-1,0){50}} 
\put(60,10){\vector(-1,0){50}} 
\put(60,10){\vector(-2,1){50}} 
\put(60,35){\vector(-2,-1){50}} 
\put(26,3){\mbox{$1-a$}} 
\put(26,39){\mbox{$1-a$}} 
\put(41,28){\mbox{$a$}} 
\put(41,13){\mbox{$a$}} 
\put(4,35){\mbox{$0$}} 
\put(4,6){\mbox{$1$}} 
\put(64,35){\mbox{$0$}} 
\put(64,6){\mbox{$1$}} 
\put(4,45){\mbox{$X$}} 
\put(64,45){\mbox{$W$}}  

\put(70,35){\vector(1,0){50}} 
\put(70,10){\vector(1,0){50}} 
\put(70,10){\vector(2,1){50}} 
\put(70,35){\vector(2,-1){50}} 
\put(82,3){\mbox{$1-a$}} 
\put(82,39){\mbox{$1-a$}} 
\put(83,29){\mbox{$a$}} 
\put(83,13){\mbox{$a$}} 
\put(123,35){\mbox{$0$}} 
\put(123,6){\mbox{$1$}} 
\put(123,45){\mbox{$Y$}} 
\end{picture}
}
\end{tabular}
\caption{Left: DSBS with crossover probability $p$. Right: Interpretation in terms of the common random variable $W$.}
\label{fig:dsbs}
\end{figure}

Equivalently,  $X=W\oplus A$ and $Y=W\oplus B$ with  $W\sim\Bern(1/2)$, $A,B\sim \Bern(a)$  mutually independent with $a:=(1-\sqrt{1-2p})/2\in (0,1/2)$ so $a\ast a =p$. Thus, similarly to $p$, as $a$ increases, the common information decreases. We can express $\alpha$ and $\beta$ in \eqref{eqn:dsbs} in terms of $a$ as $\alpha=\frac{1}{2}(a^2+(1-a)^2)$ and $\beta=a(1-a)$. In this parametrization, $W$ is the common random variable that achieves the minimum in the formula for Wyner's common information in~\eqref{eqn:wyner_sim}.   The two interpretations of the  DSBS are illustrated in Fig.~\ref{fig:dsbs}. Clearly, there is no loss in generality in restricting $p$ (or   $a$) to be in $(0,1/2)$; if not, replace $X$ by $X\oplus 1$.

\begin{figure}
\centering
\includegraphics[width = .93\columnwidth]{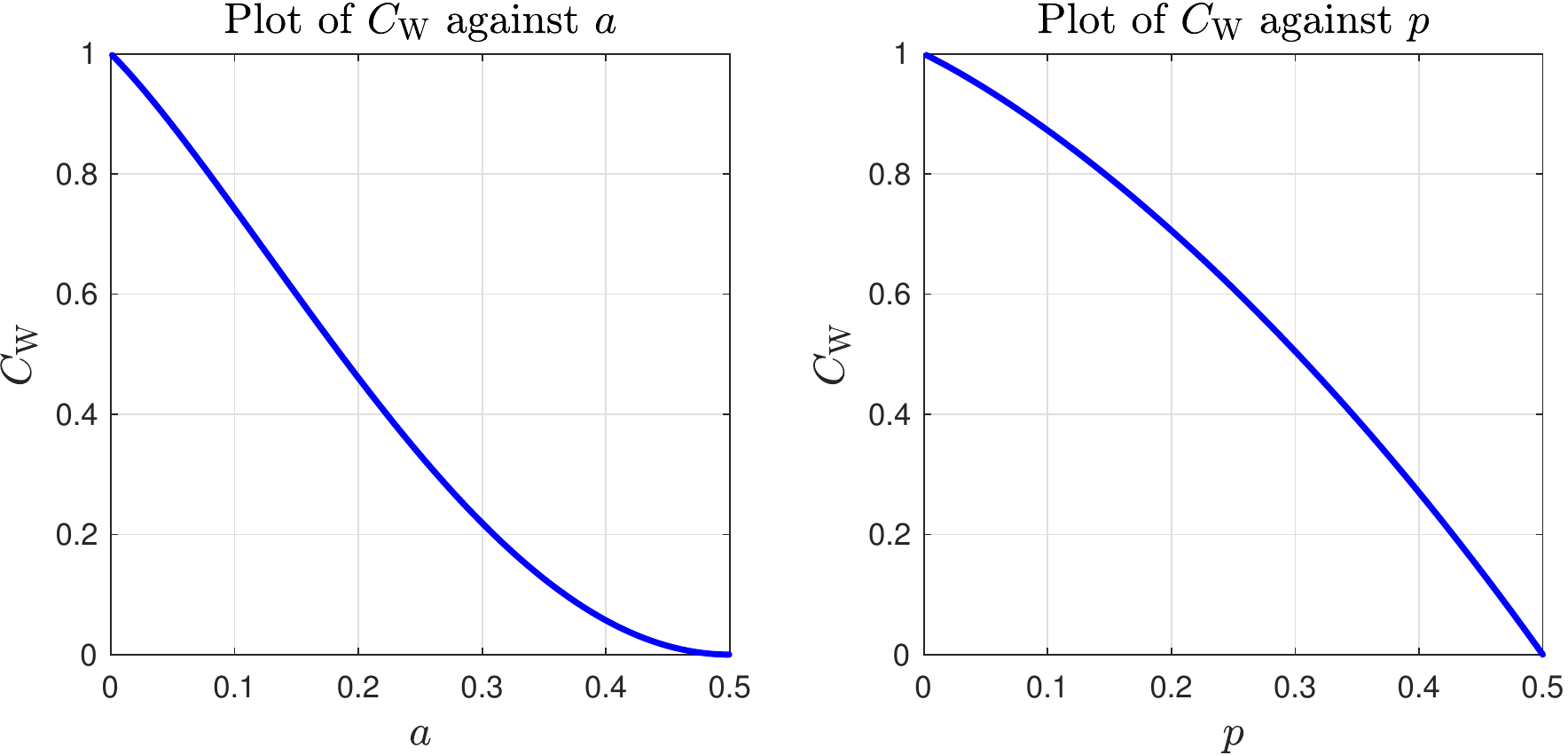}
\caption{Plots of Wyner's common information for the DSBS in terms of $p$ and $a$}
\label{fig:wyner_dsbs}
\end{figure}

\citet{WynerCI} successfully evaluated the common information for the DSBS in closed form.  
\begin{proposition} \label{prop:wyner_dsbs}
For the DSBS as described in \eqref{eqn:dsbs}, Wyner's common information is 
\begin{align}
 C_{\Wyner}(\pi_{XY})  =1+h\big(2a\bara\big)-2h(a)
\end{align}
where $h(a) := -a\log a- \bara\log \bara$ is the binary entropy function. 
\end{proposition}
This function is plotted in Fig.~\ref{fig:wyner_dsbs} and shows clearly that $C_{\Wyner}(\pi_{XY})$ for a DSBS  is decreasing in $p$ and $a$.  We conclude this section by  mentioning that \citet{Witsenhausen76} calculated $C_{\Wyner}(\pi_{XY})$ for a variety of other discrete sources. 
%\end{example}

\section{Symmetric Binary Erasure Sources} \label{sec:sbes}
The SBES is a joint source $\pi_{XY}$  with binary input $\calX=\{0,1\}$ and ternary output $\calY= \{0,\rme, 1\}$.  The ``output'' $Y$ is identical to the ``input'' $X$ with probability $1-p$ and takes on the ``erasure symbol'' $\rme$ with probability $1-p$. The input variable is uniformly distributed on $\calX$, leading to the joint distribution
\begin{equation}
\pi_{XY} = \begin{bmatrix}
(1-p)/2 & p/2 & 0 \\
0 & p/2 & (1-p)/2
\end{bmatrix} . \label{eqn:sbes_jd}
\end{equation}
This is illustrated in the left diagram of Fig.~\ref{fig:sbes}. \citet{cuff13} proved the following proposition.

\begin{figure}
\centering \setlength{\unitlength}{0.05cm} %\scalebox{2}
{ \begin{picture}(60,60) \put(4,42){%
\mbox{%
$0$%
}} \put(4,6){%
\mbox{%
$1$%
}} \put(12,45){\vector(1,0){30}} \put(12,9){\vector(1,0){30}}
\put(12,45){\vector(2,-1){30}} \put(12,9){\vector(2,1){30}}
\put(44,42){%
\mbox{%
$0$%
}} \put(44,26){%
\mbox{%
$\rme$%
}} \put(44,6){%
\mbox{%
$1$%
}} \put(4,52){%
\mbox{%
$X$%
}} \put(44,52){%
\mbox{%
$Y$%
}} \put(24,32){%
\mbox{%
$p$%
}} \put(24,20){%
\mbox{%
$p$%
}} \put(16,48){%
\mbox{%
$1-p$%
}} \put(16,3){%
\mbox{%
$1-p$%
}} \put(-8,42){%
\mbox{%
$\frac{1}{2}$%
}} \put(-8,6){%
\mbox{%
$\frac{1}{2}$%
}} \end{picture}} \hspace{1cm} \centering \setlength{\unitlength}{0.05cm}
%\scalebox{0.2}
{ \begin{picture}(100,60) \put(4,42){%
\mbox{%
$0$%
}} \put(4,6){%
\mbox{%
$1$%
}} \put(12,45){\textcolor{red}{\vector(1,0){30}}} \put(12,9){\textcolor{red}{\vector(1,0){30}}}
\put(12,45){\vector(2,-1){30}} \put(12,9){\vector(2,1){30}}
\put(44,42){%
\mbox{%
$0$%
}} \put(44,26){%
\mbox{%
$\rme$%
}} \put(44,6){%
\mbox{%
$1$%
}} \put(52,45){\vector(1,0){30}} \put(52,9){\vector(1,0){30}}
\put(52,45){\vector(2,-1){30}} \put(52,9){\vector(2,1){30}}
\put(52,27){\textcolor{red}{\vector(1,0){30}}} \put(84,42){%
\mbox{%
$0$%
}} \put(84,26){%
\mbox{%
$\rme$%
}} \put(84,6){%
\mbox{%
$1$%
}} \put(4,52){%
\mbox{%
$X$%
}} \put(44,52){%
\mbox{%
$W$%
}} \put(84,52){%
\mbox{%
$Y$%
}} \put(24,32){%
\mbox{%
$p_{1}$%
}} \put(24,20){%
\mbox{%
$p_{1}$%
}} \put(16,48){%
\mbox{%
$1-p_{1}$%
}} \put(16,3){%
\mbox{%
$1-p_{1}$%
}} \put(62,33.5){%
\mbox{%
$p_{2}$%
}} \put(62,18.5){%
\mbox{%
$p_{2}$%
}} \put(56,48){%
\mbox{%
$1-p_{2}$%
}} \put(56,3){%
\mbox{%
$1-p_{2}$%
}} \put(55,29){%
\mbox{%
$1$%
}} \put(-8,42){%
\mbox{%
$\frac{1}{2}$%
}} \put(-8,6){%
\mbox{%
$\frac{1}{2}$%
}} \end{picture}}
\caption{Left: SBES with erasure probability $p$. Right: Interpretation in terms of the common random variable $W$ due to~\citet{cuff13}. }
\label{fig:sbes}
\end{figure}
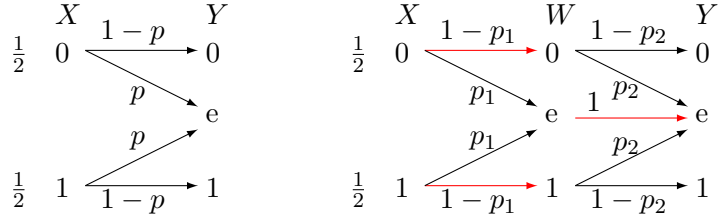
\begin{proposition} \label{prop:sbes}
For the SBES as described in~\eqref{eqn:sbes_jd}, Wyner's common information  is % for the SBES was derived by \citet{cuff13} and is equal to 
\begin{equation}
C_\Wyner(\pi_{XY}) = \left\{ \begin{array}{cc}
1 & p \le 0.5\\
 h(p) & p>0.5
\end{array} \right. . \label{eqn:cwyner_sbes}
\end{equation}
\end{proposition}
The optimal distribution $P_W P_{X|W} P_{Y|W}$ in Wyner's common information for the SBES  is shown in the right diagram of Fig.~\ref{fig:sbes} where  $X$ is uniform on $\calX$ and
%\begin{align}
%P_{W|X}(w|x)  & = \left\{ \begin{array}{cl}
%1-p_1  & w=x\in \{0,1\} \\
% p_1  & w=\rme, x\in \{0,1\} \\
%\end{array}  \right. \quad\mbox{and}\\
%P_{Y|W}(y|w)  &= \left\{ \begin{array}{cl}
%1-p_2  & y=w\in \{0,1\} \\
% p_2  & y=\rme, w\in \{0,1\} \\
% 1 & y=w=\rme
%\end{array}  \right. 
%\end{align}
  $p_1$ and $p_2 $ satisfy $(1-p_1)(1-p_2)=1-p$. Hence, the channel from $X$ to $Y$ is a concatenation of a binary erasure channel (BEC) with erasure  probability $p_1$ and a BEC-like channel with three inputs $0,\rme$, and $1$ in which, restricted to the inputs in $\{0,1\}$, it is a BEC with erasure probability $p_2$ but $\rme$ is transmitted noiselessly.   Wyner's common information  for an SBES is plotted in Fig.~\ref{fig:eci_sbes}. 

\begin{figure}
\centering
\includegraphics[width = .8\columnwidth]{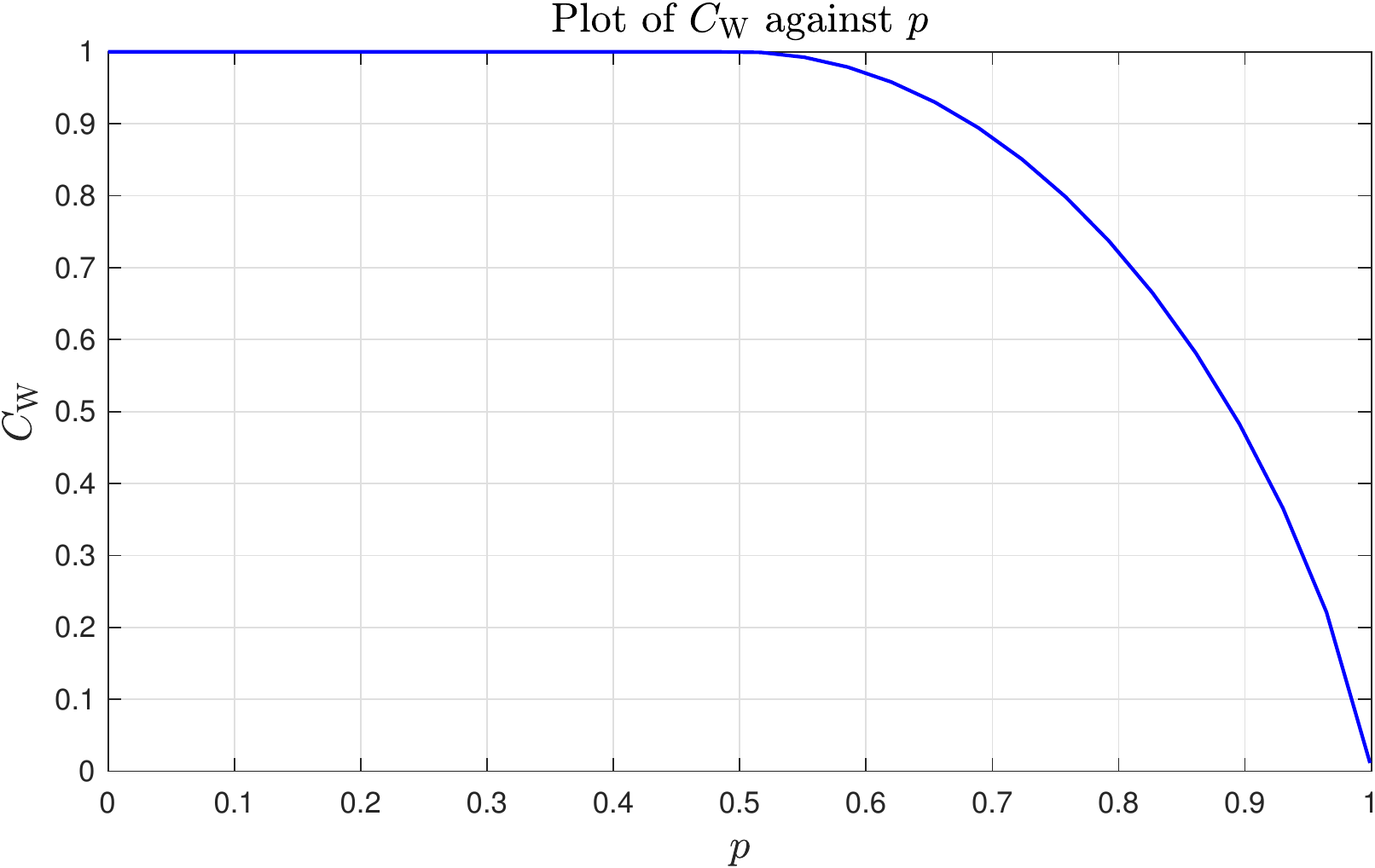}
\caption{Plot of Wyner's common information for the SBES in terms of the erasure probability $p$}
\label{fig:eci_sbes}
\end{figure}
\section{Continuous and Gaussian Sources} \label{sec:cts}
Even though the expression for Wyner's common information in~\eqref{eqn:wyner_sim} remains valid for arbitrary random variables when the $\min$ replaced by an $\inf$, i.e., 
\begin{equation}
C_{\Wyner}(\pi_{XY})=\inf_{ P_WP_{X|W}P_{Y|W}: P_{XY}=\pi_{XY}}I_P(XY;W), \label{eqn:CWyner_inf}
\end{equation}
the operational stories for Wyner's common information for continuous  random variables  are  more intricate. Indeed,  the operational interpretation  in terms of minimum common rate in the Gray--Wyner system (fixing the sum rate to be $H(XY)$) is only applicable to discrete random variables. % Note that we have to use an $\inf$ here since the cardinality of $W$ can no longer be bounded and so the constraint set is not compact.  
An operational interpretation for continuous random variables in terms of the {\em lossy} Gray--Wyner system \citep{Visw} was discovered by \citet{XuLiuChen}.  
The interpretation in terms of distributed source simulation remains valid, though the result is more subtle \citep{LiElgamal2017, YuTan2020_exact}. In this section, we first generalize Wyner's common information in this direction then discuss generalizations of the  Gray--Wyner system to be amenable to continuous sources. Finally, we justify  why these interpretations yield the same result for jointly Gaussian sources.

%An operational interpretation for continuous random variables in terms of the {\em lossy} Gray--Wyner system \citep{Visw} was discovered by \citet{XuLiuChen}. %We  briefly explain the extension  of Wyner's common information to continuous random variables in terms of the {\em lossy} Gray--Wyner system . 
%
% However, the interpretation in terms of distributed source simulation remains valid, though the result is more subtle, requiring more refined analyses. In this section, we first generalize Wyner's common information in this direction then discuss generalizations of the  Gray--Wyner system to accommodate continuous random variables. Finally, we state that both these interpretations yield the same result for the important class of jointly Gaussian random variables. 

\subsection{Distributed Source Simulation} \label{sec:dist_ss_cts}
In his seminal paper, \citet{WynerCI} characterized the common information for {\em finite alphabet} sources from the perspective of  distributed source simulation.  
%\citet{WynerCI} only characterized the common information for sources with finite alphabets. 
Here, we extend his results to arbitrary  and, in particular, continuous sources in the context of the distributed source simulation problem. The operational quantity $T(\pi_{XY})$ in the following theorem pertains to that in Definition~\ref{def:wyner_CI} (for distributed source simulation).

\begin{theorem}\label{thm:wyner_cts}
Let $(X, Y )\sim \pi_{XY}$ be a joint source with distribution
defined on the product of two arbitrary alphabets. Then we have\footnote{Since  we consider arbitrary probability spaces here, to be formal, we need to generalize several notions in probability theory, e.g., conditional distributions and conditional independence. We use  $P_{X|W}$   to denote a \emph{regular conditional probability distribution}~\cite{cinlar2011probability}. %, i.e., it consists of a class of distributions $P_{X|W=w}$ indexed by $w\in \calW$. 
Random variables   $X$ and $Y$, defined on an arbitrary measurable space, are  \emph{conditionally independent} given $W$, denoted as $X-W-Y$, if  $\sigma(X)$ and $\sigma(Y)$ are conditionally independent of $\sigma(W)$, where $\sigma(X)$ denotes the $\sigma$-algebra generated by $X$~\cite{cinlar2011probability}. When the regular conditional $P_{XY|W}$ exists, it holds that  $X-W-Y$ if and only if 
$P_{XY|W=w}$ is a product distribution (see Definition~\ref{def:pseudo_pdt}(a) for the countable alphabet case) for $P_W$-almost every $w\in \calW$.  }
\begin{equation}
\tilC_{\Wyner}(\pi_{XY})\le T(\pi_{XY})\le \hatC_{\Wyner}(\pi_{XY}) \label{eqn:wyner_ct_ineq}
\end{equation} 
where 
\begin{align}
\hspace{-.25in}\tilC_{\Wyner}(\pi_{XY}) &:= \lim_{\epsilon\downarrow 0}\inf_{\substack{ P_WP_{X|W}P_{Y|W}: \\ D(P_{XY}\|\pi_{XY})\le\epsilon}}I(XY;W)\quad\mbox{and} \label{eqn:tilC}\\
\hspace{-.25in}\hatC_{\Wyner}(\pi_{XY}) &:=\inf_{\substack{ P_WP_{X|W}P_{Y|W}: \\ P_{XY}=\pi_{XY}}}\lim_{s\downarrow 0}D_{1+s}\big(P_{X|W}P_{Y|W} \big\|P_{XY}\big|P_W \big). \label{eqn:hatC}
\end{align}
\end{theorem}
%\begin{remark} \label{rmk:wyner_cts}
%In fact, from the proof sketch below, we see that a stronger achievability (upper bound) is true. As long as $R>\hatC_{\Wyner}(\pi_{XY}) $, the unnormalized relative entropy between the synthesized distribution and the product target distribution vanishes. 
%\end{remark}

This result is due to the present authors~\cite{YuTan2020_exact}. An alternative upper bound on
the Wyner's common information of a set of continuous
random variables in terms of the dual total correlation
between them was derived by~\citet{LiElgamal2017}. We remark that when the joint source $\pi_{XY}$ is finitely supported, both $\tilC_{\Wyner}(\pi_{XY})$ and $\hatC_{\Wyner}(\pi_{XY})$ reduce to Wyner's common information $C_{\Wyner}(\pi_{XY})$ as defined in~\eqref{eqn:wyner_sim}. In particular, we recall that  as $s\downarrow 0$, the conditional R\'enyi divergence $D_{1+s} (P_{X|W}P_{Y|W}\|P_{XY}|P_W)$ reduces to the conditional relative entropy $D(P_{X|W}P_{Y|W}\|P_{XY}|P_W)$ which in turn equals  the mutual information $I(XY;W)$.

%Because the proof of this result is a recurring theme throughout Part~\ref{part:two} of the monograph, let us 

We highlight some key ideas of the proof. For the achievability part, we leverage a one-shot (non-asymptotic) soft-covering lemma that can be thought of as a strengthened version of Lemma~\ref{lem:soft-covering}. This result first appeared in the work of the present authors \cite{YuTan2019_wiretap}  en route to proving generalized security theorems for the wiretap channel~\cite{Wyn75,CK78}.
\begin{lemma}[One-Shot Soft-Covering] \label{lem:one-shot-sc}
Let $(U,W)\sim P_{UW} \in\calP(\calU\times\calW)$ be a given pair of random variables defined on some arbitrary measurable space. Consider a random codebook $\scC  = \{W (m): m \in\calM\}$ where $|\calM|=2^{ \lfloor R\rfloor}$ for some $R>0$. For each realization of the codebook $\calC= \{w(m) : m\in\calM\}$, define the synthesized distribution 
\begin{equation}
 P_{U|\scC }(u| \calC ) := \frac{1}{|\calM|}\sum_{m\in\calM}P_{U|W} ( u|w(m)) .
\end{equation}
Let $\pi_U$ be a target distribution such that for some $s\in (0,1]$, both $D_{1+s}(P_{U|W}\|\pi_U|P_W)$ and $D_{1+s}(P_U\|\pi_U)$ exist (and hence are finite). Then  for any $s\in (0,1]$,  we have
\begin{align}
&\exp\big(sD_{1+s}( P_{U|\scC} \|\pi_U |P_{\scC}) \big)\nn\\
&\;\;\;\le \exp\big(sD_{1+s}( P_{U|W} \|\pi_U |P_W)-sR \big)+\exp\big(  sD_{1+s} (P_U\|\pi_U)\big). \label{eqn:one-shot-sc}
\end{align}
\end{lemma}
By setting  $\pi_U \leftarrow \pi_{XY}^n$, $P_{U|W}\leftarrow P_{X|W}^n P_{Y|W}^n$, $P_W\leftarrow P_W^n$ and $R\leftarrow nR$, for some distribution $P_W P_{X|W} P_{Y|W}$ such that its marginal on $(X,Y)$ equals $\pi_{XY}$, Lemma \ref{lem:one-shot-sc} tells us that if 
\begin{equation}
  R> D_{1+s}( P_{X|W}P_{Y|W}\| \pi_{XY} | P_W),
  \end{equation}  
  then $D_{1+s}(P_{X^n Y^n | \scC_n}\|\pi_{XY}^n | P_{\scC_n})\to 0$. Thus, we conclude that there exists (at least) one sequence of (deterministic) codebooks $\{\calC_n\}_{n=1}^\infty$ such that 
  \begin{align} 
  & D(P_{X^nY^n|\scC_n } (\cdot|  \calC_n)\|\pi_{XY}^n) \\ 
  & \quad \le D_{1+s}(P_{X^nY^n|\scC_n}(\cdot|  \calC_n)\|\pi_{XY}^n)\to 0.
  \end{align}
   Letting   $s$ tend to $0$ (from above) and minimizing over all $P_W P_{X|W}P_{Y|W}$ concludes the proof of the achievability part. 
  
\begin{remark} \label{rmk:unnorm} The reader will observe that what we have proved is stronger than what Definition~\ref{def:wyner_CI} demands of a common information code. The one-shot soft-covering lemma as stated in Lemma~\ref{lem:one-shot-sc}  is strong enough to drive the {\em unnormalized} relative  entropy $D(P_{X^nY^n}\|\pi_{XY}^n)$ to zero   as $n\to\infty$. Compare this to~\eqref{eqn:norm_re2} in which the {\em normalized} relative entropy $\frac{1}{n}D(P_{X^nY^n}\|\pi_{XY}^n)$ is required to vanish. This strengthening will be central to our discussion in Part~\ref{part:two}.
\end{remark}
The converse follows from  standard single-letterization steps that we   outline here. Fix any code $(P_{X^n|M_n},P_{Y^n|M_n})$ per Definition~\ref{def:wyner_code}. Observe that 
\begin{align}
R &\ge \frac{1}{n}H(M_n) \ge \frac{1}{n} I(X^nY^n; M_n)\label{eqn:single_letter0}  \\*
&= \frac{1}{n}D\big(P_{ X^nY^n  M_n} \big\| P_{X^nY^n}P_{M_m} \big) \\*
&= \frac{1}{n}D\big(P_{ X^nY^n  M_n} \big\| \pi_{X^nY^n}P_{M_m} \big) - \frac{1}{n}D\big(P_{ X^nY^n} \big\| \pi_{XY}^n\big) . \label{eqn:single_letter}
\end{align}
The first term can be further lower bounded as 
\begin{align}
& \frac{1}{n}D\big(P_{ X^nY^n  M_n} \big\| \pi_{X^nY^n}P_{M_m} \big) \nn\\
&=\frac{1}{n}\sum_{i=1}^n D\big( P_{X_i Y_i | M_n X^{i-1} Y^{i-1} }\big\| \pi_{XY} \big| P_{M_nX^{i-1} Y^{i-1}} \big)\label{eqn:single_letter4}\\
&\ge\frac{1}{n} \sum_{i=1}^n D \big( P_{X_i Y_i |M_n} \big\| \pi_{XY} \big| P_{M_n}\big) \label{eqn:single_letter3}\\
&= D\big( P_{X_J Y_J | M_n J} \big\|\pi_{XY} \big| P_{M_n J } \big)=D\big( P_{XY| W} \big\|\pi_{XY} \big| P_{W } \big), \label{eqn:single_letter2}
\end{align}
where \eqref{eqn:single_letter4} follows from the chain rule for relative entropy, \eqref{eqn:single_letter3} follows from the convexity of the relative entropy, and \eqref{eqn:single_letter2} follows from introducing $J\sim \mathrm{Unif}[n]$  independent of $(M_n  , X^n,  Y^n)$ and by setting $X:=X_J$, $Y:=Y_J $ and $W:=(M_n, J)$. These identifications of the random variables satisfy the Markovity condition $X-W-Y$.  Using similar steps, we can show that $D(P_{XY}\|\pi_{XY}) \le\frac{1}{n}D(P_{X^nY^n}\|\pi_{XY}^n)$. Since   the code requires that  the final term in \eqref{eqn:single_letter}  to vanish, $D(P_{XY}\|\pi_{XY})$ also vanishes. This establishes the bound $D(P_{XY}\|\pi_{XY}) \le\epsilon$ for any $\epsilon>0$  and any $X-W-Y$ satisfying $P_{XY}=\pi_{XY}$. Taking $\epsilon\downarrow 0$ completes the proof of the converse part of Theorem~\ref{thm:wyner_cts}.

It is natural to wonder when  $C_{\Wyner}(\pi_{XY})$, $\tilC_{\Wyner}(\pi_{XY})$ and $\hatC_{\Wyner}(\pi_{XY})$, as defined in \eqref{eqn:CWyner_inf}, \eqref{eqn:tilC}, and \eqref{eqn:hatC} respectively coincide, beyond the  case in which $\pi_{XY}$ is finitely supported. This is partially addressed in the following   proposition due to the present authors~\cite{YuTan2020_exact}.  

\begin{proposition} \label{prop:reg_wyner_cts}
The following hold:
\begin{itemize}
\item If there exists a joint distribution $P_W P_{X|W} P_{Y|W}$ that attains $C_{\Wyner}(\pi_{XY})$ and satisfies $D_{1+s}(P_{X|W}P_{Y|W}\|P_{XY}|P_W)<\infty$ for some $s>0$, then $C_{\Wyner}(\pi_{XY})= \hatC_{\Wyner}(\pi_{XY})$.
\item Assume that  $\pi_{XY}$ is an absolutely continuous distribution on $\bbR^2$ with PDF $f_{XY}$  such that $C_{\Wyner}(\pi_{XY})= \hatC_{\Wyner}(\pi_{XY})$ (e.g., based on the sufficient condition in the point above), $f_{XY}$ is log-concave,\footnote{This means that $\log f_{XY}$ is concave on $\bbR^2$.} and $I(X;Y)<\infty$. For each $d>0$, define the constant 
\begin{equation}
\kappa_d :=\sup_{(x,y)\in [-d,d]^2} \left|\frac{\partial}{\partial x}\log f_{XY}(x,y) \right|+  \left|\frac{\partial}{\partial y}\log f_{XY}(x,y) \right|
\end{equation}
and $\epsilon_d:=1-\pi_{XY}( [ -d,d]^2)$. If $\epsilon_d\log(d\kappa_d)\to 0$ as $d\to\infty$, then all inequalities in \eqref{eqn:wyner_ct_ineq} become equalities. 
\end{itemize}
\end{proposition}
It holds that jointly Gaussian sources  satisfy both regularity conditions in Proposition~\ref{prop:reg_wyner_cts}. This will be discussed in detail in Section~\ref{sec:wyner_gauss}.

\subsection{Lossy Gray--Wyner System}\label{sec:lossy-gw}
%Despite the Gray--Wyner system as described in Section~\ref{sec:sim-GW} being applicable to discrete random variables, the expression for Wyner's common information in~\eqref{eqn:wyner_sim} remains valid for arbitrary random variables with the $\min$ replaced by an $\inf$, i.e., 
%\begin{equation}
%C_{\Wyner}(\pi_{XY})=\inf_{ P_WP_{X|W}P_{Y|W}: P_{XY}=\pi_{XY}}I_P(XY;W). \label{eqn:CWyner_inf}
%\end{equation}
%We have to use an $\inf$ here since the cardinality of $W$ can no longer be bounded and so the constraint set is not compact. 

 An operational interpretation for continuous random variables in terms of the {\em lossy} Gray--Wyner system \citep{Visw} was discovered by \citet{XuLiuChen}. %We  briefly explain the extension  of Wyner's common information to continuous random variables in terms of the {\em lossy} Gray--Wyner system . 
Recall that in the Gray--Wyner problem, one seeks to reconstruct a pair of sources losslessly. Obviously, this is only meaningful if the sources are discrete otherwise they cannot be reliably reconstructed with probability one for all finite rates. However, if one allows for the sources to be reconstructed to within some distortion levels, then it is meaningful to discuss the tradeoff between the rates $(R_0, R_1, R_2)$ and allowable distortions. To this end, we introduce two per-letter distortion measures $d_1:\calX\times\hat{\calX}\to[0,\infty)$ and $d_2:\calY\times\hat{\calY}\to [0,\infty)$ that operate on length-$n$ sequences as follows: $d_1(x^n , \hatx^n)=\frac{1}{n}\sum_{i=1}^n d_1(x_i,\hatx_i)$ and similarly for $d_2$.  Instead of demanding that the probability of error in~\eqref{eqn:prob_error} vanishes, in the lossy case, we only require the reconstructions $(\hatX^n,\hatY^n)$ in Fig.~\ref{fig:gw} to satisfy
% be lossily reconstructed to within distortion levels $\Delta_1$ and $\Delta_2$ respectively. That is, we require
\begin{equation}
 \limsup_{n\to\infty}\bbE\big[d_1(X^n,\hatX^n)\big]\!\le \!\Delta_1\;\;  \mbox{and}\; \;\limsup_{n\to\infty} \bbE\big[d_2(Y^n,\hatY^n)\big] \!\le \! \Delta_2 \label{eqn:gw_dist}
\end{equation}
for some permissible distortions $\Delta_1 $ and $\Delta_2$. 
This is known as the {\em lossy} Gray--Wyner system \citep{Visw}. 
Similarly to Definition~\ref{def:gw_R0},  we define the  {\em $(\Delta_1,\Delta_2)$-Pangloss-common information based on the lossy Gray--Wyner system}  $T_{\mathrm{GW}}(\pi_{XY}; \Delta_1, \Delta_2)$  to be the infimum of all common rates $R_0$ such that for each $\epsilon>0$,  there exists a sequence of Gray--Wyner codes satisfying
the distortion constraints in \eqref{eqn:gw_dist} and
\begin{equation}
R_0+R_1+R_2\le R_{XY}(\Delta_1,\Delta_2)+\epsilon, \label{eqn:pangloss_RD}
\end{equation}
for all sufficiently large $n$, 
 where  the {\em joint rate-distortion function}  is defined as 
\begin{equation}
 R_{XY}(\Delta_1,\Delta_2):= \inf_{  P_{\hatX\hatY|XY} : \bbE[d_1(X,\hatX)]\le \Delta_1,\; \bbE[d_2(Y,\hatY)]\le \Delta_2 } I(XY;\hatX\hatY). \label{eqn:joint_RDF}
\end{equation}
The {\em Pangloss plane} in this lossy case is given by the set of $(R_0,R_1,R_2)$ such that \eqref{eqn:pangloss_RD} holds with equality. The quantity  $T_{\mathrm{GW}}(\pi_{XY}; \Delta_1, \Delta_2)$, in general, depends on $(\Delta_1,\Delta_2)$.  However, \citet[Theorem~5]{XuLiuChen} showed that   in certain non-degenerate cases, this dependence vanishes.
\begin{theorem} \label{thm:gw_cts}
Let $P_WP_{X|W}P_{Y|W}$ be any distribution that achieves the infimum in the optimization problem in~\eqref{eqn:CWyner_inf}. Let the reproduction alphabets $\hat{\calX} =\calX$ and $\hat{\calY}=\calY$ and the two  distortion
measures $d_1$ and $d_2$ satisfy $d_1(x, \hatx)>d_1(x, x)=0$ for all $x\ne \hatx$ and $d_2(y, \haty)>d_2(y, y)=0$ for all $y\ne\haty$.  If the following conditions are satisfied
\begin{itemize}
\item For any $w\in\calW$, $x\in\calX$ and $y\in\calY$, $P_{W|XY}(w|x,y)>0$, 
\item There exists $\hatx\in \calX$ and $\haty\in  \calY$ such that 
\begin{equation}
\bbE[ d_1(X,\hatx)]<\infty \quad\mbox{and}\quad \bbE[ d_2(Y,\haty)]<\infty.
\end{equation}
\end{itemize}
Then there exists a positive constant $\gamma$ such that for all $0\le \Delta_1,\Delta_2\le\gamma$, 
\begin{equation}
T_{\mathrm{GW}}(\pi_{XY}; \Delta_1, \Delta_2)=C_{\Wyner}(\pi_{XY}). \label{eqn:gw_cont}
\end{equation}
\end{theorem}
In other words, under relatively mild conditions, for sufficiently small distortion levels, $T_{\mathrm{GW}}(\pi_{XY}; \Delta_1, \Delta_2)$ does not depend on $(\Delta_1,\Delta_2)$ and additionally,  there admits an operational interpretation of the expression on the right-hand side of \eqref{eqn:gw_cont}, i.e., it is the minimum common rate of the lossy Gray--Wyner system for small distortion levels. Moreover, if the regularity conditions of Proposition~\ref{prop:reg_wyner_cts} also hold, then the two operational definitions for the common information for continuous sources (as presented in Sections~\ref{sec:dist_ss_cts} and \ref{sec:lossy-gw}) coincide. This dovetails nicely with the discrete case.

From now on, we assume that $0\le \Delta_1,\Delta_2\le \gamma $ so it is permissible to write $T_{\mathrm{GW}}(\pi_{XY}; \Delta_1, \Delta_2)$ interchangeably as $T(\pi_{XY})$ or $C_{\Wyner}(\pi_{XY})$. 
\subsection{Jointly Gaussian Sources}\label{sec:wyner_gauss}
In this section, we consider jointly Gaussian sources. Our discussions up until this point inform us that there are two ways of computing Wyner's common information for such sources. In particular,  \citet{XuLiuChen}  and \citet{YuTan2020_exact}  used  Theorem~\ref{thm:gw_cts} and Proposition~\ref{prop:reg_wyner_cts}  respectively to compute  $T(\pi_{XY})$ for a jointly Gaussian source.  

%\begin{example} \label{ex:Gauss}
Let $(X,Y)\sim\pi_{XY}$ be a pair of jointly Gaussian random variables with covariance matrix given by
\begin{equation}
\bK = \begin{bmatrix}
1 & \rho \\ \rho & 1
\end{bmatrix}. \label{eqn:cov_mat}
\end{equation}
The constant $\rho \in (-1,1)$ is  known as the {\em correlation coefficient} of $X$ and $Y$. Without loss of generality, it suffices for us to consider $\rho\in [0,1)$. Otherwise, we can replace $X$ by $-X$ and the results go through {\em mutatis mutandis} with $\rho$ replaced by $-\rho$.\footnote{Equivalently, if we do not make the assumption that $\rho \in [0,1)$, the results for Gaussian sources here and in the following would hold with $\rho$ replaced by $|\rho|$.}  We expect that as $\rho\downarrow 0$, the common information  $C_{\Wyner}(\pi_{XY})$ should tend to $0$ as $X$ and $Y$ tend towards being independent. On the other hand as $\rho\uparrow 1$, $C_{\Wyner}(\pi_{XY})$ should increase as  $X$ and $Y$ tend towards being completely dependent.  The following proposition is due to \citet{XuLiuChen} and \citet{YuTan2020_exact}.
\begin{proposition}
For a jointly Gaussian source with correlation coefficient $\rho\in [0,1)$, Wyner's common information is
\begin{equation}
T(\pi_{XY})=C_{\Wyner}(\pi_{XY})= \frac{1}{2}\log\left(\frac{1+\rho}{1-\rho}\right) . \label{eqn:gaussian_wci}
\end{equation}
\end{proposition}
\begin{figure}
\centering
\includegraphics[width = .8\columnwidth]{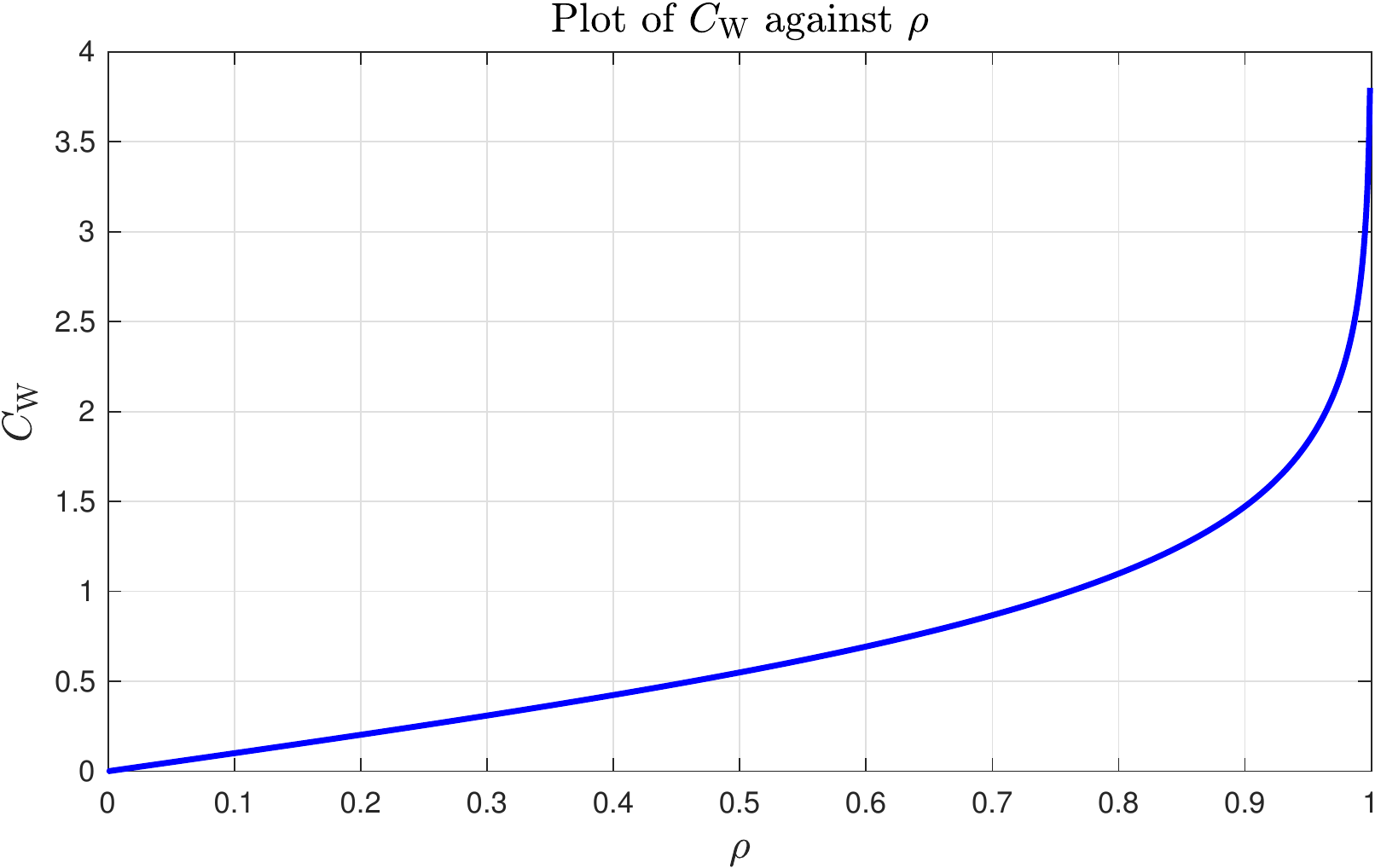}
\caption{Plot  of Wyner's common information for the jointly Gaussian source}
\label{fig:wyner_gauss}
\end{figure}

This function is plotted in Fig.~\ref{fig:wyner_gauss} and confirms our intuition about the limiting cases $\rho\downarrow 0$ and $\rho\uparrow 1$. Note that for continuous random variables, Wyner's common information $C_{\Wyner}(\pi_{XY})$ can increase without bound but for discrete random variables $C_{\Wyner}(\pi_{XY})$ is clearly overbounded (by $\log(|\calX||\calY|)$). 
%\end{example}

The test channels that achieve the infimum in~\eqref{eqn:CWyner_inf} for jointly Gaussian sources are also Gaussian. Indeed, the optimum $P_WP_{X|W}P_{Y|W}$  takes the form
\begin{equation}
X=\sqrt{\rho}\, W+\sqrt{1-\rho}\, N_1\quad\mbox{and}\quad Y=\sqrt{\rho}\, W+\sqrt{1-\rho}\, N_2,
\end{equation}
where $W$, $N_1$ and $N_2$ are independent standard Gaussian random variables. This does not come as a surprise in view of the optimum  common random variable and test channels for the DSBS; see Fig.~\ref{fig:dsbs}. Note that with this choice of test channels, and $0\le s\le \sqrt{ \frac{1+\rho}{2\rho} }$, 
\begin{equation}
D_{1+s}(P_{X|W}P_{Y|W}\|P_{XY}|P_W) =\frac{1}{2}\log\frac{1+\rho}{1-\rho}-\frac{1}{2s}\log\bigg(1- \frac{2s^2\rho}{1+\rho}\bigg).
\end{equation}
Hence, the first condition of Proposition~\ref{prop:reg_wyner_cts} is satisfied. It is also easy to verify by straightforward, albeit tedious, calculus that the second condition is satisfied, so all inequalities in~\eqref{eqn:wyner_ct_ineq} are equalities. 
 
\section{Generalizations and Applications}  \label{sec:gen_wyner}
We conclude this section by briefly mentioning some extensions of Wyner's common information and its applications that we do not discuss further in the monograph.  This list is by no means exhaustive and serves as a teaser for the reader to explore the many   generalizations of this useful quantity. 

\citet{LiuXuChen10} extended  Wyner's common information for two random variables to a quantity representing the common information among $N$ random variables, namely,
\begin{equation}
C_{\Wyner}( \pi_{X_1 X_2  \ldots  X_N}) = \min \,  I(X_1  X_2  \ldots  X_N ; W), \label{eqn:CWyner_N}
\end{equation} 
where the minimum is over all joint distributions $P_W\prod_{i=1}^N P_{X_i|W}$ such that the marginal  $P_{X_1 \ldots  X_N}$ equals   the target distribution $\pi_{X_1 \ldots  X_N}$.
This has the same operational interpretation in terms of distributed simulation of random variables and the Gray--Wyner  network with~$N$ decoders and $N+1$ encoders.  \citet{cuff13} considered a distributed channel synthesis problem and showed that in the absence of any shared common randomness between the encoder and decoder, the minimum rate required to synthesize a channel is exactly Wyner's common information. At the other extreme, if the amount of shared common randomness is sufficiently large, the rate required is the mutual information. We revisit the channel synthesis problem in Section~\ref{ch:ecs}. Recently, motivated by problems in caching, \citet{GastparSuha} found an operational interpretation of the following {\em relaxed} version of Wyner's common information
\begin{equation}
C_{\Wyner}^{(\delta)}( \pi_{XY})= \min_{P_{WXY}:  P_{XY}=\pi_{XY}, I(X;Y|W)\le\delta} \,  I(XY;W), \label{eqn:relWyn}
\end{equation}
which is parametrized by $\delta\ge 0$. Notice that if $\delta =0$, this quantity particularizes to the usual Wyner's common information as the constraint $I(X;Y|W)\le\delta$ reduces to the Markovity constraint $X-W-Y$. In another recent work, \citet{Graczyk} defined a conditional version of Wyner's common information 
\begin{equation}
C_{\Wyner} ( \pi_{XY|Z}|\pi_Z)=\min_{  P_{WZ}P_{X|WZ}P_{Y|WZ}: P_{XYZ}=\pi_{XYZ}}I(XY;W|Z),\!
\end{equation}
 which has obvious operational interpretations in terms of the distributed source simulation and Gray--Wyner problems when the terminals have access to correlated side-information $Z^n\sim\pi_Z^n$. The same authors also studied a quantity known as  the {\em relevant common information}. 
\begin{equation}
C_{\mathrm{Rel}}( \pi_{XY|S}\to \pi_S ) = \min_{\substack{P_{WXYS}: P_{XYS}=\pi_{XYS},\\ X-W-Y, \, S-(X,Y)-W}}\,  I(S;W),
\end{equation}
where the minimization is over all tuples of random variables $(X,Y,S,W)$ such that the marginal of $(X,Y,S)$ matches the given $\pi_{XY|S}\pi_S$, $X-W-Y$ and $S-(X,Y)-W$. As can be seen from the    two Markov chains, $C_{\mathrm{Rel}}( \pi_{XY|S}\to \pi_S )$ represents the common information in $(X,Y)$ that is {\em relevant} to a correlated random variable $S$. It has the interesting operational interpretation as the rate of the common randomness required
at two terminals to---through their inputs---strongly coordinate the output of a two-user multiple-access channel (MAC) according to a target distribution $\pi_S$.

\citet{tyagi2013common} introduced the notion of {\em $r$-interactive common information}, which is a variant of Wyner's common information. This quantity characterizes the minimum overall rate of interactive communication  required to generate a maximum rate secret key in an interactive manner between two parties.

Extending the seminal work of \citet{MA-Niesen} on the information-theoretic limits of caching, \citet{WangLimGastpar} formulated another caching problem from an information-theoretic perspective in which users' requests change over time. They cast the problem as a multi-terminal lossless source coding problem with side-information. For the $N$-user scenario, \citet{WangLimGastpar}  showed that the optimal caching strategy is closely related to $C_{\Wyner}( \pi_{X_1 X_2  \ldots  X_N})$ in~\eqref{eqn:CWyner_N}, which represents Wyner's common information for $N$ dependent random variables.

%The  In \cite{Graczyk}
%\section{Conditional and Relevant Common Information}
%Briefly discuss extensions of Wyner's CI to its conditional and relevant counterparts and highlight some operational interpretations. This is recent work due to Graczyk, Lapidoth and Wigger
%\section{Some Applications to Caching} \label{sec:caching}
%Discuss applications of Wyner's CI in caching through the work by Wang, Lim and Gastpar ``Information-theoretic caching: Sequential coding for computing,” 
\chapter{G\'acs--K\"orner--Witsenhausen's Common Information} \label{ch:gkw}
As mentioned at the start of Section~\ref{ch:wynerCI}, there are two well-known
notions of common information, the first  of which---Wyner's common information---has already been discussed in detail in 
Section~\ref{ch:wynerCI}. In this section,
we introduce the other classical notion of common information, namely, G\'acs--K\"orner--Witsenhausen's common
information. Recall that in the definition of Wyner's common information,
a  \emph{common} or   {\em shared source of randomness} $M_{n}$ is used to generate
a pair of random vectors $X^{n}$ and $Y^{n}$ in a distributed manner
such that the joint distribution of $(X^{n},Y^{n})$ is close to a
target product distribution $\pi_{XY}^{n}$. We now consider a counterpart
of this problem, illustrated in Fig.~\ref{fig:GKWCI}, in which a pair of random vectors $(X^{n},Y^{n})\sim\pi_{XY}^{n}$ is given, random variables $U=f(X^n)$ and $V=g(Y^n)$  
are to be extracted from $X^{n}$ and $Y^{n}$ individually using functions~$f_n$ and~$g_n$, and these random variables, called {\em common randomnesses},  should be {\em almost identical}.  This setting
was first considered by G\'acs and K\"orner  in their celebrated paper \cite{gacs1973common}  in which 
they defined the common information between $X$ and $Y$, jointly distributed as $\pi_{XY}$,
as the maximum information rate of the common randomness $U$ or,  equivalently, $V$.
This notion of common information was later coined \emph{G\'acs--K\"orner--Witsenhausen's} or \emph{GKW's common information}. In fact, G\'acs and K\"orner \cite{gacs1973common} were the first to investigate the notion of common information in 1973,
prior to Wyner's work \cite{WynerCI} in 1975. 

In this section, we review GKW's common information.
In Section~\ref{sec:GKW_CI}, we introduce the distributed randomness
extraction system, and define GKW's common
information in the context of this system. In Section~\ref{sec:property_GKW},  we introduce several properties of GKW's
common information. We also mention some probability- and graph-theoretic 
interpretations of GKW's common information.
We verify that GKW's common information is
zero for the DSBS and also for bivariate Gaussian sources; this observation motivates Part~\ref{part:three} of the monograph. In Section~\ref{sec:GW_system},
we introduce an operational interpretation of GKW's
common information in the context of the Gray--Wyner lossless source coding system~\cite{GrayWyner}. GKW's
common information turns out to be the maximum common rate under some 
conditions on the sums of the private and common rates of the messages. In Section~\ref{sec:input_dist}, we discuss an operational interpretation of GKW's
common information due to the present authors that is not too well-known. Specifically, we relate it to the channel capacity in which the input
distribution is fixed to be a given product distribution. 
%We provide
%another operational interpretation of GKW's
%common information by observing that the channel capacity for this
%channel coding problem is exactly GKW's common
%information. 
Finally, we discuss some extensions and applications in Section~\ref{sec:gen_GKW}. %sources and sources on arbitrary (e.g., Polish) alphabets.

\section{Distributed Randomness Extraction}
\label{sec:GKW_CI}

%\begin{centering}
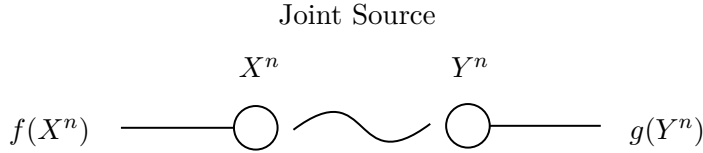
\begin{figure}
\centering

\tikzset{every picture/.style={line width=0.75pt}} %set default line width to 0.75pt        

\begin{tikzpicture}[x=0.75pt,y=0.75pt,yscale=-1,xscale=1]
%uncomment if require: \path (0,3744); %set diagram left start at 0, and has height of 3744

%Shape: Circle [id:dp7898170317060629] 
\draw   (336.01,326) .. controls (336.01,319.93) and (340.93,315.01) .. (347,315.01) .. controls (353.08,315.01) and (358,319.93) .. (358,326) .. controls (358,332.08) and (353.08,337) .. (347,337) .. controls (340.93,337) and (336.01,332.08) .. (336.01,326) -- cycle ;
%Shape: Circle [id:dp6806256956760541] 
\draw   (229.01,327) .. controls (229.01,320.93) and (233.93,316.01) .. (240,316.01) .. controls (246.08,316.01) and (251,320.93) .. (251,327) .. controls (251,333.08) and (246.08,338) .. (240,338) .. controls (233.93,338) and (229.01,333.08) .. (229.01,327) -- cycle ;
%Straight Lines [id:da7152274637866838] 
\draw    (229.01,327) -- (213.8,327) -- (172,327) ;
%Straight Lines [id:da39191847920404865] 
\draw    (358,326) -- (414,326) ;
%Curve Lines [id:da7613964163018121] 
\draw    (259,328) .. controls (299,298) and (288,355) .. (328,325) ;

% Text Node
\draw (230,290) node [anchor=north west][inner sep=0.75pt]    {$X^{n}$};
% Text Node
\draw (337,290) node [anchor=north west][inner sep=0.75pt]    {$Y^{n}$};
% Text Node
\draw (114,320) node [anchor=north west][inner sep=0.75pt]    {$f ( X^{n} )$};
% Text Node
\draw (427,320) node [anchor=north west][inner sep=0.75pt]    {$g( Y^{n} )$};
% Text Node
\draw (250,263) node [anchor=north west][inner sep=0.75pt]   [align=left] {Joint Source};
\end{tikzpicture}
\caption{\label{fig:GKWCI}The distributed randomness extraction problem}
\end{figure}

Consider the distributed  randomness extraction problem illustrated in Fig. \ref{fig:GKWCI}. For a
joint source  $(X,Y)\sim\pi_{XY}$, we use a pair
of functions $f$ and $g$, respectively acting on $X^{n}$ and
$Y^{n}$, to generate random variables $f(X^n)$ and $g(Y^n)$. Our goal is
to ensure that $f(X^n)$ and $g(Y^n)$ are equal with high probability  and, at the same time, to maximize the information rate
of $f(X^n)$ or, equivalently,  $g(Y^n)$. Formally, we define distributed extraction codes
and the $\eps$-common information as follows. These definitions are due to \citet{csiszar2000common}; we discuss the original formulation by \citet{gacs1973common} in Remark~\ref{rmk:gk}.

\begin{definition} An {\em $X$-sided $(n,R)$-distributed
extraction code} consists of a pair of (deterministic) functions\footnote{Without loss of generality, we can set the codomains of $f$ and $g$
to be the set of natural numbers $\mathbb{N}$, i.e., $f:\mathcal{X}^{n}\to\mathbb{N}$ and $g:\mathcal{Y}^{n}\to\mathbb{N}$.} $(f,g)$  defined respectively on $\mathcal{X}^{n}$ and $\mathcal{Y}^{n}$
such that 
\begin{equation}
\frac{1}{n}H\big(f(X^{n})\big)\geq R.\label{eq:GKW-2}
\end{equation}
A \emph{$Y$-sided $(n,R)$-distributed extraction code} is defined 
similarly, but with \eqref{eq:GKW-2} replaced by 
$
\frac{1}{n}H(g(Y^{n}))\geq R$.
\end{definition} 

\begin{definition} \label{def:GKW_CI} Fix $\eps \in (0,1)$.  The \emph{maximal $X$-sided $\eps$-error extraction rate $S_{\eps}^{(X)}(\pi_{XY})$} between a pair of random
variables $(X,Y)\sim\pi_{XY}$ is defined as  the supremum of all
rates $R$ such that there exists a sequence of $X$-sided $(n,R)$-distributed
extraction codes $\{(f_{n},g_{n})\}_{n\in\bbN}$ satisfying   
\begin{equation}
\Pr\big(f_{n}(X^{n})\neq g_{n}(Y^{n})\big) \le\eps,  \label{eqn:asymp_prob_disagree}
\end{equation}
%$\Pr(f_{n}(X^{n})\neq g_{n}(Y^{n}))\le\eps$
 for all sufficiently large $n$, 
where $(X^{n},Y^{n})\sim\pi_{XY}^{n}$.
The  \emph{maximal $Y$-sided $\varepsilon$-error extraction rate $S_{\eps}^{(Y)}(\pi_{XY})$
}between $(X,Y)\sim\pi_{XY}$ is defined analogously. \end{definition}

One can easily verify that the maximal $X$- and $Y$-sided  $\varepsilon$-error extraction rates
  do not differ significantly in the limit as $n\to\infty$ and $\eps\downarrow0$.
This is because, by Fano's inequality \cite[Section~2.1]{elgamal}, 
\begin{equation}
1+\eps\log|\supp(f_{n}(X^{n}))|\ge H\big(f_{n}(X^{n}) | g_{n}(Y^{n})\big).\label{eq:GKWFano}
\end{equation}
Note that $|\supp(f_{n}(X^{n}))|\le|\mathcal{X}|^{n}$ since $f_{n}$
is a deterministic function. Therefore, 
\begin{equation}
\lim_{\eps\downarrow0}\lim_{n\to\infty}\frac{1}{n}H\big( f_{n}(X^{n})|g_{n}(Y^{n})\big)=0.
\end{equation}
By symmetry, it also
holds that 
\begin{equation}
\lim_{\eps\downarrow0}\lim_{n\to\infty}\frac{1}{n}H\big( g_{n}(Y^{n})|f_{n}(X^{n}) \big)=0.
\end{equation}
Combining
these two limits yields that 
\begin{align}
&\hspace{-.25in}\lim_{\eps\downarrow0}\lim_{n\to\infty}\frac{1}{n} \big(H(f_{n}(X^{n}))-\frac{1}{n}H(g_{n}(Y^{n})\big)\nn\\*
&\hspace{-.25in}=\lim_{\eps\downarrow0}\lim_{n\to\infty} \Big(  H\big( f_{n}(X^{n})|g_{n}(Y^{n})\big)- H\big( g_{n}(X^{n})|f_{n}(Y^{n})\big)\Big)=0. \label{eqn:diff_ents}
\end{align}
%$\frac{1}{n}H(f_{n}(X^{n}))-\frac{1}{n}H(g_{n}(Y^{n}))\to0$
%as $n\to\infty$ and $\eps\downarrow0$.
 The exact expressions
for the  maximal $X$- and $Y$-sided $\eps$-extraction rate
as $\eps\downarrow0$ are given by G\'acs and K\"orner \cite{gacs1973common}.

\begin{theorem}\label{thm:GKW_CI} For a joint  source  $(X,Y)\sim\pi_{XY}$,
it holds that 
\begin{equation}
\lim_{\eps\downarrow0}S_{\eps}^{(X)}(\pi_{XY})=\lim_{\eps\downarrow0}S_{\eps}^{(Y)}(\pi_{XY})=C_{\mathrm{GKW}}(\pi_{XY}),\label{eq:GKW-1-1}
\end{equation}
where 
\begin{equation}
C_{\mathrm{GKW}}(\pi_{XY}):=\max_{f,g:f(X)=g(Y)}H \big(f(X)\big),\label{eq:GKWGKW_CI}
\end{equation}
and where the maximization   is taken over all pairs of
deterministic functions $(f,g)$  defined respectively on $\mathcal{X}$
and $\mathcal{Y}$ such that $f(X)=g(Y)$ with $\pi_{XY}$-probability one. \end{theorem}

In the literature, for example in \citet{elgamal}, \emph{$C_{\mathrm{GKW}}(\pi_{XY})$ } 
is known as  \emph{GKW's common information}. Theorem~\ref{thm:GKW_CI} says that the  maximal $X$- and $Y$-sided $\eps$-error extraction rates are equal to GKW's common information, so in the following, we will use these terminologies interchangeably. 
It is clear that   the objective  function in the maximization in \eqref{eq:GKWGKW_CI}
can be replaced by $H(g(Y))$ since $f(X)$ and $g(Y)$ are constrained to be equal almost surely. Roughly speaking, the quantity $C_{\mathrm{GKW}}(\pi_{XY})$
 corresponds to a single-letter version (i.e., $n=1$ version) of the maximal $X$-  or $Y$-sided $0$-error extraction rates (defined formally in Definition~\ref{def:gkw_0error}), in the sense that $C_{\mathrm{GKW}}(\pi_{XY})$
is equal to the supremum of all rates $R$ such that $H(f(X))\geq R$
and $\Pr\big(f(X)\neq g(Y)\big)=0$. 

\begin{proof}[Proof of Theorem \ref{thm:GKW_CI}] By symmetry, it clearly
suffices to prove that $\lim_{\eps\downarrow0}S_{\eps}^{(X)}(\pi_{XY})=C_{\mathrm{GKW}}(\pi_{XY})$.
We first prove  that $\lim_{\eps\downarrow0}\! S_{\eps}^{(X)}(\pi_{XY})\!\ge \! C_{\mathrm{GKW}}(\pi_{XY})$, the achievability part. 
Let $f^{*} : \calX\to\calU$ and $ g^{*}:\calY\to\calV$ be an  optimal pair of functions that attains the
maximum in~\eqref{eq:GKWGKW_CI}, where $\calU$ and $\calV$ are two fixed  sets (that can  be assumed to be the same). Then, let $f_{n}:x^n\in \calX^n \mapsto  (f^{*}(x_1),f^{*}(x_2),\ldots,f^{*}(x_n))\in \calU^n$
and $g_{n}:y^n\in \calY^n \mapsto  (g^{*}(y_1),g^{*}(y_2),\ldots,g^{*}(y_n))\in \calV^n$.  Then,  by the mutual independence of $X_1, X_2,\ldots, X_n$, 
\begin{equation}
\frac{1}{n}H(f_{n}(X^{n}))=H(f^{*}(X))=C_{\mathrm{GKW}}(\pi_{XY}) \label{eq:HHC}
\end{equation}
and 
\begin{equation}
\Pr(f_{n}(X^{n})=g_{n}(Y^{n}))=1. \label{eqn:agreement1}
\end{equation}
Therefore, $S_{\eps}^{(X)}(\pi_{XY})\ge C_{\mathrm{GKW}}(\pi_{XY})$
for any $\eps\in (0,1)$. 

We next prove $\lim_{\eps\downarrow0} S_{\eps}^{(X)}(\pi_{XY})\le C_{\mathrm{GKW}}(\pi_{XY})$, the converse part.
The proof is based on the following lemma due to \citet[Lemma~1.1]{csiszar2000common}. As assumed in the achievability part, let $(f^{*},g^{*})$ be
an optimal pair of functions attaining the maximization in \eqref{eq:GKWGKW_CI}.
Let $W=f^{*}(X)=g^{*}(Y) \in \calW$; this random variable is called the \emph{common part} of $(X,Y)\sim\pi_{XY}$. 
%\cite[Lemma 1.1]{csiszar2000common} 
\begin{lemma} \label{lem:csiszar}
 For $(X^{n},Y^{n})\sim\pi_{XY}^{n}$, let $U$ and $V$ be two random variables
such that $U-X^{n}-Y^{n}-V$ and 
\begin{equation}
\Pr ( U\neq V ) \leq\eps\label{eq:GKWUV}
\end{equation}
for $\eps>0$. Then 
\begin{equation}
\min_{h:\mathcal{W}^{n}\to\mathcal{U}}\Pr\big(U\neq h(W^{n})\big)\leq\delta(\eps),
\end{equation}
where $\delta:(0,\infty)\to (0,\infty)$ is a function that  only depends
on $\pi_{XY}$, is independent of $n$, and has the property that   $\delta(\eps)\downarrow0$
as $\eps\downarrow0$. \end{lemma}

This lemma is proven by the tensorization property of the \emph{conditional
maximal correlation} (the unconditional version of the maximal correlation was defined in \eqref{eq:mc}). It uses some results of \citet{witsenhausen1975sequences}, but we will not elaborate on it here; see \cite[Lemma 1.1]{csiszar2000common}. Using Lemma~\ref{lem:csiszar}, we know that $\min_{h_{n}}\Pr\left( f_{n}(X^{n})\neq h_{n}(W^{n})\right) \leq\delta(\eps)$,
where the minimization is taken over all functions $h_{n}$ defined
on $\mathcal{W}^{n}$. Similarly to \eqref{eq:GKWFano}, by Fano's inequality
\cite[Section~2.1]{elgamal}, for any function $h_{n} : \calW^n\to\calU$, 
\begin{equation}
1+\delta(\eps)\log|\supp(f_{n}(X^{n}))|\ge H\big(f_{n}(X^{n})|h_{n}(W^{n})\big).\label{eq:GKWFano-1}
\end{equation}
Following an argument similar to the one leading to \eqref{eqn:diff_ents}, we
have that 
\begin{equation}
\lim_{\eps\downarrow0}\lim_{n\to\infty}\frac{1}{n}\Big(H\big(f_{n}(X^{n})\big)- H\big(h_{n}(W^{n})\big) \Big)=0. \label{eqn:limlimfh}
\end{equation}
By combining \eqref{eqn:limlimfh} with  the fact that $H\big(h_n(W^n)\big)\le H(W^n)=nH(W)$,
\begin{align}
\lim_{\eps\downarrow0}  \limsup_{n\to\infty}\frac{1}{n}H\big(f_{n}(X^{n})\big) %&\le  \limsup_{n\to\infty}\frac{1}{n} H\big(h_{n}(W^{n})\big) \\
%& \le \limsup_{n\to\infty}\frac{1}{n}H(W^n) = 
\le H(W)  =C_{\mathrm{GKW}}(\pi_{XY}),
\end{align}
which implies that $\lim_{\eps\downarrow0}S_{\eps}^{(X)}(\pi_{XY})\le C_{\mathrm{GKW}}(\pi_{XY})$.
\end{proof}

From the proof of Theorem~\ref{thm:GKW_CI}, and in particular  \eqref{eqn:agreement1}, we know that the constraint on the   probability
of disagreement $\Pr(f_{n}(X^{n})\neq g_{n}(Y^{n}))\le\eps$  (where $\eps\in(0,1)$) can
be strengthened  significantly to the zero-error version, i.e.,  
\begin{equation}
\Pr\big(f_{n}(X^{n})\neq g_{n}(Y^{n})\big)=0 \qquad\mbox{for all} \;\, n\in\bbN. \label{eqn:zeroerr} 
\end{equation}
\begin{definition} \label{def:gkw_0error}
The \emph{maximal $X$-sided (resp.\ $Y$-sided) $0$-error extraction rate $S_{\eps}^{(X)}(\pi_{XY})$} %\emph{$X$-sided} (resp.\ \emph{$Y$-sided})  \emph{zero-error common information} $\tilS_{0}^{(X)}(\pi_{XY})$ 
(resp.\ $\tilS_{0}^{(Y)}(\pi_{XY})$) is the supremum  of all rates $R$ such that there exists a sequence of $(n,R)$-distributed extraction codes $\{(f_n,g_n)\}_{n\in\bbN}$ such that \eqref{eqn:zeroerr} holds. %  the normalized entropy of $f_n(X^n)$ (resp.\ $g_nY^n$) under the zero-error constraint in \eqref{eqn:zeroerr}.
\end{definition}
By definition, $\tilS_{0}^{(X)}(\pi_{XY})=\tilS_{0}^{(Y)}(\pi_{XY})$. Moreover, these  strengthened definitions are the same as 
the limiting values of  maximal $X$- and $Y$-sided $\eps$-error extraction rates  as $\eps \downarrow 0$, i.e., 
\begin{equation}
\tilS_0^{(X)}(\pi_{XY}) = \lim_{\eps\downarrow0}S_\eps^{(X)}(\pi_{XY})\quad \mbox{and}\quad\tilS_0^{(Y)}(\pi_{XY}) = \lim_{\eps\downarrow0}S_\eps^{(Y)}(\pi_{XY}).
\end{equation}
This is easy to see as, on one hand,  according to~\eqref{eqn:agreement1}, the functions~$f_n$ and~$g_n$, defined in the proof of Theorem~\ref{thm:GKW_CI},   satisfy the zero-error constraint. Hence, 
 $\tilS_{0}^{(X)}(\pi_{XY})\ge C_{\mathrm{GKW}}(\pi_{XY})$. 
On the other hand, observe that  the   maximal $X$-sided $\eps$-error extraction rate  is no larger than $ S_{\eps}^{(X)}(\pi_{XY})$ for any $\eps\in  (0,1)$, 
since  an error is allowed in the latter.  Combining this with Theorem  \ref{thm:GKW_CI}  yields that $\tilS_{0}^{(X)}(\pi_{XY}) \le C_{\mathrm{GKW}}(\pi_{XY})$. 
These observations are summarized in the following theorem.  

\begin{theorem} \label{thm:exactgkw} It holds that
\begin{equation}
\tilS_{0}^{(X)}(\pi_{XY})=\tilS_{0}^{(Y)}(\pi_{XY})=C_{\mathrm{GKW}}(\pi_{XY}).\label{eq:GKW-1-1-2}
\end{equation}
\end{theorem}

%\subsection{Original Formulation of GKW's Common Information}
\begin{remark} \label{rmk:gk}
The formulation of GKW's common information as presented
in Definition~\ref{def:GKW_CI} was  introduced by \citet{csiszar2000common}. This is not the original
definition introduced in \citet{gacs1973common}. In G\'acs
and K\"orner's original formulation, instead of the normalized
entropy of the common part $W_n \in \calW_n$ of $X^n$ and $Y^n$ (cf.~\eqref{eq:GKW-2}), the information rate is measured in terms of the 
 \emph{exponent} of its alphabet size $\frac{1}{n}\log|\calW_n|$. % $W_n \in \calW_n$ of $X^n$ and $Y^n$, i.e., $\frac{1}{n}\log|\calW_n|$. 
 To ensure that the exponent of the alphabet size is an ``effective''
measure of the information rate, the distribution of the random variable
$f_{n}(X^{n})$ (and $g_{n}(Y^{n})$) is required to be close to the
uniform distribution on its alphabet. Hence, lossless source coding
is used in G\'acs
and K\"orner's setting to implement this requirement. Specifically, the juxtaposition  of $f_n$ with another function $\tilf_{n}$ on $\mathcal{X}^{n}$ 
is required to be an \emph{almost optimal} fixed-length lossless source
code for~$X^{n}$. A similar constraint was also imposed for the function
$g_{n}$. These force the outputs of $f_{n}$ and $g_{n}$ to be close
to uniform on $\calW_n$. %, otherwise, the optimality cannot be achieved.
Indeed, the formulation by \citet{gacs1973common} is analogous to fixed-length source coding~\cite{Shannon48} while the formulation by \citet{csiszar2000common} is  analogous to weak variable-length source coding as studied in~\citet{han2006weak} and~\citet{koga2005weak} %, \citet{KPV15}, and \citet{sakai2021third}
among others. 
% variable-length source coding to fixed-length source
%coding. 

For their setting, G\'acs and K\"orner showed that the maximum asymptotic   exponent  $\liminf_{n\to\infty}\frac{1}{n}\log|\calW_n|$ under the asymptotic probability
of disagreement constraint in~\eqref{eqn:asymp_prob_disagree}
is   $C_{\mathrm{GKW}}(\pi_{XY})$ for all $\eps \in (0,1)$. Hence,
the strong converse holds for G\'acs and K\"orner's formulation, while
it does not hold for Csisz\'ar and Narayan's formulation (i.e., Definition~\ref{def:GKW_CI}). This observation   resembles lossless source coding in that the strong converse holds for the fixed-length version \cite{Wolfowitz} but not the weak variable-length version~\cite{koga2005weak,han2006weak, KPV15,sakai2021third}. 
\end{remark}

%This observation resembles lossless source coding,
%for which the strong converse holds for variable-length source coding,
%but does not hold for fixed-length source coding \cite{han2006weak}. 

\section{Properties of GKW's Common Information}\label{sec:property_GKW}

\enlargethispage{\baselineskip}
We next introduce several interesting properties of $C_{\mathrm{GKW}}(\pi_{XY})$ that elucidate more insights on its properties. 
%\subsection{Computation of $C_{\mathrm{GKW}}(\pi_{XY})$}
\subsection{Interpretation in terms of Markov chains and bipartite graphs}\label{sec:interp}
We first focus on the computation of $C_{\mathrm{GKW}}(\pi_{XY})$, which can be understood using Markov chains and bipartite graphs. Consider a discrete-time Markov chain $X_{1}-Y_{1}-X_{2}-Y_{2}-\ldots$ in which   $P_{X_{1}} =\pi_{X}$ is the initial distribution  and  the transition probability distributions satisfy
\begin{align}
P_{Y_{i}|X_{i}}  =\pi_{Y|X} \quad \mbox{and} \quad  P_{X_{i+1}|Y_{i}}  =\pi_{X|Y}
\end{align}
for all $i \in\bbN$.
Then, the subchain $X_{1}-X_{2}-\ldots$ is a \emph{time-homogeneous} Markov chain
with initial distribution $\pi_X$ and transition probability distribution
\begin{equation}
P_{X_{i+1}|X_{i }} =\pi_{\hat{X}|X},
\end{equation}
where
\begin{equation}
\pi_{\hat{X}|X}(x'|x):= \sum_{ y\in\calY}\pi_{X|Y}(x' |y)\pi_{Y|X}(y|x )\quad \mbox{for all}\;\, (x,x')\in\calX^2.
\end{equation}
Obviously, $\pi_{X}$ is the stationary  distribution of $X_{1}-X_{2}-\ldots$. Moreover, this subchain is \emph{reversible} since the stationary  distribution and transition probability distribution  satisfy   
\begin{equation}
\pi_{\hat{X}|X}(x'|x) \pi_X(x) = \pi_{\hat{X}|X}(x|x') \pi_X(x') \quad \mbox{for all}\;\, (x,x')\in\calX^2.
\end{equation}

Now we recap a few more definitions from Markov chains; see, for example, \citet[Chapter~4]{gallagerSP}. For two states $x,x'\in\mathcal{X}$,  $x'$ is \emph{accessible} from $x$, abbreviated as $x\to x'$, if $P_{X_N|X_1} (x'|x)>0$ for some positive integer $N$. The condition $P_{X_N|X_1} (x'|x)>0$   is  also equivalent to the fact that there exists a sequence of states (also called a {\em walk}) $(x_{1},x_{2},\ldots,x_{N})$ such
that $x_{1}=x,x_{N}=x'$, and $\pi_{\hat{X}|X}(x_{i}|x_{i-1})>0$
for $2\le i\le N$. If $\calX$ is the support of $\pi_X$, then the condition $P_{X_N|X_1} (x'|x)>0$ is also equivalent to  $P_{X_1 X_N} (x,x')>0$.
%Two elements $x,\hat{x}\in\mathcal{X}$ are called \emph{communicating},
%denoted as $x\leftrightarrow\hat{x}$, if $\Pr(X_{N}=\hat{x}|X_{1}=x)>0$ for
%some positive integer $N$. Equivalently, $x\leftrightarrow\hat{x}$ if there
%is 
%a sequence $(x_{1,x_{2},\ldots,x_{N})$ such
%that 
%$x_{1}=x,x_{N}=\hat{x}$, and $\pi_{XY}(x_{i},y_{i})>0,\pi_{XY}(y_{i},x_{i+1})>0$
%for $1\le i\le N-1$. 
Two distinct states $x$  and $x'$  \emph{communicate}, abbreviated as $x \leftrightarrow x'$, if  $x$  is
accessible from $x'$ and $x'$ is accessible from  $x$. By definition, for a stationary and reversible Markov chain (e.g., the one considered here), the joint distribution  $P_{X_i X_j}$ 
of $(X_i,X_j)$ for $i\neq j$ satisfies    $P_{X_i X_j}(x,x')=P_{X_i X_j}(x', x)$ for all $x,x'\in\mathcal{X}$. Hence,   $x\to x'$ (or $x'\to x$) is equivalent to $x\leftrightarrow x'$.
Obviously, ``$\leftrightarrow$'' is an
equivalence relation, since it satisfies the following three properties:
\begin{itemize}
\item Reflexivity: ${\displaystyle a\leftrightarrow a}$; 
\item Symmetry: $a\leftrightarrow b$ if and only if ${\displaystyle b\leftrightarrow a}$;
\item Transitivity: If $a\leftrightarrow b$ and ${\displaystyle b\leftrightarrow c}$ then ${\displaystyle a\leftrightarrow c.}$ 
\end{itemize}
This allows us to define {\em equivalence classes} for  the relation $\leftrightarrow$. In the language of Markov chains, these are known as  \emph{communicating
classes}, or simply {\em classes}. A set $\mathcal{A}\subset\mathcal{X}$ is termed a \emph{class}   of $\mathcal{X}$ if $\calA$ is non-empty and for all $x\in\calA$, each state $x'\in\calX\setminus\{x\}$ satisfies $x'\in\calA$ if $x\leftrightarrow x'$ and $x'\notin\calA$ if $x\not\leftrightarrow x'$. The classes of $\calX$, denoted as $\calX_i, i\in [r]$, form a partition\footnote{A {\em partition} of a set $\calX$ is a collection of sets $\{\calX_\alpha\}_{\alpha\in\calA}$ such that $\cup_{\alpha\in\calA}\calX_\alpha=\calX$ and $\calX_\alpha\cap\calX_{\alpha'}=\emptyset$ for all $\alpha\neq\alpha'$.} of~$\calX$.  The   classes
of $\mathcal{Y}$ are similarly denoted as $\mathcal{Y}_{j}, j \in [s]$.

Clearly, the Markov chain transitions from a state in $\calX_i$ to a state in  $\calY_j$ with  
positive probability in the sense that $\Pr(Y\in\mathcal{Y}_{j}|X\in\mathcal{X}_{i})=\bone\{i=j\}$ \cite[Theorem 4.2.9]{gallagerSP}
where $(X,Y)\sim\pi_{XY}$, and vice versa.  Hence, $r=s$ and 
\begin{equation}
\Pr\big(Y\in\mathcal{Y}_{j}\big|X\in\mathcal{X}_{i}\big)=\Pr\big(X\in\mathcal{X}_{i}\big|Y\in\mathcal{Y}_{j}\big)=\bone\{i=j\}.\label{eq:GKWtransition}
\end{equation}
Such a partition of $\mathcal{X}$ (or $\mathcal{Y}$), termed an
\emph{ergodic decomposition} \cite{gacs1973common}, is unique. 
If we denote $i^{*}(x)$ as the index $i$ such that $x\in\mathcal{X}_{i}$,
and similarly, $j^{*}(y)$ as the index $j$ such that $y\in\mathcal{Y}_{j}$,
then by \eqref{eq:GKWtransition}, $i^{*}(X)=j^{*}(Y)$. The pair of functions  $(i^{*},j^{*})$ attains the maximization
in \eqref{eq:GKWGKW_CI}. This is because, on one hand, by definition,
$C_{\mathrm{GKW}}(\pi_{XY})\ge H(i^{*}(X))$.
On the other hand, for $(f,g)$ such that $f(X)=g(Y)$ almost surely, %we have that 
\begin{equation}
f(X_{1})=g(Y_{1})=f(X_{2})=g(Y_{2})=\ldots,\label{eq:GKWfgfg}
\end{equation}
where $X_{1}-Y_{1}-X_{2}-Y_{2}-\ldots$ is the Markov chain as defined at the start of this section.
Denote the image of $f$ as $\mathcal{U}$. Then, by \eqref{eq:GKWfgfg},
 for each pair of distinct elements $(u,u')$ of $\mathcal{U}$,
we have 
\begin{equation}
\Pr\big(X_{m}\in f^{-1}(u')\big|X_{1}\in f^{-1}(u)\big)=0\qquad 
\mbox{for all}\;\,  m\in\bbN.     
\end{equation}
Hence, for each $u\in\mathcal{U}$, $f^{-1}(u) \subset\calX$
is a  class or the union of several  
classes.  This means that $f(X)$ is determined by $i^{*}(X)$, which
in turn implies that $C_{\mathrm{GKW}}(\pi_{XY})\le H(i^{*}(X))$.
 Hence, $(i^{*},j^{*})$ is the unique pair of functions (up to a
bijection) attaining the maximization in GKW's common information in~\eqref{eq:GKWGKW_CI}.

The ergodic decomposition   can be also expressed in the language
of graph theory. Without loss of generality, we may assume that $\mathcal{X}\cap\mathcal{Y}=\emptyset$.
%otherwise, $\mathcal{Y}$ can be mapped to another set by one-to-one
%correspondence satisfying this condition. 
Consider a  (undirected) bipartite
graph in which the two sets of  vertices are represented by $\calX$ and $\calY$ 
%by the elements in $\mathcal{X}$ and the elements in $\mathcal{Y}$,
and a pair of vertices $(x,y)\in\mathcal{X}\times\mathcal{Y}$
is adjacent if $\pi_{XY}(x,y)>0$. In an
undirected graph, a vertex $v\in\mathcal{X}\cup\mathcal{Y}$ is \emph{reachable}
from a vertex $u\in\mathcal{X}\cup\mathcal{Y}$ if there is a path
from $u$ to $v$. 
%The concept of ``reachable'' is almost the same
%as the concept of ``communicating'' described above, except that
%the former is defined for the set $\mathcal{X}\cup\mathcal{Y}$, while
%the latter is defined for the set $\mathcal{X}$. 
Reachability is
also an equivalence relation, and the equivalence classes of this equivalence
relation are $\mathcal{X}_{i}\cup\mathcal{Y}_{i},i\in[r]$, where
$\mathcal{X}_{i}$ and $\mathcal{Y}_{i},i\in[r]$ are the communicating
classes. The induced subgraphs formed by these equivalence classes
are known as the \emph{connected  components} of the graph. The ergodic
decomposition corresponds to the  decomposition of the graph into connected
components. Fig.~\ref{fig:GKW_decomposition} illustrates an   example of a joint distribution $\pi_{XY}$ together with its ergodic decomposition.

\input{gkw_graph}

\subsection{Connections to Other Quantities}
We now provide an alternative expression for $C_{\mathrm{GKW}}(\pi_{XY})$, which looks similar to the expression  for Wyner's common information in \eqref{eqn:wyner_sim}. This  characterization is due to \citet{ahlswede2006common}.
\begin{proposition} \label{prop:altexpr} It holds that
\begin{align}
C_{\mathrm{GKW}}(\pi_{XY}) &   =\max_{\substack{P_{WXY}:P_{XY}=\pi_{XY},\\
W-X-Y,\;W-Y-X
}
}I(XY;W).\label{eq:GKW-3}
\end{align}
\end{proposition}
We remark that the objective function in the maximization above  can be replaced by $I(X;W)$ or $I(Y;W)$, since the Markov chains $W-X-Y$ and $W-Y-X$ are assumed. 
\begin{proof}
Let $U=f^{*}(X)=g^{*}(Y)$ be the common part of $(X,Y)\sim\pi_{XY}$
with $(f^{*},g^{*})$ denoting the optimal pair of functions attaining
the maximization in $C_{\mathrm{GKW}}(\pi_{XY})$ in~\eqref{eq:GKWGKW_CI}. By setting $W=U$, we conclude that 
the right-hand side of \eqref{eq:GKW-3} is at least  $C_{\mathrm{GKW}}(\pi_{XY})$. 

To prove the opposite inequality, we first state the following lemma. 
\begin{lemma}\label{lem:markovWUX}
Every $P_{WXY}$ that satisfies the constraints in~\eqref{eq:GKW-3}    also satisfies $W-U-X$. 
\end{lemma}
Lemma~\ref{lem:markovWUX} then implies that $I(XY;W)=I(X;W)\le I(X;U)=H(U)$. Hence, the right-hand side of~\eqref{eq:GKW-3} is at most $C_{\mathrm{GKW}}(\pi_{XY})$.
Combining the two points above yields the equality in~\eqref{eq:GKW-3}.
Hence, it remains to prove Lemma~\ref{lem:markovWUX}. 

We now prove Lemma~\ref{lem:markovWUX}. Since $W-X-Y$ and $W-Y-X$, we
have  
\begin{equation}
P_{W|XY}(\cdot|x,y)=P_{W|X}(\cdot|x)=P_{W|Y}(\cdot|y)\label{eq:GKW-4}
\end{equation}
for all $(x,y)$ such that $P_{XY}(x,y)>0$. Using the graph-theoretic interpretation of GKW's common information as described in Section~\ref{sec:interp},   we assign a distribution $P_{W}^{(v)} \in\calP(\mathcal{W})$
to each vertex $v\in\mathcal{X}\cup\mathcal{Y}$ in the bipartite
graph. Here $P_{W}^{(v)}$ corresponds to
$P_{W|X}(\cdot|v)$ if $v\in\mathcal{X}$, or $P_{W|Y}(\cdot|v)$
if $v\in\mathcal{Y}$. From \eqref{eq:GKW-4} and   the assumption  that $P_{XY}=\pi_{XY}$, these distributions satisfy
that $P_{W}^{(v)}=P_{W}^{(\hat{v})}$  for any two adjacent vertices
$(v,\hat{v})$.   As a consequence, these distributions are identical
for all vertices in a connected component. As mentioned in Section
\ref{sec:interp}, $f^{*}$ is a function indicating which component~$X$
belongs to. Hence, if we denote the joint distribution of $(W,X,U)$
as $P_{WXU}$, then given $u$, the conditional distribution   $P_{W |X }( \cdot |x )$, which is equal to $P_{W |UX} ( \cdot |u, x )$,  remains
the same for all $x $ such that $u=f^{*}(x)$.  That is, given each $u$, 
\begin{equation}
P_{W|UX}(\cdot|u,x)=\sum_{x'} P_{X|U}(x'|u)P_{W|UX}(\cdot|u,x')=P_{W|U}(\cdot|u)
\end{equation}
for all $x$ such that $P_{UX}(u,x)>0$.  Hence, $W-U-X$ holds, completing
the proof of Lemma~\ref{lem:markovWUX}.  \end{proof}

We now compare $C_{\mathrm{GKW}}(\pi_{XY})$ with the mutual
information $I_{\pi}(X;Y)$ and Wyner's common information $C_{\mathrm{W}}(\pi_{XY})$.
\begin{proposition} For any joint source $(X,Y)\sim\pi_{XY}$, 
\begin{equation}
C_{\mathrm{GKW}}(\pi_{XY})\le I_{\pi}(X;Y)\le C_{\mathrm{W}}(\pi_{XY}). \label{eq:CIC}
\end{equation}
Moreover, the two inequalities become equalities if and only if  $X-U-Y$ holds, where $U$ is the common part of $X$ and $Y$. 
\end{proposition}
The leftmost inequality in~\eqref{eq:CIC} and the corresponding equality conditions were proved  by  \citet{gacs1973common}. 
\begin{proof}
The inequalities in this proposition follow directly by  their definitions. We next consider the conditions for equality. 
Obviously, if    the common part $U$  of $X$ and $Y$ satisfies $X-U-Y$, then both inequalities in \eqref{eq:CIC} are equalities. On the other hand, if
the leftmost inequality in~\eqref{eq:CIC} is an equality, then $I(X;Y) =H(U)$. Combining this with the fact $I(X;Y)=I(XU;Y)=H(U)+I(X;Y|U)$ yields that $I(X;Y|U)=0$, i.e., $X-U-Y$. 

We next assume that the rightmost  inequality in \eqref{eq:CIC} is an equality. Let $P_{WXY}$ be a distribution attaining $ C_{\mathrm{W}}(\pi_{XY})$. Then, $I(XY;W)\ge I(X;W) \ge I(X;Y)$ due to the Markov chain $X-W-Y$. By assumption, $I(XY;W)=I(X;Y)$, which implies that $I(XY;W)= I(X;W)$, i.e., $W-X-Y$ holds. By symmetry, $W-Y-X$ also holds. Combining these two conditions with Proposition \ref{prop:altexpr} yields that $ C_{\mathrm{W}}(\pi_{XY}) \le C_{\mathrm{GKW}}(\pi_{XY})$. Since the inequalities in  \eqref{eq:CIC} imply that $C_{\mathrm{W}}(\pi_{XY})$ cannot be (strictly) smaller than $C_{\mathrm{GKW}}(\pi_{XY})$, we have $ C_{\mathrm{W}}(\pi_{XY})= C_{\mathrm{GKW}}(\pi_{XY})$, which in turn implies that the random variable  $W$ under the distribution $P_{WXY}$ is   the common part of $(X,Y)\sim\pi_{XY}$.  By the choice of   $W$, $X-W-Y$ holds.
\end{proof}

%By definition, $C_{\mathrm{GKW}}(\pi_{XY})\le I_{\pi}(X;Y)$.
%Moreover, as shown in \cite{gacs1973common}, equality is attained
%if and only if $X$ and $Y$ are conditionally independent given the
%common part $W=i^{*}(X)=j^{*}(Y)$. 
\subsection{When is GKW's common information positive? }
Another interesting property of $C_{\mathrm{GKW}}(\pi_{XY})$
is its  intimate connection to the maximal correlation defined in~\eqref{eq:mc}.
%Specifically, the maximal correlation is an indicator of positivity
%of $C_{\mathrm{GKW}}(\pi_{XY})$ as stated in the following simple
%proposition whose proof we omit. 

\begin{proposition} \label{prop:CGWW0} For  $(X,Y)\sim\pi_{XY}$, the following   are equivalent.
\begin{itemize}
\item[(a)] $\rho_{\mathrm{m}}(X;Y)=1$;
\item[(b)] $C_{\mathrm{GKW}}(\pi_{XY})>0$;
\item[(c)]    There exists a pair of nonconstant functions\footnote{A nonconstant function is one whose image contains more than one element.} $(f,g)$ such that $f(X)=g(Y)$ almost surely. 
\end{itemize}
%if and only if . 
\end{proposition}
%\begin{proof}
%This   follows since both $\rho_{\mathrm{m}}(X;Y)=1$ and $C_{\mathrm{GKW}}(\pi_{XY})>0$ are equivalent to  the fact that   there exists a pair of nonconstant functions $(f,g)$ such that $f(X)=g(Y)$ almost surely. 
%%If $\rho_{\mathrm{m}}(X;Y)=1$, then there exists a pair of nonconstant functions $(f,g)$ such that $f(X)=g(Y)$. This implies $C_{\mathrm{GKW}}(\pi_{XY})>0$.
%%If $C_{\mathrm{GKW}}(\pi_{XY})>0$, then there exists a pair of nonconstant functions $(f,g)$ such that $f(X)=g(Y)$. This implies $C_{\mathrm{GKW}}(\pi_{XY})>0$.
%\end{proof}
Thus, for any source $(X,Y)\sim\pi_{XY}$ with
 maximal correlation strictly smaller than $1$, its GKW's
common information is zero. This class of sources includes the DSBS
and Gaussian sources with correlation coefficients in $(-1,1)$.  This is essentially why \citet{gacs1973common} titled their paper {\em ``Common information is far less than mutual information''}. 
Further refinements to  GKW's common information  (when it is equal to zero) that captures other aspects of the sources' correlation will be the main subject of discussion in  Part~\ref{part:three}.

\begin{example} \label{ex:VXY_GW}
Let us now revisit Example~\ref{ex:XYV_Wyner}   in which $X=(\tilde{X},V)$ and  $Y=(\tilde{Y},V)$ for mutually independent random variables $\tilde{X} \in \tilde{\calX}$, $\tilde{Y}\in\tilde{\calY}$, and $V\in \calV$. 
%$(\tilde{X},\tilde{Y},V)\in\tilde{\calX}\times\tilde{\calY}\times\calV$. 
This example
was first presented at the beginning of Section~\ref{ch:wynerCI}
to illustrate that Wyner's common information for this source coincides
with the intuitive quantity $H(V)$. Here, we can also easily observe that GKW's
common information also coincides with $H(V)$. %, i.e., $C_{\mathrm{GKW}}(\pi_{XY})=H(V)$.
This is because, the bipartite graph induced by the distribution
of $(X,Y)$ is  such that given each pair $(v,\hat{v})\in\mathcal{V}^{2}$,
two vertices $(\tilde{x},v)$ and $(\tilde{y},\hat{v})$ with $\tilde{x}\in\tilde{\mathcal{X}},\tilde{y}\in\tilde{\mathcal{Y}}$
are adjacent if and only if $v=\hat{v}$. Hence, each element in $\mathcal{V}$
identifies a unique  connected component of the graph, and {\em vice versa}.
This implies that $C_{\mathrm{GKW}}(\pi_{XY})=H(V)$ and $i^{*}(X)=j^{*}(Y)=V$
(i.e., $V$ is the common part of the joint source  $(X,Y)\sim\pi_{XY}$).
   This example implies that GKW's common
information is zero if the sources are independent (i.e., $V$ is
constant). However, the converse clearly does not hold. Indeed, for any distribution
$\pi_{XY}$ which is fully supported on $\mathcal{X}\times\mathcal{Y}$ (e.g., the DSBS with $|\rho|\ne 1$), its 
GKW's common information is  identically zero.  \end{example}
\section{The Gray--Wyner System}
\label{sec:GW_system}
In Section~\ref{sec:GKW_CI}, we saw one operational interpretation of GKW's common information. In this and the next section, we present two other operational interpretations; these sections may be omitted at a first reading as further discussions on this topic in Part \ref{part:three} depend only on Sections~\ref{sec:GKW_CI} and~\ref{sec:property_GKW}.

We now relate GKW's common information to
the common rate   in the Gray--Wyner system, 
defined in Section~\ref{sec:wyner-GW}. In the Gray--Wyner system,
the common rate is denoted as $R_{0}$ and two private rates are denoted
as $R_{1}$ and $R_{2}$. By Shannon's source coding theorem, if there exists
a Gray--Wyner code such that the source $(X,Y)$ can be reconstructed
almost losslessly by two decoders respectively, then the rate tuple
$(R_{0},R_{1},R_{2})$ of this code must satisfy $R_{0}+R_{1}\ge H(X)$
and $R_{0}+R_{2}\ge H(Y)$. Obviously, these necessary conditions
are not sufficient in general. For example, a tuple $(R_{0},R_{1},R_{2})$
such that $R_{1}=R_{2}=0$ and $R_{0}=\max\{H(X),H(Y)\}$ satisfies
these necessary conditions. However, by Shannon's source coding theorem,
the optimal rate for lossless source coding of the joint source $(X,Y)$
is $H(XY)$, which is strictly larger than $\max\{H(X),H(Y)\}$ unless
$X$ is a function of $Y$ or $Y$ is a function of $X$. Hence, in
general, there is no Gray--Wyner code with such a rate tuple $(R_{0},R_{1},R_{2})$
such that $(X,Y)$ can be reconstructed almost losslessly by the two decoders. In addition, if $(R_{0},R_{1},R_{2})=(0,H(X),H(Y))$,
then  coding   $X$ and $Y$ separately  with rates $R_{1}$
and $R_{2}$ is clearly feasible. Hence, within the transition between these
two extreme cases, there is a maximum common rate $R_{0}$ such that
$R_{0}+R_{1}=H(X)$, $R_{0}+R_{2}=H(Y)$, and the source $(X,Y)$ can
be transmitted almost losslessly to the two decoders using
a Gray--Wyner code with rate tuple $(R_{0},R_{1},R_{2})$. This maximum
common rate $R_{0}$ can be regarded as a form of common information
of $(X,Y)$. Indeed, if we consider the example $X=(\tilde{X},V)$
and $Y=(\tilde{Y},V)$ where $\tilde{X},\tilde{Y},V$ are mutually
independent (cf.\ Example~\ref{ex:VXY_GW}), then the maximum common rate coincides with the intuitive
``common information'' $H(V)$. %In analogy to the Pangloss plane (cf.\ Section~\ref{sec:wyner-GW}), we call the set of rate tuples $(R_0,R_1,R_2)$ such that $R_0+R_1=H(X)$ and  $R_0+R_2=H(Y)$, the pairwise sum rates line}. 
Formally, we define the pairwise sum rate-common information
based on the Gray--Wyner system as follows (compare to Definition~\ref{def:gw_R0}). 

\begin{definition} \label{def:gw_R0_max2} The {\em pairwise sum rate-common information
based on the Gray--Wyner system} $S_{\mathrm{GW}}(\pi_{XY})$
between a pair of random variables $(X,Y)\sim\pi_{XY}$ is the supremum
of all rates $R_{0}$ such that for all $\epsilon>0$, there exists
a sequence of $(n,R_{0},R_{1},R_{2})$ Gray--Wyner codes $\{(f_{0,n},f_{1,n},f_{2,n},\varphi_{1,n},\varphi_{2,n})\}_{n=1}^{\infty}$
such that $R_{0}+R_{1}\le H(X)+\epsilon$ and $R_{0}+R_{2}\le H(Y)+\epsilon$
and the probability of error in~\eqref{eqn:prob_error} vanishes as the length of the code $n$
tends to infinity. \end{definition}

The common information in Definition \ref{def:gw_R0_max2} differs
from G\'acs and K\"orner's formulation of the common information in
\cite{gacs1973common}  in two aspects. Firstly,  the functions
$f_{0,n},f_{1,n},f_{2,n}$ in Definition~\ref{def:gw_R0_max2} are
defined on the set $\mathcal{X}^{n}\times\mathcal{Y}^{n}$; while
the functions $f_{n},\tilf_{n}$ in \cite{gacs1973common} are defined
on $\mathcal{X}^{n}$ and  similarly, $g_{n},g_{n}'$ % in \cite{gacs1973common}
are defined on $\mathcal{Y}^{n}$. Secondly,  only \emph{one}
function $f_{0,n}$ is used to extract common randomness in Definition~\ref{def:gw_R0_max2}; while in \citet{gacs1973common}, \emph{two}
functions $f_{n}$ and $g_{n}$ are employed to extract common randomness
in a distributed way. \citet{ahlswede2006common} (and also \citet{kamath2010new}) showed that the common information in Definition~\ref{def:gw_R0_max2}
  coincides with  the one based on distributed randomness extraction in Definition~\ref{def:GKW_CI}.

\begin{theorem} The pairwise sum rate-common  information based on the Gray--Wyner system
\begin{equation}
S_{\mathrm{GW}}(\pi_{XY})=C_{\mathrm{GKW}}(\pi_{XY}).\label{eqn:wyner_gw}
\end{equation}
\end{theorem}

\begin{proof} It was shown by \citet{GrayWyner} that the closure
of the set of all rate tuples $(R_{0},R_{1},R_{2})$ such that there
exists a sequence of $(n,R_{0},R_{1},R_{2})$ Gray--Wyner codes  satisfying
that the probability of error vanishes as the length of the code $n$
tends to infinity, is the set of $(R_{0},R_{1},R_{2})$ such that
$R_{0}\ge I(XY;W),R_{1}\ge H(X|W),R_{2}\ge H(Y|W)$ for some random
variable $W$. Hence, 
\begin{align}
S_{\mathrm{GW}}(\pi_{XY}) & =\max_{\substack{P_{WXY}:P_{XY}=\pi_{XY},\\
I(XY;W)+H(X|W)=H(X),\\
I(XY;W)+H(Y|W)=H(Y)
}
}I(XY;W)\\
 & =\max_{\substack{P_{WXY}:P_{XY}=\pi_{XY},\\
W-X-Y,\;W-Y-X
}
}I(XY;W) =C_{\mathrm{GKW}}(\pi_{XY}),
\end{align}
where the final equality follows from Proposition \ref{prop:altexpr}.
\end{proof}

Similarly to Wyner's common information, $C_{\mathrm{GKW}}(\pi_{XY})$
also has (at least) two operational interpretations; one as the maximum
rate of almost identical common randomness that can be extracted from
two correlated sources independently, and the other one as the maximum
common rate of the Gray--Wyner system keeping the sum of the common
rate and the private rate at the entropy of the corresponding source.

We conclude this section by mentioning that it would be interesting to establish an analogue of Theorem~\ref{thm:wyner_cts} (due to \citet{XuLiuChen}) for GKW's common information. In particular, how is GKW's common information related to the {\em lossy} Gray--Wyner system? More specifically, is it true that, like Wyner's common information, under mild conditions and for sufficiently small distortion levels, the lossy version of GKW's common information coincides with its almost lossless counterpart?

\section{Channel Coding with an Input Distribution Constraint} \label{sec:input_dist}

\begin{figure}
\centering 
\tikzset{every picture/.style={line width=0.75pt}} %set default line width to 0.75pt        

\begin{tikzpicture}[x=0.75pt,y=0.75pt,yscale=-1,xscale=1,scale=0.9]
%uncomment if require: \path (0,2897); %set diagram left start at 0, and has height of 2897

%Straight Lines [id:da7573973765646591] 
\draw    (332.43,514.67) -- (379.43,514.67) ;
\draw [shift={(381.43,514.67)}, rotate = 180] [color={rgb, 255:red, 0; green, 0; blue, 0 }  ][line width=0.75]    (10.93,-3.29) .. controls (6.95,-1.4) and (3.31,-0.3) .. (0,0) .. controls (3.31,0.3) and (6.95,1.4) .. (10.93,3.29)   ;
%Shape: Rectangle [id:dp42304909411332003] 
\draw   (263,494.67) -- (333,494.67) -- (333,534.67) -- (263,534.67) -- cycle ;
%Shape: Rectangle [id:dp7928719697056896] 
\draw   (121,494.67) -- (191,494.67) -- (191,534.67) -- (121,534.67) -- cycle ;
%Shape: Rectangle [id:dp8247407320170581] 
\draw   (383,494.67) -- (453,494.67) -- (453,534.67) -- (383,534.67) -- cycle ;
%Straight Lines [id:da377520275148693] 
\draw    (454,514.67) -- (507,514.67) ;
\draw [shift={(509,514.67)}, rotate = 180] [color={rgb, 255:red, 0; green, 0; blue, 0 }  ][line width=0.75]    (10.93,-3.29) .. controls (6.95,-1.4) and (3.31,-0.3) .. (0,0) .. controls (3.31,0.3) and (6.95,1.4) .. (10.93,3.29)   ;
%Straight Lines [id:da845042560557947] 
\draw    (67,514.67) -- (120,514.67) ;
\draw [shift={(122,514.67)}, rotate = 180] [color={rgb, 255:red, 0; green, 0; blue, 0 }  ][line width=0.75]    (10.93,-3.29) .. controls (6.95,-1.4) and (3.31,-0.3) .. (0,0) .. controls (3.31,0.3) and (6.95,1.4) .. (10.93,3.29)   ;
%Straight Lines [id:da16180992782259773] 
\draw    (191,514.67) -- (260,514.67) ;
\draw [shift={(262,514.67)}, rotate = 180] [color={rgb, 255:red, 0; green, 0; blue, 0 }  ][line width=0.75]    (10.93,-3.29) .. controls (6.95,-1.4) and (3.31,-0.3) .. (0,0) .. controls (3.31,0.3) and (6.95,1.4) .. (10.93,3.29)   ;

% Text Node
\draw (194,486.67) node [anchor=north west][inner sep=0.75pt]    {$X^{n} \!\sim\! \pi _{X}^{n}$};
% Text Node
\draw (343,489.67) node [anchor=north west][inner sep=0.75pt]    {$Y^{n}$};
% Text Node
\draw (75,487.67) node [anchor=north west][inner sep=0.75pt]    {$M_{n}$};
% Text Node
\draw (277,503) node [anchor=north west][inner sep=0.75pt]    {$\pi _{Y|X}^{n}$};
% Text Node
\draw (127,503) node [anchor=north west][inner sep=0.75pt]    {$P_{X^{n} |M_{n}}$};
% Text Node
\draw (393,503) node [anchor=north west][inner sep=0.75pt]    {$P_{\hat{M}_{n} |Y^{n}}$};
% Text Node
\draw (467,481.67) node [anchor=north west][inner sep=0.75pt]    {$\hat{M}_{n}$};

\end{tikzpicture}
\caption{\label{fig:channelcoding}The channel coding problem with an input distribution constraint}
\end{figure}
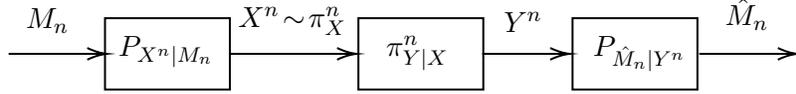

\label{sec:cc} In this section, we provide a third  and final operational interpretation
of GKW's common information in the context
of the classical problem of channel coding, but with a slight twist. Consider the   channel coding problem with
an {\em input distribution constraint} as illustrated in Fig.~\ref{fig:channelcoding}. Let $\pi_{Y|X} \in \calP(\calY|\calX)$ denote the channel,
and $M_{n}\sim\mathrm{Unif}[2^{nR}]$ a uniformly distributed message with rate $R$. 

\begin{definition} A \emph{stochastic $(n,R)$-code} consists of
two stochastic mappings, the stochastic encoder $P_{X^{n}|M_{n}}\in \calP(\mathcal{X}^{n}|[2^{nR}])$
and   the stochastic decoder $P_{\hat{M}_{n}|Y^{n}} \in \calP([2^{nR}]|\mathcal{Y}^{n})$.  \end{definition}

\begin{definition} \label{def:input_dist} %Given a distribution $\pi_{X} 
The \emph{channel
capacity  with input distribution $\pi_{X}\in \calP(\calX)$}, denoted as $C(\pi_{X})$,
is   the supremum of rates $R$ such that there exists a
sequence of stochastic \emph{$(n,R)$-}codes $\{(P_{X^{n}|M_{n}},P_{\hat{M}_{n}|Y^{n}})\}_{n\in\bbN}$
satisfying that the   distribution of $X^{n}$ is exactly equal
to $\pi_{X}^{n}$ for each $n\in\bbN$ and the average probability of error  $\Pr(\hat{M}_{n}\ne M_{n})$
%\begin{equation}
%\Pr(\hat{M}_{n}\ne M_{n})= \frac{1}{2^{\lfloor  nR\rfloor }}\sum_{ m\in [2^{nR}]} \sum_{x^n,y^n} P_{X^n | M_n}(x^n | m )\pi_{Y|X}^n(y^n|x^n)P_{\hatM_n |Y^n}( m |y^n)
%\end{equation}
vanishes as the length of the code $n$ tends to infinity. \end{definition}

This notion is markedly different from the problem of {\em channel coding with input cost}~\cite[Section~3.3]{elgamal} in which the input codewords $x^n(m) , m \in [2^{nR}]$ are required to satisfy a constraint of the form $\frac{1}{n}\sum_{i=1}^nb(x_i(m))\le B$ for some per-letter cost function $b:\calX\to [0,\infty)$ and cost constraint $B>0$. In Definition~\ref{def:input_dist}, the {\em distribution} of the channel input $X^n$, induced by $M_n$ and $P_{X^n|M_n}$, is required to be exactly equal to $\pi_X^n$. To satisfy this constraint, a {\em stochastic} encoder is required.  
The present authors showed that the channel capacity $C(\pi_{X})$ with input distribution $\pi_{X}$ is
equal to GKW's common information of $\pi_{XY}$~\cite{YuTan2019}. 

\begin{theorem} \label{thm:capacity} For any channel $\pi_{Y|X} \in \calP(\calY|\calX)$,
\begin{equation}
C(\pi_{X})=C_{\mathrm{GKW}}(\pi_{XY}).
\end{equation}
% where $C_{\mathrm{GKW}}(\pi_{XY})$ is the GKW's common information,
%defined in \eqref{eq:GKWGKW_CI}.
\end{theorem}

From this theorem, we deduce that 
\begin{equation}
C(\pi_{X})\leq I_{\pi}(X;Y)\leq C^{*}:=\max_{P_{X} \in \calP(\calX)}I(P_X, \pi_{Y|X}),
\end{equation}
where $C^{*}=C^*(\pi_{Y|X})$ denotes the Shannon capacity  of the channel $\pi_{Y|X}$ (i.e., the channel capacity
without the input distribution constraint). This channel coding problem can in fact
be reinterpreted as a {\em randomness extraction} problem. Given a bivariate
source $(X,Y)\sim\pi_{XY}$, we use two stochastic maps $P_{M_n|X^{n}}\in \calP([2^{nR} ] |\calX^n)$
and $P_{\hat{M}_{n}|Y^{n}}\in\calP( [2^{nR} ]| \calY^n  )$ to generate
a uniform random variable $M_{n}$ and an arbitrary random variable
$\hat{M}_{n}$  such that $\Pr(\hat{M}_{n}\ne M_n)\to0$ as $n\to\infty$.
We aim to maximize the rate $R$ of $M_n$. This variation  of the randomness
extraction problem differs from limiting case in which $\eps \downarrow 0$  (vanishing error probability) in Definition~\ref{def:GKW_CI}
 in Section~\ref{sec:GKW_CI} in two aspects. Firstly, the
maps in the channel coding with input distribution constraint problem are \emph{stochastic}, while
 the maps in the distributed randomness extraction problem in Section~\ref{sec:GKW_CI} are \emph{deterministic}.
Secondly, the output $M_n$ from $P_{M_n|X^n}$  here is a \emph{uniform} random
variable, while the output $f(X^{n})$ is not   necessarily uniform.
In spite of these two differences, we observe that the % maximum achievable rate of extracted
%randomness do not differ.  That is, the 
maximum achievable rates of the extracted randomnesses
for these two problems coincide, and are both equal
to $C_{\mathrm{GKW}}(\pi_{XY})$. 

\begin{proof}[Proof of Theorem \ref{thm:capacity}]   The inequality
$C(\pi_{X})\le C_{\mathrm{GKW}}(\pi_{XY})$, which represents the converse,
can be proved by combining Lemma~\ref{lem:csiszar} and Fano's inequality,
just as in the proof of Theorem~\ref{thm:GKW_CI}. We omit the details
here and refer the interested reader to \citet[Section~VI]{YuTan2019}.  

We next prove the more interesting part $C(\pi_{X})\ge C_{\mathrm{GKW}}(\pi_{XY})$, which represents the achievability.
Observe that the distributions of $X^{n}$ and $M_{n}$ are respectively
$\pi_{X}^{n}$ and $\mathrm{Unif}[2^{nR}]$, both of which are given.
Hence, designing a stochastic map $P_{X^{n}|M_{n}} \in \calP(\calX^n|[2^{nR}])$
(or $P_{M_{n}|X^{n}} \in \calP(  [2^{nR}]|\mathcal{X}^{n})$) is equivalent to
designing a {\em coupling} (cf.\ Section~\ref{sec:coupling})  of $\pi_{X}^{n}$ and $\mathrm{Unif}[2^{nR}]$.
On the other hand, let $U$ be the common part of $(X,Y)\sim\pi_{XY}$.
By definition,    $C_{\mathrm{GKW} }( \pi_{XY} ) = H ( U )$. Moreover, $U^n$, which corresponds to the common part of $X^n$ and $Y^n$,  can be generated
from $X^n$ and $Y^n$ individually. %$U^{n}$ can be generated from $X^{n}$ or $Y^{n}$
%individually, and $C_{\mathrm{GKW}}(\pi_{XY})=H(U)$.
To prove that $C(\pi_{X})\ge C_{\mathrm{GKW}}(\pi_{XY})$,
it suffices to construct a  sequence of couplings  $\{P_{U^nM_n}\}_{n\in\bbN}$ of the distributions of $U^{n}$ and $M_{n}$
such that 
\begin{equation}
\lim_{n\to\infty}\;\min_{\varphi:\mathcal{U}^{n}\to[2^{nR}]}\;\Pr\big( M_{n}\neq\varphi(U^{n})\big) =0.\label{eq:GKWerrMW}
\end{equation}
 In other words, we only utilize the common part $U^{n}$ of $X^{n}$ and $Y^{n}$ to transmit the message.
To this end, we leverage the maximal guessing coupling equality in Lemma \ref{thm:max_guess_coup}. 
% To this end, we need to introduce the maximal guessing coupling
%and a related lemma. 
%
%\begin{definition} \label{def:maxguessing} The \emph{maximal guessing
%probability }over couplings of $(P_{X},P_{Y})$ is defined as 
%\begin{equation}
%\mathcal{G}(P_{X},P_{Y}):=\max_{P_{XY}\in C(P_{X},P_{Y})}\max_{\varphi:{\cal X}\to{\cal Y}}\Pr(Y=\varphi(X)).\label{eq:GKW-45}
%\end{equation}
%Any $P_{XY}\in C(P_{X},P_{Y})$ achieving $\mathcal{G}(P_{X},P_{Y})$
%is called a \emph{maximal guessing coupling} of $(P_{X},P_{Y})$.
%\end{definition}
%\begin{lemma}[Maximal Guessing Coupling Equality] \label{lem:equivalence}
%It holds that 
%\begin{equation}
%\mathcal{G}(P_{X},P_{Y})=1-\min_{\varphi}|P_{Y}-P_{\varphi(X)}|,\label{eq:GKW-8}
%\end{equation}
%where $|\cdot|$ denotes the TV distance. \end{lemma}
%
%This lemma can be easily proven by swapping the two maximizations
%in \eqref{eq:GKW-45}; see details in \cite{YuTan2019}. 

The minimization on the right-hand side of \eqref{eq:max_guess_coup} is termed the \emph{distribution approximation} or \emph{random number generation} problem  \cite[Chapter~2]{Han10}, in which a random variable $X \sim P_{X}$
is used to simulate another random variable $Y \sim P_{Y}$ using a
function $\varphi:{\cal X}\to{\cal Y}$ such that  the TV distance
between the distribution $P_{\varphi(X)}$ of the generated random variable
$\varphi(X)$ and the target distribution $P_{Y}$ is  minimized. When $Y$ is uniform, this problem reduces to the {\em intrinsic randomness} problem  in which a well-known result 
%A well-known result concerning the distribution approximation problem 
due to \citet{vembu1995generating} %(see also \citet{Ste96})
is the following. For an i.i.d.\ source sequence  $X^n\sim P_{X}^{n}$ and the uniform distribution
$\mathrm{Unif}[2^{nR}]$, if $R<H(X)$, 
\begin{equation}
\lim_{n\to\infty}\;\min_{\varphi_n : \calX^n\to [2^{nR}]}\;\big|\mathrm{Unif}[2^{nR}]-P_{\varphi_n(X^{n})}\big|=0.
\end{equation}
%(with $X^{n}\sim P_{X}^{n}$) as $n\to\infty$ if $H(X)>R$.  
Hence,
in our setting,  with $U^n\sim \pi_{U}^n$ and $R<H_\pi(U)$, 
\begin{equation}
\lim_{n\to\infty}\;\min_{\varphi_n: \calU^n\to [2^{nR}]}\; \big|\mathrm{Unif}[2^{nR}]-\pi_{\varphi_n(U^{n})}\big|=0.
\end{equation}
%with $U^{n}\sim\pi_{U}^{n}$ as $n\to\infty$ if $H_{\pi}(U)>R$.
Combining this with Lemma~\ref{thm:max_guess_coup} yields that if $R<H_{\pi}(U)$, then 
\begin{equation}
\lim_{n\to\infty}\max_{ P_{U^{n}M_{n}} \in \calC( \pi_{U}^{n},\mathrm{Unif}[2^{nR}])}\; \max_{\varphi_n:\calU^n\to [2^{nR}]}\; \Pr\big(M_n=\varphi_n( U^n )\big)=1. \label{eqn:cc_error}
\end{equation}
%$\mathcal{G}(\pi_{U}^{n},\mathrm{Unif}[2^{nR}])\to1$,
In other words, if $R<H_\pi(U)$, there exists a sequence of couplings $\{ P_{U^{n}M_{n}} \}_{n\in\bbN}$ of $\pi_{U}^{n}$
and $\mathrm{Unif}[2^{nR}]$ such that \eqref{eq:GKWerrMW} holds. Let
$\varphi_{n}^{*} : \calU^n\to[2^{nR}]$ be any function that attains $\max_{\varphi_{n}}\Pr ( M_{n}=\varphi_{n}(U^{n}))$
for $(U^{n},M_{n})\sim P_{U^{n}M_{n}}$. Now define the encoder as  
\begin{equation}
P_{X^n|M_n}(x^n|m):=\sum_{u^{n} \in \calU^n}P_{U^{n}|M_{n}}(u^{n}|m)P_{X|U}^{n}(x^n|u^{n}), \label{eqn:cc_enc}
\end{equation}
and the decoder as
%as the encoder, and 
\begin{equation}
P_{\hatM_n|Y^n} (m|y^n) := \sum_{u^{n}\in \calU^n}P_{U|Y}^{n}(u^{n}|y^n)\;\bone\big\{ m=\varphi_{n}^{*}(u^{n})\big\}.
\end{equation}
See the coding scheme in Fig.~\ref{fig:cc_ach}. Then, as a result of~\eqref{eqn:cc_error}--\eqref{eqn:cc_enc}, the constraints $X^{n}\sim\pi_{X}^{n}$ and
$\Pr(\hat{M}_{n}\ne M_{n})\to0$ as $n\to\infty$ are  respectively satisfied, which implies that 
$C(\pi_{X})\ge H_{\pi}(U)=C_{\mathrm{GKW}}(\pi_{XY})$.
\end{proof}

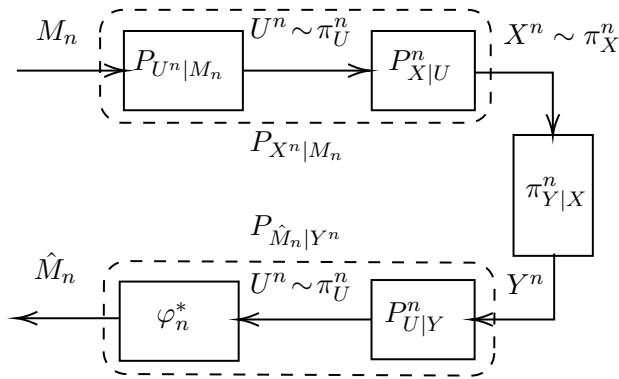
\begin{figure}[!ht]
\centering 
\tikzset{every picture/.style={line width=0.75pt}} %set default line width to 0.75pt        

\begin{tikzpicture}[x=0.75pt,y=0.75pt,yscale=-1,xscale=1]
%uncomment if require: \path (0,2897); %set diagram left start at 0, and has height of 2897

%Shape: Rectangle [id:dp9647989682998661] 
\draw   (357.5,722.17) -- (357.5,783.17) -- (317.5,783.17) -- (317.5,722.17) -- cycle ;
%Shape: Rectangle [id:dp8310975459148642] 
\draw   (121,669.67) -- (181,669.67) -- (181,709.67) -- (121,709.67) -- cycle ;
%Straight Lines [id:da3436828911752461] 
\draw    (67,689.67) -- (120,689.67) ;
\draw [shift={(122,689.67)}, rotate = 180] [color={rgb, 255:red, 0; green, 0; blue, 0 }  ][line width=0.75]    (10.93,-3.29) .. controls (6.95,-1.4) and (3.31,-0.3) .. (0,0) .. controls (3.31,0.3) and (6.95,1.4) .. (10.93,3.29)   ;
%Straight Lines [id:da3864440636545805] 
\draw    (298,690.67) -- (337.33,690.67) -- (337.33,720.67) ;
\draw [shift={(337.33,722.67)}, rotate = 270] [color={rgb, 255:red, 0; green, 0; blue, 0 }  ][line width=0.75]    (10.93,-3.29) .. controls (6.95,-1.4) and (3.31,-0.3) .. (0,0) .. controls (3.31,0.3) and (6.95,1.4) .. (10.93,3.29)   ;
%Shape: Rectangle [id:dp7094755025580186] 
\draw   (246,669.67) -- (298,669.67) -- (298,709.67) -- (246,709.67) -- cycle ;
%Straight Lines [id:da08265360051109094] 
\draw    (181,689.67) -- (243,689.67) ;
\draw [shift={(245,689.67)}, rotate = 180] [color={rgb, 255:red, 0; green, 0; blue, 0 }  ][line width=0.75]    (10.93,-3.29) .. controls (6.95,-1.4) and (3.31,-0.3) .. (0,0) .. controls (3.31,0.3) and (6.95,1.4) .. (10.93,3.29)   ;
%Rounded Rect [id:dp6353645327880333] 
\draw  [dash pattern={on 4.5pt off 4.5pt}] (109,670.4) .. controls (109,664.1) and (114.1,659) .. (120.4,659) -- (294.6,659) .. controls (300.9,659) and (306,664.1) .. (306,670.4) -- (306,704.6) .. controls (306,710.9) and (300.9,716) .. (294.6,716) -- (120.4,716) .. controls (114.1,716) and (109,710.9) .. (109,704.6) -- cycle ;
%Shape: Rectangle [id:dp13233146833786402] 
\draw   (118.88,796.15) -- (178.88,796.15) -- (178.88,836.15) -- (118.88,836.15) -- cycle ;
%Straight Lines [id:da8193225763040202] 
\draw    (338,782) -- (338,815.29) -- (298.14,815.29) ;
\draw [shift={(296.14,815.29)}, rotate = 360] [color={rgb, 255:red, 0; green, 0; blue, 0 }  ][line width=0.75]    (10.93,-3.29) .. controls (6.95,-1.4) and (3.31,-0.3) .. (0,0) .. controls (3.31,0.3) and (6.95,1.4) .. (10.93,3.29)   ;
%Shape: Rectangle [id:dp5669867227750232] 
\draw   (245.88,795.37) -- (297.88,795.37) -- (297.88,835.37) -- (245.88,835.37) -- cycle ;
%Straight Lines [id:da756381810284495] 
\draw    (246.14,815.29) -- (180.14,815.29) ;
\draw [shift={(178.14,815.29)}, rotate = 360] [color={rgb, 255:red, 0; green, 0; blue, 0 }  ][line width=0.75]    (10.93,-3.29) .. controls (6.95,-1.4) and (3.31,-0.3) .. (0,0) .. controls (3.31,0.3) and (6.95,1.4) .. (10.93,3.29)   ;
%Rounded Rect [id:dp1325682929816674] 
\draw  [dash pattern={on 4.5pt off 4.5pt}] (307.78,832.26) .. controls (307.75,838.55) and (302.62,843.63) .. (296.33,843.6) -- (122.13,842.81) .. controls (115.83,842.78) and (110.75,837.66) .. (110.78,831.36) -- (110.94,797.16) .. controls (110.96,790.87) and (116.09,785.79) .. (122.39,785.81) -- (296.59,786.6) .. controls (302.88,786.63) and (307.96,791.76) .. (307.93,798.06) -- cycle ;
%Straight Lines [id:da9167677862561971] 
\draw    (118.4,814.66) -- (68.4,814.66) ;
\draw [shift={(66.4,814.66)}, rotate = 360] [color={rgb, 255:red, 0; green, 0; blue, 0 }  ][line width=0.75]    (10.93,-3.29) .. controls (6.95,-1.4) and (3.31,-0.3) .. (0,0) .. controls (3.31,0.3) and (6.95,1.4) .. (10.93,3.29)   ;

% Text Node
\draw (311,664.67) node [anchor=north west][inner sep=0.75pt]    {$X^{n} \sim \pi_{X}^{n}$};
% Text Node
\draw (312.43,790.97) node [anchor=north west][inner sep=0.75pt]    {$Y^{n}$};
% Text Node
\draw (75,662.67) node [anchor=north west][inner sep=0.75pt]    {$M_{n}$};
% Text Node
\draw (322,742) node [anchor=north west][inner sep=0.75pt]    {$\pi_{Y|X}^{n}$};
% Text Node
\draw (124,678) node [anchor=north west][inner sep=0.75pt]    {$P_{U^{n} |M_{n}}$};
% Text Node
\draw (74,779.67) node [anchor=north west][inner sep=0.75pt]    {$\hat{M}_{n}$};
% Text Node
\draw (183,661.67) node [anchor=north west][inner sep=0.75pt]    {$U^{n} \!\sim\! \pi_{U}^{n}$};
% Text Node
\draw (253,678) node [anchor=north west][inner sep=0.75pt]    {$P_{X |U}^{n}$};
% Text Node
\draw (183,790.21) node [anchor=north west][inner sep=0.75pt]  [rotate=-0.1]  {$U^{n} \!\sim\! \pi_{U}^{n}$};
% Text Node
\draw (250.9,805) node [anchor=north west][inner sep=0.75pt]  [rotate=-0.1]  {$P_{U|Y}^{n}$};
% Text Node
\draw (136,805) node [anchor=north west][inner sep=0.75pt]    {$\varphi_{n}^{*}$};

\draw (183,719) node [anchor=north west][inner sep=0.75pt]    {$P_{X^n | M_n}$};
\draw (183,761) node [anchor=north west][inner sep=0.75pt]    {$P_{\hatM_n | Y^n}$};

\end{tikzpicture}
\caption{\label{fig:cc_ach}A channel coding scheme}
\end{figure}

\section{Generalizations and Applications}\label{sec:gen_GKW}
We conclude this section by introducing  several generalizations of GKW's common information. The expression of GKW's common information in~\eqref{eq:GKW-3} implies that it can also written as the maximum of the mutual information $I(X;U)$ over all distributions $P_{UXY}$ such that $P_{XY}=\pi_{XY}$, $U-X-Y$, and $H(U|Y)=0$. If we relax the constraint $H(U|Y)=0$ to $H(U|Y)\le \delta$ for a given $\delta>0$, then we arrive at  the \emph{approximate  GKW's common information} of $\pi_{XY}$,  namely, %, and keep other constraints and the objective function unchanged, then we arrive at the following quantity. For any $0\le \delta \le H_{\pi}(X|Y)$,
\begin{equation}
C_{\mathrm{GKW}}^{(\delta)}(\pi_{XY}) := \max_{P_{UXY}:P_{XY}=\pi_{XY}, U-X-Y,H(U|Y) \le  \delta} I(X;U).
\end{equation}
%This quantity $C_{\mathrm{GKW}}^{\delta}(\pi_{XY})$ is known as   the \emph{approximate  GKW's common information} of $\pi_{XY}$,  proposed by \citet{salamatian2020approximate}. 
%Such a quantity was used by 
This quantity was proposed and used by \citet{salamatian2020approximate} to characterize an  achievability result for    zero-error coding   in the  \emph{distributed lossless compression with helper} \cite{sw73} problem. 

\citet{yu2016generalized} introduced another generalization of GKW's common information. Let $U$ be  the common part of $X$ and $Y$. Then, for every $u\in\calU$,  on the event $\{U=u\}$, the maximal correlation of $X$ and $Y$, i.e., the maximal correlation of $(X',Y')\sim \pi_{XY|U=u}$ which is denoted as $\rho_{\rmm}(X;Y|U=u)$, is strictly less than~$1$. This is because, otherwise, for each $u\in\calU$, one can extract a common part $V_u$ of $(X',Y')\sim \pi_{XY|U=u}$   such that $H(V_U|U)>0$. Note that  $(U,V_U)$   also forms a common part of $X$ and $Y$, and moreover, $H(U,V_U)>H(U)$. This contradicts   the assumption that $U$ is the common part of $X$ and $Y$. 
%Hence, the observation above is true, i.e., for each $u$, the maximal correlation of $(X',Y')\sim \pi_{XY|U=u}$ is strictly smaller than $1$. 
This inspires \citet{yu2016generalized}  to provide another characterization of GKW's common information as follows. 
For each $\beta \in [0,1]$, define the \emph{information-correlation function} as 
\begin{equation}
C_{\mathrm{IC}}^{(\beta)}(\pi_{XY}) := \max_{P_{UXY}:P_{XY}=\pi_{XY}, \rho_{\rmm}(X;Y|U) \le  \beta} I(XY;U) ,  %\quad \mbox{for all } \beta \in [0,1], 
\end{equation}
where $\rho_{\rmm}(X;Y|U):=\sup_{u\in\calU}\rho_{\rmm}(X;Y|U=u)$ denotes the {\em conditional maximal correlation} of $X$ and $Y$ given $U$. By the support lemma \cite{elgamal},   it suffices to consider a  variable $U$ with alphabet size $|\calU| \le  |\calX ||\calY| + 1$.  
Then, GKW's common information can be expressed as 
\begin{equation}
C_{\mathrm{GKW}}(\pi_{XY})= \lim_{\beta\uparrow 1}C_{\mathrm{IC}}^{(\beta)}(\pi_{XY}).
\end{equation}
If we consider  the other end point $\beta=0$, then the constraint in the definition of the  information-correlation function reduces to $\rho_{\rmm}(X;Y|U) = 0$, which is   equivalent to $X-U-Y$. Hence, we recover Wyner's common information from the information-correlation function, i.e., 
\begin{equation}
C_{\mathrm{IC}}^{(0)}(\pi_{XY})=C_{\mathrm{W}}(\pi_{XY}).
\end{equation}
Thus the information-correlation  function $C_{\mathrm{IC}}^{(\beta)}$ interpolates between GKW's and Wyner's common information as $\beta$ decreases from~$1$ to~$0$. 
The generalization of Wyner's common information by \citet{GastparSuha}, also given in \eqref{eqn:relWyn}, is defined in the same spirit as $C_{\mathrm{IC}}^{(\beta)}$, in which  the conditional maximal correlation in the constraint is replaced by the conditional mutual information. 

The distributed common randomness extraction problem formulated by \citet{gacs1973common} is a type of key agreement or key generation problem, in which interactive communication is not allowed and also secrecy is not considered. 
These two assumptions are usually not applicable to practical secret key agreement systems.  To ameliorate these limitations, \citet{csiszar2000common} generalize GKW's common information  to the setting  in which communication is allowed, 
and moreover, a helper  assists  the extractors to extract a higher  rate of common randomness   from the sources. 
Besides, they also consider another setting in which the communication can be observed by a wiretapper, and the secrecy is measured by certain information-theoretic quantities. The latter setting is known as the \emph{secret key agreement problem}. The regions of achievable rate tuples for these two settings are characterized in terms of certain mutual information quantities, which recover  GKW's common information as extreme cases. 
Further generalizations of GKW's common information  to more complicated networks have also been investigated in the literature; see for example the comprehensive surveys by  \citet{sudan2020}, \citet{liang2009information}, and \citet{bloch2011physical}. Other interesting quantities defined based on Gray--Wyner system, such as the K\"orner graph entropy \cite{korner1973coding}, the privacy funnel \cite{Makhdoumi2014}, and the  excess functional information \cite{LiElgamal2018},
can be found in \citet{li2017extended}. These quantities can also considered as generalizations of GKW's and Wyner's common information.

\part{Extensions of Wyner's Common Information} \label{part:two}
\chapter{R\'enyi and Total Variation  Common Information}
\label{ch:renyi}
%This is based on our paper ``Wyner's common information under R\'enyi divergence measures''. More than just a generalization of R\'enyi's CI, it serves as a convenient bridge between Wyner's CI and Exact CI. Will discuss the TV (variational distance) approximate version and its strong converse together with unnormalized and normalized  versions of the the R\'enyi divergence constraints.

In this section, we extend the notion of Wyner's common information by modifying the discrepancy measure used to quantify the distance between the synthesized distribution $P_{X^n Y^n}$ and  the $n$-fold product of the target distribution $\pi_{XY}^n$. We analyze how the minimum amount of shared randomness in the distributed source simulation problem (the rate of~$M_n$ in Fig.~\ref{fig:source_sim}) changes when we  employ R\'enyi divergences of orders $1+s$  where $s \in [-1,1]\cup\{\infty\}$, their normalized versions, and  the total variation (TV) distance in place of the normalized relative entropy in~\eqref{eqn:norm_re2}.

The reader might naturally wonder what the value is in going beyond the traditional normalized relative entropy. For one, in security problems, one is usually not content with having the {\em normalized} amount of leaked information vanish as the length of the code grows; this is known as {\em weak secrecy}. Systems that satisfy weak secrecy  nevertheless allow an unbounded number of bits to be leaked to a potentially malicious party. This is clearly undesirable. In practical systems, we seek to design codes such that the  {\em unnormalized} amount of leaked information vanishes. This is known as {\em strong secrecy}~\cite{Maurer00, BlochLaneman} in which the average number of bits that is leaked vanishes. Analogously, requiring the normalized relative entropy between  $P_{X^n Y^n}$ and $\pi_{XY}^n$ to vanish is usually not a criterion that is sufficiently stringent. Our objective is to design and analyze codes that drive the unnormalized relative entropy to zero. It turns out that there is typically no  additional cost to satisfy this more stringent criterion compared to the normalized case. %, so this is definitely a worthwhile academic endeavor. 

More importantly, prior to our work that this section is based on \cite{YuTan2018,yu2020corrections}, the unnormalized relative entropy was the {\em strongest} or {\em most stringent} criterion for measuring the discrepancy between $P_{X^n Y^n}$  and $\pi_{XY}^n$. Are there families of divergences that further strengthen the unnormalized relative entropy? It turns out that the answer is yes. Since the R\'enyi divergence $D_{1+s}$ is monotonically non-decreasing in its order $1+s$, if we increase $s$ from zero  to infinity and mandate that $D_{1+s}(P_{X^n Y^n}\|\pi_{XY}^n)$ vanishes, we obtain a family of Wyner-inspired common information measures that strengthens the original Wyner's common information.

En route to proving coding theorems for the common information when the R\'enyi divergence (in both its normalized and unnormalized forms) is employed, we find it convenient to segue to working with the TV distance. Via   Pinsker-like inequalities relating the R\'enyi divergence to the TV distance, results concerning the TV distance can be rather conveniently translated to those for the R\'enyi divergence and {\em vice versa}. 

It would be remiss for us to not mention that the study of common information under R\'enyi divergence measures leads us to one of the main results of this part of the monograph. In particular, we show in Section~\ref{ch:exact} that the R\'enyi common information of order $\infty$ is exactly the same as the so-called exact common information. This allows us to interpret the latter quantity in a whole new different light, thus providing a pathway to computing it and showing that the exact common information can be strictly larger than Wyner's common information for some joint sources. 

Finally, it is worth noting that it is quite natural to use various divergences to measure the discrepancy between two distributions. For instance, \citet{Hayashi06,Hayashi11} and \citet{YuTan2019_wiretap} respectively used
the KL divergence and the R\'enyi divergence to study the
channel resolvability problem. The latter also applied their
results to study the capacity of the wiretap channel under the condition that the security requirement is measured by these generalized measures. Special instances of R\'enyi entropies and divergences including the relative entropy, the collision entropy, and the min-entropy (corresponding to the R\'enyi divergence of order $\infty$) were used to study various  problems in  probability theory, cryptography, and quantum information recently. See  \citet{Bobkov19}, \citet{dodis2013overcoming}, \citet{iwamoto}, \citet{HayashiTan17}, \citet{TanHayashi18}, and \citet{beigi2014quantum}, and references therein for a non-exhaustive list.

This section starts by  formally defining some useful quantities and stating some of their   properties. These quantities are used to express bounds or exact expressions for the R\'enyi and $\eps$-TV common informations (to be defined in Definition~\ref{def:TVCI})  in terms of single-letter quantities, rendering their computations for a variety  of joint sources feasible. We evaluate the R\'enyi common information for the DSBS. We show that for R\'enyi orders greater than~$1$ (resp.\ in $(0,1]$), the  R\'enyi common information generally exceeds  (resp.\ coincides with) Wyner's  common information. This section, being technical in nature, also provides glimpses of how various proofs are intertwined and hinge on some basic results introduced in Sections \ref{ch:intro} and \ref{ch:wynerCI}. 

\newcommand{\oGamma}{\overline{\Psi}}
\newcommand{\uGamma}{\underline{\Psi}}

\section{Preliminary Definitions}
We commence by stating a couple of definitions  that are used extensively to characterize the common information quantities of interest in this and the following sections. 
\begin{definition} \label{def:max_cross_ent} %[Maximal $s$-Mixed Shannon-Cross Entropy]
For $s>0$, the {\em maximal $s$-mixed cross entropy} with respect to $\pi_{XY}$ over all couplings of $P_X$ and $P_Y$ is 
\begin{align}
&\rvH_{s}(P_X,P_Y\|\pi_{XY}) \nn\\*
&:= \max_{Q_{XY}\in\calC(P_X,P_Y)}\sum_{x,y}Q_{XY}(x,y)\log\frac{1}{\pi_{XY}(x,y)}+\frac{1}{s}H(Q_{XY}). \label{eqn:maximal_mixed}
\end{align}
When $s=\infty$, the above definition % maximal $s$-mixed cross entropy 
reduces to the {\em maximal cross entropy} with respect to $\pi_{XY}$ over all couplings of $P_X$ and $P_Y$, i.e.,
\begin{equation}
  \rvH_\infty(P_X,P_Y\|\pi_{XY})\!:=\!\max_{Q_{XY}\in\calC(P_X,P_Y)}\sum_{x,y}Q_{XY}(x,y)\log\frac{1}{\pi_{XY}(x,y)}.\!\! \label{eqn:maximal_cross_ent0}
\end{equation}
\end{definition}
%Observe that in the special case of $s=-1$, the above quantity reduces to 
%\begin{align}
%H_1(P_X,P_Y\|\pi_{XY}) = \max_{Q_{XY}\in\calC(P_X,P_Y)} D(Q_{XY}\| \pi_{XY})
%\end{align}
Some intuition for these quantities can be gleaned by considering the case $s=\infty$ in~\eqref{eqn:maximal_cross_ent0}. Consider a sequence of pairs of marginal types $\big\{(T_X^{(n)}, T_Y^{(n)})\big\}_{n\in\bbN}\subset \calP(\calX)\times \calP(\calY)$ such that $T_X^{(n)} \in\calP_n(\calX)$ converges to $P_X$ and  $T_Y^{(n)}\in\calP_n(\calY)$ converges to $P_Y$ as $n\to\infty$ (in TV distance, for example). The minimum $\pi_{XY}^n$-probability of $(x^n,y^n)$ such that the {\em marginal} types of $x^n$ and $y^n$ are  $T_X^{(n)}$ and $T_Y^{(n)}$  respectively is given by 
\begin{align}
& \min_{T_{x^n} = T_X^{(n)}  ,T_{y^n} =  T_Y^{(n)}  } \pi_{XY}^n\left(  x^n,y^n \right)\nn\\
%& \doteq \exp\left( -n \sum_{x,y} T_{x^n, y^n}(x,y)\log\frac{1}{\pi_{XY}(x,y)} \right)
%\nn\\
 &\quad\doteq \exp \bigg(-n \max_{Q_{XY}\in\calC(P_X,P_Y)}\sum_{x,y}Q_{XY}(x,y)\log\frac{1}{\pi_{XY}(x,y)}\bigg)\\
  &\quad= \exp \big(-n \rvH_{\infty}(P_X,P_Y\|\pi_{XY})\big). \label{eqn:intuition_maximal}
\end{align}
%which explains the first part of \eqref{eqn:maximal_mixed}. 
The intuitive reason why we consider the {\em minimum} $\pi_{XY}^n$-probability leading to {\em maximal}  cross entropy is that, as alluded to in the introduction of this section, we are considering {\em strengthenings} of Wyner's common information using the discrepancy measures $D_{1+s}$ for $s\in (0,\infty]$. Consequently, the required resolution  or common information rate of~$M_n$ would, in general, need  to be larger than that for Wyner's common information. In fact, it is determined by the {\em minimum} of the $\pi_{XY}^n$-probability of certain type classes. This will be made   clear when we discuss the notion of exact common information in Section~\ref{ch:exact}.

We now state a few properties of the maximal cross-entropy.
\begin{lemma} \label{lem:prop_max_cross_ent}
Let $\pi_{XY}$ be a joint distribution on a finite alphabet $\calX\times\calY$ and with marginals $\pi_X$ and $\pi_Y$. 
\begin{enumerate}
\item We have
\begin{equation}
\rvH_\infty(\pi_X,\pi_Y\|\pi_{XY})\ge H(\pi_{XY}) \label{eqn:prop_max_cross_ent1}
\end{equation}
where equality holds if and only if $\pi_{XY}=\pi_X\pi_Y$. 
\item Assume that $\supp(\pi_{XY})=\calX\times\calY$. Then for any pair of distributions $P_X$ and $P_Y$ such that $\supp(P_X)=\calX$ and $\supp(P_Y)=\calY$,  we have
\begin{equation}
\rvH_\infty(P_X,P_Y\|\pi_{XY})\ge\sum_{x,y}P_X(x)P_Y(y)\log \frac{1}{\pi_{XY}(x,y)} \label{eqn:prop_max_cross_ent2}
\end{equation}
where equality holds if and only if $\pi_{XY}=\pi_X\pi_Y$. 
\end{enumerate}
\end{lemma}

We now provide a couple of examples to show that $\rvH_\infty$ can be calculated in closed form for some archetypal joint sources.
\begin{example} \label{ex:dsbs_renyi}
Consider the DSBS in Section~\ref{sec:dsbs}. Fix $P_X =\mathrm{Bern}(a)$ and $P_Y =\mathrm{Bern}(b)$ where $a,b\in [0,1]$. Then 
\begin{align}
\rvH_\infty(P_X,P_Y\|\pi_{XY}) &=\log\frac{1}{\alpha}+\big(\min\{a,\barb\}+\min\{\bara,b\} \big)\log\frac{\alpha}{\beta}\\
&=\log\frac{1}{\alpha}+\min\{a+b,\bara+\barb\} \log\frac{\alpha}{\beta},
\end{align}
where $\alpha=(1-p)/2$ and $\beta =p/2$.
Furthermore, when $P_X$ and $P_Y$ are particularized to $\pi_X=\mathrm{Bern}(1/2)$ and $\pi_Y=\mathrm{Bern}(1/2)$, 
 \begin{equation}
\rvH_\infty(\pi_X,\pi_Y\|\pi_{XY}) =\log\frac{1}{\beta}.
\end{equation}
In contrast, the joint entropy is 
\begin{equation}
H(\pi_{XY}) =2\alpha\log\frac{1}{\alpha} + 2\beta\log\frac{1}{\beta} \le \rvH_\infty(\pi_X,\pi_Y\|\pi_{XY})
\end{equation}
with equality if and only if $p=1/2$, i.e., $\alpha=\beta=1/4$.
\end{example}
\begin{example}
Let $(X,Y)\sim \pi_{XY}$ be a jointly Gaussian source with  $\bbE[X]=\bbE [Y]=0$, $\var(X)=\var(Y)=1$, and  correlation coefficient~$\rho$. Let $P_X=\calN(\mu_X,\sigma_X^2)$ and $P_Y=\calN(\mu_Y,\sigma_Y^2)$. Then, 
\begin{align}
&\rvH_\infty(P_X,P_Y\|\pi_{XY})\nn\\*
&=\log\big(2\pi\sqrt{1-\rho^2}\big)+\frac{1-\rho \big(\min_{P_{XY}\in\calC(P_X,P_Y)}\bbE[XY]\big)}{1-\rho^2}\,\log\rme \\
&=\log\big(2\pi\sqrt{1-\rho^2}\big)+\frac{1+\rho(\sigma_X \sigma_Y-\mu_X\mu_Y)}{1-\rho^2}\,\log\rme,  
\end{align}
where the last equality easily follows from the  condition for equality in the Cauchy--Schwarz inequality. We note that this step is equivalent to computing the Wasserstein distance of order $2$ between $X$ and $-Y$; see \cite[Example~3.2.14]{Rachev98}. Furthermore, if $P_X=\pi_X=\calN(0,1)$   and $P_Y=\pi_Y=\calN(0,1)$ (so $\mu_X=\mu_Y=0$ and $\sigma_X=\sigma_Y=1$), 
\begin{equation}
\rvH_\infty(\pi_X,\pi_Y\|\pi_{XY})=\log\big(2\pi\sqrt{1-\rho^2}\big)+\frac{\log\rme}{1-\rho}.
\end{equation}
In contrast, the joint (differential) entropy of $\pi_{XY}$ is 
\begin{equation}
H(\pi_{XY})=\log\big(2\pi e\sqrt{1-\rho^2}\big)\le \rvH_\infty(\pi_X,\pi_Y\|\pi_{XY}),
\end{equation}
with equality if and only if $\rho=0$, i.e., $(X,Y)\sim \pi_{XY}$ is  a pair of  independent  Gaussian random variables.
\end{example}

%In other words, $H_{-1}(P_X, P_Y\|\pi_{XY})$ characterizes the exponent that a randomly generated pair of sequences $(X^n,Y^n)$ with joint distribution $\pi_{XY}^n$  where $X^n$ and $Y^n$ respectively have marginals $\pi_X^n$ and $\pi_Y^n$  [TBD]

The quantities in Definition~\ref{def:max_cross_ent} are used to characterize the following upper and lower bounds on the R\'enyi common information. 
\begin{definition} \label{def:pseudo_CI} %[Upper and Lower Pseudo-Common Information]
For $s>0$, define the {\em upper pseudo-common information of order $(1+s)$} as 
\begin{align}
\oGamma_{1+s}(\pi_{XY}) & := \min_{ \substack{P_W P_{X|W} P_{Y|W} :\\ P_{XY}=\pi_{XY}}}-\frac{1+s}{s}H(XY|W)\nn\\
&\qquad\qquad+ \bbE_{ P_W} \big[ \rvH_{s}(P_{X|W }, P_{Y|W }\|\pi_{XY})\big],\label{eqn:Gamma_up}
\end{align}
where the expectation on the second line can be explicitly written as $$\sum_w P_W(w)  \rvH_{s}(P_{X|W =w}, P_{Y|W =w}\|\pi_{XY}).$$  Similarly, define the {\em lower pseudo-common information of order $(1+s)$} as 
\begin{align}
\hspace{-.2in}\uGamma_{1+s}(\pi_{XY}) & := \inf_{ \substack{P_W P_{X|W} P_{Y|W} :\\ P_{XY}=\pi_{XY}}}-\frac{1+s}{s}H(XY|W)\nn\\
%&\quad+ \inf_{\substack{Q_{WW'}\in \\ \calC(P_W,P_W)}}\sum_{w,w'}Q_{WW' }(w,w') \rvH_{s}(P_{X|W=w}, P_{Y|W=w'}\|\pi_{XY}).\label{eqn:Gamma_down}
&\hspace{-.2in}\qquad+ \inf_{\substack{Q_{WW'} \\ \in \calC(P_W,P_W)}}\bbE_{Q_{WW' }} \big[\rvH_{s}(P_{X|W }, P_{Y|W' }\|\pi_{XY})\big],\label{eqn:Gamma_down}
\end{align}
where the expectation on the second line can be explicitly written as $$\sum_{w,w'}Q_{WW'}(w,w')\rvH_{s}(P_{X|W =w}, P_{Y|W'=w' }\|\pi_{XY}).$$

Also define $\oGamma_{1}(\pi_{XY})$, $\uGamma_{1}(\pi_{XY})$,  $\oGamma_{\infty}(\pi_{XY})$, and $\uGamma_{\infty}(\pi_{XY})$ to be the limits of $\oGamma_{1+s}(\pi_{XY})$ and $\uGamma_{1+s}(\pi_{XY})$ as $s\downarrow 0$ or $s\to\infty$. 
\end{definition}
Observe that these two definitions are rather similar. Indeed, if the inner infimum in \eqref{eqn:Gamma_down} is achieved by a coupling  $Q_{WW'} \in\calC(P_W,P_W)$ such that $Q_{WW'}(w,w')=P_W(w)\mathbbm{1}\{ w=w'\}$ for all $(w,w')\in\calW^2$, then  we have the favorable scenario in which  $\oGamma_{1+s}(\pi_{XY})=\uGamma_{1+s}(\pi_{XY})$ for all $s>0$.  This coupling is known as the  {\em equality coupling}.

Even though the expression in \eqref{eqn:Gamma_up} is somewhat involved, for some special classes of distributions, $\oGamma_{1+s}(\pi_{XY})$ turns out to be equal to Wyner's common information $C_{\Wyner}(\pi_{XY})$ for all $s\in (0,\infty]$. To state the desired result, we now define a hierarchy of product-like distributions.

\begin{definition} \label{def:pseudo_pdt} %[Pseudo-Product Distribution]
   Consider the following hierarchy of joint distributions.
\begin{enumerate}
\item[(a)]  A {\em product distribution} $\pi_{XY} \in  \calP(\calX\times\calY)$ is one in which  %$\pi_{XY}(x,y)=\pi_X(x)\pi_Y(y)$ for all %
there exists     functions $\alpha:\calX\to [0,\infty)$ and $\beta:\calY\to [0,\infty)$ such that $\pi_{XY} (x,y)=\alpha(x)\beta(y)$ for all
 $(x,y)\in\calX\times\calY$. In this case, $\alpha(x) = \sum_{y\in\calY}\pi_{XY}(x,y)$ is the marginal of $\pi_{XY}$ on $\calX$ and similarly for $\beta(y)$. The matrix of probabilities corresponding to $\pi_{XY}$ has rank~$1$ and $X$ is independent of $Y$.

\item[(b)]  A  {\em pseudo-product distribution} $\pi_{XY} \in  \calP(\calX\times\calY)$ is one in which there  exists some subset $\calA\subset\calX\times\calY$ such that 
\begin{equation}
\pi_{XY}(x,y) = \left\{ \begin{array}{cc}
\alpha(x)\beta(y) & (x,y)\in\calA\\
0 & \mbox{otherwise}
\end{array} \right.  \label{eqn:pseudo_prod}
\end{equation}
for some functions $\alpha:\calX\to [0,\infty)$ and $\beta:\calY\to [0,\infty)$ such that $\sum_{(x,y)\in\calA}\alpha(x)\beta(y)=1$.   %In this case, we say that $\pi_{XY}$ is a {\em product distribution on $\calA$}.  %If $\calA=\calX\times\calY$, then we say that $\pi_{XY}$ is a {\em product distribution}. 
%Additionally, we say that    $\pi_{XY}$ is a {\em Wyner-product distribution} if there exists an optimal distribution $P_W P_{X|W} P_{Y|W}$ attaining  the infimum in the definition of $C_{\Wyner}(\pi_{XY})$ in \eqref{eqn:CWyner_inf} such that $\pi_{XY}$ is a product distribution on $\calA_w=\supp(P_{X|W=w}) \times\supp(P_{Y |W=w})$ for each $w \in \supp(P_W)$.
%, i.e.,    \eqref{eqn:pseudo_prod} holds with $\calA$ replaced by $\calA_w$ for each $w\in\supp(P_W)$. 
%  \eqref{eqn:pseudo_prod} 
%\begin{equation}
%\pi_{XY}(\cdot,\cdot| \calA_w )= \left\{ \begin{array}{cc}
%\alpha_w(x)\beta_w(y) & (x,y)\in\calA_w\\
%0 & \mbox{otherwise}
%\end{array} \right. 
%\end{equation}
%for some functions $\alpha_w$ and $\beta_w$.
% In other words, the matrix   describing the joint distribution $\pi_{XY}(\cdot,\cdot| \supp(P_{X|W=w}) \times  \supp(P_{Y |W=w}))$ has rank one for each $w\in\supp(P_W)$. 
\item[(c)]   A {\em Wyner-product distribution} $\pi_{XY}\in  \calP(\calX\times\calY)$ is one  in which  there exists a  distribution $P_W P_{X|W} P_{Y|W}$ attaining  the infimum in the definition of $C_{\Wyner}(\pi_{XY})$ in \eqref{eqn:CWyner_inf} such that $\pi_{XY}$ restricted to $\calA_w=\supp(P_{X|W=w}) \times\supp(P_{Y |W=w})$ is a product distribution for all % on $\calA_w=\supp(P_{X|W=w}) \times\supp(P_{Y |W=w})$ for each
 $w \in \supp(P_W)$. In other words,  
\begin{equation}
  \pi_{XY}(x,y|\calA_w) := \frac{\pi_{XY}(x,y)}{\pi_{XY}( \calA_w)} \bone\{(x,y)\in\calA_w\} \label{eqn:wyner_prod}
 \end{equation}
 is a product distribution for all $w\in\supp(P_W)$.  
\end{enumerate}
\end{definition}
It can be seen that a pseudo-product distribution is a Wyner-product distribution. This is because  $\calA_w\subset \calA$ for all $w$, otherwise   there is some $(x,y)\in\calX\times\calY$ such that $P_{XY}(x,y)>0$ but $\pi_{XY}(x,y)=0$. Since $\pi_{XY}$ has the property in \eqref{eqn:pseudo_prod}, so does it on each $\calA_w$.   

%fulfilling the condition of a pseudo-product distribution.  % for some set $\calA$ in the definition of the pseudo-product distribution. 
%Such a set $\calA$ exists because otherwise, there is some $(x,y)\in\calX\times\calY$ such that $P_{XY}(x,y)>0$ but $\pi_{XY}(x,y)=0$. %to be any $\calA_w$ %for $w\in\supp(P_W)$ in \eqref{eqn:wyner_prod}), but, 
%%However,  the reverse implication may not hold in general. 

Obviously, a product distribution is a pseudo-product distribution (take $\calA$ to be $\calX\times\calY$).  However,
  a pseudo-product distribution need not be a {\em bona fide} product distribution as the next example shows.   Nevertheless, if $\supp(\pi_{XY})$ is a {\em product set} (i.e., a set $\calA$ that can be written as the Cartesian product $\calX'\times\calY'$ where $\calX'\subset\calX$ and $\calY'\subset\calY$), then a pseudo-product distribution is a product distribution. A Venn diagram of these classes of distributions is shown in Fig.~\ref{fig:products}.

\begin{figure}[!ht]
\vspace{-.2in}
\centering
\begin{overpic}[width=.7\textwidth]{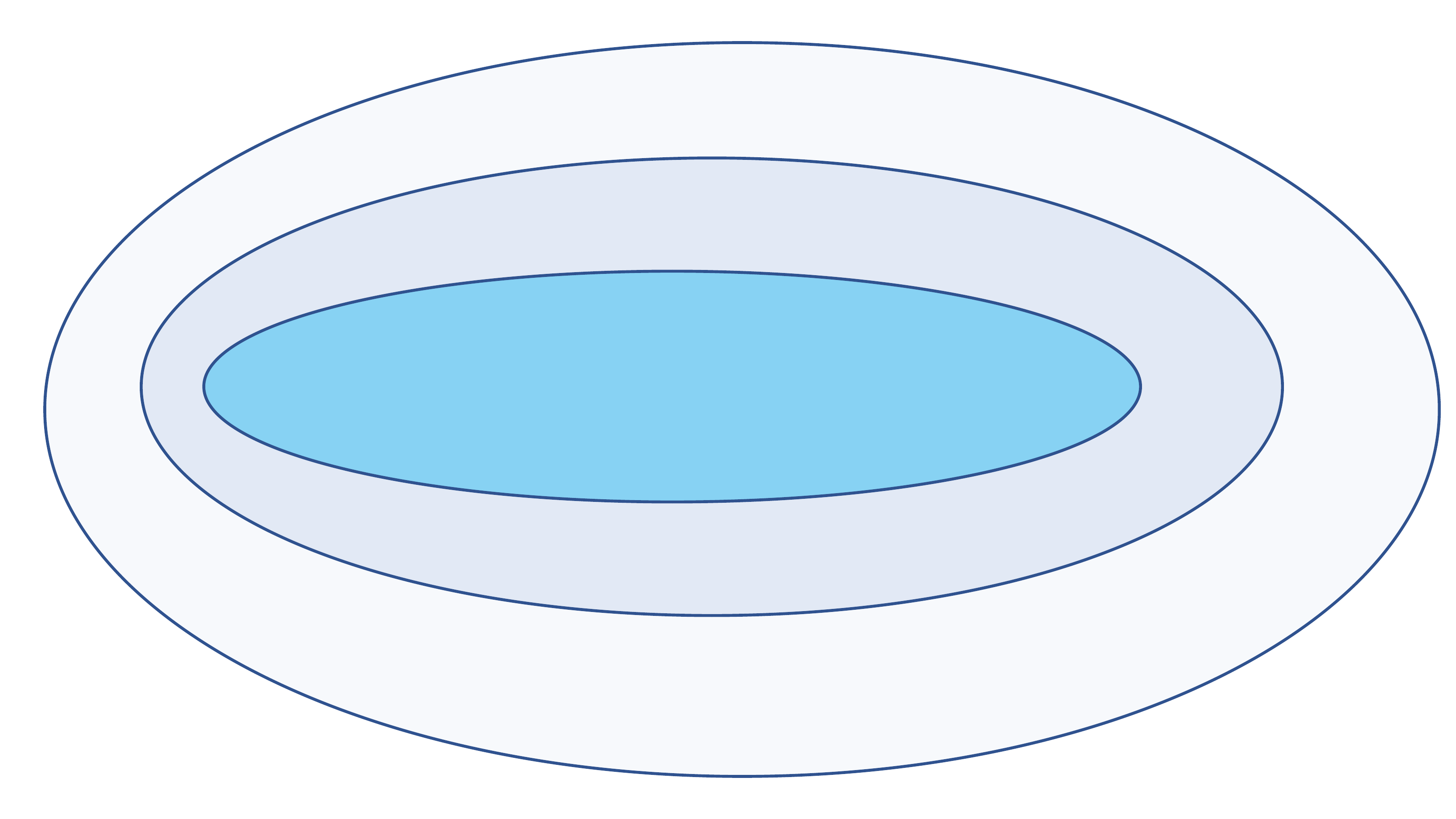}
{\footnotesize
\put(27,9){Wyner-product distributions}
\put(24,18.5){Pseudo-product distributions}
\put(27,29){Product distributions}}
\end{overpic}
\vspace{-.1in}
\caption{Venn diagram of various types of product distributions}
\label{fig:products}
\end{figure}

\begin{example}
Consider the joint distribution supported on $\{0,1\}^2$ with matrix of probabilities given by
\begin{equation}
\pi_{XY}=\frac{1}{\alpha_0\beta_0+\alpha_0\beta_1+ \alpha_1\beta_0}\begin{bmatrix}
\alpha_0\beta_0 & \alpha_0\beta_1 \\ \alpha_1\beta_0 & 0
\end{bmatrix}, \label{eqn:matrix_pseudo}
\end{equation}
where $\alpha_x,\beta_y >0$ for $x,y\in\{0,1\}$.
This is a  pseudo-product distribution but not a product distribution. To show the former statement,  take the set $\calA$ to be $\{0,1\}^2\setminus\{(1,1)\}$ and functions $\alpha(x)\propto\alpha_x$ and $\beta(y) \propto \beta_y$ for $x,y\in \{0,1\}$. For the latter statement, note that since the rank of the matrix in~\eqref{eqn:matrix_pseudo} is not  one, $\pi_{XY}$ is not a product distribution. 
\end{example}
We now state some useful properties of $\oGamma_{1+s}$ and $\uGamma_{1+s}$.
\begin{lemma} \label{lem:prop_Gamma}
The upper and lower pseudo-common information quantities satisfy the following properties.
\begin{enumerate}
\item[(a)] For the optimization in \eqref{eqn:Gamma_up} that defines $\oGamma_{1+s}(\pi_{XY})$, it suffices to restrict the cardinality $|\calW|\le|\calX||\calY|$. 
\item[(b)] The functions $s\mapsto \oGamma_{1+s}(\pi_{XY})$ and $s\mapsto \uGamma_{1+s}(\pi_{XY})$ are non-decreasing in $s > 0$. 
\item[(c)] As $s\downarrow 1$, the following limiting case  holds:
\begin{align}
\uGamma_1(\pi_{XY})\le \oGamma_1(\pi_{XY})=C_{\Wyner}(\pi_{XY}).
\end{align}
\item[(d)] As $s\to\infty$, the following limiting cases  hold:
\begin{equation}\label{eqn:oGamma_infty}
    \begin{aligned}
&\hspace{-.15in} \lim_{s\to\infty}\oGamma_{1+s}(\pi_{XY})   =\oGamma_{\infty}(\pi_{XY})\\
&\hspace{-.1in}=  \min_{ \substack{P_W P_{X|W} P_{Y|W} :\\ P_{XY}=\pi_{XY}}} -H(XY|W) + \bbE_{P_W}\big[ \rvH_\infty(P_{X|W },P_{Y|W  }\|\pi_{XY}) \big], 
% \times\max_{ \substack{Q_{XY}: \\ \in\calC(P_{X|W=w}, P_{Y|W=w})}}\sum_{x,y}Q_{XY}(x,y)\log\frac{1}{\pi_{XY}(x,y)}.
\end{aligned}
\end{equation}
and 
\begin{equation} \label{eqn:uGamma_infty}
    \begin{aligned}
&\hspace{-.35in}\lim_{s\to\infty}\uGamma_{1+s}(\pi_{XY})  =\uGamma_{\infty}(\pi_{XY})\\
&\hspace{-.35in}\qquad= \inf_{ \substack{P_W P_{X|W} P_{Y|W} :\\ P_{XY}=\pi_{XY}}}-H(XY|W)\nn\\
&\hspace{-.35in}\qquad\qquad+ \inf_{\substack{Q_{WW'}\\ \in \calC(P_W,P_W)}} \bbE_{Q_{WW'}}\big[ \rvH_\infty(P_{X|W},P_{Y|W'}\|\pi_{XY})\big]. 
\end{aligned}
\end{equation}
\item[(e)] For any $s\in (0,\infty]$, $\oGamma_{1+s}(\pi_{XY})=C_{\Wyner}(\pi_{XY})$ if and only if $\pi_{XY}$ is a Wyner-product distribution.
\end{enumerate}
\end{lemma}
Statement~(a) says that  $\oGamma_{1+s}(\pi_{XY})$ is efficiently computable as there is a cardinality bound on the auxiliary random variable $W$. Statement~(b) is clear and mirrors that of the operational definition of the R\'enyi common information we state later. Statements~(c) and~(d) say that the limit operations (as $s\downarrow 1$ and $s\to\infty$) ``commute'' with the minimizations. Finally, Statement~(e) says that the upper pseudo-common information of any order greater than $1$ is the same as Wyner's common information when the target distribution is a Wyner-product distribution so $\oGamma_{1+s}(\pi_{XY})$  offers another representation of Wyner's common information  for this class of distributions. 
\section{R\'enyi Common Information}
In this section, we formally define and state some known results on the R\'enyi common information. The following definition,  which differs from an alternative one in \citet{graczyk2020gray}, mirrors that of the Wyner's common information from the perspective of the distributed simulation problem (Definition~\ref{def:wyner_CI}). In the following, we only consider $s \in (-1,1]\cup\{\infty\}$.

\begin{definition} \label{def:RCI}
The  {\em normalized R\'enyi common information\footnote{To be analogous to Definition~\ref{def:wyner_CI}, we should term $T_{1+s}(\pi_{XY})$ as the {\em minimal normalized R\'enyi distributed simulation rate of order $1+s$}. However, since  we have established that the mimimal distributed simulation rate is Wyner's common information  in Theorem~\ref{thm:wynerCI_sim}, henceforth, to avoid having too many different terminologies, we refer to such fundamental limits (operational definitions) as common information quantities. In other words, we define common information quantities  operationally.}  of order $1+s$}  between a pair of random variables $(X,Y)\sim \pi_{XY}$, denoted as  $T_{1+s}(\pi_{XY})$, is the infimum of all rates $R$ such that there exists a sequence of $(n,R)$-fixed-length  distributed source simulation  codes $\{ (P_{X^n | M_n},P_{Y^n | M_n})\}_{n\in\bbN}$ (Definition~\ref{def:wyner_code}) satisfying 
\begin{equation}
\lim_{n\to\infty} \frac{1}{n} D_{1+s}\left(P_{X^nY^n} \middle\| \pi_{XY}^n\right)=0. \label{eqn:norm_rci}
\end{equation}
Similarly, the {\em unnormalized R\'enyi common information of order $1+s$}, denoted as $\tilT_{1+s}(\pi_{XY})$, is analogous to the normalized version except that the criterion in \eqref{eqn:norm_rci} is replaced with the more stringent condition
\begin{equation}
\lim_{n\to\infty}  D_{1+s}\left(P_{X^nY^n} \middle\| \pi_{XY}^n\right)=0. \label{eqn:unnorm_rci}
\end{equation}
\end{definition}
A few remarks on Definition~\ref{def:RCI} are in order. First, note that $T_1(\pi_{XY})=T(\pi_{XY})$ is exactly Wyner's common information, as defined in Definition~\ref{def:wyner_CI}. Second, since the unnormalized criterion in \eqref{eqn:unnorm_rci} is more stringent than that of the normalized one in~\eqref{eqn:norm_rci}, we have 
\begin{equation}
T_{1+s}(\pi_{XY})\le\tilT_{1+s}(\pi_{XY}) \quad\mbox{for all}\;\, s\ge -1 .
\end{equation}
Furthermore, by the monotonically non-decreasing nature of the R\'enyi divergence in its parameter, we see that $T_{1+s}$ and $\tilT_{1+s}$ are also monotonically non-decreasing  in their parameter, i.e., 
\begin{equation}
T_{1+s}(\pi_{XY})\le T_{1+s'}(\pi_{XY})\quad\mbox{and}\quad \tilT_{1+s}(\pi_{XY})\le \tilT_{1+s'}(\pi_{XY}) \label{eqn:monoT}
\end{equation}
for all $-1\le s\le s' \le\infty$. Finally, for the special case in which $s=0$, we obtain 
\begin{equation}
T(\pi_{XY})=T_{1} (\pi_{XY})=\tilT_1 (\pi_{XY})=C_{\Wyner}(\pi_{XY}),
\end{equation}
where the statement for the unnormalized case (the final equality) comes from the one-shot soft covering lemma as discussed in Remark~\ref{rmk:unnorm}.

We are now ready to state the main result of this section; this result is due to the present authors \citep{YuTan2018,yu2020corrections}.
\begin{theorem}[Bounds on R\'enyi common information] \label{thm:RCI}  %[Single-Letter Characterizations of R\'enyi Common Information]
The following hold:
\begin{enumerate}
\item[(a)] For $s=-1$, 
\begin{equation}
\tilT_0(\pi_{XY})= T_0(\pi_{XY})=0. \label{eqn:rci0}
\end{equation}
\item[(b)] For $s\in (-1,0]$, 
\begin{equation}
\tilT_{1+s}(\pi_{XY})=T_{1+s}(\pi_{XY})=C_{\Wyner}(\pi_{XY}). \label{eqn:rci_neg}
\end{equation}
\item[(c)] For $s\in (0,1] \cup\{\infty\}$, 
\begin{align}
\hspace{-.32in}\tilT_{1+s}(\pi_{XY}) \!\ge\! T_{1+s}(\pi_{XY})\!\ge\!\max\big\{ \uGamma_{1+s}(\pi_{XY}),C_{\Wyner}(\pi_{XY})\big\}, \label{eqn:rci_lower}
\end{align}
and 
\begin{align}
 T_{1+s}(\pi_{XY})\le \tilT_{1+s}(\pi_{XY})\le \oGamma_{1+s}(\pi_{XY}) .\label{eqn:rci_upper}
\end{align}
\end{enumerate}
\end{theorem}
For $s \in [-1, 0]$, we have tight characterizations of the normalized and unnormalized R\'enyi common information. For $s\in (0,1]\cup\{\infty\}$, we only have bounds in general. %Note that the bound $T_{1+s}(\pi_{XY})\ge C_{\Wyner}(\pi_{XY})$ for $s\in (0,1] \cup\{\infty\}$ in \eqref{eqn:rci_lower} is obvious in view of the monotonicity of $s\mapsto T_{1+s}$ as stated in \eqref{eqn:monoT}. 
Despite only having bounds for this case, combining~\eqref{eqn:rci_lower} and Lemma~\ref{lem:prop_Gamma}(e) yields  the following corollary.
\begin{corollary}[Sufficient condition for equality of R\'enyi and Wyner's common information] \label{cor:wyner_pdt}
Let $s\in (-1,1]\cup\{\infty\}$.  For any Wyner-product distribution $\pi_{XY}$,  
\begin{equation}
T_{1+s}(\pi_{XY}) = \tilT_{1+s}(\pi_{XY})=C_{\Wyner}(\pi_{XY}) . \label{eqn:TeqCWyner}
\end{equation}
%The equality in~\eqref{eqn:TeqCWyner} also holds for pseudo-product distributions.
\end{corollary}
%Since a pseudo-product distribution is a 
By the inclusions shown in Fig.~\ref{fig:products}, the equality in \eqref{eqn:TeqCWyner} also holds for pseudo-product distributions. 

\section{TV Common Information and Its Strong Converse} \label{sec:tvci}

Interestingly, the converse part of the proof of Part (b) of Theorem~\ref{thm:RCI} requires an auxiliary result concerning the so-called $\eps$-TV common information. We formally define this quantity in the following. 
\begin{definition} \label{def:TVCI}
For $0\le\eps<1$, the {\em $\eps$-TV common information} $T_{\eps}^{\mathrm{TV}}(\pi_{XY})$ between a pair random variables $(X,Y)\sim \pi_{XY}$ is the infimum of all rates~$R$ such that there exists a sequence of $(n,R)$-fixed-length  distributed source simulation  codes $\{ (P_{X^n | M_n},P_{Y^n | M_n})\}_{n\in\bbN}$ (Definition~\ref{def:wyner_code}) satisfying 
\begin{equation}
\limsup_{n\to\infty} \left| P_{X^nY^n}-\pi_{XY}^n \right|\le\eps. \label{eqn:tv_ci}
\end{equation}
We abbreviate the $0$-TV common information as the {\em TV common information}.
If $T_{\eps}^{\mathrm{TV}}(\pi_{XY})$ does not depend on $0\le\eps<1$, we say that {\em the strong converse property} holds.
\end{definition}

If the strong converse property~\cite{Han10,Wolfowitz57}  holds, there is a sharp phase transition in rates such that the TV distance between the synthesized and target distributions can be made arbitrarily small and those rates such that the TV distance necessarily tends to one as the blocklength grows. This is usually a very pleasing phenomenon in Shannon theory because in this case, there is no tradeoff between a permissible error and the rate, at least in the first-order sense; see Fig.~\ref{fig:tv_ci}. The tradeoff between the error probability and rate can be seen in the {\em second-order coding rate}~\cite{PPV10,Hayashi08, Hayashi09,TanBook}.

\begin{figure}[t]
\centering
\begin{picture}(115, 113)
\setlength{\unitlength}{.47mm}
\put(0, 10){\vector(1, 0){110}}
\put(10, 0){\vector(0,1){90}}
\put(110,8){ $R$}
\put(-25, 84){TV Dist.} 
%\put(95, 8){\line(0, 1){4}}
\put(60, 8){\line(0, 1){4}}
\linethickness{0.4mm}
\put(60, 75){\line(-1, 0){50}}
\put(60, 75){\line(0, -1){65}}
\put(60, 10){\line(1, 0){45}}
\put(2, 73){$1$}
%\put(94, 1){$1$}
\put(45, 0){$C_{\Wyner}(\pi_{XY})$}
\put(2, 1){$0$}
\end{picture}
\caption{Plot of the asymptotic TV distance $\lim_{n\to\infty}\left| P_{X^nY^n}-\pi_{XY}^n \right|$ against the rate $R$. Observe the sharp phase transition at $C_{\Wyner}(\pi_{XY})$.}
\label{fig:tv_ci}
\end{figure}
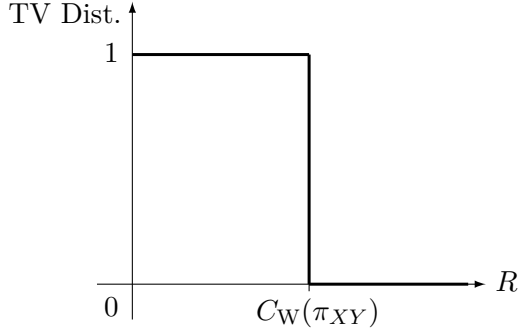

\enlargethispage{-\baselineskip}
Unlike the R\'enyi common information with orders $(1,2]\cup\{\infty\}$, a full characterization of $T_{\eps}^{\mathrm{TV}}(\pi_{XY})$ is available. The strong converse part is due to the present authors~\citep{YuTan2018} leveraging an ingenious information spectrum-based~\cite{VH94, Han10}, single-letterization technique by \citet{oohama18}, while the achievability part  can be obtained  using  arguments in \citet{Hayashi06} or \citet[Lemma~IV.1]{cuff13}. 

\begin{theorem}[$\eps$-TV Common Information] \label{thm:tv_ci}
The following hold:
\begin{enumerate}
\item[(a)] For any $\eps\in [0,1)$, 
\begin{equation}
T_{\eps}^{\mathrm{TV}}(\pi_{XY})=C_{\Wyner}(\pi_{XY }).
\end{equation}
\item[(b)] Let $R>C_{\Wyner}(\pi_{XY })$. Then, there exists a sequence of rate-$R$  codes such that $\left| P_{X^nY^n}-\pi_{XY}^n \right|$ converges to $0$ exponentially fast (i.e., $\limsup_{n\to\infty} \frac{1}{n}\log\left| P_{X^nY^n}-\pi_{XY}^n \right|<0$). 
\item[(c)] Let  $R<C_{\Wyner}(\pi_{XY })$. Then,   all sequences of rate-$R$ codes result in   $\left| P_{X^nY^n}-\pi_{XY}^n \right|$ converging to $1$ exponentially fast  (i.e., $\limsup_{n\to\infty} \frac{1}{n}\log (1-\left| P_{X^nY^n}-\pi_{XY}^n \right|)<0$).  
\end{enumerate}
\end{theorem}

This theorem is illustrated in an alternative way in Fig.~\ref{fig:tv_ci}. In fact, Parts~(b) and~(c) say that not only do we have matching achievability and strong converse results, these results are also {\em exponentially} strong in the sense that the TV distance converges to $0$ and $1$ exponentially fast. This rate of convergence, for the exponentially strong converse part, has implications for the converse proof of Part~(b) of Theorem~\ref{thm:RCI}. We discuss this in its proof sketch in Section~\ref{sec:sneg}.

\section{Doubly Symmetric Binary Sources} \label{sec:renyi_dsbs}
 
We now consider the DSBS with crossover probability $p$ as depicted in Fig.~\ref{fig:dsbs}. Since the R\'enyi common information for the case $s\in (-1,0]$ is exactly Wyner's common information, we can see how it depends on $p$ from  Fig.~\ref{fig:wyner_dsbs}. Thus, we will only be concerned with the case $s\in(0,1]\cup\{\infty\}$. Here we show  that for the DSBS, we have strong numerical evidence that  the upper bound on the R\'enyi common information coincides with the lower bound. Recall the definitions of $a$, $\alpha$ and $\beta$ from Section~\ref{sec:dsbs}.
\begin{proposition}\label{prop:RCI_dsbs}
If $\pi_{XY}$ is  a DSBS with crossover probability $p$  and $s\in (0,1]$, the R\'enyi common information can be upper bounded as 
\begin{align}
T_{1+s}(\pi_{XY})&\le \tilT_{1+s}(\pi_{XY}) \nn\\
&\le-\frac{1+s}{s}\cdot 2 h (a)+\frac{1}{s}\Big[ h_4(q^*, a-q^*, a-q^*, 1+q^*-2a)  \nn\\*
&\qquad\qquad\qquad\qquad\qquad\qquad-2s(a-q^*)\log\beta\Big], \label{eqn:RCI_dsbs_s}
\end{align}
where  $h_4(a_1,a_2,a_3,a_4)=-\sum_{i=1}^4 a_i\log a_i$ is the {\em quaternary entropy} and 
\begin{align}
q^* &:=\frac{\sqrt{\kappa(\bara-a)^2+4\kappa a\bara} -(\kappa(\bara-a)+2a) }{2(\kappa-1) } \label{eqn:q_star}
\end{align}
where $\kappa :=\big( {\alpha}/{\beta}\big)^{2s}$.
For $s=\infty$, 
\begin{align}
\hspace{-.2in}T_{\infty}(\pi_{XY})&= \tilT_{\infty}(\pi_{XY})\label{eqn:RCI_infty0}\\
\hspace{-.2in} & =-2h (a)-(1 -  2a)\log\bigg( \frac{a^2+\bara^2}{2} \bigg)- 2a\log(a\bara). \label{eqn:RCI_infty}
\end{align}
\end{proposition}

The idea of the proof of~\eqref{eqn:RCI_dsbs_s} in Proposition~\ref{prop:RCI_dsbs} is straightforward but tedious. It involves considering the Markov chain as shown on the right plot of Fig.~\ref{fig:dsbs} and noticing for the random variables $(X,W,Y)$ (such that $X-W-Y$ forms a Markov chain), the coupling set is
\begin{equation}
\calC(P_{X|W=w}, P_{Y|W=w})=\left\{\begin{bmatrix}
a & a-q \\ a-q & 1+q -2a
\end{bmatrix}: 0\le q\le a  \right\}.
\end{equation}
By noticing this, we can then evaluate  the maximal $s$-mixed cross entropy $\rvH_{s}(P_{X|W=w}, P_{Y|W=w}\|\pi_{XY})$ by optimizing over the   scalar parameter $0\le q\le a$ to yield $q^*$ in~\eqref{eqn:q_star}. This shows~\eqref{eqn:RCI_dsbs_s}. We defer the discussion and justification of the  R\'enyi common information  of order~$\infty$ in \eqref{eqn:RCI_infty0}--\eqref{eqn:RCI_infty} to the next section. 

\begin{figure}[!ht]
\centering
\includegraphics[width = .8\columnwidth]{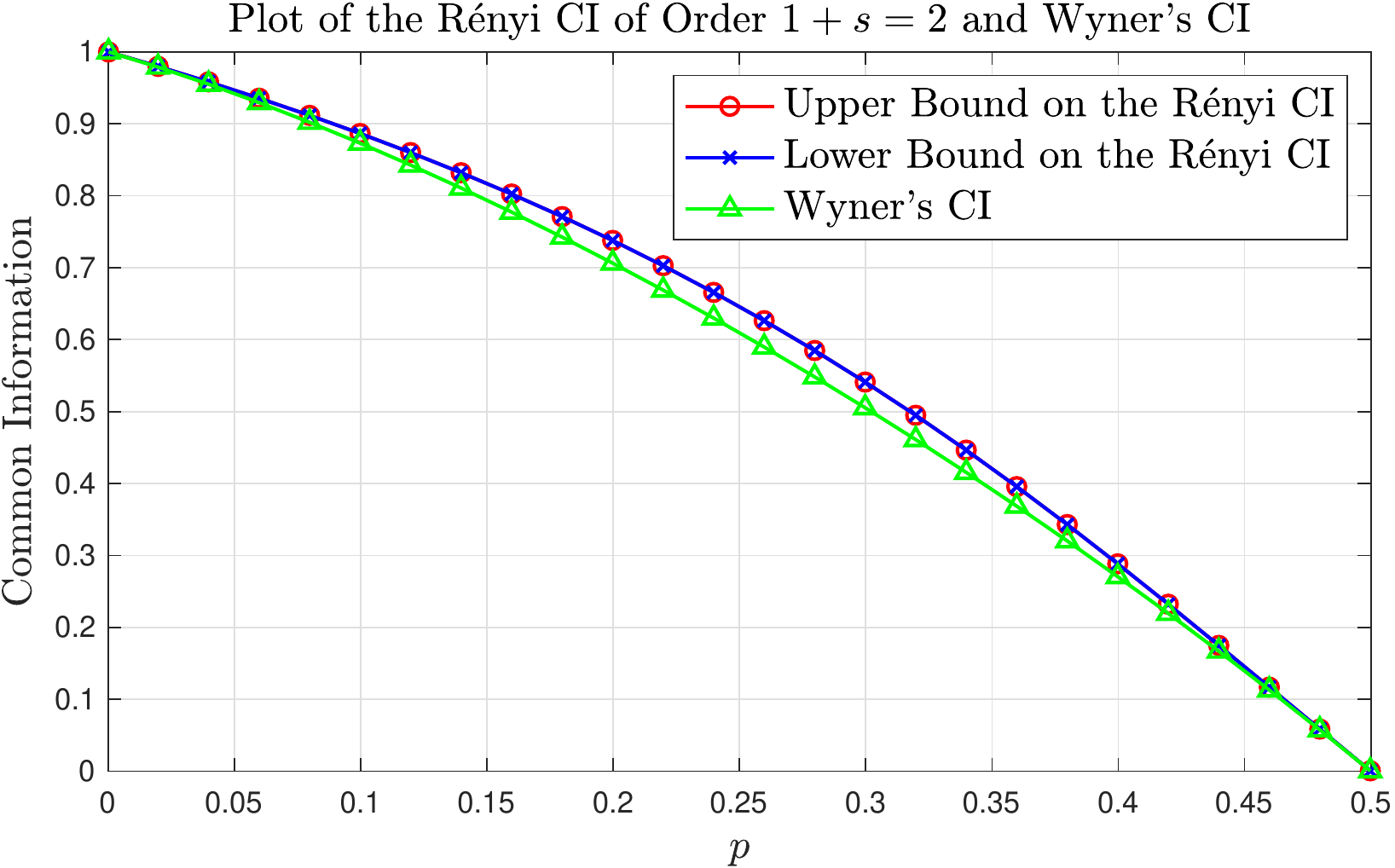}
\caption{Plots of the upper bound in \eqref{eqn:RCI_dsbs_s} and lower bound in \eqref{eqn:rci_lower}
of the R\'enyi common information with order $1+s=2$, and Wyner's common information}
\label{fig:renyi_dsbs}
\end{figure}

Upper and lower bounds for the R\'enyi common
information of order $1+s=2$, as well as Wyner's common information for the
DSBS are illustrated in Figs.~\ref{fig:renyi_dsbs} and~\ref{fig:renyi_dsbs_s}. Unlike the upper bound, we  do not have a closed-form expression for the lower bound so we resort to numerical optimization to evaluate \eqref{eqn:rci_lower}. To do so, we   gradually increase the alphabet size of $W$ from $2 $ to $10$ and we notice for the DSBS that  this does not change the resulting curve and in fact it appears to coincide with the upper bound. Hence it is
natural to conjecture the upper bound in~\eqref{eqn:RCI_dsbs_s} in Proposition~\ref{prop:RCI_dsbs} for the
DSBS is tight. Finally, we note that  the  R\'enyi common information   of orders larger than $1$ for the DSBS are strictly larger than Wyner's common information. %This can be seen  in Fig.~\ref{fig:renyi_dsbs_s}.

\begin{figure}[!ht]
\centering
\includegraphics[width = .8\columnwidth]{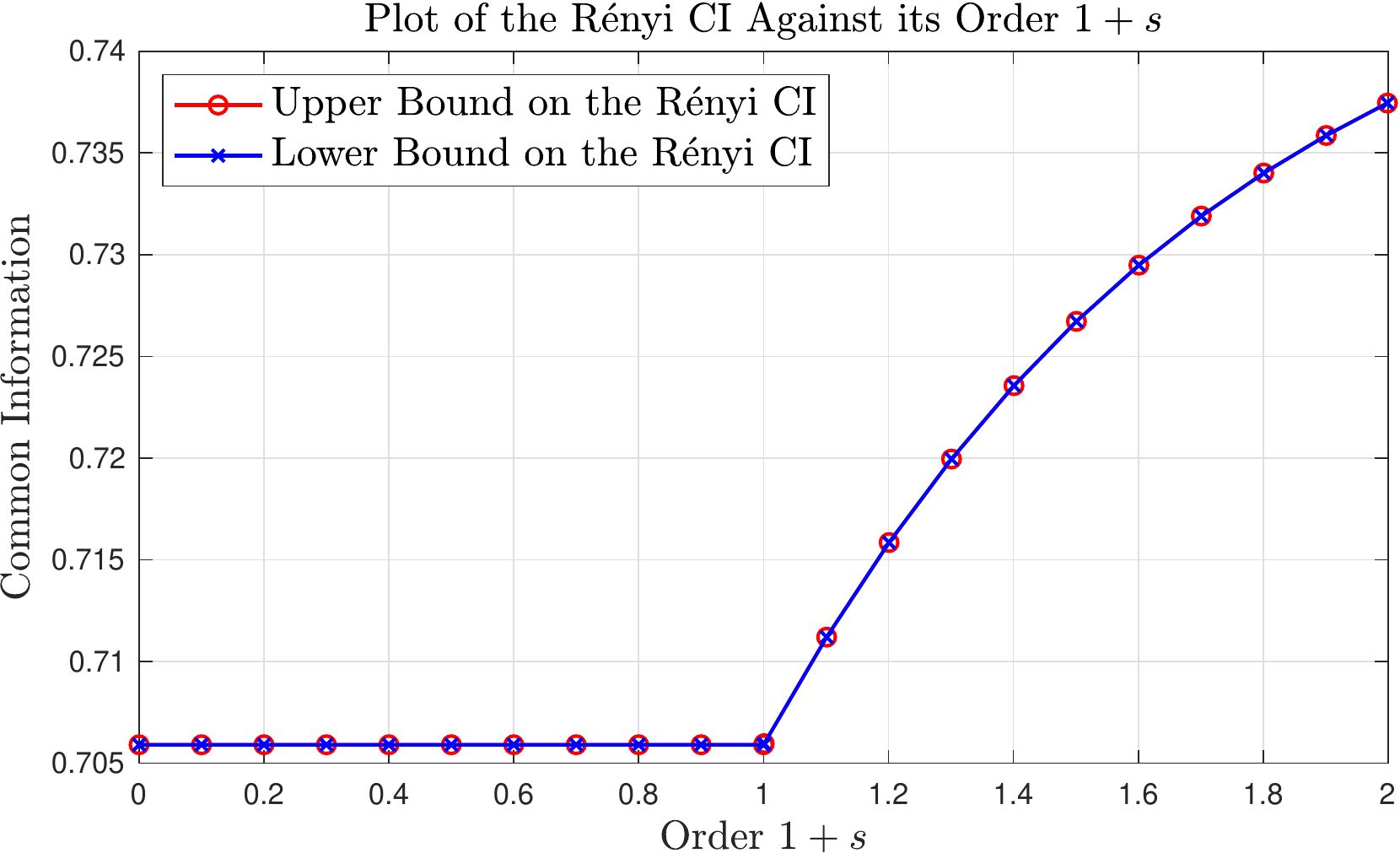}
\caption{Illustrations of the upper bound in \eqref{eqn:RCI_dsbs_s} and lower bound in \eqref{eqn:rci_lower}
of the R\'enyi common information as functions of $s\in (-1,1]$ (or order $1+s\in (0,2]$) for the DSBS with crossover probability $p=0.2$. Notice that for $s\in (-1,0]$, the R\'enyi common information $T_{1+s}(\pi_{XY})$ is Wyner's common information $T(\pi_{XY})=C_{\Wyner}(\pi_{XY})$  so the curve is constant in this range. }
\label{fig:renyi_dsbs_s}
\end{figure}

\section{Proof Sketches} \label{sec:prf_rci}
In this section, which may be skipped at a first reading, we provide  proof   sketches  of Theorems~\ref{thm:RCI} and~\ref{thm:tv_ci}. See \cite{YuTan2018} and \cite{yu2020corrections}  for details.  High-level ideas and interconnections among the  proofs  are presented in Table~\ref{tab:prf_rci}.

\begin{table}[t]
\caption{Summary of proof ideas for the various cases of Theorems~\ref{thm:RCI} and~\ref{thm:tv_ci}. RCI, WCI, norm.\ and unnorm.\ respectively stand for R\'enyi and Wyner's common information, normalized and unnormalized.}
\label{tab:prf_rci}
\resizebox{\textwidth}{!}{
\centering{\small
\begin{tabular}{|c|c|c|}
\hline
                                                                                & Achievability                                                                                                               & Converse                                                                                                     \\ \hline\hline
TV CI                                                                           & \begin{tabular}[c]{@{}c@{}}Soft-covering~\cite{Hayashi06,cuff13} or \\ implied by unnorm.\ WCI\end{tabular} & \begin{tabular}[c]{@{}c@{}}Information spectrum \\ method~\cite{oohama18}\end{tabular}   \\ \hline
\begin{tabular}[c]{@{}c@{}}Unnorm.\ RCI \\ $s\in (-1,0]$\end{tabular}            & Implied by unnorm.\  WCI & \begin{tabular}[c]{@{}c@{}}Implied by norm. RCI \\ $s\in (-1,0]$\end{tabular}                                \\ \hline
\begin{tabular}[c]{@{}c@{}}Norm.\ RCI \\ $s\in (-1,0]$\end{tabular}              & \begin{tabular}[c]{@{}c@{}}Implied by unnorm.\ RCI \\ $s\in (-1,0]$\end{tabular}                                               & \begin{tabular}[c]{@{}c@{}}TV exp.\ strong converse \\ \& Pinsker's inequality~\eqref{eqn:pinsker}\end{tabular} \\ \hline
\begin{tabular}[c]{@{}c@{}}Unnorm.\ RCI\\ $s\in (0,1]\cup\{\infty\}$\end{tabular} & \begin{tabular}[c]{@{}c@{}} $\!\!$ Soft-covering \& $\!\!$ \\ truncated product dist.\ \end{tabular}                                          & \begin{tabular}[c]{@{}c@{}}Implied by norm.\ RCI \\ $s\in (0,1]\cup\{\infty\}$\end{tabular}                    \\ \hline
\begin{tabular}[c]{@{}c@{}}Norm.\ RCI \\ $s\in (0,1]\cup\{\infty\}$\end{tabular}  & \begin{tabular}[c]{@{}c@{}}Implied by unnorm.\ RCI \\ $s\in (0,1]\cup\{\infty\}$\end{tabular}                                   &  \begin{tabular}[c]{@{}c@{}} Lem~\ref{lem:lb_renyi} \& Chain rule \\$\!\!$ for coupling sets (Lem.~\ref{lem:coupling}) $\!\!$  \end{tabular}            \\ \hline
\end{tabular}}}
\end{table}

\subsection{Sketch of the Achievability   of Theorem~\ref{thm:tv_ci} } \label{sec:ach_tv}
%There are several ways to prove that if $R> C_{\Wyner}(\pi_{XY})$, then there exists a sequence of distributed source simulation  codes $\{ (P_{X^n | M_n},P_{Y^n | M_n})\}_{n=1}^\infty$ such that the TV distance $| P_{X^nY^n}-\pi_{XY}^n |$ vanishes. 

First, we note from Theorem~\ref{thm:wyner_cts} and Remark~\ref{rmk:unnorm}  that if $R>C_{\Wyner}(\pi_{XY})$, there   exists  a sequence of codes $\{ (P_{X^n | M_n},P_{Y^n | M_n})\}_{n\in\bbN}$ such that the unnormalized relative entropy  $D(P_{X^nY^n}\|\pi_{XY}^n)$ vanishes. By Pinsker's inequality in~\eqref{eqn:pinsker} (which says that the relative entropy dominates the TV distance), we see that the TV distance also vanishes. 

Alternatively, one can  directly leverage  the soft-covering lemma for the TV distance   (as stated in~\eqref{eqn:soft-covering-tv} in Lemma~\ref{lem:soft-covering}) to show that if $R>C_{\Wyner}(\pi_{XY})$, there exists a sequence  of rate-$R$ distributed source simulation codes such that TV distance between $P_{X^nY^n}$ and $\pi_{XY}^n$ vanishes. 
%This result, stated in \eqref{eqn:soft-covering-tv}, is completely analogous to Wyner's soft covering lemma (Lemma~\ref{lem:soft-covering}) except that the guarantee is on the TV distance instead of the normalized relative entropy in~\eqref{eqn:soft-covering-approx}.

\subsection{Sketch of the Exponential Strong Converse   of Theorem~\ref{thm:tv_ci} }
The proof of the exponential strong converse requires a careful application of the  information spectrum method~\cite{VH94} due to~\citet{oohama18}, who used this technique to provide the first proof of the strong converse for the Wyner-Ziv problem~\cite{wynerziv}. The idea is to express the TV distance  $| P_{X^nY^n}-\pi_{XY}^n |$  in terms of the probability of some ``error events''. Roughly speaking, for any synthesis code with 
\begin{equation}
\frac{1}{n}\log |\calM_n|\le R, \label{eqn:logMleR}
\end{equation}
we can lower bound the TV distance as 
\begin{equation}
|P_{X^nY^n}-\pi_{XY}^n |\ge 1-\Pr\big(  \calA_1\cap\calA_2\cap \calA_3 ~\big|~  \calS \big)- 3\cdot  2^{-n\eta},  \label{eqn:lower_bd_TV}
\end{equation}
where for any $\eta>0$ and $Q_{X^nY^n}$ and $Q_{X^nY^n|M_n}$, the sets above are defined as
\begin{align}
\calA_1 &:=\bigg\{ (x^n,y^n) :\frac{1}{n}\log\frac{\pi_{XY}^n(x^n,y^n)}{Q_{X^nY^n}(x^n,y^n)} \ge-\eta\bigg\}\times\calM_n \\
\calA_2 &:=\bigg\{ (x^n,y^n,m) \!:\!\frac{1}{n}\log\frac{ P_{X^n|M_n}(x^n|m)P_{Y^n|M_n}(y^n|m)}{Q_{X^nY^n|M_n}(x^n,y^n|m)}\! \ge\!-\eta\bigg\}  \\
\calA_3 &:=\bigg\{ (x^n,y^n,m) :\frac{1}{n}\log\frac{ Q_{X^nY^n|M_n}(x^n,y^n|m) }{ \pi_{XY}^n(x^n,y^n)} \le R+\eta\bigg\}  ,
\end{align}
and 
%\tilde{\calA}_1 &:=\quad\tilde{\calA}_2 :=\supp(P_{X^nY^nM_n})\quad\mbox{and}\quad
$\calS := \big(\supp(\pi_{XY}^n)\times \calM_n\big)\cap\supp\big(P_{X^nY^nM_n}\big)$. 
 The rest of the proof follows from choosing  $Q_{X^nY^n}$ and $Q_{X^nY^n|M_n}$ appropriately; the freedom to allow us to do so  in converse proofs was first noticed by~\citet{Hayashi03}. One then applies a Chernoff bound to the probability in~\eqref{eqn:lower_bd_TV} and single-letterizes the resultant exponent. All in all, we obtain that under~\eqref{eqn:logMleR}, 
\begin{equation}
|P_{X^nY^n}-\pi_{XY}^n |\ge 1-2^{-nF(R)},
\end{equation}
where $F(R)$ is an exponent function that is strictly positive when $R<C_{\Wyner}(\pi_{XY})$ and equal to $0$ otherwise. This completes the proof of the exponential strong converse.

\subsection{Sketch of the Achievability   of Theorem~\ref{thm:RCI}(b)}
For any $s\in (-1,0]$, the fact that $\tilT_{1+s}(\pi_{XY}) \le C_{\Wyner}(\pi_{XY})$ is obvious due to monotonically non-decreasing nature of  $s\mapsto\tilT_{1+s}(\pi_{XY})$ and the fact that $\tilT_{1}(\pi_{XY})  =C_{\Wyner}(\pi_{XY})$.
\subsection{Sketch of the Converse   of Theorem~\ref{thm:RCI}(b)} \label{sec:sneg}
The converse part for the case $s\in (-1,0]$, i.e., that $T_{1+s}(\pi_{XY}) \ge C_{\Wyner}(\pi_{XY})$ is more interesting and leverages a Pinsker-like relationship between the TV distance and the R\'enyi divergence due to \citet{sasonRenyi}, which we restate here as it may be of independent interest. 

\begin{lemma}[Pinsker-like inequality for R\'enyi divergence] 
\label{lem:Renyiinequality} For any $s>-1$,
\begin{align}
  \inf_{P_{X},Q_{X}:\left|P_{X}-Q_{X}\right|\geq\epsilon}D_{1+s}(P_{X}\|Q_{X}) % & =\inf_{P_{X},Q_{X}:\left|P_{X}-Q_{X}\right|=\epsilon}D_{1+s}(P_{X}\|Q_{X})\\
 & =\inf_{q\in[0,1-\epsilon]}d_{1+s}(q+\epsilon\|q),\label{eqn:renyi_eq_tv}
\end{align}
and for any $s\in(-1,0)$,
\begin{align}
  \inf_{q\in[0,1-\epsilon]}d_{1+s}(q+\epsilon\|q)
 & \geq\left[\min\left\{ 1,\frac{1 \! +\! s}{s}\right\} \log\frac{1}{1\! -\! \epsilon}\! +\! \frac{1}{s}\log 2\right]^{+},\label{eqn:lower_bd_renyi}
\end{align}
where
\begin{equation}
d_{1+s}(p\|q):=\left\{\begin{array}{cc}
\displaystyle\frac{1}{s}\log\left(p^{1+s}q^{-s}+\bar{p}^{1+s}\bar{q}^{-s} \right), & s\geq-1,s\neq0 \vspace{.03in}\\
\displaystyle p\log\frac{p}{q}+\bar{p}\log\frac{\bar{p}}{\bar{q}}, & s=0
\end{array}
\right.  \label{eqn:bin_renyi}
\end{equation}
denotes the {\em binary R\'enyi divergence of order $1+s$}.
\end{lemma}
\begin{remark} \label{rmk:sason}
%Define the function $g_{1+s}(\epsilon):=  \inf_{q\in[0,1-\epsilon]}d_{1+s}(q+\epsilon\|q)$ where  $\epsilon$ denotes the TV distance between $\mathrm{Bern}(q+\epsilon)$ and $\mathrm{Bern}(q)$.  
\citet{sasonRenyi} showed that 
\begin{equation}
\inf_{q\in[0,1-\epsilon]}d_{1/2}(q+\epsilon\|q)= \log\frac{1}{1-\epsilon^2}.%\approx\epsilon^2 \quad\mbox{as}\quad\epsilon\downarrow 0. \label{eqn:sason_pinskder}
\end{equation}
%g_{1/2}(\epsilon)=-\log(1-\epsilon^2)\approx\epsilon^2 $ as $\epsilon\to 0^+$. 
%Lemma~\ref{lem:Renyiinequality} generalizes Pinsker's inequality~\eqref{eqn:pinsker}.
\end{remark}

\begin{remark} \label{rmk:gila}
\citet{Gilardoni} showed for $\alpha=1+s\in (0,1)$ that %derived Pinsker-type lower bounds on the R\'enyi divergence of order $\alpha=1+s\in (0,1)$ as follows
\begin{equation}
\inf_{P_{X},Q_{X}:\left|P_{X}-Q_{X}\right|\geq\epsilon} D_{\alpha}(P_X\|Q_X)\ge\frac{1}{2}\alpha\epsilon^2+\frac{1}{9}\alpha(1+5\alpha	-5\alpha^2)\epsilon^4. \label{eqn:gilardoni}
\end{equation}
These two remarks imply that the  minimal R\'enyi divergence of order less than $1$ subject to the TV distance between the two distributions  having TV distance $\epsilon$  behaves quadratically in $\epsilon$. %The same is true for the R\'enyi divergences of other orders. 
Thus, these can be considered as Pinsker-type inequalities for the R\'enyi divergence.
\end{remark}
Using Lemma~\ref{lem:Renyiinequality}, the converse part for Theorem~\ref{thm:RCI}(b) is obvious. If $R< C_{\Wyner}(\pi_{XY})$, the TV distance converges to $1$ exponentially fast. In other words, 
\begin{equation}
| P_{X^nY^n} -\pi_{XY}^n|\ge 1-2^{-n\delta_n}
\end{equation}
for some sequence $\{\delta_n\}_{n\in\bbN}\subset [0,\infty)$ satisfying $\liminf_{n\to\infty}\delta_n>0$. Thus using Lemma~\ref{lem:Renyiinequality} (with $1-2^{-n\delta_n}$ in place of $\epsilon$), 
\begin{align}
\liminf_{n\to\infty}\frac{1}{n}D_{1+s}( P_{X^nY^n} \|\pi_{XY}^n)&\ge\liminf_{n\to\infty} \left\{ \min\bigg\{ 1,\frac{1+s}{s}\bigg\}\delta_n +\frac{\log 2}{ns}   \right\} \nn\\*
&> 0 ,
\end{align}
showing that if the rate $R$ is strictly smaller than Wyner's common information $C_{\Wyner}(\pi_{XY})$, the normalized R\'enyi divergence cannot converge to zero. Thus, $T_{1+s}(\pi_{XY})\ge  C_{\Wyner}(\pi_{XY})$ for any $s\in (-1,0)$. The case $s=0$ follows from the converse for Wyner's common information. 

%Finally, we see that the bounds in Remarks~\ref{rmk:sason} and \ref{rmk:gila} are not sufficient for showing the converse for the unnormalized R\'enyi common information of order less than one. The lower bound in~\eqref{eqn:lower_bd_renyi} is necessary. 
\subsection{Sketch of the Achievability   of Theorem~\ref{thm:RCI}(c)} \label{sec:ach_RCIc}
We only consider $s\in (0,1]$ since the proof ideas for   $s =\infty$ are similar to those for  $s\in (0,1]$. The achievability of Theorem~\ref{thm:RCI}(c)  follows by carefully evaluating the one-shot soft-covering result in Lemma~\ref{lem:one-shot-sc}. We set $\pi_U , P_{U|W}, P_W$, and $R$ to be $\pi_{XY}^n$, $P_{X^nY^n|W^n}=P_{X^n|W^n}P_{Y^n|W^n}$, $P_{W^n}$ and $nR$ respectively. Note that if there exists a sequence of distributions $\{ P_{W^n} P_{X^n|W^n}P_{Y^n|W^n}  \}_{n\in\bbN}$ such that $D_{1+s}(P_{X^nY^n}\|\pi_{XY}^n)\to 0$ and 
\begin{equation}
R> \limsup_{n\to\infty}\frac{1}{n}D_{1+s}(P_{X^nY^n|W^n}\|\pi_{XY}^n|P_{W^n}) , \label{eqn:Rlimsup}
\end{equation}
then from Lemma~\ref{lem:one-shot-sc}, there exists a sequence of  distributed source simulation codes $\{  (P_{X^n|M_n},P_{Y^n|M_n}  )\}_{n\in\bbN}$  such that 
\begin{align}
&\limsup_{n\to\infty}D_{1+s}\big( P_{X^nY^n|M_n}\big\|\pi_{X^nY^n}|P_{M_n} \big)\nn\\
&\le\limsup_{n\to\infty}\frac{1}{s}\log\Big[\exp \big(sD_{1+s}(P_{X^nY^n|W^n}\|\pi_{XY}^n|P_{W^n}) -nsR\big) \nn\\*
&\qquad+\exp\big(sD_{1+s}(P_{X^nY^n}\|\pi_{XY}^n )\big)\Big]\\
&\le\limsup_{n\to\infty}  \frac{1}{s}\! \log\Big[\! \exp \! \big(sD_{1+s}(P_{X^nY^n|W^n}\|\pi_{XY}^n|P_{W^n})\!  -\! nsR\big) \!+\! 1\Big]\! \! \\*
&=0,
\end{align}
where the last equality follows from \eqref{eqn:Rlimsup}.
Thus, for $s\in (0,1]$, 
\begin{equation}
\tilT_{1+s}(\pi_{XY})\le \inf \,\limsup_{n\to\infty}\frac{1}{n}D_{1+s}(P_{X^nY^n|W^n}\|\pi_{XY}^n|P_{W^n}), \label{eqn:upper_Tpos}
\end{equation}
where the infimum is over all sequences  $\{ P_{W^n} P_{X^n|W^n}P_{Y^n|W^n}  \}_{n\in\bbN}$ such that $D_{1+s}(P_{X^nY^n}\|\pi_{XY}^n)\to 0$.  

\enlargethispage{-\baselineskip}
As a result, the achievability proof reduces to finding a tractable  joint  distribution $P_{W^n} P_{X^n|W^n}P_{Y^n|W^n} $ such that the conditional R\'enyi divergence in \eqref{eqn:upper_Tpos} can be single-letterized. We first choose a distribution $Q_{WXY} \in \calP(\calW\times\calX\times\calY)$ such that $Q_{XY}=\pi_{XY}$ and 

\noindent
\begin{align}
P_{W^n}(w^n) &\propto Q_W^n(w^n)\mathbbm{1} \left\{ w^n\in\calT_{\epsilon'}(Q_W) \right\} ,\label{eqn:trunc1}\\
P_{X^n|W^n}(x^n|w^n) &\propto Q_{X|W}^n(x^n|w^n)\mathbbm{1} \left\{ x^n\in\calT_{\epsilon}(Q_{WX}|w^n) \right\},\label{eqn:trunc2}\\
P_{Y^n|W^n}(x^n|w^n) &\propto Q_{Y|W}^n(y^n|w^n)\mathbbm{1} \left\{ y^n\in\calT_{\epsilon}(Q_{WY}|w^n) \right\},\label{eqn:trunc3}
\end{align}
where $0<\epsilon'<\epsilon\le 1$. This triple is known as a {\em truncated product  distribution}, also used by \citet{Vellambi} and has two desirable features. First, it behaves like a {\em bona fide} product distribution. Indeed, 
\begin{equation}
P_{X^nY^n}(x^n,y^n)\le\frac{\pi_{XY}^n(x^n,y^n)}{1-\gamma_n}\quad\mbox{for all}\;\, (x^n,y^n) \in \calX^n\times\calY^n,
\end{equation}
where, roughly speaking, $\gamma_n=o(1)$ represents the probability of   atypical sets.  This property ensures that the constraint   $D_{1+s}(P_{X^nY^n}\|\pi_{XY}^n)\to 0$ is satisfied and the single-letterization of the conditional R\'enyi divergence in~\eqref{eqn:upper_Tpos} is tractable. Secondly,  any triple of sequences $(w^n,x^n,y^n)$ generated from the truncated product distribution  has marginal types $T_{w^n,x^n}$ and $T_{w^n,y^n}$ that are close to $Q_{WX}$ and $Q_{WY}$ respectively, so necessary approximations of types by distributions can be done to yield that the right-hand side of \eqref{eqn:upper_Tpos} is not larger than $\oGamma_{1+s}(\pi_{XY})$,  which in turn implies that $\tilT_{1+s}(\pi_{XY})\le\oGamma_{1+s}(\pi_{XY})$ for $s\in (0,1]$. 

Truncated product distributions will also be used extensively in the next section on exact common information; see Section~\ref{sec:type_over}.
\subsection{Sketch of the Converse   of Theorem~\ref{thm:RCI}(c)}
Note that for $s\in (0,1]$, $T_{1+s}(\pi_{XY})\ge C_{\Wyner}(\pi_{XY})$ is obvious in view of the monotonically non-decreasing nature of $s\mapsto T_{1+s}$. Hence, we only have to show that $T_{1+s}(\pi_{XY})\ge\uGamma_{1+s}(\pi_{XY})$ for $s\in (0,1]$. This proceeds in a few steps and we highlight the key ideas.

First, we derive a lower bound for $T_{1+s}(\pi_{XY})$ in terms of a  multi-letter expression. This hinges on the following non-asymptotic converse lemma which is due to the present authors~\cite{YuTan2019_wiretap}. 
\begin{lemma}\label{lem:lb_renyi}
Let $M$ be a uniform random variable on the set $[L]$ and let $P_{X|M}$ be an arbitrary stochastic map, whence $P_{XM}(x,m) =L^{-1}P_{X|M}(x|m)$ for all $(x,m)\in\calX\times[L]$. Then for $s\in [0,\infty]$ and any distribution $\pi_X$, we have 
\begin{align}
D_{1+s}(P_X\|\pi_X) \ge\max\big\{ D_{1+s}(P_{MX}\|P_M\pi_X)-\log L, D_{1+s}(P_X\|\pi_X)\big\}.
\end{align}
\end{lemma}
By particularizing $\pi_X$, $P_{X|M}$, $P_M$, and $\log L$ to be $\pi_{XY}^n$, $P_{X^n|M_n}$, $P_{Y^n|M_n}$, $P_{M_n}$ and $nR$ respectively, Lemma~\ref{lem:lb_renyi} implies that 
\begin{equation}
T_{1+s}(\pi_{XY})\ge\inf\, \limsup_{n\to\infty}\,\frac{1}{n}D_{1+s}\big(P_{M_n X^n Y^n}\big\|P_{M_n}\pi_{XY}^n\big), \label{eqn:lb_T}
\end{equation}
where the infimum runs over all distributed source simulation codes  $\{(P_{X^n|M_n},  P_{Y^n|M_n})\}_{n\in\bbN}$ such that $\frac{1}{n}D_{1+s}(P_{X^n Y^n}\|\pi_{XY}^n)\to 0$.

Now, it is easy to check by elementary calculus (see, for example, \citet{shayevitz2011renyi} and~\citet{Anant2018}) that the R\'enyi divergence admits a variational representation of the form 
\begin{equation}
D_{1+s}(P\|Q)= \max_{\tilQ \in\calP(\calX)}  \frac{1}{s}\bigg\{\sum_{x \in\calX} \tilQ(x)\log \big(P^{1+s}(x)Q^{-s}(x)\big) + H(\tilQ) \bigg\}. \label{eqn:var_Renyi}
\end{equation}
By particularizing $P$, $Q$, and $\tilQ$   above to $P_{M_nX^nY^n}$, $P_{M_n}\pi_{XY}^n$, and $\tilQ=P_{M_n}Q_{X^nY^n|M_n}$ respectively, and performing some algebraic manipulations,~\eqref{eqn:lb_T} yields
\begin{align}
&\hspace{-.3in}T_{1+s}(\pi_{XY})\ge\inf\, \limsup_{n\to\infty} \, -\frac{1+s}{s}H(X^nY^n|M_n) \nn\\*
&\hspace{-.3in}\;\;+\frac{1}{s}\!\max_{ \substack{Q_{X^nY^n|M_n} \in \\ \calC(P_{X^n|M_n}, P_{Y^n|M_n})}} \!\bbE_Q\left[  \log\frac{1}{(\pi^n(X^n,Y^n))^s Q(X^n,Y^n|M_n) }\right], \label{eqn:multi-letter1}
\end{align}
where the infimum runs over the  same sequence of distributions under the same constraints as in~\eqref{eqn:lb_T}.

The multi-letter expression in~\eqref{eqn:multi-letter1} consists of two parts. The entropy term can be single-letterized using standard techniques in network information theory. In particular,  
\begin{align}
\!\frac{1}{n}H(X^n Y^n |M_n) & =  \frac{1}{n}\sum_{i=1}^n H(X_i|X^{i-1}M_n) \!+\!\frac{1}{n}\sum_{i=1}^n H(Y_i|Y^{i-1}M_n) \!\\
&=  H(X_J | X^{J-1} M_n J)+H(Y_J | Y^{J-1} M_n J),
\end{align}
where we introduced the random variable $J \sim \mathrm{Unif}[n] $ which is independent of $(M_n, X^n,Y^n)$.
 %; see \eqref{eqn:single_letter0}--\eqref{eqn:single_letter2} for an example. 
 The second term is more involved but the main ingredient  for simplifying it is  the chain rule for couplings (Lemma~\ref{lem:coupling}) which implies that for any function $f:\calP(\calX^n\times\calY^n)\to\bbR$, 
\begin{align}
&\hspace{-.2in}\max_{\substack{ Q_{X^nY^n|M_n} \in\calC(P_{X^n|M_n}, P_{Y^n|M_n})}}f(Q_{X^nY^n|M_n}) \nn\\*
&\hspace{-.2in}\quad\ge \max_{\substack{ Q_{X^nY^n|M_n} \in \\ \prod_{i=1}^n\calC(P_{X_i|X^{i-1}M_n}, P_{Y_i|Y^{i-1}M_n})}}f\bigg(\prod_{i=1}^n Q_{X_iY_i|X^{i-1}Y^{i-1}W} \bigg) . \label{eqn:consequence_chain}
\end{align} 
  Observe that for a fixed $m \in \calM_n$ and  $Q_{X^n Y^n |M_n=m}$, we have
\begin{align}
&\bbE_{Q_{X^n Y^n |M_n=m}}\left[  \log\frac{1}{(\pi^n(X^n,Y^n))^s Q(X^n,Y^n|M_n) }\right] \nn\\*
&=\sum_{i=1}^n \sum_{x_i, y_i} \sum_{x^{i-1}, y^{i-1}} Q(x^{i-1}, y^{i-1}|m)  Q(x_i, y_i | x^{i-1},y^{i-1}, m) \nn\\*
&\qquad\times \log\frac{1}{\pi(x_i,y_i)^s  Q(x_i, y_i | x^{i-1},y^{i-1}, m)} \\
&\ge \sum_{i=1}^n \!\min_{ \substack{\tilQ_{X^{i-1}Y^{i-1}|M_n}\in \\  \calC( P_{X^{i-1} | M_n}, P_{Y^{i-1} | M_n}) }}\! \sum_{x_i, y_i} \sum_{x^{i-1}, y^{i-1}}\tilQ_{X^{i-1}Y^{i-1}|M_n}(x^{i-1},y^{i-1} |m)  \nn\\*
&\;\;  \times Q(x_i, y_i | x^{i-1}, y^{i-1}, m)  \log\frac{1}{\pi(x_i,y_i)^s  Q(x_i, y_i | x^{i-1},y^{i-1}, m)} . \label{eqn:simplifyQ}
\end{align}
Now, the main idea from here onwards is to use the consequence of the chain rule for couplings in~\eqref{eqn:consequence_chain} and then to justify swapping the maximization in~\eqref{eqn:multi-letter1} and the minimization in~\eqref{eqn:simplifyQ}. Then we see that we will be  left with an inner maximization over couplings $Q_{XY|UVW} \in\calC(P_{X|UW}, P_{Y|VW})$ where $U := X^{J-1}$, $ V=Y^{J-1}$, $X:=X_J$, $Y:=Y_J$ and $W := (M_n, J)$. These ideas, together with   a few additional approximation arguments, gives rise to the maximization over couplings $Q_{XY} \in\calC(P_{X|W=w},P_{Y|W=w'})$  in the definition of $\rvH_s$ and minimization over couplings  $Q_{WW'} \in \calC(P_W, P_W)$ in~\eqref{eqn:Gamma_down}. This   completes our sketch of the proof of the lower bound (converse) $T_{1+s}(\pi_{XY})\ge\uGamma_{1+s}(\pi_{XY})$  for the case $s\in (0,1]$.
 
\chapter{Exact Common Information}
\label{ch:exact}
\newcommand{\Ex}{\mathrm{Ex}}
\newcommand{\overG}{\overline{G}}
In this section, we depart from two key assumptions that we employed in the previous sections on Wyner's and R\'enyi common information. These assumptions are that {\em fixed-length} codes are used and {\em approximate} generation of the target distribution $\pi_{XY}^n$ is desired. By {\em fixed-length}, we mean that the shared or common randomness $M_n$ in Fig.~\ref{fig:source_sim} takes on values in the set $\calM_n$, which contains no more than $2^{nR}$ elements. Equivalently the bit string that corresponds to $M_n$ has length no larger than $nR$. By {\em approximate} generation, we mean that we only demand that some notion of the ``discrepancy'' between $P_{X^nY^n}$ and $\pi_{XY}^n$ converges to zero as the length of the code grows. The metrics that govern  the discrepancy include the (normalized and unnormalized) R\'enyi divergence and the TV distance. 

In this section, we consider the distributed source simulation problem under the assumptions that the codes used are allowed to be {\em variable-length} and we demand that the reconstruction $P_{X^nY^n}$ of the target distribution $\pi_{XY}^n$ be {\em exact} for some blocklength  $n\in\bbN$; this formulation is due to \citet{KLE2014}.  These distinctions are  analogous to the problem of lossless source coding~\cite{Cov06} in which there are also two  formulations. The first, which mirrors our discussion in Section~\ref{ch:renyi}, is of {\em fixed-length lossless source coding} with approximate reconstruction of the source. \citet{Shannon48} showed using ideas from what is now known as the {\em asymptotic equipartition property} (AEP) \cite[Ch.~3]{Cov06} that the minimum rate of compression is the entropy of the source. The second formulation, analogous to the current section, is {\em variable-length source coding} in which each source symbol is allowed to be encoded to bit strings of varying lengths and the {\em minimal average codeword length} is sought  under the constraint of zero-error reconstruction. In this case, the asymptotic minimal average per symbol codeword length is also the entropy of the source. This can be achieved via a variety of schemes including the Shannon--Fano--Elias code~\cite[Ch.~5]{Cov06} or the Huffman code~\cite{Huffman}.

One of the key benefits of variable-length coding in  data compression is the ability to obtain {\em exact} reconstructions. In contrast, if we are constrained to use fixed-length codes for the distributed source simulation problem, then we would require  a much higher rate to obtain an exact reconstruction, namely, $ \min\{\log |\calX|,\log|\calY|\}$ in the worst case. 
%we would not be able to obtain an exact reconstruction unless the rate is at least $ \min\{\log |\calX|,\log|\calY|\}$ (in the worst case). 
Since the fundamental limits for the lossless and zero-error source compression problems are the same---i.e., the Shannon entropy $H(X)$---it is natural to wonder whether the same is true for the distributed source simulation problem. This was an open problem posed at the 2014 International Symposium on Information Theory by \citet{KLE2014}.

%because in the worst case, Wyner's CI can be also equal to min{log|X|,log|Y|}. 

In this section, we answer this question in the negative. The way we do so is to 
%is to first formally define the {\em exact common information} as the minimal rate of common randomness such that a target distribution $\pi_{XY}^n$ can be {\em exactly} generated via variable-length codes. Then, we 
show a surprising equivalence between the unnormalized R\'enyi common information of order $\infty$ and the  exact common information. This is done by relating both problems at the operational level. To wit, we show that if there exists a rate-$R$ exact common information code, this code can be suitably modified to be a rate-$R$ order-$\infty$ R\'enyi common information code and {\em  vice versa}.  Thus, the family of R\'enyi common information provides a {\em bridge} between Wyner's common information and the exact common information; see Fig.~\ref{fig:bridge}. We recall that the R\'enyi common information is monotonically non-decreasing in its order and as we have seen from Section~\ref{sec:renyi_dsbs} for the DSBS, it can be {\em strictly} increasing.  Thus, exact generation of a joint source requires strictly larger rate compared to approximate generation in general, answering the open problem posed by \citet{KLE2014}. We identify classes of sources for which the exact common information is equal to Wyner's common information and provide intuition for why no extra rate is needed for exact generation of these sources~\cite{Vellambi}.  We extend our discussion to sources with continuous alphabets and provide bounds on the exact common information for the bivariate Gaussian source.

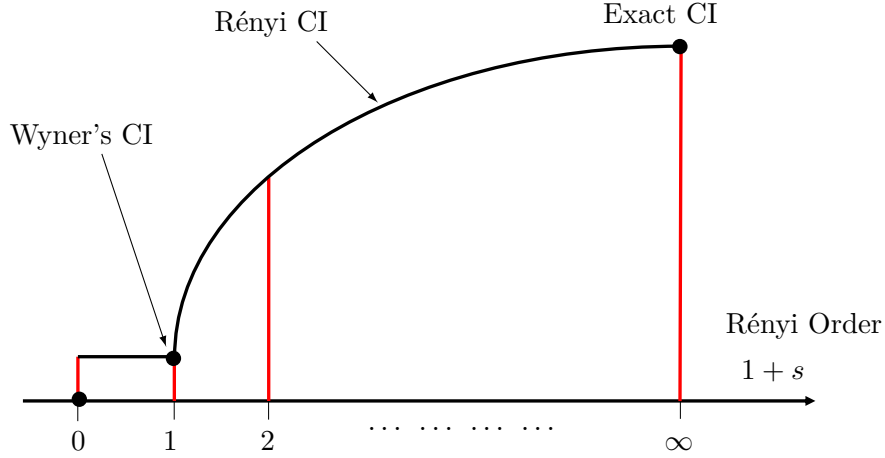
\begin{figure}[!ht]
\centering
\begin{overpic}[width=.9\textwidth]{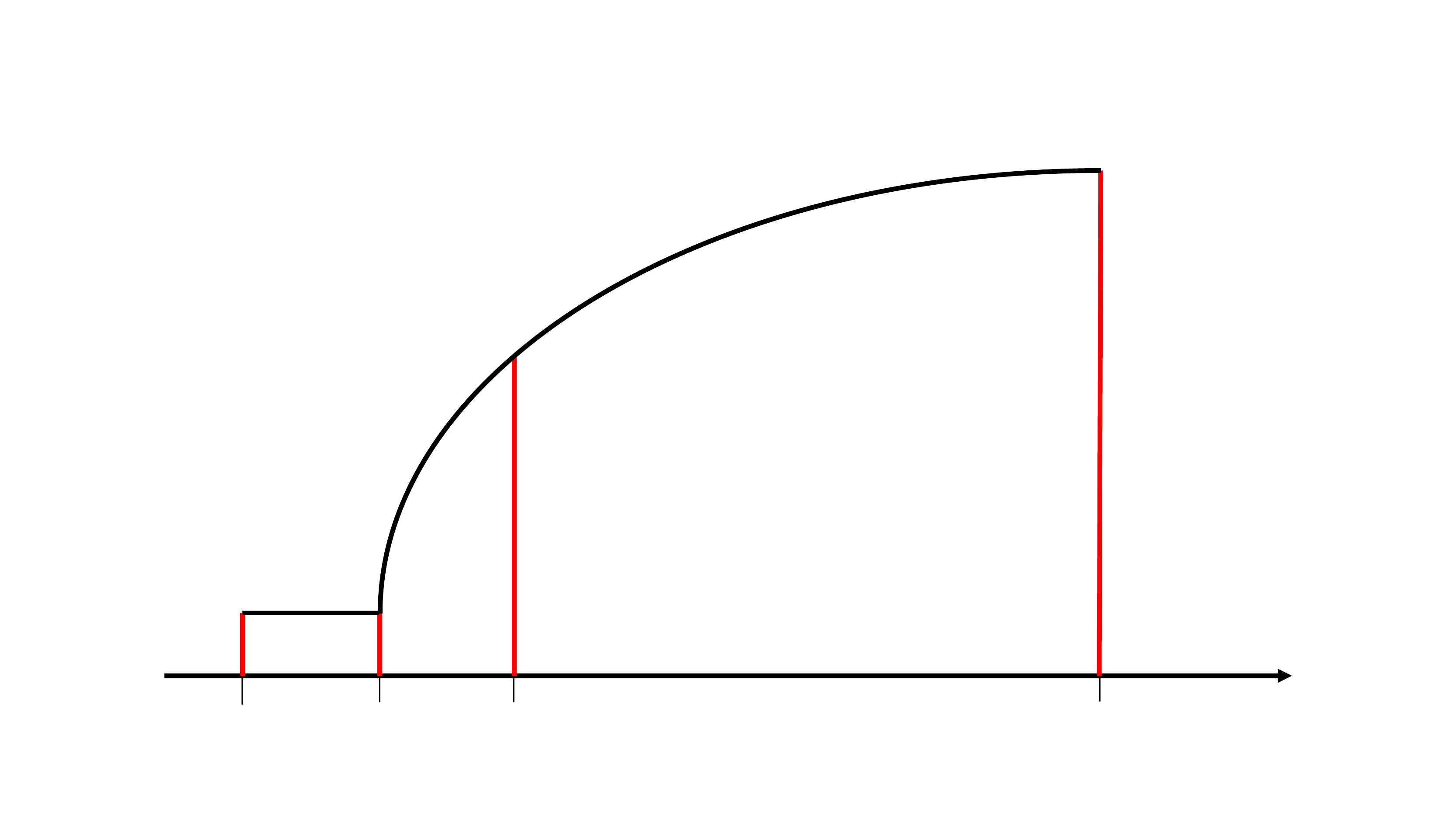}
\put(80,16.5){R\'enyi Order}
\put(81.5,12){$1+s$}
\put(74,5){$\infty$}
\put(50,7){$\ldots$}
\put(55,7){$\ldots$}
\put(45,7){$\ldots$}
\put(60,7){$\ldots$}
\put(34.5,5){$2$}
\put(25,5){$1$}
\put(16,5){$0$}
\put(68,46.8){Exact CI}
\put(10,35){Wyner's CI}
\put(19,34){\vector(1,-3){6.3}}
\put(26,14){\circle*{1.5}}
\put(16.8,10){\circle*{1.5}}
\put(40,45){\vector(1,-1){6}}
\put(30,46){R\'enyi CI}
\put(75.5,44.5){\circle*{1.5}}
\end{overpic}
\caption{A schematic showing that the R\'enyi common information provides a bridge between Wyner's common information and the exact common information}
\label{fig:bridge}
\end{figure}

%Introduce Kumar, Li and El Gamal's notion of exact common information. Show that the order-$1$ R\'enyi's CI reduces to Wyner's CI and the order $\infty$ R\'enyi's CI reduces to Exact CI. Thus showing that Exact CI exceeds Wyner's CI for some sources. This part is based on the authors' paper ``On Exact and $\infty$-R\'enyi Common Informations''

\section{Preliminary Definitions}
%In the formulation of the R\'enyi common information problem, fixed-length block codes and approximate generation of the target distribution $\pi_{XY}^n$ are assumed. In contrast, 
%In the exact common information problem that \citet{KLE2014} considered,
%variable-length codes and the exact generation of $\pi_{XY}^n$ for all $n$. 
%The objective  is also to find the fundamental limit on the minimum amount of common randomness satisfying this  stringent requirement. The amount here is quantified in terms of the per-letter expected codeword length,
%rather than the exponent of alphabet size (rate of the common random variable $M_n$) described in Chapter~\ref{ch:renyi}.

Define $\{0,1\}^* = \bigcup_{n\in\bbN} \{0,1\}^n$ to be the set of all finite-length binary strings. Denote the alphabet of the common random variable $W_n$ as the countable set $\calW_n$. We assume, without loss of generality, that $\calW_n \subset \bbN$.  Recall that a {\em prefix-free code} is a source code in which no codeword is a prefix of another. Consider any prefix-free code $f:\calW_n\to\{0,1\}^*$ which yields the codebook $\calC=\{f(w):w\in\calW_n\}$ whose elements are known as {\em codewords}. Then for each symbol $w\in\calW_n$ and the code $f$, let $\ell_f(w)$ be the {\em length} of the codeword~$f(w)$. 
\begin{example} \label{ex:length}
Let $\calW_n = \{1,2,3,4\}$. Consider the  prefix-free code
\begin{equation}
f(1) = 0,\quad f(2) = 10,\quad f(3) = 110,  \quad\mbox{and}\quad  f(4) = 111.
\end{equation}
This code has corresponding lengths 
\begin{equation}
\ell_f(1) = 1,\quad \ell_f(2) = 2,\quad \ell_f(3) = 3,  \quad\mbox{and}\quad  \ell_f(4)=3.
\end{equation}
\end{example}
\begin{definition} \label{def:length}
The {\em expected codeword length} of a code $f:\calW_n\to\{0,1\}^*$ for compressing the source $W_n\sim P_{W_n}$ is 
\begin{equation}
L_f(W_n) = \bbE[\ell_f(W_n)] =\sum_{w\in\calW_n}P_{W_n}(w)\ell_f(w).
\end{equation}
\end{definition}

\begin{definition} \label{def:vlcode}
An {\em $(n,R)$-variable-length distributed source simulation code} $(P_{W_n},f, P_{X^n | W_n},P_{Y^n | W_n})$ consists of 
\begin{itemize}
\item A distribution $P_{W_n}$ supported on a   countable set $\calW_n\subset\bbN$;
\item A prefix-free source code $f:\calW_n\to\{0,1\}^*$;
\item A pair of random   mappings  called {\em processors}  $P_{X^n | W_n} \in\calP( \calX^n |\calW_n)$ and  $P_{Y^n | W_n} \in\calP( \calY^n |\calW_n)$; 
\end{itemize}
such that the per-symbol expected codeword length 
\begin{equation}
\frac{L_f(W_n )}{n}\le R. \label{eqn:expected_length}
\end{equation}
\end{definition}
As usual, $n$ and $R$ are known as the {\em blocklength} and {\em rate} respectively. Observe that if $f$  in Definition~\ref{def:vlcode} is constrained to output codewords whose lengths do not exceed $nR$ and $W_n$ is constrained to be uniform on $\calW_n$, the expected length constraint in \eqref{eqn:expected_length} is automatically satisfied and the definition reverts to that for a  fixed-length distributed source simulation code (cf.\ Definition~\ref{def:wyner_code}).

Using a variable-length code, we assume that the common random variable $W_n$ is transmitted in an error-free manner to the two processors which then generate the {\em synthesized distribution} 
\begin{equation}
P_{X^nY^n}(x^n,y^n)=\sum_{w\in\calW_n}P_{W_n}(w ) P_{X^n|W_n}(x^n|w)P_{Y^n|W_n}(y^n|w). \label{eqn:syn_dist_exact}
\end{equation}
\begin{definition} \label{def:ECI}
The  {\em exact common information} $T_{\Ex}(\pi_{XY})$ between a pair of random variables $(X,Y)\sim \pi_{XY}$ is the infimum of all rates $R$ such that there exists an $(n,R)$-variable-length  distributed source simulation  code $(P_{W_n},f_n, P_{X^n | W_n},P_{Y^n | W_n})$ satisfying 
\begin{equation}
P_{X^nY^n}=\pi_{XY}^n  \quad\mbox{for some}~n\in\bbN, \label{eqn:exact}
\end{equation}
where $P_{X^nY^n}$ denotes the synthesized distribution in \eqref{eqn:syn_dist_exact}. 
\end{definition}
\begin{remark} Note that since we assume that $f_n$ is a {\em prefix-free} code, we can synthesize target distributions of arbitrarily long lengths by concatenating codewords $f_n(W_n)$ and decoding them {\em uniquely}. 
\end{remark} 
It is well known \cite[Sec.~5.4]{Cov06} that the minimal per-letter expected codeword length $L_{f_n}(W_n)$  for a prefix-free code $f_n$ satisfies 
\begin{equation}
H(W_n)\le L_{f_n}(W_n)< H(W_n)+1. \label{eqn:bounds_length}
\end{equation}
The lower bound follows from Kraft's inequality~\cite{Kraft} while the upper bound follows from Shannon's code assignment $\ell_{f_n}(w) = \left\lceil -\log   P_{W_n}(w)\right\rceil$. The bounds in \eqref{eqn:bounds_length} are colloquially known as {\em Shannon's zero-error compression theorem}.
%From \eqref{eqn:bounds_length}, we immediately have 
%\begin{equation}
%\lim_{n\to\infty}\frac{1}{n}\big( L_{f_n}(W_n)-H(W_n) \big)=0. \label{eqn:LH}
%\end{equation}

Define  the {\em common entropy} of the joint source $(X,Y)\sim\pi_{XY}$ as 
\begin{equation}
G(\pi_{XY}):=\min_{P_WP_{X|W}P_{Y|W}:P_{XY}=\pi_{XY}} H(W).\label{eqn:common_ent}
\end{equation} 
%The {\em joint common entropy} of $(X^n,Y^n)\sim\pi_{XY}^n$ is analogously defined as
%\begin{equation}
%G(\pi_{XY}^n):=\min_{P_{W_n}P_{X^n|W^n}P_{Y^n|W_n}:P_{X^nY^n}=\pi_{XY}^n} H(W_n).\label{eqn:common_ent_joint}
%\end{equation} 
It can be shown that $\lim_{n\to\infty} \frac{1}{n} G(\pi_{XY}^n) =\inf_{n\in\bbN}\frac{1}{n} G(\pi_{XY}^n)$ so we can define the  {\em common entropy rate}  of the source $(X, Y)\sim\pi_{XY}$ as
\begin{equation}
\overG(\pi_{XY}):= \lim_{n\to\infty} \frac{G(\pi_{XY}^n)}{n} =\inf_{n\in\bbN}\frac{G(\pi_{XY}^n)}{n}. \label{eqn:common_ent_rate}
\end{equation}
%The limit here exists due to Fekete's Lemma~\cite{Fekete} and the subadditivity of $G(\pi_{XY}^n)$. It is also easy to check that the limiting quantity in~\eqref{eqn:common_ent_rate} is $\inf_{n\in\bbN}G(\pi_{XY}^n)/n$.    
\citet{KLE2014} showed the following proposition. 
\begin{proposition} \label{prop:KLE}
% we see from~\eqref{eqn:LH} that the exact common information admits the following multi-letter characterization 
The exact common information 
\begin{align}
T_{\Ex}(\pi_{XY}) = \overG(\pi_{XY}). \label{eqn:multi-letter3}
\end{align}
\end{proposition}
Since the proof, due to \citet{KLE2014}, is brief and insightful, we reproduce it here. %his result can be shown as follows. 

\begin{proof}
For the achievability, fix any $R>\overG(\pi_{XY})$. From~\eqref{eqn:common_ent_rate}, we see that for sufficiently large $n$, $\overG(\pi_{XY})\ge\frac{1}{n} (G(\pi_{XY}^n)+1)$. By the upper bound in  Shannon's zero-error compression theorem in~\eqref{eqn:bounds_length}, we see that it is possible to exactly generate $(X^n,Y^n)\sim \pi_{XY}^n$ with rate at most $\frac{1}{n}(G(\pi_{XY}^n)+1)$. Hence, $R$ is achievable. 

For the converse part, assume $R$ is achievable. Then there exists a simulation code $(P_{W_n},f_n, P_{X^n | W_n},P_{Y^n | W_n}) $ with large enough blocklength $n$  that exactly generates $(X^n,Y^n)\sim\pi_{XY}^n$. Therefore, by the lower bound in  Shannon's zero-error compression theorem, $R\ge \frac{1}{n} G(\pi_{XY}^n)$ for some $n$. Thus, by \eqref{eqn:common_ent_rate}, $R\ge\overG(\pi_{XY})$.
\end{proof}

%where the second equality holds via Fekete's Lemma~\cite{Fekete} and the subadditivity of $G(\pi_{XY}^n)$. 
In view of Proposition~\ref{prop:KLE},  a variable-length synthesis code can be represented
by the triple $(P_{W_n},P_{X^n | W_n}, P_{Y^n | W_n} )$  and the dependence on the
prefix-free source code $f_n$ can be omitted. %The multi-letter expressions in~\eqref{eqn:multi-letter3} are due to \citet{KLE2014}. %Note that the one-shot version of  these multi-letter expressions   $$ is known as the {\em common entropy}.

A particularly important property of the common entropy rate is stated in the following lemma~\cite{KLE2014}.
\begin{lemma} \label{lem:common_ent}
The common entropy rate is an upper bound on Wyner's common information, i.e., 
\begin{equation}
\overG(\pi_{XY})\ge C_{\Wyner}(\pi_{XY}). \label{eqn:common_ent_Wyner}
\end{equation}
\end{lemma}
This result can be shown by first defining $W_n^*$ to be the common random variable achieving the  common entropy of the product source $G(\pi_{XY}^n)$. Then it follows that 
\begin{align}
\overG(\pi_{XY})&=\lim_{n\to\infty}\frac{1}{n}H(W_n^*)\\
&\ge \lim_{n\to\infty}\frac{1}{n}I(W_n^*;X^nY^n)\\
&\ge\lim_{n\to\infty}\frac{1}{n}\sum_{i=1}^n I(W_n^*;X_i Y_i) \label{eqn:XYiid}\\
&\ge \min_{ P_W P_{X|W} P_{Y|W}: P_{XY}=\pi_{XY}}I(W;XY) \\
&=C_{\Wyner}(\pi_{XY}), \label{eqn:exact_ge_wyner}
\end{align}
where~\eqref{eqn:XYiid} follows because the source $\{(X_i,Y_i)\}_{i=1}^\infty$ is memoryless.
The central question of this section is whether the inequality in \eqref{eqn:common_ent_Wyner} is strict for some sources $\pi_{XY}$.  We answer this in the affirmative in Section~\ref{sec:exact_dsbs}.

\section{Equivalence}% and Alternative Multi-Letter Expression}
We now establish a somewhat  surprising equivalence between the exact and unnormalized
R\'enyi common information of order $\infty$ and characterize them via  
an alternative multi-letter expression. 
%First recall the definition of the {\em maximal cross-entropy} in \eqref{eqn:maximal_cross_ent0} which we restate for the reader's convenience
%\begin{equation}
%  \rvH_{\infty}(P_X,P_Y\|\pi_{XY})=\max_{Q_{XY}\in\calC(P_X,P_Y)}\!\sum_{x,y}\! Q_{XY}(x,y)\log\frac{1}{\pi_{XY}(x,y)}.\!\! \label{eqn:maximal_cross_ent}
%\end{equation}
This equivalence was noticed by the present authors~\cite{YuTan2020_exact}.

\begin{theorem}[Equivalence of exact and $\infty$-R\'enyi common information]\label{thm:equivalence}
For a source with distribution $\pi_{XY}$ defined on a finite alphabet $\calX\times\calY$,
\begin{equation}
T_{\Ex}(\pi_{XY})=\tilT_{\infty}(\pi_{XY}). \label{eqn:Tequivalence}
\end{equation}
Furthermore, the quantities in \eqref{eqn:Tequivalence} are equal to  %the multi-letter quantity
\begin{equation}
\lim_{n\to\infty}\frac{\oGamma(\pi_{XY}^n)}{n}, \label{eqn:multiletter_equiv}
\end{equation}
where $\oGamma$ is the upper pseudo-common information of order $\infty$, i.e., %defined in \eqref{eqn:oGamma_infty}, i.e.,  
\begin{align}
\hspace{-.35in}\oGamma(\pi_{XY}):=\oGamma_\infty (\pi_{XY})& \stackrel{\eqref{eqn:oGamma_infty}}{=} \min_{ \substack{P_W P_{X|W} P_{Y|W} :\\ P_{XY}=\pi_{XY}}}-H(XY|W)\nn\\
&\qquad\quad + \bbE_{P_W}\big[ \rvH_{\infty}(P_{X|W },P_{Y|W }\|\pi_{XY}) \big]  \label{eqn:Gamma_fn}
\end{align}
and $\rvH_{\infty}$ is the maximal cross-entropy defined in \eqref{eqn:maximal_cross_ent0}.
\end{theorem}
Similarly to the common entropy, the function $\oGamma(\pi_{XY}^n)$ can also be shown to be subadditive; hence, the limit in \eqref{eqn:multiletter_equiv} exists due to Fekete's lemma~\cite{Fekete}.  We also remark that to compute the minimization in \eqref{eqn:Gamma_fn}, we can restrict the cardinality of $W$ to be no more than $|\calX||\calY|$.

Theorem~\ref{thm:equivalence} says that for any joint  source defined on a finite alphabet, the exact common information is equal to the unnormalized R\'enyi common information of order $\infty$. The former is defined in Definition~\ref{def:ECI}. The latter, on the other hand, involves the seemingly stringent condition    $D_\infty ( P_{ X^nY^n} \| \pi_{XY}^n)\to 0$ which is equivalent to
\begin{equation}
\max_{(x^n,y^n)\in \supp(P_{ X^nY^n})}\frac{P_{ X^nY^n}(x^n,y^n)}{\pi_{XY}^n(x^n,y^n)}=1+o(1). \label{eqn:Dinf}
\end{equation}
% Recall that the former is defined as the minimal rate of the common randomness $M_n$ such that
%\begin{equation
%D_\infty\big( P_{ X^nY^n}\big\| \pi_{XY}^n\big)=\log \max_{(x^n,y^n)\in \supp(\pi_{XY}^n)}\frac{P_{ X^nY^n}(x^n,y^n)}{\pi_{XY}^n(x^n,y^n)}\to 0
%\end{equation}
%using fixed-length distributed source simulation codes and $P_{X^nY^n}$ is the synthesized distribution defined in \eqref{eqn:syn_dist}. The latter is defined per Definition~\ref{def:ECI}. 
 This is surprising as two aspects of the definition have changed, yet they serendipitously resulted in common information quantities that coincide. The theorem also presents an alternative multi-letter expression for the exact common information in \eqref{eqn:multiletter_equiv}. This comes about due to the evaluation of $\tilT_{\infty}(\pi_{XY})$ instead of $T_{\Ex}(\pi_{XY})$ and is more useful than the common entropy rate in~\eqref{eqn:common_ent_rate} for the purposes of single-letterization. 
 
\subsection{Sketch of the Proof of $T_{\Ex}(\pi_{XY})=\tilT_{\infty}(\pi_{XY})$ } \label{sec:sketch_equiv}
Because the equality in~\eqref{eqn:Tequivalence} is particularly important, we sketch its proof in this subsection. For the impatient reader, this subsection can be omitted at a first reading.

We first show that $T_{\Ex}(\pi_{XY})\le\tilT_{\infty}(\pi_{XY})$. To so, we let $R$ be achievable rate for a fixed-length distributed source simulation code for which $D_\infty(P_{X^nY^n}\|\pi_{XY}^n)\to 0$. This means that for every $\epsilon>0$, for all sufficiently large $n$,  there exists a fixed-length simulation code $( P_{X^n |M_n},P_{Y^n |M_n})$ with rate $R$ such that  $D_\infty(P_{X^nY^n}\|\pi_{XY}^n)\le\epsilon$, where the synthesized distribution $P_{X^nY^n}$ is defined in \eqref{eqn:syn_dist}.  We show that $R$ is also achievable for exact reconstruction using a variable-length code.  The idea is to consider a ``mixing'' scheme in which  with high probability, we use the given fixed-length code, and with low probability, we use a completely lossless code. 

By the definition of $D_\infty$,  we have $P_{X^nY^n}(x^n,y^n)\le 2^\epsilon \pi_{XY}^n(x^n,y^n)$ for all $(x^n,y^n)\in\calX^n\times\calY^n$.
%\begin{equation}
%P_{X^nY^n}(x^n,y^n)\le 2^\epsilon \pi_{XY}^n(x^n,y^n)\qquad\forall\, (x^n,y^n)\in\calX^n\times\calY^n. 
%\end{equation}
Define 
\begin{equation}
P_{\hatX^n\hatY^n}(x^n,y^n):=\frac{2^\epsilon\pi_{XY}^n(x^n,y^n)-P_{X^nY^n}(x^n,y^n)}{2^\epsilon-1},
\end{equation}
which is a valid distribution (as it is non-negative and sums to one). 
Note now that    $\pi_{XY}^n$ is a {\em mixture distribution} that can be written as a convex combination  of $P_{X^nY^n}$ and $P_{\hatX^n\hatY^n}$ as follows
\begin{equation}
\pi_{XY}^n(x^n,y^n)=2^{-\epsilon}P_{X^nY^n}(x^n,y^n)+(1-2^{-\epsilon}) P_{\hatX^n\hatY^n}(x^n,y^n). \label{eqn:mix_dist}
\end{equation}
The variable-length code first generates a Bernoulli random variable $U\sim\mathrm{Bern}(2^{-\epsilon})$ which can be described by  $1$ bit. It then  transmits $U$ to the two processors. If $U = 1$,   the encoder also generates $M_n  \sim\mathrm{Unif}[2^{nR}]$ and uses the  given  fixed-length code $( P_{X^n |M_n},P_{Y^n |M_n})$  with rate $R$ to generate $P_{X^nY^n}$. Otherwise (if $U=0$), the encoder generates $(\hatX^n,\hatY^n)\sim P_{\hatX^n\hatY^n}$ using $\log (|\calX||\calY|)$ bits per source symbol. By the law of total probability, the distribution generated is {\em exactly} $\pi_{XY}^n$ in~\eqref{eqn:mix_dist}. Since $\Pr(U=1)=2^{-\epsilon}$,  the average codeword length   required is 
\begin{equation}
%\!\frac{1}{n}+\! P_U(1)R+ P_U(0)\log (|\calX||\calY|)=
\frac{1}{n}+ 2^{-\epsilon} R+(1 - 2^{-\epsilon})\log (|\calX||\calY|). \label{eqn:tot_rate}
\end{equation}
 Taking $n\to\infty$ and then $\epsilon\downarrow 0$ yields the conclusion  that $R$ is an achievable rate for the exact synthesis of $\pi_{XY}^n$.  Because we use a mixture distribution in \eqref{eqn:mix_dist}, this technique is known as the {\em mixture decomposition technique} and will also be used in Section~\ref{ch:ecs}.
%a uniform random variable $M_n \sim\mathrm{Unif}[2^{nR}]$, and the
%encoder and two generators use the fixed-length synthesis
%codes with rate $R$ to generate PXnY n. If U = 0, then
%the encoder generates (Xn, Y n)   PXnY n, and uses a
%variable-length compression code with rate  log |X ||Y| to
%generate

We next argue that $\tilT_{\infty}(\pi_{XY})\le T_{\Ex}(\pi_{XY})$. For this purpose, assume that there exists a $(k,R)$-variable-length distributed source simulation code  $(P_{W_k} ,P_{X^k | W_k},P_{Y^k | W_k})$ that exactly generates $\pi_{XY}^k$, i.e., 
\begin{equation}
\pi_{XY}^k(x^k,y^k) =\sum_{w_k} P_{W_k}(w_k)P_{X^k | W_k}(x^k|w_k) P_{Y^k | W_k}(y^k|w_k) . \label{eqn:pi_k}
\end{equation}
%We know that the infimum $R^*$ of all such $R$ satis
%\begin{equation}
%\lim_{k\to\infty}\frac{H(W_k)}{k}= R. \label{eqn:limitH}
%\end{equation}
For every $\epsilon>0$, there exists $k\in\bbN$ such that $R$ can come arbitrarily close to $\frac{1}{k} H(W_k)$. In particular, we can assume  
\begin{equation}
R\le \frac{H(W_k)}{k}(1+2\epsilon). 
\end{equation}

\begin{figure}
\centering
\setlength{\unitlength}{0.05cm}{ %\scalebox{2}
\begin{picture}(160,26) %\linethickness{1pt}
\put(0,00){\line(1,0){160}} 
\put(0,15){\line(1,0){160}}
\put(0,0){\line(0,1){15}}
\put(160,0){\line(0,1){15}}

\put(20,0){\line(0,1){15}}
\put(40,0){\line(0,1){15}}
\put(140,0){\line(0,1){15}}
\put(120,0){\line(0,1){15}}

\put(0,20){\vector(1,0){160}}
%\put(40,20){\vector(1,0){40}}
%\put(80,20){\vector(1,0){40}}
%\put(120,20){\vector(1,0){40}}

\put(160,20){\vector(-1,0){160}}
%\put(120,20){\vector(-1,0){40}}
%\put(80,20){\vector(-1,0){40}}
%\put(40,20){\vector(-1,0){40}}

\put(9,5.5){$k$}
\put(29,5.5){$k$}
\put(129,5.5){$k$}
\put(149,5.5){$k$}

\put(66,5.5){$\ldots$}
\put(73,5.5){$\ldots$}
\put(80,5.5){$\ldots$}
\put(87,5.5){$\ldots$}
\put(77,23){$nk$}
\end{picture}}
\caption{Code construction for the proof that $\tilT_{\infty}(\pi_{XY})\le T_{\Ex}(\pi_{XY})$.}
\label{fig:supercode}
\end{figure}
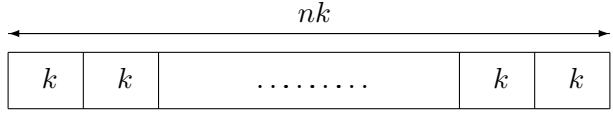

%Let $\epsilon>0$ and $k$ be so large such that 
Using the above variable-length code,  we construct a fixed-length super-code which is the concatenation of $n$ independent length-$k$ blocks  with rate 
\begin{equation}
R':=\frac{H(W_k)}{k} (1+\epsilon). \label{eqn:rate_super}
\end{equation}
 See Fig.~\ref{fig:supercode}.  We will now verify that this super-code has the desired property in~\eqref{eqn:Dinf}. The common random variable $W_k^n = (W_{k1},\ldots, W_{kn})\sim P_{W_k}^n$ and $(P_{X_k|W_k}^n,  P_{Y_k|W_k}^n)$ is the pair of processors. The main idea of the proof is to suitably ``shape'' a uniform random variable $M_n$ into the non-uniform $W_k^n$ so that the variable-length code $(P_{W_k}^n, P_{X_k|W_k}^n,  P_{Y_k|W_k}^n)$ can be used subsequently.

%For a fixed but sufficiently large  $k$, 
We now  design a function $f$ such that given a uniform random variable $M_n\sim\mathrm{Unif}(\calM_n)$ where $\calM_n=[2^{nkR'}]$, the function $f$ applied to $M_n$ simulates $P_{W_k}^n$ in the sense that the R\'enyi divergence of order  $\infty$  from $P_{f(M_n)}$ to $P_{W_k}^n$ vanishes. This function  is constructed as follows~\cite[Theorem~7]{YuTan_sim}. It maps multiple elements of $\calM_n$ to each sequence in the weakly typical set $\calA_\epsilon^{(n)}(P_{W_k})$. For each sequence $w_k^n$, we control the number of elements that are mapped to $w_k^n$ to be directly proportional to $P_{W_k}^n (w_k^n)$.
 By the asymptotic equipartition property~\cite{Cov06}, $W_k^n\sim P_{W_k}^n$ is distributed almost uniformly on $\calA_\epsilon^{(n)}(P_{W_k})$. According to the theory of (R\'enyi) source resolvability~\cite{HV93,Ste96, YuTan_sim}, since $\frac{1}{n}\log|\calM_n|= kR'\ge (1+\epsilon) H(W_k)$ (cf.~\eqref{eqn:rate_super}), 
%for any $\epsilon>0$, if $\frac{1}{n}\log|\calM_n|> H(W_k) +\epsilon$ or  equivalently, $R' >\frac{1}{k}(H(W_k) +\epsilon)$,   %(source resolvability~\cite{HV93,Ste96}), % (which is possible in view of~\eqref{eqn:limitH}), 
\begin{equation}
\lim_{n\to\infty}D_\infty\big( P_{f(M_n)}\big\| P_{W_k}^n \big)=0.
\end{equation}
This essentially follows because the $P_{W_k}^n$-probability of the weakly typical set $\calA_\epsilon^{(n)}(P_{W_k})$ converges to  one exponentially fast as $n\to\infty$. 
%For the fixed $k$, define the {\em truncated i.i.d.\ distribution}
%\begin{equation}
%Q_{W_k^n}(w_k^n) = \frac{P_{W_k}^n(w_k^n ) \mathbbm{1}\big\{ w_k^n \in \calA_\epsilon^{(n)}(P_{W_k}) \big\}}{ P_{W_k}^n \big( \calA_\epsilon^{(n)}(P_{W_k}) \big)}.
%\end{equation}
%We note the following two important properties 
%\begin{equation}
%P_{W_k}^n \big(  \calA_\epsilon^{(n)}(P_{W_k}) \big) \to 1\quad \mbox{and}\quad \big|\calA_\epsilon^{(n)}(P_{W_k})\big|\le 2^{n (H(W_k)+\epsilon)}. \label{eqn:two_props}
%\end{equation}
Now, we consider the concatenation scheme in Fig.~\ref{fig:concat_codes}. From~\eqref{eqn:pi_k} and the constructed fixed-length code,  we have
\begin{align}
P_{W_k}^n &\rightarrow  P_{X^k|W_k}^n P_{Y^k|W_k}^n \rightarrow\pi_{XY}^{kn} \quad\mbox{and}\\*
P_{f(M_n)}&\rightarrow  P_{X^k|W_k}^n P_{Y^k|W_k}^n \rightarrow P_{X^{kn}  Y^{kn}} ,
\end{align}
where $P_X\rightarrow V_{Y|X}\rightarrow P_Y$ means that $P_Y $ is the induced output distribution when the input distribution is $P_X$ and the stochastic kernel (channel) is $V_{Y|X}$. 
%\begin{align}
%\pi_{XY}^{kn} (x^{kn},y^{kn}) &=\sum_{w_k^n}P_{W_k}^n(w_k^n) P_{X^k|W_k}^n(x^{kn}|w_k^n) P_{Y^k|W_k}^n(y^{kn}|w_k^n)\\
% P_{X^{kn}, Y^{kn}}(x^{kn},y^{kn}) &=\sum_{w_k^n}P_{f(M_n)}(w_k^n) P_{X^k|W_k}^n(x^{kn}|w_k^n) P_{Y^k|W_k}^n(y^{kn}|w_k^n).
%\end{align}
Thus, by the data-processing inequality for the R\'enyi divergence,
\begin{equation}
D_\infty\big(P_{X^{kn}Y^{kn}}\big\|\pi_{XY}^{kn}\big)\le D_\infty\big( P_{f(M_n)}\big\| P_{W_k}^n \big)\to0\quad\mbox{as }n\to\infty.
\end{equation}
This concludes the proof that the R\'enyi divergence of order $\infty$ converges to zero along blocklengths  $n$ that are integers multiples of $k$. For other $n$'s, a standard approximation argument suffices. Thus, $R$ is an achievable rate for the   approximate synthesis problem under the R\'enyi divergence of order $\infty$, which in turn implies $R\ge \tilde{T}_\infty (\pi_{XY})$. %  R\'enyi common information of order $\infty$  approximate synthesis problem. 

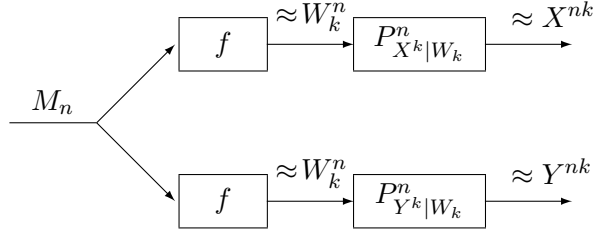
\begin{figure}[t]
\centering \setlength{\unitlength}{0.058cm} { \begin{picture}(140,60)
%\linethickness{1pt}
\put(5,30){\line(1,0){20}} \put(25,30){\vector(1,1){18}}
\put(25,30){\vector(1,-1){18}} \put(44,42){\framebox(20,12){$f$}}
\put(44,6){\framebox(20,12){$f$}} \put(64,48){\vector(1,0){20}}
\put(64,12){\vector(1,0){20}} \put(66,53){%
\mbox{%
$\approx\! W_k^{n}$%
}} \put(66,17){%
\mbox{%
$\approx\! W_k^{n}$%
}} \put(84,42){\framebox(30,12){$P_{X^{k}|W_k}^{n}$}} \put(84,6){\framebox(30,12){$ P_{Y^{k}|W_k}^{n}$}}
\put(114,48){\vector(1,0){20}} \put(114,12){\vector(1,0){20}}
\put(10,33){%
\mbox{%
$M_n$%
}} \put(120,52){%
\mbox{%
$\approx X^{nk}$%
}} \put(120,16){%
\mbox{%
$\approx Y^{nk}$%
}} \end{picture}}
\caption{Concatenation scheme for the proof that $\tilT_{\infty}(\pi_{XY})\le T_{\Ex}(\pi_{XY})$.}
\label{fig:concat_codes} 
\end{figure}

\section{Single-Letter Bounds for Exact Common Information}
In anticipation of evaluating $T_{\Ex}(\pi_{XY})=  \tilT_{\infty}(\pi_{XY})$ for various sources $\pi_{XY}$, we provide single-letter bounds on these common information 
quantities. Recall the definition of the  upper  pseudo-common information of order $\infty$, namely $\oGamma(\pi_{XY}) =\oGamma_\infty(\pi_{XY})$  in~\eqref{eqn:Gamma_fn}. Additionally, we rename the lower pseudo-common information of order $\infty$  as
\begin{align}
\hspace{-.2in}\uGamma(\pi_{XY}) & :=\uGamma_\infty(\pi_{XY}) \nn\\*
\hspace{-.2in} &=  \inf_{ \substack{P_W P_{X|W} P_{Y|W} :\\ P_{XY}=\pi_{XY}}}-H(XY|W)\nn\\
\hspace{-.2in}&\hspace{-.2in}\qquad\qquad+ \inf_{\substack{Q_{WW'}\\ \in \calC(P_W,P_W)}} \bbE_{Q_{WW'}}\big[ \rvH_{\infty}(P_{X|W},P_{Y|W'}\|\pi_{XY})\big], \label{eqn:uGamma_inf}
\end{align}
where the last equality is the same as \eqref{eqn:uGamma_infty}.
Note that the only difference between $\uGamma(\pi_{XY})$ and $\oGamma(\pi_{XY})$ is the inner sum; the former is an optimization over all couplings $Q_{WW'}\in\calC(P_W,P_W)$ while latter replaces this optimization with $P_W$. 

%Thus, if the inner optimization in \eqref{eqn:uGamma_inf} is achieved at the equality coupling $Q_{WW'}(w,w')=P_W(w)\bone\{w=w'\}$, we have the favourable scenario in which $\uGamma(\pi_{XY})=\oGamma(\pi_{XY})$. 

We are now ready to state single-letter bounds on the (unnormalized) R\'enyi common information of order $\infty$ which is the same as the exact common information (cf.\ Theorem~\ref{thm:equivalence}). 
\begin{theorem}[Bounds on exact common information] \label{thm:exact_sl}
For a source with distribution $\pi_{XY}$ defined on a finite alphabet $\calX\times\calY$,
\begin{align}
\max\big\{ \uGamma(\pi_{XY}),C_{\Wyner} (\pi_{XY}) \big\}&\le T_\infty(\pi_{XY})\le\tilT_\infty  (\pi_{XY}) \\
&=T_{\Ex}(\pi_{XY})\le\oGamma(\pi_{XY}). \label{eqn:T_Ex_ach}
\end{align}
\end{theorem}
We note that this theorem is just a combination of Theorem~\ref{thm:RCI}(c) (concerning bounds on the  R\'enyi common information of orders in $ (1,2]\cup\{\infty\}$) and Theorem~\ref{thm:equivalence}.% (on the equivalence of $T_\infty(\pi_{XY})$ and $T_{\Ex}(\pi_{XY})$). 
 
 \subsection{Coding Scheme and Type Overflow Phenomenon} \label{sec:type_over}
We now   comment on the coding scheme used to achieve the upper bound $\tilT_\infty(\pi_{XY})\le\oGamma(\pi_{XY})$ in~\eqref{eqn:T_Ex_ach}. It shares many similarities to the achievability of the R\'enyi common information for orders in $(1,2]$ as outlined in Section~\ref{sec:ach_RCIc}. We use   truncated product distributions. Sequences generated from these distributions are useful in upper bounding $T_\infty$ and Wyner's common information. This is because under both scenarios, $X^n -W_n - Y^n$ forms a Markov chain.  Hence given $W_n = w$, the support of $P_{X^n|W_n}(\cdot|w)P_{Y^n|W_n}(\cdot|w)$ is a {\em product set}, i.e., $\calA \times \calB$ where $\calA\subset\calX^n$ and $\calB\subset\calY^n$. Thus the support of $P_{X^nY^n}$ is the {\em union of product sets}. 
This union consists of not only the jointly typical set $\calT_\epsilon^{(n)}(P_{XY})$ but also {\em other} joint type classes. This is what we term as the {\em type overflow phenomenon}. See Fig.~\ref{fig:type_overflow} for a schematic. Designing a  synthesis code that achieves Wyner's common information (under the relative entropy measure) only requires sequences
in the jointly typical set $\calT_\epsilon^{(n)}(P_{XY})$ to be well-simulated. However, $\infty$-R\'enyi approximate
synthesis requires {\em all} the sequences in the support
of $P_{X^nY^n}$ to be well-simulated; see~\eqref{eqn:Dinf}.  Hence, the type overflow phenomenon does not affect Wyner's synthesis asymptotically, but plays a critical role in determining the optimal rate for $\infty$-R\'enyi approximate synthesis (or equivalently, exact synthesis). Truncated i.i.d.\ coding turns out to be  a convenient approach to control all possible types of the output sequence of a code  to mitigate the effects of type overflow.

\begin{figure}[t]
%\vspace{-.1in}
\centering
\begin{overpic}[width=1.02\textwidth]{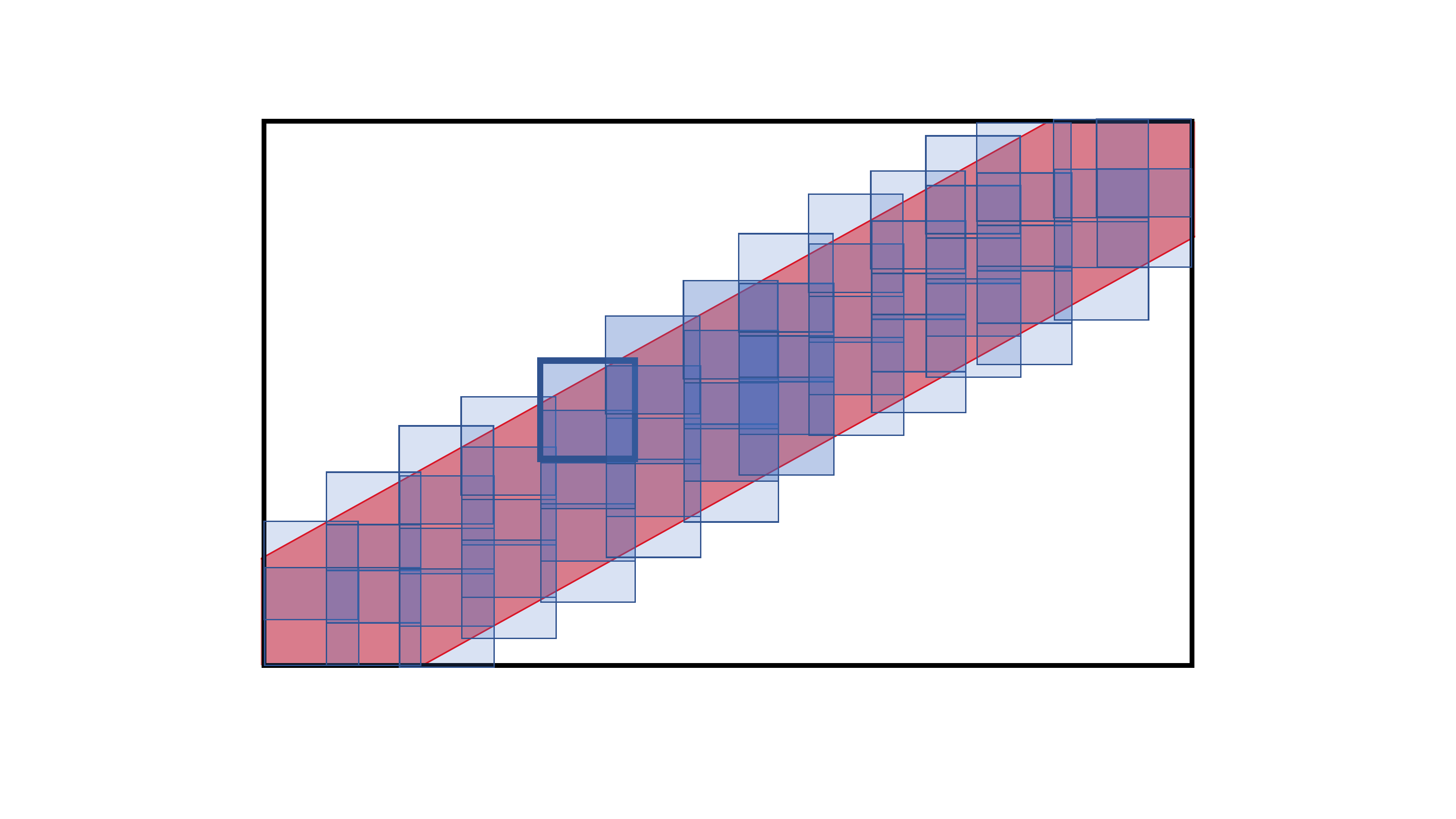}
\put(45,5){$\calT_\epsilon^{(n)}(P_X)$}
\put(4,27){$\calT_\epsilon^{(n)}(P_Y)$}
\put(63,18){$\calT_\epsilon^{(n)}(P_{XY})$}
\put(70,21){\vector(-1,1){13}}
\put(20,51){$\calT_\epsilon^{(n)}(P_{XW}|w^n )\times \calT_\epsilon^{(n)}(P_{YW}|w^n)$}
\put(38,50){\vector(0,-1){18}}
\put(60,14){\vector(-1,1){8}}
\put(50,12){Type overflow}
\end{overpic}
\vspace{-.28in}
\caption{Illustration of the type overflow phenomenon.  The rectangle represents the Cartesian product of the marginal typical sets $\calT_\epsilon^{(n)}(P_X)\times\calT_\epsilon^{(n)}(P_Y)$. The jointly typical set is the  shaded diagonal area $\calT_\epsilon^{(n)}(P_{XY})$. Each small square  is the Cartesian product of conditionally typical sets $\calT_\epsilon^{(n)}(P_{XW}|w^n)\times \calT_\epsilon^{(n)}(P_{YW}|w^n)$ indexed by the codewords $w^n$ in $\calC=\{w^n(m):m\in[2^{nR}]\}$. As can be seem from this schematic, the union of   small squares can be a strict superset of the jointly typical set.  }
\label{fig:type_overflow}
\end{figure}
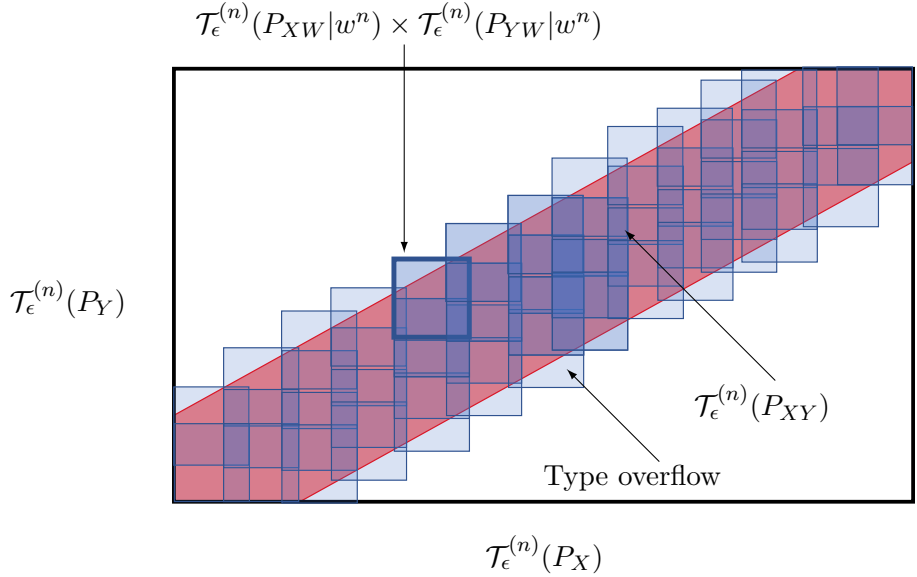

\subsection{Intuition for  the upper bound (Achievability)} \label{sec:intuition_ub}
Let us provide some intuition for the upper bound in~\eqref{eqn:T_Ex_ach}. Exact synthesis requires that $P_{X^nY^n}$  multiplicatively  approximates $\pi_{XY}^n$ pointwise for all $(x^n,y^n)\in\supp(P_{ X^nY^n})$; see \eqref{eqn:Dinf}. By using the  truncated i.i.d.\ coding technique, we can essentially restrict our attention to random variables $(W^n,X^n)\in\calT_\epsilon^{(n)}(P_{WX})$ and $(W^n,Y^n)\in\calT_\epsilon^{(n)}(P_{WY})$. Let  $M_n \in\calM_n$ be the common randomness   for approximate synthesis based on the R\'enyi divergence of order $\infty$. Then, for sufficiently large~$n$, 
\begin{align}
&\hspace{-.2in}P_{X^nY^n}(x^n,y^n) \nn\\*
&\hspace{-.2in}\quad\approx \frac{1}{|\calM_n|}\sum_{m\in\calM_n} P_{X^n|W^n}(x^n|w^n(m)) P_{Y^n|W^n}(y^n|w^n(m))\\
&\hspace{-.2in}\quad\approx \exp( -nR) N(x^n,y^n)\exp\Big(-n\big(H(X|W)+H(Y|W)\big)\Big), \label{eqn:type_counting}
\end{align}
where $N(x^n,y^n) $ is the number of $w^n(m)$ sequences in the codebook~$\calC$ that are jointly typical with $x^n$ and jointly typical with $y^n$ (individually).  On the other hand, by a similar intuition for the maximal cross-entropy in \eqref{eqn:intuition_maximal}, we have 
\begin{align}
& \min_{(x^n,y^n) \in\supp(P_{X^nY^n})} \pi_{XY}^n\left(  x^n,y^n \right)\nn\\*
&\quad\approx\min_{ \substack{(w^n,x^n,y^n) :\\ T_{w^nx^n}\approx P_{WX}, \, T_{w^ny^n}\approx P_{WY}}}\pi_{XY}^n\left(  x^n,y^n \right)\\*
&\quad\approx \exp\Big(-n \bbE_W\big[\rvH_{\infty} (P_{X|W}, P_{Y|W}\|\pi_{XY} )\big]\Big). \label{eqn:pi_intuition}
\end{align}
Since $N(x^n,y^n)\ge 1$ for $(x^n,y^n)\in\supp(P_{X^nY^n})$  and  $H(X|W)+H(Y|W)=H(XY|W)$, combining~\eqref{eqn:Dinf}, \eqref{eqn:type_counting} and~\eqref{eqn:pi_intuition} yields that  any rate $R$ satisfying 
\begin{equation}
R\ge-H(XY|W)+ \bbE_W\big[\rvH_{\infty} (P_{X|W}, P_{Y|W}\|\pi_{XY} )\big] \label{eqn:lb_R}
\end{equation}
is achievable.
Taking the minimum of the right-hand side over all joint distributions $P_WP_{X|W}P_{Y|W}$ such that $P_{XY}=\pi_{XY}$ and noticing that the resultant expression is  $\oGamma(\pi_{XY})$ (defined in~\eqref{eqn:Gamma_fn})  completes the proof that   $\tilT_\infty(\pi_{XY})\le\oGamma(\pi_{XY})$.

\section{Equality of Exact and Wyner's Common Information} \label{sec:equality}
As we have seen from Section~\ref{sec:renyi_dsbs}, the R\'enyi common information for orders   larger than $1$ can be strictly larger than Wyner's common information. We now discuss various  conditions under which Wyner's common information $C_\Wyner(\pi_{XY})$ is equal to the exact common information $T_{\Ex}(\pi_{XY})$.  Under these conditions, in view of the monotonicity of $\tilT_{1+s}(\pi_{XY})$ for $s\ge-1$, the entire family of R\'enyi common information for all positive orders is equal to $C_\Wyner(\pi_{XY})$. % In view of Theorem~\ref{thm:exact_sl},  this is so if $\oGamma(\pi_{XY})=C_{\Wyner}(\pi_{XY})$. 

\begin{theorem} \label{thm:exact_eq_wyn}
For every Wyner-product distribution  $\pi_{XY} \in\calP(\calX\times\calY)$ (see Definition~\ref{def:pseudo_pdt}), 
\begin{equation}
T_{\Ex}(\pi_{XY})=C_{\Wyner}(\pi_{XY}). \label{eqn:exact_eq_wyn}
\end{equation}
\end{theorem}
This theorem, due to the present authors~\cite{YuTan2020_exact}, follows easily by combining Lemma~\ref{lem:prop_Gamma}(e) and Theorem~\ref{thm:equivalence}. The former for the case $s=\infty$ is restated here for ease of reference. %the case $s=\infty$.
\begin{lemma} \label{lem:oGamma_CWy}
The equality  $\oGamma(\pi_{XY})=C_{\Wyner}(\pi_{XY})$  holds if and only if $\pi_{XY}$ is a  Wyner-product distribution. 
\end{lemma}
 Since every pseudo-product distribution is a  Wyner-product distribution (cf.\ Fig.~\ref{fig:products}), the equality in~\eqref{eqn:exact_eq_wyn} also applies to pseudo-product distributions. The fact that pseudo-product distributions result  in the equality $T_{\Ex}(\pi_{XY})=C_{\Wyner}(\pi_{XY})$ was also realized by  \citet{Vellambi2018}, albeit via a different consideration. 
%Even though Theorem~\ref{thm:exact_eq_wyn} provides a {\em sufficient} condition for equality of $T_{\Ex}$ and $C_{\Wyner}$,  Lemma~\ref{lem:prop_Gamma}(e) says that if $\oGamma(\pi_{XY})$ is tight for the exact common information $T_{\Ex}(\pi_{XY})$, then the condition that  $\pi_{XY}$ is a Wyner-product distribution is also {\em necessary} for  \eqref{eqn:exact_eq_wyn} to hold. 

We now provide a brief justification  of Lemma~\ref{lem:oGamma_CWy}.
\begin{proof}[Proof Sketch of Lemma~\ref{lem:oGamma_CWy}]
If $\pi_{XY}$ is a Wyner-product distribution, by the second part of Lemma~\ref{lem:prop_max_cross_ent}, 
\begin{equation}
\rvH_{\infty}(P_{X|W=w},P_{Y|W=w}\|\pi_{XY})=\sum_{x,y}P(x|w)P(y|w)\log\frac{ 1}{\pi(x,y)}, \label{eqn:Hinfty_w}
\end{equation}
where $P_W P_{X|W} P_{Y|W}$ is a joint distribution that attains  the infimum in $C_\Wyner(\pi_{XY})$.
Taking the expectation with respect to $P_W$, and noticing that $P_{XY}=\pi_{XY}$, we obtain
\begin{equation}
\bbE\big[\rvH_{\infty}(P_{X|W},P_{Y|W }\|\pi_{XY}) \big]=H(XY). \label{eqn:Hinf_eq}
\end{equation}
Substituting $P_W P_{X|W} P_{Y|W}$ into the definition of $\oGamma(\pi_{XY})$ in~\eqref{eqn:Gamma_fn}, we obtain $\oGamma(\pi_{XY})\le C_\Wyner(\pi_{XY})$. Obviously (see Theorem~\ref{thm:exact_sl}), the reverse inequality holds and so $\oGamma(\pi_{XY}) = C_\Wyner(\pi_{XY})$.

Now suppose that $\oGamma(\pi_{XY}) = C_\Wyner(\pi_{XY})$. Let $P_W P_{X|W} P_{Y|W}$  attain the infimum in the upper pseudo-common information of order $\infty$, namely $\oGamma(\pi_{XY})$. Then for every $w\in\supp(P_W)$, $\supp(P_{X|W=w})\times \supp(P_{Y|W=w})\subset\supp(\pi_{XY})$. Otherwise, $\bbE [\rvH_{\infty}(P_{X|W },P_{Y|W }\|\pi_{XY}) ]=\infty$, contradicting the optimality of $P_W P_{X|W} P_{Y|W}$. At the same time,
\begin{align}
\uGamma(\pi_{XY}) &= -H(XY|W)+ \bbE\big[\rvH_{\infty}(P_{X|W},P_{Y|W}\|\pi_{XY}) \big] \label{eqn:uGamma_bd1} \\
&\ge -H(XY|W)+H(XY)\ge C_{\Wyner}(\pi_{XY}), \label{eqn:uGamma_bd2}
\end{align}
where the first inequality follows from \eqref{eqn:prop_max_cross_ent2}. Thus, all inequalities above  are equalities. In particular,  $P_W P_{X|W} P_{Y|W}$ also attains the infimum in $C_{\Wyner}(\pi_{XY})$ and~\eqref{eqn:Hinf_eq} holds. This implies that \eqref{eqn:Hinfty_w} holds for all $w\in\supp(P_W)$. By the second part of Lemma~\ref{lem:prop_max_cross_ent}, for all $w\in\supp(P_W)$, $\pi_{XY}$ is a product distribution on $\calA_w = \supp(P_{X|W=w})\times\supp( P_{Y|W=w})$. Hence $\pi_{XY}$ is a Wyner-product distribution.
\end{proof} 

 We now state a couple of other easy-to-verify sufficient conditions for Wyner's common information to be equal to the exact common information. These conditions are  due to \citet{Vellambi}. % the condition that $\pi_{XY}$ is a Wyner-product distribution (as in Theorem~\ref{thm:exact_eq_wyn}) subsumes the following sufficient conditions.

\begin{corollary} \label{cor:equiv}
Let $\pi_{XY}$ be a distribution defined on a finite alphabet. Let $P_WP_{X|W}P_{Y|W}$ achieve the infimum in $C_{\Wyner}(\pi_{XY})$. If either 
\begin{align}
H(W|XY) &=0 \quad\mbox{or}\label{eqn:WgivenXY}\\
\sum_{w\in\calW}H(X|W=w)H(Y|W=w)&=0, \label{eqn:XYgivenW}
\end{align}
then $T_{\Ex}(\pi_{XY})=C_{\Wyner}(\pi_{XY})$.
\end{corollary}

If either of these conditions hold, it is easy to see that $\pi_{XY}$ is a Wyner-product distribution; thus Theorem~\ref{thm:exact_eq_wyn} generalizes these sufficient conditions.  Indeed, if $H(W|XY)=0$ (i.e.,~\eqref{eqn:WgivenXY} holds), $\calX\times\calY$ can be    partitioned into a collection of subsets $\{\calA_w :w\in\calW\}$. For each $w$, $P_{XY|W=w}$ is the restriction of $\pi_{XY}$ to $\calA_w$, defined in \eqref{eqn:wyner_prod}. Since by assumption, $X-W-Y$ holds, we have $P_{XY|W=w}=P_{X|W=w}P_{Y|W=w}$. This implies the restriction of $ \pi_{XY} $ to each $\calA_w$ can be written as a product distribution, i.e., $\pi_{XY}$ is a Wyner-product distribution.

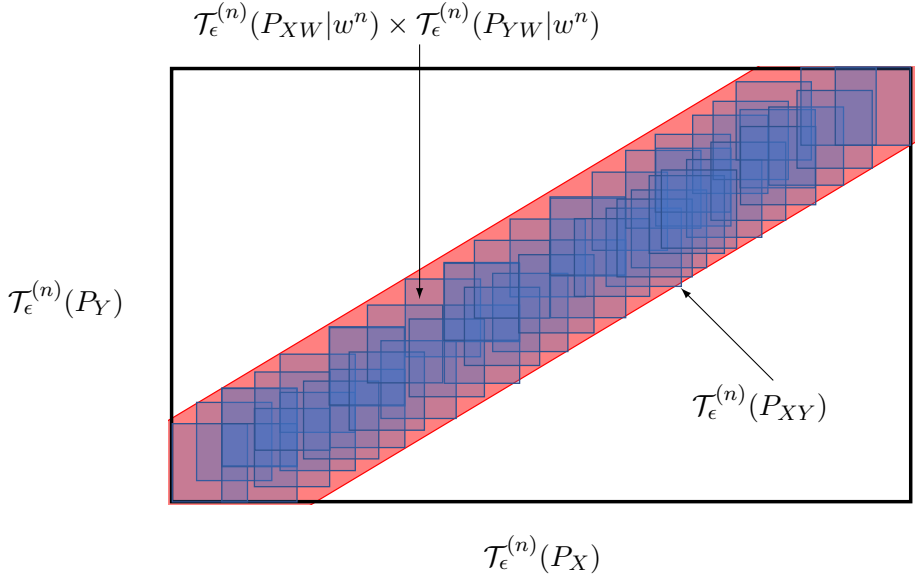
\begin{figure}[t]
\vspace{-.1in}
\centering
\begin{overpic}[width=1.02\textwidth]{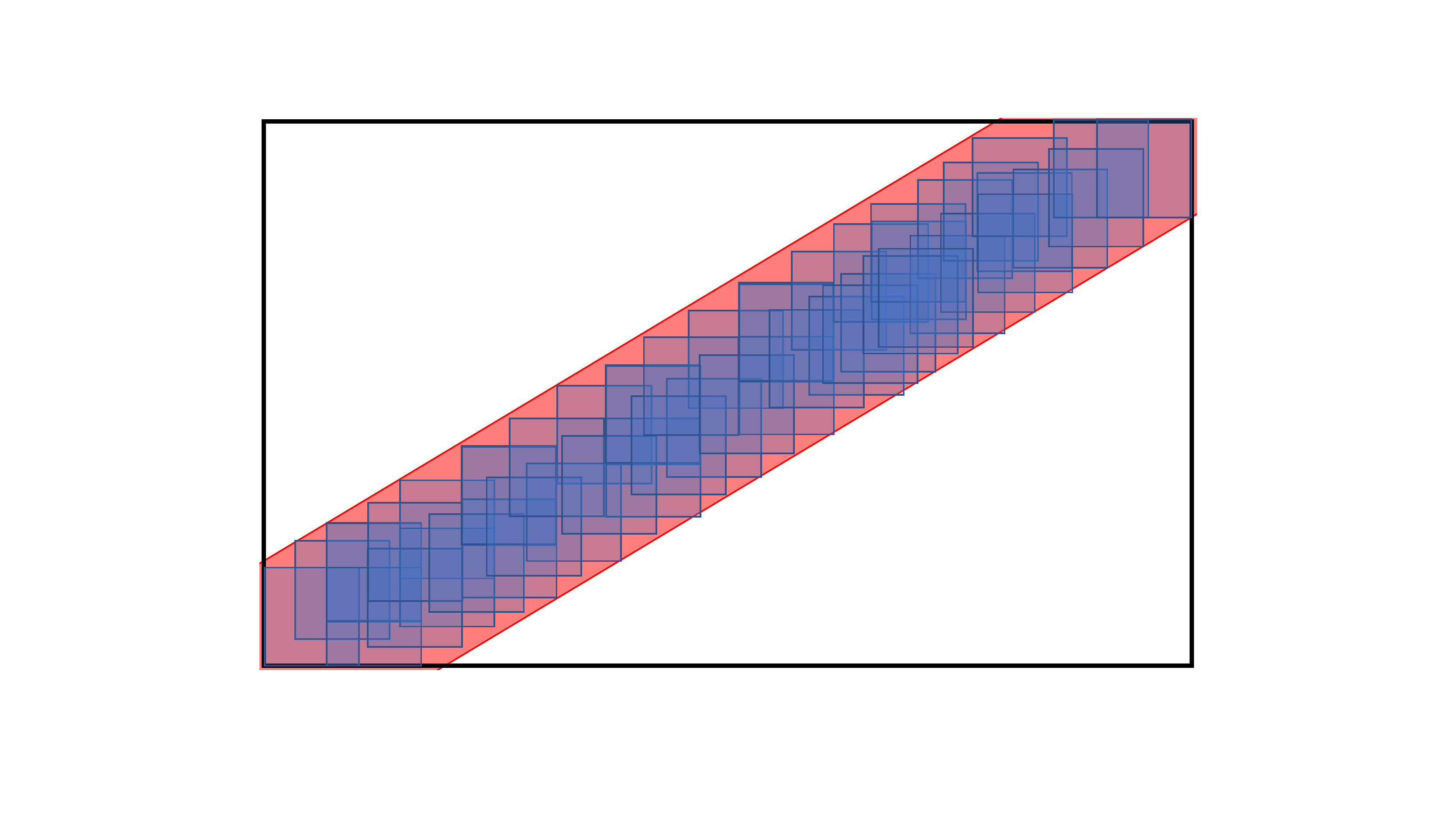}
\put(45,5){$\calT_\epsilon^{(n)}(P_X)$}
\put(4,27){$\calT_\epsilon^{(n)}(P_Y)$}
\put(63,18){$\calT_\epsilon^{(n)}(P_{XY})$}
\put(70,21){\vector(-1,1){8}}
\put(20,51){$\calT_\epsilon^{(n)}(P_{XW}|w^n )\times \calT_\epsilon^{(n)}(P_{YW}|w^n)$}
\put(39.5,50){\vector(0,-1){22}}
%\put(60,14){\vector(-1,1){8}}
%\put(50,12){Type overflow}
\end{overpic}
\vspace{-.28in}
\caption{Illustration of no type overflow. If the condition in~\eqref{eqn:XYgivenW} holds, there is no type overflow. Indeed, the jointly typical set $\calT_\epsilon^{(n)}(P_{XY})$ is approximately the union of the Cartesian product of conditional typical sets as written in \eqref{eqn:union_Cart}. }
\label{fig:type_non_overflow}
\end{figure}

On the other hand, if~\eqref{eqn:XYgivenW} holds, either the support of $P_{X|W=w}$ or the support of $P_{Y|W=w}$ (or both) is a singleton.   Hence,  the restriction of any joint distribution to $\supp(P_{X|W=w})\times  \supp(P_{Y|W=w})$ can be written as $P_X(x)\mathbbm{1}\{y=y_0\}$ or $P_Y(y)\mathbbm{1}\{x=x_0\}$ for some $(P_X , P_Y)\in\calP(\calX)\times\calP(\calY)$ and some $(x_0,y_0)\in\calX\times\calY$, i.e., $\pi_{XY}$ is a Wyner-product distribution.   Another way of seeing this, and as illustrated in Fig.~\ref{fig:type_non_overflow}, is that if $H(X|W=w)H(Y|W=w)=0$ for each $w$, then the coupling set $\calC(P_{X|W},P_{Y|W})$  is a singleton consisting solely of the distribution $P_{X|W}P_{Y|W}$. In other words, the jointly typical set   is approximately the union of   Cartesian products of conditionally typical sets, i.e.,  
\begin{equation}
\calT_\epsilon^{(n)}(P_{XY})\approx\bigcup_{ w^n \in \calC} \big( \calT_\epsilon^{(n)}(P_{XW}|w^n )\times \calT_\epsilon^{(n)}(P_{YW}|w^n )\big). \label{eqn:union_Cart}
\end{equation}
Hence, the jointly typical set  $\calT_\epsilon^{(n)}(P_{XY})$ is approximately $\supp(P_{X^nY^n})$, nullifying the type overflow phenomenon as discussed in Section~\ref{sec:type_over}. Thus, the equality $T_{\Ex}(\pi_{XY})=C_{\Wyner}(\pi_{XY})$ holds. 

\enlargethispage{\baselineskip}
In the remaining sections, we turn to  examples to illustrate the exact common information for various joint sources.

\section{Symmetric Binary Erasure Sources} \label{sec:exact_dses}

Recall the SBES introduced in Section~\ref{sec:sbes}. For this source in which its Wyner's common information is stated in Proposition~\ref{prop:sbes}, observe the following important feature from Fig.~\ref{fig:sbes}. 
%Here is the important observation. Notice that 
If $W=0$, then we know for sure that $X=0$. Similarly if $W=1$, we also know  that $X=1$. The final possibility is that $W=\rme$, in which case $Y=\rme$. That is to say, for all $w\in \calW=\{0,1,\rme\}$, either $H(X|W=w)=0$ or $H(Y|W=w)=0$ (indicated by the red arrows in Fig.~\ref{fig:sbes}). Thus, by the sufficient condition in~\eqref{eqn:XYgivenW} in Corollary~\ref{cor:equiv}, we know that $T_{\Ex}(\pi_{XY})=C_{\Wyner}(\pi_{XY})$. This is summarized in the following proposition, which was originally proved from first principles (i.e., without using Corollary~\ref{cor:equiv}) by \citet{KLE2014}.
\begin{proposition}
The exact common information for the SBES with erasure probability $p$ is 
\begin{equation}
T_{\Ex}(\pi_{XY}) = C_\Wyner(\pi_{XY}) = \left\{ \begin{array}{cc}
1 & p \le 0.5\\
 h(p) & p>0.5
\end{array} \right. .
\end{equation}
\end{proposition}
%\begin{figure}
%\centering
%\includegraphics[width = .8\columnwidth]{figs/ExactBEC}
%\caption{Plot of the exact common information for the SBES}
%\label{fig:eci_sbes}
%\end{figure}
%This function is plotted in Fig.~\ref{fig:eci_sbes}. Note that the exact and Wyner's common informations are the same.

This function is illustrated in Fig. \ref{fig:eci_sbes}.

\section{Doubly Symmetric Binary Sources}\label{sec:exact_dsbs}
As we have just seen, in the case of the SBES, the exact common information can be computed in closed form and is equal to Wyner's common information. This begs the following two questions. Are there any other sources for which the exact common information $T_\Ex(\pi_{XY})$ can be computed in closed form?  From what we have gathered up to this point, in general, $T_\Ex(\pi_{XY})$ can only be expressed via a {\em multi-letter} form (in Proposition~\ref{prop:KLE}) or via single-letter {\em bounds}  (in Theorem~\ref{thm:exact_sl}). In addition, are there sources for which the exact common information  is strictly larger than Wyner's common information? The latter is the content of an open question posed by \citet{KLE2014}.

\iffalse\begin{figure}\vspace{-.1in}
\centering
\begin{overpic}[width=.8\textwidth]{figs/PartII_flowchart.pdf}
{\small \put(15,39.5){Theorem \ref{thm:RCI}}
\put(15,17.5){Theorem \ref{thm:exact_sl}}
\put(60,39.5){Theorem \ref{thm:equivalence}}
\put(62,17.5){Prop.~\ref{prop:dsbs_exact}}
}
\end{overpic}
\vspace{-.4in}
\caption{Flowchart of results leading to the main result in Part \ref{part:two}, namely Proposition~\ref{prop:dsbs_exact}}
\label{fig:flowchart_partII}
\end{figure}\fi

In this section, we consider the DSBS with crossover probability $p \in (0,1/2)$ as described in Section~\ref{sec:dsbs}. Surprisingly, the exact common information can also be evaluated in closed form. Recall  from Section~\ref{sec:dsbs} that  $a=a(p) \in (0,1/2)$ is defined as the unique number satisfying $a\ast a=p$.
\begin{proposition}\label{prop:dsbs_exact}
The exact common information of  the DSBS   with crossover probability $p$ is 
\begin{align}
T_{\Ex}(\pi_{XY}) =  -2h(a) - (1\!-\! 2a)\log\left(\frac{1}{2}  \big(a^2\!+\!\bara^2\big)\right) -2a \log\big( a\bara\big).
\end{align}
\end{proposition}

\begin{figure}
\centering
\includegraphics[width = .8\columnwidth]{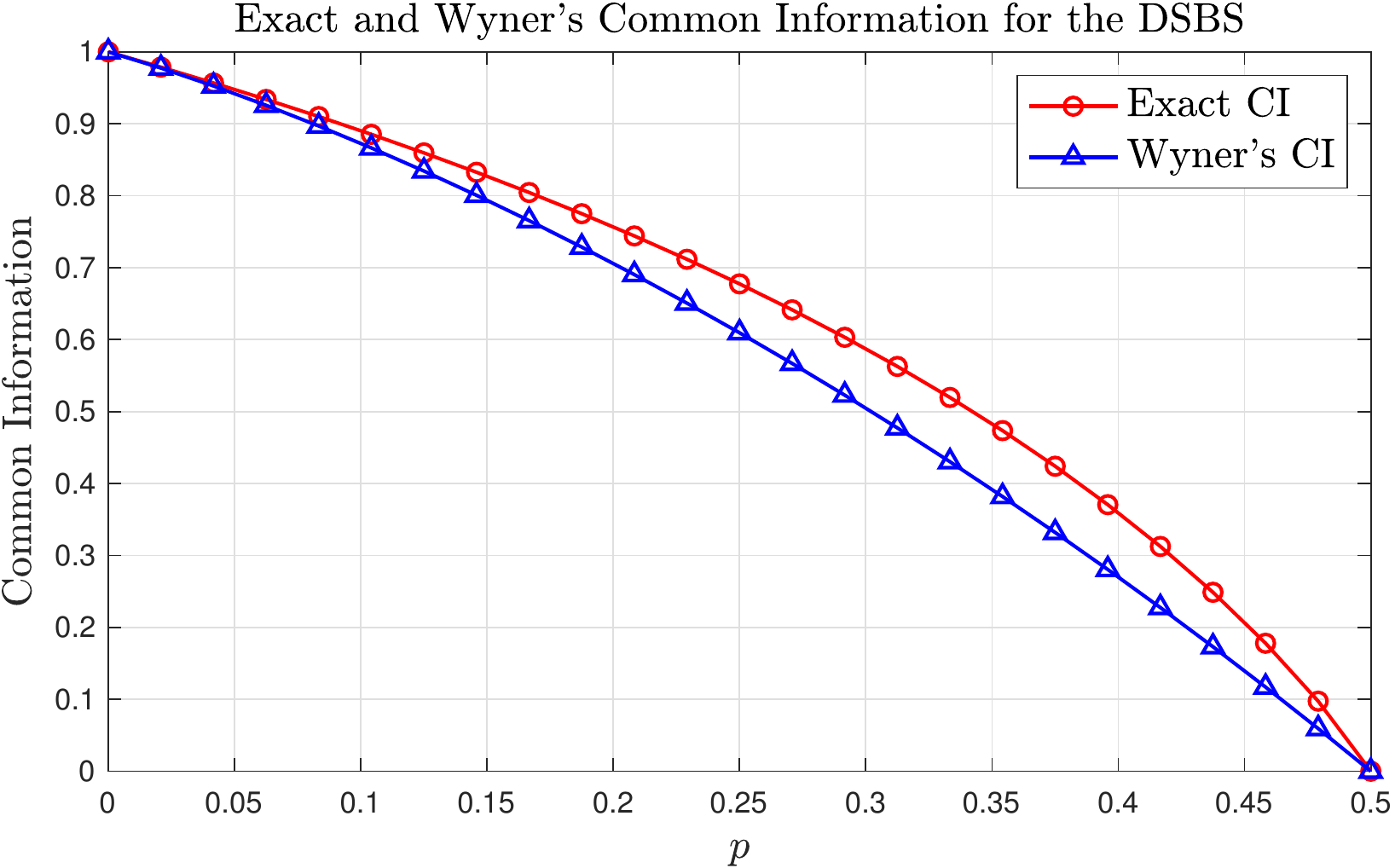}
\caption{Plot of the exact and Wyner's common information for the DSBS}
\label{fig:eci_dsbs}
\end{figure}

This result follows by Theorem~\ref{thm:equivalence} and~\eqref{eqn:RCI_infty0}--\eqref{eqn:RCI_infty} in Proposition~\ref{prop:RCI_dsbs} where we stated $T_\infty$ and $\tilT_\infty$ for the DSBS. From Proposition~\ref{prop:dsbs_exact} and Proposition~\ref{prop:wyner_dsbs}, we see that the difference between the exact and Wyner's common information is
\begin{equation}
T_{\Ex}(\pi_{XY}) - C_\Wyner(\pi_{XY}) =2a^2\log\bigg( \frac{a^2+\bara^2}{2a\bara} \bigg)>0. \label{eqn:dsbs_gap}
\end{equation}
This difference is positive for all $a\in (0,1/2)$; equivalently, $p\in (0,1/2)$. This answers the open problem posed by \citet{KLE2014}.  We conclude that there exists sources (namely the DSBS with $p\in (0,1/2)$) for which the exact common information strictly exceeds Wyner's common information.  Note that the DSBS does not satisfy any of the sufficient conditions in Section~\ref{sec:equality}. The two common information quantities and their gap are illustrated in Fig.~\ref{fig:eci_dsbs}.

%  Since we did not provide a sketch for the proof of~\eqref{eqn:RCI_infty0}--\eqref{eqn:RCI_infty}, we   do so here. 
%where $h_{\rmb}(a) = -a\log a- \bara\log \bara$ is the binary entropy function. 

\begin{proof}[Proof Sketch of Proposition~\ref{prop:dsbs_exact}]
Because $T_{\Ex}(\pi_{XY})=\tilT_\infty(\pi_{XY})$ (Theorem~\ref{thm:equivalence}), it suffices to prove   \eqref{eqn:RCI_infty0}--\eqref{eqn:RCI_infty}.  The crux in the evaluation of both bounds is in the understanding of the maximal cross-entropy terms in $\oGamma(\pi_{XY})$ and $\uGamma(\pi_{XY})$. 
 
In view of \eqref{eqn:rci_upper}, we first evaluate the upper bound $\oGamma(\pi_{XY})$ for the DSBS. We set $P_WP_{X|W}P_{Y|W}$ as the distribution that achieves the minimum in Wyner's common information. Hence $W\sim\mathrm{Bern}(1/2)$ and $X = W\oplus A$ and $Y = W\oplus B$ where $A$ and $B$ are mutually independent $\mathrm{Bern}(a)$ random variables. The key terms in $\oGamma(\pi_{XY})$ are thus the  maximal cross-entropies for each $w$. For a fixed $w$, this  can be simplified as follows 
\begin{align}
&\rvH_{\infty}(P_{X|W=w}, P_{Y|W=w}\|\pi_{XY}) \nn\\*
 &=\max_{Q_{XY}\in\calC(P_{X|W=w}, P_{Y|W=w})}\sum_{x,y}Q_{XY}(x,y)\log\frac{1}{\pi_{XY}(x,y)}\\
 &=\log\frac{1}{\alpha}+2\min\{a,\bara\}\log\frac{\alpha}{\beta}=\log\frac{1}{\alpha}+2a\log\frac{\alpha}{\beta}.
\end{align}
See Example~\ref{ex:dsbs_renyi} for details of this calculation. Hence, we have
\begin{align}
\oGamma(\pi_{XY})\le -2h(a) +\log\frac{1}{\alpha}+2a\log\frac{\alpha}{\beta}.
\end{align}
Recalling that $\alpha=\frac{1}{2}(a^2 + \bara^2)$ and $\beta=a\bara$ completes the proof of the upper bound.

The evaluation of the lower bound in~\eqref{eqn:uGamma_inf} is more involved but is essentially inspired by \citet{WynerCI} in his evaluation of Wyner's common information for the DSBS. Let $\alpha_w :=\Pr(X=0|W=w)$ and $\beta_w=\Pr(Y=0|W=w)$. The condition that    $P_{XY}=\pi_{XY}$ implies that $\bbE[\alpha_W]=\bbE[\beta_W]= \Pr(X=0)=\Pr(Y=0)= 1/2$ and $\bbE[\alpha_W\beta_W]=\Pr(X=0,Y=0)=\alpha$. In view of these equalities, we  lower bound the maximal cross-entropy for each $w$ as follows
\begin{align}
&\rvH_{\infty}(P_{X|W=w}, P_{Y|W=w'}\|\pi_{XY}) \nn\\*
 &=\max_{Q_{XY}\in\calC(P_{X|W=w}, P_{Y|W=w'})}\sum_{x,y}Q_{XY}(x,y)\log\frac{1}{\pi_{XY}(x,y)}\\
 &=\log\frac{1}{\alpha}+\Big(\min\{\alpha_w,\overline{\beta_{w'}}\} +\min\{\overline{\alpha_w}, {\beta_{w'}}\} \Big)\log\frac{\alpha}{\beta}\\
%  &=\log\frac{1}{\alpha}+\Big(\min\{\alpha_w, \beta_{w'}\} +\min\{\overline{\alpha_w}, \overline{\beta_{w'}}\} \Big)\log\frac{\alpha}{\beta}\\
  &\ge\log\frac{1}{\alpha}+\Big(\min\{\alpha_w, \overline{\alpha_{w}}\} +\min\{\beta_{w'}, \overline{\beta_{w'}}\} \Big)\log\frac{\alpha}{\beta}  .
\end{align}
Now, we plug this lower bound into the definition of $\uGamma(\pi_{XY})$ in~\eqref{eqn:uGamma_inf}. We conclude by leveraging ideas from \citet{WynerCI}; these ideas  include the concavity of the functions $x \in (0,1/2)\mapsto h(x)$ and $x\in\bbR_+\mapsto\sqrt{x}$, to solve the  optimization problem in~\eqref{eqn:uGamma_inf}. See \cite{YuTan2020_exact} for details. 
\end{proof} 

%Resolution of Kumar, Li and El Gamal's open problem. 
\section{Jointly Gaussian Sources} \label{sec:exact_gauss}
In this final section, we briefly discuss the generalization of the concept of exact common information to continuous sources  and, specifically, the important family of jointly Gaussian sources. Per the theme of this section, we aim to establish that the unnormalized R\'enyi common information of order $\infty$ is equal to the exact common information.  However, this is not true in general for arbitrary continuous sources. Nevertheless, the proof that $T_{\Ex}(\pi_{XY})\ge\tilT_\infty(\pi_{XY})$ (in the second half of Section~\ref{sec:sketch_equiv}) goes through verbatim as the weakly typical set and its properties,  which are applicable to arbitrary sources, are exploited therein. It also holds that $T_{\Ex}(\pi_{XY})= \tilT_\infty(\pi_{XY})$ for sources with countable alphabets; this follows from another typicality and truncation argument. Hence, it remains to establish some mild regularity conditions such that $T_{\Ex}(\pi_{XY})\ge\tilT_\infty(\pi_{XY})$ holds for sources with uncountable alphabets.  

\enlargethispage{2\baselineskip}
In this section, we use $f_{XY}$ to denote the PDF of the distribution $\pi_{XY}$, which is assumed to be absolutely continuous with respect to the Lebesgue measure on $\bbR^2$. 
To state the results succinctly, for each $\epsilon>0$ and $n\in\bbN$, we define
\begin{equation}
\kappa_{\epsilon, n} :=\sup_{(x,y)\in \calI_{\epsilon,n}^2 } \left\{  \Big| \frac{\partial}{\partial x} \log f_{XY}(x,y) \Big|+\Big| \frac{\partial}{\partial y} \log f_{XY}(x,y) \Big| \right\},
\end{equation}
where $\calI_{\epsilon,n}$ is the interval $\big[-\sqrt{n (1+\epsilon)}, \sqrt{n (1+\epsilon)}\big]$. The following lemma and Proposition~\ref{prop:gauss_exact} to follow are due to the present authors~\cite{YuTan2020_exact}.
\begin{lemma}\label{lem:regularity}
Assume that the joint source $\pi_{XY}$ satisfies the following three assumptions. 
\begin{enumerate}
\item[(A1)] $\pi_{XY}$ is absolutely continuous on $\bbR^2$ with  $\bbE[X^2 ], \bbE[Y^2]<\infty$; %Without loss of generality, assume that $\bbE[X^2]=\bbE[Y^2]=1$;
\item[(A2)] The PDF $f_{XY}$ is log-concave and continuously differentiable and that $I(X;Y)$ exists (and thus is finite);
\item[(A3)] $\log \kappa_{\epsilon,n}$ is sub-exponential in $n$ (i.e., $\frac{1 }{n}\log\log \kappa_{\epsilon,n}\to 0$ as $n\to\infty$).
\end{enumerate}

\noindent If there exists a sequence of fixed-length distributed source simulation codes with rate $R$  (Definition~\ref{def:wyner_code}) that generates $P_{X^n Y^n}$ (defined in~\eqref{eqn:syn_dist}) such that 
\begin{equation}
D_{\infty}\big(P_{X^nY^n} \big\| \pi_{XY}^n \big) = o\bigg( \frac{1}{n+\log \kappa_{\epsilon,n} }\bigg) ,\label{eqn:rate_Dinf}
\end{equation}
then there   exists a sequence of  variable-length distributed source simulation codes  with rate $R$ (Definition~\ref{def:vlcode}) that  generates $\pi_{XY}^n$ exactly. In other words, $T_\Ex(\pi_{XY})\le\tilT_\infty^{\mathrm{cts}}(\pi_{XY})$  where $\tilT_\infty^{\mathrm{cts}}(\pi_{XY})$ is the infimum of all rates $R$ such that \eqref{eqn:rate_Dinf} holds for all $\epsilon>0$.  
\end{lemma}
In short, if the continuous source $\pi_{XY} \in\calP(\bbR^2)$  satisfies Assumptions (A1)--(A3) and the R\'enyi divergence of order $\infty$ vanishes sufficiently rapidly relative to the smoothness of the source density (captured by~$\kappa_{\epsilon, n}$), we are able to relate $T_{\Ex}$ to $\tilT_\infty^{\mathrm{cts}}$, a proxy of $\tilT_\infty$. 

One important example satisfying the conditions in Lemma~\ref{lem:regularity} is the class of jointly Gaussian sources as described in Section~\ref{sec:wyner_gauss}. Consider two jointly Gaussian random variables $X$ and $Y$  that have zero means and unit variances, and the pair $(X,Y)$ has correlation coefficient $\rho\in (0,1)$.\footnote{The results also hold for negative correlation coefficients in which case $\rho$ should be replaced by $|\rho|$.} In this case, it is easy to check that 
\begin{equation}
\kappa_{\epsilon, n} = \sup_{(x,y)\in \calI_{\epsilon,n}^2}  \left| \frac{x-\rho y}{1-\rho^2}\right|+\left| \frac{y-\rho x}{1-\rho^2}\right|=\frac{2\sqrt{ n (1+\epsilon)}}{1-\rho}.
\end{equation}
Hence, for every fixed $\epsilon>0$ and $\rho\in (0,1)$, $\log \kappa_{\epsilon, n} = O(\log n)$ is clearly sub-exponential in $n$. Furthermore, $(n+\log \kappa_{\epsilon, n})^{-1} =\Theta( 1/n)$. Hence, by Lemma~\ref{lem:regularity}, if there exists  a sequence of fixed-length codes of rate $R$ such that $D_\infty(P_{X^nY^n}\| \pi_{XY}^n) = o(1/n)$, then there also exists a sequence of rate-$R$ variable-length codes that exactly generates $\pi_{XY}^n$.

Using Lemma~\ref{lem:regularity}, we are able to provide bounds for the exact common information of jointly Gaussian sources.
%\clearpage
\begin{proposition}\label{prop:gauss_exact}
For a jointly Gaussian source with correlation coefficient $\rho\in (0,1)$, 
\begin{align}
\hspace{-.2in}\frac{1}{2}\log\left(\frac{1+\rho}{1-\rho}\right) &= C_{\Wyner}(\pi_{XY}) \le T_\infty(\pi_{XY})\le\tilT_\infty  (\pi_{XY}) \\
\hspace{-.2in}&=T_{\Ex}(\pi_{XY})\le \frac{1}{2}\log\left(\frac{1+\rho}{1-\rho}\right)+ \frac{\rho\log \rme}{1+\rho}. \label{eqn:Texact_gauss}
\end{align}
\end{proposition}

Thus, the upper and lower bounds differ by $\rho/(1+\rho)$. %The multiplicative gap (ratio of the upper and lower bounds) is no more than~$1.7$. 
These bounds are illustrated in Fig.~\ref{fig:eci_gauss}. % which is no more than $1$ bit per source symbol. 

\begin{figure}[!ht]
\centering
\includegraphics[width = .8\columnwidth]{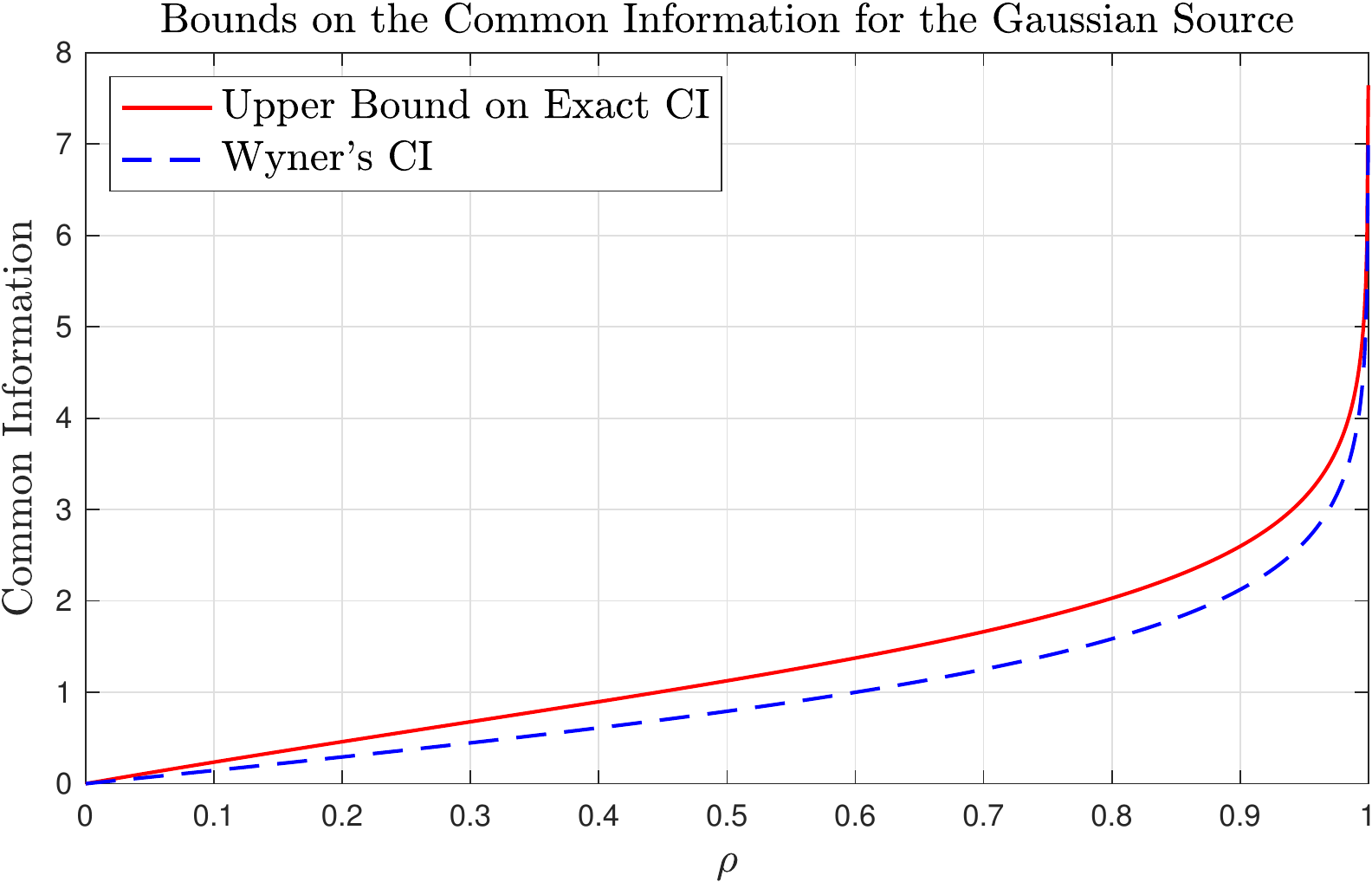}
\caption{Plot of the upper bound on exact common information~\eqref{eqn:Texact_gauss} and Wyner's common information (lower bound on the exact common information) for the jointly Gaussian source as stated in~\eqref{eqn:gaussian_wci}. The exact  common information lies between these two curves and Conjecture~\ref{conj:gauss} says the upper bound on the exact  common information is tight. }
\label{fig:eci_gauss}
\end{figure}

\begin{remark}
\citet{LiElgamal2017} showed using a dyadic decomposition scheme that 
\begin{align}
T_{\Ex}(\pi_{XY})\le  G(\pi_{XY})&\le I(X;Y)+24 =  \frac{1}{2}\log\left(\frac{1}{1-\rho^2}\right)+24.
\end{align}
This  bound by \citet{LiElgamal2017} is based on a {\em one-shot} scheme and hinges on upper bounding the common entropy $G(\pi_{XY})$, defined in~\eqref{eqn:common_ent}. The coding scheme involved in proving Proposition~\ref{prop:gauss_exact}, however, utilizes  {\em multiple copies} of the source and hence,   naturally results in a better upper bound. In fact, simple algebra yields that  for all $\rho \in (0,1)$
\begin{align}
\!\left[\frac{1}{2}\log\left(\frac{1}{1\!- \!\rho^2}\right)\!+\! 24\right]-\left[ \frac{1}{2}\log\left(\frac{1\!+ \!\rho}{1\!-\! \rho}\right) +  \frac{\rho\log\rme}{1+ \rho}\right]  
 \ge 22.28 \mbox{ bits/symb}.
\end{align}
The strategy to achieve the upper bound in \eqref{eqn:Texact_gauss}, which we will not describe in detail here,  is a combination of Li and El Gamal's  dyadic decomposition scheme \cite{LiElgamal2017} and the construction of a sequence of fixed-length codes that yields  $\{ P_{X^nY^n} \}_{n\in\bbN}$ satisfying $D_\infty(P_{X^nY^n}\|\pi_{XY}^n)= o(1/n)$ (as hinted by Lemma~\ref{lem:regularity}). 
\end{remark}

We remark that for the DSBS, the upper bound in Proposition~\ref{prop:dsbs_exact} is tight. It is thus natural to conjecture that the upper bound  in~\eqref{eqn:Texact_gauss} is  also tight which implies that the gap between Wyner's common information and the exact common information for the  bivariate Gaussian source is exactly $(\rho\log \rme)/(1+\rho)$. We state this as a conjecture. 
\begin{conjecture}[Exact common information for a jointly Gaussian source] \label{conj:gauss}
The exact common information for a jointly Gaussian source with correlation coefficient $\rho\in (0,1) $ is 
\begin{equation}
T_{\Ex}(\pi_{XY}) \stackrel{?}{=} \frac{1}{2}\log\left(\frac{1+\rho}{1-\rho}\right)+ \frac{\rho\log \rme}{1+\rho}. \label{eqn:Texact_gauss_conj}
\end{equation}
\end{conjecture}

\chapter{Approximate and Exact Channel Synthesis}
\label{ch:ecs}

How much information is required to create correlation {\em remotely}? How much interaction is necessary to create such correlation? These questions form the basis of this section. This setup is depicted in Fig.~\ref{fig:dcs}. It shows that an observer or encoder observes a sequence of i.i.d.\ random variables $X^n\sim\pi_X^n$ and describes it using a bit string with a certain rate $R$ to the decoder which itself produces another sequence $Y^n$. It is the hope that even though the encoder and decoder are remotely located, they can leverage a source of shared randomness $K_n$ to reduce the rate of jointly synthesizing a random process $(X^n,Y^n)\sim \pi_X^n P_{Y^n|X^n}$ such that its joint distribution $\pi_X^n P_{Y^n|X^n}$ is close to (or exactly equal) to a target distribution $\pi_{XY}^n=\pi_X^n\pi_{Y|X}^n$. Since the $X$-marginals of $\pi_{XY}^n$ and $\pi_X^n P_{Y^n|X^n}$ are identical, the spotlight is then shone on the generated conditional distribution $P_{Y^n |X^n}$ that is mandated to be close  (or exactly equal) to the target conditional distribution or {\em channel} $\pi_{Y|X}^n$. For this reason, this problem is termed as the {\em distributed channel synthesis} or {\em communication complexity of correlation} problem and has been studied in  \cite{winter02compression,cuff13,Bennett02,bennett14quantum,Harsha10} among others. 

Aiding the reconstruction of the channel is a source of {\em shared}  or {\em common randomness}  which we denote by $K_n$ in Fig.~\ref{fig:dcs}. This random variable is uniformly distributed on the index set $\calK_n=[2^{nR_0}]$; equivalently it has rate $R_0$.  It can be seen that there is a tradeoff between $R_0$ and $R$. Indeed, generally the larger the amount of shared randomness $R_0$, the more resources the encoder and decoder jointly have, and consequently,  the rate of communication $R$ required for synthesizing $\pi_{Y|X}^n$ (exactly or approximately) is usually smaller. The purpose of this section is to quantify this tradeoff precisely.
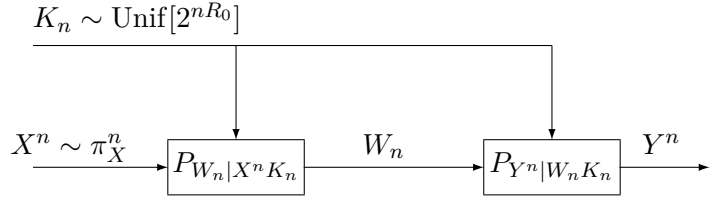
\begin{figure}[t]
\centering \setlength{\unitlength}{0.06cm} { \begin{picture}(150,40)
%\linethickness{1pt}
\put(-5,10){\vector(1,0){30}} \put(-10,13){%
\mbox{%
$X^{n}\sim\pi_{X}^{n}$%
}} \put(25,4){\framebox(30,12){$P_{W_{n}|X^{n}K_{n}}$}} \put(55,10){\vector(1,0){40}}
\put(68,13){%
\mbox{%
$W_{n}$%
}} \put(95,4){\framebox(30,12){$P_{Y^{n}|W_{n}K_{n}}$}} \put(125,10){\vector(1,0){20}}
\put(130,13){%
\mbox{%
$Y^{n}$%
}} \put(40,37){\vector(0,-1){21}} \put(-5,41){%
\mbox{%
$K_{n}\sim\mathrm{Unif}[2^{nR_{0}}]$%
}} \put(110,37){\vector(0,-1){21}} \put(-5,37){\line(1,0){115}}
 \end{picture}}
\caption{The   channel synthesis problem. The goal is to ensure that $P_{X^nY^n}$ is either approximately or exactly equal to $\pi_{XY}^n$.}
\label{fig:dcs}
\end{figure}

In the spirit of the previous sections, we study the problems of {\em approximately} and  {\em exactly} synthesizing the ($n$-fold product of the) target channel $\pi_{Y|X}^n$. The approximate version consists in quantifying the tradeoff between $R$ and $R_0$ such that the TV distance between $P_{X^nY^n}$ and $\pi_{XY}^n$ converges to zero as the blocklength $n$ increases without bound. This problem was studied by \citet{Bennett02}, \citet{winter02compression}, \citet{cuff13}, and \citet{bennett14quantum} among others. In particular, \citet{cuff13} showed that if $R_0=0$, then the minimum amount of communication rate required for TV-approximate synthesis is $R=C_\Wyner(\pi_{XY})$. In essence, when there is no common randomness, the problem of channel synthesis reduces to the distributed source simulation problem (Section~\ref{sec:dist_sim}). On the other hand, if $R_0 =\infty$, the corresponding minimum amount of rate is $R=I(X;Y)$. See Table~\ref{tab:ecs}. Thus by varying $R_0$, one traces out a tradeoff curve that interpolates between two familiar notions of correlation, namely Wyner's common information and the mutual information.  We elaborate on this in Section~\ref{sec:approxCS}.

We are also concerned with synthesizing the channel $\pi_{Y|X}$ {\em exactly} using {\em variable-length} codes. This problem was also studied in several works, including by \citet{Bennett02}, \citet{Harsha10} and  \citet{LiElgamal2018}. \citet{Bennett02} showed that when there is unlimited shared randomness, the minimum rate of communication is $I(X;Y)$. At the other extreme, if there is no shared randomness, the problem of exact channel synthesis reduces to  the exact common information problem. From Section~\ref{ch:exact}, we saw that for the DSBS, exact channel synthesis (with a uniform source $X$)
requires a strictly larger communication rate $\tilT_\infty(\pi_{XY})=T_\Ex(\pi_{XY})$ compared to  that required
for the TV-approximate version $\tilT_1(\pi_{XY})=C_\Wyner(\pi_{XY})$ (Theorem~\ref{thm:tv_ci}). These results are also summarized in Table~\ref{tab:ecs}.

\begin{table}[t]
\caption{Summary of results for the minimum communication rate for the extreme cases of the common randomness (CR) rate $R_0=0$ and $R_0=\infty$}
\label{tab:ecs}
\centering{\small
\begin{tabular}{|l|c|c|}
\hline
\diagbox{Syntheses}{CR Rate}                         & $R_0=\infty$ & $R_0=0$                \\ \hline \hline 
TV Approx.\ Synthesis & $I(X;Y)$ \cite{Bennett02, winter02compression, cuff13}      & $C_{\Wyner}(\pi_{XY})$ \cite{cuff13}  \\ \hline
Exact Synthesis          & $I(X;Y)$  \cite{Bennett02}   & $T_{\Ex}(\pi_{XY})$ \cite{KLE2014}    \\ \hline
\end{tabular}}
\end{table}

In this section, we are concerned with {\em refinements} to these extreme cases. Some results in the literature are worth highlighting. 
\citet{Harsha10}  used
a rejection sampling scheme to study the one-shot version of
exact simulation for the discrete source $(X, Y )$. The authors showed that the
number of bits of the shared randomness can be limited to
$O(\log \log |\calX| + \log |\calY|)$ if the expected description length of $X$ is
increased by $O(\log (I (X; Y ) + 1)+\log \log |\calY|)$ bits from the
mutual information lower bound $I (X; Y )$. \citet{LiElgamal2018} showed that if
the expected description length is increased by $\log(I (X; Y )+
1) + 5$ bits from $I (X; Y )$, then the number of bits of 
shared randomness can be upper bounded by $\log(|\calX|(|\calY|-1)+2)$. 
%It can be seen that if there is no shared randomness, the problem of exact channel synthesis reduces to that of the exact common information. From Chapter~\ref{ch:exact}, we see that for the DSBS, exact channel synthesis (with a uniform source)
%requires a strictly larger communication rate $\tilT_\infty(\pi_{XY})=T_\Ex(\pi_{XY})$ than that required
%for the TV-approximate version $C_\Wyner(\pi_{XY})$. This chapter studies the tradeoff between the communication rate and shared randomness rate for exact channel synthesis. We note that two extreme points of this tradeoff have been studied---for the case of unlimited shared randomness, this has been addressed in \cite{Bennett02} and for the case of no shared randomness, this was studied by \cite{KLE2014}; see Table~\ref{tab:ecs} for a summary of the results. 
This section is   concerned with the fundamental limits of the amount of shared randomness when the sequence of communication rates
is required to approach the minimum rate $I (X; Y )$ {\em only
asymptotically} as $n\to\infty$. In this case, what is the minimum
amount of shared randomness required to realize exact synthesis? \citet{bennett14quantum} conjectured that an exponential number
of bits (and hence an infinite rate) of shared randomness
is necessary. This was disproved by \citet{Harsha10} and \citet{LiElgamal2018} where finite bounds on the rate were established. This section, and in particular Section~\ref{sec:singleletterCS}, surveys advances on this question and provides the best known bounds on the minimum amount of shared randomness in Section~\ref{sec:tradeoff_rates}. We supplement our discussions with numerical examples using the DSBS and the bivariate Gaussian source.

Besides the works surveyed above, local TV-approximate
simulation of a channel was studied by \citet{Ste96}.  TV-approximate simulation of a ``bidirectional''
channel via interactive communication was studied by \citet{YGA15}. Both the exact and TV-approximate versions
of the simulation of a channel over another noisy channel were
studied by \citet{Haddadpour17}. In particular, \cite{Haddadpour17} addressed
the case of exact simulation of a binary symmetric channel
over a binary erasure channel. The relationship between
the problem of exact channel simulation over another channel
and the problem of zero-error capacity was studied by
\citet{Cubitt11}.

\section{Approximate Channel Synthesis} \label{sec:approxCS}
In this section, we set the stage by describing the problem of approximate channel synthesis. The problem is depicted in Fig.~\ref{fig:dcs} in which the encoder provides a description of   the source sequence $X^n\sim\pi_{X}^n$ at a certain rate $R$. The rate-$R$ description, also known as the {\em message}, is denoted as $W_n$. A rate-$R_0$ random variable $K_n$, uniformly distributed on $\calK_n$, represents  {\em common randomness} available to {\em both} the encoder and decoder. The decoder generates a sequence $Y^n$ based on the message $W_n$ and the common randomness~$K_n$. 

The following definition is parallel to   Definition~\ref{def:wyner_code} for fixed-length distributed source simulation codes. 
\begin{definition}
An {\em $(n,R,R_0)$-fixed-length   channel synthesis code} consists of a pair of random mappings  $P_{W_n | X^n K_n}\in\calP(\calW_n | \calX^n\times \calK_n)$  and $P_{Y^n  | W_n  K_n}\in\calP(\calY^n | \calW_n\times \calK_n)$ such that 
\begin{equation}
\frac{1}{n}\log |\calW_n|\le   R\quad\mbox{and}\quad \frac{1}{n}\log |\calK_n|\le R_0.
\end{equation}
These two mappings are   known as the {\em encoder} and {\em decoder} respectively.
\end{definition}
Given   a code $(P_{W_n | X^n K_n},P_{Y^n  | W_n  K_n})$, the joint distribution of the message $W_n$ and output $Y^n$ given $(X^n, K_n)$ is  
\begin{equation}
P_{Y^n W_n|X^n K_n} = P_{Y^n  | W_n  K_n}  P_{W_n | X^n K_n}. \label{eqn:syn_exact0}
\end{equation}
The joint distribution of all the random variables $(X^n,Y^n, W_n, K_n)$ is 
\begin{equation}
P_{X^nY^n W_n K_n} =P_{Y^n W_n|X^n K_n}P_{X^nK_n}, 
\end{equation}
where, by definition, 
\begin{equation}
P_{X^nK_n}(x^n,k)=\frac{\pi_X^n(x^n)}{|\calK_n|}  \quad\mbox{for all} \;\, (x^n ,k)\in \calX^n\times  \calK_n.
\end{equation}
 Given a code, the {\em synthesized distribution} is
\begin{equation}
P_{X^nY^n} (x^n,y^n ) := \sum_{(w,k) \in \calW_n\times\calK_n} P_{X^nY^n W_n K_n}(x^n,y^n,w,k). \label{eqn:syn_exact}
\end{equation}
\begin{definition} \label{def:tv_approx_syn}
The pair $(R,R_0) \in \bbR_+^2$ is said to be {\em achievable for synthesizing the channel $\pi_{Y|X}$ with input $\pi_X$} if there exists a sequence of $(n,R,R_0)$-fixed-length   channel synthesis codes such that the TV distance between the synthesized distribution in \eqref{eqn:syn_exact} and the target distribution $\pi_{XY}^n$ vanishes, i.e.,
\begin{equation}
\lim_{n\to\infty}\left| P_{X^nY^n}-\pi_{XY}^n \right|=0.
\end{equation}
Define the {\em optimal rate region} $\calT(\pi_{XY}) \subset\bbR_+^2$ to be the closure of the set of achievable rate pairs $(R, R_0)$ for synthesizing $\pi_{Y|X}$ with input~$\pi_X$. 
\end{definition}
We remark that this definition is generally more stringent than the analogous one for distributed source synthesis in Definition~\ref{def:TVCI} as we only require that the TV distance vanishes. In contrast, in Definition~\ref{def:TVCI}, the TV distance is only required to be asymptotically   bounded by $\eps \in [0,1)$. To  state  the next result succinctly, let us define the following  set:
\begin{align}
\hspace{-.2in}\calC_{\Wyner}(\pi_{XY}) :=\bigcup_{\substack{ P_W P_{X|W} P_{Y|W} : \\ P_{XY}=\pi_{XY} }} \left\{  (R, R_0)   \, :  \, \parbox[c]{1.4 in}{$\hspace{.95cm} R\ge I(X;W)$  \vspace{0.03 in}\\ $R+R_0\ge I(XY;W)$ }  \right\}.  \label{eqn:calC}
\end{align}
Here, just like in~\eqref{eqn:wyner_sim} for Wyner's common information, the union runs over all triples of random variables $(X,W,Y)$ such that $X-W-Y$ forms a Markov chain in this order and $P_{XY}=\pi_{XY}$. To exhaust the rate region, it suffices to take $|\calW|\le|\calX||\calY|+1$.
 \citet{cuff13} proved the following fundamental result. 
\begin{theorem}\label{thm:dist_chsyn}
For any joint distribution $\pi_{XY}$ defined on a finite alphabet $\calX\times\calY$, 
\begin{equation}
\calT(\pi_{XY}) =\calC_{\Wyner}(\pi_{XY}) .
\end{equation}
\end{theorem}

\begin{figure}[t]
%\vspace{-.25in}
\centering
\begin{overpic}[width=.65\textwidth]{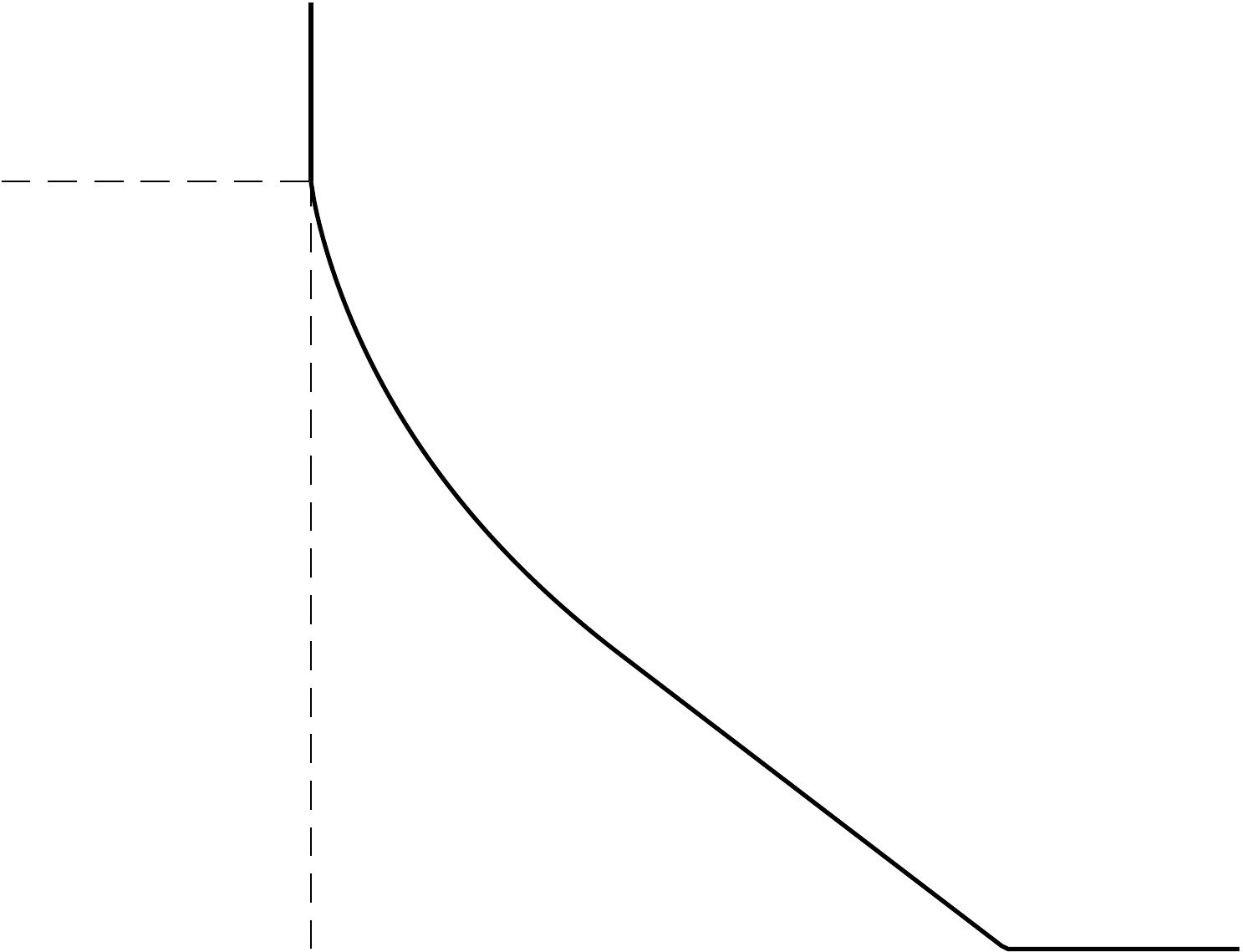}
\put(100,4){$R$}
\put(-20,61.5){$H(Y|X)$}
\put(-8,75){$R_0$}
\put(34,10){$I(X;Y)$}
\put(25,0){\circle*{2}}
\put(82,0){\circle*{2}}
\put(82,20){\vector(0,-1){18}}
\put(33,9){\vector(-1,-1){7}}
\thicklines
\put(-3,0){\vector(1,0){106}}
\put(0,-3){\vector(0,1){80}}

\put(71,22){$C_{\Wyner}(\pi_{XY})$}
{\Large\put(70,60){$\calC_{\Wyner}(\pi_{XY})$}}
\vspace{.1in}
\end{overpic}
\caption{A schematic of the region $\calC_{\Wyner}(\pi_{XY})$ defined in \eqref{eqn:calC}}
\label{fig:dist_chsyn}
\end{figure}

%\begin{figure}[t]
%\vspace{-.25in}
%\centering
%\begin{overpic}[width=.9\textwidth]{figs/dist_chsys.pdf}
%\put(80,9){$R$}
%\put(18,49){$R_0$}
%\put(33,2.5){$I(X;Y)$}
%\put(60,2.5){$C_{\Wyner}(\pi_{XY})$}
%\put(10,27){$H(Y|X)$}
%\put(60,35){$\calC_{\Wyner}(\pi_{XY})$}
%\end{overpic}
%\vspace{-.1in}
%\caption{A schematic of the region $\calC_{\Wyner}(\pi_{XY})$ defined in \eqref{eqn:calC}}
%\label{fig:dist_chsyn}
%\end{figure}

 Let us examine the extreme points of the region $\calC_{\Wyner}(\pi_{XY})$. When there is no common randomness, i.e., $R_0=0$, the second inequality in~\eqref{eqn:calC} dominates and Theorem~\ref{thm:dist_chsyn} says that rate of communication needs to be at least $I(XY;W)$ where the joint distribution $P_{W} P_{X|W}P_{Y|W}$ satisfies $P_{XY}=\pi_{XY}$. This rate is precisely Wyner's common information $C_{\Wyner}(\pi_{XY})$. On the other hand, when $R_0 =\infty$, the second inequality is inactive and we can easily see that the minimum communication rate is $R=I(X;Y)$. This can be rigorously justified as follows.  By the data processing inequality for the mutual information and the Markov chain $X-W-Y$, we have $I(X;W)\ge I(X;Y)$. This inequality can be met with equality by choosing $W=Y$. Furthermore, this choice implies that 
\begin{align}
 R_0 &=I(XY;W )- R = I(XY ;Y) -  R \label{eqn:R0_req0}\\
 &= H(Y) -  I(X;Y)=  H(Y|X) \label{eqn:R0_req}
 \end{align} 
 is a sufficient common randomness rate for achieving $R=I(X;Y)$. A schematic of the region $\calC_{\Wyner}(\pi_{XY})$ is shown in Fig.~\ref{fig:dist_chsyn}. Hence, we see that the approximate   channel synthesis problem provides us with a tuning knob $R_0$   to obtain a continuum of common information measures that interpolate from the mutual information to Wyner's common information. 
 
 We now devote the final paragraphs of this section to sketch the achievability proof of Theorem~\ref{thm:dist_chsyn}. The main idea is to invoke the TV distance version of the soft-covering lemma (cf.\ \eqref{eqn:soft-covering-tv} in Lemma~\ref{lem:soft-covering} and Section~\ref{sec:ach_tv}) multiple times, together with  some properties of the TV distance. 
 
 We proceed by a  random selection  (random coding) argument.  Fix any distribution $P_W P_{X|W} P_{Y|W}$ such that $P_{XY} =\pi_{XY}$. Randomly and independently generate a codebook $\scC_{n} = \{W^n(m, k): m \in \calW_n, k \in \calK_n\}$ where $\log |\calW_n|=nR$ and $\log |\calK_n|=nR_0$ and where each codeword  $W^n(m,k)$ is generated independently from the $n$-fold product distribution $P_W^n$.  Using $\scC_n$, define the  (random) distribution 
 \begin{align}
 Q_{X^n Y^n W_n K_n} (x^n,y^n,m,k) = \frac{ P_{X|W}^n (x^n | W^n(m,k)) P_{Y|W}^n(y^n | W^n(m,k))}{2^{n(R+R_0)}} .
 \end{align}
Based on $Q_{X^n Y^n W_n K_n}$, define the {\em synthesized distribution} as 
\begin{equation} 
P_{X^n Y^n W_n K_n} (x^n,y^n,m,k)= \frac{\pi_X^n(x^n) Q_{Y^n W_n|X^nK_n}(y^n,m | x^n,k)}{2^{nR_0}} . \label{eqn:syn_exact2}
\end{equation}
This distribution satisfies all the properties in~\eqref{eqn:syn_exact0}--\eqref{eqn:syn_exact}. % By construction $X^n-(W_n,K_n)-Y^n$ forms a Markov chain. 
 
 By the soft-covering lemma for the TV distance, if 
\begin{equation}
 R+R_0>I_P(XY;W),  \label{eqn:rate_constraint1}
 \end{equation} 
then the expectation of the  TV distance between $Q_{X^nY^n}$ and $\pi_{XY}^n$ vanishes, i.e., 
 \begin{equation}
\lim_{n\to\infty} \bbE\Big[\big| Q_{X^nY^n}-\pi_{XY}^n \big|\Big]=0. \label{eqn:tv1}
 \end{equation}
 Unfortunately, $Q_{X^nY^n}$ is not the synthesized distribution $P_{X^nY^n}$ so we must do a little more. Applying  the soft-covering lemma for the TV distance again, we see that if 
\begin{equation}
 R>I_P (X;W), \label{eqn:rate_constraint2}
 \end{equation} 
 then for all $k \in\calK_n$, 
 \begin{equation}
\lim_{n\to\infty} \bbE\Big[\big|Q_{X^n|K_n=k }-\pi_X^n \big|\Big] =0.
 \end{equation}
 Consequently,  by invoking the definition of the TV distance, 
 \begin{equation}
 \lim_{n\to\infty} \bbE\Big[\big|Q_{X^n K_n  }-\pi_X^n Q_{K_n} \big|\Big] =0, \label{eqn:tv2}
 \end{equation}
 where $Q_{K_n}=\mathrm{Unif}[2^{nR_0}]$. Now, we compare the synthesized distribution to the target distribution as follows
 \begin{align}
 &\big|P_{X^n Y^n}-\pi_{XY}^n  \big| \nn\\*
 & \le  \big|P_{X^n Y^n  }-Q_{X^nY^n} \big|+\big|Q_{X^nY^n}- \pi_{XY}^n  \big| \label{eqn:tri}\\
 & \le  \big|P_{X^n Y^n W_n K_n  }-Q_{X^nY^n W_n K_n} \big|+\big|Q_{X^nY^n}- \pi_{XY}^n  \big| \label{eqn:tv_common}\\
  & =   \big|P_{X^n  K_n  }-Q_{X^n  K_n} \big|+\big|Q_{X^nY^n}- \pi_{XY}^n  \big| \label{eqn:tv_common2}\\
    & =   \big|\pi_X^n Q_{K_n} -Q_{X^n  K_n} \big|+\big|Q_{X^nY^n}- \pi_{XY}^n  \big|, \label{eqn:tv_common3} 
 \end{align}
 where \eqref{eqn:tri} follows from the triangle inequality for the TV distance, \eqref{eqn:tv_common} follows from the fact that the TV distance between joint distributions is at least as large as the TV distance  between marginal distributions, \eqref{eqn:tv_common2} follows from the fact that $P_{Y^n W_n |X^n K_n}=Q_{Y^n W_n |X^n K_n}$ by the construction in~\eqref{eqn:syn_exact2} and finally, \eqref{eqn:tv_common3} follows from the definition of $P_{X^n Y^n W_n K_n} $ in \eqref{eqn:syn_exact2}. We can  now take expectations on both sides of  the above chain of inequalities.  The expectations of both terms in~\eqref{eqn:tv_common3} vanish due to~\eqref{eqn:tv1} and~\eqref{eqn:tv2}, which means that the synthesized distribution $P_{X^n Y^n}$ is arbitrarily close in TV distance to the target distribution  $\pi_{XY}^n$ as $n\to\infty$ if \eqref{eqn:rate_constraint1} and \eqref{eqn:rate_constraint2} are satisfied. 
\section{Exact Channel Synthesis} \label{sec:exactCS}
In this section, we consider an exact synthesis counterpart to that considered in Section~\ref{sec:approxCS}. That is, we require that the decoder in Fig.~\ref{fig:dist_chsyn} outputs a sequence of random variables $Y^n$ whose joint distribution with the source sequence $X^n$ is {\em exactly} $\pi_{XY}^n$. Just as we discussed in Section~\ref{ch:exact} on exact common information, to ensure exact reconstruction, one has to be given the freedom to use {\em variable-length} codes. In this channel synthesis setting, the notion of variable-length codes is parallel to that in Section~\ref{ch:exact} albeit slightly  more involved due to the presence of the common randomness $K_n$.  Our objective is to compare and contrast the optimal rate regions for approximate and exact reconstructions of the channel $\pi_{Y|X}$. 

Formally, let the alphabet of the common randomness $K_n$ be $\calK_n = [2^{nR_0}]$. In other words, $K_n$ can be represented by $nR_0$ bits and this length is kept {\em fixed}. The length that is allowed to vary  is that of $W_n \sim P_{W_n | X^n K_n}(\cdot| x^n,k)$ whose alphabet we denote by $\calW_n$. This alphabet, without loss of generality, can be regarded as a subset of $\bbN$. We now consider a set
of source codes  $\mathbf{f}=\{f_k : k\in\calK_n\}$, where each element of~$\mathbf{f}$ is a prefix-free code~\cite{Cov06}  $f_k:\calW_n\to\{0,1\}^*$ indexed by $k\in\calK_n$. Then for each message-common randomness pair $(w,k)\in\calW_n\times\calK_n$, and the set of codes $\mathbf{f}$, let $\ell_{\mathbf{f}}(w|k)$ denote the  {\em length} of the codeword $f_k(w)$ (see Example~\ref{ex:length}) where $f_k$ is the $k^{\mathrm{th}}$ component of $\mathbf{f}$. %$(w,k)\in\calW_n\times\calK_n$ %and the sequence of %codes $\mathbf{f}$, let $\ell_{\mathbf{f}}(w|k)$ denote the length of the codeword $f_k(w)$, where $f_k$ is the $k^{\mathrm{th}}$ component of $\mathbf{f}$. 

\begin{definition} \label{def:length_ecs}
The {\em expected codeword length} $L_{\mathbf{f}}$  of a code $\mathbf{f}=\{f_k : k\in\calK_n\}$ for compressing the source $W_n$ given $K_n$  is 
\begin{equation}
\!L_{\mathbf{f}}(W_n|K_n) \!=\!\bbE\big[\ell_{\mathbf{f}}(W_n|K_n)\big] \!= \!\sum_{(w,k)\in\calW_n\times\calK_n} \!\! P_{W_nK_n}(w,k) \ell_{\mathbf{f}}(w|k), 
\end{equation}
where the joint distribution between the message and  uniformly distributed common randomness (i.e., $P_{K_n}(k)=|\calK_n|^{-1}$ for all $k\in\calK_n$) is
\begin{equation}
P_{W_n K_n}(w,k)=\sum_{x^n\in\calX^n}\frac{1}{|\calK_n|}\pi_{X}^n (x^n)P_{W_n|X^n K_n}(w | x^n,k).
\end{equation}
\end{definition}
Note that if $\calK_n=\emptyset$, this definition reduces to that in Definition~\ref{def:length} for the exact common information problem. 

\begin{definition} \label{def:vlcode_ecs}
An {\em $(n,R, R_0)$-variable-length channel simulation code} $(\mathbf{f},P_{W_n|X^n K_n},P_{Y^n | W_n K_n})$ consists of 
\begin{itemize}
%\item A distribution $P_{W_n}$ supported on a   countable set $\calW_n\subset\bbN$;
\item A set of prefix-free source codes $\mathbf{f}= \{f_k:\calW_n\to\{0,1\}^*\}_{k\in\calK_n}$;
\item A pair of random   mappings   $P_{W_n|X^n K_n}\in\calP( \calW_n |\calX^n\times\calK_n)$ and  $P_{Y^n | W_n K_n} \in\calP( \calY^n |\calW_n\times\calK_n)$ called the {\em encoder} and {\em decoder} respectively; 
\end{itemize}
such that the per-symbol expected codeword length 
\begin{equation}
\frac{1}{n}  L_{\mathbf{f}}(W_n |K_n) \le R, \label{eqn:expected_length_ecs}
\end{equation}
and the rate of the common randomness $|\calK_n|$ satisfies
\begin{equation}
\frac{1}{n} \log|\calK_n| \le R_0. \label{eqn:cr_size}
\end{equation}
\end{definition}
By the variable-length nature of the code, $W_n$ can be transmitted to the decoder error-free. The {\em synthesized channel} is then  given by 
\begin{equation}
P_{Y^n|X^n}(y^n|x^n) \!= \!\sum_{ (w,k) \in \calW_n\times\calK_n } \!\frac{P_{W_n|X^n K_n}(w|x^n, k) P_{Y_n|W_n K_n} (y^n|w,k)}{|\calK_n|}. \label{eqn:syn_con_dis}
\end{equation}
In the exact channel synthesis problem we consider in this section, $P_{Y^n |X^n}$ is required to be {\em exactly} equal to $\pi_{Y|X}^n$ for some large enough $n$.  It is worth noting that under the assumption that $K_n\sim\mathrm{Unif}(\calK_n)$, the synthesized channel depends only on the code $(P_{W_n|X^n K_n},P_{Y_n|W_n K_n})$ and not the   distribution $\pi_X^n$.  However, the code rate $R$ induced by a given channel simulation code (Definition~\ref{def:vlcode_ecs})  depends on $\pi_X$. 

% However, the code rate  $R$ induced by a given synthesis code   depends on $\pi_X$. 

\begin{definition} \label{def:ECS}
The pair $(R,R_0) \in \bbR_+^2$ is said to be {\em achievable for exactly  synthesizing the channel $\pi_{Y|X}$ with input $\pi_X$} if there exists an  $(n,R,R_0)$-variable-length   channel synthesis code  such that  the synthesized distribution in \eqref{eqn:syn_con_dis} and the target distribution $\pi_{Y|X}^n$ are equal, i.e.,  
\begin{equation}
 P_{Y^n|X^n}= \pi_{Y|X}^n  \quad\mbox{for some }   n\in\bbN. \label{eqn:exact_ecs}
\end{equation}
Define the {\em optimal rate region} $\calT_{\Ex}(\pi_{XY}) \subset\bbR_+^2$ to be the closure of the set of achievable rate pairs $(R, R_0)$ for exactly synthesizing $\pi_{Y|X}$ with input $\pi_X$. 
\end{definition}
The central goal of this section is to characterize $\calT_{\Ex}(\pi_{XY})$ for various sources $(X,Y)\sim\pi_{XY}=\pi_X \pi_{Y|X}$.

We first perform a simple observation that is parallel to that of Lemma~\ref{lem:common_ent} for the exact common information problem. Observe from the law of total expectation that $L_{\mathbf{f}}(W_n|K_n)  =\bbE[ \bbE[  \ell_{\mathbf{f}}(W_n|K_n)  | K_n ]]$. Hence, to minimize the expected codeword length $L_{\mathbf{f}}(W_n|K_n)$, it suffices to minimize $\bbE[  \ell_{\mathbf{f}}(W_n|K_n)  | K_n =k]$ for each $k \in\calK_n$. By applying Shannon's zero-error compression theorem for every $k$, we have the bounds
\begin{equation}
H(W_n|K_n=k)\le \bbE[  \ell_{\mathbf{f}}(W_n|K_n)  | K_n =k] < H(W_n|K_n=k)+1. \label{eqn:bounds_length_exp}
\end{equation}
Hence, by taking the expectation over $K_n$, for a set of optimal prefix-free codes $\mathbf{f}^*= \{f_k^* : k \in\calK_n\}$, one has 
\begin{equation}
H(W_n|K_n ) \le L_{\mathbf{f}^*}(W_n|K_n) < H(W_n|K_n )+1.
\end{equation}
Consequently, 
\begin{equation}
\lim_{n\to\infty} \frac{ L_{\mathbf{f}^*}(W_n|K_n) }{n}= \lim_{n\to\infty}\frac{H(W_n|K_n)}{n}. \label{eqn:limitKH}
\end{equation}
Hence, completely analogous to Lemma~\ref{lem:common_ent}, we have the following multi-letter characterization of $\calT_{\Ex}(\pi_{XY})$; this is due to the present authors~\cite{YuTan2020b}. 
\begin{lemma}\label{lem:common_ent_ecs}
The optimal rate region for the exact channel synthesis problem is 
\begin{equation}
\calT_{\Ex}(\pi_{XY}) \!=\!  \mathrm{Cl} \left(\bigcup_{n\in\bbN } \left\{  (R, R_0)   :   \parbox[c]{1.67 in}{$ \exists\, (P_{W_n|X^n K_n},P_{Y^n | W_n K_n})$   \vspace{0.03 in}\\ $P_{Y^n |X^n} =\pi_{Y|X}^n$  \vspace{0.03 in}\\ $R\ge \frac{1}{n} H(W_n|K_n)$  }  \right\} \right). 
\end{equation}
\end{lemma}
Because of \eqref{eqn:limitKH}, the multi-letter expression presented in Lemma~\ref{lem:common_ent_ecs} does not depend on the set of prefix-free codes $\mathbf{f}$ and thus $\mathbf{f}$ may be omitted from Definition~\ref{def:vlcode_ecs} in our consideration of the optimal rate region (per Definition~\ref{def:ECS}). We notice that the limit of $H(W_n|K_n)/n$ can be interpreted as the  {\em conditional common entropy rate} of the process $\{W_n\}_{n\in\bbN}$ given another process $\{K_n\}_{n\in\bbN}$. While this lemma presents a characterization of $\calT_{\Ex}(\pi_{XY})$, it is far from explicit and intractable to calculate given a $\pi_{XY}$. In the following, we present alternative characterizations of and bounds on~$\calT_{\Ex}(\pi_{XY})$ that are  more explicit. 
\section{Multi-Letter Characterization for Exact Channel Synthesis} \label{sec:multiLetterECS}
In this section, we present an alternative multi-letter characterization in terms of the maximal cross-entropy defined (see \eqref{eqn:maximal_cross_ent0} in Definition~\ref{def:max_cross_ent}), which as we have seen from Sections \ref{ch:renyi} and \ref{ch:exact}, plays a crucial role in the characterization of fundamental limits of common information problems when {\em exact} reconstruction is required.  To do so, we define  $\ucalR(\pi_{XY} )$ to be the set of   rate pairs $(R, R_0)\in\bbR_+^2$ such that there exists a joint distribution $P_W P_{X |W} P_{Y |W}$ with $P_{XY}=\pi_{XY}$ and 
\begin{align}
R &\ge I(W;X )  \label{eqn:ecs_ml1}\\
 R_0+R &\ge  -H(XY;W) \!+\! \bbE_{P_W} \big[ \rvH_{\infty}( P_{X|W}, P_{Y|W} \|\pi_{XY}) \big].\label{eqn:ecs_ml2}
\end{align}
%where we remind the reader that the expectation in \eqref{eqn:ecs_ml2} can be written  as $\sum_{w\in\calW} P_W(w)  \rvH_{\infty}( P_{X^n|W=w}, P_{Y^n|W=w} \|\pi_{XY}^n)$.
%\begin{equation}
%%\ucalR(\pi_{XY}) \!=\!   \left\{  (R, R_0)   :   \parbox[c]{2.0 in}{$ \exists\, (P_W, P_{X|W},P_{Y|W})$   \vspace{0.03 in}\\ $P_{XY}=\pi_{XY}$  \vspace{0.03 in}\\ $R\ge I(W;X)$ \vspace{0.03 in}\\ $R_0+R\ge -H(XY;W) +%\bbE \big[ \rvH_{\infty}( P_{X|W}, P_{Y|W} \|\pi_{XY}) \big]$  }  \right\}  .
%\ucalR(\pi_{XY}) \!=\!   \left\{   (R,R_0): \parbox[c]{3.3 in}{  
%$\exists \, (P_W, P_{X|W} , P_{Y|W})$ such that $P_{XY}=\pi_{XY}$ \vspace{0.03 in} \\
%$R\ge I(W;X)$\vspace{0.03 in}\\ 
%$R_0+R\ge  -H(XY;W) + \bbE \big[ \rvH_{\infty}( P_{X|W}, P_{Y|W} \|\pi_{XY}) \big]$
% } \right\}
%\end{equation}
To exhaust the region $\ucalR(\pi_{XY})$, it suffices to take $|\calW|\le|\calX||\calY|+1$. The bound in \eqref{eqn:ecs_ml2} is analogous to the  upper pseudo-common information of order $\infty$ (Definition~\ref{def:pseudo_CI}). The following theorem is also parallel to~\eqref{eqn:multiletter_equiv} in Theorem~\ref{thm:equivalence} and is due to the present authors~\cite{YuTan2020b}.

\begin{theorem} \label{thm:ecs_ml}
For a source with distribution $\pi_{XY}$ defined on a finite alphabet $\calX\times\calY$,
\begin{equation}
\calT_{\Ex}(\pi_{XY}) = \mathrm{Cl}\Bigg( \bigcup_{n\in\bbN} \frac{1}{n} \ucalR  \big(\pi_{XY}^n\big) \Bigg).
\end{equation}
\end{theorem}

The intuition for the achievability part of this  result (i.e., that $\calT_{\Ex}(\pi_{XY})\supset \frac{1}{n}\ucalR  \big(\pi_{XY}^n\big)$ for any $n\in\bbN$) is similar to that sketched in Section~\ref{sec:intuition_ub} for the exact common information problem.  In particular, the constraint in~\eqref{eqn:ecs_ml2}  results from the fact that we can use the {\em pair} $(W_n, K_n) $ as the common randomness for the exact synthesis of $\pi_{XY}^n$. Recall that the exact common information is equal to the R\'enyi common information of order $\infty$ (Theorem~\ref{thm:equivalence}). Hence, we require $D_\infty( P_{X^nY^n}\|\pi_{XY}^n)$ to vanish. This is equivalent to 
\begin{equation}
\max_{(x^n,y^n) \in\supp(P_{X^n Y^n})} \frac{P_{X^nY^n}(x^n, y^n)}{\pi_{XY}^n(x^n, y^n)}  =1+o(1).
\end{equation}
 According to the discussion in Section~\ref{sec:intuition_ub}, for this condition to hold using truncated i.i.d.\ codes within a mixture decomposition framework, we need the  total rate of the available common randomness $R_0+R$ to satisfy~\eqref{eqn:ecs_ml2}. 
%\begin{equation}
%R_0+R\ge-H(XY|W)+ \bbE_W\big[\rvH_{\infty} (P_{X|W}, P_{Y|W}\|\pi_{XY} )\big].
%\end{equation}
 The constraint that $R\ge I(W; X)$  in~\eqref{eqn:ecs_ml1} is  similar, albeit simpler. It is required to ensure that $\pi_X^n$ is close to $P_{X^n}$ in the sense that 
\begin{equation}
\max_{x^n \in\calT_\epsilon^{(n)}(\pi_X)} \frac{P_{X^n}(x^n)}{\pi_X^n(x^n)}=1+o(1).
 \end{equation} 
 %This is a standard source resolvability or soft-covering problem (e.g., in the relative entropy sense). 
Putting these ideas together yields the fact that $\calT_{\Ex}(\pi_{XY})\supset \ucalR  (\pi_{XY})$. Using the above coding scheme and following steps similar to the approximate synthesis case (i.e., the proof of Theorem~\ref{thm:dist_chsyn}) on source blocks of length $n$ yields  that $\calT_{\Ex}(\pi_{XY})\supset \frac{1}{n}\ucalR  \big(\pi_{XY}^n\big)$ for all $n\in\bbN$, which is the achievability part of  Theorem~\ref{thm:ecs_ml}. 

\section{Single-Letter Bounds for Exact Channel Synthesis} \label{sec:singleletterCS}
In this section, we present single-letter  inner and outer bounds on the optimal rate region for exact channel synthesis $\calT_{\Ex}(\pi_{XY})$.  

To state the bounds succinctly, we present a definition that is analogous to $\ucalR(\pi_{XY})$ in~\eqref{eqn:ecs_ml1}--\eqref{eqn:ecs_ml2}. Let $\ocalR(\pi_{XY})$ be the set of   rate pairs $(R, R_0)\in\bbR_+^2$ such that there exists a joint distribution $P_W P_{X |W} P_{Y |W}$ with $P_{XY}=\pi_{XY}$ satisfying \eqref{eqn:ecs_ml1} and  additionally, 
\begin{align}
 \hspace{-.2in}
 R_0+ R & \ge  -H(XY;W) \nn\\*
  \hspace{-.2in}& \qquad+\inf_{\substack{Q_{WW'}\\ \in \calC(P_W,P_W)}}  \bbE_{Q_{WW'}}\big[ \rvH_{\infty}(P_{X|W},P_{Y|W'}\|\pi_{XY})\big]. \label{eqn:ecs_ml_u2}
\end{align}
This bound is analogous to the  lower pseudo-common information of order $\infty$ (Definition~\ref{def:pseudo_CI}). 
The difference between $\ocalR(\pi_{XY})$ and $\ucalR(\pi_{XY})$  is also similar  to the difference between $\uGamma(\pi_{XY})$ and $\oGamma(\pi_{XY})$ defined in~\eqref{eqn:uGamma_inf} and \eqref{eqn:Gamma_fn} respectively.  Clearly, $\ucalR(\pi_{XY})\subset\ocalR(\pi_{XY})$ and equality is achieved, for example, when every  point  on the boundary of $\ucalR(\pi_{XY})$  induces an  optimal coupling  in~\eqref{eqn:ecs_ml_u2} that is the equality coupling, i.e., $Q_{WW'}(w,w') = P_W(w) \bone\{w=w'\}$ for all $(w,w')\in\calW^2$. 

The following theorem, due to the present authors~\cite{YuTan2020b}, is analogous to Theorem~\ref{thm:exact_sl} for the exact common information problem.  
\begin{theorem}[Bounds on exact channel synthesis region] \label{thm:ecs_sl}
For a source with distribution $\pi_{XY}$ defined on a finite alphabet $\calX\times\calY$,
\begin{equation}
\ucalR(\pi_{XY})\subset \calT_{\Ex}(\pi_{XY}) \subset \ocalR(\pi_{XY})\cap \calC_{\Wyner}(\pi_{XY}).
\end{equation}
\end{theorem}
\begin{remark}
To alleviate any possible confusion, we remark that in Theorem \ref{thm:exact_sl} in which the optimal rate (exact common information) $T_{\Ex}(\pi_{XY})$ is sought, the {\em achievability} part (resp.\ converse part) corresponds to the {\em upper} bound  (resp.\ lower bound) on $T_{\Ex}(\pi_{XY})$. In contrast,   in Theorem \ref{thm:ecs_sl} in which the rate region   $\calT_{\Ex}(\pi_{XY})$ is sought, the {\em achievability} part (resp.\ converse part) corresponds to the {\em inner} bound  (resp.\ outer bound) on $T_{\Ex}(\pi_{XY})$. 
\end{remark}
The difference between the inner bound $\ucalR(\pi_{XY})$ and $\calC_\Wyner(\pi_{XY})$ is the bound on the sum rate. In the former, the sum rate is lower bounded by $-H(XY;W)+\bbE_{P_W} \big[ \rvH_{\infty}( P_{X|W}, P_{Y|W} \|\pi_{XY}) \big]$ (see \eqref{eqn:ecs_ml2}) while for the latter (see \eqref{eqn:calC}), the sum rate is lower bounded by $I(XY;W)$. It can easily be seen that the inner bound is indeed a subset of $\calC_\Wyner(\pi_{XY})$. This is because 
\begin{align}
&\sum_{w\in\calW}P_W(w)\rvH_{\infty}(P_{X|W=w},P_{Y|W=w}\|\pi_{XY}) \nn\\*
& \ge\sum_{w\in\calW}P_W(w) \sum_{x,y}P_{X|W }(x|w)P_{Y|W }(y|w)\log\frac{1}{\pi_{XY}(x,y)}\\
&= \sum_{x,y} P_{XY}(x,y)\log\frac{1}{\pi_{XY}(x,y)}=H_\pi( XY ), \label{eqn:lower_bd_Hinfty}
\end{align}
where the final equality follows from the fact that $P_{XY}=\pi_{XY}$.  As a result, 
\begin{align}
I(XY;W)& = H(XY) - H(XY;W) \\
&\le \bbE_{P_W} \big[ \rvH_{\infty}( P_{X|W}, P_{Y|W} \|\pi_{XY}) \big] - H(XY;W).
\end{align}
Thus the lower bound on the sum rate in $\ucalR(\pi_{XY})$ is at least as large as that on the sum rate in $\calC_{\Wyner}(\pi_{XY})$, which implies that  $\ucalR(\pi_{XY}) \subset \calC_{\Wyner}(\pi_{XY})$. 

\subsection{Ideas for the Proof of Theorem \ref{thm:ecs_sl}}
In this section, we provide brief sketches of the set inclusions in Theorem~\ref{thm:ecs_sl}; this section can be omitted at a first reading.  In Section~\ref{sec:multiLetterECS}, we have already provided a sketch of the proof  that $\ucalR(\pi_{XY})\subset\calT_{\Ex}(\pi_{XY})$  (the achievability part) so we proceed to show the other inclusions.  

Let us now reason that $\calT_{\Ex}(\pi_{XY})\subset \calC_{\Wyner}(\pi_{XY})$. By the  chain of inequalities leading to \eqref{eqn:lower_bd_Hinfty}, we see that 
\begin{equation}
\sum_{w\in\calW}P_W(w)\rvH_{\infty}(P_{X^n|W=w},P_{Y^n|W=w}\|\pi_{XY}^n)\ge n H_\pi(XY)
\end{equation}
and so 
\begin{equation}
 I(X^n Y^n;W)\le - H(X^nY^n|W) +\bbE_{P_W}\big[ \rvH_{\infty} (P_{X^n|W},P_{Y^n|W}\|\pi_{XY}^n)\big]. 
\end{equation}
Hence, $\ucalR(\pi_{XY}^n)\subset \calC_{\Wyner}(\pi_{XY}^n)$. Furthermore, \citet{cuff13} showed that the set $\calC_{\Wyner}(\pi_{XY}^n)$ {\em tensorizes}, i.e., $\frac{1}{n}\calC_{\Wyner}(\pi_{XY}^n)=\calC_{\Wyner}(\pi_{XY} )$. Thus, by the multi-letter characterization of $\calT_{\Ex} (\pi_{XY})$ in terms of the union of the sets $\frac{1}{n}\ucalR(\pi_{XY}^n)$ for $n\in\bbN$ in Theorem~\ref{thm:ecs_ml}, we see that $\calT_{\Ex} (\pi_{XY})\subset \calC_{\Wyner}(\pi_{XY})$. This bound is completely analogous to the fact that the exact common information is at least as large as Wyner's common information; recall the derivation of this in \eqref{eqn:exact_ge_wyner}. 

Thus, it remains to prove the alternative outer bound $\calT_{\Ex} (\pi_{XY})\subset \ocalR(\pi_{XY})$ (converse). This requires a single-letterization result in \cite[Theorem~2]{YuTan2020_exact} which is restated here for the reader's convenience. 
\begin{lemma}\label{lem:sl_lb}
For a triple of random variables $(X^n, Y^n, Z) \in\calX^n\times\calY^n\times\calZ$ such that $(X^n, Y^n) \sim \pi_{XY}^n$ and $X^n-Z-Y^n$, we have 
\begin{align}
 &\hspace{-.3in}- \frac{1}{n} H(X^n Y^n|Z)+\frac{1}{n}\sum_{z \in\calZ} P_{Z}(z) \rvH_{\infty}(P_{X^n |Z=z}, P_{Y^n |Z=z}\|\pi_{XY}^n)\\*
 &\hspace{-.3in}\ge -H(XY|W) + \inf_{\substack{Q_{WW'}\\ \in \calC(P_W,P_W)}}  \bbE_{Q_{WW'}}\big[ \rvH_{\infty}(P_{X|W},P_{Y|W'}\|\pi_{XY})\big],   \label{eqn:sl_lb}
\end{align}
where $W:=(Z ,J, X^{J-1}, Y^{J-1})$, $X:=X_J$, $Y:=Y_J$, and $J\sim\mathrm{Unif}[n]$ denotes a random index independent of $(X^n, Y^n, Z)$. 
\end{lemma}
This lemma says that   the $n$-letter expression whose expectation is over a single distribution $P_{Z}$ can be lower bounded by a single-letter expression at the additional ``cost'' of an optimization over couplings $Q_{WW'}\in\calC(P_W,P_W)$.  From this lemma, we see that the multi-letter expression of the sum rate in $\ucalR(\pi_{XY}^n)$ can be lower bounded by the single-letter expression in \eqref{eqn:sl_lb}. On the other hand, observe that 
\begin{align}
P_{W_n K_n X^i Y^{i-1}}  &= P_{W_n K_n}P_{X^i|W_n K_n}P_{Y^{i-1} | W_n K_n}\\
&= P_{W_n K_n}P_{X_i|W_n K_n}P_{X^{i-1}|W_n K_n} P_{Y^{i-1} | W_n K_n},
\end{align}
so that $X_i -( W_n , K_n,   X^{i-1})-Y^{i-1}$ forms a Markov chain for all $i \in [n]$.  As in Lemma~\ref{lem:sl_lb}, let $J \sim \mathrm{Unif}[n]$ be a  random index independent of the random variables in $(X^n, Y^n , W_n, K_n)$ and let $X:=X_J$, $Y:=Y_J$, and $W:=(W_n, K_n, X^{J-1} ,Y^{J-1}, J )$.  Hence, 
\begin{align}
nR &\ge H(W_n|K_n) \\
&\ge I( X^n;W_n|K_n) \\
&= I( X^n;W_n,K_n) \\
& = \sum_{i=1}^n I( X_i;W_n K_n|X^{i-1}) \\
& = \sum_{i=1}^n I( X_i;W_n K_n X^{i-1}) \label{eqn:source_iid} \\
& = \sum_{i=1}^n I( X_i;W_n K_n X^{i-1}  Y^{i-1}) \label{eqn:bec_mc} \\
& = n  I( X_J;W_n K_n ,X^{J-1}  Y^{J-1}|J ) \\
& = n  I( X_J;W_n K_n X^{J-1} Y^{J-1} J ) \\
&=nI(X;W), 
\end{align}
where \eqref{eqn:source_iid} follows from the fact that $\{X_i\}_{i=1}^n$ is a memoryless process and \eqref{eqn:bec_mc} follows because   $X_i -( W_n  ,K_n  , X^{i-1})-Y^{i-1}$ forms a Markov chain for all $i \in [n]$. This
completes the proof that $\calT_{\Ex} (\pi_{XY})\subset \ocalR(\pi_{XY})$.

\subsection{Tradeoff Between the Communication and Common \\ Randomness Rates} \label{sec:tradeoff_rates}
We now examine the tradeoff between the communication rate $R$ and the common randomness rate $R_0$ in the optimal region for exact channel synthesis $\calT_\Ex(\pi_{XY})$. For this  purpose, define the two optimal rates
\begin{align}
T^*(R_0) &:= \inf \left\{  R\in\bbR_+:  (R,R_0) \in \calT_\Ex(\pi_{XY})  \right\} \quad\mbox{and} \label{eqn:T_starR0} \\
T_0^*(R) &:= \inf \left\{  R_0\in\bbR_+:  (R,R_0) \in \calT_\Ex(\pi_{XY})  \right\} .
\end{align}
From the inner and outer bounds in Theorem \ref{thm:ecs_sl}, we see that 
\begin{equation}
T^*(\infty)=I_\pi(X;Y), \label{eqn:T_Star}
\end{equation}
where $I_\pi$ denotes the mutual information computed with respect to the target distribution $\pi_{XY}$. 
This is because when $R_0=\infty$, the sum rate bound is inactive. Eqn.~\eqref{eqn:T_Star}   is consistent with the observation in~\citet{Bennett02}. Namely, when there is unlimited shared or common randomness  at the encoder and the decoder, the target channel $\pi_{Y|X}$ can
be successfully synthesized by some protocol if and only if the minimum
asymptotic communication rate is  at least 
the mutual information $I_\pi(X;Y)$ (refer to Table~\ref{tab:ecs}).  This is the same as approximate channel synthesis in the TV metric (refer to Fig.~\ref{fig:dist_chsyn}). More interestingly, \citet{Bennett02} showed that an {\em exponential} number of bits  of common randomness  {\em suffices} for~\eqref{eqn:T_Star} to hold. This condition is rather different from what we have  seen from~\eqref{eqn:R0_req} in the context of approximate channel synthesis in which a shared randomness rate of $H_\pi(Y|X)$ is  needed for us to ensure that the communication rate is $I_\pi(X;Y)$. \citet{Bennett02} also conjectured that an exponential number of bits of common randomness (which implies that $R_0=\infty$) is, in fact, {\em necessary} for~\eqref{eqn:T_Star} to hold. 

 \citet{Harsha10} and \citet{LiElgamal2018} disproved this conjecture for $(X, Y )\sim\pi_{XY}$ with finite alphabets. These authors showed that shared randomness with rate $\log|\calY|$ is sufficient to realize \eqref{eqn:T_Star}, i.e., 
\begin{equation}
T^*(\log |\calY|) \le I_\pi(X;Y). \label{eqn:T_Star2}
\end{equation}
The result in Theorem \ref{thm:ecs_sl}, in fact, yields a   better bound. Consider,
\begin{align}
\hspace{-.2in} T_0^*(I_\pi(X;Y))  &= \inf \left\{  R_0:  (I_\pi(X;Y),R_0) \in \calT_\Ex(\pi_{XY})  \right\} \\
\hspace{-.2in}&\le \min_{P_{W|Y}: X-W-Y}-H(X)-H(Y|W) \nn\\*
\hspace{-.2in}&\qquad\qquad+\bbE_{P_W}\big[ \rvH_{\infty}(P_{X|W}, P_{Y|W}\|\pi_{XY}) \big] \label{eqn:follows_from_inner}\\
\hspace{-.2in}&\le H_\pi(Y|X), \label{eqn:setWY}
\end{align} 
where \eqref{eqn:follows_from_inner}  results from the inner bound in Theorem~\ref{thm:ecs_sl}  (setting the communication rate as $R=I(X;Y)$) and \eqref{eqn:setWY} follows from setting $W=Y$ (so that $\bbE_{P_W} [ \rvH_{\infty}(P_{X|W}, P_{Y|W}\|\pi_{XY}) ]=H(XY)$).  We will see from Section~\ref{sec:ecs_dsbs} that the bound in \eqref{eqn:setWY} is tight for the DSBS. % A schematic of the region $\calT_\Ex(\pi_{XY})$ is shown in Fig.~\ref{fig:exact_chsys}.

%\begin{figure}[t]
%\vspace{-.25in}
%\centering
%\begin{overpic}[width=.9\textwidth]{figs/exact_chsys.pdf}
%\put(89,9){$R$}
%\put(12.5,48){$R_0$}
%\put(26,2.5){$I(X;Y)$}
%\put(48,2.5){$C_{\Wyner}(\pi_{XY})$}
%\put(66,2.5){$T_{\Ex}(\pi_{XY})$}
%\put(2,27){$H(Y|X)$}
%\put(60,35){$\calT_{\Ex}(\pi_{XY})$}
%\end{overpic}
%\vspace{-.05in}
%\caption{A schematic of the region $\calT_{\Ex}(\pi_{XY})$. Compare this region to that in Fig.~\ref{fig:dist_chsyn} and observe that in general, $ \calT_{\Ex}(\pi_{XY})\subsetneq\calC_{\Wyner}(\pi_{XY})$. }
%\label{fig:exact_chsys}
%\end{figure}

Why is  the amount of common randomness yielded by Theorem~\ref{thm:ecs_sl}     smaller than those in the works~\cite{Bennett02,Harsha10,LiElgamal2018} prior to that of the present authors? In the coding scheme of~\cite{Bennett02}, the shared randomness is used to generate a codebook. However, as described in the   sketch of the coding scheme for Theorem~\ref{thm:ecs_ml}, we use the so-called mixture decomposition technique (cf.\ Section~\ref{sec:sketch_equiv}) to construct a
variable-length exact synthesis code. This  is a mixture of
a fixed-length approximate synthesis code that   ensures the R\'enyi divergence of order $ \infty$  vanishes and a completely lossless code % ensures that the R\'enyi 
 of rate $\log( |\calX||\calY|)$ (see~\eqref{eqn:tot_rate}). The performance of this scheme is dominated by   that of  the approximate synthesis code which requires a much lower rate of shared randomness compared to the scheme in \cite{Bennett02}. Furthermore, the codes in~\cite{Harsha10} and~\cite{LiElgamal2018} are such that $Y^n$ is required to be a  {\em deterministic} function of $W_n$ and~$K_n$. In contrast, we allow $Y^n$ to be a {\em stochastic} function of $(W_n,K_n)$ (cf.\ Definition~\ref{def:vlcode_ecs}).  Hence, naturally, the rate of common randomness is reduced. 

Finally, we mention that $\calT_{\Ex}(\pi_{XY})$, in general, is a strict subset of $\calT(\pi_{XY}) = \calC_\Wyner (\pi_{XY})$. This is because, from the operational definitions,  
\begin{equation}
T^*(0) = T_{\Ex}(\pi_{XY}) ,
\end{equation}
and as we have seen from Section~\ref{sec:exact_dsbs} for the DSBS with crossover probability $p\in (0,1/2)$,
\begin{equation}
T_{\Ex}(\pi_{XY}) > C_{\Wyner}(\pi_{XY}).
\end{equation}
We   evaluate the region $\calT_{\Ex}(\pi_{XY})$ for the DSBS in Section~\ref{sec:ecs_dsbs} and compare it to $\calC_{\Wyner}(\pi_{XY})$.
\section{Symmetric Binary Erasure Sources} \label{sec:ecs_erasure}
In this section, we revisit the SBES as discussed in Section~\ref{sec:sbes}. Recall that this is a source with uniform $X \in\{0,1\}$ and that $Y$ is connected to $X$ via a binary erasure channel with   erasure probability~$p$. The joint distribution is given in \eqref{eqn:sbes_jd}. We saw that the exact common information of the SBES is equal to its Wyner's common information because Condition~\eqref{eqn:XYgivenW} in Corollary~\ref{cor:equiv}, namely that $\sum_{w\in\calW} H(X|W=w)H(Y|W=w)=0$, is satisfied.

For the SBES, \citet{cuff13} evaluated the optimal rate region for TV approximate synthesis (cf.\ Definition~\ref{def:tv_approx_syn}). Unsurprisingly, the region is the same as that for exact channel synthesis (cf.\ Definition~\ref{def:ECS}).
\begin{proposition} 
For the SBES with erasure probability $p$, we have 
\begin{align}
&\hspace{-.3in}\calT(\pi_{XY}) = \calT_\Ex(\pi_{XY}) = \calC_\Wyner (\pi_{XY})  \\
&\hspace{-.3in}=\bigcup_{  1-p\le r\le r^*} \left\{  (R, R_0)   \, :  \, \parbox[c]{2.25in}{$\hspace{.81cm} R\ge r$  \vspace{0.015 in}\\ $\displaystyle R+R_0\ge h (p)+r\bigg(  1- h \Big(\frac{1-p}{r}\Big)\bigg)$ }  \right\}, \label{eqn:TEx_sbes}
\end{align}
where $r^*=\min\{2(p-1),1\}$ and $h (\cdot)$ is the binary entropy function.
\end{proposition}
Recall  that for the SBES, the optimal distribution attaining Wyner's common information is a concatenation of a BEC with  erasure probabilities $p_1$ and a BEC-like channel with erasure probability  $p_2$ such that $(1-p_1)(1-p_2)=1-p$. The parameter $r$  in \eqref{eqn:TEx_sbes} represents the term $1-p_1$. The various terms in \eqref{eqn:TEx_sbes} are simply the evaluations of $I(W;X)$ and $I(W;XY)$ in the description of $\calC_{\Wyner}(\pi_{XY})$ with this parametrization.

\begin{figure}
\centering
\includegraphics[width = .8\columnwidth]{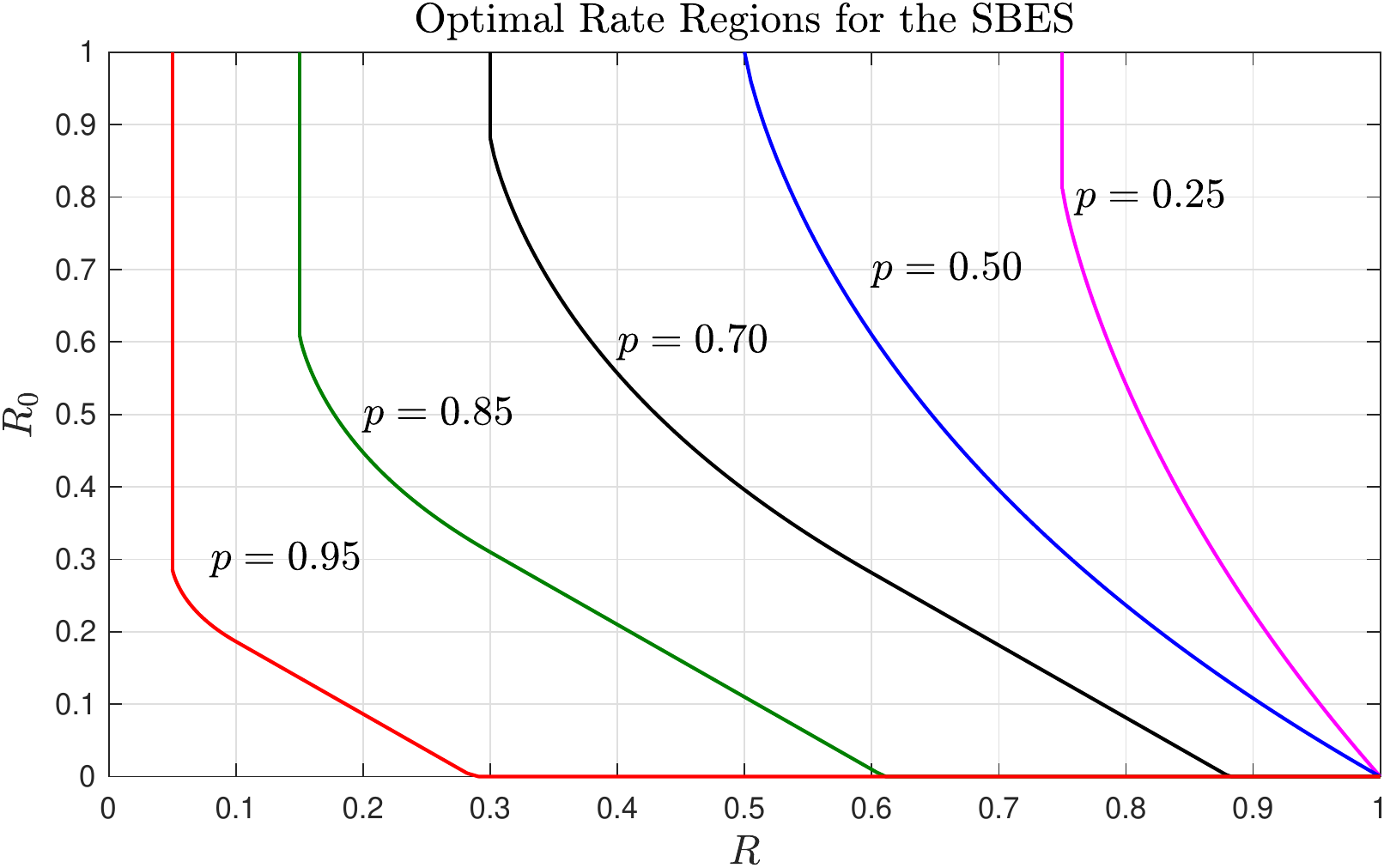}
\caption{Plot of the optimal rate regions (top right hand regions of the boundaries) for approximate and exact channel synthesis for the SBES}
\label{fig:ecs_sbes}
\end{figure}

The regions for various $p\in [0,1]$ are illustrated in Fig.~\ref{fig:ecs_sbes}. The boundaries cross the horizontal axis at $C_\Wyner(\pi_{XY})$, which is equal to $1$ for all $0\le p\le 0.5$ (see~\eqref{eqn:cwyner_sbes} and Fig.~\ref{fig:eci_sbes}). 

\section{Doubly Symmetric Binary Sources} \label{sec:ecs_dsbs}
In this section, we consider the DSBS, a prototypical example in which exact synthesis requires larger rate than approximate synthesis. This is a source in which $X$ is uniform on $\{0,1\}$ and $Y$ is connected to $X$ via a binary symmetric channel with crossover probability $p$. In Section~\ref{sec:dsbs}, we mentioned that an alternative representation of this source is in terms of the optimal distribution $P_W P_{X|W} P_{Y|W}$ that attains Wyner's common information. This takes the form $W\sim \mathrm{Bern}(1/2)$, $X = W\oplus A$ (or equivalently, $W=X\oplus A$), and $Y = W\oplus B$ where $A$ and $B$ are independent $\mathrm{Bern}(a)$ random variables such that $a \ast a = p$.  This representation is useful but in this section, as $X$ and $Y$ are not treated symmetrically in the context of distributed channel synthesis, we find it convenient to reparameterize the representation as  $X = W\oplus A$ and $Y = W\oplus B$ where $A\sim \mathrm{Bern}(a)$  and $B\sim \mathrm{Bern}(b)$ such that $a \ast b = p$ and $a,b\in (0,p)$     are not necessarily equal. For a given $a\in (0,p)$, the corresponding $b$ is 
\begin{equation}
b =\frac{p-a}{1-2a}.
\end{equation}
%Also set $\alpha= (1-p)/2$ and $\beta= p/2$ (cf.~\eqref{eqn:dsbs}).

Since the exact common information is strictly larger than that of Wyner's common information for the DSBS with crossover probability $p\in (0,1/2)$ (Proposition~\ref{prop:dsbs_exact}), the optimal rate regions in the context of distributed channel synthesis are also different; in particular  $\calT_\Ex(\pi_{XY})\subsetneq\calC_\Wyner(\pi_{XY})$ for all $p\in (0,1/2)$. %This is summarized as follows.
\begin{proposition}\label{prop:dsbs_ecs}
For the DSBS with crossover probability $p$, the optimal rate region for  TV approximate channel synthesis (Definition~\ref{def:tv_approx_syn})
\begin{align}
\!\calC_\Wyner(\pi_{XY})= \bigcup_{  0\le a\le p} \left\{ \! (R, R_0)   : \parbox[c]{2.15in}{$\hspace{.92cm} R\ge 1-h (a)$  \vspace{0.015 in}\\ $\displaystyle R+R_0\ge 1+h (p) -h (a) - h (b) $   }  \right\}. \!\label{eqn:Tapprox_dsbs}
\end{align}
The optimal rate region for  exact synthesis (Definition~\ref{def:ECS})
\begin{align}
\!\calT_\Ex(\pi_{XY}) = \bigcup_{  0\le a\le p} \! \left\{ \! (R, R_0)   : \parbox[c]{2.33in}{$\hspace{.8cm} R\ge 1-h (a)$  \vspace{0.02in}\\ $\displaystyle R+R_0\ge  \log\frac{2}{1\!-\! p}+(a+b) \log\frac{1\!-\! p}{p} $   \vspace{.0051in}\\  $ \textcolor{white}{.....................} -h (a) -h (b)$ }  \right\}\label{eqn:TEx_dsbs}.
\end{align}
\end{proposition}
These regions are illustrated in Fig.~\ref{fig:exact_syn_dsbs}. We computed $\calT_\Ex(\pi_{XY})$ by fixing a point  $R_0\in [0,1]$ on the ordinate. Then  we compute $T^*(R_0)$, defined in \eqref{eqn:T_starR0}, as 
\begin{align}
T^*(R_0)=\min_{0\le a\le p }\max \bigg\{ 1-h (a), &\,\, \log\frac{2}{1-p}+(a+b) \log\frac{1-p}{p} \nn\\*
&\qquad\quad - h (a) -h (b)-R_0\bigg\}.
\end{align}
The same can be done for  $\calC_{\Wyner}(\pi_{XY})$.  It can be seen that for $p \in (0,1/2)$,   $\calT_\Ex(\pi_{XY})\subsetneq\calC_\Wyner(\pi_{XY})$. Intuitively,
this strict inclusion is a consequence of the type overflow
phenomenon described in Section~\ref{sec:type_over}. This observation also confirms that type overflow does not affect
the fundamental limits of TV approximate channel synthesis, but it does affect  the fundamental limits of  exact  channel synthesis   in the sense of strictly increasing the optimal communication rate for a fixed common randomness rate $R_0 \in [0, H_\pi(Y|X))$.
 
Finally, if we let $R=I_\pi(X;Y) =1-h (p)$ in $\calT_\Ex(\pi_{XY})$, we get that $a=p$ and $b=0$. Hence, the   rate of the common randomness is lower bounded as $R \ge h (p)$. Combining this with~\eqref{eqn:setWY} shows that 
\begin{equation}
T_0^*(1-h (p)) = h (p) = H_\pi(Y|X).
\end{equation}
Thus, for the DSBS, we have identified the optimal rate of the common randomness when the  communication rate  approaches its optimal value $I_\pi(X;Y)=1-h(p)$. In other words, for the DSBS,~\eqref{eqn:setWY} is tight. In fact, one can see that this critical rate  $H_\pi(Y|X)$ is the same as that for {\em approximate} channel synthesis (see Fig.~\ref{fig:dist_chsyn} and~\eqref{eqn:R0_req}).\\

\begin{figure}[!ht]
%\vspace{-.25in}
\centering
\begin{overpic}[width=.8\columnwidth]{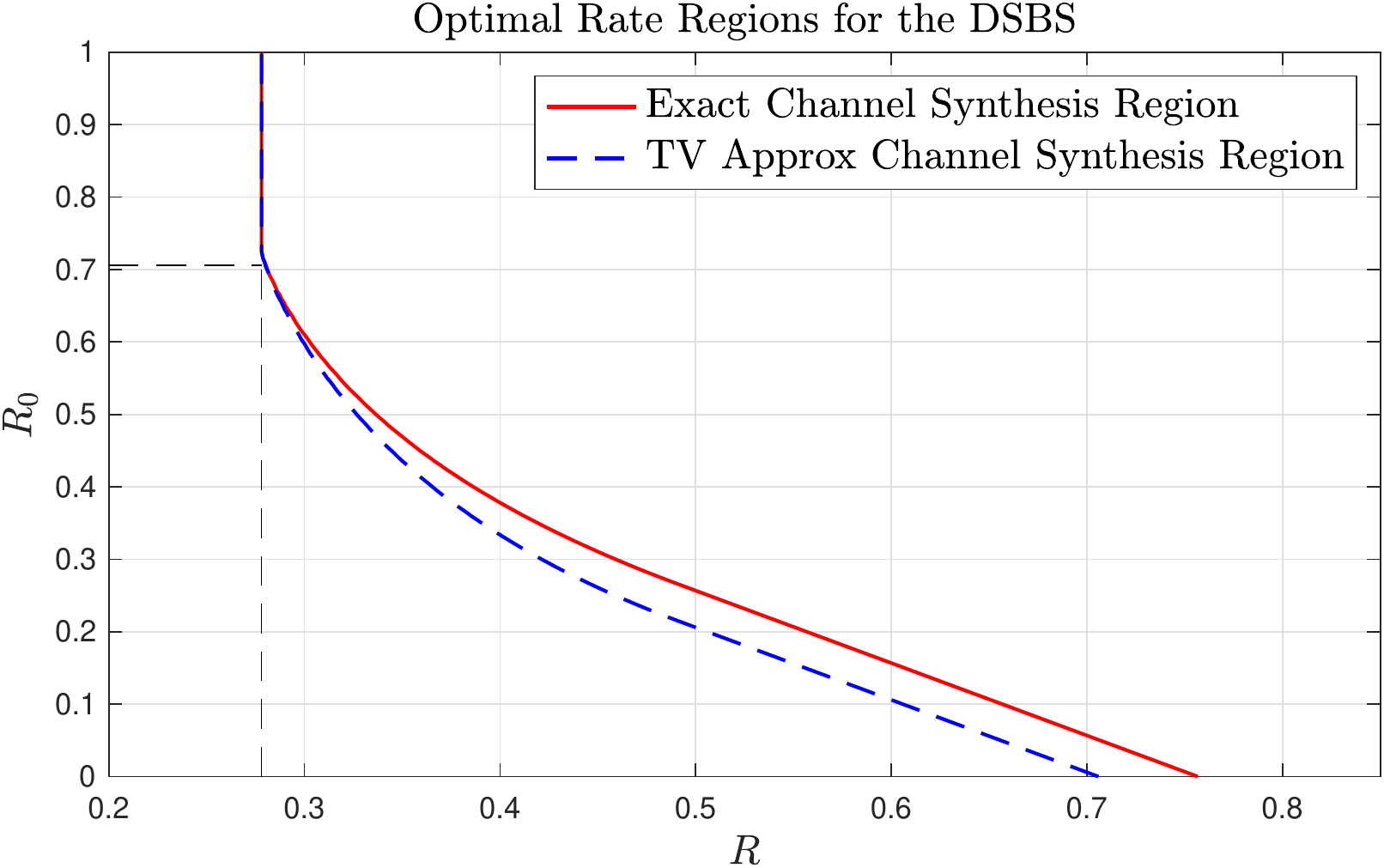}
\put(19,6.5){\circle*{1}}

\put(87,6.5){\circle*{1}}
\put(87,27){\vector(0,-1){20}}
\put(27,14){\vector(-1,-1){7}}

\put(80,6.5){\circle*{1}}
\put(80,27){\vector(0,-1){20}}
%\put(70,60){$C_{\Wyner}(\pi_{XY})$}
{\footnotesize
\put(-9,42){$H(Y|X)$}
\put(28,14.5){$I(X;Y)$}
\put(66,29){$C_{\Wyner}(\pi_{XY})$}
\put(84,29){$T_{\Ex}(\pi_{XY})$}}
\vspace{.1in}
\end{overpic}
\caption{Plot of the optimal rate regions (top right hand regions of the boundaries) for TV approximate and exact channel synthesis for the DSBS with crossover probability~$p=0.2$}
\label{fig:exact_syn_dsbs}
\end{figure}
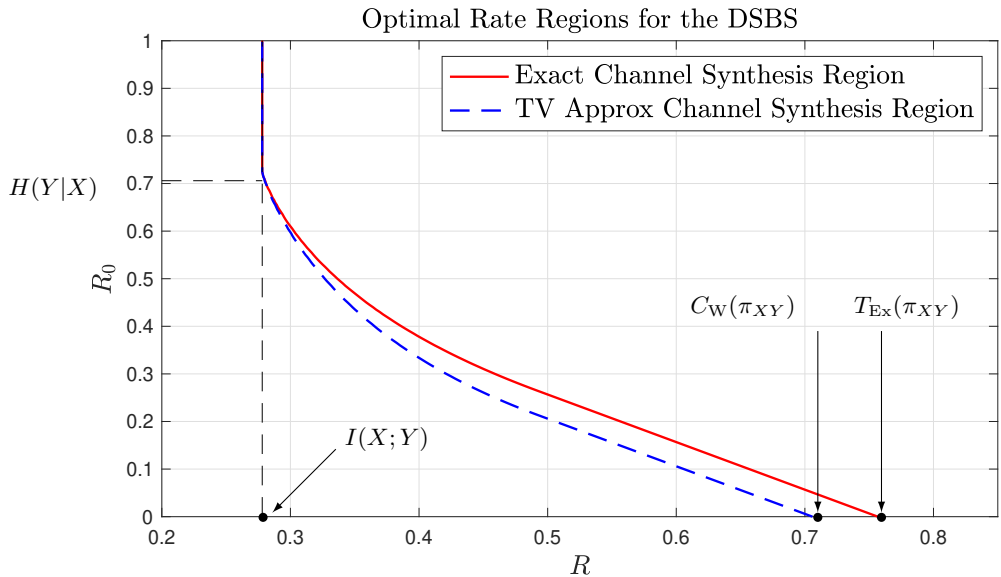

\section{Jointly Gaussian Sources} \label{sec:ecs_gau}

We conclude this section by revisiting the jointly Gaussian source $\pi_{XY}$ with correlation coefficient $\rho \in (0,1)$. The full description of the source is available in Section~\ref{sec:wyner_gauss} so we will not repeat the details here apart to remind the reader that the source is assumed to have zero mean and correlation coefficient $\rho$. We saw from Section~\ref{sec:dist_ss_cts} that to evaluate Wyner's common information for sources with uncountable alphabets, it should satisfy some regularity conditions (cf.\ Proposition~\ref{prop:reg_wyner_cts}). Fortunately, the ubiquitous and canonical jointly Gaussian source satisfies these regularity conditions. As a result, Wyner's common information in terms of the information expression in~\eqref{eqn:CWyner_inf} can be evaluated and interpreted as the minimum rate of the description so that $\pi_{XY}$ can be simulated in a distributed fashion.

For the distributed channel synthesis problem, a similar set of regularity conditions~\cite[Corollary~2]{YuTan2020b} has to be verified to ensure that optimal rate region under which the TV distance between the synthesized distribution $P_{X^nY^n}$ and the target distribution $\pi_{XY}^n$ vanishes is equal to $\calC_\Wyner(\pi_{XY})$ in \eqref{eqn:calC}. Jointly Gaussian distributions do indeed satisfy these regularity conditions, yielding the following proposition.
\begin{proposition}
For the jointly Gaussian source with correlation coefficient $\rho \in (0,1)$, the optimal rate region for  TV approximate channel synthesis $\calT(\pi_{XY}) =\calC_\Wyner(\pi_{XY})$ (Definition~\ref{def:tv_approx_syn}) is the  set of rate pairs $(R , R_0) \in\bbR_+^2$ satisfying  
\begin{align}
 R &\ge \frac{1}{2}\log\bigg(\frac{1}{1-\alpha^2}\bigg)\label{eqn:acs_gauss1}\qquad\mbox{and}\\
 R+R_0&\ge  \frac{1 }{2}\log\bigg( \frac{1-\rho^2}{(1-\alpha^2)(1-\beta^2)} \bigg)  \label{eqn:acs_gauss2}%+\frac{\rho\sqrt{ (1-\alpha^2)(1-\beta^2) }}{1-\rho^2}
\end{align}
for some $\alpha\in [\rho,1]$ and $\beta = \rho/\alpha$. 
%\begin{align}
% \calC_\Wyner(\pi_{XY}) = \bigcup_{  \rho\le \alpha\le 1} \left\{ \! (R, R_0)   : \parbox[c]{2.2in}{$\hspace{.81cm} \displaystyle R\ge \frac{1}{2}\log\bigg(\frac{1}{1-\alpha^2}\bigg)$  \vspace{0.015 in}\\ $\displaystyle R+R_0\ge \frac{1 }{2}\log\bigg( \frac{1-\rho^2}{(1\!-\!\alpha^2)(1\!-\!\beta^2)} \bigg) $   }  \right\}, \label{eqn:acs_gauss}
%\end{align}
%where $\beta = \rho/\alpha$. 
\end{proposition}
Recall that for Wyner's common information, the optimal distribution $P_W P_{X|W} P_{Y|W}$ takes the form that $W$ is standard Gaussian and $X$ and $Y$ are connected to $W$ as 
\begin{equation}
X=\sqrt{\rho}\, W+\sqrt{1-\rho}\, N_1\quad\mbox{and}\quad Y=\sqrt{\rho}\, W+\sqrt{1-\rho}\, N_2,
\end{equation}
where $N_1$ and $N_2$ are independent standard Gaussian random variables. For the distributed channel synthesis problem, $X$ and $Y$ are not treated symmetrically and so we have to consider a different and more general parametrization of $P_W P_{X|W} P_{Y|W}$.  Similarly to the DSBS  in Section~\ref{sec:ecs_dsbs}, the channels  from $W$ to $X$  and  from $W$ to $Y$   are respectively 
\begin{equation}
X = \alpha\, W+ \sqrt{1-\alpha^2}\, N_1\quad \mbox{and}\quad Y=\beta\, W+\sqrt{1-\beta^2}\, N_2,
\end{equation}
where $\alpha\beta=\rho$. Note that $X$ and $Y$ have zero mean and unit variance. By considering this parametrization and evaluating the mutual information terms $I(W;X)$ and $I(W;XY)$, we obtain the expressions in the lower bounds in~\eqref{eqn:acs_gauss1} and~\eqref{eqn:acs_gauss2}.

Similarly to the case for the exact common information, we do not yet have a complete characterization of the optimal rate region for exact channel synthesis for jointly Gaussian sources. It is clearly the case that $\calC_\Wyner(\pi_{XY})$ constitutes an outer bound to $\calT_\Ex(\pi_{XY})$. To derive the inner bound, we have to verify a set of regularity conditions similar to those in Lemma~\ref{lem:regularity} and  to construct codes such that the conditional R\'enyi divergence of order $\infty$ satisfies
\begin{align}
D_\infty\big(P_{Y^n|X^n}\big\| \pi_{Y|X}^n \big|\tilde{\pi}_{X^n}\big)  = o\Big(\frac{1}{n}\Big), 
\end{align}
%D_\infty(P_{Y^n|X^n}\| \pi_{Y|X}^n |\tilde{\pi}_{X^n})  = o(1/n)$ 
where  the {\em truncated} distribution $\tilde{\pi}_{X^n}$ on the $X$-marginal has PDF
\begin{equation}
\tilde{f}_{X^n}(x^n) \propto \bigg(\prod_{i=1}^n f_X (x_i)\bigg)\bone\big\{ x^n\in\calA_\epsilon^{(n)} \big\},
\end{equation}
and where $f_X$ and $\calA_\epsilon^{(n)}$ are the PDF and the $\epsilon$-weakly  typical set of $\pi_X$ respectively.  By using source synthesis codes, one can construct a reliable sequence of approximate channel synthesis codes in 
%We have constructed a reliable sequence of  approximate synthesis codes  in 
 the sense that their R\'enyi divergences of order $\infty$ vanish (sufficiently rapidly); this translates to a sequence of codes that guarantees exact channel synthesis.  

%To define the inner bound, it is useful to define 
%\begin{align}
%&\Lambda_\eps(P_W,P_{X|W}, P_{Y|W}, \pi_{XY}) \nn\\*
%&:= -H(XY|W) +  \sup_{\substack{Q_{XYW} :  \\ D(Q_{WX} \| P_{WX})\le \eps \\ D(Q_{WY} \| P_{WY})\le \eps } }\sum_{w,x,y} P(w) Q(x,y|w)\log\frac{1}{\pi(x,y)}. \label{eqn:Lambda_eps}
%\end{align}
%Comparing this quantity to~\eqref{eqn:Gamma_fn}, we observe  that 
%\begin{equation}
%\oGamma(\pi_{XY})  = \min_{ \substack{P_W P_{X|W} P_{Y|W} :\\ P_{XY}=\pi_{XY}}} \Lambda_0(P_W,P_{X|W}, P_{Y|W}, \pi_{XY}).
%\end{equation}%\min_{ \substack{P_W P_{X|W} P_{Y|W} :\\ P_{XY}=\pi_{XY}}}
%Hence, $\Lambda_\eps(P_W, P_{X|W}, P_{Y|W}, \pi_{XY})$ for $\eps>0$ can be regarded as a relaxed and pre-optimized version of the upper pseudo-common information of order $\infty$. Now define 

\begin{proposition}
For the jointly Gaussian source with correlation coefficient $\rho \in (0,1)$, the optimal rate region for  exact channel synthesis (Definition~\ref{def:ECS}) satisfies 
\begin{align}
 \calT_\Ex^{(\mathrm{in})}(\pi_{XY})\subset \calT_\Ex(\pi_{XY}) \subset \calC_\Wyner(\pi_{XY}), \label{eqn:ecs_gauss0}
\end{align}
where  $\calT_\Ex^{(\mathrm{in})}(\pi_{XY})$ is the set of rate pairs $(R , R_0) \in\bbR_+^2$ satisfying  
\begin{align}
\hspace{-.3in} R &\ge \frac{1}{2}\log\bigg(\frac{1}{1-\alpha^2}\bigg)\label{eqn:ecs_gauss1}\qquad\mbox{and}\\
\hspace{-.3in} R+R_0&\ge  \frac{1 }{2}\log\bigg( \frac{1-\rho^2}{(1\!-\!\alpha^2)(1\!-\!\beta^2)} \bigg)\!+\!\frac{\rho\sqrt{ (1-\alpha^2)(1-\beta^2) }}{1-\rho^2}\log\rme \label{eqn:ecs_gauss2}
\end{align}
for some $\alpha\in [\rho,1]$ and $\beta = \rho/\alpha$. 
%\begin{align}
%\! \calT_\Ex^{(\mathrm{in})}(\pi_{XY})\!:=\! \bigcup_{  \rho\le \alpha\le 1}  \left\{ \! (R, R_0)   : \parbox[c]{2.2in}{$\hspace{.81cm} \displaystyle R\ge \frac{1}{2}\log\bigg(\frac{1}{1-\alpha^2}\bigg)$  \vspace{0.015 in}\\ $\displaystyle R+R_0\ge \frac{1 }{2}\log\bigg( \frac{1-\rho^2}{(1\!-\!\alpha^2)(1\!-\!\beta^2)} \bigg) $ \vspace{0.015 in}\\  $\displaystyle \textcolor{white}{.................}+\frac{\rho\sqrt{ (1-\alpha^2)(1-\beta^2) }}{1-\rho^2} $  }\right\}. \label{eqn:ecs_gauss}
%\end{align}
\end{proposition}

\enlargethispage{-\baselineskip}
The additional term in the inequality in~\eqref{eqn:ecs_gauss2} (over the one in~\eqref{eqn:acs_gauss2}) is analogous to the additional term of $(\rho\log\rme)/(1+\rho)$ of the upper bound on the exact common information for jointly Gaussian sources in Proposition~\ref{prop:gauss_exact}. By the intuition gleaned from the DSBS in Proposition~\ref{prop:dsbs_ecs} (in which $ \calT_\Ex(\pi_{XY})$ was  characterized exactly) and the type overflow phenomenon, we conjecture that the inner bound  $\calT_\Ex^{(\mathrm{in})}(\pi_{XY})$ is tight and there is a gap between $ \calT_\Ex(\pi_{XY})$ and $\calC_\Wyner(\pi_{XY})$, i.e., that $\calT_\Ex^{(\mathrm{in})}(\pi_{XY}) =  \calT_\Ex(\pi_{XY}) \subsetneq  \calC_\Wyner(\pi_{XY})$.  The inner and outer bounds on $\calT_\Ex(\pi_{XY})$ for a jointly Gaussian source with $\rho = 0.5$ is shown in Fig.~\ref{fig:exact_syn_gauss}.

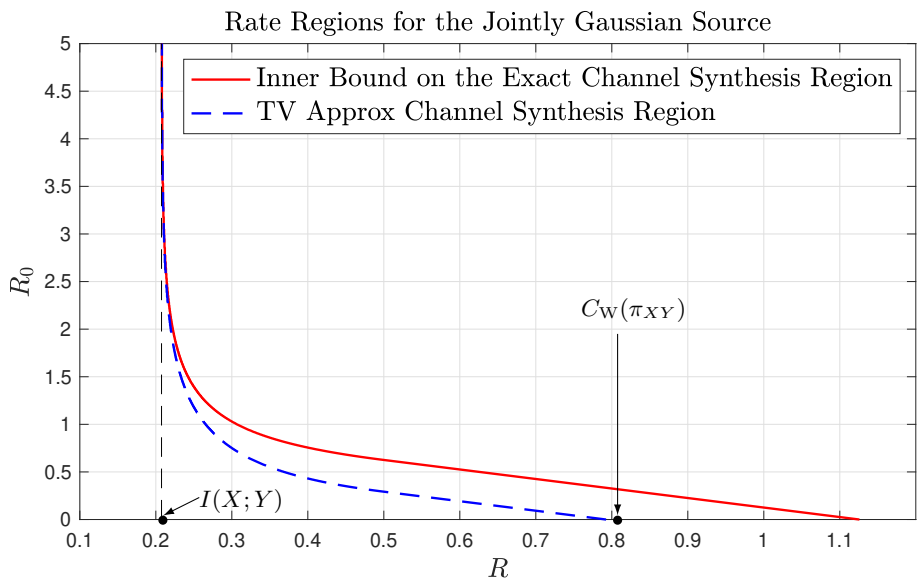
\begin{figure}[!ht]
%\vspace{-.25in}
\centering
\begin{overpic}[width=.8\columnwidth]{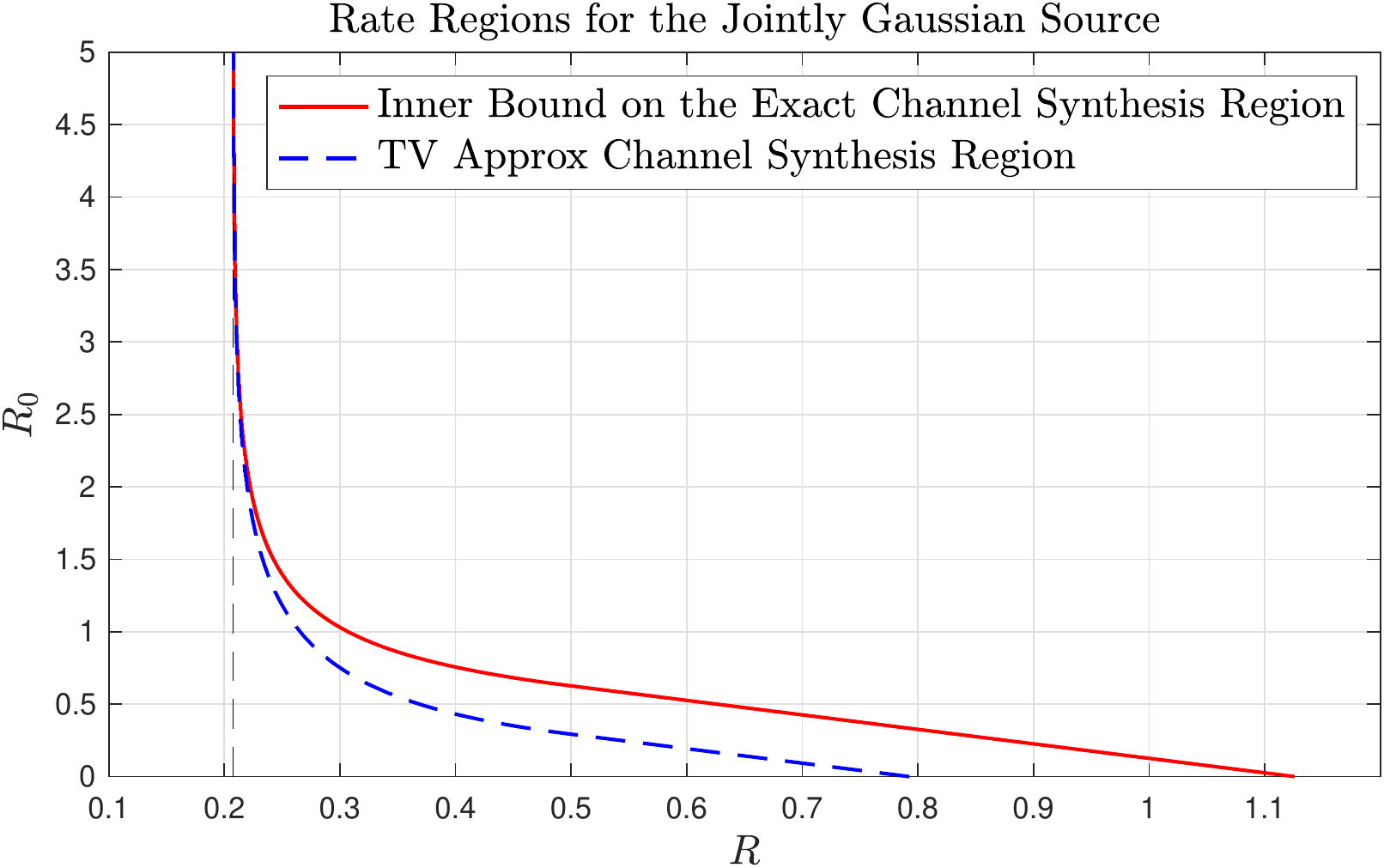}
\put(17,6.5){\circle*{1}}

%\put(87,6.5){\circle*{1}}
%\put(87,27){\vector(0,-1){20}}
\put(21,9){\vector(-2,-1){4}}

\put(67,6.5){\circle*{1}}
\put(67,27){\vector(0,-1){20}}
%\put(70,60){$C_{\Wyner}(\pi_{XY})$}
{\footnotesize
%\put(-9,42){$H(Y|X)$}
\put(21,8){$I(X;Y)$}
\put(63,29){$C_{\Wyner}(\pi_{XY})$}
%\put(84,29){$T_{\Ex}(\pi_{XY})$}
}
\vspace{.1in}
\end{overpic}
\caption{Plots of the TV approximate synthesis rate region and the inner bound to the  exact channel synthesis region for the jointly Gaussian source with correlation coefficient~$\rho=0.5$. The boundary of the optimal rate region for exact channel synthesis lies between lines. We conjecture that it coincides with the inner bound.}
\label{fig:exact_syn_gauss}
\end{figure}

%\clearpage
Finally, for a jointly Gaussian source $(X,Y)\sim\pi_{XY}$ with correlation coefficient $\rho\in [0,1)$, the mutual information between  $X$ and $Y$ is 
\begin{equation}
I_\pi(X;Y) =  \frac{1}{2}\log \frac{1}{1-\rho^2}.
\end{equation}
 Hence, in this case, under the condition that the sequence of communication rates approaches $I_\pi(X;Y)$ asymptotically, the minimum rate of shared randomness for exact channel synthesis 
\begin{equation}
T_0^*(I_\pi(X;Y))\ge\inf\big\{ R_0: (I_\pi(X;Y), R_0) \in \calC_{\Wyner}(\pi_{XY}) \big\}=\infty.
\end{equation}
This is attained when $\alpha\downarrow\rho$ and $\beta \uparrow 1$ in~\eqref{eqn:acs_gauss2}.  The same is true for TV approximate channel synthesis. Hence, for a Gaussian source, if the rate of shared randomness $R_0$ is finite, it is impossible to realize either exact or approximate channel synthesis unless the asymptotic communication rate $R$ is strictly larger than $I_\pi(X;Y)$.   
%Hence, for a Gaussian source, any finite rate of shared randomness $R_0$ can neither realize exact nor approximate channel synthesis when there is no penalty on the asymptotic communication rate $R$ over its minimum possible value~$I_\pi(X;Y)$. 
 This can also be seen from the vertical asymptote in Fig.~\ref{fig:exact_syn_gauss}.

\chapter{Common Information and Nonnegative Rank}\label{ch:nr}
This section completes our discussion of the extensions and generalizations of Wyner's common information. Instead of focusing on   coding-inspired operational interpretations of various common information measures, we describe   somewhat surprising connections between these measures and a fundamental problem in  numerical linear algebra, signal processing, and machine learning, known as {\em nonnegative matrix factorization} or NMF. The NMF problem was popularized  in a landmark paper by \citet{LS99} and has received significant attention since its inception. It has numerous applications to audio signal processing,  hyperspectral imaging, bioinformatics, and text clustering, among others.  See the excellent books by \citet{gillisBook}  and \citet{Cichocki09} for overviews.  

Simply put, in NMF, one is given a nonnegative matrix $\bM\in \bbR_+^{m\times k}$ and one is required to find a factorization of $\bM$ into two nonnegative matrices $\bU\in\bbR_+^{m\times r}$ and $\bV\in\bbR_+^{r\times k}$  such that $\bM$ is {\em exactly} or {\em approximately} equal to the product $\bU \bV$, i.e., 
\begin{equation}
  \bM=\bU\bV\quad\mbox{or}\quad\bM\approx\bU\bV. \label{eqn:nmf_equal_approx}
  \end{equation}  The matrices $\bU$ and $\bV$ are usually referred to as the {\em dictionary} and {\em coefficient} matrices respectively and the minimum $r$ such that there exists $\bU$ and $\bV$ such that $\bM = \bU\bV$ is known as the {\em nonnegative rank} of $\bM$. The study of NMF in machine learning and signal processing usually consists in developing algorithms to find $\bU$ and $\bV$ efficiently and accurately. En route, one typically uses heuristic (e.g., Bayesian) methods~\cite{TF12} to find good approximations of the nonnegative rank. In this section, however, we are concerned with {\em exact} factorizations  and we discuss information-theoretic interpretations of the nonnegative rank. We will see that the nonnegative rank is intimately connected to several common information quantities  that we encountered in the previous sections (such as Wyner's common information and the exact common information).  

\section{Nonnegative Rank}
We now formally define the nonnegative rank  based on the minimal number of  rank one factors that sum to the given matrix $\bM$. 

\begin{definition} \label{def:nn_rank}
The {\em nonnegative rank}  of a nonnegative matrix $\bM \in \bbR_+^{m\times k}$, denoted as  $\rank_+(\bM)$, is the smallest integer $r$ such that $\bM$ can be represented as 
\begin{equation}
\bM = \sum_{w=1}^r \bu_w \bv_w^\top \label{eqn:rank_one}
\end{equation}
for some nonnegative vectors $\bu_w\in\bbR_+^m$ and $\bv_w\in\bbR_+^k$. Here $\bu_w$ and $\bv_w^\top$ respectively represent the $w^{\mathrm{th}}$ column of $\bU$ and $w^{\mathrm{th}}$ row of $\bV$ where $\bU$ and $\bV$ are the dictionary and coefficient matrices in \eqref{eqn:nmf_equal_approx}.  % where $w \in [r]$. 
\end{definition}
As shown by \citet{Vavasis}, the computation of the nonnegative rank is NP-hard. See \citet{Moitra} for some  positive results. For example, checking whether the nonnegative rank is equal to a fixed  value $r$ (not part of the input) can be done in polynomial time in the dimensions of the input matrix, namely in time $O( (mk)^{r^2} )$. The nonnegative rank is of tremendous significance in computational complexity and combinatorial optimization. Of particular importance is  the fundamental {\em factorization theorem} of \citet{Yannakakis} which states that the nonnegative rank of the slack matrix of a polytope equals its extension complexity  (i.e., the minimum number of facets in a higher-dimensional polytope from which the original one can be obtained as a  linear  projection). We will not delve into such issues in this section as we focus on various other interesting information-theoretic interpretations of the nonnegative rank.  % $\calP= \mathrm{conv}(\bv_1,\ldots,\bv \{\bx: \bA\bx\le\mathbf{b}\}$,  %In view of this difficulty,  we seek bounds or asymptotic approximations. 

We note that the usual (linear) rank is a trivial lower bound to the nonnegative rank as the vectors $\bu_w$ and $\bv_w$ are no longer constrained to be nonnegative, i.e., 
\begin{equation}
\rank(\bM)\le\rank_+(\bM). \label{eqn:rank_ineq}
\end{equation}
It is known from~\citet[Theorem 4.1]{cohenR} that if $\bM$ has rank at most two, this inequality is tight. However, this inequality is not tight in general as the following example from~\cite{cohenR} shows. 
\begin{example} \label{ex:nn_rank}
Let 
\begin{equation}
\bM = \begin{bmatrix}
1 & 1 & 0 & 0 \\
1 & 0 & 1 & 0 \\
0 & 1 & 0 & 1 \\
0 & 0 & 1 & 1 
\end{bmatrix}.
\end{equation}
A direct computation shows that $\rank(\bM)=3$. We now claim that $\rank_+(\bM)=4$. On the one hand, the nonnegative rank of any matrix of size $m$ by $k$ cannot exceed $\min\{m,k\}$, which is $4$ in this case. To justify the lower bound, we call a pair of entries $M_{a,b}$ and $M_{c,d}$ {\em pairwise independent} if $M_{a,b}M_{c,d}>0$ and $M_{a,d} M_{c,b}=0$. Then by~\eqref{eqn:rank_one}, we see that if $\bM$ contains a set of $q$ pairwise independent entries, then  $\rank_+(\bM)\ge q$; see  \citet[Section~3.4.1]{gillisBook} for a detailed justification of this fact. In this example, the entries $M_{1,1}$, $M_{2,3}$, $M_{3,2}$, and $M_{4,4}$	 are pairwise independent, so  $\rank_+(\bM)\ge 4$. Hence, $\rank_+(\bM)=4$, which is strictly larger than the rank of~$\bM$.
\end{example}

In fact, it is known~\cite{Beasley} that the nonnegative rank can be arbitrarily larger than the rank. A canonical example is the family of {\em distance matrices}. 
\begin{example}
For a set of real numbers $\{a_1, \ldots, a_m\}\subset\bbR$, let $\bM=\bM(a_1,\ldots, a_m)$ be the $m\times m$ symmetric matrix
\begin{equation}
\bM = \begin{bmatrix}
0  & (a_1-a_2)^2 &(a_1-a_3)^2 & \dots & (a_1-a_m)^2\\
(a_2-a_1)^2  & 0 &(a_2-a_3)^2 & \dots & (a_2-a_m)^2\\
\vdots & \vdots &\vdots & \ddots & \vdots \\
(a_m-a_1)^2  & (a_m-a_2)^2  &(a_m-a_3)^2 & \dots & 0 
\end{bmatrix}.
\end{equation}
Thus, the $(i,j)^{\mathrm{th}}$ entry of $\bM$ is the square of the  distance between $a_i $ and $a_j$. For obvious reasons, $\bM$ is known as a {\em distance matrix}. Let 
\begin{equation}
\bB  := \begin{bmatrix}
a_1^2  & 1 & -2a_1 \\
a_2^2  & 1 & -2a_2 \\
\vdots & \ddots & \vdots\\
a_m^2 & 1 & -2a_m
\end{bmatrix} \quad \mbox{and}\quad \bC  := \begin{bmatrix}
1  & 1 & \dots & 1 \\
a_1^2  & a_2^2 & \dots & a_m^2 \\
%\vdots & \vdots & \ddots & \vdots\\
a_1 & a_2 &\dots & a_m
\end{bmatrix} .
\end{equation}
Then $\bM =\bB\bC$ so the rank of $\bM$ is at most $3$. If $|\{a_1, a_2, \ldots, a_m\}|\ge 3$, then the rank of $\bM $ is exactly $3$. However, \citet{Beasley} showed that $\rank_+(\bM)=\Omega(\log m)$, so the gap between $\rank(\bM)$ and $\rank_+(\bM)$ is arbitrarily large as $m\to\infty$. Also see the work of \citet{Hrubes} who showed that $\rank_+(\bM)\le 2\log m+2$. We mention in passing that there is yet another notion of rank known as the {\em positive semidefinite (PSD) rank}. The PSD rank of distance matrices is $2$ \cite{Fawzi2015}.
\end{example}
% Indeed, for every $r\in\bbN$, there exists a matrix $\bM$ such that
%$\rank(\bM) = 3$ and $\rank_+(\bM)\ge r$.  

Most of the existing lower bounds on the nonnegative rank are based only on the {\em support} of the matrix, i.e., the sparsity   pattern of the entries as in Example~\ref{ex:nn_rank}. See~\citet{Braun17} and \citet[Chapter~3]{gillisBook} for reviews and \citet{FawziParrilo} for an interesting exception using   norm-based methods. The  sole utilization  of the support has obvious shortcomings as the values of the elements of~$\bM$ are ignored. In the rest of this section, we take a deeper look at the nonnegative rank from a common information-theoretic perspective. 

\section{Wyner's Common Information as Amortized Nonnegative Rank}\label{sec:7.2}
In this section, we describe a  connection between Wyner's common information and the  nonnegative rank of nonnegative matrices. This connection was discovered by \citet{Braun17}, \citet{Braun13}, and \citet{Jain13}. 

To make this connection, for a nonnegative matrix $\bM \in\bbR_+^{m\times k}$, we define its {\em induced distribution} as 
\begin{equation}
\pi_{XY}(x,y) := \frac{M_{x,y}}{\|\bM\|_1} \quad\mbox{for all}\;\, (x,y)\in [m] \times [k],\label{eqn:induced_dist}
\end{equation}
where $\|\bM\|_1 := \sum_{x,y} M_{x,y}$ denotes the $\ell_1$ norm or the  sum of the (absolute values of the) elements of~$\bM$.  Note that the distribution $\pi_{XY}$ defined in \eqref{eqn:induced_dist} is   valid     because $\bM$ is nonnegative and, with the normalization by $\|\bM\|_1$, $\pi_{XY}$ sums to unity. Furthermore, we deliberately use the symbol $\pi_{XY}$  to draw an analogue to the target distribution discussed in Sections~\ref{ch:wynerCI} and~\ref{ch:renyi}--\ref{ch:ecs}. In this section, we write $\calX = [m]$ and $\calY=[k]$ to denote the finite alphabets of $X$ and $Y $ respectively. 

A discrete random variable $W$ with support (or alphabet) $\calW$ is said to be a {\em seed} for the pair of random variables $(X,Y)\sim \pi_{XY}$,  or equivalently the matrix $\bM$, if $X$ and $Y$ are conditionally independent given $W$. Given $\bM$ and a seed $W$ for $\bM$, define the collection of matrices $\{\bM_w:w\in\calW\}$, each with elements 
\begin{equation}
[\bM_w]_{x,y} := \Pr(X=x, Y=y, W=w) \|\bM\|_1 \label{eqn:defbMw}.% \quad\mbox{for all}\;\, (x, y)\in\calX\times \calY. \label{eqn:defbMw}
\end{equation}
Every seed $W$ for $(X,Y)$ induces an NMF by writing $\bM = \sum_{w} \bM_w$ since $\bM_w$ is rank one (a consequence of the Markov chain $X-W-Y$). In Section~\ref{ch:wynerCI}, we   referred to $W$, the seed, as the {\em common random variable} in the definition of Wyner's common  information.  %  where $\bM_w$ is defined in \eqref{eqn:defbMw}. 

We also note that every NMF of $\bM=\sum_w \bM_w=\sum_w \bu_w\bv_w^\top$ induces a seed $W$ for $\bM$ by extending the induced distribution $\pi_{XY}$ via
\begin{equation}
\Pr(X=x,Y=y,W=w) :=\frac{[\bM_w]_{x,y}}{\|\bM\|_1}.
\end{equation}
By virtue of the fact that $\bM_w=\bu_w\bv_w^\top$ for every $w$, we see that $X$ and $Y$ are conditionally independent given $W$.  Due to the    connection between a nonnegative matrix~$\bM$ and its induced distribution $\pi_{XY}$ in~\eqref{eqn:induced_dist}, we can define {\em Wyner's common information for $\bM$} as 
\begin{equation}
C_\Wyner(\bM) := C_\Wyner(\pi_{XY}).
\end{equation}

We start with a simple observation due to \citet{Jain13} and \citet{Braun13} which reinforces the definitions above. 
\begin{proposition} \label{prop:wyner_nnr}
Wyner's common information of $\bM\in\bbR_+^{m\times k}$ is upper bounded by the logarithm of the nonnegative rank of $\bM$, i.e., 
\begin{equation}
C_\Wyner(\bM)\le \log\rank_+(\bM).  \label{eqn:nn_rank_bd}
\end{equation}
\end{proposition}
\begin{proof}
Let $\bM$ have an NMF given by 
\begin{equation}
\bM=\sum_{w\in\calW} \bu_w\bv_w^\top. \label{eqn:given_NMF}
\end{equation}
  Define the seed or common random variable $W$ with conditional distribution $P_{W|XY}$ as 
\begin{equation}
P_{W|XY}(w|x,y) =\left\{ \begin{array}{cc}
\displaystyle\frac{[\bu_w]_x [\bv_w]_y}{M_{x,y}}  & M_{x,y}>0 \vspace{0.03 in} \\
\mbox{arbitrary} & M_{x,y}=0 
\end{array}  \right.  . \label{eqn:PWXY}
\end{equation}
This is a valid conditional distribution because for every $(x,y)$ such that $M_{x,y}>0$,
\begin{equation}
\sum_{w\in\calW} P_{W|XY}(w|x,y)=\sum_{w\in\calW}  \frac{[\bu_w]_x [\bv_w]_y}{M_{x,y}}  = \frac{1}{M_{x,y}}\sum_{w\in\calW} [\bu_w]_x [\bv_w]_y=1,
\end{equation}  where the last equality follows from \eqref{eqn:given_NMF}. 
Define the joint distribution $P_{WXY} := P_{W|XY}\pi_{XY}$, where $\pi_{XY}$ is the induced distribution of $\bM$.  By construction, $(W,X,Y)\sim P_{WXY}$  satisfies $P_{XY}=\pi_{XY}$. Furthermore, by combining \eqref{eqn:induced_dist} and~\eqref{eqn:PWXY}, we obtain
\begin{equation}
P_{XY|W}(x,y|w) = \frac{[\bu_w]_x [\bv_w]_y}{\sum_{x',y'} [\bu_w]_{x'} [\bv_w]_{y'}}, %\quad\mbox{for all}\;\, (x,y)\in\calX\times\calY,
\end{equation}
which, for every fixed $w$, is clearly  a product distribution (cf.\ Definition~\ref{def:pseudo_pdt}(a)). Hence, $X-W-Y$ forms a Markov chain. To complete the proof, recall that 
\begin{equation}
C_\Wyner(\bM)=C_\Wyner(\pi_{XY}) = \min_{P_{W}P_{X|W}P_{Y|W}: P_{XY}=\pi_{XY}} I_P(XY;W).
\end{equation}
Hence, by choosing a  minimal factorization of $\bM$ in \eqref{eqn:given_NMF} (i.e., one that has the smallest $|\calW|$), one has 
\begin{equation}
C_\Wyner(\bM)\le I_P(XY;W)\le H(W)\le \log |\calW| = \log\rank_+(\bM),
\end{equation}
completing the proof of \eqref{eqn:nn_rank_bd}.
\end{proof}
One natural question arising from Proposition \ref{prop:wyner_nnr} concerns the tightness of the bound in \eqref{eqn:nn_rank_bd}. This bound can be arbitrarily loose as the following example from \citet{Braun17} demonstrates.  Our justification of the upper   bound on $C_\Wyner(\bM)$ differs from~\cite{Braun17} and, in particular, does not use require the notion of {\em rectangle covering}~\cite{Yannakakis}. 
\begin{example}
For a fixed natural number $m$, let $\bM \in \bbR_+^{m\times m}$ be the   diagonal matrix with diagonal elements $M_{i,i}= 2^i /\sum_{ j \in [m] } 2^j$ for $i \in [m]$. Then, it is clear that $\rank_+(\bM)=m$. However, Wyner's common information for this matrix $\bM$, which is normalized, can be bounded as
\begin{align}
C_\Wyner(\bM )% &= C_\Wyner(\pi_{XY})\\
&\le H_\pi(XY)  \label{eqn:wyner_XY}\\
&= H(\pi_X) \label{eqn:XY_det}\\
 &= H\left( \frac{2}{\sum_{j \in [m]} 2^j} , \frac{2^2}{\sum_{j \in [m]} 2^j}, \ldots, \frac{2^m}{\sum_{j \in [m]} 2^j}\right)\\
 &= -\sum_{i \in [m]} \frac{2^i}{\sum_{j \in [m]} 2^j }\log\bigg(\frac{2^i}{\sum_{j \in [m]} 2^j }\bigg),\label{eqn:ent_m}
\end{align}
where \eqref{eqn:wyner_XY} follows because $I(XY;W)\le H_\pi(XY)$ and \eqref{eqn:XY_det} follows because $X=Y$ in the joint distribution $\pi_{XY}$ induced by $\bM$. The final expression in~\eqref{eqn:ent_m} can be shown to be no larger than $2$ for all $m$ (and in fact converges to $2$ as $m\to\infty$). Thus,  $C_\Wyner(\bM ) \le 2$ for all $m\in\bbN$ and the gap between $C_\Wyner(\bM ) $ and $\log\rank_+(\bM)=\log m$ in~\eqref{eqn:nn_rank_bd} can be made arbitrarily large as~$m$ tends to infinity. 
\end{example}

This somewhat pathological phenomenon can, however, be remedied by considering  small $\ell_1$ perturbations of the $n$-fold Kronecker power of the given matrix $\bM$. In this case, the limit of the normalized logarithm of the nonnegative rank of the perturbed matrix can be shown to be upper bounded by Wyner's common information of $\bM$.  This fundamental result is due to \citet{Braun17} who used the term {\em amortization} to describe the perturbation and limiting operations. 
\begin{theorem}[Amortized nonnegative rank and Wyner's common information]\label{thm:nn_perturb}
Let $\bM \in\bbR_+^{m\times k}$ be a matrix with  $\|\bM\|_1= \sum_{x,y} M_{x,y} =\ell$. Then for any $\epsilon>0$ and $\delta\in (0,1)$, for every 
\begin{equation}
n\ge \max\left\{ \Omega\bigg(\frac{\log^2 (mk)}{\epsilon^2 C_\Wyner(\bM)^2} \log\Big(\frac{1}{\delta} \Big) \bigg),  \, \Omega \bigg(\frac{\delta}{\epsilon}  \bigg)  \right\}  ,
\end{equation}
there exists a nonnegative matrix $\bM_{\epsilon,\delta,n} \in\bbR_+^{m^n \times k^n}$ with 
\begin{align}
\frac{1}{n}\log\rank_+(\bM_{\epsilon,\delta,n} )\le (1+\epsilon)\ C_\Wyner(\bM)+ O\Big( \delta^3 \log\frac{1}{\delta}\Big)\ \frac{\log n}{n}
\end{align}
and 
\begin{align}
\big\|\bM^{\otimes n}  -\bM_{\epsilon,\delta,n}\big\|_1\le\delta \ell^n. \label{eqn:tv_M}
\end{align}
In particular,  for every $\delta\in (0,1)$,  one has  
\begin{equation}
 \lim_{\epsilon\downarrow 0}  \lim_{n\to\infty}\frac{1}{n}\log\rank_+(\bM_{\epsilon,\delta,n} )=C_\Wyner(\bM). \label{eqn:l1_perturbed_nnr}
\end{equation}
\end{theorem}
Thus, Wyner's common information of $\bM$  (or of $(X,Y)\sim \pi_{XY}$) admits yet another operational interpretation, namely   normalized logarithm of the nonnegative rank  of an $\ell_1$-perturbed version $\bM^{\otimes n}$. This is in addition to its two more familiar operational interpretations in terms of (i) the  minimum rate of common randomness of the Gray-Wyner system when the sum rate is constrained to be no larger than the joint entropy and (ii) the minimum amount of common randomness to simulate a joint source in a distributed manner (cf.\ Section~\ref{ch:wynerCI}).  

The  reader will notice that Theorem~\ref{thm:nn_perturb} is analogous to the fact that the {\em TV common information} (introduced in  Definition~\ref{def:TVCI}) is equal to {\em Wyner's common information} (see \citet{cuff13} and Section~\ref{sec:tvci}). Indeed, the $\ell_1$-relaxation of $\bM^{\otimes n}$ to $\bM_{\epsilon,\delta,n}$ in~\eqref{eqn:tv_M} is analogous to the discrepancy between the target distribution $\pi_{XY}^n$ and the synthesized distribution $P_{X^nY^n}$ in~\eqref{eqn:tv_ci}. So, in some sense, we have ``come full circle'' in this part of the monograph.

The proof of Theorem~\ref{thm:nn_perturb} involves approximating $\pi_{XY}^{n}$, the $n$-fold product of the induced distribution given $\bM$, by a collection of ``better behaving'' distributions so that we can bound the log-likelihood ratio $\log   P_{XY|W  } (X,Y |W)-\log\pi_{XY}(X,Y)$ whose expectation under $(X,Y,W)\sim\pi_{XY}P_{W|XY}$ yields $I(XY;W)$  in the expression for Wyner's common information $C_\Wyner(\pi_{XY})=\min_{X-W-Y} I(XY;W)$. For this purpose, several concentration bounds, such as Chernoff bounds, are used to show that certain well-behaved distributions exist with high probability. As the details are rather involved and delicate, we refer the reader to \citet{Braun17}.
\section{Exact R\'enyi Common Information as Nonnegative Rank}
In Section \ref{sec:7.2}, we related the nonnegative rank of a matrix $\bM$ to its Wyner's common information. Given our discussion of the {\em exact} common information  in Section~\ref{ch:exact}, it is natural to wonder whether the nonnegative rank has any relation to the exact  common information. The purpose of this section is to elaborate on this.   

Recall  from Proposition~\ref{prop:KLE} that the exact common information admits the multi-letter characterization in terms of the {\em common entropy rate} (previously defined in \eqref{eqn:common_ent_rate})  as follows
\begin{equation}
T_\Ex(\pi_{XY}) = \lim_{n\to\infty} \frac{ G(\pi_{XY}^n)}{n}\label{eqn:TEx_comm_ent} ,
\end{equation}
where the {\em common entropy} of $(X,Y)\sim\pi_{XY}$ (previously defined in~\eqref{eqn:common_ent}) is 
\begin{equation}
G(\pi_{XY} ) =  \min_{  P_{W  }P_{X |W }P_{Y |W }:  P_{X Y  } = \pi_{XY }} H(W ) . \label{eqn:cmn_ent}
\end{equation}
We can define the {\em common R\'enyi entropy of order $\alpha \in [0,\infty]$}  as
\begin{equation}
G_\alpha (\pi_{XY}) :=  \min_{  P_{W  }P_{X |W }P_{Y |W }:  P_{X Y  } = \pi_{XY }} H_\alpha(W ) ,\label{eqn:cmn_ent_alpha} 
\end{equation}
where $H_\alpha$ is the R\'enyi entropy of order $\alpha$ (defined in \eqref{eqn:renyi_ent} and \eqref{eqn:renyi_ent2}). Then the exact common information can be generalized to the {\em exact R\'enyi common information of order $\alpha$}, similarly to \eqref{eqn:TEx_comm_ent}, as follows
\begin{equation}
T_\Ex^{(\alpha)}(\pi_{XY}) := \lim_{n\to\infty} \frac{ G_\alpha(\pi_{XY}^n)}{n} . \label{eqn:TEx_comm_ent_alpha}
\end{equation}
The existence of the limit in \eqref{eqn:TEx_comm_ent_alpha} follows by the
subadditivity of the sequence $\{ G_\alpha(\pi_{XY}^n)/n\}_{n\in\bbN}$ and Fekete's lemma~\cite{Fekete}.  Note, by definition, that $T_\Ex^{(1)}(\pi_{XY})=T_\Ex(\pi_{XY})$. Since $\alpha\mapsto H_\alpha (\pi_{XY})$ is non-increasing, $T_\Ex^{(\alpha)}(\pi_{XY}) $ is also non-increasing in $\alpha \in [0,\infty]$. We will be concerned  with $T_\Ex^{(\alpha)}(\pi_{XY}) $ for values of $\alpha \in \{0,1,\infty\}$. The following proposition was shown by the present authors in~\cite{YuTan2020_exact}.

\begin{proposition}\label{prop:ECI_nnr}
We have 
\begin{align}
G_\alpha(\pi_{XY})  = \left\{ \begin{array}{cl}
\log\rank_+(\pi_{XY}) & \alpha=0  \vspace{0.03in} \\
G(\pi_{XY}) & \alpha=1  \vspace{0.03 in} \\
\displaystyle\min_{Q_X, Q_Y} D_\infty(Q_X Q_Y\|\pi_{XY}) & \alpha=\infty
\end{array} \right. .
\end{align}
\end{proposition}

The first statement  ($\alpha=0$) in Proposition~\ref{prop:ECI_nnr} follows by first noticing that $H_0(W) = \log |\calW|$, where the support of $W$ is $\calW$. Hence, we see  that the exact R\'enyi common information of order $0$ corresponds
to the minimum common randomness rate for exact
generation of the target distribution in which the common
randomness is only allowed to be compressed by {\em fixed-length}
codes.  In contrast to Theorem~\ref{thm:nn_perturb}, this is a {\em one-shot} characterization of the nonnegative rank and no perturbation or limiting operations are needed. The second statement  ($\alpha=1$) follows by comparing the definitions of $G(\pi_{XY})$ and $G_\alpha(\pi_{XY})$  in~\eqref{eqn:cmn_ent} and~\eqref{eqn:cmn_ent_alpha} respectively. 

The final statement  ($\alpha=\infty$) requires a short calculation which we sketch here. Since $H_\infty(W ) = -\log \max_w P_W(w)$, 
\begin{align}
G_\infty(\pi_{XY}) = -\log \max_{  P_{W  }P_{X |W }P_{Y |W }:  P_{X Y  } = \pi_{XY }} \max_w P_W(w).
\end{align}
Swapping the maximization operations, 
\begin{align}
G_\infty(\pi_{XY})  &= -\log\max_w \max_{  P_{W  }P_{X |W }P_{Y |W }:  P_{X Y  } = \pi_{XY }} P_W(w) \\
 &\ge -\log\max_w \max_{ \substack{P_{X|W}P_{Y|W}: \\ P_W(w) P_{X|W}(x|w)P_{Y|W}(y|w)\\ \le \pi_{XY}(x,y) \,\forall (x,y) } } P_W(w) \\
 &\ge \min_w  \min_{P_{X|W=w}, P_{Y|W=w}}D_\infty\big(P_{X|W=w} P_{Y|W=w} \big\| \pi_{XY} \big) \\
 &\ge \min_{Q_X, Q_Y}D_\infty\big( Q_X Q_Y \big\| \pi_{XY} \big).
\end{align}
In the other direction, we let $(Q_X^*, Q_Y^*)$ achieve the minimization in the optimization problem defining $G_\infty(\pi_{XY})$. Let $\epsilon:= D_\infty(Q_X^*Q_Y^*\|\pi_{XY})$. Then by the mixture decomposition technique (as described in Section~\ref{sec:sketch_equiv}),  we see that 
\begin{equation}
\pi_{XY} = 2^{-\epsilon}Q_X^*Q_Y^*+ (1-2^{-\epsilon})P_{\hatX\hatY},
\end{equation}
where $P_{\hatX\hatY} \in\calP(\calX\times\calY)$ is a joint distribution defined as
\begin{equation}
P_{\hatX\hatY} := \left\{  \begin{array}{cl}
\mbox{arbitrary} & \epsilon=0\vspace{.03in}\\
\displaystyle\frac{\pi_{XY}-2^{-\epsilon} Q_X^*Q_Y^* }{1-2^{-\epsilon}} & \epsilon \in (0,\infty)\vspace{.03in}\\
\pi_{XY}& \epsilon=\infty
\end{array} \right. .
\end{equation}
Now, we choose the common random variable $W$ having alphabet $\calW=(\calX\times\calY)\cup\{w_0\}$ where $w_0\notin\calX\times\calY$ and $W$ has distribution 
\begin{equation}
P_W(w) =  \left\{  \begin{array}{cl}
2^{-\epsilon} & w=w_0\vspace{.03in}\\
(1-2^{-\epsilon})P_{\hatX\hatY}(\hatx,\haty) & w=(\hatx,\haty)\in\calX\times\calY
\end{array} \right. .
\end{equation}
We construct $P_{X|W}$ and $P_{Y|W}$ as 
\begin{align}
P_{X|W}(x|w) &=  \left\{  \begin{array}{cl}
Q_X^*(x) & w=w_0\vspace{.03in}\\
\bone\{x=\hatx\} & w=(\hatx,\haty)\in\calX\times\calY
\end{array} \right.  
\quad\mbox{and}\\
P_{Y|W}(y|w) &=  \left\{  \begin{array}{cl}
Q_Y^*(y) & w=w_0\vspace{.03in}\\
\bone\{y=\haty\} & w=(\hatx,\haty)\in\calX\times\calY
\end{array} \right.   .
\end{align}
By construction, the joint distribution $P_W P_{X|W} P_{Y|W}$ satisfies
\begin{align}
P_{XY}=\pi_{XY}\quad\mbox{and}\quad H_\infty(W)\le\epsilon.
\end{align}
Thus, $G_\infty(\pi_{XY}) \le \epsilon= D_\infty(Q_X^*Q_Y^*\|\pi_{XY})$ as desired. 

When we consider the $n$-fold product distribution $\pi_{XY}^n$ (or equivalently the $n$-fold Kronecker product $\pi_{XY}^{\otimes n}$ of the matrix of joint probabilities $\pi_{XY}$), we obtain the following corollary. 
 
 \begin{corollary}[Exact R\'enyi common information]  \label{cor:ECI_nnr}
We have 
\begin{align}
T_\Ex^{(\alpha)}(\pi_{XY})  = \left\{ \begin{array}{cl}
\displaystyle\lim_{n\to\infty}\frac{1}{n}\log\rank_+\big(\pi_{XY}^{\otimes n}\big) & \alpha=0  \vspace{0.03 in} \\
T_\Ex(\pi_{XY}) & \alpha=1  \vspace{0.03 in} \\
\displaystyle\min_{Q_X, Q_Y} D_\infty(Q_X Q_Y\|\pi_{XY}) & \alpha=\infty
\end{array} \right. .
\end{align}
\end{corollary}
We note that the first statement ($\alpha=0$) requires a limiting operation because unlike the linear rank, it is, in general, not true that $\rank_+(\bM^{\otimes  n}) = \big(\rank_+(\bM)\big)^{n}$; see \citet{Vandaele2016} for  a discussion and related conjectures. The second statement ($\alpha=1$)  comes from   the multi-letter characterization of the exact common information given in \citet{KLE2014}  (Proposition~\ref{prop:KLE}) while the last  ($\alpha=\infty$)  requires some single-letterization steps; see \citet{YuTan2020_exact}.  Corollary~\ref{cor:ECI_nnr}, illustrated in Fig.~\ref{fig:bridge_exact}, implies that the  exact R\'enyi common information of $\alpha$ interpolates between the nonnegative rank (when $\alpha=0$),   the exact common information (when $\alpha=1$), and   $\min_{Q_X, Q_Y} D_\infty(Q_X Q_Y\|\pi_{XY}) $ (when $\alpha=\infty$). This is somewhat analogous to the fact that  the R\'enyi common information forms a bridge between Wyner's common information and the exact common information (see Fig.~\ref{fig:bridge}).

It is important to note a key distinction between Corollary~\ref{cor:ECI_nnr} and Theorem~\ref{thm:nn_perturb}. The former tells us  that the asymptotic exponent of the nonnegative rank of $\pi_{XY}^{\otimes n}$ can be interpreted as the exact R\'enyi common information of order $0$. The latter, on the other hand, tells us that the asymptotic exponent of the nonnegative rank of {\em an $\ell_1$ perturbed version}   of $\pi_{XY}^{\otimes n}$  is Wyner's common information; see~\eqref{eqn:l1_perturbed_nnr}.

\begin{figure}[!ht]
\centering
\begin{overpic}[width=.9\textwidth]{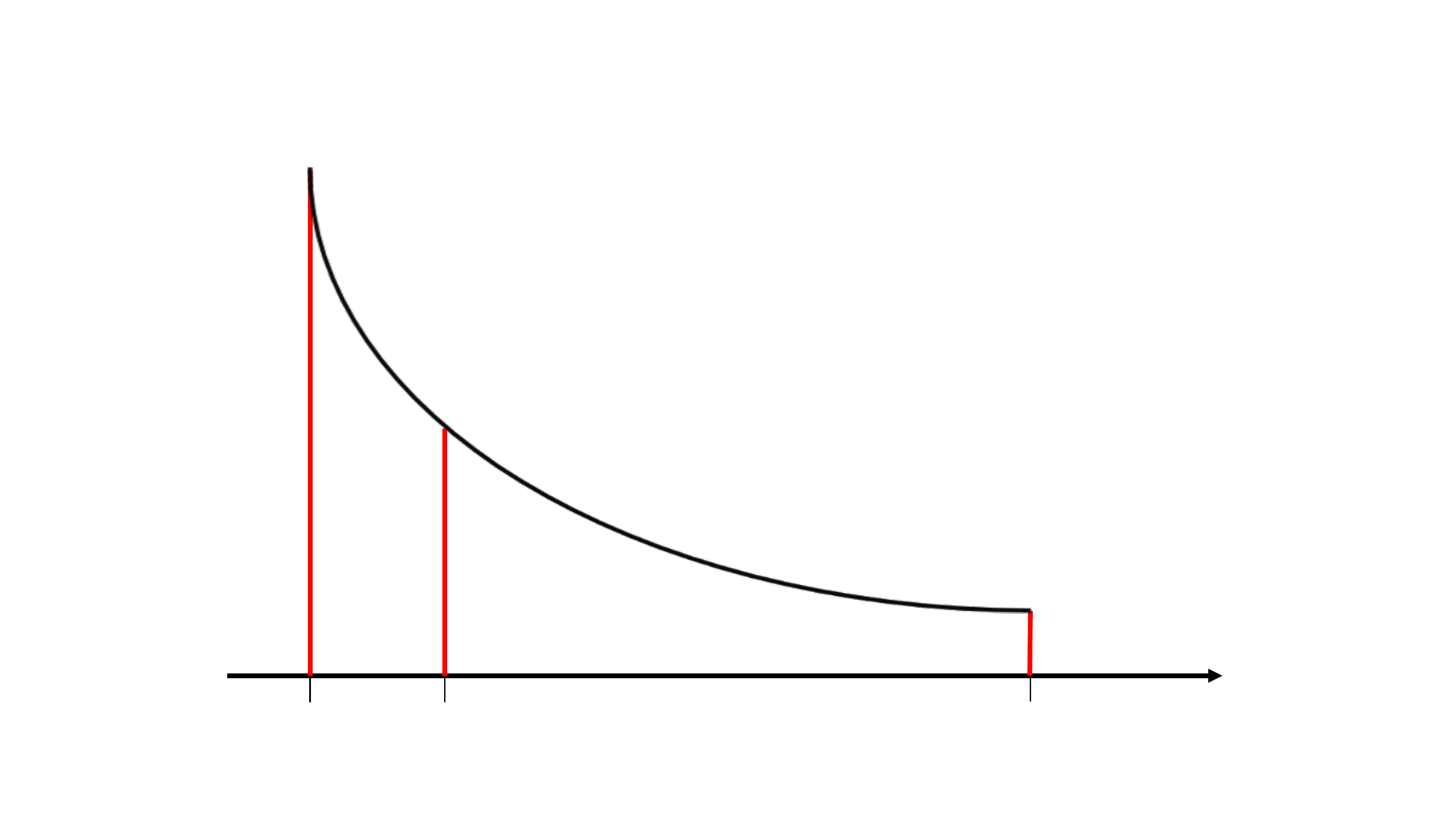}
%\put(80,15.5){}
{\small\put(80,11.5){R\'enyi Order $\alpha$}}
\put(70,5){$\infty$}
\put(45,7){$\ldots$}
\put(50,7){$\ldots$}
\put(40,7){$\ldots$}
\put(55,7){$\ldots$}
%\put(34.5,5){$2$}
\put(30,5){$1$}
\put(21,5){$0$}
{\small\put(64,19){\mbox{$\displaystyle\min_{Q_X, Q_Y}\! D_\infty (Q_XQ_Y\|\pi_{XY})$}}}
{\small\put(25,39.5){Exact CI}}
\put(33.5,38){\vector(-1,-4){2.7}}
%\put(13,34){\vector(10,-1){17}}
\put(21.5,44.5){\circle*{1.5}}
\put(30.5,27){\circle*{1.5}}
\put(54,29.5){\vector(0,-1){13.5}}
%\put(30.5,38){\vector(0,-1){8.5}}
{\small\put(43,31){Exact R\'enyi CI $T_{\Ex}^{ (\alpha) }(\pi_{XY})$ }}
\put(71,14.5){\circle*{1.5}}
{\small\put(2,49){Exponent of}
\put(2,45){Nonnegative}
\put(7.5,41){Rank}}
\end{overpic}
\caption{A schematic showing that the exact  R\'enyi common information forms a bridge between the nonnegative rank, the exact common information, and $\min_{Q_X, Q_Y}D_\infty (Q_XQ_Y\|\pi_{XY})$}
\label{fig:bridge_exact}
\end{figure}
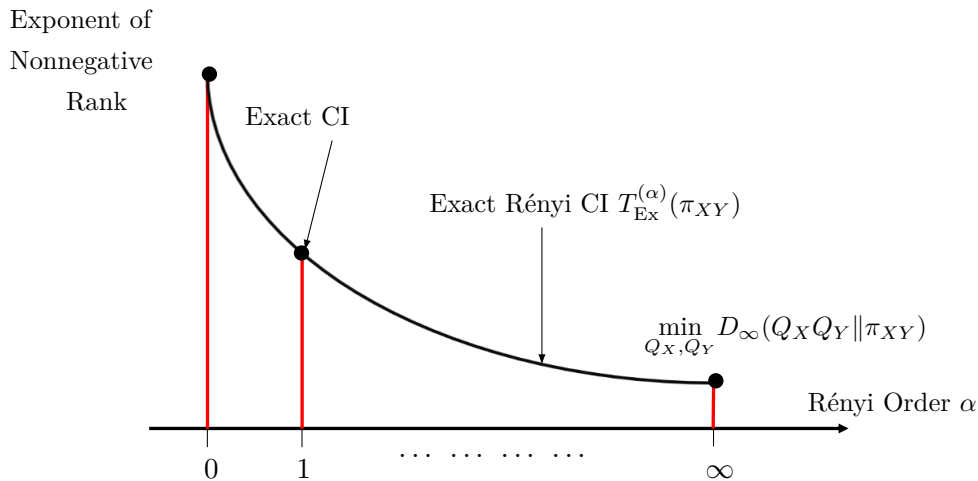

%In summary, we see from Corollary~\ref{cor:ECI_nnr} that the asymptotic exponent of the nonnegative rank of $\pi_{XY}^{\otimes n}$ can be interpreted as the exact R\'enyi common information of order $0$. In contrast, with some $\ell_1$ perturbation, the nonnegative rank of $\pi_{XY}^{\otimes n}$ is Wyner's common information; see Theorem~\ref{thm:nn_perturb}. 
% 
\section{Nonnegative $\alpha$-Rank}
\label{sec:nnar}
We conclude this section by briefly mentioning a common information-theoretic generalization of the nonnegative rank. Recall from Proposition~\ref{prop:ECI_nnr} that the logarithm of the nonnegative rank is the exact R\'enyi common information of order $0$. Inspired by this relationship, we can generalize the notion of the nonnegative rank to
the nonnegative $\alpha$-rank as follows. For a nonnegative matrix
(but non-zero matrix) $\bM$ and $\alpha\in [-\infty,\infty]$, we define the
{\em nonnegative $\alpha$-rank} of  $\bM$ as
\begin{equation}
\rank_+^{(\alpha)} (\bM) := 2^{G_\alpha(\pi_{XY}) }, \label{eqn:rank_alpha0}
\end{equation}
where $\pi_{XY}$ is the induced distribution of $\bM$, defined in \eqref{eqn:induced_dist}. The normalization by $\|\bM\|_1$ is required as any reasonable notion of rank should be invariant to the scale of the matrix. This definition reduces to that of $\rank_+(\bM)$ when $\alpha=0$ by Proposition~\ref{prop:ECI_nnr}.

For a diagonal matrix $\bD$, let $\|\bD\|_\alpha$ be the $\alpha$-norm of its diagonal, i.e., $\|\bD\|_\alpha  = \big(\sum_i D_{i,i}^\alpha\big)^{1/\alpha}$. By appealing to the definition of the R\'enyi entropy, we see that the nonnegative $\alpha$-rank of $\bM\in\bbR_+^{m\times k}$ can be equivalently expressed as
\begin{equation}
\rank_+^{(\alpha)} (\bM) =\min_{\bU, \bD, \bV}\big\|\bD\big\|_\alpha^{\frac{\alpha}{1-\alpha}}, \label{eqn:rank_alpha1}
\end{equation}
where the minimization runs over all triples of matrices $\bU \in \bbR_+^{m\times r}$,  $\bV \in \bbR_+^{k\times r}$ and diagonal $\bD\in\bbR_+^{r\times r }$ for some $r\in\bbN$ such that 
\begin{align}
\sum_{x=1}^m [\bU]_{xw} &= 1\quad\mbox{for all} \;\,  w \in [r], \label{eqn:cond_prob_U} \\
\sum_{y=1}^k [\bV]_{yw} &= 1\quad\mbox{for all} \;\,  w \in [r], \quad\mbox{and}\label{eqn:cond_prob_V}\\
  \bU\bD\bV^\top&= \frac{ \bM }{\|\bM\|_1} .
\end{align}
The equality conditions in~\eqref{eqn:cond_prob_U} and~\eqref{eqn:cond_prob_V} are to ensure that each column of $\bU$ and each column of $\bV$ is a PMF. The alternative definition of $\rank_+^{(\alpha)}(\bM)$ in \eqref{eqn:rank_alpha1} 
  can be seen to be a generalization of the usual nonnegative rank as defined in Definition~\ref{def:nn_rank}. Indeed, 
\begin{equation}
 \lim_{\alpha\downarrow 0} \rank_+^{(\alpha)} (\bM)=\rank_+ (\bM).
\end{equation}
The properties of $\rank_+^{(\alpha)} (\bM)$ as defined in \eqref{eqn:rank_alpha0} or \eqref{eqn:rank_alpha1} are not well understood and constitute a fertile avenue for further investigations.

\part[Extensions of G\'acs--K\"orner--Witsenhausen's Common \\ Information]{Extensions of G\'acs--K\"orner--Witsenhausen's Common Information} \label{part:three}
\global\long\def\cvx{\bbL}%
 
\global\long\def\cve{\bbU}%

\chapter{Non-Interactive Correlation Distillation} \label{ch:NICD}
%This is based on our paper ``Wyner's common information under R\'enyi divergence measures''. More than just a generalization of R\'enyi's CI, it serves as a convenient bridge between Wyner's CI and Exact CI. Will discuss the TV (variational distance) approximate version and its strong converse together with unnormalized and normalized  versions of the the R\'enyi divergence constraints.

In this section, we consider an extension of GKW's common information,
termed {\em Non-Interactive Correlation Distillation}. We recall
that GKW's common information measures the amount of ``almost identical''
randomnesses that can be extracted individually from a pair of correlated
sources. By G\'acs and K\"orner's theorem~\cite{gacs1973common} (also recall Proposition~\ref{prop:CGWW0}), the
GKW's common information of a joint source $(X,Y)$ is positive if and only if there
exists a pair of non-constant functions $(f,g)$ such that $f(X)=g(Y)$
almost surely. Unfortunately, GKW's common information is zero for
many common pairs of sources, such as jointly Gaussian sources
and doubly symmetric binary sources (DSBS) with  correlation
coefficients $\rho\in(-1,1)$. For these joint sources, even if we wish to extract
{\em a single pair} of identical bits from these sources individually,
this innocuous task still turns out to be infeasible.
 %Thus, in some sense, the GKW common information does not serve the purpose of quantifying the common information and we have to resort to other methods.

This observation begs the following natural question: {\em How can
we refine the quantification of common information for these and other
sources such that it resembles the GKW's common information and yet
is non-zero?} Even though any randomnesses extracted from these sources
individually cannot agree almost surely, the extracted randomnesses
can indeed agree with a certain probability, which, in this section,
we quantify via various probability limit theorems such as the central
limit and large deviations theorems. In other
words, the extracted randomnesses can be correlated. It is thus natural to quantify the ``common information''
by the maximal correlation of a pair of random bits that can be extracted
from the sources individually. In the literature, determining this
maximal correlation is coined the  {\em  Noise Stability Problem} (two-set version), the {\em Non-Interactive Correlation
Distillation} or NICD problem. Other names include the {\em Non-Interactive
Binary Simulation Problem} and the {\em Binary Decision Problem}.
This problem was studied by \citet{kamath2016non}, \citet{yang2007possibility},
\citet{mossel2006non} and \citet{witsenhausen1975sequences} among
others.

In this section, we focus mainly on the doubly symmetric binary source
(DSBS) parametrized by its correlation coefficient $\rho\in(-1,1)$.
Even though this source is simple, the NICD problem for this source is nontrivial and insights can be drawn from it. In Section~\ref{sec:nicd2}, we define the $2$-user
NICD problem for the DSBS. Based on the means of the extracted random bits, we define several asymptotic regimes of interest,
including the central limit, moderate, and large deviations regimes.
In Section~\ref{sec:nicd_ach}, we discuss various achievability
schemes for the NICD problem based on certain geometric structures
in Hamming space; these include subcubes and Hamming spheres. These
geometric structures are useful to prove existence results in the
above-mentioned asymptotic regimes. In Sections~\ref{sec:nicd_conv}, \ref{sec:sse}, and~\ref{sec:ldr} we discuss the optimality
of these schemes. 
%We then extend the NICD problem to the multiple
%user case in Section~\ref{sec:nicd_multi}. We discuss the related
%$q$-stability problems and some contemporary conjectures such as
%the Courtade--Kumar conjecture. 
Finally, in Section~\ref{sec:nicd_arb}
we discuss known results in the NICD problem for other sources such
as 
%sources on finite alphabets and 
 bivariate Gaussians. %sources and sources on arbitrary (e.g., Polish) alphabets.

\section{Non-Interactive Correlation Distillation with $2$ Users}

\label{sec:nicd2} Consider a doubly symmetric binary distribution
$\pi_{XY}$ on the alphabet $\calX\times\calY=\{0,1\}^{2}$ with   
correlation coefficient $\rho\in(0,1)$, i.e., %\begin{equation}
%\pi_{XY}=\begin{array}{ccc}
%X\backslash Y & 0 & 1\\
%0 & \frac{1+\rho}{4} & \frac{1-\rho}{4}\\
%1 & \frac{1-\rho}{4} & \frac{1+\rho}{4}
%\end{array}.\label{eq:NICDDSBS}
%\end{equation}
\begin{align}
\pi_{XY}(x,y)=\left\{ \begin{array}{cc}
{\displaystyle \frac{1+\rho}{4}} & x=y\vspace{0.03in}\\
{\displaystyle \frac{1-\rho}{4}} & x\ne y
\end{array}\right..%
%\pi_{XY}(1,0)
\label{eq:NICDDSBS}
\end{align}
With this parametrization, the {\em correlation coefficient} of
$(X,Y)$, defined in \eqref{eqn:pcc}, 
%\begin{equation}
%\rho(X;Y):=\frac{\cov(X,Y)}{\sqrt{\var(X)\var(Y)}}\label{eqn:pcc}
%\end{equation}
is indeed~$\rho$. The pair of random variables $(X,Y)\sim \pi_{XY}$
corresponds to the DSBS as described in Section~\ref{sec:dsbs} with
{\em crossover probability} $p=(1-\rho)/2\in(0,1/2)$. In this
section, we find it convenient to parametrize the DSBS by its correlation
coefficient~$\rho$ instead of its crossover probability~$p$. It
suffices to consider positive $\rho$ as the results carry over to
the case for negative $\rho$ by replacing $X$ with $1-X$. Throughout
this section except for Section~\ref{sec:nicd_arb}, we  let $(X^{n},Y^{n})$ be distributed as the $n$-fold
product distribution~$\pi_{XY}^{n}$.

%\begin{centering}
\begin{figure}
\centering \tikzset{every picture/.style={line width=0.75pt}}
%set default line width to 0.75pt  
\resizebox{0.95\textwidth}{!}{      
\begin{tikzpicture}[x=0.75pt,y=0.75pt,yscale=-1,xscale=1] %uncomment if require: \path (0,2353); %set diagram left start at 0, and has height of 2353
%Shape: Circle [id:dp15810919403524903] 
\draw   (958.37,566) .. controls (958.37,559.93) and (963.29,555.01) .. (969.37,555.01) .. controls (975.44,555.01) and (980.36,559.93) .. (980.36,566) .. controls (980.36,572.08) and (975.44,577) .. (969.37,577) .. controls (963.29,577) and (958.37,572.08) .. (958.37,566) -- cycle ; %Shape: Circle [id:dp061990753128097875] 
\draw   (851.37,567) .. controls (851.37,560.93) and (856.29,556.01) .. (862.37,556.01) .. controls (868.44,556.01) and (873.36,560.93) .. (873.36,567) .. controls (873.36,573.08) and (868.44,578) .. (862.37,578) .. controls (856.29,578) and (851.37,573.08) .. (851.37,567) -- cycle ; %Straight Lines [id:da2261531493900839] 
\draw    (851.37,567) -- (796.36,567) ;
\draw [shift={(794.36,567)}, rotate = 360] [color={rgb, 255:red, 255; green, 255; blue, 255 }  ][line width=0.75]    (10.93,-3.29) .. controls (6.95,-1.4) and (3.31,-0.3) .. (0,0) .. controls (3.31,0.3) and (6.95,1.4) .. (10.93,3.29)   ; %Straight Lines [id:da5517669597085024] 
\draw    (980.36,566) -- (1034.36,566) ;
\draw [shift={(1036.36,566)}, rotate = 180] [color={rgb, 255:red, 255; green, 255; blue, 255 }  ][line width=0.75]    (10.93,-3.29) .. controls (6.95,-1.4) and (3.31,-0.3) .. (0,0) .. controls (3.31,0.3) and (6.95,1.4) .. (10.93,3.29)   ; %Curve Lines [id:da07343264222723489] 
\draw    (881.36,568) .. controls (921.36,538) and (910.36,595) .. (950.36,565) ;
% Text Node
\draw (852.36,531) node [anchor=north west][inner sep=0.75pt]    {$X^n$}; % Text Node
\draw (964.36,531) node [anchor=north west][inner sep=0.75pt]    {$Y^n$}; % Text Node
\draw (682.36,556) node [anchor=north west][inner sep=0.75pt]    {$f(X^n) \!\sim\!\mathrm{Bern}(a)$}; % Text Node
\draw (1040.36,554) node [anchor=north west][inner sep=0.75pt]    {$g(Y^n)\! \sim\! \mathrm{Bern}( b)$}; % Text Node
\draw (820.36,615) node [anchor=north west][inner sep=0.75pt]    {$\max/\min \, \Pr( f(X^n) =g(Y^n))$}; % Text Node
\draw (814.36,508) node [anchor=north west][inner sep=0.75pt]   [align=left] {DSBS with correlation coefficient $\rho $}; % Text Node
% \draw (971.36,615) node [anchor=north west][inner sep=0.75pt]    {$\mathrm{max} \ \Pr( f(\mathbf{X}) =g(\mathbf{Y}) =1)$}; % Text Node
% \draw (860.36,616) node [anchor=north west][inner sep=0.75pt]   [align=left] {or equivalently,}; % Text Node
\draw (816,548) node [anchor=north west][inner sep=0.75pt]    {$f$}; % Text Node
\draw (1000,550) node [anchor=north west][inner sep=0.75pt]    {$g$};
\end{tikzpicture}} %\end{centering}
\caption{\label{fig:NICD-with-2}The Non-Interactive Correlation Distillation
problem with $2$ users }
\end{figure}

%Recall that the bases of logarithms are set to $2$.

We now introduce the NICD problem with $2$ users. This problem is
illustrated in Fig.~\ref{fig:NICD-with-2}, in which a source sequence $(X^n,Y^n)$ generated by a DSBS is given, and two random bits $f(X^{n})$ and $g(Y^{n})$ are generated
in a distributed manner using a pair of Boolean functions  $f,g:\{0,1\}^{n}\to\{0,1\}$.
The objective of the NICD problem is to maximize or minimize  the
{\em agreement probability} of~$f(X^{n})$ and~$g(Y^{n})$, i.e.,
$\Pr(f(X^{n})=g(Y^{n}))$, under the condition that the means of $f(Y^{n})$
and $g(Y^{n})$ are bounded. 
\begin{definition}\label{def:for_rev_jp}
Given $a,b\in[0,1]$,  the {\em forward joint
probability} is 
\begin{align}
\hspace{-.2in}\overline{\Gamma}^{(n)}(a,b)\!:=\!\max_{\substack{f,g:\{0,1\}^{n}\to\{0,1\}:\Pr(f(X^{n})=1)\leq a,\\
\Pr(g(Y^{n})=1)\leq b
}
}\!\Pr\big(f(X^{n})\!=\!g(Y^{n})\!=\!1\big).\label{eqn:jp_forward}
\end{align}
Similarly, define the {\em reverse joint probability} as 
\begin{align}
\hspace{-.2in}\underline{\Gamma}^{(n)}(a,b)\!:=\!\min_{\substack{f,g:\{0,1\}^{n}\to\{0,1\}:\Pr(f(X^{n})=1)\ge a,\\
\Pr(g(Y^{n})=1)\ge b
}
}\!\Pr\big(f(X^{n})\!=\!g(Y^{n})\!=\!1\big).\label{eqn:jp_reverse}
\end{align}
\end{definition} In Definition~\ref{def:for_rev_jp}, we maximize
or minimize the probability that both generated bits are equal to
one, i.e.,  $\Pr(f(X^{n})=g(Y^{n})=1)$, rather than $\Pr(f(X^{n})=g(Y^{n}))$,
since by noting that the marginal probabilities $\Pr(f(X^{n})=1)$
and $\Pr(g(Y^{n})=1)$ are constrained in~\eqref{eqn:jp_forward}
and~\eqref{eqn:jp_reverse}, determining the former is equivalent
to that of the latter.

\subsection{Optimizing over Supports of Boolean Functions}

Instead of optimizing over the Boolean functions $f$ and $g$, in
the following, we find it convenient for the sake of exploiting the
properties of geometric structures (such as Hamming balls and spheres)
to optimize over their {\em supports}. The {\em support} of
a Boolean function $f:\{0,1\}^{n}\to\{0,1\}$ is defined as the set
$\calA:=\{x^{n}\in\{0,1\}^{n}:f(x^{n})=1\}$.

If we denote the supports of $f$ and $g$ as $\calA$ and $\calB$
respectively, then one can rewrite \eqref{eqn:jp_forward} and \eqref{eqn:jp_reverse}
respectively as 
\begin{align}
\overline{\Gamma}^{(n)}(a,b) & =\underset{\calA,\calB\subset\{0,1\}^{n}\, :\, \pi_{X}^{n}(\calA)\le a,\, \pi_{Y}^{n}(\calB)\le b}{\max}\pi_{XY}^{n}(\calA\times\calB),\label{eqn:jp_for}
\end{align}
and 
\begin{align}
\underline{\Gamma}^{(n)}(a,b)=\underset{\calA,\calB\subset\{0,1\}^{n}\, :\, \pi_{X}^{n}(\calA)\geq a,\, \pi_{Y}^{n}(\calB)\geq b}{\min}\pi_{XY}^{n}(\calA\times\calB).\label{eqn:jp_rev}
\end{align}
Let $\overline{\Gamma}^{(\infty)}$ and $\underline{\Gamma}^{(\infty)}$  respectively
denote the pointwise limits of $\overline{\Gamma}^{(n)}$ and $\underline{\Gamma}^{(n)}$
as $n\to\infty$, i.e., 
\begin{equation}
\overline{\Gamma}^{(\infty)}(a,b)\!:=\!\lim_{n\to\infty}\overline{\Gamma}^{(n)}(a,b)\quad\mbox{and}\quad\underline{\Gamma}^{(\infty)}(a,b):=\lim_{n\to\infty}\underline{\Gamma}^{(n)}(a,b).\label{eqn:Gamma_infty}
\end{equation}
These are respectively known as the {\em asymptotic forward} and
{\em asymptotic reverse joint probabilities}.

By definition,   the forward and reverse joint probabilities are
non-decreasing in each of the parameters when the other is fixed.
This implies that there exists an optimal pair of sets $\calA,\calB\subset\{0,1\}^{n}$
(or Boolean functions $(f,g)$) attaining the forward joint probability
such that 
\begin{equation}
\pi_{X}^{n}(\calA)=\frac{\lfloor a\cdot2^{n}\rfloor}{2^{n}}\quad\mbox{and}\quad \pi_{Y}^{n}(\calB)=\frac{\lfloor b\cdot2^{n}\rfloor}{2^{n}}.
\end{equation}
Indeed, if either of these statements were not true, we can enlarge $\calA$ (resp.\ $\calB$)
to make its $\pi_{X}^{n}$-probability (resp.\ $\pi_{Y}^{n}$-probability)
closer to $a$ (resp.\ $b$). %To see this, note that  for an optimal pair of subsets $(A,B) \subsets \{0,1\}^n$, we can enlarge
%the subsets to make $\pi_{X}^{n}(A),\pi_{Y}^{n}(B)$
%respectively close to $a,b$. Note that this enlarging 
%operation  does not influence the optimality of $(A,B)$. 
Similarly, there exists an optimal pair $(\calA,\calB)$ (or Boolean
functions $(f,g)$) attaining the reverse joint probability such that
\begin{equation}
\pi_{X}^{n}(\calA)=\frac{\lceil a\cdot2^{n}\rceil}{2^{n}}\quad\mbox{and}\quad \pi_{Y}^{n}(\calB)=\frac{\lceil b\cdot2^{n}\rceil}{2^{n}}.
\end{equation}
As a consequence, for dyadic rationals $a$ and $b$ (i.e., $a={M}/{2^{n}},b={N}/{2^{n}}$
with integers $M,N\in\{0,1,\ldots,2^{n}\}$), the inequalities in
the constraints in the definitions of forward and reverse probabilities
(i.e., $\overline{\Gamma}^{(n)}(a,b)$ and $\underline{\Gamma}^{(n)}(a,b)$)
can be replaced by equalities, without affecting their values. %and $\underline{\Gamma}^{(n)}(a,b)$. 
These observations also allow  us to conclude that 
\begin{align}
\overline{\Gamma}^{(n)}(1-a,b)=b-\underline{\Gamma}^{(n)}(a,b)\quad\mbox{for all dyadic rationals}\;  a,b .
\end{align}
When we consider the asymptotic case in which $n\to\infty$, i.e.,
the quantities in~\eqref{eqn:Gamma_infty}, %When we take limits as $n\to\infty$, 
the requirement that $a$ and $b$ are dyadic rationals can be removed.
This implies that for any $a,b\in[0,1]$, 
\begin{align}
\overline{\Gamma}^{(\infty)}(1-a,b)=b-\underline{\Gamma}^{(\infty)}(a,b).\label{eqn:equiv_for_rev}
\end{align}
Hence, for all $(a,b)\in[0,1]^{2}$, determining the asymptotic forward
joint probability in~\eqref{eqn:Gamma_infty} is equivalent to determining
the asymptotic reverse joint probability and vice versa. % for all $(a,b)$. 

\subsection{Asymptotic Regimes and Exponents of Interest}

The identification of the optimal pairs $(\calA,\calB)$ that attain
the forward or reverse joint probabilities in \eqref{eqn:jp_for}
and \eqref{eqn:jp_rev} constitutes a combinatorial problem and is
thus difficult in general. Hence, we focus on the limiting cases as
$n\to\infty$ as this simplifies the problem, and the resultant problems
are also information-theoretic in nature. Specifically, the following
three asymptotic regimes will be considered. 
\begin{enumerate}
\item \underline{Central limit} (CL) regime: We set $a$ and $b$ to be
constants. We write $a=2^{-\alpha}$ and $b=2^{-\beta}$ for a pair
of constants $(\alpha,\beta)\in[0,\infty)^{2}$. %, in order to be consistent
%with the other two limiting cases. 
\item \underline{Large deviations} (LD) regime: We set $a$ and $b$ to
be sequences that vanish exponentially fast as $n\to\infty$. In particular,
we write $a=2^{-n\alpha}$ and $b=2^{-n\beta}$ for a pair of constants
$(\alpha,\beta)\in[0,1]^{2}$. 
\item \label{item:MDregime} \underline{Moderate deviations} (MD) regime: We set $a$ and $b$
to be sequences that vanish subexponentially fast as $n\to\infty$.
More precisely, $a=2^{-\theta_{n}\alpha},b=2^{-\theta_{n}\beta}$
for a pair of constants $(\alpha,\beta)\in[0,\infty)^{2}$, where
$\{\theta_{n}\}_{n\in\bbN}$ is a positive sequence satisfying $\theta_{n}\to\infty$
and ${\theta_{n}}/{n}\to0$, henceforth called an {\em MD sequence}. 
%The   sequence $\{ \theta_{n}\} $ widely appears in   the usual moderate deviation theory. 
\end{enumerate}
%It is easily observed that the MD regime falls in between the CL and
%LD regimes. In the MD regime, $a,b$ vanish but slower than the
%exponentially, while in the CL  and LD regimes, $a,b$ are fixed or
%vanish exponentially. 

The MD regime straddles between the CL and LD regimes. It is usually
the case if one solves a certain information-theoretic problem in
the CL or the LD regimes, a result for the MD regime can be derived
as a corollary, for example, by appealing to Taylor's theorem; see~\citet{altug10, PV10, Tan12} for example. We
will see that this is also the case for the NICD problem.

In the following section, we will set $\calA$ and $\calB$ to be
subcubes, Hamming balls, and Hamming spheres. These are prototypical
subsets in the Hamming space that are amenable to analyses. We will then
apply various probabilistic limit theorems---such as the central
limit theorem and large and moderate deviations theorems---to derive
the ``performances'' of these subsets in attaining the forward and
reverse joint probabilities. %will introduce them  in details, and will apply the central limit theorems, large deviation theory, and moderate deviation theory to derive the performance of Hamming balls or spheres. These will give the reason  why we call the regimes as ``central limit'', ``large deviation'',
%and ``moderate deviation'' regimes.  Before doing that, we first
We formally define several exponents of interest. %These will serve
%as our performance metrics. 
\begin{definition}\label{def:for_rev_exp}
Consider the following exponents: 
\begin{enumerate}
\item \underline{Forward and reverse CL exponents}: For $\alpha,\beta\in[0,\infty)$,
\begin{align}
\underline{\Upsilon}_{\mathrm{CL}}^{(n)}(\alpha,\beta) & :=-\log\overline{\Gamma}^{(n)}(2^{-\alpha},2^{-\beta})\quad\mbox{and}\label{eqn:CL_for}\\
\overline{\Upsilon}_{\mathrm{CL}}^{(n)}(\alpha,\beta) & :=-\log\underline{\Gamma}^{(n)}(2^{-\alpha},2^{-\beta}).\label{eqn:CL_rev}
\end{align}
\item \underline{Forward and reverse LD exponents}: For $\alpha,\beta\in[0,1]$,
\begin{align}
\underline{\Upsilon}_{\mathrm{LD}}^{(n)}(\alpha,\beta) & :=-\frac{1}{n}\log\overline{\Gamma}^{(n)}(2^{-n\alpha},2^{-n\beta})\quad\mbox{and}\label{eqn:LD_for}\\
\overline{\Upsilon}_{\mathrm{LD}}^{(n)}(\alpha,\beta) & :=-\frac{1}{n}\log\underline{\Gamma}^{(n)}(2^{-n\alpha},2^{-n\beta}).\label{eqn:LD_rev}
\end{align}
\item \underline{Forward and reverse  MD exponents}: 
%Given a positive sequence
%$\{\theta_{n}\}_{n\in\bbN}$ such that $\theta_{n}\to\infty$ and
%${\theta_{n}}/{n}\to0$, and 
Given an MD sequence  $\{\theta_{n}\}$, and for $\alpha,\beta\in[0,\infty)$, 
\begin{align}
\underline{\Upsilon}_{\mathrm{MD}}^{(n)}(\alpha,\beta) & :=-\frac{1}{\theta_{n}}\log\overline{\Gamma}^{(n)}(2^{-\theta_{n}\alpha},2^{-\theta_{n}\beta})\quad\mbox{and}\label{eqn:MD_for}\\
\overline{\Upsilon}_{\mathrm{MD}}^{(n)}(\alpha,\beta) & :=-\frac{1}{\theta_{n}}\log\underline{\Gamma}^{(n)}(2^{-\theta_{n}\alpha},2^{-\theta_{n}\beta}).\label{eqn:MD_rev}
\end{align}
\item Define $\underline{\Upsilon}_{\mathrm{CL}}^{(\infty)}$, $\overline{\Upsilon}_{\mathrm{CL}}^{(\infty)}$,
$\underline{\Upsilon}_{\mathrm{LD}}^{(\infty)}$, $\overline{\Upsilon}_{\mathrm{LD}}^{(\infty)}$,
$\underline{\Upsilon}_{\mathrm{MD}}^{(\infty)}$, and $\overline{\Upsilon}_{\mathrm{MD}}^{(\infty)}$
as the pointwise limits of the above exponents as $n\to\infty$. 
\end{enumerate}
\end{definition} The reader may notice that the  definitions
in \eqref{eqn:CL_for}--\eqref{eqn:MD_rev} appear to be redundant,
since each of the forward (resp.\ reverse) exponents is equivalent
to the forward (resp.\ reverse) joint probability in the sense that
if the forward (resp.\ reverse) joint probability has been determined,
then each of the forward (resp.\ reverse) exponents has also been
determined. %so  each of the forward (resp. reverse) exponents, and vice versa.
This also means the forward (resp.\ reverse) exponents are also ``equivalent''.
For example, for each~$n \in\bbN$, 
%\begin{equation}
 $\underline{\Upsilon}_{\mathrm{LD}}^{(n)}(\alpha,\beta)=\frac{1}{n}\underline{\Upsilon}_{\mathrm{CL}}^{(n)}(n\alpha,n\beta)$ and $\underline{\Upsilon}_{\mathrm{MD}}^{(n)}(\alpha,\beta)=\frac{1}{\theta_{n}}\underline{\Upsilon}_{\mathrm{CL}}^{(n)}(\theta_{n}\alpha,\theta_{n}\beta)$.
We introduce these notations because in the sequel, we will introduce
several {\em dimension-free} bounds (e.g., Theorem~\ref{thm:sse}) that can be conveniently expressed in
terms of the exponents defined in~\eqref{eqn:CL_for}--\eqref{eqn:MD_rev}. Here, 
a dimension-free bound is one that is independent of the dimension (or blocklength) $n$, but is valid for all dimensions $n$. 

In the following, we introduce bounds on the NICD exponents in~\eqref{eqn:CL_for}--\eqref{eqn:MD_rev}.
As is conventional in information theory, there are two parts to this
endeavor. In the achievability part that will be discussed in Section~\ref{sec:nicd_ach},
we construct subsets $\calA$ and $\calB$ that upper bound the forward exponents and lower bound the reverse exponents.
In the converse parts that will be discussed in Section~\ref{sec:nicd_conv}--\ref{sec:ldr},
we demonstrate impossibility results, i.e., lower bounds on the forward exponents and upper bounds on the reverse exponents.
%Thes  bounds match the bounds in the achievability parts in some special cases. 
The achievability  and converse bounds match  in some special cases. 
%The achievability parts are typically easier than the converse parts.

%These bounds are classified into 
%into two , as is conventionally done in information theory. These parts include  the achievability
%part (in Section~\ref{sec:nicd_ach}) and the converse part (in Section~\ref{sec:nicd_conv}). 

\section{Achievability: Subcubes, Hamming Balls, and Spheres}
\label{sec:nicd_ach}

We now consider the achievability parts, i.e., deriving lower bounds
for the forward joint probability and upper bounds for the reverse
joint probability. For these parts, we consider three canonical types
of subsets in Hamming space---subcubes, Hamming balls, and Hamming
spheres.% (or spherical shells).

\subsection{Subcubes}
\label{sec:subcubes}

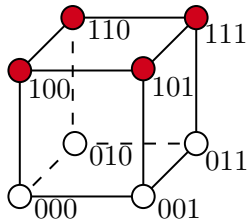
\begin{figure}
\centering \tikzset{every picture/.style={line width=0.75pt}}
%set default line width to 0.75pt        
\begin{tikzpicture}[x=0.75pt,y=0.75pt,yscale=-1,xscale=1,scale=0.50]  %uncomment if require: \path (0,2353); %set diagram left start at 0, and has height of 2353
%Shape: Cube [id:dp49838998304487325] 
\draw   (611.03,1556.01) -- (664.42,1502.62) -- (789,1502.62) -- (789,1629.78) -- (735.61,1683.17) -- (611.03,1683.17) -- cycle ;  \draw   (789,1502.62) -- (735.61,1556.01) -- (611.03,1556.01) ;  \draw   (735.61,1556.01) -- (735.61,1683.17) ; %Straight Lines [id:da5028555286951217] 
\draw  [dash pattern={on 4.5pt off 4.5pt}]  (664.42,1502.62) -- (664.42,1632.01) ; %Straight Lines [id:da10759638466390542] 
\draw  [dash pattern={on 4.5pt off 4.5pt}]  (666.65,1629.78) -- (611.03,1685.41) ; %Straight Lines [id:da9894502640174585] 
\draw  [dash pattern={on 4.5pt off 4.5pt}]  (664.42,1629.78) -- (789,1629.78) ; %Shape: Ellipse [id:dp3865675349690556] 
\draw  [fill={rgb, 255:red, 208; green, 2; blue, 27 }  ,fill opacity=1 ] (724.08,1554.22) .. controls (724.08,1547.74) and (729.21,1542.5) .. (735.54,1542.5) .. controls (741.87,1542.5) and (747,1547.74) .. (747,1554.22) .. controls (747,1560.69) and (741.87,1565.93) .. (735.54,1565.93) .. controls (729.21,1565.93) and (724.08,1560.69) .. (724.08,1554.22) -- cycle ; %Shape: Ellipse [id:dp7931030190994239] 
\draw  [fill={rgb, 255:red, 255; green, 255; blue, 255 }  ,fill opacity=1 ] (655.19,1629.78) .. controls (655.19,1623.31) and (660.32,1618.06) .. (666.65,1618.06) .. controls (672.98,1618.06) and (678.11,1623.31) .. (678.11,1629.78) .. controls (678.11,1636.25) and (672.98,1641.5) .. (666.65,1641.5) .. controls (660.32,1641.5) and (655.19,1636.25) .. (655.19,1629.78) -- cycle ; %Shape: Ellipse [id:dp7257815674863457] 
\draw  [fill={rgb, 255:red, 208; green, 2; blue, 27 }  ,fill opacity=1 ] (599.57,1556.01) .. controls (599.57,1549.54) and (604.7,1544.29) .. (611.03,1544.29) .. controls (617.35,1544.29) and (622.48,1549.54) .. (622.48,1556.01) .. controls (622.48,1562.48) and (617.35,1567.73) .. (611.03,1567.73) .. controls (604.7,1567.73) and (599.57,1562.48) .. (599.57,1556.01) -- cycle ; %Shape: Ellipse [id:dp9223342681548201] 
\draw  [fill={rgb, 255:red, 208; green, 2; blue, 27 }  ,fill opacity=1 ] (652.96,1502.62) .. controls (652.96,1496.14) and (658.09,1490.9) .. (664.42,1490.9) .. controls (670.75,1490.9) and (675.88,1496.14) .. (675.88,1502.62) .. controls (675.88,1509.09) and (670.75,1514.33) .. (664.42,1514.33) .. controls (658.09,1514.33) and (652.96,1509.09) .. (652.96,1502.62) -- cycle ; %Shape: Ellipse [id:dp05944363535186503] 
\draw  [fill={rgb, 255:red, 255; green, 255; blue, 255 }  ,fill opacity=1 ] (599.57,1683.17) .. controls (599.57,1676.7) and (604.7,1671.45) .. (611.03,1671.45) .. controls (617.35,1671.45) and (622.48,1676.7) .. (622.48,1683.17) .. controls (622.48,1689.64) and (617.35,1694.89) .. (611.03,1694.89) .. controls (604.7,1694.89) and (599.57,1689.64) .. (599.57,1683.17) -- cycle ; %Shape: Ellipse [id:dp9267848797173599] 
\draw  [fill={rgb, 255:red, 255; green, 255; blue, 255 }  ,fill opacity=1 ] (724.15,1683.17) .. controls (724.15,1676.7) and (729.28,1671.45) .. (735.61,1671.45) .. controls (741.94,1671.45) and (747.07,1676.7) .. (747.07,1683.17) .. controls (747.07,1689.64) and (741.94,1694.89) .. (735.61,1694.89) .. controls (729.28,1694.89) and (724.15,1689.64) .. (724.15,1683.17) -- cycle ; %Shape: Ellipse [id:dp7613806876467912] 
\draw  [fill={rgb, 255:red, 255; green, 255; blue, 255 }  ,fill opacity=1 ] (777.54,1629.78) .. controls (777.54,1623.31) and (782.67,1618.06) .. (789,1618.06) .. controls (795.33,1618.06) and (800.46,1623.31) .. (800.46,1629.78) .. controls (800.46,1636.25) and (795.33,1641.5) .. (789,1641.5) .. controls (782.67,1641.5) and (777.54,1636.25) .. (777.54,1629.78) -- cycle ; %Shape: Ellipse [id:dp6593416367588012] 
\draw  [fill={rgb, 255:red, 208; green, 2; blue, 27 }  ,fill opacity=1 ] (777.54,1502.62) .. controls (777.54,1496.14) and (782.67,1490.9) .. (789,1490.9) .. controls (795.33,1490.9) and (800.46,1496.14) .. (800.46,1502.62) .. controls (800.46,1509.09) and (795.33,1514.33) .. (789,1514.33) .. controls (782.67,1514.33) and (777.54,1509.09) .. (777.54,1502.62) -- cycle ;
% Text Node
\draw (622.48,1683.17) node [anchor=north west][inner sep=0.75pt]  [font=\normalsize]  {$000$}; % Text Node
\draw (615.97,1562.19) node [anchor=north west][inner sep=0.75pt]  [font=\normalsize]  {$100$}; % Text Node
\draw (795,1635.5) node [anchor=north west][inner sep=0.75pt]  [font=\normalsize]  {$011$}; % Text Node
\draw (678.11,1629.78) node [anchor=north west][inner sep=0.75pt]  [font=\normalsize]  {$010$}; % Text Node
\draw (747.07,1683.17) node [anchor=north west][inner sep=0.75pt]  [font=\normalsize]  {$001$}; % Text Node
\draw (794,1506.5) node [anchor=north west][inner sep=0.75pt]  [font=\normalsize]  {$111$}; % Text Node
\draw (675.88,1502.62) node [anchor=north west][inner sep=0.75pt]  [font=\normalsize]  {$110$}; % Text Node
\draw (741,1560.22) node [anchor=north west][inner sep=0.75pt]  [font=\normalsize]  {$101$};
\end{tikzpicture}

\caption{\label{fig:subcube}A subcube (shaded) in $\{0,1\}^{3}$ with the
first component fixed to $1$}
\end{figure}

An {\em $(n-k)$-subcube} $\bbC_{n-k}$ is a set of vectors
$x^{n}\in\{0,1\}^{n}$ with $k$ components held fixed. For example,
if we fix the first $k$ components to $1$, then we get the $(n-k)$-subcube
$\{1^{k}\}\times\{0,1\}^{n-k}$, where $1^{k}$ denotes the length-$k$
all-ones vector. For any set $\calA\subset\{0,1\}^{n}$, we say that
its {\em indicator}, denoted as $\bone_{\calA}$, is the function $f:\{0,1\}^{n}\to\{0,1\}$
such that $f(x^{n})=1$ for all $x^{n}\in\calA$ and $f(x^{n})=0$
for all $x^{n}\notin\calA$.  The indicator of the subcube $\{{1}^{k}\}\times\{0,1\}^{n-k}$
is $x^{n}\in\{0,1\}^{n}\mapsto\prod_{i=1}^{k}x_{i}$. An important
class of subcubes is the class of $(n-1)$-subcubes, e.g., $\{1\}\times\{0,1\}^{n-1}$.
An $(n-1)$-subcube with $n=3$ is illustrated in Fig.~\ref{fig:subcube}.
The indicators
of $( n-1)$-subcubes are the functions $x^n \mapsto  x_i$ or $x^n\mapsto 1-x_i$ for $i \in  [n]$. Such functions are known as
{\em dictator} functions.

%The indicators of $(n-1)$-subcubes are the functions $x^{n}\mapsto x_{i}$
%for $i\in[n]$. Such functions, together with their complements $x^{n}\mapsto1-x_{i}$
%for $i\in[n]$, are known as {\em dictator} functions.

We now return to the NICD problem. For $a=b=2^{-k}$ for a  positive integer~$k$,
we choose $\calA$ and $\calB$ as a pair of identical $(n-k)$-subcubes. By referring to the joint distribution in \eqref{eq:NICDDSBS}, we see that  the joint probability induced by  $(\calA, \calB)$  is 
\begin{equation}
\pi_{XY}^{n}(\calA\times\calB)=\pi_{XY}(1,1)^{k}=\Big(\frac{1+\rho}{4}\Big)^{k}.\label{eqn:subcube_id}
\end{equation}
On the other hand, if we choose $\calA$ and $\calB$ as a pair of
anti-symmetric $(n-k)$-subcubes, i.e., $\calA=1^{n}-\calB=\mathcal{C}_{n-k}$,
then the induced joint probability is 
\begin{equation}
\pi_{XY}^{n}(\calA\times\calB)=\pi_{XY}(1,0)^{k}=\Big(\frac{1-\rho}{4}\Big)^{k}.\label{eqn:subcube_anti}
\end{equation}
For the more general case in which $a=2^{-k_{1}}$ and $b=2^{-k_{2}}$
for integers $0\le k_{1}\le k_{2}$, if we choose $(\calA,\calB)$
as a pair of ``nested'' subcubes, i.e., $\calA=\{1^{k_{1}}\}\times\{0,1\}^{n-k_{1}}$
and $\calB=\{1^{k_{2}}\}\times\{0,1\}^{n-k_{2}}$, then the induced
joint probability 
\begin{equation}
\pi_{XY}^{n}(\calA\times\calB)=\Big(\frac{1}{2}\Big)^{k_{2}-k_{1}}\Big(\frac{1+\rho}{4}\Big)^{k_{1}}.\label{eqn:subcube_nested}
\end{equation}
For the same case, if we choose $(\calA,\calB)$ as a pair of ``anti-nested''
subcubes, i.e., $\calA=\{1^{k_{1}}\}\times\{0,1\}^{n-k_{1}}$ and
$\calB=\{0^{k_{2}}\}\times\{0,1\}^{n-k_{2}}$, then 
%the induced joint
%probability 
\begin{equation}
\pi_{XY}^{n}(\calA\times\calB)=\Big(\frac{1}{2}\Big)^{k_{2}-k_{1}}\Big(\frac{1-\rho}{4}\Big)^{k_{1}}.\label{eqn:subcube_anti_nested}
\end{equation}

We now discuss the case in which $a$ and $b$ are dyadic rationals
(i.e., $a={M}/{2^{n}},b={N}/{2^{n}}$ for some integers $M,N$). Observe
that if a dyadic rational $a$ is not equal to $2^{-k}$ for some
integer $k$, then there is no subcube with $\pi_{X}^{n}$-probability
{\em exactly} equal to $a$. Hence, to achieve better performances, a generalization
of subcubes $\{0^{k}\}\times\{0,1\}^{n-k}$ and $\{1^{k}\}\times\{0,1\}^{n-k}$,
called {\em lexicographic sets}, turns out to be useful. A subset
of $\{0,1\}^{n}$ is called {\em lexicographic} if the elements
are selected as the first sequences in some lexicographic order (either
ascending or descending). A Boolean function is called {\em
lexicographic} if its support is a lexicographic set. %Note that we only compare two sequences of the same length. For
%this case, lexicographic ordering is the same as (if $0<1$ is assumed)
%or opposite to (if $0>1$ is assumed) the numerical ordering. 
By setting $\calA$ and $\calB$ to be two lexicographic sets both
in ascending (or descending) order, we can obtain a relatively  large
joint probability $\pi_{XY}^{n}(\calA\times\calB)$. On the other hand,
if we set $\calA$ and~$\calB$ to be two lexicographic sets such
that one is chosen in ascending order and the other in descending
order, we can obtain a relatively  small joint probability $\pi_{XY}^{n}(\calA\times\calB)$.
The explicit expressions for these two joint probabilities are complicated,
and thus we omit them. % Assume that we adopt the convention $0<1$. 
A lexicographic set chosen in ascending order can then be written
as $\{x^{n}\in\{0,1\}^{n}:\sum_{i=1}^{n}2^{i-1}x_{i}\le r\}$ for
some~$r$. This is a special case of so-called {\em linear threshold
functions}, which is discussed in detail in \cite{ODonnell14analysisof}.

%This is a special case of {\em linear threshold
%functions} (LTF). is defined as $\{ \mathbf{x}\in\{0,1\}^{n}:\sum_{i=1}^{n}\omega_{i}x_{i}\le r\} $
%for an arbitrary vector $(\omega_{i})$ and an arbitrary value $r$.
%We do not discuss LTFs here. Interested readers can refer to \cite{ODonnell14analysisof}
%for more details on LTFs. 

\subsection{Hamming Balls} \label{sec:ball}

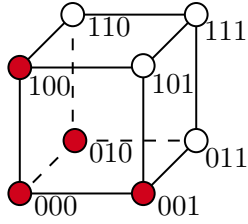
\begin{figure}
\centering \tikzset{every picture/.style={line width=0.75pt}}
%set default line width to 0.75pt        
\begin{tikzpicture}[x=0.75pt,y=0.75pt,yscale=-1,xscale=1,scale=0.50]    %uncomment if require: \path (0,2353); %set diagram left start at 0, and has height of 2353
%Shape: Cube [id:dp7810144077812278] 
\draw   (233.03,1897.01) -- (286.42,1843.62) -- (411,1843.62) -- (411,1970.78) -- (357.61,2024.17) -- (233.03,2024.17) -- cycle ;  \draw   (411,1843.62) -- (357.61,1897.01) -- (233.03,1897.01) ;  \draw   (357.61,1897.01) -- (357.61,2024.17) ; %Straight Lines [id:da16150338451359714] 
\draw  [dash pattern={on 4.5pt off 4.5pt}]  (286.42,1843.62) -- (286.42,1973.01) ; %Straight Lines [id:da7557989581632345] 
\draw  [dash pattern={on 4.5pt off 4.5pt}]  (288.65,1970.78) -- (233.03,2026.41) ; %Straight Lines [id:da5900254841833406] 
\draw  [dash pattern={on 4.5pt off 4.5pt}]  (286.42,1970.78) -- (411,1970.78) ; %Shape: Ellipse [id:dp5351142727503628] 
\draw  [fill={rgb, 255:red, 255; green, 255; blue, 255 }  ,fill opacity=1 ] (346.08,1895.22) .. controls (346.08,1888.74) and (351.21,1883.5) .. (357.54,1883.5) .. controls (363.87,1883.5) and (369,1888.74) .. (369,1895.22) .. controls (369,1901.69) and (363.87,1906.93) .. (357.54,1906.93) .. controls (351.21,1906.93) and (346.08,1901.69) .. (346.08,1895.22) -- cycle ; %Shape: Ellipse [id:dp24970151488861947] 
\draw  [fill={rgb, 255:red, 208; green, 2; blue, 27 }  ,fill opacity=1 ] (277.19,1970.78) .. controls (277.19,1964.31) and (282.32,1959.06) .. (288.65,1959.06) .. controls (294.98,1959.06) and (300.11,1964.31) .. (300.11,1970.78) .. controls (300.11,1977.25) and (294.98,1982.5) .. (288.65,1982.5) .. controls (282.32,1982.5) and (277.19,1977.25) .. (277.19,1970.78) -- cycle ; %Shape: Ellipse [id:dp6581665423810694] 
\draw  [fill={rgb, 255:red, 208; green, 2; blue, 27 }  ,fill opacity=1 ] (221.57,1897.01) .. controls (221.57,1890.54) and (226.7,1885.29) .. (233.03,1885.29) .. controls (239.35,1885.29) and (244.48,1890.54) .. (244.48,1897.01) .. controls (244.48,1903.48) and (239.35,1908.73) .. (233.03,1908.73) .. controls (226.7,1908.73) and (221.57,1903.48) .. (221.57,1897.01) -- cycle ; %Shape: Ellipse [id:dp19242973943382546] 
\draw  [fill={rgb, 255:red, 255; green, 255; blue, 255 }  ,fill opacity=1 ] (274.96,1843.62) .. controls (274.96,1837.14) and (280.09,1831.9) .. (286.42,1831.9) .. controls (292.75,1831.9) and (297.88,1837.14) .. (297.88,1843.62) .. controls (297.88,1850.09) and (292.75,1855.33) .. (286.42,1855.33) .. controls (280.09,1855.33) and (274.96,1850.09) .. (274.96,1843.62) -- cycle ; %Shape: Ellipse [id:dp11934107900724844] 
\draw  [fill={rgb, 255:red, 208; green, 2; blue, 27 }  ,fill opacity=1 ] (221.57,2024.17) .. controls (221.57,2017.7) and (226.7,2012.45) .. (233.03,2012.45) .. controls (239.35,2012.45) and (244.48,2017.7) .. (244.48,2024.17) .. controls (244.48,2030.64) and (239.35,2035.89) .. (233.03,2035.89) .. controls (226.7,2035.89) and (221.57,2030.64) .. (221.57,2024.17) -- cycle ; %Shape: Ellipse [id:dp6213531125665153] 
\draw  [fill={rgb, 255:red, 208; green, 2; blue, 27 }  ,fill opacity=1 ] (346.15,2024.17) .. controls (346.15,2017.7) and (351.28,2012.45) .. (357.61,2012.45) .. controls (363.94,2012.45) and (369.07,2017.7) .. (369.07,2024.17) .. controls (369.07,2030.64) and (363.94,2035.89) .. (357.61,2035.89) .. controls (351.28,2035.89) and (346.15,2030.64) .. (346.15,2024.17) -- cycle ; %Shape: Ellipse [id:dp8441980301881826] 
\draw  [fill={rgb, 255:red, 255; green, 255; blue, 255 }  ,fill opacity=1 ] (399.54,1970.78) .. controls (399.54,1964.31) and (404.67,1959.06) .. (411,1959.06) .. controls (417.33,1959.06) and (422.46,1964.31) .. (422.46,1970.78) .. controls (422.46,1977.25) and (417.33,1982.5) .. (411,1982.5) .. controls (404.67,1982.5) and (399.54,1977.25) .. (399.54,1970.78) -- cycle ; %Shape: Ellipse [id:dp5386645565311878] 
\draw  [fill={rgb, 255:red, 255; green, 255; blue, 255 }  ,fill opacity=1 ] (399.54,1843.62) .. controls (399.54,1837.14) and (404.67,1831.9) .. (411,1831.9) .. controls (417.33,1831.9) and (422.46,1837.14) .. (422.46,1843.62) .. controls (422.46,1850.09) and (417.33,1855.33) .. (411,1855.33) .. controls (404.67,1855.33) and (399.54,1850.09) .. (399.54,1843.62) -- cycle ;
% Text Node
\draw (244.48,2024.17) node [anchor=north west][inner sep=0.75pt]  [font=\normalsize]  {$000$}; % Text Node
\draw (237.97,1903.19) node [anchor=north west][inner sep=0.75pt]  [font=\normalsize]  {$100$}; % Text Node
\draw (417,1976.5) node [anchor=north west][inner sep=0.75pt]  [font=\normalsize]  {$011$}; % Text Node
\draw (300.11,1970.78) node [anchor=north west][inner sep=0.75pt]  [font=\normalsize]  {$010$}; % Text Node
\draw (369.07,2024.17) node [anchor=north west][inner sep=0.75pt]  [font=\normalsize]  {$001$}; % Text Node
\draw (416,1847.5) node [anchor=north west][inner sep=0.75pt]  [font=\normalsize]  {$111$}; % Text Node
\draw (297.88,1843.62) node [anchor=north west][inner sep=0.75pt]  [font=\normalsize]  {$110$}; % Text Node
\draw (363,1901.22) node [anchor=north west][inner sep=0.75pt]  [font=\normalsize]  {$101$};
\end{tikzpicture}

\caption{\label{fig:ball} A Hamming ball (shaded) in $\{0,1\}^{3}$ centered
at $(0,0,0)$ with radius $1$ }
\end{figure}

A {\em Hamming ball} centered at $y^{n}\in\{0,1\}^{n}$ with radius
$r\in\{0,1,\ldots,n\}$ takes the form $\mathbb{B}_{r}(y^{n}):=\{x^{n}\in\{0,1\}^{n}:d_{\mathrm{H}}(x^{n},y^{n})\le r\}$,
where $d_{\mathrm{H}}(x^{n},y^{n}):=\sum_{i=1}^{n}\bone\{x_{i}\neq y_{i}\}$
denotes the {\em Hamming distance} between vectors $x^{n}$ and
$y^{n}$. An example of a Hamming ball with radius $1$ is illustrated
in Fig.~\ref{fig:ball}. In the following, we only consider Hamming
balls that are centered at $0^{n}=(0,0,\ldots,0)$ or $1^{n}=(1,1,\ldots,1)$.
For these Hamming balls (with radius $r$), we can rewrite them as
$\{x^{n}\in\{0,1\}^{n}:\sum_{i=1}^{n}x_{i}\le r\}$ and $\{x^{n}\in\{0,1\}^{n}:\sum_{i=1}^{n}x_{i}\ge n-r\}$
respectively.

We now set $\calA$ and $\calB$ in the NICD problem to be Hamming
balls. We first consider the CL regime in which we choose $\calA$
and $\calB$ to be a pair of {\em concentric} Hamming balls. More
specifically, $\calA_{n}:=\mathbb{B}_{r_{n}}(0^{n})$ and $\calB_{n}=\mathbb{B}_{s_{n}}(0^{n})$
for some sequences $\{r_{n}\}_{n\in\bbN}$ and $\{s_{n}\}_{n\in\bbN}$.
We append the subscript $n$ to $\calA$ and $\calB$, to indicate
that these two sets depend on $n$. We can rewrite $\calA_{n}$ as
$\{x^{n}:\sum_{i=1}^{n}x_{i}\le r_{n}\}$. Hence, the marginal probability
$\pi_{X}^{n}(\calA_{n})$ can be written as $\Pr(\sum_{i=1}^{n}X_{i}\le r_{n})$
where $\{X_{i}\}_{i=1}^{n}$ are i.i.d.\ with each $X_{i}\sim\mathrm{Bern}(1/2)$.
To calculate the limiting value of this probability as $n\to\infty$,
one may apply several well-known concentration of measure theorems,
including the central limit theorem or various large deviations theorems.
Since we focus on the CL regime here, we require that $\pi_{X}^{n}(\calA_{n})$
tends to a non-vanishing constant. Hence, we set the radius $r_{n}=\frac{n}{2}+\frac{\lambda\sqrt{n}}{2}$
for some $\lambda\in\mathbb{R}$. Then, the (univariate) central limit
theorem yields 
\begin{align}
\lim_{n\to\infty}\pi_{X}^{n}(\calA_{n})=\Phi(\lambda),\label{eqn:cl1}
\end{align}
where $\Phi(\cdot)$ is the cumulative distribution function (CDF)
of the standard univariate Gaussian distribution. Similarly, if we
set the radius $s_{n}=\frac{n}{2}+\frac{\mu\sqrt{n}}{2}$ for some
$\mu\in\mathbb{R}$, we obtain 
\begin{align}
\lim_{n\to\infty}\pi_{Y}^{n}(\calB_{n})=\Phi(\mu).\label{eqn:cl2}
\end{align}

We now estimate the asymptotic value of the joint probability $\pi_{XY}^{n}(\calA_{n}\times\calB_{n})$
where $\calA_{n}$ and $\calB_{n}$ are concentric spheres with radii
$r_{n}$ and~$s_{n}$ respectively. Note that this probability can
be restated as $\Pr(\sum_{i=1}^{n}X_{i}\le r_{n},\sum_{i=1}^{n}Y_{i}\le s_{n})$
where $(X^n,Y^n)=\{(X_{i},Y_{i})\}_{i=1}^{n}$ is a source sequence generated by a DSBS with correlation coefficient $\rho$. 
%denotes a sequence of i.i.d.\ random
%variables distributed according to the doubly symmetric binary distribution
%$\pi_{XY}$ with correlation coefficient $\rho$ defined in~\eqref{eq:NICDDSBS}.
The multivariate central limit theorem then yields 
\begin{align}
\lim_{n\to\infty}\pi_{XY}^{n}(\calA_{n}\times\calB_{n})=\Phi_{\rho}(\lambda,\mu),%:=\int_{-\infty}^{\mu}\int_{-\infty}^{\lambda}\frac{1}{2\pi\sqrt{1-\rho^{2}}}\exp%x
%{bmatrix}^{\top}\begin{bmatrix}\end{bmatrix}
%1
%{bmatrix}^{-1}\begin{bmatrix}\end{bmatrix}
%x
%{bmatrix}
\label{eqn:biv}
\end{align}
where $\Phi_{\rho}(\cdot,\cdot)$ is the joint CDF of the zero-mean
bivariate Gaussian distribution with covariance matrix 
\begin{equation}
\bK:=\begin{bmatrix}1 & \rho\\
\rho & 1
\end{bmatrix}.\label{eqn:cov_mat_rho}
\end{equation}
%Recall that $\rho$ is the correlation coefficient of $(X,Y)\sim \pi_{XY}$ where the doubly symmetric binary distribution is defined in \eqref{eq:NICDDSBS}.

Based on the asymptotic results in \eqref{eqn:cl1}--\eqref{eqn:biv}, one can obtain a lower bound on the forward
joint probability in the NICD problem \cite[Ex.~9.24 and 10.5]{ODonnell14analysisof}.
\begin{proposition} \label{prop:CLsphere} For $a,b\in(0,1)$, 
\begin{equation}
\overline{\Gamma}^{(\infty)}(a,b)\ge\Lambda_{\rho}(a,b),\label{eq:NICDCL_ball}
\end{equation}
where 
\begin{equation}
\Lambda_{\rho}(a,b):=\Phi_{\rho}\big(\Phi^{-1}(a),\Phi^{-1}(b)\big).\label{eq:NICDcopula}
\end{equation}
\end{proposition} Here $\Lambda_{\rho}(\cdot,\cdot)$ is known as
the {\em bivariate normal copula} or the {\em Gaussian quadrant
probability function}. Thanks to the equivalence between the forward
and reverse joint probabilities as stated in~\eqref{eqn:equiv_for_rev},
\eqref{eq:NICDCL_ball} can alternatively be expressed in terms of
the reverse joint probability as 
\begin{align}
 & \underline{\Gamma}^{(\infty)}(a,b)\le\Lambda_{-\rho}(a,b).\label{eq:NICDCL_ball-1}
\end{align}
The upper bound $\Lambda_{-\rho}(a,b)$ is achieved by a sequence
of pairs of anti-concentric balls $\calA_{n}=\mathbb{B}_{r_{n}}(0^{n})$
and $\calB_{n}=\mathbb{B}_{s_{n}}(1^{n})$.

Considering the exponents of the probabilities in~\eqref{eq:NICDCL_ball}
and \eqref{eq:NICDCL_ball-1}, 
\begin{align}
\underline{\Upsilon}_{\mathrm{CL}}^{(\infty)}(\alpha,\beta) & \le\underline{\Upsilon}_{\mathrm{CL}}(\alpha,\beta):=-\log\Lambda_{\rho}(2^{-\alpha},2^{-\beta})\quad\mbox{and}\label{eqn:Theta_under}\\
\overline{\Upsilon}_{\mathrm{CL}}^{(\infty)}(\alpha,\beta) & \ge\overline{\Upsilon}_{\mathrm{CL}}(\alpha,\beta):=-\log\Lambda_{-\rho}(2^{-\alpha},2^{-\beta}).\label{eqn:Theta_over}
\end{align}
%where  
%\begin{align}
%\underline{\Upsilon}_{\mathrm{CL}}(\alpha,\beta)  & \quad\mbox{and}\\
%\overline{\Upsilon}_{\mathrm{CL}}(\alpha,\beta)  &
%\end{align}

We next consider the LD and MD regimes. Although it is certainly possible
to set $\calA_{n}$ and $\calB_{n}$ to be Hamming balls to obtain
achievability results for these two regimes, we prefer not to do so
here. This is because, it is much easier to derive the same results
by using Hamming \emph{spheres} or \emph{spherical shells}. Therefore, we consider the LD and
MD regimes in the following subsection after we introduce Hamming
spheres.

\subsection{Hamming Spheres}

\label{sec:spheres}

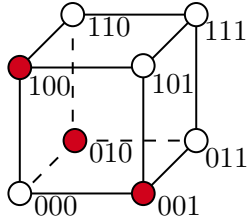
\begin{figure}
\centering \tikzset{every picture/.style={line width=0.75pt}}
%set default line width to 0.75pt        
\begin{tikzpicture}[x=0.75pt,y=0.75pt,yscale=-1,xscale=1,scale=0.50]    %uncomment if require: \path (0,2353); %set diagram left start at 0, and has height of 2353
%Shape: Cube [id:dp9231330005018561] 
\draw   (618.03,1896.01) -- (671.42,1842.62) -- (796,1842.62) -- (796,1969.78) -- (742.61,2023.17) -- (618.03,2023.17) -- cycle ;  \draw   (796,1842.62) -- (742.61,1896.01) -- (618.03,1896.01) ;  \draw   (742.61,1896.01) -- (742.61,2023.17) ; %Straight Lines [id:da8253252631689396] 
\draw  [dash pattern={on 4.5pt off 4.5pt}]  (671.42,1842.62) -- (671.42,1972.01) ; %Straight Lines [id:da132106199299159] 
\draw  [dash pattern={on 4.5pt off 4.5pt}]  (673.65,1969.78) -- (618.03,2025.41) ; %Straight Lines [id:da26350787400345] 
\draw  [dash pattern={on 4.5pt off 4.5pt}]  (671.42,1969.78) -- (796,1969.78) ; %Shape: Ellipse [id:dp8552713018854012] 
\draw  [fill={rgb, 255:red, 255; green, 255; blue, 255 }  ,fill opacity=1 ] (731.08,1894.22) .. controls (731.08,1887.74) and (736.21,1882.5) .. (742.54,1882.5) .. controls (748.87,1882.5) and (754,1887.74) .. (754,1894.22) .. controls (754,1900.69) and (748.87,1905.93) .. (742.54,1905.93) .. controls (736.21,1905.93) and (731.08,1900.69) .. (731.08,1894.22) -- cycle ; %Shape: Ellipse [id:dp6294182802804282] 
\draw  [fill={rgb, 255:red, 208; green, 2; blue, 27 }  ,fill opacity=1 ] (662.19,1969.78) .. controls (662.19,1963.31) and (667.32,1958.06) .. (673.65,1958.06) .. controls (679.98,1958.06) and (685.11,1963.31) .. (685.11,1969.78) .. controls (685.11,1976.25) and (679.98,1981.5) .. (673.65,1981.5) .. controls (667.32,1981.5) and (662.19,1976.25) .. (662.19,1969.78) -- cycle ; %Shape: Ellipse [id:dp9033501877334644] 
\draw  [fill={rgb, 255:red, 208; green, 2; blue, 27 }  ,fill opacity=1 ] (606.57,1896.01) .. controls (606.57,1889.54) and (611.7,1884.29) .. (618.03,1884.29) .. controls (624.35,1884.29) and (629.48,1889.54) .. (629.48,1896.01) .. controls (629.48,1902.48) and (624.35,1907.73) .. (618.03,1907.73) .. controls (611.7,1907.73) and (606.57,1902.48) .. (606.57,1896.01) -- cycle ; %Shape: Ellipse [id:dp26486504342467954] 
\draw  [fill={rgb, 255:red, 255; green, 255; blue, 255 }  ,fill opacity=1 ] (659.96,1842.62) .. controls (659.96,1836.14) and (665.09,1830.9) .. (671.42,1830.9) .. controls (677.75,1830.9) and (682.88,1836.14) .. (682.88,1842.62) .. controls (682.88,1849.09) and (677.75,1854.33) .. (671.42,1854.33) .. controls (665.09,1854.33) and (659.96,1849.09) .. (659.96,1842.62) -- cycle ; %Shape: Ellipse [id:dp31360876037658536] 
\draw  [fill={rgb, 255:red, 255; green, 255; blue, 255 }  ,fill opacity=1 ] (606.57,2023.17) .. controls (606.57,2016.7) and (611.7,2011.45) .. (618.03,2011.45) .. controls (624.35,2011.45) and (629.48,2016.7) .. (629.48,2023.17) .. controls (629.48,2029.64) and (624.35,2034.89) .. (618.03,2034.89) .. controls (611.7,2034.89) and (606.57,2029.64) .. (606.57,2023.17) -- cycle ; %Shape: Ellipse [id:dp27311877460589495] 
\draw  [fill={rgb, 255:red, 208; green, 2; blue, 27 }  ,fill opacity=1 ] (731.15,2023.17) .. controls (731.15,2016.7) and (736.28,2011.45) .. (742.61,2011.45) .. controls (748.94,2011.45) and (754.07,2016.7) .. (754.07,2023.17) .. controls (754.07,2029.64) and (748.94,2034.89) .. (742.61,2034.89) .. controls (736.28,2034.89) and (731.15,2029.64) .. (731.15,2023.17) -- cycle ; %Shape: Ellipse [id:dp6614656136197323] 
\draw  [fill={rgb, 255:red, 255; green, 255; blue, 255 }  ,fill opacity=1 ] (784.54,1969.78) .. controls (784.54,1963.31) and (789.67,1958.06) .. (796,1958.06) .. controls (802.33,1958.06) and (807.46,1963.31) .. (807.46,1969.78) .. controls (807.46,1976.25) and (802.33,1981.5) .. (796,1981.5) .. controls (789.67,1981.5) and (784.54,1976.25) .. (784.54,1969.78) -- cycle ; %Shape: Ellipse [id:dp9983602851465285] 
\draw  [fill={rgb, 255:red, 255; green, 255; blue, 255 }  ,fill opacity=1 ] (784.54,1842.62) .. controls (784.54,1836.14) and (789.67,1830.9) .. (796,1830.9) .. controls (802.33,1830.9) and (807.46,1836.14) .. (807.46,1842.62) .. controls (807.46,1849.09) and (802.33,1854.33) .. (796,1854.33) .. controls (789.67,1854.33) and (784.54,1849.09) .. (784.54,1842.62) -- cycle ;
% Text Node
\draw (629.48,2023.17) node [anchor=north west][inner sep=0.75pt]  [font=\normalsize]  {$000$}; % Text Node
\draw (622.97,1902.19) node [anchor=north west][inner sep=0.75pt]  [font=\normalsize]  {$100$}; % Text Node
\draw (802,1975.5) node [anchor=north west][inner sep=0.75pt]  [font=\normalsize]  {$011$}; % Text Node
\draw (685.11,1969.78) node [anchor=north west][inner sep=0.75pt]  [font=\normalsize]  {$010$}; % Text Node
\draw (754.07,2023.17) node [anchor=north west][inner sep=0.75pt]  [font=\normalsize]  {$001$}; % Text Node
\draw (801,1846.5) node [anchor=north west][inner sep=0.75pt]  [font=\normalsize]  {$111$}; % Text Node
\draw (682.88,1842.62) node [anchor=north west][inner sep=0.75pt]  [font=\normalsize]  {$110$}; % Text Node
\draw (748,1900.22) node [anchor=north west][inner sep=0.75pt]  [font=\normalsize]  {$101$};
\end{tikzpicture} \caption{\label{fig:sphere}A Hamming sphere (shaded) in $\{0,1\}^{3}$ centered
at $(0,0,0)$ with radius~$1$ }
\end{figure}

A {\em Hamming sphere} centered at $y^{n}\in\{0,1\}^{n}$ with
radius $r\in\{0,1,\ldots,n\}$   takes the form $\mathbb{S}_{r}(y^{n}):=\{x^{n}\in\{0,1\}^{n}:d_{\mathrm{H}}(x^{n},y^{n})=r\}$.
See Fig.~\ref{fig:sphere} for an illustration.
The definition of Hamming spheres differs from that for Hamming balls
in the condition $d_{\mathrm{H}}(x^{n},y^{n})=r$ in which {\em
equality} is mandated. Similarly to the previous subsection, here
we also only consider Hamming spheres centered at either $0^{n}$
or~$1^{n}$, for which we can rewrite them respectively as $\{x^{n}:\sum_{i=1}^{n}x_{i}=r\}$
and $\{x^{n}:\sum_{i=1}^{n}x_{i}=n-r\}$. These Hamming
spheres can be regarded as {\em type classes} with types $(\bar{\lambda},\lambda)$
and $(\lambda,\bar{\lambda})$ respectively in   Hamming space, where
$\lambda:=\frac{r}{n}$ and $\bar{\lambda}:=1-\lambda$. Observe that
  $\mathbb{S}_{r}(0^{n})$ is the same as 
$\mathbb{S}_{n-r}(1^{n})$. %It means that the Hamming sphere at center $\mathbf{0}$ with radius
%$r$ is identical to the Hamming sphere at center $\mathbf{1}$ with
%radius $n-r$. 
Notwithstanding this equivalence, we term a pair of spheres $\mathbb{S}_{r_{1}}(0^{n})$
and $\mathbb{S}_{r_{2}}(0^{n})$ as a pair of {\em concentric}
spheres if $r_{1},r_{2}\le n/2$ or $r_{1},r_{2}\ge n/2$, and as
a pair of {\em anti-concentric} spheres if $r_{1}\le n/2\le r_{2}$
or $r_{2}\le n/2\le r_{1}$.

For the LD regime, we choose $\calA_{n}$ and $\calB_{n}$ to be a
pair of concentric or anti-concentric Hamming spheres, i.e., $\calA_{n}=\mathbb{S}_{r_{n}}(0^{n})$
and $\calB_{n}=\mathbb{S}_{s_{n}}(0^{n})$ with $r_{n}=\lfloor\lambda n\rfloor$
or $\lceil\lambda n\rceil$ and $s_{n}=\lfloor\mu n\rfloor$ or $\lceil\mu n\rceil$,
where $\lambda,\mu\in[0,1]$. By Sanov's theorem~\cite{Dembo} (stated in Theorem~\ref{thm:sanov}), 
\begin{align}
\lim_{n\to\infty}-\frac{1}{n}\log \pi_{X}^{n}(\calA_{n}) & =D((\bar{\lambda},\lambda)\|\pi_{X})\quad\mbox{and}\\
\lim_{n\to\infty}-\frac{1}{n}\log \pi_{Y}^{n}(\calB_{n}) & =D((\bar{\mu},\mu)\|\pi_{Y}).
\end{align}
Since $X$ is uniform on $\{0,1\}$, we can write $D((\bar{\lambda},\lambda)\|\pi_{X})=1-h(\lambda)$.
%,
%where $h:t\in[0,1]\mapsto-t\log_{2}t-(1-t)\log_{2}(1-t)$ denotes
%the binary entropy function. 

For the joint probability, observe that the set $\calA_{n}\times\calB_{n}$
is a union of joint type classes with types $T_{XY}$ satisfying the
condition that its marginals $T_{X}$ and $T_{Y}$ are equal to $(\bar{\lambda},\lambda)$
and $(\bar{\mu},\mu)$ respectively. Hence, by Sanov's theorem,
the joint probability satisfies 
\begin{align}
\lim_{n\to\infty}-\frac{1}{n}\log \pi_{XY}^{n}(\calA_{n}\times\calB_{n})=\rvD((\bar{\lambda},\lambda),(\bar{\mu},\mu)\|\pi_{XY}),
\end{align}
%\clearpage
\noindent where, in analogy to Definition~\ref{def:max_cross_ent}, the {\em
minimal relative entropy} with respect to $\pi_{XY}$ over all couplings
of $Q_{X}$ and $Q_{Y}$ is defined as 
\begin{align}
\rvD(Q_{X},Q_{Y}\|\pi_{XY}):=\min_{Q_{XY}\in\mathcal{C}(Q_{X},Q_{Y})}D(Q_{XY}\|\pi_{XY}).\label{eqn:min_rel_ent}
\end{align}
%with $\mathcal{C}(Q_{X},Q_{Y})$
%denoting the coupling set. 
%Here we call $\rvD(Q_{X},Q_{Y}\|\pi_{XY})$
%as the minimum-relative-entropy over couplings of $(Q_{X},Q_{Y})$. 

Optimizing the exponent $\rvD((\bar{\lambda},\lambda),(\bar{\mu},\mu)\|\pi_{XY})$
over all feasible pairs of $(\lambda,\mu)$, yields the following
achievability result.

\begin{proposition}\label{prop:LDsphere} For all $\alpha,\beta\in(0,1)$,
\begin{align}
\hspace{-.2in} \underline{\Upsilon}_{\mathrm{LD}}^{(\infty)}(\alpha,\beta) & \le\underline{\Upsilon}_{\mathrm{LD}}(\alpha,\beta)\\
&:=\min_{\substack{Q_X, Q_Y: \\ D(Q_{X}\|\pi_{X})\ge\alpha, \, D(Q_{Y}\|\pi_{Y})\ge\beta}} \rvD(Q_{X},Q_{Y}\|\pi_{XY}),\label{eqn:LDsphere1}
\end{align}
and
\begin{align}
\hspace{-.2in}\overline{\Upsilon}_{\mathrm{LD}}^{(\infty)}(\alpha,\beta) & \ge\overline{\Upsilon}_{\mathrm{LD}}(\alpha,\beta) \\
&:=\min_{\substack{Q_X, Q_Y: \\ D(Q_{X}\|\pi_{X})\le\alpha, \, D(Q_{Y}\|\pi_{Y})\le\beta}}\rvD(Q_{X},Q_{Y}\|\pi_{XY}).\label{eqn:LDsphere2}
\end{align}
%where the minimum in \eqref{eqn:LDsphere1} is over the set of  $(Q_{X},Q_{Y}) \in\calP(\calX)\times\calP(\calY)$ satisfying 
%\begin{equation}
%D(Q_{X}\|\pi_{X})\ge\alpha\quad\mbox{and}\quad D(Q_{Y}\|\pi_{Y})\ge\beta, \label{eqn:ld_ineqs}
%\end{equation}
%and the maximum in \eqref{eqn:LDsphere2} is over  $(Q_{X},Q_{Y})$ but with the directions of the inequalities  in~\eqref{eqn:ld_ineqs} reversed.
%\begin{equation}
%D(Q_{X}\|\pi_{X})\le\alpha\quad\mbox{and}\quad D(Q_{Y}\|\pi_{Y})\le\beta.
%\end{equation}
\end{proposition}

The bounds in \eqref{eqn:LDsphere1} and \eqref{eqn:LDsphere2} are
 attained by sequences of concentric and anti-concentric
Hamming spheres respectively. By the method of types, it is easy to observe that
they also can be respectively attained by sequences of concentric
and anti-concentric {\em balls} 
%(since the type class corresponding
%to the sphere dominates the type classes that characterize the ball).
(since a  Hamming  ball consists of several spheres and there is one sphere that dominates the others in the sense of the exponent).
The above inequalities were conjectured to be tight by \citet{ordentlich2020note}.
We refer to this as the {\em OPS conjecture} in the sequel.

\begin{conjecture}[OPS Conjecture] \label{conj:ordentlich2020note}
%[Ordentlich--Polyanskiy--Shayevitz's (OPS's) Conjecture] \cite{ordentlich2020note}
For the DSBS and $\alpha,\beta\in(0,1)$, 
\begin{align}
\underline{\Upsilon}_{\mathrm{LD}}^{(\infty)}(\alpha,\beta)\stackrel{?}{=}\underline{\Upsilon}_{\mathrm{LD}}(\alpha,\beta)\quad\mbox{and}\quad\overline{\Upsilon}_{\mathrm{LD}}^{(\infty)}(\alpha,\beta)\stackrel{?}{=}\overline{\Upsilon}_{\mathrm{LD}}(\alpha,\beta).
\end{align}
\end{conjecture}

In Section~\ref{sec:ldr}, we   discuss the optimality of Hamming spheres
in the LD regime, leading to the proof this conjecture. However, before
doing this, we first focus on achievability results by Hamming spherical shells in the MD regime.

For the MD regime, we choose the sets in the NICD problem to be two
spherical shells (annuli), with thickness in the order of $\sqrt{n\theta_{n}}$.
Specifically, for a fixed and small $\epsilon>0$, we choose 
\begin{equation}
\!\!\calA_{n}=\bigcup_{r\in n/2+[\lambda,\lambda+\epsilon]\sqrt{n\theta_{n}}}\mathbb{S}_{r}(0^{n})\;\;\;\;\mbox{and}\;\;\;\;\calB_{n} =\bigcup_{s\in n/2+[\mu,\mu+\epsilon]\sqrt{n\theta_{n}}}\mathbb{S}_{s}(0^{n}),
\end{equation}
%$r_{n}=\lfloor \frac{n}{2}+\lambda\sqrt{n\theta_{n}}\rfloor$ or $\lceil \frac{n}{2}+\lambda\sqrt{n\theta_{n}}\rceil$, and  $s_{n}=\lfloor \frac{n}{2}+\mu\sqrt{n\theta_{n}}\rfloor$ or $\lceil \frac{n}{2}+\mu\sqrt{n\theta_{n}}\rceil$,
where $\{\theta_{n}\}$ is an MD sequence, and $\lambda,\mu\in\mathbb{R}$.
%a positive sequence that satisfies
%$\theta_{n}\to\infty$ and ${\theta_{n}}/{n}\to0$, and $\lambda,\mu\in\mathbb{R}$.
In other words, we choose $\calA_{n}$ and $\calB_{n}$ to be unions
of type classes induced by types $Q_{X}=\pi_{X}+\sqrt{{\theta_{n}}/{n}}\, \eta_{X}$
and $Q_{Y}=\pi_{Y}+\sqrt{{\theta_{n}}/{n}}\, \eta_{Y}$ respectively,
where $\eta_{X}$ and $\eta_{Y}$ are  functions such that $\sum_{x\in\{0,1\}}\eta_{X}(x)=0$
and $\sum_{y\in\{0,1\}}\eta_{Y}(y)=0$ and $\eta_{X}(1)\in[\lambda,\lambda+\epsilon]$
and $\eta_{Y}(1)\in[\mu,\mu+\epsilon]$. %converge to  $\eta_{X}=(-\lambda,\lambda) $ and $\eta_{Y}=(-\mu,\mu)$ respectively. 
Let 
\begin{equation}
\hat{\chi}^{2}(\eta\|\pi):=\sum_{x\in\{0,1\}}\frac{\eta(x)^{2}}{\pi(x)}.%\pi(x)\bigg(\frac{\eta(x)}{\pi(x)}\bigg)^{2},
\end{equation}
and notice that $\hat{\chi}^2(Q-\pi \|\pi)$ is the chi-squared divergence from $Q$ to~$\pi$.
%for a signed measure $\eta$ and a probability measure $P$. 
In analogy to the minimal relative entropy in \eqref{eqn:min_rel_ent},
we define %  the {\em minimal chi-squared divergence} as
\begin{equation}
\hat{\rvX}^{2}(\eta_{X},\eta_{Y}\|\pi_{XY}):=\inf_{\eta_{XY}\in\overline{\mathcal{C}}(\eta_{X},\eta_{Y})}\hat{\chi}^{2}(\eta_{XY}\|\pi_{XY}),\label{eq:NICDminpsi}
\end{equation}
where $\overline{\mathcal{C}}(\eta_{X},\eta_{Y})$ is the set of all
bivariate functions $\eta_{XY}:\{0,1\}^{2}\to\bbR$ such that their
$X$- and $Y$-marginals are equal to $\eta_{X}$ and $\eta_{Y}$ respectively
and $\sum_{x,y}\eta_{XY}(x,y)=0$. %\begin{align}
%\bar{\mathcal{C}}(\eta_{X},\eta_{Y}):=\{ \textrm{signed measure }\eta_{XY}\textrm{ with marginals }\eta_{X},\eta_{Y}\} 
%\end{align}
%denote the signed coupling set of signed measures $(\eta_{X},\eta_{Y})$.
%The marginals here are defined as $\sum_{y}\eta_{XY}(x,y)$ and $\sum_{x}\eta_{XY}(x,y)$.
\enlargethispage{\baselineskip} Then, letting $\theta_{n}\to\infty$ and then $\epsilon\downarrow0$,
by the moderate deviations theorem~\cite{wu1994large,Dembo},
\begin{align}
\lim_{n\to\infty}-\frac{1}{\theta_{n}}\log \pi_{X}^{n}(\calA_{n}) & =\frac{1}{2}\hat{\chi}^{2}(\eta_{X}\|\pi_{X}),\label{eq:NICDmd}\\
\lim_{n\to\infty}-\frac{1}{\theta_{n}}\log \pi_{Y}^{n}(\calB_{n}) & =\frac{1}{2}\hat{\chi}^{2}(\eta_{Y}\|\pi_{Y}), \quad\mbox{and}\\
\lim_{n\to\infty}-\frac{1}{\theta_{n}}\log \pi_{XY}^{n}(\calA_{n}\times\calB_{n}) & =\frac{1}{2}\hat{\rvX}^{2}(\eta_{X},\eta_{Y}\|\pi_{XY}).\label{eq:NICDmd-1}
\end{align}
In fact,~\eqref{eq:NICDmd-1} requires the continuity of $(\eta_{X},\eta_{Y})\mapsto\hat{\rvX}^{2}(\eta_{X},\eta_{Y}\|\pi_{XY})$;
this follows from the following lemma. 
\begin{lemma} \label{lem:min_chi_square}
For $\eta_{X}=(-\lambda,\lambda)$ and $\eta_{Y}=(-\mu,\mu)$, we have 
\begin{align}
\hat{\rvX}^{2}(\eta_{X},\eta_{Y}\|\pi_{XY}) & =\frac{2(\lambda+\mu)^{2}}{1+\rho}+\frac{2(\lambda-\mu)^{2}}{1-\rho}.\label{eq:NICDX2}
\end{align}
\end{lemma}
\begin{proof}
One can calculate that the optimal $\eta_{XY}$ attaining the
maximum in the definition of $\hat{\rvX}^{2}(\eta_{X},\eta_{Y}\|\pi_{XY})$
is 
\begin{equation}
\eta_{XY}=\begin{bmatrix}p-\lambda-\mu & \mu-p\\
\lambda-p & p
\end{bmatrix},\label{eq:NICDDSBS-1}
\end{equation}
where $p=(\lambda+\mu)/{2}$. Hence, \eqref{eq:NICDX2} follows. 
\end{proof}

% explicit expression of $\hat{\rvX}^{2}(\eta_{X},\eta_{Y}\|\pi_{XY})$
%given in~\eqref{eq:NICDX2} to follow. 
Optimizing the exponent $\frac{1}{2}\hat{\rvX}^{2}(\eta_{X},\eta_{Y}\|\pi_{XY})$
over all feasible $\eta_{X}=(-\lambda,\lambda)$ and $\eta_{Y}=(-\mu,\mu)$
yields the following proposition. 

 \begin{proposition} \label{prop:MDsphere} For $\alpha,\beta>0$,
\begin{align}
\underline{\Upsilon}_{\mathrm{MD}}^{(\infty)}(\alpha,\beta)\le\underline{\Upsilon}_{\mathrm{MD}}(\alpha,\beta) & :=\inf\hat{\rvX}^{2}(\eta_{X},\eta_{Y}\|\pi_{XY})\quad\mbox{and}\label{eq:NICDmdexponent}\\
\overline{\Upsilon}_{\mathrm{MD}}^{(\infty)}(\alpha,\beta)\ge\overline{\Upsilon}_{\mathrm{MD}}(\alpha,\beta) & :=\sup\hat{\rvX}^{2}(\eta_{X},\eta_{Y}\|\pi_{XY}).\label{eq:NICDmdexponent-1}
\end{align}
where the $\inf$ in \eqref{eq:NICDmdexponent} is over the set of
functions $\eta_{X},\eta_{Y}:\{0,1\}\to\bbR$ such that $\sum_{x}\eta_{X}(x)=\sum_{y}\eta_{Y}(y)=0$
and 
\begin{align}
\hat{\chi}^{2}(\eta_{X}\|\pi_{X})\ge\alpha\quad\mbox{and}\quad\hat{\chi}^{2}(\eta_{Y}\|\pi_{Y})\ge\beta,\label{eq:NICDx2const}
\end{align}
and the $\sup$ in \eqref{eq:NICDmdexponent-1} is over the same set
of functions $(\eta_{X},\eta_{Y})$ but with the directions of the  inequalities
in \eqref{eq:NICDx2const} reversed. \end{proposition} The bounds
in \eqref{eq:NICDmdexponent} and \eqref{eq:NICDmdexponent-1} are
respectively attained by sequences of concentric and anti-concentric
Hamming spheres or balls. The reader may have noticed that
the constant $1/2$ in \eqref{eq:NICDmd}--\eqref{eq:NICDmd-1} has
been removed in \eqref{eq:NICDmdexponent} and \eqref{eq:NICDmdexponent-1}.
This is because, by definition, $\underline{\Upsilon}_{\mathrm{MD}}$
and $\overline{\Upsilon}_{\mathrm{MD}}$ are {\em homogeneous} (of
degree $1$), i.e., for any $\gamma>0$, 
\begin{align}
\underline{\Upsilon}_{\mathrm{MD}}(\gamma\alpha,\gamma\beta) & =\gamma\underline{\Upsilon}_{\mathrm{MD}}(\alpha,\beta)\quad\mbox{and}\label{eq:NICD-6}\\
\overline{\Upsilon}_{\mathrm{MD}}(\gamma\alpha,\gamma\beta) & =\gamma\overline{\Upsilon}_{\mathrm{MD}}(\alpha,\beta).\label{eq:NICD-12}
\end{align}
%Furthermore, $\underline{\Upsilon}_{\mathrm{MD}}(\alpha,\beta)$ and
%$\overline{\Upsilon}_{\mathrm{MD}}(\alpha,\beta)$ are finite for any
%$\alpha,\beta>0$. Hence, without loss of optimality, it suffices
%to consider bounded~$\eta_{XY}$. % that satisfy $|\eta_{XY}(x,y)|\le M_{\alpha,\beta}$ for all $(x,y)$
%%for some large enough $M_{\alpha,\beta}$, which further implies that $|\eta_{X}(x)|\le {M}_{\alpha,\beta}^X$ for all $x$, $|\eta_{Y}(y)|\le {M}_{\alpha,\beta}^Y$ for all $y$
%for some large enough ${M}_{\alpha,\beta}^X$ and ${M}_{\alpha,\beta}^Y$.
%Furthermore, due the finite alphabet assumption, the infimum and supremum in~\eqref{eq:NICDmdexponent}
%and~\eqref{eq:NICDmdexponent-1} respectively, and also the corresponding
%infimum in~\eqref{eq:NICDminpsi}, are attained. %, and hence
%they are in fact minimum and maximum. 

The bounds in \eqref{eq:NICDmdexponent} and \eqref{eq:NICDmdexponent-1}
can be further simplified as follows.

\begin{lemma} For $\alpha,\beta>0$, 
\begin{align}
\hspace{-.25in}\underline{\Upsilon}_{\mathrm{MD}}(\alpha,\beta) & =\begin{cases}
{\displaystyle \frac{\alpha+\beta-2\rho\sqrt{\alpha\beta}}{1-\rho^{2}}} & \rho^{2}\alpha\le\beta\le\frac{\alpha}{\rho^{2}}\\
\alpha & \beta<\rho^{2}\alpha\\
\beta & \alpha<\rho^{2}\beta
\end{cases}\quad\mbox{and}\label{eq:NICD-MD}\\
\hspace{-.25in}\overline{\Upsilon}_{\mathrm{MD}}(\alpha,\beta) & =\frac{\alpha+\beta+2\rho\sqrt{\alpha\beta}}{1-\rho^{2}}.\label{eq:NICD-MD2}
\end{align}
\end{lemma}

\begin{proof} Observe by the uniformity of $\pi_{X}$ and $\pi_{Y}$
that $\hat{\chi}^{2}(\eta_{X}\|\pi_{X})=4\lambda^{2}$ and $\hat{\chi}^{2}(\eta_{Y}\|\pi_{Y})=4\mu^{2}$. Combining these with Lemma \ref{lem:min_chi_square} yields that  
\begin{align}
\underline{\Upsilon}_{\mathrm{MD}}(\alpha,\beta) & =\min_{\lambda,\mu:4\lambda^{2}\ge\alpha,4\mu^{2}\ge\beta}\frac{2(\lambda+\mu)^{2}}{1+\rho}+\frac{2(\lambda-\mu)^{2}}{1-\rho}\quad\mbox{and}\\
\overline{\Upsilon}_{\mathrm{MD}}(\alpha,\beta) & =\max_{\lambda,\mu:4\lambda^{2}\le\alpha,4\mu^{2}\le\beta}\frac{2(\lambda+\mu)^{2}}{1+\rho}+\frac{2(\lambda-\mu)^{2}}{1-\rho}.
\end{align}
By the rearrangement inequality and by symmetry, it suffices to consider
$\lambda,\mu\ge0$ for $\underline{\Upsilon}_{\mathrm{MD}}(\alpha,\beta)$
and $\lambda\le0\le\mu$ for $\overline{\Upsilon}_{\mathrm{MD}}(\alpha,\beta)$.
This results in 
\begin{align}
\underline{\Upsilon}_{\mathrm{MD}}(\alpha,\beta) & =\min_{\lambda\ge\frac{\sqrt{\alpha}}{2},\mu\ge\frac{\sqrt{\beta}}{2}}\frac{2(\lambda+\mu)^{2}}{1+\rho}+\frac{2(\lambda-\mu)^{2}}{1-\rho}\quad\mbox{and}\label{eqn:calculus1}\\
\overline{\Upsilon}_{\mathrm{MD}}(\alpha,\beta) & =\max_{-\frac{\sqrt{\alpha}}{2}\le\lambda\le0\le\mu\le\frac{\sqrt{\beta}}{2}}\frac{2(\lambda+\mu)^{2}}{1+\rho}+\frac{2(\lambda-\mu)^{2}}{1-\rho}.\label{eqn:calculus2}
\end{align}
By calculus, one can verify that the right-hand sides of~\eqref{eqn:calculus1}
and~\eqref{eqn:calculus2} are respectively equal to the right-hand sides
of~\eqref{eq:NICD-MD} and~\eqref{eq:NICD-MD2}. \end{proof}

We conclude this section by discussing the relationships between the MD and CL 
exponents as well as the MD and LD exponents. 
%As is well known from the connections
%between the central limit, large and moderate deviations theorems~\cite{Dembo},
We can recover the MD exponents from the CL or LD exponents  if the  MD sequence $\{\theta_{n}\}$ additionally satisfies $(\log n)/\theta_{n}\to0$
as $n\to\infty$. Roughly speaking, in the MD regime, we chose the
radii $r_{n}$ and $s_{n}$ of Hamming spheres such that $\frac{r_{n}}{n}\approx\frac{1}{2}+\lambda\sqrt{\epsilon}$
and $\frac{s_{n}}{n}\approx\frac{1}{2}+\mu\sqrt{\epsilon}$, where
$\epsilon:=\frac{\theta_{n}}{n}\to0$ as $n\to\infty$. This implies
that the types corresponding to the spheres are $Q_{X}\approx \pi_{X}+\sqrt{\epsilon} \, \eta_{X}$
and $Q_{Y}\approx \pi_{Y}+\sqrt{\epsilon} \, \eta_{Y}$ as $\epsilon\downarrow0$.
%Define the $\hat{\chi}^{2}_{+}$-divergence from $Q$ to $P$ as $\hat{\chi}^{2}_{+}(Q\|P):=\sum_{x}P(x)(\frac{Q(x)}{P(x)}-1)^{2}$.  
 Note that in Sanov's theorem, the LD exponent of the probability
of a Hamming sphere with type $Q_{X}$ is $D(Q_{X}\|\pi_{X})+O\big(\frac{\log n}{n}\big)$.
Hence, if the MD sequence  $\{\theta_n\}$ additionally satisfies  $(\log n)/\theta_{n}\to0$ as $n\to\infty$,
this LD exponent is dominated by the term $D(Q_{X}\|\pi_{X})$, which
allows us to omit the $O\big(\frac{\log n}{n}\big)$ term. Moreover, by Taylor's theorem,
\begin{align}
D(Q_{X}\|\pi_{X}) & =\frac{\epsilon}{2}\hat{\chi}^{2}(\eta_{X}\|\pi_{X})+o(\epsilon),\\
D(Q_{Y}\|\pi_{Y}) & =\frac{\epsilon}{2}\hat{\chi}^{2}(\eta_{Y}\|\pi_{Y})+o(\epsilon),
\end{align}
and similarly, 
\begin{align}
\rvD(Q_{X},Q_{Y}\|\pi_{XY}) & =\frac{\epsilon}{2}\hat{\rvX}^{2}(\eta_{X},\eta_{Y}\|\pi_{XY})+o(\epsilon)\quad\mbox{as }\epsilon\downarrow0.
\end{align}
We obtain the MD exponents by replacing $D$ and $\rvD$ in the LD
exponents with $\frac{\epsilon}{2}\hat{\chi}^{2}$ and $\frac{\epsilon}{2}\hat{\rvX}^{2}$
respectively. Formally, 
\begin{align}
\lim_{\epsilon\downarrow0}\frac{1}{\epsilon}\underline{\Upsilon}_{\mathrm{LD}}(\epsilon\alpha,\epsilon\beta) & =\underline{\Upsilon}_{\mathrm{MD}}(\alpha,\beta)\quad\mbox{and}\label{eq:NICD-13a}\\
\lim_{\epsilon\downarrow0}\frac{1}{\epsilon}\overline{\Upsilon}_{\mathrm{LD}}(\epsilon\alpha,\epsilon\beta) & =\overline{\Upsilon}_{\mathrm{MD}}(\alpha,\beta).\label{eq:NICD-13}
\end{align}
Furthermore, the MD exponents can be also recovered from the CL exponents.
By the Berry--Esseen theorem \cite{Berry41,Esseen42}, under  the  condition that the MD sequence $\{\theta_n\}$ satisfies $(\log n)/\theta_{n}\to0$
as $n\to\infty$, the probability of a Hamming ball is dominated by
the term involving the Gaussian cumulative distribution function $\Phi(\cdot)$. In other words, the additive error term in the Berry--Esseen theorem, which scales as
$O(\frac{1}{\sqrt{n}})$, is negligible asymptotically. 
On the other hand, \citet[Ex.~9.24 and~10.5]{ODonnell14analysisof}
shows that  %by taking $\theta$ to be vanishingly small.  That is, 
\begin{align}
\lim_{\theta\to\infty}\frac{1}{\theta}\underline{\Upsilon}_{\mathrm{CL}}(\theta\alpha,\theta\beta) & =\underline{\Upsilon}_{\mathrm{MD}}(\alpha,\beta)\quad\mbox{and}\label{eq:NICD-14a}\\
\lim_{\theta\to\infty}\frac{1}{\theta}\overline{\Upsilon}_{\mathrm{CL}}(\theta\alpha,\theta\beta) & =\overline{\Upsilon}_{\mathrm{MD}}(\alpha,\beta),\label{eq:NICD-14}
\end{align}
where $\underline{\Upsilon}_{\mathrm{CL}}$ and $\overline{\Upsilon}_{\mathrm{CL}}$
are defined in \eqref{eqn:Theta_under} and \eqref{eqn:Theta_over}
respectively. %Recall that here $\underline{\Upsilon}_{\mathrm{CL}}(\alpha,\beta)=-\log\Lambda_{\rho}(2^{-\alpha},2^{-\beta})$
%with 
%\begin{equation}
%\Lambda_{\rho}(a,b)=\Phi_{\rho}(\Phi^{-1}(a),\Phi^{-1}(b)),
%\end{equation}
%and $\overline{\Upsilon}_{\mathrm{CL}}(\alpha,\beta)$ is
%defined similarly but with $\rho\leftarrow-\rho$.  
%Intuitively, Equation \eqref{eq:NICD-14} can be verified as follows. Observe that 
%for this case, $\Phi(2^{-\theta\alpha})\doteq\varphi(2^{-\theta\alpha})$
%and $\Phi_{\rho}(2^{-\theta\gamma_{1}},2^{-\theta\gamma_{2}})\approx\varphi_{\rho}(2^{-\theta\gamma_{1}},2^{-\theta\gamma_{2}})$
%as $\theta\to\infty$ for given $\alpha,\gamma_{1},\gamma_{2}>0$. 

\begin{figure}
\centering %
\begin{tabular}{cc}
\hspace{-0.2in}\includegraphics[width=0.51\columnwidth]{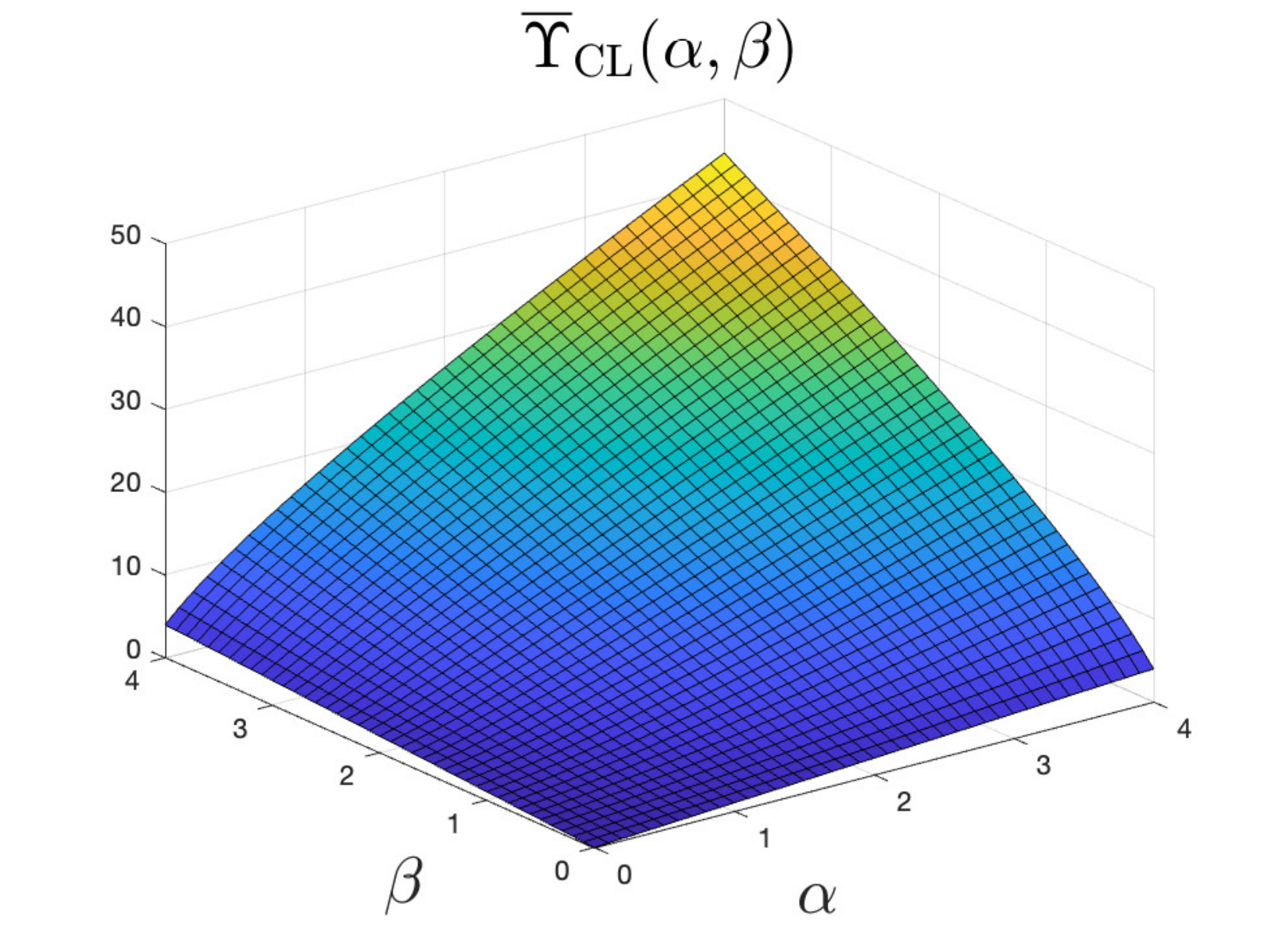}  & \hspace{-0.32in} \includegraphics[width=0.51\columnwidth]{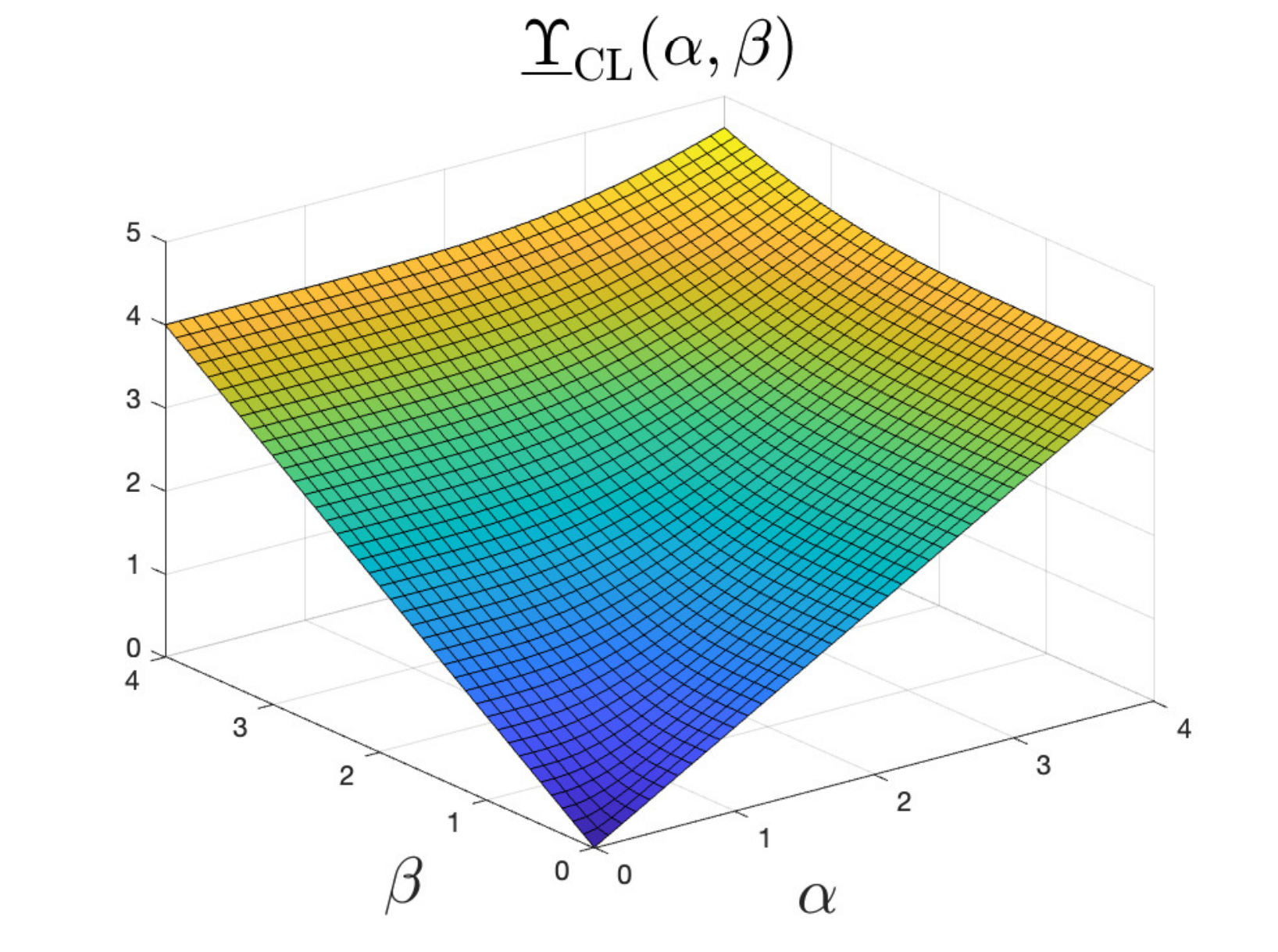} \tabularnewline
\hspace{-0.2in}\includegraphics[width=0.51\columnwidth]{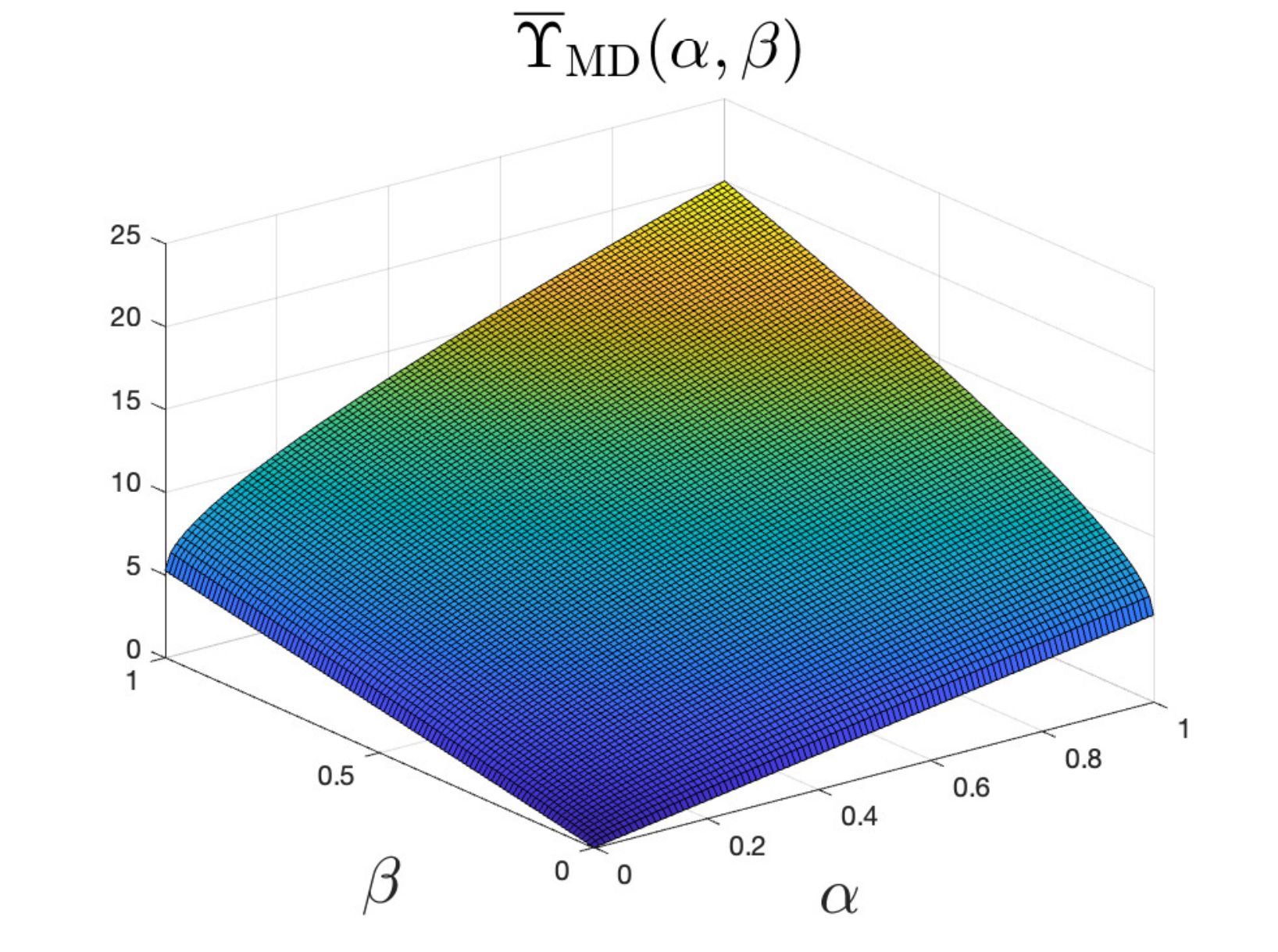}  & \hspace{-0.32in} \includegraphics[width=0.51\columnwidth]{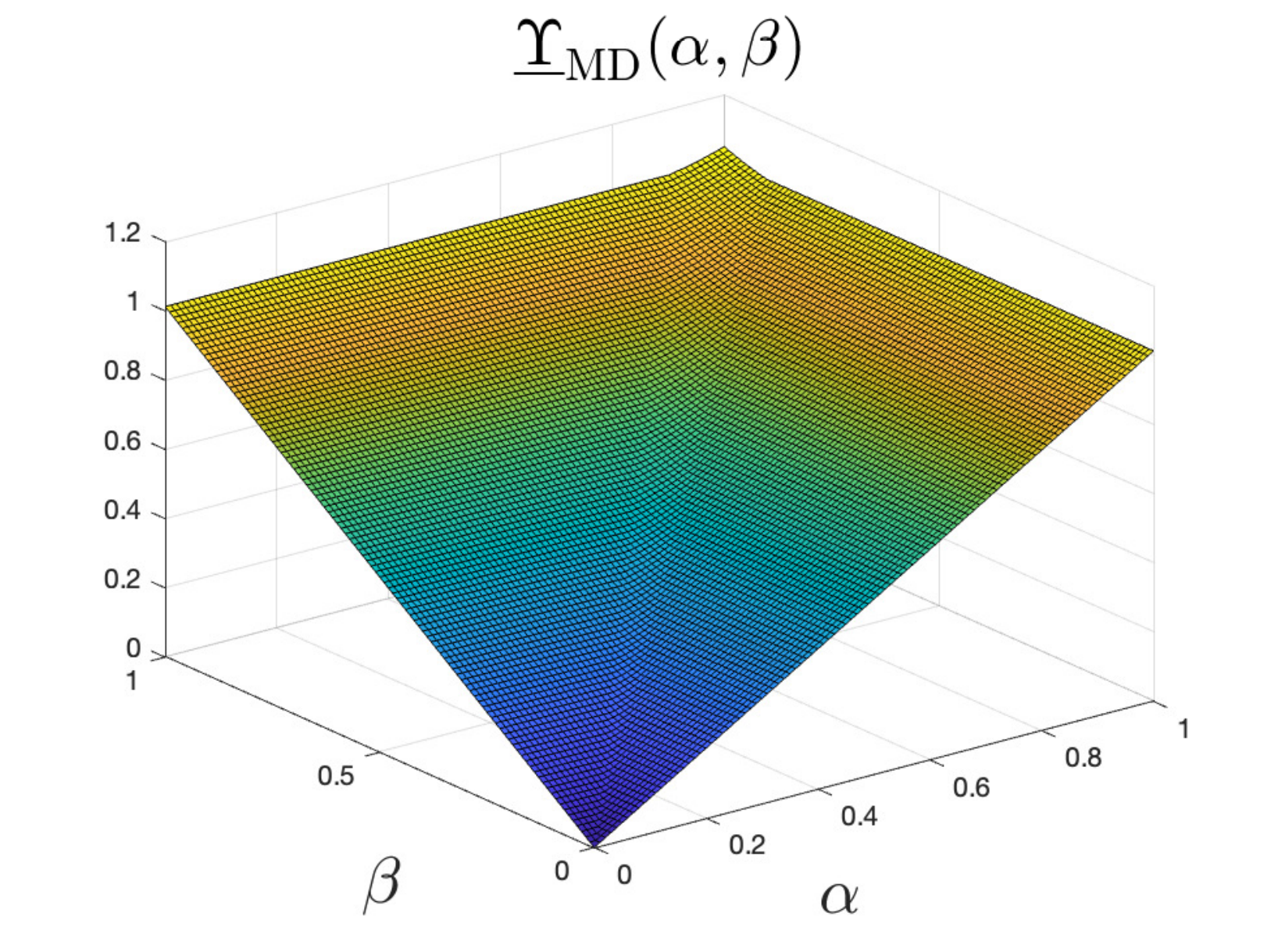} \tabularnewline
\hspace{-0.2in}\includegraphics[width=0.51\columnwidth]{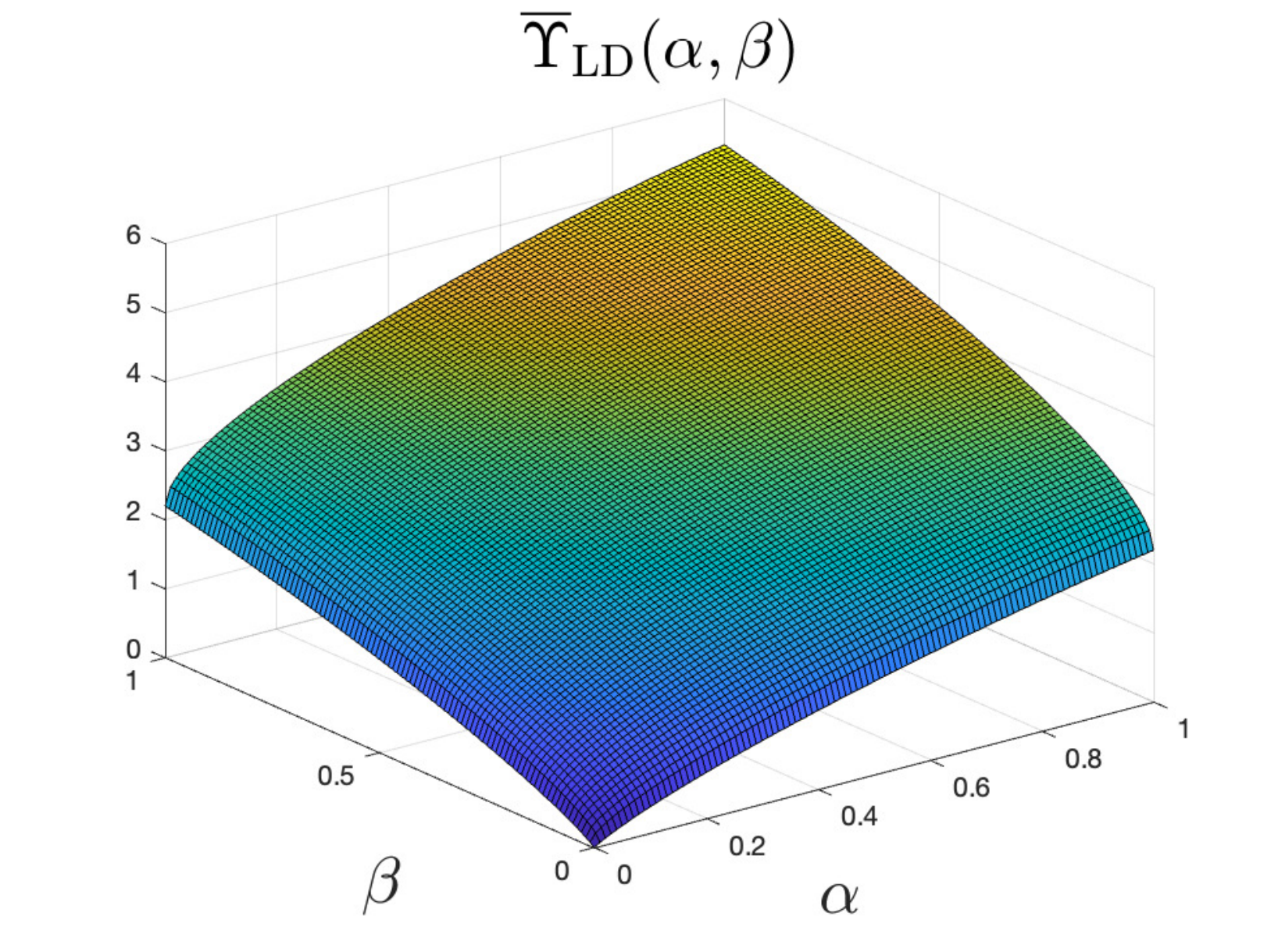}  & \hspace{-0.32in} \includegraphics[width=0.51\columnwidth]{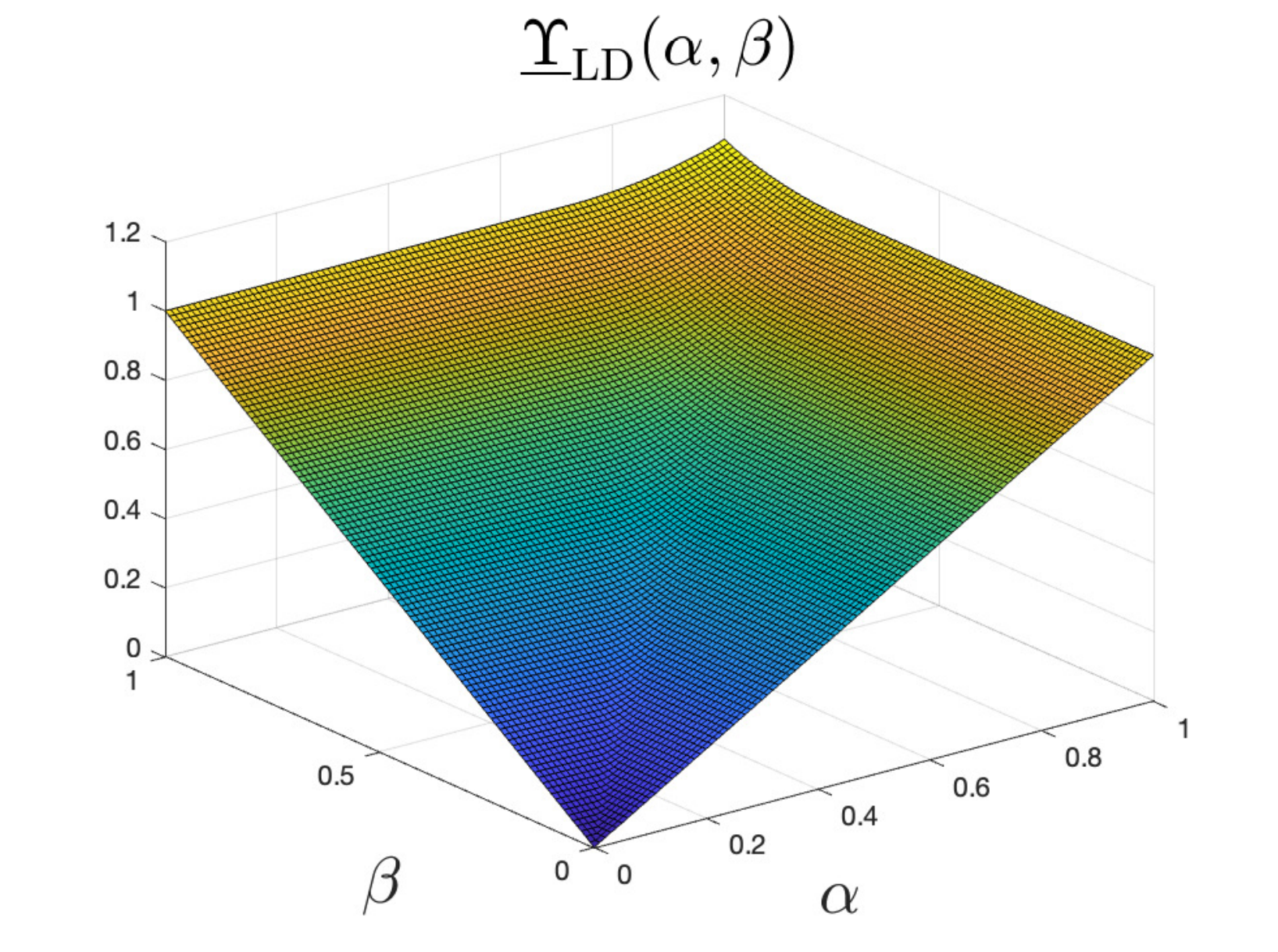}
% \tabularnewline
% \multicolumn{3}{c}{{\footnotesize{}$\underline{\Upsilon}_{\mathrm{CL}}(\alpha,\beta)$}{\small{}
% $\qquad$ $\overset{\textrm{Zoom out}}{\Longrightarrow}$ $\qquad$
% $\qquad$}{\footnotesize{}$\underline{\Upsilon}_{\mathrm{MD}}(\alpha,\beta)$}{\small{}$\qquad$}{\footnotesize{}
% }{\small{} $\stackrel[\textrm{a neighborhood of the origin}]{\textrm{Zoom in to }}{\Longleftarrow}$
% $\qquad$}{\footnotesize{}$\underline{\Upsilon}_{\mathrm{LD}}(\alpha,\beta)$}}\tabularnewline
\tabularnewline
\end{tabular}\caption{Forward and reverse CL, MD, and LD exponents induced by Hamming balls
(or spheres) for $\rho=0.9$. Observe that $\underline{\Upsilon}_{\mathrm{MD}}$
and $\underline{\Upsilon}_{\mathrm{LD}}$ appear to be convex while~$\overline{\Upsilon}_{\mathrm{MD}}$
and $\overline{\Upsilon}_{\mathrm{LD}}$ appear to be concave. The convexity
and concavity   of~$\underline{\Upsilon}_{\mathrm{LD}}$ and
$\overline{\Upsilon}_{\mathrm{LD}}$ respectively have implications
for the OPS conjecture (Conjecture~\ref{conj:ordentlich2020note})
whose resolution is provided in Section~\ref{sec:ldr}.}
%Theorem~\ref{thm:ops_resolve}
\label{fig:FExponents} 
\end{figure}

\subsection{Numerical Results and Comparisons}

We now evaluate the various exponents for the DSBS with correlation
coefficient~$\rho$. Define $\kappa:=(\frac{1+\rho}{1-\rho})^{2}$,
\begin{align}
%D(a)
D_{a,b}(p) & :=D\left(\begin{bmatrix}1+p-a-b & b-p\\
a-p & p
\end{bmatrix}~\middle\|~\pi_{XY}\right)\quad\mbox{and}\\
\rvD(a,b) & :=\min_{\max\{0,a+b-1\}\le p\le\min\{a,b\}}D_{a,b}(p)=D_{a,b}(p_{a,b}^{*}),
\end{align}
where $h(\cdot)$ is the binary entropy function, and 
\begin{equation}
p_{a,b}^{*}:=\frac{(\kappa - 1)(a+b)+1-\sqrt{((\kappa-1)(a  +b)+1)^{2}- 4\kappa(\kappa- 1)ab}}{2(\kappa-1)}.
\end{equation}
For the DSBS, $\underline{\Upsilon}_{\mathrm{LD}}$ and $\overline{\Upsilon}_{\mathrm{LD}}$,
defined in~\eqref{eqn:LDsphere1} and \eqref{eqn:LDsphere2}, respectively
can be written in closed form as 
\begin{align}
\underline{\Upsilon}_{\mathrm{LD}}(\alpha,\beta) & =\rvD \big(h^{-1}(1-\alpha),h^{-1}(1-\beta)\big)\quad\mbox{and}\\
\overline{\Upsilon}_{\mathrm{LD}}(\alpha,\beta) & =\rvD\big(h^{-1}(1-\alpha),1-h^{-1}(1-\beta)\big),
\end{align}
where $h^{-1}:[0,1]\to[0,1/2]$ is the inverse of the binary entropy
function~$h$ when its domain is restricted to $[0,1/2]$.

%\begin{figure}[t]
%\centering %
%\begin{tabular}{rcl}
%\includegraphics[width=0.3\columnwidth]{figs/CL_upper} & \includegraphics[width=0.3\columnwidth]{figs/MD_upper} & \includegraphics[width=0.3\columnwidth]{figs/LD_upper}
%% \tabularnewline
%% \multicolumn{3}{c}{{\footnotesize{}$\overline{\Upsilon}_{\mathrm{CL}}(\alpha,\beta)$}{\small{}
%% $\qquad$ $\overset{\textrm{Zoom out}}{\Longrightarrow}$ $\qquad$
%% $\qquad$}{\footnotesize{}$\overline{\Upsilon}_{\mathrm{MD}}(\alpha,\beta)$}{\small{}$\qquad$
%% $\stackrel[\textrm{a neighborhood of the origin}]{\textrm{Zoom in to }}{\Longleftarrow}$
%% $\qquad$}{\footnotesize{}$\overline{\Upsilon}_{\mathrm{LD}}(\alpha,\beta)$}}\tabularnewline
%\end{tabular}
%\caption{\label{fig:RExponents}Reverse CL, MD, and LD exponents induced by
%Hamming balls or spheres for $\rho=0.9$}
%\end{figure}
We plot the CL exponents achieved by Hamming balls, and the MD and
LD exponents achieved by Hamming balls, spheres,  or spherical shells in Fig.~\ref{fig:FExponents}.
By the homogeneity property in~\eqref{eq:NICD-6} and~\eqref{eq:NICD-12},
the surfaces corresponding to $\underline{\Upsilon}_{\mathrm{MD}}$
and $\overline{\Upsilon}_{\mathrm{MD}}$ are formed by an infinite number
of half-lines from the origin to   infinity. Furthermore, by the relation
between the MD, LD and CL exponents in~\eqref{eq:NICD-13a}--\eqref{eq:NICD-13}
and~\eqref{eq:NICD-14a}--\eqref{eq:NICD-14}, the surfaces of
$\underline{\Upsilon}_{\mathrm{MD}}$ and $\overline{\Upsilon}_{\mathrm{MD}}$
can be recovered from the surfaces of $\underline{\Upsilon}_{\mathrm{CL}}$
and $\overline{\Upsilon}_{\mathrm{CL}}$ by zooming them out, or recovered
from $\underline{\Upsilon}_{\mathrm{LD}}$ and $\overline{\Upsilon}_{\mathrm{LD}}$
by zooming into a neighborhood of the origin. However, the surfaces
of $\underline{\Upsilon}_{\mathrm{CL}}$ and $\overline{\Upsilon}_{\mathrm{CL}}$
as well as the surfaces of $\underline{\Upsilon}_{\mathrm{LD}}$ and
$\overline{\Upsilon}_{\mathrm{LD}}$ {\em cannot} be recovered from
those of $\underline{\Upsilon}_{\mathrm{MD}}$ and $\overline{\Upsilon}_{\mathrm{MD}}$.
In other words, $\underline{\Upsilon}_{\mathrm{MD}}$ and $\overline{\Upsilon}_{\mathrm{MD}}$
contain much less information compared to $\underline{\Upsilon}_{\mathrm{CL}}$
and $\overline{\Upsilon}_{\mathrm{CL}}$ as well as $\underline{\Upsilon}_{\mathrm{LD}}$
and $\overline{\Upsilon}_{\mathrm{LD}}$. This is not unexpected as
the MD regime can be thought of a limiting case of the LD and CL regimes.
Numerical results in Fig.~\ref{fig:FExponents} suggest that~$\underline{\Upsilon}_{\mathrm{MD}}$
and $\underline{\Upsilon}_{\mathrm{LD}}$ are convex, and $\overline{\Upsilon}_{\mathrm{MD}}$
and $\overline{\Upsilon}_{\mathrm{LD}}$ are concave, but $\underline{\Upsilon}_{\mathrm{CL}}$
and~$\overline{\Upsilon}_{\mathrm{CL}}$ are neither convex nor concave.
In  Section~\ref{sec:ldr}, we discuss
these issues rigorously in the context of the OPS conjecture (Conjecture~\ref{conj:ordentlich2020note}).

\begin{table}[t] 
\caption{Comparison of subcubes and Hamming balls or, equivalently, spheres}
\label{tab:Comparison-of-subcubes} 
 \centering{\small
% \begin{tabular}{|c|>{\centering}p{2cm}|c|c|c|}
\begin{centering}
\begin{tabular}{|>{\centering}p{2.0cm}|>{\centering}p{1.5cm}|>{\centering}p{1.5cm}|>{\centering}p{1.9cm}|>{\centering}p{1.9cm}|}
\hline 
Regimes  & \multicolumn{2}{c|}{Central limit} & Moderate deviations  & Large deviations\tabularnewline
\hline 
$a,b$  & Fixed and large  & Fixed and small  & Subexp.\ vanishing  & Exp.\ vanishing \tabularnewline
\hline 
\hline 
Subcubes  & Better  & Worse  & Worse  & Worse \tabularnewline
\hline 
Balls/Spheres  & Worse  & Better  & Better  & Better\tabularnewline
\hline 
\end{tabular}
\end{centering}}
\end{table}

%\begin{center}
%\begin{figure}
%\centering \includegraphics[width=0.7\columnwidth]{figs/SubcubeVsBall2}
%\caption{\label{fig:comparison}The forward joint probabilities achieved by
%subcubes and Hamming balls with $a=b$ and $\rho=0.9$}
%\end{figure}
%\par\end{center}
\begin{figure}[t]
\centering 
\begin{overpic}[width=0.95\columnwidth]{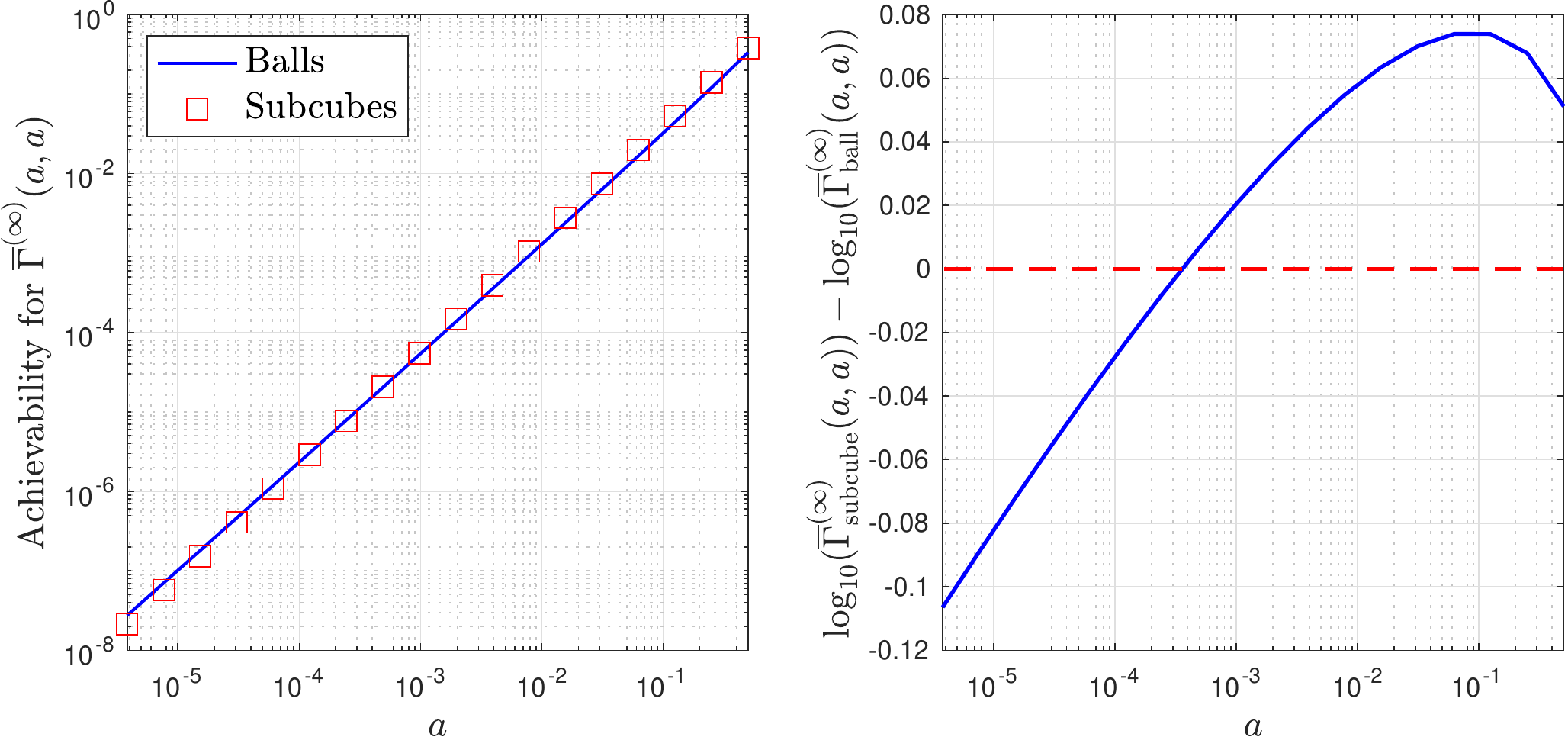} 
{\footnotesize
\put(86,38){Subcubes}
\put(86,35){better}
\put(63,10){Balls better}
}
\end{overpic}
\caption{\label{fig:comparison}Left: The forward joint probabilities achieved
by subcubes and Hamming balls with $a=b$ and $\rho=0.5$; Right:
The difference between the logarithms of the forward joint probabilities
achieved by subcubes and Hamming balls which shows that subcubes outperform
balls for large $a$ and vice versa.}
\end{figure}

\enlargethispage{\baselineskip}
We now compare the performances of subcubes, Hamming balls, and Hamming
spheres (or spherical shells). We illustrate the forward joint probabilities achieved by
subcubes and Hamming balls in Fig.~\ref{fig:comparison}. As the
gaps between the probabilities are visually imperceptible, we also illustrate their differences on the right plot of Fig.~\ref{fig:comparison}. Based on the numerical
comparisons, we observe that for large $a$ and~$b$, subcubes are
better. However, for small~$a$ and~$b$, Hamming balls are better.
We summarize the performances of various geometric structures under
different asymptotic regimes in Table~\ref{tab:Comparison-of-subcubes}.
Based on these results, it is natural to ask whether subcubes are
{\em optimal} for large $a$ and $b$, and whether Hamming balls or spheres
are {\em optimal} for small $a$ and $b$. In the following  sections, we  
provide answers to these questions.

%\section{Converses: Optimality of Subcubes and Spheres?}
\section{Converses in the Central Limit Regime}
\label{sec:nicd_conv}

In this and the next two sections, we discuss the optimality of subcubes, Hamming
balls, and spheres (or spherical shells) in the various asymptotic regimes for the forward
and reverse joint probabilities. In this section, we consider the CL regime in which we
are interested in determining whether subcubes are optimal in for the 
NICD problem for $a=b\in\{{1}/{2},{1}/{4}\}$. The case $a=b=1/2$ is relatively
well known and solved by \citet{witsenhausen1975sequences}. The case
$a=b={1}/{4}$, however, is more challenging and, in fact, was posed
as an open problem by   E.~Mossel in 2017 \cite{mossel2017}; see also \citet[Problem~2.6]{mossel2021probabilistic}. Here, we term the case $a=b={1}/{4}$, 
as the \emph{mean-$1/4$ stability problem}. In the CL regime, it
is also natural to ask whether Hamming balls are optimal for small 
but fixed $a$ and $b$ (i.e., $0<a,b<1/4$). Since this case behaves similarly to that in  the MD regime, we will discuss it in the next section concerning the MD regime. %, instead of this.% section. 

\subsection{Case of $a=b={1}/{2}$:  Maximal Correlation Method} 
\label{sec:clt_half}

We first consider the optimality of subcubes (or Boolean functions)
for the case $a=b={1}/{2}$ in the NICD problem. By using the properties
of the {\em maximal correlation}, the non-asymptotic optimality
of subcubes for this basic case was confirmed positively by \citet{witsenhausen1975sequences}.
We recall from \eqref{eq:mc} in  the introduction that the {\em Hirschfeld--Gebelein--R\'enyi} (or {\em
HGR}) {\em maximal correlation} \cite{hirschfeld1935connection,gebelein1941statistische,renyi1959measures}
between two random variables $X$ and $Y$ is defined as 
\begin{align}
\rho_{\mathrm{m}}(X;Y):=\sup_{f,g}\rho(f(X);g(Y)),
\end{align}
where $\rho(U;V)$ denotes the correlation coefficient between $U$ and~$V$ (defined in~\eqref{eqn:pcc}), and
the supremum is taken over all real-valued functions $f$ and $g$
such that $0<\var(f(X)),\var(g(Y))<\infty$. It is well-known that
the maximal correlation satisfies several desirable properties, including
tensorization and the data processing inequality. 
\begin{enumerate}
\item \underline{Tensorization}: For a sequence of independent pairs of
random variables $(X^{n},Y^{n})=\{(X_{i},Y_{i})\}_{i=1}^{n}$, we
have 
\begin{align}
\rho_{\mathrm{m}}(X^{n};Y^{n})=\max_{ i \in [n]}\rho_{\mathrm{m}}(X_{i};Y_{i}).\label{eqn:tensor_prop}
\end{align}
%where $X^n=(X_1,\ldots, X_n)$ and $Y^n= (Y_1,\ldots, Y_{n})$. 
\item \underline{Data processing inequality (DPI)}: For the Markov chain
$U-X-Y-V$, we have 
\begin{align}
\rho_{\mathrm{m}}(U;V)\leq\rho_{\mathrm{m}}(X;Y).\label{eqn:dpi_prop}
\end{align}
\item \label{enu:Binary-Cases:-For} \underline{Binary random variables}:
For binary $X$ and $Y$, we have 
\begin{equation}
\rho_{\mathrm{m}}(X;Y)=\left|\rho(X;Y)\right|.\label{eq:NICD-16}
\end{equation}
\end{enumerate}
Using these properties, \citet{witsenhausen1975sequences} proved
the following theorem. \begin{theorem} \label{thm:wit} Let $\pi_{XY}$
be the doubly symmetric binary distribution with correlation coefficient
$\rho$ as defined in~\eqref{eq:NICDDSBS}. For any $\calA$ and $\calB$
with $\pi_{X}^{n}(\calA)=a$ and $\pi_{Y}^{n}(\calB)=b$, 
\begin{align}
ab-\rho\sqrt{a\bar{a}b\bar{b}}\le \pi_{XY}^{n}(\calA\times\calB)\le ab+\rho\sqrt{a\bar{a}b\bar{b}}.\label{eqn:max_cor}
\end{align}
\end{theorem}

\begin{proof} Let $(X^{n},Y^{n})\sim \pi_{XY}^{n}$. Define $U:=\bone_{\calA}(X^{n})$
and $V:=\bone_{\calB}(Y^{n})$. Then we have the Markov chain $U-X^{n}-Y^{n}-V$.
Consider, 
\begin{align}
\frac{\left|\pi_{XY}^{n}(\calA\times\calB)-ab\right|}{\sqrt{a\bar{a}}\sqrt{b\bar{b}}} & =|\rho(U;V) |\nn\\*
 & =\rho_{\mathrm{m}}(U;V)\label{eqn:use_bin1}\\
 & \le\rho_{\mathrm{m}}(X^{n};Y^{n})\label{eqn:use_DPI}\\
 & =\rho_{\mathrm{m}}(X_1;Y_1)\label{eqn:use_tensor}\\
 & =\rho,\label{eqn:use_bin2}
\end{align}
where \eqref{eqn:use_bin1} and \eqref{eqn:use_bin2} follow from~\eqref{eq:NICD-16},
\eqref{eqn:use_DPI} follows from the data processing inequality in~\eqref{eqn:dpi_prop}, 
and \eqref{eqn:use_tensor} follows from the tensorization property
in \eqref{eqn:tensor_prop} (since all pairs of random variables are
{\em identically distributed}, the max in \eqref{eqn:tensor_prop}
is simply $\rho_{\rmm}(X_{1};Y_{1})$). \end{proof}
From  Theorem~\ref{thm:wit}, one  deduces
that for $a=b=1/2$,  
\begin{equation}
\frac{1-\rho}{4}\le \pi_{XY}^{n}(\calA\times\calB)\le\frac{1+\rho}{4}. \label{eqn:mc_half}
\end{equation}
Based on the discussion around~\eqref{eqn:subcube_id}--\eqref{eqn:subcube_anti},
the upper bound is achieved by a pair of identical dictator functions,
i.e., $f(x^{n})=g(x^{n})=x_{i}$ (or $1-x_{i}$) for all $i\in[n]$. Moreover,   the
lower bound is achieved by a pair anti-symmetric dictator functions,
i.e., $f(x^{n})=1-g(x^{n})=x_{i}$ for all $i\in[n]$. Hence, 
\begin{equation}
\overline{\Gamma}^{(n)}\Big(\frac{1}{2},\frac{1}{2}\Big)=\frac{1+\rho}{4}\quad\mbox{and}\quad\underline{\Gamma}^{(n)}\Big(\frac{1}{2},\frac{1}{2}\Big)=\frac{1-\rho}{4}\quad \mbox{for all}\;\,  n\ge1.
\end{equation}

%It is worth noting that 
This result also can be proven by the
hypercontractivity method and Fourier analysis; these are discussed in the next two subsections. 

% We now briefly
%introduce the hypercontractivity method. The basics of Boolean Fourier
%analysis will be introduced in the next subsection to address the
%case in which $a=b={1}/{4}$.

\subsection{Case of $a=b={1}/{2}$: Hypercontractivity Method}  \label{sec:ab_half}

The classic {\em hypercontractivity inequalities} form an important
class of functional inequalities. These inequalities
play a fundamental role in the NICD problem when
the means of the Boolean functions are assumed to be either large or  small. The forward and reverse parts of the hypercontractivity
inequalities for the DSBS are stated in Theorem~\ref{thm:hyper2}  which
follow from \citet{gross1975logarithmic}, \citet{borell1982positivity},
and \citet{ODonnell14analysisof}.

We commence with some definitions. For  $f:\mathcal{X}^{n}\to[0,\infty)$ and $g:\mathcal{Y}^{n}\to[0,\infty)$, denote their \emph{inner product} 
\begin{align}
\langle f,g\rangle:=\mathbb{E}[f(X^{n})g(Y^{n})],
\end{align} 
where the expectation is taken with respect to $\pi_{XY}^n$. 
Define the \emph{$L^{p}$-norm} for $p\in [1,\infty)$ and  the \emph{pseudo $L^{p}$-norm} for $p\in (-\infty,1)\backslash \{0\} $ as 
%Denote the  $L^{p}$-norm (or  the pseudo $L^{p}$-norm)  as 
\begin{align}
\Vert f\Vert_{p}:=(\mathbb{E}[f^{p}(X^{n})])^{1/p}. \label{eq:LpNorm}
\end{align} 
For $p\in \{0,\pm \infty\}$, $\Vert f\Vert_{p}$ is defined by its continuous extensions. Specifically, 
\begin{align}
\Vert f\Vert_{0} & := \rme^{\bbE [\ln f(X^n)]},\\
\Vert f\Vert_{\infty } & := \max_{x^n\in \calX^n} f(x^n), \quad \mbox{and}  \\
\Vert f\Vert_{-\infty } &:= \min_{x^n\in \calX^n} f(x^n), %\label{eq:LinftyNorm}
\end{align} 
where $\Vert f\Vert_{0}$ is known as the \emph{geometric mean} of $f$. Note that $\Vert f\Vert_{p} =0$ for $p<0$ if $f$ is not positive $\pi_X$-almost everywhere. % w.r.t. the underlying distribution.   

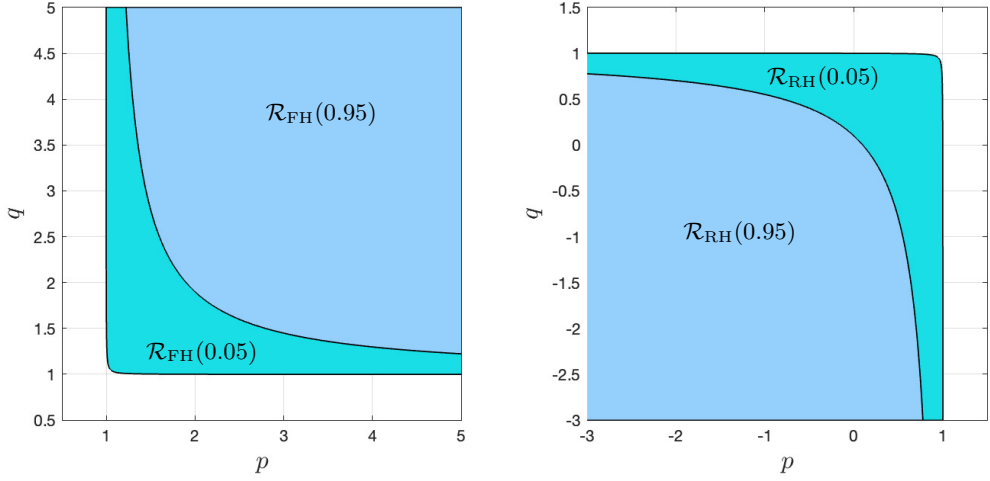
\begin{figure}
\centering
\begin{overpic}[width=1.05\columnwidth]{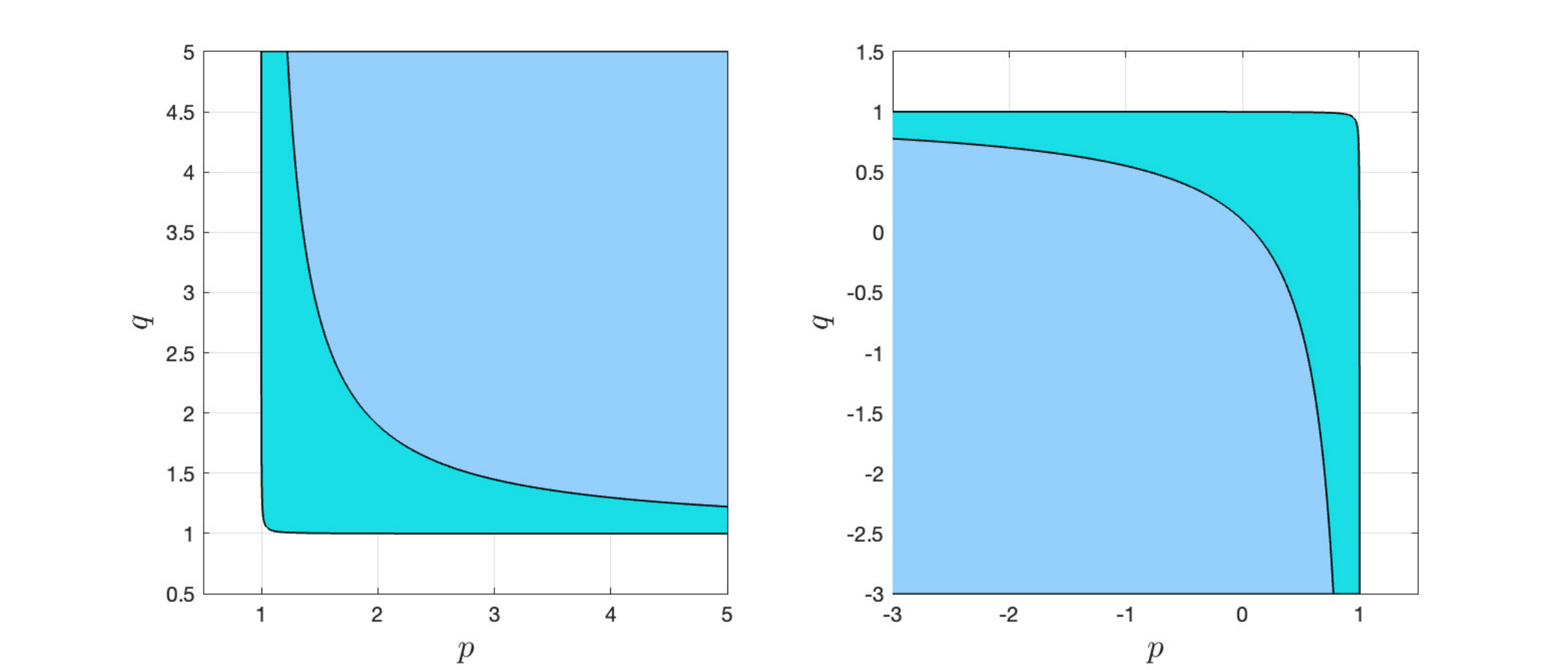}
\put(30,30){{\footnotesize \mbox{$\mathcal{R}_{\mathrm{FH}} (0.95)$}}}
\put(20,10){{\footnotesize \mbox{$\mathcal{R}_{\mathrm{FH}} (0.05)$}}}
\put(65,20){{\footnotesize \mbox{$\mathcal{R}_{\mathrm{RH}} (0.95)$}}}
\put(72,33){{\footnotesize \mbox{$\mathcal{R}_{\mathrm{RH}} (0.05)$}}}
\end{overpic} 
\caption{Plots of  the forward (left) and reverse (right) hypercontractivity regions in \eqref{eq:FIFHR_nicd} and~\eqref{eq:FIRHR_nicd} for $\rho=0.05$ and $0.95$}
\label{fig:hc_regions}
\end{figure}

For the DSBS $(X,Y)\sim \pi_{XY}$ with correlation coefficient $\rho$, define 
\begin{align}
\hspace{-.25in}\mathcal{R}_{\mathrm{FH}} (\rho) & :=\big\{(p,q)\in[1,\infty]^{2}:(p-1)(q-1)\ge\rho^{2}\big\},\;\;\mbox{and}\label{eq:FIFHR_nicd}\\
\hspace{-.25in}\mathcal{R}_{\mathrm{RH}}(\rho)&:= \big\{(p,q)\!\in\![-\infty,1]^{2}:(p-1)(q-1)\ge\rho^{2} \big\}.\label{eq:FIRHR_nicd}
\end{align}
These regions are respectively called the {\em forward} and {\em reverse hypercontractivity regions} for the DSBS and are illustrated in Fig.~\ref{fig:hc_regions}.

\begin{theorem}[Hypercontractivity: DSBS and Two-Function Version]
\label{thm:hyper2} Let  $(X^{n},Y^{n})\sim \pi_{XY}^{n}$ be a source sequence generated by a DSBS  with correlation coefficient $\rho$.
%The following hold: 
\begin{enumerate}
\item The inequality 
\begin{align}
\langle f,g\rangle & \leq\Vert f\Vert_{p}\Vert g\Vert_{q}\label{eq:NICDFHC}
\end{align}
holds for all $f:\{0,1\}^n\to[0,\infty)$ and $g:\{0,1\}^n\to[0,\infty)$, if and only if  $(p,q)\in \mathcal{R}_{\mathrm{FH}} (\rho)$.
\item The inequality  
\begin{align}
\langle f,g\rangle & \ge\Vert f\Vert_{p}\Vert g\Vert_{q}\label{eq:NICDRHC}
\end{align}
holds for all $f:\{0,1\}^n\to[0,\infty)$ and  $g:\{0,1\}^n\to[0,\infty)$, if and only if  $(p,q)\in \mathcal{R}_{\mathrm{RH}} (\rho)$.
\end{enumerate}
\end{theorem}

%\begin{theorem}[Hypercontractivity: DSBS and Two-Function Version]
%\label{thm:hyper2} Let $f:\mathcal{X}^{n}\to[0,\infty)$, $g:\mathcal{Y}^{n}\to[0,\infty)$
%and $(X^{n},Y^{n})$ be a DSBS with correlation coefficient $\rho$.
%The following hold: 
%\begin{enumerate}
%\item For $p,q\ge1$ such that $(p-1)(q-1)\ge\rho^{2}$, 
%\begin{align}
%\langle f,g\rangle & \leq\Vert f\Vert_{p}\Vert g\Vert_{q},\label{eq:NICDFHC}
%\end{align}
%where $\langle f,g\rangle:=\mathbb{E}[f(X^{n})g(Y^{n})]$ and $\Vert f\Vert_{p}:=(\mathbb{E}[f^{p}(X^{n})])^{1/p}$.
%%and $\Vert g\Vert _{q}:=(\mathbb{E}[g^{q}(Y^n)])^{1/q}$. 
%\item For $p,q\le1$ such that $(p-1)(q-1)\ge\rho^{2}$, 
%\begin{align}
%\langle f,g\rangle & \ge\Vert f\Vert_{p}\Vert g\Vert_{q}.\label{eq:NICDRHC}
%\end{align}
%\end{enumerate}
%\end{theorem}

These two inequalities (due to \cite{gross1975logarithmic, borell1982positivity, ODonnell14analysisof}) are known as the {\em two-function versions}
of the  hypercontractivity inequalities for the DSBS. These
inequalities are equivalent to the following {\em single-function
versions} of the hypercontractivity inequalities for the DSBS. 
%These inequalities
%also follow from \citet{gross1975logarithmic}, \citet{borell1982positivity},
%and \citet{ODonnell14analysisof}.

Before we describe these single-function versions, we introduce some additional notation. Denote $q'=\frac{q}{q-1}$ as  the \emph{H\"older conjugate} of $q$ for $q\ne 1$; for $q=1$, both $q=\pm\infty$ are H\"older conjugates of $q$. For  a   DSBS sequence  $(X^{n},Y^{n})\sim \pi_{XY}^{n} = \pi_{X|Y}^n\times \pi_Y^n$ with correlation coefficient $\rho$, the {\em noise operator} or {\em conditional expectation
operator} $T_{\rho}$ (or $\pi_{X|Y}^{n}$)  as 
\begin{equation}
T_{\rho}f(y^{n}):=\mathbb{E}[f(X^{n})\mid Y^{n}=y^{n}]=\sum_{x^n \in\calX^n}f(x^n ) \pi_{X  |Y }^n(x^n|y^n).\label{eqn:noise_op}
\end{equation}
One can easily check that $T_{\rho_1\rho_2}=T_{\rho_1}T_{\rho_2}$ for all $\rho_1,\rho_2 \in [0,1]$.

\begin{theorem}[Hypercontractivity: DSBS and Single-Function Version]
\label{thm:hyper_single} Let $(X^{n},Y^{n})\sim \pi_{XY}^{n}$ be a source sequence generated by a DSBS with correlation coefficient $\rho$.
%The following hold: 
\begin{enumerate}
\item The inequality 
\begin{align}
\Vert T_{\rho}f\Vert_{q} & \leq\Vert f\Vert_{p}\label{eq:NICDFHC-1}
\end{align}
holds for all $f:\{0,1\}^n\to[0,\infty)$, if and only if  $(p,q')\in \mathcal{R}_{\mathrm{FH}} (\rho)$ (with $1':=\infty$).
\item The inequality 
\begin{align}
\Vert T_{\rho}f\Vert_{q} & \ge\Vert f\Vert_{p}\label{eq:NICDRHC-1}
\end{align}
holds for all $f:\{0,1\}^n\to[0,\infty)$, if and only if  $(p,q')\in \mathcal{R}_{\mathrm{RH}} (\rho)$ (with $1':=-\infty$).
\end{enumerate}
\end{theorem}

Here we do not delve deeper into the equivalence between the single- and two-function versions of hypercontractivity inequalities, 
since we will discuss the equivalence in detail in Section \ref{subsec:Single-Function-Version}.

By applying the hypercontractivity inequalities,  \citet[Eqns.~(28) and~(29)]{kamath2016non} provided the following bounds.

\begin{theorem}[Hypercontractivity bound for the DSBS] \label{thm:hyper}  Define  the function
\begin{align}
\varphi_{a,b}(s,t,p):=\frac{(s^{p}a+\overline{a})^{\frac{1}{p}}(t^{q}b+\overline{b})^{\frac{1}{q}}-1}{(s-1)(t-1)}-\frac{a}{t-1}-\frac{b}{s-1}
\end{align}
with $q:=1+{\rho^{2}}/{(p-1)}$. Then, for any  sets $\calA$ and $\calB$
with $\pi_{X}^{n}(\calA)=a$ and $\pi_{Y}^{n}(\calB)=b$, 
\begin{align}
 & \sup_{s,t>0, p:(s-1)(t-1)(p-1)<0}\varphi_{a,b}(s,t,p)\le \pi_{XY}^{n}(\calA\times\calB)\label{eqn:hyper_con0}\\
 & \qquad\qquad\qquad\qquad\leq\inf_{s,t>0,p:(s-1)(t-1)(p-1)>0}\varphi_{a,b}(s,t,p).\label{eqn:hyper_con}
\end{align}
\end{theorem}

\begin{proof}
This theorem follows by setting $f$ and $g$ in Theorem \ref{thm:hyper2}  to be $\{s,1\}$-valued and $\{t,1\}$-valued functions respectively. Note that changing the range of the functions $f$ and $g$ from $\{0,1\}$ to the sets $\{s,1\}$  and $\{t,1\}$ respectively does not affect the values of the probability masses of the joint distribution of $(f(X^n),g(Y^n))$.
\end{proof}

\begin{figure}
\centering \includegraphics[width=0.8\columnwidth]{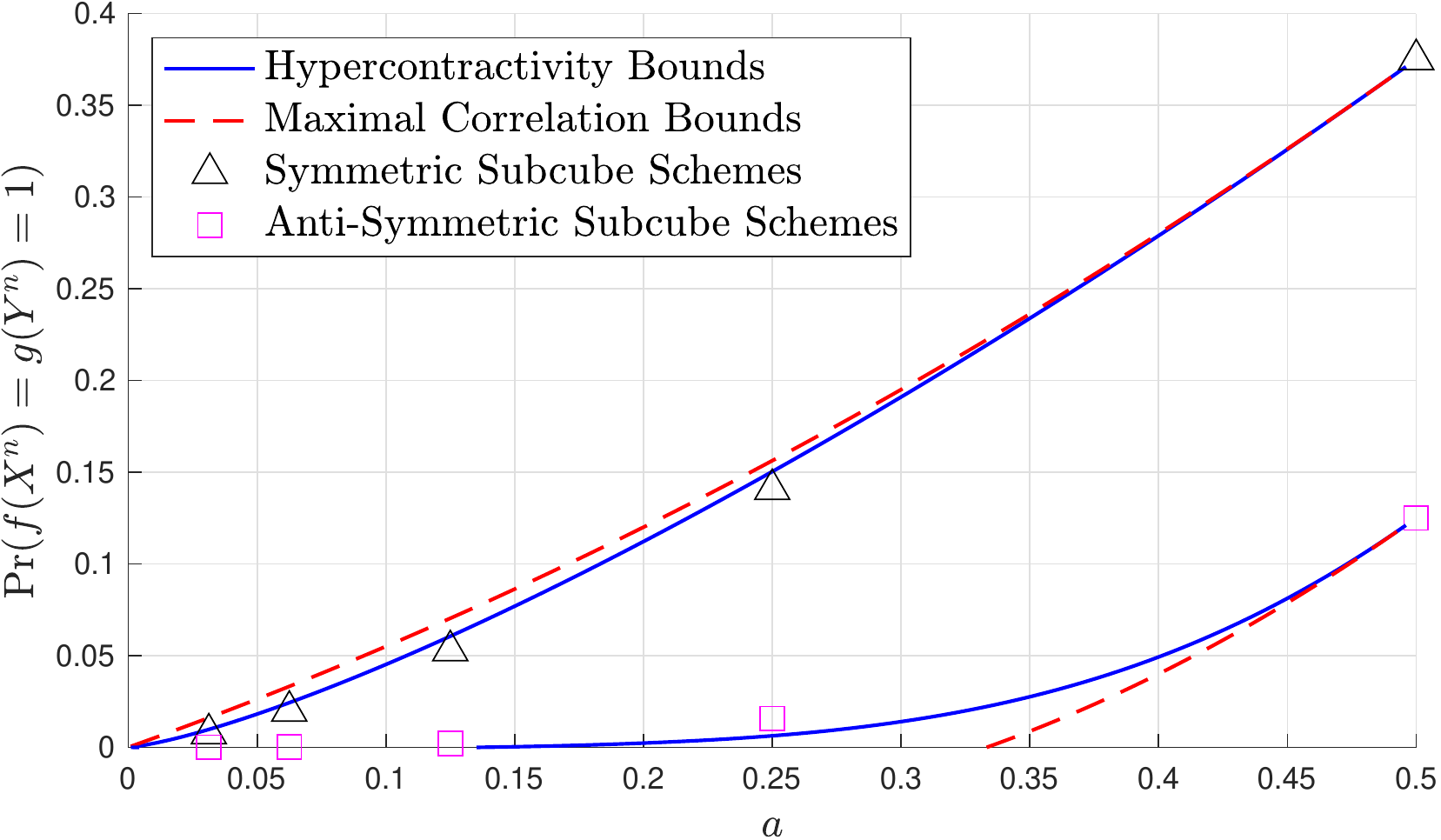}
\caption{Illustration the maximal correlation bounds in \eqref{eqn:max_cor},
the hypercontractivity bounds in \eqref{eqn:hyper_con0}--\eqref{eqn:hyper_con}
as well as the performances of symmetric and anti-symmetric subcube
schemes }
\label{fig:hc_mc} 
\end{figure}

It can be shown analytically that the hypercontractivity bounds are
no worse than the maximal correlation bounds in Theorem~\ref{thm:wit}
for any $a,b\in[0,1]$; see Fig.~\ref{fig:hc_mc} for a numerical
comparison. Moreover, for $a=b=1/2$, the hypercontractivity bounds
in \eqref{eqn:hyper_con0} and \eqref{eqn:hyper_con} reduce to the
sharp bounds $\frac{1-\rho}{4}\le \pi_{XY}^{n}(\calA\times\calB)\le\frac{1+\rho}{4}$,
which correspond to the bounds given by the maximal correlation technique in~\eqref{eqn:mc_half}. 

\subsection{Case of $a=b={1}/{4}$: Boolean Fourier Analysis} \label{sec:quarter}

We now consider the case $a=b={1}/{4}$, and we answer the forward
part of Mossel's mean-$1/4$ stability problem. Mossel's mean-$1/4$
stability problem \cite{mossel2017, mossel2021probabilistic} consists in the determination of $\overline{\Gamma}^{(n)}({1}/{4},1/{4})$
(forward part) and $\underline{\Gamma}^{(n)}({1}/{4},{1}/{4})$ (reverse
part) for $n\ge2$, and also the optimal Boolean functions that attain the maximum
and minimum that define these two quantities.

%Mossel's mean-1/4 stability problem: Which Boolean functions attain
%$\overline{\Gamma}^{(n)}(\frac{1}{4},\frac{1}{4})$ and
%$\underline{\Gamma}^{(n)}(\frac{1}{4},\frac{1}{4})$ for
%$n\ge2$?

The forward part of this problem was resolved by the present authors
in~\cite{yu2021non,yu2019improved} using elements of {\em Boolean
Fourier analysis}. We recap some fundamentals of this study here. Given
a Boolean function $f:\{0,1\}^{n}\to\{0,1\}$, its {\em Fourier
coefficients} are defined as 
\begin{align}
\hat{f}_{\calS}:=\bbE[ f(X^{n})\chi_{\calS}(X^n)] =\frac{1}{2^n}\sum_{x^n\in \{0,1\}^{n}}f(x^n)\, \chi_{\calS}(x^n)\quad\mbox{for all}\;\, \calS\!\subset\![n],
\end{align}
where the {\em (Fourier) basis functions} are
\begin{equation}
\chi_{\calS}(x^n):=(-1)^{\sum_{i\in \calS}x_i}\quad\mbox{for all}\;\, x^n\in \{0,1\}^{n},
\end{equation}
and $X^n\sim \Unif\{0,1\}^{n}$. 
%, as assumed almost throughout this chapter. 
%where $\langle x^{n},y^{n}\rangle=\sum_{i=1}^{n}x_{i}y_{i}$ is the
%{\em inner product} of $x^{n}$ and $y^{n}$. 
%The {\em Fourier
%expansion} of $f$ can be expressed in terms of the Fourier coefficients
%as 
The function $f$ can be expressed in terms of the Fourier coefficients as
\begin{align}
f(x^{n})=\sum_{\calS \subset[n]}\hat{f}_\calS  \ \chi_\calS (x^n) \quad\mbox{for all} \;\, x^{n}\in\{0,1\}^{n},
\end{align}
which is known as the {\em Fourier expansion} of $f$.
For $0\le k \le n$, define the {\em degree-$k$ Fourier weight} of $f$ as 
\begin{align}
\mathbf{W}_{k}[f]:=\sum_{\calS \subset [n]:|\calS|=k}\hat{f}_\calS^{2}. \label{eq:NICD-weight}
\end{align}
%where $|y^{n}|:=d_{\rmH}(y^{n},0^{n})$ denotes the {\em Hamming
%weight} of $y^{n}$. 
%One can interpret $\mathbf{W}_{k}[f]$ as the amount
%that the degree-$k$ terms in the Fourier expansion of $f$ contribute
%to $f(x^{n})$. 
%\red{(Lei: The above sentence is not so clear to me. Could you refine it?)}
%The {\em mean} of a Boolean function $f$ is $\bbE[ f(X^{n})]$.
%%$2^{-n}\sum_{x^{n}}f(x^{n})$.
%Intuitively,  $\bW_k[f]$ represents the amount that the degree $k$ terms in $f$'s Fourier expansion contribute to ascertaining the value of $f$. 
It is easy to check that if we define the {\em degree-$k$ part of $f$} as $f_k (x^n) := \sum_{\calS \subset [n]:|\calS|=k} \hat{f}_\calS\ \chi_\calS ( x^n )$, then $\bbE[f_k (X^n)^2] = \bW_k[f]$. Hence, $\bW_k[f]$ represents the  ``energy'' of the degree-$k$ part in $f$'s Fourier expansion. 

The Fourier weights satisfy the following properties. Proofs of these
properties can be found in the delightful exposition of Boolean functions
by \citet{ODonnell14analysisof}.

\begin{lemma} \label{lem:W1} For a Boolean function $f:\{0,1\}^{n}\to\{0,1\}$
with mean $a$, 
\begin{align}
\mathbf{W}_{0}[f]=a^{2}\quad\mbox{and}\quad\sum_{k=0}^{n}\mathbf{W}_{k}[f]=a.\label{eqn:wt_prop}
\end{align}
Furthermore, if $(X^{n},Y^{n})\sim \pi_{XY}^{n}$ is a source sequence
of the DSBS with correlation coefficient~$\rho$, then
for any pair of Boolean functions $f,g:\{0,1\}^{n}\to\{0,1\}$, 
\begin{align}
\Pr(f(X^{n})=g(Y^{n})=1) & =\sum_{k=0}^{n}\rho^{k}\sum_{\calS \subset [n]:|\calS |=k}\hat{f}_\calS\,  \hat{g}_\calS \quad\mbox{and}\\
\Pr(f(X^{n})=f(Y^{n})=1) & =\sum_{k=0}^{n}\mathbf{W}_{k}[f]\, \rho^{k}.
\end{align}
\end{lemma}

For $\rho\in(0,1)$, lower degree Fourier weights have a higher contribution
to the joint probability $\Pr(f(X^{n})=f(Y^{n})=1)$ than higher degree
weights. Hence, to bound this joint probability, we can focus on bounding
the lower degree Fourier weights of $f$. Observe from~\eqref{eqn:wt_prop}
that given the mean of $f$, the degree-$0$ Fourier weight is fully
specified. Hence, it is instructive to estimate the {\em second
most important} Fourier weight. In particular, we are interested
in the {\em degree-$1$ Fourier weight}~$\mathbf{W}_{1}[f]$ under the condition
that the mean of $f$ is specified. In the literature, there exist
several bounds on $\mathbf{W}_{1}[f]$. These include Chang's bound, which can
be found in~\cite[Level-1 Inequality]{ODonnell14analysisof} and~\cite{chang2002polynomial}
and the linear programming (LP) bounds of \citet{fu2001minimum} and
\citet{yu2019improved}. In particular, the LP bounds state that 
\begin{align}
\mathbf{W}_{1}[f]\leq\varphi(a):=\begin{cases}
{\displaystyle 2a(\sqrt{a}-a)} & 0\le a\le1/4\vspace{.03in}\\
{a}/{2} & {1}/{4}<a\le{1}/{2}
\end{cases}.\label{eqn:lp_bounds}
\end{align}
By the Cauchy--Schwarz inequality, one easily observes that 
\begin{align}
 &\hspace{-.2in} \Pr(f(X^{n})=g(Y^{n})=1)\nn\\*
 &\hspace{-.2in} \;\le\max\big\{\!\Pr(f(X^{n})\!=\! f(Y^{n})\!=\! 1),\Pr(g(X^{n})\!=\! g(Y^{n})\!=\! 1)\big\}.\label{eqn:cs_iq}
\end{align}
This inequality implies that in the determination of $\overline{\Gamma}^{(n)}(a,a)$
(the symmetric case in which $a=b$), it suffices to consider a pair
of {\em identical} Boolean functions.

By combining the ideas in Lemma~\ref{lem:W1}, the LP bound in~\eqref{eqn:lp_bounds}
and~\eqref{eqn:cs_iq}, the present authors proved the following
result \cite{yu2019improved,yu2021non}. \begin{theorem} For all
$a\in[0,1]$ and $n\ge2$, 
\begin{align}
\overline{\Gamma}^{(n)}(a,a) & \le a^{2}+\rho\varphi(a)+\rho^{2}\big(a-a^{2}-\varphi(a)\big).
\end{align}
\end{theorem}

Particularizing this upper bound to $a=b=1/4$, we obtain $\overline{\Gamma}^{(n)}(1/4,1/4)\le(\frac{1+\rho}{4})^{2}$.
Per the discussion leading to~\eqref{eqn:subcube_id}, this upper
bound is attained by a pair of identical $(n-2)$-subcubes. Hence,
\begin{equation}
\overline{\Gamma}^{(n)}\Big(\frac{1}{4},\frac{1}{4}\Big)=\Big(\frac{1+\rho}{4}\Big)^{2}\quad\mbox{for all}\;\,  n\ge2, \label{eq:quarter_quarter}
\end{equation}
resolving the forward part of Mossel's mean-$1/4$ stability problem.
However, the reverse part of the same problem (i.e., which Boolean
functions attain $\underline{\Gamma}^{(n)}(1/4,1/4)$) remains open.

\section{Converse in the  Moderate Deviations Regime}
\label{sec:sse}
% and Central Limit Regimes with small~$a$ and~$b$

%and similarly, in the MD and LD regimes,
%whether Hamming balls and spheres are optimal in attaining the asymptotic
%exponents; cf.\ Table~\ref{tab:Comparison-of-subcubes}. As mentioned
%prior to the statement of Conjecture~\ref{conj:ordentlich2020note},
%the optimality of Hamming balls and spheres in the LD case corresponds
%to the OPS conjecture. We shall address these questions in the following
%subsections.

We now consider the optimality of Hamming balls and spheres in the
MD regime and the CL regime with small~$a$ and~$b$. To address
this question, we resort to two key ideas, namely the hypercontractivity
inequalities in Theorem \ref{thm:hyper2} and the  {\em small set expansion (SSE)} theorem. 

%We have already
%seen the former in the context of Theorem~\ref{thm:hyper}. We further
%elaborate on it to address the regimes in question.

A well-known result to address the optimality of Hamming balls and
spheres in the MD regime and the CL regime with small~$a$ and~$b$
is the  SSE theorem  \cite{ODonnell14analysisof,mossel2006non}, which is a consequence of the hypercontractivity
inequalities in Theorem \ref{thm:hyper2}. 
\begin{theorem}[Small set expansion: DSBS version] \label{thm:sse}
For any $n\ge1$ and $\alpha,\beta>0$, 
\begin{align}
\underline{\Upsilon}_{\mathrm{MD}}^{(n)}(\alpha,\beta) & \ge\underline{\Upsilon}_{\mathrm{MD}}(\alpha,\beta)\quad\mbox{and}\\
\overline{\Upsilon}_{\mathrm{MD}}^{(n)}(\alpha,\beta) & \leq\overline{\Upsilon}_{\mathrm{MD}}(\alpha,\beta),
\end{align}
where $\underline{\Upsilon}_{\mathrm{MD}}$ and $\overline{\Upsilon}_{\mathrm{MD}}$
are %defined in \eqref{eq:NICDmdexponent} and \eqref{eq:NICDmdexponent-1}
%respectively, and also 
expressed in closed form for the DSBS in~\eqref{eq:NICD-MD}
and~\eqref{eq:NICD-MD2} respectively. \end{theorem} 

The reader might wonder about the term ``small set expansion'' that is used to describe Theorem~\ref{thm:sse}. This term refers to  
%a consequence 
a curious phenomenon of the Hamming cube being a ``small set expander'' in the sense that any small subset $\calA\subset\{0,1\}^n$ has an usually large (or expanded) boundary. Here, the Hamming cube is regarded as an
edge-weighted  complete   graph, known as the \emph{$\rho$-stable hypercube graph}, in which each edge $(x^n,y^n)$ is assigned a  weight equal to the probability $\pi_{XY}^n(x^n,y^n)$.
The limiting case as $\rho  \downarrow 0$ of this phenomenon is   quantified by the edge-isoperimetric  inequality which will be stated in Theorem~\ref{thm:EdgeIsoper}. We refer readers to \citet{ODonnell14analysisof} for more intuition about the  term ``small set expansion''. 
\begin{proof}[Proof Sketch of Theorem \ref{thm:sse}] 
%The SSE theorem is a direct consequence
%of the classic hypercontractivity inequalities in \eqref{eq:NICDFHC}
%and \eqref{eq:NICDRHC}. 
Substituting the indicator functions $f\leftarrow\bone_{\calA}$
and $g\leftarrow\bone_{\calB}$ into~\eqref{eq:NICDFHC} and~\eqref{eq:NICDRHC} respectively,
and optimizing over $(p,q)$, we obtain the inequalities as stated in the SSE
theorem. \end{proof}

Due to the equivalence among the CL, MD, and LD exponents for all~$n\in\bbN$ (as discussed after Definition~\ref{def:for_rev_exp})
and the homogeneity property in~\eqref{eq:NICD-6} and~\eqref{eq:NICD-12}, 
$\underline{\Upsilon}_{\mathrm{MD}}^{(n)}(\alpha,\beta)$ and~$\overline{\Upsilon}_{\mathrm{MD}}^{(n)}(\alpha,\beta)$
in Theorem~\ref{thm:sse} can be replaced by $\underline{\Upsilon}_{\mathrm{CL}}^{(n)}(\alpha,\beta)$
and~$\overline{\Upsilon}_{\mathrm{CL}}^{(n)}(\alpha,\beta)$ respectively, or by $\underline{\Upsilon}_{\mathrm{LD}}^{(n)}(\alpha,\beta)$
and~$\overline{\Upsilon}_{\mathrm{LD}}^{(n)}(\alpha,\beta)$ respectively.

%since the bounds above satisfy the homogeneity property in~\eqref{eq:NICD-6} and~\eqref{eq:NICD-12}.

The bounds in the SSE theorem are achieved by sequences of Hamming
balls or spherical shells. 
 Hence, these geometric objects % sequences of Hamming balls or spherical shells
are optimal in attaining the   MD exponents. 

\section{Converse in the  Large Deviations Regime}
\label{sec:ldr} 
We now address the final asymptotic regime of interest,
namely, the large deviations regime. First, we introduce some terminology. Let
$\calI\subset\bbR^{d}$ be a convex subset of $d$-dimensional Euclidean space.
We recall that for a function $f:\calI\to\bbR$, its {\em lower
convex envelope} $\cvx[f]$ is the function defined at each point
of $\calI$ as the supremum of all convex functions that lie under $f$, i.e., for every $\bx\in\bbR^d$,  $\cvx[f](\bx):=\sup\{g(\bx):g\mbox{ is convex},g\le f\mbox{ on }\calI\}$.
By Carath\'eodory's theorem, equivalently, 
\begin{align} 
\cvx[f](\bx) & =\inf_{\{\bx_i\}_{i=1}^{d+1}\subset\calI,~\{\lambda_i\}_{i=1}^{d+1}}\sum_{i=1}^{d+1}\lambda_i f(\bx_i),\label{eqn:cara}
%\Big\{\sum_{i=1}^{m+1}\lambda_{i}f(x_{i}^{m}): %\sum_{i=1}^{m+1}\lambda_{i}x_{i}^{m}=x^{m},\sum_{i=1}^{m+1}%\lambda_{i}=1,\nonumber \\
%  x_{i}^{m}\in\calI,\lambda_{i}\ge0,~\forall~ i\in[m+1]\Big\}.
\end{align}
where $\{ \lambda_i\}_{i=1}^{d+1}$ is a $(d+1)$-dimensional probability mass function with $\sum_{i=1}^{d+1}\lambda_i\bx_i=\bx$. 
The {\em upper concave envelope} $\cve[f]$ is defined as $-\cvx[-f]$.
The SSE theorem in Theorem~\ref{thm:sse} can be strengthened to
the following result, known as the {\em strong SSE theorem}; see
\citet{yu2021graphs} and \citet{yu2021strong}. \begin{theorem}[Strong small set expansion: DSBS version]
\label{thm:strongSSE_DSBS} \label{thm:strongsse} For any $n\ge1$
and $\alpha,\beta\in(0,1]$, 
\begin{align}
\underline{\Upsilon}_{\mathrm{LD}}^{(n)}(\alpha,\beta) & \ge\cvx[\underline{\Upsilon}_{\mathrm{LD}}](\alpha,\beta)\quad\mbox{and}\label{eq:NICDLD}\\
\overline{\Upsilon}_{\mathrm{LD}}^{(n)}(\alpha,\beta) & \leq\cve[\overline{\Upsilon}_{\mathrm{LD}}](\alpha,\beta).\label{eq:NICDLD2}
\end{align}
%where $\cvx[f]$ and $\cve[f]$ respectively
%denote the lower convex and upper concave envelopes of the function~$f : [0,1]^2\to\bbR$.  
\end{theorem} The proof  of this theorem (and also its  generalization to the finite alphabet case in Theorem \ref{thm:strongsse-2}) will be provided in Section
\ref{sec:NICD-Stability}. The proof is
based on the information-theoretic characterizations of hypercontractivity
inequalities (also discussed in Section~\ref{ch:funineq}). 
%In fact, Theorem~\ref{thm:strongsse} is generalized
%to any finite alphabet case in Section \ref{sec:NICD-Stability} . 

By Carath\'eodory's representation of the lower convex and upper concave envelopes in \eqref{eqn:cara}, the bounds in Theorem~\ref{thm:strongsse} can be asymptotically achieved by
%sequences of  
``time-sharing'' at most three (since $d=2$ in our case) concentric or anti-concentric Hamming spheres (or balls) for each length $n$.
Specifically, let $(\lambda_1,\lambda_2,\lambda_3)$ be a PMF, i.e.,  $\lambda_i \ge 0$ for all $i \in [3]$  and $\sum_{i=1}^3 \lambda_i =1$. For each blocklength $n\in\bbN$, this strategy uses certain
concentric or anti-concentric Hamming spheres $\bbS^{(i)}$ for a period of length $\lfloor n \lambda_i \rfloor,i\in[3]$.
%, and   then varying
%the distribution $(p_1,p_2,p_3)$. 
Since time-sharing of certain Hamming spheres
is optimal in the LD regime, this confirms a weaker version of the
OPS conjecture (Conjecture~\ref{conj:ordentlich2020note}) in which the convexification and concavification operations in~\eqref{eq:NICDLD} and~\eqref{eq:NICDLD2} respectively are permitted.

Theorem~\ref{thm:strongsse} is known as the {\em strong} SSE
theorem because the bounds given in Theorem~\ref{thm:strongsse}
are asymptotically sharp in the LD regime. This is in contrast to
the ones given in the vanilla SSE theorem (Theorem~\ref{thm:sse})
which are not sharp in the LD regime. Furthermore, both these two theorems
are asymptotically sharp in the MD regime, since the bounds in the
strong SSE theorem reduce to the ones in the SSE theorem, as shown
in  \eqref{eq:NICD-13a} and \eqref{eq:NICD-13}. Hence,
Theorem~\ref{thm:strongsse} is stronger than the SSE theorem (Theorem~\ref{thm:sse}), in the sense that for all $\alpha, \beta \in [0,1]$ and $\gamma >0$,
\begin{align}
\cvx[\underline{\Upsilon}_{\mathrm{LD}}](\alpha,\beta) & \ge\underline{\Upsilon}_{\mathrm{MD}}(\gamma\alpha,\gamma\beta)\quad\mbox{and}\\
\cve[\overline{\Upsilon}_{\mathrm{LD}}](\alpha,\beta) & \le\overline{\Upsilon}_{\mathrm{MD}}(\gamma\alpha,\gamma\beta).
\end{align}

%\begin{table}
%\centering %
%\begin{tabular}{|>{\centering}m{2cm}|>{\centering}m{2cm}|>{\centering}m{2cm}|>{\centering}m{2cm}|>{\centering}m{2cm}|}
%\hline 
%\multirow{1}{2cm}{\centering{}Regimes } & \multicolumn{2}{c|}{Central Limit} & \centering{}Moderate Deviations  & \centering{}Large Deviations\tabularnewline
%\hline 
% & \centering{}Fixed and large $a,b$  & \centering{}Fixed but small $a,b$  & \centering{}Subexp.\ vanishing $a,b$  & \centering{}Exp.\ vanishing $a,b$\tabularnewline
%\hline 
%\hline 
%\centering{}Maximal Correlation  & \centering{}Sharp for $a=b=1/2$  & \centering{}Not sharp  & \centering{}Not sharp  & \centering{}Not sharp\tabularnewline
%\hline 
%\centering{}Fourier Analysis  & \centering{}Sharp for $a=b=1/2$ and $a=b=1/4$  & \centering{}Not sharp  & \centering{}Not sharp  & \centering{}Not sharp\tabularnewline
%\hline 
%\centering{}SSE  & \centering{}Not sharp  & \multirow{2}{2cm}{\centering{}Essentially sharp} & \multirow{2}{2cm}{\centering{}Sharp} & \centering{}Not sharp\tabularnewline
%\cline{1-2} \cline{2-2} \cline{5-5} 
%\centering{}Strong SSE  & \centering{}Not sharp &  &  & \centering{}Sharp\tabularnewline
%\cline{1-2} \cline{2-2} \cline{5-5} 
%\end{tabular}\caption{\label{tab:Converse-results-for}Converse (optimality) results, together
%with techniques used to prove them, for the $2$-user NICD problem
%in the CL, MD, and LD regimes }
%\end{table}

To prove the OPS conjecture, we need to remove the operations
of taking the lower convex and upper concave envelopes in the strong
SSE theorem. This was done by  the first
author of this monograph~\cite{yu2021convexity}.  In particular, he showed that  $\underline{\Upsilon}_{\mathrm{LD}}$ is convex
and $\overline{\Upsilon}_{\mathrm{LD}}$ is concave. 
Combining this result with the strong SSE theorem (Theorem~\ref{thm:strongsse})  allows us to conclude that the OPS conjecture
is unconditionally true and that Hamming balls or spheres (without
time-sharing) are optimal in the LD regime \cite{yu2021convexity}. 
That is, for the DSBS and $\alpha,\beta\in(0,1)$,
\begin{align}
\hspace{-.2in}\underline{\Upsilon}_{\mathrm{LD}}^{(\infty)}(\alpha,\beta)=\underline{\Upsilon}_{\mathrm{LD}}(\alpha,\beta)\quad\mbox{and}\quad\overline{\Upsilon}_{\mathrm{LD}}^{(\infty)}(\alpha,\beta)=\overline{\Upsilon}_{\mathrm{LD}}(\alpha,\beta). \label{eq:opstrue}
\end{align}

Several special cases of \eqref{eq:opstrue} were established in
the literature prior to the most general result of \citet{yu2021convexity}. The limiting cases
as $\rho\downarrow0$ and $\rho\uparrow1$ were shown by \citet{ordentlich2020note}.
The ``symmetric'' special case with $\alpha=\beta$ was shown by \citet{kirshner2021moment}.
We   introduce these results in Section \ref{ch:funineq}, since
they are consequences of strengthened versions of the hypercontractivity inequalities. 

We summarize all converse results  discussed in Sections~\ref{sec:nicd_conv}--\ref{sec:ldr} and techniques used to prove them
in Table \ref{tab:Converse-results-for}.

\begin{table}[!ht]
\caption{\label{tab:Converse-results-for}Converse (optimality) results and techniques for the $2$-user NICD problem
in the CL, MD, and LD regimes }
\centering{\small
\begin{tabular}{|>{\centering}m{1.85cm}|>{\centering}m{1.85cm}|>{\centering}m{1.85cm}|>{\centering}m{2.05cm}|>{\centering}m{2.05cm}|}
\hline 
\multirow{2}{1.9cm}{\centering{}Regimes  } & \multicolumn{2}{c|}{Central Limit} & \centering{}Moderate Deviations & \centering{}Large Deviations\tabularnewline
\cline{2-5} \cline{3-5} \cline{4-5} \cline{5-5} 
 & \centering{}Fixed and large $a,b$ & \centering{}Fixed but small $a,b$ & \centering{}Subexp.\ vanishing $a,b$ & \centering{}Exp.\ vanishing $a,b$\tabularnewline
\hline 
\hline 
\centering{}Maximal Correlation & \centering{}Sharp for $a=b=1/2$ & \centering{}Not sharp & \centering{}Not sharp & \centering{}Not sharp\tabularnewline
\hline 
\centering{}Fourier Analysis & \centering{}Sharp for $a=b=1/2$ and $a=b=1/4$ & \centering{}Not sharp & \centering{}Not sharp & \centering{}Not sharp\tabularnewline
\hline 
\centering{}SSE & \centering{}Not sharp & \multirow{2}{2cm}{\centering{}Essentially sharp} & \multirow{2}{2cm}{\centering{}Sharp} & \centering{}Not sharp\tabularnewline
\cline{1-2} \cline{2-2} \cline{5-5} 
\centering{}Strong SSE & \centering{}Not sharp&  &  & \centering{}Sharp\tabularnewline
\hline 
\end{tabular}}
\end{table}

\section{Extensions to   Sources Beyond the DSBS}
\label{sec:nicd_arb}

Thus far, we have only considered the DSBS. Can the results in Sections \ref{sec:nicd_ach}-\ref{sec:ldr} be extended to other bivariate memoryless sources?
Indeed, the SSE and strong SSE theorems, can be extended to sources
on Polish spaces (separable completely metrizable topological space).
We refer the reader to \cite{yu2021strong} for details. Here for
simplicity, we   discuss analogues of the preceding results for the finite alphabet and bivariate Gaussian cases. The NICD problem for the latter case has been completely solved by \citet{borell1985geometric} and \citet{mossel2015robust}. 

%alphabet case, we will also introduce the Gaussian case, for which
%the NICD problem has been completely solved. 

\subsection{Finite Alphabets }
\label{subsec:finite}
In this section, we generalize the NICD problem to the finite alphabet
case in which $\calX$ and $\calY$ are finite sets. %We still require the source to be memoryless.
 Let $\pi_{XY}\in\calP(\calX\times\calY)$. 
For simplicity, we assume that the supports of $\pi_{X}$ and $\pi_{Y}$
are $\mathcal{X}$ and $\mathcal{Y}$ respectively.   Given $\pi_{X}$ and $\pi_{Y}$,
define their maximum exponents of ``atomic events'' as 
\begin{equation}
    \begin{aligned}
 \alpha_{\max}(\pi_{X})&:=\max_{x\in\calX}\;\log\frac{1}{\pi_{X}(x)}\qquad\mbox{and}\\ 
 \beta_{\max}(\pi_{Y})&:=\max_{y\in\calY}\;\log\frac{1}{\pi_{Y}(y)}. \label{eqn:alpha_max}
\end{aligned}
\end{equation}
For $n\ge1$, $\alpha\in(0,\alpha_{\max}(\pi_{X})]$ and $\beta\in(0,\beta_{\max}(\pi_{Y})]$,
re-define the {\em forward} and {\em reverse LD exponents}  respectively as
\begin{align}
\hspace{-.35in}\underline{\Upsilon}_{\mathrm{LD}}^{(n)}(\alpha,\beta) & :=-\frac{1}{n}\log\max_{\substack{\mathcal{A}\subset\mathcal{X}^{n},\mathcal{B}\subset\mathcal{Y}^{n}:\\
\pi_{X}^{n}(\mathcal{A})\leq2^{-n\alpha},\pi_{Y}^{n}(\mathcal{B})\leq2^{-n\beta}
}
}\!  \pi_{XY}^{n}(\mathcal{A}\times\mathcal{B})\;\;\mbox{and}\label{eq:FI-72}\\
\hspace{-.35in}\overline{\Upsilon}_{\mathrm{LD}}^{(n)}(\alpha,\beta) & :=-\frac{1}{n}\log\min_{\substack{\mathcal{A}\subset\mathcal{X}^{n},\mathcal{B}\subset\mathcal{Y}^{n}:\\
\pi_{X}^{n}(\mathcal{A})\ge2^{-n\alpha},\pi_{Y}^{n}(\mathcal{B})\ge2^{-n\beta}
}
}\!  \pi_{XY}^{n}(\mathcal{A}\times\mathcal{B}).\label{eq:FI-71}
\end{align}
Let $\underline{\Upsilon}_{\mathrm{LD}}^{(\infty)}$ and $\overline{\Upsilon}_{\mathrm{LD}}^{(\infty)}$
be their pointwise limits as $n\to\infty$. These are the same as the forward
and reverse LD exponents in \eqref{eqn:LD_for} and \eqref{eqn:LD_rev}
but here, $\pi_{XY}$ is no longer restricted to be a DSBS.  

\begin{theorem}[Strong small set expansion: General version] \label{thm:strongsse-2}
For any joint distribution on a finite alphabet $\pi_{XY}$, any blocklength $n\ge1$, $\alpha\in(0,\alpha_{\max}(\pi_{X})]$, and $\beta\in(0,\beta_{\max}(\pi_{Y})]$,
\eqref{eq:NICDLD} and \eqref{eq:NICDLD2} remain true, with $\underline{\Upsilon}_{\mathrm{LD}}$
and $\overline{\Upsilon}_{\mathrm{LD}}$ defined in \eqref{eqn:LDsphere1}
and \eqref{eqn:LDsphere2} for  $\pi_{XY}$, i.e., 
\begin{equation}
\underline{\Upsilon}_{\mathrm{LD}}(\alpha,\beta)=\min_{Q_X, Q_Y:D(Q_X\|\pi_X)\ge\alpha, D(Q_Y\|\pi_Y)\ge\beta  } \rvD(Q_X, Q_Y\|\pi_{XY}) \label{eq:Upsilon_LD}
\end{equation}
and analogously for $\overline{\Upsilon}_{\mathrm{LD}}$. 
 Moreover,
the inequalities in \eqref{eq:NICDLD} and \eqref{eq:NICDLD2} 
remain asymptotically tight in the limit as $n\to\infty$.  \end{theorem}

However, in general, $\underline{\Upsilon}_{\mathrm{LD}}$ and $\overline{\Upsilon}_{\mathrm{LD}}$
are not necessarily convex and concave, respectively. Hence, unlike
the case of the DSBS, for sources on finite alphabets, the operations of
taking the lower convex and upper concave envelopes in \eqref{eq:NICDLD}
and \eqref{eq:NICDLD2} cannot be removed in general. Nevertheless,  the bounds $\cvx[\underline{\Upsilon}_{\mathrm{LD}}](\alpha,\beta)$
and $\cve[\overline{\Upsilon}_{\mathrm{LD}}](\alpha,\beta)$ can be
asymptotically attained by  \emph{time-sharing} the use of  at most three type classes (cf.\ the discussion after Theorem~\ref{thm:strongsse}). 

%  at most three
%type classes for each dimension $n$, but using only type class for each dimension $n$ is  in general not optimal.  

% the "three" is not explained well. I would remove and just keep the "time-sharing" which is easy to understand. 

Theorem~\ref{thm:strongsse-2}  was first proven by \citet{yu2021graphs} by using information-theoretic and coupling techniques. In this monograph, we will provide a simple proof of Theorem~\ref{thm:strongsse-2}, which
is based on the information-theoretic characterizations of hypercontractivity
inequalities as discussed  in Section~\ref{sec:NICD-Stability}.

Similarly, one can generalize the DSBS-specific definitions in \eqref{eqn:MD_for}
and \eqref{eqn:MD_rev} to an arbitrary distribution $\pi_{XY}$ on
a finite alphabet. Then, the
SSE theorem (Theorem \ref{thm:sse}) can be also generalized
to the finite alphabet case. 

\begin{theorem}[Small set expansion: General version] \label{thm:sse-1}
For any $n\ge1$ and $\alpha,\beta>0$, 
\begin{align}
\underline{\Upsilon}_{\mathrm{MD}}^{(n)}(\alpha,\beta) & \ge\lim_{\epsilon\downarrow0}\frac{1}{\epsilon}\cvx[\underline{\Upsilon}_{\mathrm{LD}}](\epsilon\alpha,\epsilon\beta)\quad\mbox{and}\label{eq:NICDMD}\\
\overline{\Upsilon}_{\mathrm{MD}}^{(n)}(\alpha,\beta) & \leq\lim_{\epsilon\downarrow0}\frac{1}{\epsilon}\cve[\overline{\Upsilon}_{\mathrm{LD}}](\epsilon\alpha,\epsilon\beta).\label{eq:NICDMD2}
\end{align}
Moreover, the inequalities in \eqref{eq:NICDMD} and \eqref{eq:NICDMD2}
are asymptotically tight in the limit as $n\to\infty$. \end{theorem}

Since, in general,  $\underline{\Upsilon}_{\mathrm{LD}}$ and $\overline{\Upsilon}_{\mathrm{LD}}$
are not necessarily convex and concave, respectively, the operations
of taking the lower convex and upper concave envelopes in~\eqref{eq:NICDMD}
and \eqref{eq:NICDMD2} cannot be removed  as well. As a consequence,
for this case, \eqref{eq:NICDMD} and \eqref{eq:NICDMD2} cannot be
written as in the variational expressions that appear  on right-hand sides of~\eqref{eq:NICDmdexponent} and~\eqref{eq:NICDmdexponent-1}. %, i.e., the infimization and supremization of $\hat{\rvX}^{2}(\eta_{X},\eta_{Y}\|\pi_{XY})$ subject to the appropriate constraints on $\eta_X$ and $\eta_Y$. 

% and they are also not equal to the right-hand-sides of \eqref{eq:NICD-MD}
%and \eqref{eq:NICD-MD2} in general. 

\subsection{Gaussian Sources}
\label{subsec:NICD-Gaussian}

 We next consider  memoryless bivariate Gaussian sources with
correlation coefficient $\rho\in(-1,1)\setminus\{0\}$. For such sources, the NICD
problem was completely solved by \citet{borell1985geometric} (for the  symmetric cases in which $a=b$) and
\citet{mossel2015robust} (for the asymmetric cases) for all $(a,b)\in[0,1]^{2}$ and {\em
non-asymptotically}, i.e., for all   $n$. Let $\pi_{XY}$ be the bivariate Gaussian
distribution with mean $(0,0)$ and covariance matrix~$\bK$ given
in \eqref{eqn:cov_mat_rho}, where the correlation coefficient $\rho\in(-1,1)\setminus\{0\}$.
%For this case, $\pi_{X}=\pi_{Y}$ are the  standard univariant Gaussian distribution.
As usual, let $(X^{n},Y^{n})\sim \pi_{XY}^{n}$. %The following holds. 
\begin{theorem}[Borell's isoperimetric theorem]
\label{thm:borell_nicd} For any $n\ge 1$ and  $a,b\in[0,1]$, 
\begin{align}
\overline{\Gamma}^{(n)}(a,b)=\Lambda_{\rho}(a,b)\quad\mbox{and}\quad\underline{\Gamma}^{(n)}(a,b)=\Lambda_{-\rho}(a,b),
\end{align}
where the bivariate normal copula $\Lambda_{\rho}(\cdot,\cdot)$ is
defined in \eqref{eq:NICDcopula}. \end{theorem} 
Moreover, it has been shown by \citet{mossel2015robust} that  
  the optimal subsets $(\calA,\calB)$
attaining $\overline{\Gamma}^{(n)}$ or $\underline{\Gamma}^{(n)}$
must be equal to parallel halfspaces (almost everywhere).
%if $\calA$ and $\calB$ with probabilities respectively equal to $a,b$ achieve $\overline{\Gamma}^{(n)}(a,b)$ (or $\underline{\Gamma}^{(n)}(a,b)$), then $\calA$ and $\calB$ must be almost surely equal to parallel halfspaces.
 
Specialized to the case of
$a=b=1/2$, this theorem implies that 
\begin{align}
\hspace{-.3in}\overline{\Gamma}^{(n)}\Big(\frac{1}{2},\frac{1}{2}\Big)=\frac{1}{2}-\frac{\arccos\rho}{2\pi}\quad\mbox{and}\quad\underline{\Gamma}^{(n)}\Big(\frac{1}{2},\frac{1}{2}\Big)=\frac{\arccos\rho}{2\pi}. \label{eq:half_half}
\end{align}
The optimal $(\calA,\calB)$ attaining $\overline{\Gamma}^{(n)}(1/2,1/2)$
 correspond to  a pair of    identical
halfspaces through the origin. In contrast, the optimal $(\calA,\calB)$
attaining $\underline{\Gamma}^{(n)}(1/2,1/2)$   correspond to a pair
of   complementary halfspaces through
the origin. 

Next, we provide a proof sketch of Theorem \ref{thm:borell_nicd} which is 
due to \citet{mossel2015robust}.  In fact,
they also proved the following equivalent form of Theorem
\ref{thm:borell_nicd}. 

\begin{theorem}\label{thm:functional-borell} For any $n\ge 1$,  any pair of measurable
functions $f,g:\bbR^{n}\to[0,1]$, and any $0<\rho<1$, 
\begin{equation}
\bbE\big[\Lambda_{\rho}(f(X^{n}),g(Y^{n}))\big]\le\Lambda_{\rho}\big(\bbE[f(X^{n})],\bbE[g(Y^{n})]\big).\label{eq:functional-borell}
\end{equation}
If $-1<\rho<0$,   the inequality in~\eqref{eq:functional-borell}
is reversed. \end{theorem}

To see that Theorem~\ref{thm:functional-borell} implies 
Theorem~\ref{thm:borell_nicd}, set $f=\bone_{\mathcal{A}}$ and $g=\bone_{\mathcal{B}}$
for two   sets $\mathcal{A},\mathcal{B}\subset\bbR^{n}$ such
that $\bbE[f(X^{n})]=a$ and $\bbE[g(Y^{n})]=b$ in Theorem~\ref{thm:functional-borell}.
Observe that $\Lambda_{\rho}(0,0)=\Lambda_{\rho}(1,0)=\Lambda_{\rho}(0,1)=0$,
and $\Lambda_{\rho}(1,1)=1$. Therefore, $\Lambda_{\rho}(f(X^{n}),g(Y^{n}))=\bone_{\mathcal{A}\times\mathcal{B}}(X^{n},Y^{n})$,
which implies that $\overline{\Gamma}^{(n)}(a,b)\le\Lambda_{\rho}(a,b)$.
Obviously,  by definition, $\overline{\Gamma}^{(n)}(a,b)\ge\Lambda_{\rho}(a,b)$ follows
by setting $\mathcal{A}$ and $\mathcal{B}$ to be two parallel halfspaces.
Hence, $\overline{\Gamma}^{(n)}(a,b)=\Lambda_{\rho}(a,b)$. 

We now argue that Theorem~\ref{thm:borell_nicd} implies Theorem~\ref{thm:functional-borell}. For this purpose, 
given $f,g:\bbR^{n}\to[0,1]$, define $\mathcal{A}$ and $\mathcal{B} $ (subsets of $\bbR^{n+1}$)
to be  the respective hypographs\footnote{The {\em hypograph} $\mathrm{hyp}(h)$ of a function $h:\calX\to\bbR$ is the set of points of $\calX\times\bbR$ lying on or below its graph, i.e., $\mathrm{hyp}(h) := \{(x,r)\in\calX\times\bbR: r\le h(x)\}.$ } of $\Phi^{-1}\circ f: \bbR^n\to\bbR$ and $\Phi^{-1}\circ g : \bbR^n\to\bbR$, where recall that $\Phi(\cdot)$ is the cumulative distribution function
of the standard Gaussian and $\Phi^{-1} : (0,1) \to \bbR$ is its inverse. It can be readily  checked that 
\begin{align}
 \bbE\big[\Lambda_{\rho}(f(X^{n}),g(Y^{n}))\big] &=\Pr\big(X_{n+1}\!\le\!\Phi^{-1}\circ f(X^{n}),Y_{n+1}\!\le\!\Phi^{-1}\circ g(Y^{n})\big) \nn\\*
 & =\pi_{XY}^{n+1}(\mathcal{A}\times\mathcal{B}),
\end{align}
where $(X^{n+1},Y^{n+1})\sim \pi_{XY}^{n+1}$. On the other hand, $\bbE [f(X^{n})]=\pi_{X}^{n+1}(\mathcal{A})$
and $\bbE [g(Y^{n})]=\pi_{Y}^{n+1}(\mathcal{B})$, and hence, the right-hand side of \eqref{eq:functional-borell}
satisfies 
\begin{align}
\Lambda_{\rho}\big(\bbE[f(X^{n})],\bbE[g(Y^{n})]\big)=\Lambda_{\rho}(\pi_{X}^{n+1}\big(\mathcal{A}),\pi_{Y}^{n+1}(\mathcal{B})\big).
\end{align}
Thus, Theorem~\ref{thm:borell_nicd} in $n+1$ dimensions implies Theorem~\ref{thm:functional-borell}
in $n$ dimensions.

Hence, to prove  Theorem~\ref{thm:borell_nicd}, it suffices to prove Theorem~\ref{thm:functional-borell}. In their proof of Theorem~\ref{thm:functional-borell},  Mossel
and Neeman~\cite{mossel2015robust} first
constructed an   Ornstein--Uhlenbeck semigroup, then defined $R_t$, 
an auxiliary function for this semigroup  that connects the two sides
of \eqref{eq:functional-borell} as limiting cases. Lastly, they showed
that $R_t$ is monotone. A similar idea was also used in \citet{bakry1996levy}.

\begin{proof}[Proof Sketch of Theorem \ref{thm:functional-borell}]  
For every $t\ge0$, define the  operator $P_{t}$ that acts on functions $f:\bbR^n\to [0,1]$ as
\begin{equation}
(P_{t}f)(x^{n}):=\int_{\bbR^{n}}f\big(\rme^{-t}\,x^{n}+\sqrt{1-\rme^{-2t}}\,y^{n}\big)\,\mathrm{d}\pi_{Y}^{n}(y^{n}).\label{eq:P_t}
\end{equation}
This operator is known as the \emph{Ornstein--Uhlenbeck semigroup operator}.
Note that $P_{t}f\!\to\! f$ pointwise  as $t\!\to\! 0$ and $P_{t}f\!\to\!\bbE [f]$ pointwise as~$t\!\to\!\infty$.

Let $f_{t}:=P_{t}f$ and $g_{t}:=P_{t}g$, and consider the quantity 
\begin{equation}
R_{t}:=\bbE\big[\Lambda_{\rho}(f_{t}(X^{n}),g_{t}(Y^{n}))\big].\label{eq:R_t-def}
\end{equation}
As $t\to0$, $R_{t}$ converges to the left-hand side of~\eqref{eq:functional-borell};
as $t\to\infty$, $R_{t}$ converges to the right-hand side of~\eqref{eq:functional-borell}.
Hence, to establish Theorem~\ref{thm:functional-borell}, it suffices
to prove that ${\mathrm{d}R_{t}}/{\mathrm{d}t}\ge0$ for all $t>0$.
This point can be checked by careful calculations, as shown in the
following lemma due to \citet{mossel2015robust}.

\begin{lemma}\label{lem:diff-R_t} 
The function $t\in [0,\infty)\mapsto R_t$, defined in \eqref{eq:R_t-def}, satisfies
\begin{equation}
\frac{\mathrm{d}R_{t}}{\mathrm{d}t}=\frac{\rho}{2\pi\sqrt{1-\rho^{2}}}\ \bbE\bigg[\exp\bigg(-\frac{v_{t}^{2}+w_{t}^{2}-2\rho v_{t}w_{t}}{2(1-\rho^{2})}\bigg)\bigg]\  \left\|\grad v_{t}-\grad w_{t}\right\|^{2},
\end{equation}
where $v_{t}:=\Phi^{-1}\circ f_{t}: \bbR^n\to\bbR$, $w_{t}:=\Phi^{-1}\circ g_{t}: \bbR^n\to\bbR$,
and $\grad$ denotes the gradient operator. Hence, the  derivative of $R_t$ for $t\ge0$ is nonnegative. \end{lemma}  

This completes the proof sketch of Theorem \ref{thm:functional-borell}. \end{proof}
\global\long\def\cvx{\bbL}%
\global\long\def\cve{\bbU}%
 
%\global\long\def\oGamma{{\Gamma}}%
%\global\long\def\uGamma{{\Gamma}}%

\chapter{$q$-Stability} \label{ch:Stability} %This is based on our paper ``Wyner's common information under R\'enyi divergence measures''. More than just a generalization of R\'enyi's CI, it serves as a convenient bridge between Wyner's CI and Exact CI. Will discuss the TV (variational distance) approximate version and its strong converse together with unnormalized and normalized  versions of the the R\'enyi divergence constraints.

In Section~\ref{ch:NICD}, we discussed the $2$-user
NICD problem. In this section, we extend the NICD problem to the multi-user
case, and consider two versions of these extensions. In the {\em symmetric}
version, we maximize the agreement probability of the random bits
generated individually by the users. In   the {\em asymmetric} version,
we maximize the joint probability that all the random bits are equal to $1$.
These two maximization problems are equivalent in the $2$-user setting (see discussion following~\eqref{eqn:jp_reverse}),
but  are {\em not} equivalent in the setting involving $3$ or more users. This distinction results in the upcoming set of problems being significantly more challenging,  but they provide more insight into the NICD and related problems. 

%These extensions of the NICD problem to the multi-user
%case are non-trivial, since in the multi-user setting, determining
%the maximum agreement probability and the maximum joint probability
%are significantly more challenging compared to the $2$-user case. However, they provide more insight into the NICD and related problems. 

Indeed, these extensions 
% provide
% an insightful perspective on the NICD and related problems.
%Specifically, such extensions 
have inspired researchers to define a more
general concept known as the {\em $q$-stability}. This is done by generalizing
the number of users in the NICD problem from an integer $k$ to an arbitrary real number $q\ge 1$. %Thus, the $q$-stability is a family of quantities
%indexed by the parameter $q$. 
 The \emph{max $q$-stability} problem
concerns the identification of Boolean functions that most ``stable''---measured in terms of the $q$-stability---under the action
of a noise operator. %Here, stability is measured in terms of the $q$-stability.
Such a problem not only significantly generalizes the $2$-user NICD problem to
a version parametrized by an arbitrary real number $q\ge1$, but   more importantly, it seamlessly connects to several
interesting contemporary conjectures in information theory and discrete probability, including the  Mossel--O'Donnell conjecture \cite{mossel2005coin}, the Courtade--Kumar
conjecture \cite{courtade2014boolean}, and the Li--M\'edard conjecture
\cite{li2021boolean}. Hence, the study of $q$-stability provides us a comprehensive and unified understanding
of these  conjectures. 

Similar to Section~\ref{ch:NICD}, in this section, we focus mainly on the doubly symmetric binary
source (DSBS) with  correlation coefficient $\rho\in(-1,1)$.
In Section~\ref{sec:nicd_multi}, we formulate the multi-user NICD
problem for the DSBS. We define the  asymmetric and symmetric forward joint
probabilities, and also generalize them to various max $q$-stabilities by
relaxing the number of users to an arbitrary real number $q\ge1$.
In Section~\ref{sec:conj}, we introduce several important conjectures
concerning the max $q$-stability problem for the case in that the  Boolean functions in question are {\em balanced}. These include the  Mossel--O'Donnell, Courtade--Kumar, and Li--M\'edard
conjectures. %We also unify them into two conjectures. 
In Section~\ref{sec:extreme},
we describe resolutions for the conjectures in the extreme cases in which  the correlation
coefficient $\rho\downarrow 0$ or $\rho\uparrow 1$. Interestingly, in these two extreme
cases, the conjectures are characterized by  the classic {\em edge-isoperimetric inequality}
and the {\em maximal degree-$1$ Fourier weight}. Hence, related concepts in
discrete geometry, e.g., influences and edge boundaries, will also be introduced. In
Section~\ref{sec:balanced}, we describe recent progress on partial resolutions of these conjectures. In Section~\ref{sec:MD-and-LD},
we introduce the solutions to the max $q$-stability problem in the moderate
 and large deviations regimes. Finally, in Section~\ref{sec:stab_arb},
we discuss known results on the max $q$-stability problem for sources beyond the DSBS including  bivariate Gaussian sources.

\section{The Multi-User NICD Problem and $q$-Stability}
\label{sec:nicd_multi}

\subsection{Formulation}
\label{sec:form_q_stability}

\begin{figure}
\centering \tikzset{every picture/.style={line width=0.75pt}}
%set default line width to 0.75pt        
\begin{tikzpicture}[x=0.75pt,y=0.75pt,yscale=-1,xscale=1,scale=0.9] %uncomment if require: \path (0,2761); %set diagram left start at 0, and has height of 2761
%Shape: Circle [id:dp33671381330382033] 
\draw   (1204.01,126) .. controls (1204.01,119.93) and (1208.93,115.01) .. (1215,115.01) .. controls (1221.08,115.01) and (1226,119.93) .. (1226,126) .. controls (1226,132.08) and (1221.08,137) .. (1215,137) .. controls (1208.93,137) and (1204.01,132.08) .. (1204.01,126) -- cycle ; %Shape: Circle [id:dp5414765473126235] 
\draw   (1279.01,199) .. controls (1279.01,192.93) and (1283.93,188.01) .. (1290,188.01) .. controls (1296.08,188.01) and (1301,192.93) .. (1301,199) .. controls (1301,205.08) and (1296.08,210) .. (1290,210) .. controls (1283.93,210) and (1279.01,205.08) .. (1279.01,199) -- cycle ; %Shape: Circle [id:dp19308041676604537] 
\draw   (1180.01,199) .. controls (1180.01,192.93) and (1184.93,188.01) .. (1191,188.01) .. controls (1197.08,188.01) and (1202,192.93) .. (1202,199) .. controls (1202,205.08) and (1197.08,210) .. (1191,210) .. controls (1184.93,210) and (1180.01,205.08) .. (1180.01,199) -- cycle ; %Shape: Circle [id:dp9426200737876957] 
\draw   (1130.01,199) .. controls (1130.01,192.93) and (1134.93,188.01) .. (1141,188.01) .. controls (1147.08,188.01) and (1152,192.93) .. (1152,199) .. controls (1152,205.08) and (1147.08,210) .. (1141,210) .. controls (1134.93,210) and (1130.01,205.08) .. (1130.01,199) -- cycle ; %Straight Lines [id:da44557460499974866] 
\draw    (1215,137) -- (1142.65,186.87) ;
\draw [shift={(1141,188.01)}, rotate = 325.41999999999996] [color={rgb, 255:red, 0; green, 0; blue, 0 }  ][line width=0.75]    (10.93,-3.29) .. controls (6.95,-1.4) and (3.31,-0.3) .. (0,0) .. controls (3.31,0.3) and (6.95,1.4) .. (10.93,3.29)   ; %Straight Lines [id:da032694881316852165] 
\draw    (1215,137) -- (1191.86,186.2) ;
\draw [shift={(1191,188.01)}, rotate = 295.2] [color={rgb, 255:red, 0; green, 0; blue, 0 }  ][line width=0.75]    (10.93,-3.29) .. controls (6.95,-1.4) and (3.31,-0.3) .. (0,0) .. controls (3.31,0.3) and (6.95,1.4) .. (10.93,3.29)   ; %Straight Lines [id:da614536631441102] 
\draw    (1215,137) -- (1288.35,186.88) ;
\draw [shift={(1290,188.01)}, rotate = 214.22] [color={rgb, 255:red, 0; green, 0; blue, 0 }  ][line width=0.75]    (10.93,-3.29) .. controls (6.95,-1.4) and (3.31,-0.3) .. (0,0) .. controls (3.31,0.3) and (6.95,1.4) .. (10.93,3.29)   ; %Straight Lines [id:da21032126287743735] 
\draw    (1141,210) -- (1141,265.2) ;
\draw [shift={(1141,267.2)}, rotate = 270] [color={rgb, 255:red, 0; green, 0; blue, 0 }  ][line width=0.75]    (10.93,-3.29) .. controls (6.95,-1.4) and (3.31,-0.3) .. (0,0) .. controls (3.31,0.3) and (6.95,1.4) .. (10.93,3.29)   ; %Straight Lines [id:da7961240277132049] 
\draw    (1191,210) -- (1191,265.2) ;
\draw [shift={(1191,267.2)}, rotate = 270] [color={rgb, 255:red, 0; green, 0; blue, 0 }  ][line width=0.75]    (10.93,-3.29) .. controls (6.95,-1.4) and (3.31,-0.3) .. (0,0) .. controls (3.31,0.3) and (6.95,1.4) .. (10.93,3.29)   ; %Straight Lines [id:da5976635697758215] 
\draw    (1290,210) -- (1290,265.2) ;
\draw [shift={(1290,267.2)}, rotate = 270] [color={rgb, 255:red, 0; green, 0; blue, 0 }  ][line width=0.75]    (10.93,-3.29) .. controls (6.95,-1.4) and (3.31,-0.3) .. (0,0) .. controls (3.31,0.3) and (6.95,1.4) .. (10.93,3.29)   ;
% Text Node
\draw (1203,91) node [anchor=north west][inner sep=0.75pt]    {$Y^n\sim \mathrm{Bern}(\frac{1}{2})^{n}$}; % Text Node
\draw (1253,139) node [anchor=north west][inner sep=0.75pt]    {$\text{Independent} \,\, \mathrm{BSC}\Big( \frac{1-\rho}{2}\Big)^{n}$}; % Text Node
\draw (1110,210) node [anchor=north west][inner sep=0.75pt]    {$X^n_{1}$}; % Text Node
\draw (1160,210) node [anchor=north west][inner sep=0.75pt]    {$X^n_{2}$}; % Text Node
\draw (1260,210) node [anchor=north west][inner sep=0.75pt]    {$X^n_{k}$}; % Text Node
\draw (1129,272) node [anchor=north west][inner sep=0.75pt]    {$U_{1}$}; % Text Node
\draw (1179,272) node [anchor=north west][inner sep=0.75pt]    {$U_{2}$}; % Text Node
\draw (1280,272) node [anchor=north west][inner sep=0.75pt]    {$U_{k}$}; % Text Node
\draw (1222,185) node [anchor=north west][inner sep=0.75pt]   [align=left] {......}; % Text Node
\draw (1227,267) node [anchor=north west][inner sep=0.75pt]   [align=left] {......}; % Text Node
\draw (1220,335) node [anchor=north west][inner sep=0.75pt]    {$\mathrm{max} \ \Pr( U_{1} =U_{2} =\ldots =U_{k} =1)$}; % Text Node
\draw (1091.78,298.67) node [anchor=north west][inner sep=0.75pt]    {$\sim \mathrm{Bern}( a)$}; % Text Node
\draw (1175.78,297.67) node [anchor=north west][inner sep=0.75pt]    {$\sim \mathrm{Bern}( a)$}; % Text Node
\draw (1270.78,297.67) node [anchor=north west][inner sep=0.75pt]    {$\sim \mathrm{Bern}( a)$}; % Text Node
\draw (1062.73,334.91) node [anchor=north west][inner sep=0.75pt]   [align=left] {Asymmetric Version:}; % Text Node
\draw (1220,365) node [anchor=north west][inner sep=0.75pt]    {$\mathrm{max} \ \Pr( U_{1} =U_{2} =\ldots=U_{k})$}; % Text Node
\draw (1072.73,365.91) node [anchor=north west][inner sep=0.75pt]   [align=left] {Symmetric Version:};
\end{tikzpicture} \caption{The Non-Interactive Correlation Distillation problem with $k$ users}
\label{fig:NICD-with-k} 
\end{figure}

\enlargethispage{\baselineskip}
Before formally introducing the $k$-user NICD problem, we first introduce
a class of  Boolean functions, known as \emph{majority functions}.
For an odd number $m \in [ n]$, let $\mathrm{Maj}_{m}:\{0,1\}^{n}\to\{0,1\}$
be the majority function on the first $m$ bits which is given by  $\mathrm{Maj}_{m}(x^{n}):=\bone\left\{ \sum_{i=1}^{m}x_{i}\ge m/2\right\} $
for each $x^{n}\in\{0,1\}^{n}$. Then, clearly, $\mathrm{Maj}_{1}$ is
a dictator function, and $\mathrm{Maj}_{n}$ is the indicator of the Hamming
ball $\mathbb{B}_{n/2}(1^{n})$ (as introduced in Section~\ref{sec:ball}). Hence,   majority functions are
generalizations of dictator functions and indicators of Hamming balls. Furthermore,
we say that a Boolean function $f:\{0,1\}^{n}\to\{0,1\}$ is \emph{anti-symmetric} (or \emph{odd})
if 
\begin{equation}
f(x^{n})+f(\barx^{n})=1 \quad\mbox{for all} \;\, x^{n}\in\{0,1\}^{n}, \label{eqn:asymmetric}
\end{equation}
where $\barx^n := 1^n-x^n$ is the bitwise negation of $x^n$. 
Equivalently, for an anti-symmetric Boolean
function $f$, $\supp(f)^{\rmc}=1^{n}-\supp(f)$, where $1^{n}-\calA:=\{1^{n}-x^{n}:x^{n}\in \calA\}$
for any set ${\cal A}\subset\{0,1\}^{n}$. By definition, for any
odd $m\in[n]$, the majority function $\mathrm{Maj}_{m}$ is
anti-symmetric. 

\enlargethispage{\baselineskip}
The $k$-user NICD problem, which is illustrated in Fig.~\ref{fig:NICD-with-k},  was   investigated by \citet{mossel2005coin}
for the symmetric version, and by \citet{li2021boolean} for the asymmetric
version. 
There are $k$ correlated memoryless sources $X_{1},X_{2},\ldots,X_{k}$   generated from a common
memoryless Bernoulli source $Y\sim\mathrm{Bern}(\frac{1}{2})$
through~$k$ independent binary symmetric channels  with crossover
probability $p=(1-\rho)/2$; hence, $0<\rho<1$ is the correlation coefficient between $X_{j,i}$ and $Y_i$ for all $j\in [k]$ and $i\in [n]$. 
A Boolean function $f_{i} : \{0,1\}^n\to\{0,1\}$
is applied to each source sequence\footnote{Here, we use the notation $X_{i}^{n}$ to denote the $i^{\mathrm{th}}$
(out of $k$) length-$n$ correlated source sequences instead of the random
vector $(X_{i},X_{i+1},\ldots,X_{n})$.} $X_{i}^{n}$ to generate a random bit $U_{i}=f_{i}(X_{i}^{n})$.
\begin{definition}
For a dyadic rational $a=M/2^{n} \in [0,1]$ (in which $M\in\{0,1,\ldots,2^{n}\}$),
define the {\em forward joint probability at mean $a$} as 
\begin{align}
\hspace{-.2in}\Gamma_{\rho}^{(k)}(a) & :=\max_{\substack{\textrm{Boolean }f_{i},1\le i\le k:\\
\Pr(f_{i}(X_{i}^{n})=1)=a
}
}\Pr\big(f_{1}(X_{1}^{n})=\ldots=f_{k}(X_{k}^{n})=1\big).\label{eq:for_jp}
\end{align}
\end{definition}
Since  we do not consider the reverse counterpart
of the forward joint probability in \eqref{eq:for_jp} throughout this section, we omit the overline on~$\Gamma$ (cf.\ the notation $\overline{\Gamma}^{(n)}$ used for the forward joint
probability in~\eqref{eqn:jp_forward}) but we make the number of users $k$ and the correlation coefficient $\rho$ explicit in the notation. To avoid notational overload, we also omit the superscript $n$ that indexes the blocklength. Table~\ref{tab:notation_stability} lists  commonly encountered operational quantities in this section.

\begin{table}[!ht]
\caption{Table of commonly used operational quantities in this section}
\label{tab:notation_stability}
\centering
\begin{tabular}{|c|c|c|}
\hline
Name          & Symbol              & Definition(s)             \\ \hline\hline
Forward joint probability  at $a$    & $\Gamma_{\rho}^{(k)}(a)$  & \eqref{eq:for_jp}, \eqref{eqn:asymp_max_k} \\ \hline
$q$-stability of $f$ & $\bS_\rho^{(q)}[f] $ & \eqref{eq:q-stability} \\ \hline
Asymmetric max $q$-stability    at $a$          &       $\Gamma_{\rho}^{(q)}(a)$              &  \eqref{eqn:asymp_max_q_0}\\ \hline
Symmetric max $q$-stability  at $a$    &$ \breve{\Gamma}_{\rho}^{(q)}(a)$ &  \eqref{eqn:sym_max_q} \\ \hline
Symmetric $q$-stability of $f$ & $\breve{\bS}_\rho^{(q)}[f] $ & \eqref{eqn:sym_q_stab_f} \\ \hline
Symmetric forward joint probability  at $a$    & $\breve{\Gamma}_{\rho}^{(k)}(a)$ &  \eqref{eqn:sym_max_k} \\ \hline
$\Phi$-stability  of $f$    & $\bS_\rho^{(\Phi)}[f] $ &  \eqref{eqn:Phi_stab} \\ \hline
$\Phi$-asymmetric max $q$-stability  at $a$    & $\Pi_\rho^{(q)}(a) $ &  \eqref{eq:-3}\\ \hline
$\Phi$-symmetric max $q$-stability at $a$     & $\breve{\Pi}_\rho^{(q)}(a) $ &  \eqref{eq:-4}\\ \hline
LD exponent & ${\Upsilon}_{q,\mathrm{LD}}^{(n)}(\alpha) $  & \eqref{eq:NICD-FLD} \\ \hline
MD exponent & ${\Upsilon}_{q,\mathrm{MD}}^{(n)}(\alpha) $  &  \eqref{eq:NICD-RMD} \\ \hline
\end{tabular}
\end{table}

%Since the DSBS satisfies the property that if $X_1 - Y -$ 
It clearly holds that every pair   $(X_{j}^{n},X_{\ell}^{n})$ with  $j\neq \ell$  is  a source sequence generated by a   DSBS with correlation coefficient $\rho^2$ (because $X_j -Y - X_\ell$). This implies that $\Gamma_{\rho}^{(2)}(a)$ corresponds to the forward joint probability  defined in~\eqref{eqn:jp_forward} for the DSBS with correlation coefficient $\rho^2$. 
% $\Gamma_{\rho}^{(2)}(a)$ $\overline{\Gamma}^{(n)}$
%
%As implied by the following proposition, the $k$ functions that
%attain the forward joint probability are necessarily identical. This point can be  easily proven by using
%Mossel and O'Donnell's proof idea of  Proposition 3.3 of \cite{mossel2005coin}, and it  
%was also proven by Li and M\'edard
%\cite{li2021boolean} using a different method. 

Due to the apparent symmetry of the problem, one may naturally wonder whether the $k$ functions $f_1,\ldots, f_k$ that attain the forward joint probability are necessarily identical. This is positively confirmed in the following proposition which can be proved using either the idea  in  \cite[Proposition~3]{mossel2005coin} or \cite{li2021boolean}. We provide a self-contained proof.

\begin{proposition} \label{prop:identity2-1} Let $\calF $ be any class of Boolean functions. Let  $k,n\ge 1$ and $\rho\in (0,1)$.
%Let $\mathcal{F}$ be any class of Boolean functions on $\{0,1\}^{n}$.
%Subject to the restriction that $f_{1},\ldots,f_{k}\in\mathcal{F}$,
Every tuple of functions $(f_{1},\ldots,f_{k}) \in \calF^k$ that maximizes $\Pr\big(f_{1}(X_{1}^{n})=\ldots=f_{k}(X_{k}^{n})=1\big)$
satisfies $f_{1}=f_{2}=\ldots=f_{k}$. \end{proposition}

\begin{proof} Since $\calF$ is finite, we may enumerate its elements as  $\mathcal{F}=\{g_{j}:j\in[M]\}$ where $M\ge2$ to avoid the trivial case in which $M=1$. 
%the trivial cases, we assume $M\ge2$. 
Suppose that among the $k$
users, $g_{j}$ is used by $kp_{j}$ of them. Then clearly, $\{p_{j}:j\in[M]\}$
forms a distribution on $\mathcal{F}$ or, isomorphically,  on $[M]$. On the other hand, the
joint probability induced by this scheme is 
\begin{align}
\hspace{-.2in} \Pr\big(f_{1}(X_{1}^{n})=\ldots=f_{k}(X_{k}^{n})=1\big)=\mathbb{E}_{Y^{n}}\bigg[\prod_{j=1}^{M}(T_{\rho}g_{j}(Y^{n}))^{kp_{j}}\bigg],\label{eq:-5-1}
\end{align}
where $T_\rho$ is the noise operator defined in~\eqref{eqn:noise_op}. 
On the other hand, given $a_{1},\ldots,a_{M}>0$, the map $(p_{1},\ldots,p_{M}) \in \calP([M]) \mapsto\prod_{j=1}^{M}a_{j}^{p_{j}}$
is  convex. Hence, the expression in~\eqref{eq:-5-1} is
{\em convex} in $(p_{1},\ldots,p_{M})$. {\em Maximizing}~\eqref{eq:-5-1} over $(p_{1},\ldots,p_{M})$
on the probability simplex $\calP([M])$, we see that the maximum is attained
at a vertex of $\calP([M])$. This in turn implies that
the maximum of $\Pr\big(f_{1}(X_{1}^{n})=\ldots=f_{k}(X_{k}^{n})=1\big)$
over all $(f_{1},\ldots,f_{k})\in \mathcal{F}^k$  is attained by some $(f_{1},\ldots,f_{k})$
such that $f_{1}=f_{2}=\ldots=f_{k}$. The necessity of the identity
of Boolean functions in attaining this maximum can also be  verified;
see \citet{mossel2005coin}.  \end{proof}

By particularizing $\mathcal{F}$ in Proposition \ref{prop:identity2} to be the set of 
Boolean functions with mean $a$, any tuple of $k$ functions $(f_{1},\ldots,f_{k})$
that attains the forward joint probability necessarily satisfies $f_{1}=f_{2}=\ldots=f_{k}$.
This observation draws our attention to the following related quantity known as the 
{\em $q$-stability} \cite{eldan2015two,li2021boolean}. 
\begin{definition} \label{def:qstab}
For any $q\in[1,\infty)$ and a Boolean function $f: \{0,1\}^n\to\{0,1\}$, the {\em $q$-stability}
of $f$ is defined as 
\begin{align}
\mathbf{S}^{(q)}_{\rho}[f]:=\mathbb{E}_{Y^{n}}\big[(T_{\rho}f(Y^{n}))^{q}\big].\label{eq:q-stability}
\end{align}
%where the noise operator $T_{\rho}$ is defined in \eqref{eqn:noise_op}.
\end{definition}
%Recall that for a Boolean function $f$, $T_{\rho}f(y^{n})=\pi_{X|Y}^{n}(\calA|y^{n})$ with $\calA$
%being the support of $f$; see~\eqref{eqn:noise_op}. 
For $q=2$, the $q$-stability reduces to the {\em correlation} $\bbE[f(X^n)f(\hat{X}^n)]$, or equivalently, the {\em joint probability} $\Pr(f(X^n)=f(\hat{X}^n)=1)$, where $(X^n,\hat{X}^n)$ is a source sequence of the DSBS with correlation coefficient~$\rho^2$. Hence,    $\bbE[f(X^n)f(\hat{X}^n)]$ is  termed the \emph{noise stability} of the Boolean function $f$ with parameter $\rho^2$, which is denoted as $\mathbf{S}_{\rho^2}[f]$. As mentioned in the discussion following~\eqref{eq:for_jp}, 
\begin{align}
\mathbf{S}_{\rho^2}[f]=\mathbf{S}^{(2)}_{\rho}[f]. \label{eq:stability_stability}
\end{align}

To better understand the concept of the
$q$-stability, we now compute it for two  functions.

\begin{example}
For  the dictator
function $\mathrm{Maj}_{1}$, %can be verified that 
\begin{align}
\mathbf{S}^{(q)}_{\rho}[\mathrm{Maj}_{1}] & =\mathbf{S}^{(q)}_{\rho}[X_{1}]\\
 & =\mathbb{E}_{Y_{1}}\big[(\mathbb{E}[X_{1}|Y_{1}])^{q}\big]\\
 & =\frac{1}{2}\Big(\frac{1+\rho}{2}\Big)^{q}+\frac{1}{2}\Big(\frac{1-\rho}{2}\Big)^{q}.
\end{align}
\end{example}
\begin{example}
For the  indicator of the Hamming ball $\mathrm{Maj}_{n}$,
%For $\mathrm{Maj}_{1}$, it can be  verified that 
% For $\mathrm{Maj}_{n}$, 
   it is not easy to derive the exact value of
 its $q$-stability for each dimension $n \in\bbN$. However, one can   determine the limit of the
$q$-stability of $\mathrm{Maj}_{n}$ as $n\to \infty$. 
By the (multivariate) central limit theorem,  %we have the following convergence in distribution
\begin{equation} 
\frac{2}{\sqrt{n}}\bigg(\sum_{i=1}^{n}\begin{bmatrix}
X_i	\\ Y_i
\end{bmatrix}-\frac{n}{2}\begin{bmatrix}
1 \\ 1
\end{bmatrix}\bigg) \stackrel{\mathrm{d}}{\longrightarrow}\calN\bigg( \begin{bmatrix}
0 \\ 0 
\end{bmatrix},\bK \bigg),
\end{equation}
%converges in distribution to a zero-mean bivariate Gaussian with 
%% a pair of Gaussian random variables $(U,V)$
%%with mean $(0,0$) and 
where the covariance matrix $\bK$ is  defined 
in~\eqref{eqn:cov_mat_rho}.  Define the \emph{Gaussian $q$-stability
function} $\Lambda_{\rho}^{(q)} : [0,1]\to [0,1]$ as 
\begin{align}
\hspace{-.25in}\Lambda_{\rho}^{(q)}(a)&:=\mathbb{E}\big[\Pr(U\le\Phi^{-1}(a)|V)^{q}\big]  \label{eq:GaussianStabFunc}  = \bbE \bigg[\Phi\Big( \frac{ \Phi^{-1}(a)-\rho V } { \sqrt{ 1-\rho^2}} \Big)^q \bigg] ,
\end{align}
where $(U,V)$ is a pair of jointly Gaussian random variables with zero mean and covariance matrix $\bK$. 
%and $\Phi$ is the CDF of the standard univariate Gaussian. 
Therefore, for every $(\rho, q) \in (-1,1)\times [0,1]$, 
%the argument
%above justifies that given $(\rho,q)$,  
the limit of the $q$-stability of $\mathrm{Maj}_{n}$
is 
\begin{equation}
\lim_{n\to\infty}\mathbf{S}^{(q)}_{\rho}[\mathrm{Maj}_{n}]=\Lambda_{\rho}^{(q)}(1/2).\label{eq:StabMAJn}
\end{equation}
%For completeness, the proof of \eqref{eq:StabMAJn} will be given
%in Section xx. 
\end{example}

We   relate the $q$-stability to the NICD problem by observing
that for an integer $k$,  the forward joint probability in \eqref{eq:for_jp} can be rewritten as
\begin{align}
\Gamma_{\rho}^{(k)}(a) & =\max_{\textrm{Boolean }f:\mathbb{E}[f(X^{n})]=a}\mathbf{S}^{(k)}_{\rho}[f].\label{eqn:asymp_max_k}
\end{align}
Hence, it is natural to term $\Gamma_{\rho}^{(k)}(a)$ as the
{\em asymmetric max $k$-stability at mean~$a$}. If we replace
the integer~$k$ in~\eqref{eqn:asymp_max_k} with an arbitrary
real number $q\in[1,\infty)$, we can define the asymmetric
max $q$-stability at mean $a$~\cite{eldan2015two,li2021boolean}
as follows.
\begin{definition} \label{def:asym_max}
For $q\in[1,\infty)$,  the {\em asymmetric
max $q$-stability at mean~$a$} is defined as 
\begin{align}
\Gamma_{\rho}^{(q)}(a)  & :=\max_{\textrm{Boolean }f:\mathbb{E}[f(X^{n})]=a}\mathbf{S}^{(q)}_{\rho}[f]\label{eqn:asymp_max_q_0} \\
 & =\max_{\calA\subset\{0,1\}^{n}:\pi_{X}^{n}(\calA)=a}\mathbb{E}_{Y^{n}}\big[\pi_{X|Y}^{n}(\calA|Y^{n})^{q}\big].\label{eqn:asymp_max_q}
\end{align}
The equality in \eqref{eqn:asymp_max_q} follows because for a Boolean function $f$, $T_{\rho}f(y^{n})=\pi_{X|Y}^{n}(\calA|y^{n})$ with $\calA$ being the support of $f$; see~\eqref{eqn:noise_op}. 
\end{definition}
A few remarks concerning this definition are in order. First, for fixed  $a\in[0,1]$, the function $q \in [1,\infty)\mapsto \Gamma_{\rho}^{(q)}(a)$ is
nonincreasing. % (because the probability $\pi_{X|Y}^{n}(\calA|Y^{n})\le 1$). 
 Second, given $q\ge1$ and two correlation coefficients
$0\le\rho\le\hat{\rho}\le1$, for any $y^n\in \{0,1\}^n$,  we have 
\begin{align}
(T_{\rho} f)^{q}(y^{n}) &= (T_{\rho/\hat{\rho}}T_{\hat{\rho}}f)^{q}(y^{n}) \le T_{\rho/\hat{\rho}}(T_{\hat{\rho}}f)^{q}(y^{n}), \label{eqn:jens}
\end{align}
where the equality follows from the fact that $T_{\rho_1\rho_2}=T_{\rho_1}T_{\rho_2}$ for all $\rho_1,\rho_2 \in [0,1]$, and  the inequality follows by Jensen's inequality ($x\mapsto x^q$ is convex for $q\ge1$). From~\eqref{eqn:jens}, we obtain 
\begin{align}
\mathbf{S}^{(q)}_{\rho}[f]  & \le\mathbb{E}_{Y^{n}}\big[T_{\rho/\hat{\rho}}(T_{\hat{\rho}}f)^{q}(Y^{n})\big]\label{eqn:jens0}\\
 & =\mathbb{E}_{Z^{n}}\big[(T_{\hat{\rho}}f)^{q}(Z^{n})\big] \qquad (Z^n\sim\mathrm{Unif}\{0,1\}^n)\label{eqn:construct_bsc}\\
 & =\mathbf{S}^{(q)}_{\hat{\rho}}[f],\label{eqn:construct_bsc1}
\end{align}
%\begin{align}
%\mathbf{S}^{(q)}_{\rho}[f] & =\mathbb{E}_{Y^{n}}\big[(T_{\rho/\hat{\rho}}T_{\hat{\rho}}f(Y^{n}))^{q}\big]\label{eqn:jens0}\\
% & \le\mathbb{E}_{Y^{n}}\big[T_{\rho/\hat{\rho}}(T_{\hat{\rho}}f)^{q}(Y^{n})\big]\label{eqn:jens}\\
% & =\mathbb{E}_{Z^{n}}\big[(T_{\hat{\rho}}f(Z^{n}))^{q}\big] \qquad (Z^n\sim\mathrm{Unif}\{0,1\}^n)\label{eqn:construct_bsc}\\
% & =\mathbf{S}^{(q)}_{\hat{\rho}}[f],\label{eqn:construct_bsc1}
%\end{align}
where~\eqref{eqn:construct_bsc} follows because if the input to a binary symmetric channel is uniform, so is its output.\footnote{The block of inequalities  in \eqref{eqn:jens0}--\eqref{eqn:construct_bsc1} can also be re-interpreted as follows. Given a DSBS $(X,Y)$ with correlation coefficient $\rho \in [0,1]$, we can construct a Markov chain $X-Z-Y$ with correlation coefficient between $X$ and $Z$ being $\hat{\rho}\in [0,\rho]$ such that for any $q\ge1$,  we have 
$
\bbE [\bbE[f(X)|  Y]^q ] \le \bbE[ \bbE[\bbE[f(X)|  Z]^q] |  Y] ] = \bbE [\bbE[f(X)|  Z]^q]
$.}
Hence, given $q\ge 1$ and $a\in[0,1]$, the function $\rho\in [0,1]\mapsto \Gamma_{\rho}^{(q)}(a)$ is nondecreasing.  Finally,
if $\rho=1$ (i.e., there is no noise), then  $\Gamma_{1}^{(q)}(a)=a$. If instead
$\rho=0$, then $\Gamma_{0}^{(q)} (a)=a^{q}$.

%The max $q$-stability
%problem (solution of the equivalent optimization problems in Definition~\ref{def:asym_max})  

To find the solution to the asymmetric  max $q$-stability problem (i.e., the  equivalent optimization problems in Definition~\ref{def:asym_max}), we have to  identify  Boolean functions that are the ``most
stable'' under the action of the noise operator~$T_\rho$, with the stability being measured by the $q$-stability $\mathbf{S}^{(q)}_{\rho}[f]$. 

Analogously to the asymmetric max $q$-stability at mean $a$, one
can define a {\em symmetric} version of this stability notion by
maximizing the sum of the $q$-stabilities of $f$ and $1-f$. 
\begin{definition}
For $q>1$, the {\em symmetric max $q$-stability
at mean~$a$} is\footnote{We use the breve accent on symbols  that signify {\em symmetric} quantities (e.g., $\breve{\Gamma}_{\rho}^{(q)}$). The breve serves as a mnemonic as it is symmetric about a vertical axis.}
\begin{align}
\breve{\Gamma}_{\rho}^{(q)}(a):=\max_{\textrm{Boolean }f:\mathbb{E}[f(X^{n})]=a}\breve{\mathbf{S}}^{ (q)}_{\rho}[f],\label{eqn:sym_max_q}
\end{align}
where 
\begin{align}
\breve{\mathbf{S}}^{ (q)}_{\rho}[f]:=\mathbf{S}^{(q)}_{\rho}[f]+\mathbf{S}^{(q)}_{\rho}[1-f] \label{eqn:sym_q_stab_f}
\end{align}
is the {\em symmetric $q$-stability of $f$}. 
\end{definition}
%Let $\barx^n:=1^n-x^n$ be the bitwise negation of $x^n\in\{0,1\}^n$.
 Let $f$ be an anti-symmetric Boolean function. Consider, %One can verify 
%that for an anti-symmetric Boolean function $f$, 
\begin{align}
\mathbf{S}^{(q)}_{\rho}[1-f] &= \mathbf{S}^{(q)}_{\rho}[f(1^{n}-\cdot)] \label{eqn:use_asymmetric}\\
 & = \frac{1}{2^n}\sum_{y^{n}\in\{0,1\}^{n}} \big(\mathbb{E}[f(\bar{X}^{n})|Y^{n}=y^{n}]\big)^{q}\\
 & = \frac{1}{2^n}\sum_{\bar{y}^{n}\in\{0,1\}^{n}}\big(\mathbb{E}[f(\bar{X}^{n})|\bar{Y}^{n}=\bar{y}^{n}]\big)^{q}\\
 & =\mathbf{S}^{(q)}_{\rho}[f], \label{eqn:same_dist}
\end{align}
% & =2^{-n}\sum_{y^{n}\in\{0,1\}^{n}}(\mathbb{E}[f(\hat{X}^{n})|\hat{Y}^{n}=1^{n}-y^{n}])^{q}\\
where~\eqref{eqn:use_asymmetric} follows from~\eqref{eqn:asymmetric}, and  %we use the notation $\bar{y}^{n}:=1^{n}-y^{n}$ for any $y^{n}\in\{0,1\}^{n}$,
%and the last line 
\eqref{eqn:same_dist} follows because $(\bar{X}^{n},\bar{Y}^{n})$ has the
same joint distribution as $(X^{n},Y^{n})$. Hence, for an anti-symmetric
Boolean function $f$, 
\begin{align}
\breve{\bS}_{\rho}^{(q)}[f]   =2\,\mathbf{S}^{(q)}_{\rho}[f].\label{eq:stability-asym-func}
\end{align}

Furthermore, similarly to the asymmetric case, the symmetric max $q$-stability
also admits an important operational interpretation in the $k$-user
NICD problem; see~\eqref{eqn:sym_max_k}.  To describe this, we need to first introduce the following  proposition, which is the symmetric counterpart of Proposition~\ref{prop:identity2-1} and is 
due to \citet{mossel2005coin}. The proof is almost the
same as that of Proposition \ref{prop:identity2-1} and hence, is omitted. 

%\begin{proposition} \label{prop:identity2} Fix $k,n\ge 1,$ and $\rho\in [0,1]$.
%Let $\mathcal{F}$ be any class of Boolean functions on $\{0,1\}^{n}$.
%Subject to the restriction that $f_{1},\ldots,f_{k}\in\mathcal{F}$,
%every tuple $(f_{1},\ldots,f_{k})$ which maximizes $\Pr(f_{1}(X_{1}^{n})=\ldots=f_{k}(X_{k}^{n}))$
%has $f_{1}=f_{2}=\ldots=f_{k}$. \end{proposition}

\begin{proposition}\label{prop:identity2} 
Let $\calF $ be any class of Boolean functions. Let  $k,n\ge 1$ and $\rho\in (0,1)$.
Every tuple of functions $(f_{1},\ldots,f_{k}) \in \calF^k$ that maximizes $\Pr\big(f_{1}(X_{1}^{n})=\ldots=f_{k}(X_{k}^{n})\big)$
satisfies $f_{1}=f_{2}=\ldots=f_{k}$.
\end{proposition}

By choosing $\mathcal{F}$  in Proposition \ref{prop:identity2} to be the set of 
Boolean functions with mean $a$, we deduce that the symmetric max $q$-stability with $q=k$ (an integer) satisfies
\begin{align}
\breve{\Gamma}_{\rho}^{(k)}(a)& =\max_{\substack{\textrm{Boolean }f_{i},1\le i\le k:\\
\Pr(f_{i}(X_{i}^{n})=1)=a
}
}\Pr\big(f_{1}(X_{1}^{n})=\ldots=f_{k}(X_{k}^{n}) \big).\label{eqn:sym_max_k}
\end{align}
This is also called the {\em symmetric forward joint probability}
in the $k$-user NICD problem. Thus, the symmetric max $k$-stability $\breve{\Gamma}_{\rho}^{(k)}$ quantifies the {\em maximum agreement probability} over all Boolean functions with a fixed mean in the $k$-user NICD problem (Fig.~\ref{fig:NICD-with-k}).  In contrast, the forward joint probability or asymmetric max $k$-stability  $\Gamma_{\rho}^{(k)}$ (in~\eqref{eq:for_jp} and \eqref{eqn:asymp_max_k}) quantifies the maximum agreement probability {\em when the generated bits take on the value $1$}. 

\subsection{Variants of $q$-Stabilities}

The reader will notice that the definitions of the asymmetric
and symmetric max $q$-stabilities in~\eqref{eqn:asymp_max_q} and~\eqref{eqn:sym_max_q}
are trivial for the case $q=1$, since for this case, any Boolean
$f$ such that $\mathbb{E}[f(X^{n})]=a$ satisfies 
\begin{align}
\mathbf{S}^{(1)}_{\rho}[f]=a\quad\mbox{and}\quad\breve{\bS}_{\rho}^{(1)}[f] =1. \label{eqn:simpleS}
\end{align}
Hence, the asymmetric and symmetric max $1$-stabilities at mean $a$
are attained by {\em any} Boolean functions with mean $a$. Are there any
``more meaningful'' notions of asymmetric and symmetric max $q$-stabilities for $q=1$?
We answer this question in the affirmative by defining variants of
the max  $q$-stabilities. These variants connect the $q$-stabilities
to the {\em most informative Boolean functions} problem of \citet{courtade2014boolean},
one of the most important open problems in information theory at the
time of the writing of this monograph.

To introduce these variants, for $q\ge1$, define\footnote{These functions are not to be confused with the Gaussian cumulative distribution function which is also denoted as $\Phi(\cdot)$.}  $\Phi_{q}, \breve{\Phi}_{q}: (0,1)\to\bbR$ as 
\begin{align}
\Phi_{q}(t) & :=t\cdot \frac{\ln_{q}(t)}{\ln 2} \quad\mbox{and}\quad
\breve{\Phi}_{q} (t)   :=\Phi_{q}(t)+\Phi_{q}(1-t), \label{eqn:def_Phiq}
\end{align}
where  $\ln_q :(0,\infty)\to\bbR$ is defined as 
\begin{equation}
\ln_{q}(t):=\begin{cases}
\ln(t) &  q=1 \vspace{0.03in}\\
\displaystyle\frac{t^{q-1}-1}{q-1} &q>1
\end{cases}
\end{equation}
and is known as the \emph{$q$-logarithm} introduced by \citet{tsallis1994numbers},
but with a slight reparameterization. Note that  for $\Phi_{q}$ and $\breve{\Phi}_{q}$,
the case of $q=1$ is the continuous extension  of the case $q>1$.
\begin{definition}
  For a Boolean function $f:\{0,1\}^{n}\to\{0,1\}$ and another function $\Phi: (0,1)\to\bbR$, define the
\emph{$\Phi$-stability  of $f$} with respect to a correlation parameter $\rho$ as 
\begin{align}
\mathbf{S}_{\rho}^{( \Phi) }  [f] & =\mathbb{E}_{Y^n}\big[\Phi(T_{\rho}f(Y^{n}))\big]. \label{eqn:Phi_stab}
\end{align}
\end{definition}
Thus, this definition is analogous to that of the $q$-stability (Definition~\ref{def:qstab}) as we recover the latter when we instantiate $\Phi(t) = t^q$. We are, however, going to consider $\Phi$ to be the functions in~\eqref{eqn:def_Phiq}.
\begin{definition} \label{def:Phi_versions}
Define the   \emph{$\Phi$-asymmetric} and \emph{$\Phi$-symmetric max
$q$-stabilities   at mean $a$} as 
\begin{align}
{\Pi}_{\rho}^{(q)} (a) & :=\max_{\textrm{Boolean }f:\mathbb{E}[f(X^{n})]=a}\mathbf{S}_{\rho}^{ ( \Phi_{q})}[f]\quad\mbox{and}\label{eq:-3}\\
\breve{\Pi}_{\rho}^{(q)} (a) & :=\max_{\textrm{Boolean }f:\mathbb{E}[f(X^{n})]=a}\mathbf{S}_{\rho}^{ ( \breve{\Phi}_{q})}[f].\label{eq:-4}
\end{align}
\end{definition}

For $q>1$, it is easy to verify that 
\begin{align}
{\Pi}_{\rho}^{(q)} (a)=\frac{{\Gamma}^{(q)}_{\rho}(a)-a}{(q-1)\ln2}\quad\mbox{and}\quad\breve{\Pi}_{\rho}^{(q)}(a)=\frac{\breve{\Gamma}_{\rho}^{(q)}(a)-1}{(q-1)\ln2},\label{eqn:qge1}
\end{align}
where ${\Gamma}_{\rho}^{(q)}(a)$ and $\breve{\Gamma}_{\rho}^{(q)}$
are the asymmetric and symmetric max $q$-stabilities defined in \eqref{eqn:asymp_max_q}
and \eqref{eqn:sym_max_q} respectively. For $q=1$,  the $\Phi$-asymmetric and $\Phi$-symmetric max $1$-stabilities at mean $a$ can be expressed  respectively  as  
\begin{equation}
{\Pi}_{\rho}^{(1)} (a)=\max_{\textrm{Boolean }f:\mathbb{E}[f(X^{n})]=a}\mathbb{E}_{Y^{n}}\big[T_{\rho}f(Y^{n})\log T_{\rho}f(Y^{n})\big],\label{eq:NICD-2}
\end{equation}
and 
\begin{align}
\breve{\Pi}_{\rho}^{(1)}(a)& =\max_{\textrm{Boolean }f:\mathbb{E}[f(X^{n})]=a}-H(f(X^{n})|Y^{n})\label{eq:NICD-3}\\
 & =\max_{\textrm{Boolean }f:\mathbb{E}[f(X^{n})]=a}I(f(X^{n});Y^{n})-h(a).\label{eq:NICD-8}
\end{align}

The objective function in \eqref{eq:NICD-2} is known as the \textit{entropy
(functional)} of the noisy Boolean function $T_{\rho}f$, and the
objective function in \eqref{eq:NICD-3} is the negative \textit{conditional
Shannon entropy} of $f(X^{n})$ given $Y^{n}$. For dictator functions
$f$, the conditional Shannon entropy $H(f(X^{n})|Y^{n})$ is equal
to $H(X_1|Y_1) =h((1-\rho)/2)$. 

The maximization in~\eqref{eq:NICD-8} for $a=1/2$ corresponds to
the balanced version of the {\em most informative Boolean function}
problem which was first studied in the papers of Courtade and Kumar~\cite{kumar2013boolean,courtade2014boolean}. They conjectured that the
maximum in \eqref{eq:NICD-8} is attained by dictator functions (cf.\ Section~\ref{sec:subcubes}).
In the following section, we provide more details on  this
conjecture, and we also review several related conjectures concerning   the
$q$-stabilities. Observe from the one-to-one relationships in \eqref{eqn:qge1} that for $q>1$,
the original definitions of the asymmetric and symmetric max $q$-stabilities
${\Gamma}_{\rho}^{(q)}$ and $\breve{\Gamma}_{\rho}^{(q)}$
in~\eqref{eqn:asymp_max_q} and~\eqref{eqn:sym_max_q} are ``equivalent''
to their $\Phi$-versions ${\Pi}_{\rho}^{(q)}$ and $\breve{\Pi}_{\rho}^{(q)}$
defined respectively in \eqref{eq:-3} and \eqref{eq:-4}, in the sense that once
the former (resp.\ the latter) has  been determined, the latter (resp.\
the former) will also be  determined. Hence, throughout this section,
for $q>1$, we refer to  ${\Gamma}_{\rho}^{(q)}$
and ${\Pi}_{\rho}^{(q)}$  interchangeably for the asymmetric case. We will also refer to $\breve{\Gamma}_{\rho}^{(q)}$ and $\breve{\Pi}_{\rho}^{(q)}$ interchangeably for the symmetric case.
However, for  $q=1$, we only consider
the quantities ${\Pi}_{\rho}^{(1)}$ and $\breve{\Pi}_{\rho}^{(1)}$
since the definitions of ${\Gamma}_{\rho}^{(1)}$ and $\breve{\Gamma}_{\rho}^{(1)}$
are trivial for this case; see~\eqref{eqn:simpleS}.

\section{Related Conjectures}
\label{sec:conj}

In this section, we introduce several prominent conjectures on the
max $q$-stabilities. We first consider the optimality of dictator
functions in attaining the asymmetric and symmetric max $q$-stabilities
at mean $a=1/2$ (also called the {\em balanced} case). For ease of reference, we first state
a corollary to Witsenhausen's classical result~\cite{witsenhausen1975sequences} in Theorem \ref{thm:wit}. This corollary implies that dictator functions are optimal in attaining
asymmetric or symmetric max $q$-stabilities for $q=2$ and $a=1/2$ (the $2$-user NICD problem in the CL regime with $a=b=1/2$).  

\begin{corollary} \label{thm:wit-stability} For $q=2$ and $\rho\in(0,1)$,
both ${\Gamma}_{\rho}^{(q)}( {1}/{2})$ and $\breve{\Gamma}_{\rho}^{(q)}( {1}/{2})$
are attained by dictator functions.  \end{corollary}

There was no further progress on the max $q$-stability problem for
almost $30$ years since Witsenhausen's seminal work~\cite{witsenhausen1975sequences} in $1975$.  In $2005$,  \citet{mossel2005coin} considered the 
%and O'Donnell's work \cite{mossel2005coin} published in 2005. The
symmetric max $q$-stability problem with $q\in\{3,4,5,\ldots\}$ and
%was studied by \citet{mossel2005coin}, and push forward 
made  progress
on this problem. They resolved the case of $q=3$ for the
balanced case (i.e., $a=1/2$) using a cute reduction argument.

\begin{theorem} \label{thm:mos-stability} For $q=3$ and $\rho\in(0,1)$,
$\breve{\Gamma}_{\rho}^{(q)}( {1}/{2})$ is attained by
dictator functions.  \end{theorem}

\begin{proof} %Note that for $q=2$, $\breve{\Gamma}_{\rho}^{(q)}( {1}/{2})$
%is attained by dictator functions (Corollary~\ref{thm:wit-stability}). 
Theorem~\ref{thm:mos-stability} can be proved
by reducing the problem involving $q=3$ to the (simpler)  problem in which $q=2$. By the equivalence between the NICD problem and 
max $q$-stability, we consider
the $3$-user NICD problem with $3$ (possibly) distinct functions $(f_1, f_2, f_3)$. For brevity, denote the values of the joint  probability mass function of  $(U_1, U_2, U_3) = (f_{1}(X_{1}^{n}),f_{2}(X_{2}^{n}),f_{3}(X_{3}^{n}))$
as $\left\{ p_{000},p_{001},\ldots,p_{111}\right\} $. Then, we see that the following identity holds:
\begin{align}	
&  3  +  \sum_{(i, j) \in[3]^2 : i \ne j } \Pr (U_i = U_j )     = 5  +  4\,\Pr(U_1 = U_2 = U_3). \label{eqn:q3q2}
\end{align}
This identity can be verified by   bookkeeping the probability masses. For example, note that $\Pr(U_1 = U_2) = p_{000} + p_{001} + p_{110} + p_{111} $ and $\Pr(U_1 = U_2 = U_3) = p_{000}  + p_{111}$. 
Having established this, leveraging the case for $q=2$ (Corollary~\ref{thm:wit-stability}),
we know that the left-hand side of \eqref{eqn:q3q2} is maximized by identical dictator functions
over all balanced Boolean functions (i.e., $\bbE[f_i(X_i^n)]=1/2$); hence, so is the right-hand side. \end{proof}

Based on Corollary \ref{thm:wit-stability} and Theorem~\ref{thm:mos-stability},
one may na\"ively conjecture that dictator functions are   
optimal in attaining the asymmetric or symmetric max $q$-stability
at mean $1/2$ for any integer $q\ge 2$. However, this  was disproved by  
\citet{mossel2005coin}. Specifically, using computer-assisted calculations, they constructed    counterexamples,
as shown in the following proposition, such that when $q=10$, dictator
functions are not optimal in attaining the asymmetric and symmetric max $q$-stabilities
at mean $a=1/2$. 

\begin{proposition} \label{prop:counterexample} For $q=10$ and $\rho=0.48$,
it holds that 
\begin{align}
\mathbf{S}^{(q)}_{\rho}[\mathrm{Maj}_{3}] & > \max \big\{\mathbf{S}^{(q)}_{\rho}[\mathrm{Maj}_{1}],\mathbf{S}^{(q)}_{\rho}[\mathrm{Maj}_{5}]\big\} \quad\textrm{and }\label{eq:-7}\\
\breve{\mathbf{S}}^{(q)}_{\rho}[\mathrm{Maj}_{3}] & > \max\big\{\breve{\mathbf{S}}^{(q)}_{\rho}[\mathrm{Maj}_{1}], \breve{\mathbf{S}}^{(q)}_{\rho}[\mathrm{Maj}_{5}]\big\}.\label{eq:-6}
\end{align}
  \end{proposition}
%\mathbf{S}^{\mathrm{sym}(q)}_{\rho}[\mathrm{Maj}_{3}] & \ge\max\{\mathbf{S}^{\mathrm{sym}(q)}_{\rho}[\mathrm{Maj}_{1}], nonumber \\
%& \qquad \mathbf{S}^{\mathrm{sym}(q)}_{\rho}[\mathrm{Maj}_{5}]\}

\begin{proof} By computer-assisted calculations, for $q=10$ and $\rho=0.48$, one finds that
$\breve{\mathbf{S}}^{(q)}_{\rho}[\mathrm{Maj}_{1}]\le0.0493$, $\breve{\mathbf{S}}^{(q)}_{\rho}[\mathrm{Maj}_{5}]\le0.0488$,
and $\breve{\mathbf{S}}^{(q)}_{\rho}[\mathrm{Maj}_{3}]\ge0.0496$.
Hence, the inequality in \eqref{eq:-6} holds. The inequality in \eqref{eq:-7}
follows from~\eqref{eq:-6} since $\breve{\mathbf{S}}^{(q)}_{\rho}[\mathrm{Maj}_{m}]=2\, \mathbf{S}^{(q)}_{\rho}[\mathrm{Maj}_{m}]$
for any odd $m\le n$; see~\eqref{eq:stability-asym-func}. \end{proof}

Since $\mathbf{S}^{(q)}_{\rho}[\mathrm{Maj}_{3}] > \mathbf{S}^{(q)}_{\rho}[\mathrm{Maj}_{1}]$ and $\mathrm{Maj}_1$ is a dictator function, dictators are not optimal for $q=10$ and $\rho=0.48$. Furthermore,  since the indicators of subcubes and the indicators of Hamming balls (or spheres) have been shown to be optimal or asymptotically optimal in several cases for the NICD problem (Sections~\ref{sec:nicd_conv}--\ref{sec:ldr}), one may wonder whether the max $q$-stability is always exactly attained by these functions. The inequality $\mathbf{S}^{(q)}_{\rho}[\mathrm{Maj}_{3}] > \mathbf{S}^{(q)}_{\rho}[\mathrm{Maj}_{5}]$ implies a negative answer to this question. 
For $n=5$, $q=10$, $a=1/2$, and $\rho=0.48$, both the indicators of subcubes and Hamming balls in the $5$-dimensional Hamming cube are not optimal.
In fact, $\mathrm{Maj}_{3}$ corresponds to the indicator of a set formed by multiplying Hamming balls in the $3$-dimensional cube and the $2$-dimensional cube.

Now things have become relatively clearer. For small $q$, e.g., $q=2$ or
$q=3$, dictator functions are optimal in attaining the asymmetric (for
$q=2$) or symmetric (for $q=2,3$) max $q$-stabilities at mean $a=1/2$.
On the other hand, for large $q$, e.g., $q=10$, dictator functions
are not optimal. \citet{mossel2005coin}
conjectured that for all $q\in\{4,5,\ldots,9\}$, dictator functions maximize the symmetric $q$-stability $\breve{\Gamma}_{\rho}^{(q)}( {1}/{2})$ 
over all balanced Boolean functions.

The symmetric max $1$-stability problem  at mean $a=1/2$ (i.e., $\breve{\Pi}_{\rho}^{(1)}(1/2)$) was
studied by \citet{kumar2013boolean} and \cite{courtade2014boolean}.
This question concerns the identification of the class of balanced Boolean
functions that maximize the mutual information $I(f(Y^{n});X^{n})$;
cf.\ \eqref{eq:NICD-8}. The authors conjectured that dictator functions
maximize the symmetric $1$-stability. 
We note that this is a weaker version of the original conjecture posed
by Courtade and Kumar. In the original version of their conjecture,
the Boolean functions are not restricted to be balanced. Along these
lines, \citet{li2021boolean} conjectured that for $q\in(1,2)$ (non-integer), the
max asymmetric $q$-stability is still attained by dictator functions.
Here we summarize and generalize this family of conjectures in the
following two conjectures. 

\begin{conjecture}[Asymmetric max 
$q$-stability] \label{conj:AsymmetricStability} For $\rho\in[0,1]$
and $q\in[1,9]$, ${\Pi}_{\rho}^{(q)}(1/2)$ is attained by dictator
functions. \end{conjecture} \begin{conjecture}[Symmetric max 
$q$-stability] \label{conj:SymmetricStability} For $\rho\in[0,1]$
and $q\in[1,9]$, $\breve{\Pi}_{\rho}^{(q)}(1/2)$ is attained
by dictator functions. \end{conjecture} Observe that dictator functions
are anti-symmetric. Hence, \eqref{eq:stability-asym-func} holds for dictator
functions, which implies that if Conjecture~\ref{conj:AsymmetricStability}
is true, so is Conjecture~\ref{conj:SymmetricStability}. Conjectures~\ref{conj:AsymmetricStability}
and~\ref{conj:SymmetricStability} together consist of three (named)
conjectures, as summarized in Table~\ref{tab:Illustration-of-various}.

\begin{table}
\caption{\label{tab:Illustration-of-various} Illustration of  the various named conjectures
on max $q$-stabilities; these  constitute Conjectures \ref{conj:AsymmetricStability}
and \ref{conj:SymmetricStability}  }
\begin{centering}
\begin{tabular}{|c|>{\centering}p{6cm}|}
\hline 
$q$  & Are dictators optimal in attaining ${\Pi}_{\rho}^{(q)}(1/2)$ (or
$\breve{\Pi}_{\rho}^{(q)}(1/2)$)?\tabularnewline
\hline 
\hline 
$q=1$  & Courtade--Kumar conjecture (balanced version) \cite{courtade2014boolean} \tabularnewline
\hline 
$1<q<2$  & Li--M\'edard conjecture \cite{li2021boolean} \tabularnewline
\hline 
$q=2$ ($2$-User NICD)  & True and shown by \citet{witsenhausen1975sequences} (cf.\ Section~\ref{sec:clt_half}) \tabularnewline
\hline 
$2<q\le9$  & Mossel--O'Donnell conjecture \cite{mossel2005coin} \tabularnewline
\hline 
\end{tabular}
\par\end{centering}
\end{table}

%\citet{barnes2020courtade} showed that Conjecture \ref{conj:AsymmetricStability}
%with $q=1$ and Conjecture \ref{conj:AsymmetricStability} with $q\in(1,2)$
%are equivalent, and Conjecture \ref{conj:SymmetricStability} with
%$q=1$ and Conjecture \ref{conj:SymmetricStability} with $q\in(1,2)$
%are equivalent. That is, the asymmetric (resp.\ symmetric) version
%of Courtade--Kumar conjecture and the asymmetric (resp.\ symmetric)
%version of Li--M\'edard conjecture are equivalent. In fact, 
\citet{barnes2020courtade}
proved an  interesting dichotomy concerning these conjectures.

%~\ref{conj:AsymmetricStability}
%and~\ref{conj:SymmetricStability}. 

\begin{lemma} \label{lem:BarnesOzgur} For $a=1/2$, there are two
thresholds $q_{\min}$ and $q_{\max}$ satisfying $1\le q_{\min}\le  2\le q_{\max}$
such that dictator functions are optimal in attaining the asymmetric
max  $q$-stability with $q\ge 1$ if and only if $q\in[q_{\min},q_{\max}]$.
This statement also holds for the symmetric max $q$-stability
but with possibly different thresholds $\breve{q}_{\min}$
and $\breve{q}_{\max}$ satisfying the same condition  $1\le \breve{q}_{\min}\le 2\le\breve{q}_{\max}$.
\end{lemma}

\begin{proof}[Proof Sketch of Lemma \ref{lem:BarnesOzgur}] For any
$q\in\mathbb{R}$ (not necessarily greater than or equal to $1$), define 
\begin{equation}
N_{q}(f):= 2^n\,\big\|T_\rho f\big\|_q^q= \sum_{y^{n}\in\{0,1\}^{n}}\big(T_{\rho}f(y^{n})\big)^{q}.\label{eq:alphanorm}
\end{equation}
%Note that here $q$ is relaxed to an arbitrary real number, not necessarily
%to be greater than $1$. 
Let $f_{0}$ be a dictator  function,
e.g., $f_{0}=\mathrm{Maj}_{1}$. Define 
\begin{equation}
g_{f}(q):=N_{q}(f)-N_{q}(f_{0}).\label{eq:g}
\end{equation}
By using a result due to Laguerre \cite{laguerre1884theorie}, one can find that  the sum of exponentials
   $g_{f}(q)$    has at most
four roots. 
% (unless, of course, they are identically roots). 
 Observe that 
\begin{align}
g_{f}(0) & =0,\quad g_{f}(1)   =0,\quad g_{f}(2)  \le0, \quad\mbox{and}\\*
g_{f}(q),\ g_{f}(-q) & >0\;\;\textrm{ for sufficiently large }q.
\end{align}
From these observations, we know that $g_{f}(q)$ has a root  at $q_{1}\ge2$,
another at $q_{2}=1$, and another at $q_{3}=0$. \enlargethispage{-3\baselineskip} Moreover, the remaining
  root $q_{4}$ satisfies $q_{4}\le2$. Hence, $g_{f}(q)\le0$
for all $q$ in the interval $[\max\{q_{4},1\},q_{1}]$; see Fig.~\ref{fig:gf}. Taking the
intersection of  these intervals for all non-dictator  functions $f$, we
  obtain the interval $\left[q_{\min},q_{\max}\right]$, where  $q_{\min}$
and $q_{\max}$ are  the desired thresholds. The symmetric case follows
similarly.  \end{proof}

\begin{figure}[!ht]
\centering 
\includegraphics[width=0.85\columnwidth]{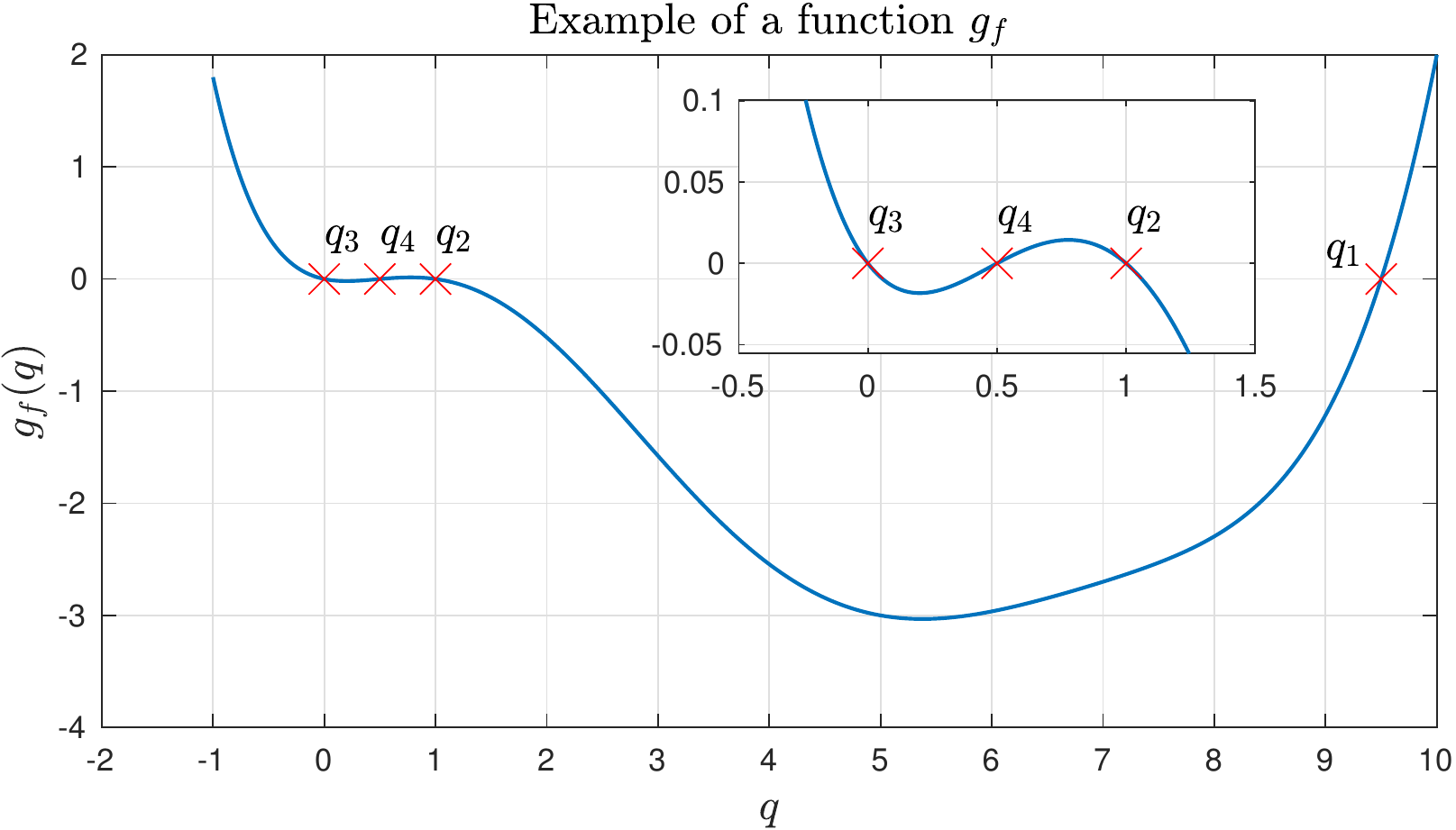}
\caption{\label{fig:gf}Example of a function $g_f$ in \eqref{eq:g}.} % Dictators are optimal in the interval $[q_4, q_1]$ for this example.}
\end{figure}

\begin{figure}[!ht]
\centering \centering \tikzset{every picture/.style={line width=0.75pt}}
%set default line width to 0.75pt        
\begin{tikzpicture}[x=0.75pt,y=0.75pt,yscale=-1,xscale=1,scale=0.9]%uncomment if require: \path (0,2761); %set diagram left start at 0, and has height of 2761
%Straight Lines [id:da8898493418186564] 
\draw    (248,2616) -- (419,2616) -- (630,2616) ;
\draw [shift={(632,2616)}, rotate = 180] [color={rgb, 255:red, 0; green, 0; blue, 0 }  ][line width=0.75]    (10.93,-3.29) .. controls (6.95,-1.4) and (3.31,-0.3) .. (0,0) .. controls (3.31,0.3) and (6.95,1.4) .. (10.93,3.29)   ; %Curve Lines [id:da8479236912410413] 
\draw    (302,2611.67) .. controls (342,2581.67) and (451,2586.14) .. (490,2612.14) ; %Curve Lines [id:da8702449668579595] 
\draw    (269,2649.67) .. controls (273,2663) and (292,2662) .. (296,2649.67) ; %Curve Lines [id:da565756033663503] 
\draw    (502,2646) .. controls (523,2659.86) and (571,2660) .. (615,2660) ; %Curve Lines [id:da5768253643614716] 
\draw    (412,2683) .. controls (451.6,2653.3) and (494.14,2691.23) .. (533.8,2662.88) ;
\draw [shift={(535,2662)}, rotate = 503.13] [color={rgb, 255:red, 0; green, 0; blue, 0 }  ][line width=0.75]    (10.93,-3.29) .. controls (6.95,-1.4) and (3.31,-0.3) .. (0,0) .. controls (3.31,0.3) and (6.95,1.4) .. (10.93,3.29)   ; %Curve Lines [id:da6543594174939376] 
\draw    (395,2682) .. controls (371.36,2656.39) and (316.67,2690.93) .. (282.54,2664.26) ;
\draw [shift={(281,2663)}, rotate = 400.46000000000004] [color={rgb, 255:red, 0; green, 0; blue, 0 }  ][line width=0.75]    (10.93,-3.29) .. controls (6.95,-1.4) and (3.31,-0.3) .. (0,0) .. controls (3.31,0.3) and (6.95,1.4) .. (10.93,3.29)   ;
% Text Node
\draw (336,2625) node [anchor=north west][inner sep=0.75pt]    {$2$}; % Text Node
\draw (336,2608.82) node [anchor=north west][inner sep=0.75pt]   [align=left] {$\times$}; % Text Node
\draw (612,2619.67) node [anchor=north west][inner sep=0.75pt]    {$q$}; % Text Node
\draw (479,2621.67) node [anchor=north west][inner sep=0.75pt]  [color={rgb, 255:red, 208; green, 2; blue, 27 }  ,opacity=1 ]  {$q_{\max}$}; % Text Node
\draw (282,2621.67) node [anchor=north west][inner sep=0.75pt]  [color={rgb, 255:red, 208; green, 2; blue, 27 }  ,opacity=1 ]  {$q_{\min}$}; % Text Node
\draw (262,2625) node [anchor=north west][inner sep=0.75pt]    {$1$}; % Text Node
\draw (262,2608.82) node [anchor=north west][inner sep=0.75pt]   [align=left] {$\times$}; % Text Node
\draw (488,2608.82) node [anchor=north west][inner sep=0.75pt]  [color={rgb, 255:red, 208; green, 2; blue, 27 }  ,opacity=1 ] [align=left] {$\times$}; % Text Node
\draw (288,2608.82) node [anchor=north west][inner sep=0.75pt]  [color={rgb, 255:red, 208; green, 2; blue, 27 }  ,opacity=1 ] [align=left] {$\times$}; % Text Node
\draw (331,2566.67) node [anchor=north west][inner sep=0.75pt]   [align=left] {Dictators optimal}; % Text Node
\draw (326,2682.67) node [anchor=north west][inner sep=0.75pt]   [align=left] {Dictators not optimal};
\end{tikzpicture} \caption{\label{fig:Illustration-of-Lemma}Illustration of Lemma \ref{lem:BarnesOzgur}.}
\end{figure}
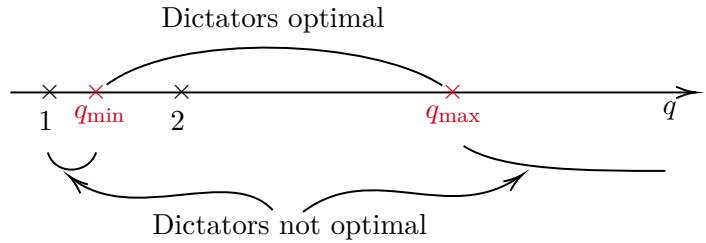

This lemma is illustrated in Fig.~\ref{fig:Illustration-of-Lemma}.

\enlargethispage{-\baselineskip}
\begin{remark} \label{rmk:imp_BO}
This lemma has several important implications. 
\begin{enumerate}
\item[(a)] Firstly, this lemma implies that both the Courtade--Kumar conjecture
and the Li--M\'edard conjecture are equivalent to the statements that $q_{\min}=1$ for
the asymmetric version and $\breve{q}_{\min}=1$ for the symmetric
version. Hence, the Courtade--Kumar conjecture and the Li--M\'edard
conjecture are also equivalent (to each other). 
\item[(b)] Secondly, it also implies that the Mossel--O'Donnell conjecture is
equivalent to the statements that $q_{\max}\ge9$ for the asymmetric
version and $\breve{q}_{\max}\ge9$ for the symmetric version.
On the other hand, from Proposition  \ref{prop:counterexample},
we see that $\max\{q_{\max},\breve{q}_{\max}\}<10$. 
\item[(c)] Lastly, by Lemma~\ref{lem:BarnesOzgur} and Theorem \ref{thm:mos-stability}
(i.e., Conjecture~\ref{conj:SymmetricStability} holds for $q=3$),
Conjecture~\ref{conj:SymmetricStability} also holds for any $q\in[2,3]$.
In other words, Conjecture~\ref{conj:SymmetricStability} is only
open for $q\in[1,2)\cup(3,9]$. 
\end{enumerate}
\end{remark}
Combining all points in Remark~\ref{rmk:imp_BO} yields that $2\le q_{\max}<10$ and $3\le \breve{q}_{\max}<10$.
If Conjectures~\ref{conj:AsymmetricStability} and~\ref{conj:SymmetricStability}
hold, then the estimates of $q_{\max}$ and $\breve{q}_{\max}$
can be improved to $9\le q_{\max},\breve{q}_{\max}<10$, as
shown in Fig.~\ref{fig:Illustration-of-conjectures}. 

%We will introduce related results on the balanced case in the next
%subsection. Besides the balanced case, we will also introduce the
%LD case and the MD case in Section \ref{subsec:MD-and-LD}.
%\begin{center}

\begin{figure}[!ht]
\centering \tikzset{every picture/.style={line width=0.75pt}}
%set default line width to 0.75pt        
\begin{tikzpicture}[x=0.75pt,y=0.75pt,yscale=-1,xscale=1,scale=0.9] %uncomment if require: \path (0,2761); %set diagram left start at 0, and has height of 2761
%Straight Lines [id:da4639412408407191] 
\draw    (270,2400) -- (441,2400) -- (652,2400) ;
\draw [shift={(654,2400)}, rotate = 180] [color={rgb, 255:red, 0; green, 0; blue, 0 }  ][line width=0.75]    (10.93,-3.29) .. controls (6.95,-1.4) and (3.31,-0.3) .. (0,0) .. controls (3.31,0.3) and (6.95,1.4) .. (10.93,3.29)   ; %Curve Lines [id:da5664650278878078] 
\draw    (296,2392.14) .. controls (336,2362.14) and (473,2370.14) .. (512,2396.14) ; %Curve Lines [id:da637844441536136] 
\draw    (521,2392.14) .. controls (533,2376.14) and (592,2374.67) .. (633,2372.67) ;
% Text Node
\draw (358,2409) node [anchor=north west][inner sep=0.75pt]    {$2$}; % Text Node
\draw (358,2392.82) node [anchor=north west][inner sep=0.75pt]   [align=left] {$\times$}; % Text Node
\draw (634,2403.67) node [anchor=north west][inner sep=0.75pt]    {$q$}; % Text Node
\draw (501,2405.67) node [anchor=north west][inner sep=0.75pt]  [color={rgb, 255:red, 208; green, 2; blue, 27 }  ,opacity=1 ]  {$q_{\max}$}; % Text Node
\draw (234,2410.67) node [anchor=north west][inner sep=0.75pt]  [color={rgb, 255:red, 208; green, 2; blue, 27 }  ,opacity=1 ]  {$q_{\min} =1$}; % Text Node
\draw (284,2392.82) node [anchor=north west][inner sep=0.75pt]  [color={rgb, 255:red, 208; green, 2; blue, 27 }  ,opacity=1 ] [align=left] {$\times$}; % Text Node
\draw (510,2392.82) node [anchor=north west][inner sep=0.75pt]  [color={rgb, 255:red, 208; green, 2; blue, 27 }  ,opacity=1 ] [align=left] {$\times$}; % Text Node
\draw (482,2392.82) node [anchor=north west][inner sep=0.75pt]   [align=left] {$\times$}; % Text Node
\draw (483,2410) node [anchor=north west][inner sep=0.75pt]    {$9$}; % Text Node
\draw (546,2392.82) node [anchor=north west][inner sep=0.75pt]   [align=left] {$\times$}; % Text Node
\draw (547,2410) node [anchor=north west][inner sep=0.75pt]    {$10$}; % Text Node
\draw (330,2346.67) node [anchor=north west][inner sep=0.75pt]   [align=left] {Dictators optimal}; % Text Node
\draw (529,2350.67) node [anchor=north west][inner sep=0.75pt]   [align=left] {Dictators not optimal};
\end{tikzpicture} \caption{\label{fig:Illustration-of-conjectures}Illustration of the range
of $q$ for the optimality of dictator functions if Conjectures~\ref{conj:AsymmetricStability}
and~\ref{conj:SymmetricStability}
are true.}
\end{figure}
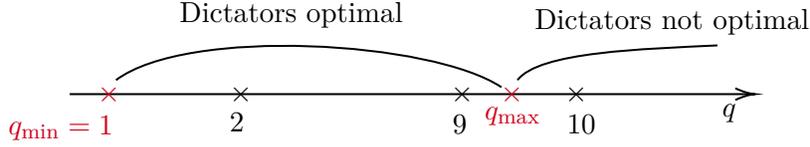

\section{Extreme Cases of the Correlation Coefficient}

\label{sec:extreme}
To better understand the max $q$-stabilities, and also to connect
them to several well-known concepts  
in the analysis of Boolean functions, we first focus our attention on the extreme cases
in which the correlation coefficient
$\rho$ tends to~$0$ or~$1$, but the   dimension (or blocklength) $n$ is kept fixed. To illustrate the intuition as to why some results hold, we introduce
the concepts of {\em influences} and   {\em edge-isoperimetric inequalities}. 

%; these concepts 
%will be used to derive  the extreme cases of the max $q$-stabilities in terms of $\rho$. 

\subsection{Influences} \label{sec:influ}
For a vector $x^{n}\in\{0,1\}^{n}$, we denote the vector with the $i^{\mathrm{th}}$ bit flipped as $(x^n)^{\oplus i}:=(x_{1},\ldots,x_{i-1},1-x_{i},x_{i+1},\ldots,x_{n})$. Denote the length-$(n-1)$ with the $i^{\mathrm{th}}$ component removed as  $x^{\backslash i }:= (x_1, \ldots, x_{i-1},x_{i+1}, \ldots x_n)$. 

%as the resulting vector by flipping the $i^{\mathrm{th}}$ component of $x^{n}$. 

\begin{definition}\label{def:influence} The \emph{influence  of coordinate $i \in [n]$} on a Boolean function  $f:\{0,1\}^{n}\to\{0,1\}$
is defined as
\begin{equation}
\mathbf{I}_{i}[f]:=\Pr\big(f(X^{n})\neq f((X^n)^{\oplus i})\big),
\end{equation}
where $X^{n}\sim\Unif\{0,1\}^{n}$. \end{definition}

Let $f$ be a Boolean function that depends only on $x^{\backslash i} $. This 
means that the value of $f$ evaluated at every $x^n$ is independent of the $i^{\mathrm{th}}$  component
$x_{i}$. For such an $f$, clearly, $\mathbf{I}_{i}[f]=0$. On the
other hand, if $f$  depends only on the $i^{\mathrm{th}}$ component, i.e., the
dictator functions $f(x^{n})=x_{i}$ or $1-x_{i}$, then $\mathbf{I}_{i}[f]=1$.
Hence, the influence of coordinate $i$ measures how much a function
is influenced by the $i^{\mathrm{th}}$ coordinate of the input; this coincides
with the literal meaning of ``influence''.  Furthermore, the influence
can be also expressed in terms of the discrete derivative operator
 as follows:

\begin{definition} Let $x^{i\mapsto b}:=(x_{1},\ldots,x_{i-1},b,x_{i+1},\ldots,x_{n})$. The $i^{\mathrm{th}}$ \emph{discrete derivative operator}
$D_{i}$ maps a function $f:\{0,1\}^{n}\to\mathbb{R}$ to the function
$D_{i}f:\{0,1\}^{n}\to\mathbb{R}$ defined as
\begin{equation}
D_{i}f(x^{n}):=f(x^{i\mapsto1})-f(x^{i\mapsto0}).
\end{equation}
\end{definition}

One observes that 
\begin{equation}
\mathbf{I}_{i}[f]=\mathbb{E}\big[D_{i}f(X^{n})^{2}\big]= \|D_{i}f \|_{2}^{2}.\label{eq:Infi-Di}
\end{equation}
This formula enables us to generalize the definition of the influence
from a Boolean function to an arbitrary real-valued function defined
on $\{0,1\}^{n}$; see \citet{ODonnell14analysisof}. We do not discuss
this generalization here, since we only mainly focus on Boolean functions.

\begin{definition} \label{def:total_influence} The \emph{total influence} (or \emph{average sensitivity})
of a Boolean function  $f:\{0,1\}^{n}\to\{0,1\}$ is defined as
\begin{equation}
\mathbf{I}[f]:=\sum_{i=1}^{n}\mathbf{I}_{i}[f].
\end{equation}
\end{definition}

\enlargethispage{\baselineskip}
The quantities $D_{i}f$, $\mathbf{I}_{i}[f]$, and $\mathbf{I}[f]$
admit the following Fourier-analytic representations. 

\begin{theorem} For a Boolean function  $f:\{0,1\}^{n}\to\{0,1\}$ and $i\in[n]$, 
\begin{align}
D_{i}f(x^{n}) & =-2\sum_{\calS\subset[n]:\calS\ni i}\hat{f}_{\calS}\cdot\chi_{\calS\backslash\{i\}}(x^{n}),\label{eq:Di}\\
\mathbf{I}_{i}[f] & =4\sum_{\calS\subset[n]:\calS\ni i}\hat{f}_{\calS}^{2}, \qquad\mbox{and}\label{eq:Infi}\\
\mathbf{I}[f] & =4\sum_{\calS\subset[n]}|\calS|\hat{f}_{\calS}^{2}=4\sum_{k=0}^{n}k\cdot\mathbf{W}_{k}[f].\label{eq:InfluenceFourier}
\end{align}
where  $\mathbf{W}_{k}[f]$ denotes the  degree-$k$ Fourier weight of $f$ defined in \eqref{eq:NICD-weight}. \end{theorem}

%$\chi_{S}(x^{n}):=(-1)^{\bigl\langle x^{n},y^{n}\bigr\rangle}$
%with $S=\{i\in [n]:y_{i}=1\}$, and
%\red{Lei: The notations for Fourier coefficients is inconsistent with the ones used in NICD chapter. I will change the ones in NICD chapter to these notations, namely using a set $S$ to denote a frequency-variable. }

\begin{proof} Since $D_{i}$ is a linear operator,  \eqref{eq:Di}
follows by expressing $f$ in terms of its Fourier coefficients, and then
applying the following identity 
\begin{equation}
D_{i} \chi_{\calS}(x^{n})=\begin{cases}
-2\,\chi_{\calS\backslash\{i\}}(x^{n}) & i\in \calS\vspace{.03in}\\
0 & i\notin \calS
\end{cases}.
\end{equation}
The identity in~\eqref{eq:Infi} follows from \eqref{eq:Infi-Di} and~\eqref{eq:Di}, and  the fact that for any sets $\calS,\calT\subset [n]$, % $\chi_{\calS}$ and $\chi_{\calT}$ are {\em orthogonal}, i.e., 
\begin{equation}
\bbE[\chi_{\calS}(X^n)\chi_{\calT}(X^n)]=\bone\{\calS=\calT\}. \label{eq:inner_prod_chi}
\end{equation}
The identity in~\eqref{eq:InfluenceFourier} follows from \eqref{eq:Infi} and Definition \ref{def:total_influence}.
\end{proof}

\enlargethispage{-\baselineskip}
The quantities $\mathbf{I}_{i}[f]$ and $\mathbf{I}[f]$ also admit
interesting graph-theoretic interpretations. Consider the  undirected 
graph in which the vertices consist of all vectors in $\{0,1\}^{n}$,
and two vertices $x^{n},y^{n}\in\{0,1\}^{n}$ are joined 
by an edge if the Hamming distance between them is exactly $1$, i.e.,
$d_{\mathrm{H}}(x^{n},y^{n})=1$. This graph is known as the {\em Hamming graph}; see Fig.~\ref{fig:ham_graph} for the Hamming graph when $n=3$.  

\begin{figure}
\centering
%\begin{overpic}[width = .94\columnwidth]{figs/HammingGraph2}
%\put(6,12){\mbox{$\calA$}}
%\put(11,30){\mbox{$\partial\calA$}}
%\put(30,30){\mbox{$\partial\calA$}}
%\put(51,30){\mbox{$\partial\calA$}}
%\put(69,30){\mbox{$\partial\calA$}}
%\put(17,9){\mbox{$101$}}
%\put(74,9){\mbox{$111$}}
%\put(17,48){\mbox{$001$}}
%\put(74,48){\mbox{$011$}}
%\put(35,20){\mbox{$100$}}
%\put(56,20){\mbox{$110$}}
%\put(35,37){\mbox{$000$}}
%\put(56,37){\mbox{$010$}}
%\end{overpic}

\tikzset{every picture/.style={line width=0.75pt}} %set default line width to 0.75pt        

\begin{tikzpicture}[x=0.75pt,y=0.75pt,yscale=-1,xscale=1,scale=0.5]
%uncomment if require: \path (0,3293); %set diagram left start at 0, and has height of 3293

%Straight Lines [id:da43556285341207435] 
\draw [line width=3]    (424,2799.44) -- (424,2926.61) ;
%Straight Lines [id:da1839590120577328] 
\draw    (424,2799.44) -- (370.61,2852.84) ;
%Straight Lines [id:da46714376614367836] 
\draw    (424,2926.61) -- (370.61,2980) ;
%Straight Lines [id:da06831157969924972] 
\draw    (370.61,2980) -- (246.03,2980) ;
%Straight Lines [id:da41541170370943603] 
\draw [line width=3]    (370.61,2852.84) -- (370.61,2980) ;
%Straight Lines [id:da03028344533394911] 
\draw    (370.61,2852.84) -- (246.03,2852.84) ;
%Straight Lines [id:da8213516123197466] 
\draw    (299.42,2799.44) -- (246.03,2852.84) ;
%Straight Lines [id:da025734512785617225] 
\draw    (424,2799.44) -- (299.42,2799.44) ;
%Straight Lines [id:da186465824736753] 
\draw [line width=3]    (246.03,2852.84) -- (246.03,2980) ;
%Straight Lines [id:da918514409051447] 
\draw [color={rgb, 255:red, 0; green, 0; blue, 0 }  ,draw opacity=1 ][line width=3]  [dash pattern={on 7.88pt off 4.5pt}]  (299.42,2799.44) -- (299.42,2928.84) ;
%Straight Lines [id:da6208761372079388] 
\draw  [dash pattern={on 4.5pt off 4.5pt}]  (301.65,2926.61) -- (246.03,2982.23) ;
%Straight Lines [id:da610086549999129] 
\draw  [dash pattern={on 4.5pt off 4.5pt}]  (299.42,2926.61) -- (424,2926.61) ;
%Shape: Ellipse [id:dp2977970013806801] 
\draw  [fill={rgb, 255:red, 208; green, 2; blue, 27 }  ,fill opacity=1 ] (359.08,2851.04) .. controls (359.08,2844.57) and (364.21,2839.33) .. (370.54,2839.33) .. controls (376.87,2839.33) and (382,2844.57) .. (382,2851.04) .. controls (382,2857.52) and (376.87,2862.76) .. (370.54,2862.76) .. controls (364.21,2862.76) and (359.08,2857.52) .. (359.08,2851.04) -- cycle ;
%Shape: Ellipse [id:dp8891854526938032] 
\draw  [fill={rgb, 255:red, 255; green, 255; blue, 255 }  ,fill opacity=1 ] (290.19,2926.61) .. controls (290.19,2920.13) and (295.32,2914.89) .. (301.65,2914.89) .. controls (307.98,2914.89) and (313.11,2920.13) .. (313.11,2926.61) .. controls (313.11,2933.08) and (307.98,2938.32) .. (301.65,2938.32) .. controls (295.32,2938.32) and (290.19,2933.08) .. (290.19,2926.61) -- cycle ;
%Shape: Ellipse [id:dp17677074804176152] 
\draw  [fill={rgb, 255:red, 208; green, 2; blue, 27 }  ,fill opacity=1 ] (234.57,2852.84) .. controls (234.57,2846.37) and (239.7,2841.12) .. (246.03,2841.12) .. controls (252.35,2841.12) and (257.48,2846.37) .. (257.48,2852.84) .. controls (257.48,2859.31) and (252.35,2864.55) .. (246.03,2864.55) .. controls (239.7,2864.55) and (234.57,2859.31) .. (234.57,2852.84) -- cycle ;
%Shape: Ellipse [id:dp9367721235345978] 
\draw  [fill={rgb, 255:red, 208; green, 2; blue, 27 }  ,fill opacity=1 ] (287.96,2799.44) .. controls (287.96,2792.97) and (293.09,2787.73) .. (299.42,2787.73) .. controls (305.75,2787.73) and (310.88,2792.97) .. (310.88,2799.44) .. controls (310.88,2805.92) and (305.75,2811.16) .. (299.42,2811.16) .. controls (293.09,2811.16) and (287.96,2805.92) .. (287.96,2799.44) -- cycle ;
%Shape: Ellipse [id:dp7479191412700514] 
\draw  [fill={rgb, 255:red, 255; green, 255; blue, 255 }  ,fill opacity=1 ] (234.57,2980) .. controls (234.57,2973.53) and (239.7,2968.28) .. (246.03,2968.28) .. controls (252.35,2968.28) and (257.48,2973.53) .. (257.48,2980) .. controls (257.48,2986.47) and (252.35,2991.72) .. (246.03,2991.72) .. controls (239.7,2991.72) and (234.57,2986.47) .. (234.57,2980) -- cycle ;
%Shape: Ellipse [id:dp7729310493324129] 
\draw  [fill={rgb, 255:red, 255; green, 255; blue, 255 }  ,fill opacity=1 ] (359.15,2980) .. controls (359.15,2973.53) and (364.28,2968.28) .. (370.61,2968.28) .. controls (376.94,2968.28) and (382.07,2973.53) .. (382.07,2980) .. controls (382.07,2986.47) and (376.94,2991.72) .. (370.61,2991.72) .. controls (364.28,2991.72) and (359.15,2986.47) .. (359.15,2980) -- cycle ;
%Shape: Ellipse [id:dp8681827359403937] 
\draw  [fill={rgb, 255:red, 255; green, 255; blue, 255 }  ,fill opacity=1 ] (412.54,2926.61) .. controls (412.54,2920.13) and (417.67,2914.89) .. (424,2914.89) .. controls (430.33,2914.89) and (435.46,2920.13) .. (435.46,2926.61) .. controls (435.46,2933.08) and (430.33,2938.32) .. (424,2938.32) .. controls (417.67,2938.32) and (412.54,2933.08) .. (412.54,2926.61) -- cycle ;
%Shape: Ellipse [id:dp1343475850215008] 
\draw  [fill={rgb, 255:red, 208; green, 2; blue, 27 }  ,fill opacity=1 ] (412.54,2799.44) .. controls (412.54,2792.97) and (417.67,2787.73) .. (424,2787.73) .. controls (430.33,2787.73) and (435.46,2792.97) .. (435.46,2799.44) .. controls (435.46,2805.92) and (430.33,2811.16) .. (424,2811.16) .. controls (417.67,2811.16) and (412.54,2805.92) .. (412.54,2799.44) -- cycle ;

% Text Node
\draw (257.48,2980) node [anchor=north west][inner sep=0.75pt]  [font=\normalsize]  {$000$};
% Text Node
\draw (250.97,2859.02) node [anchor=north west][inner sep=0.75pt]  [font=\normalsize]  {$100$};
% Text Node
\draw (430,2932.32) node [anchor=north west][inner sep=0.75pt]  [font=\normalsize]  {$011$};
% Text Node
\draw (313.11,2926.61) node [anchor=north west][inner sep=0.75pt]  [font=\normalsize]  {$010$};
% Text Node
\draw (382.07,2980) node [anchor=north west][inner sep=0.75pt]  [font=\normalsize]  {$001$};
% Text Node
\draw (429,2803.32) node [anchor=north west][inner sep=0.75pt]  [font=\normalsize]  {$111$};
% Text Node
\draw (310.88,2799.44) node [anchor=north west][inner sep=0.75pt]  [font=\normalsize]  {$110$};
% Text Node
\draw (376,2857.04) node [anchor=north west][inner sep=0.75pt]  [font=\normalsize]  {$101$};
% Text Node
\draw (255,2785.6) node [anchor=north west][inner sep=0.75pt]    {$\mathcal{A}$};
% Text Node
\draw (190,2901.6) node [anchor=north west][inner sep=0.75pt]    {$\partial \mathcal{A}$};
\end{tikzpicture}
\caption{Hamming graph for $n=3$. For the dictator function $\mathrm{Maj}_1(x^3) = x_1$, the set $\calA = \{(1,0,0) ,(1,0,1), (1,1,0), (1,1,1) \}$ (indicated in red balls). The edge boundary $\partial\calA$ of $\calA$ is indicated as the  four thick edges. Each boundary edge is a dimension-$1$ edge.  }
\label{fig:ham_graph}
\end{figure}
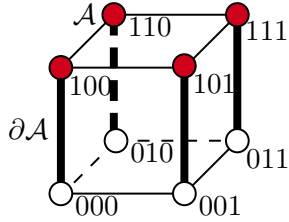

\begin{definition} For a set $\mathcal{A}\subset\{0,1\}^{n}$, define
its \emph{edge boundary} $\partial\mathcal{A}$ as the set of edges
in the Hamming graph such that one of its endpoints belongs to $\mathcal{A}$
while the other one belongs to $\mathcal{A}^{\rmc}$. Every edge that belongs to $\partial\calA$ is called a {\em boundary edge}. An boundary edge $\{x^n, y^n\} \in \partial \calA$ is known as a  {\em dimension-$i$ edge} if $y^n = (x^n)^{\oplus i}$, i.e., $x^n$ and $y^n$ are identical except in their $i^{\mathrm{th}}$ coordinates. \end{definition}

For a set $\mathcal{A}\subset\{0,1\}^{n}$, one observes the following
facts. 
\begin{enumerate}
\item \label{item:fraction_dim} The fraction of   dimension-$i$
 edges that are boundary edges of $\mathcal{A}$ 
in the Hamming graph is equal to $\mathbf{I}_{i}[\bone_{\mathcal{A}}]$.
\item \label{item:fraction} The   fraction of edges
in the Hamming graph that are boundary edges of $\calA$ is equal to $\frac{1}{n}\mathbf{I}[\bone_{\mathcal{A}}]$. This implies that  $|\partial\mathcal{A}|=2^{n-1}\mathbf{I}[\bone_{\mathcal{A}}]$,
since the total number of edges in the Hamming graph is $n\,2^{n-1}$. 
\end{enumerate}

\begin{example}
Let $n=3$. The Hamming graph is shown  in Fig.~\ref{fig:ham_graph}. This graph has $3\cdot 2^{3-1}=12$ edges. Consider the dictator function $\mathrm{Maj}_1(x^3) = x_1$. The support of $f$ is the set $\calA =  \{(1,0,0) ,(1,0,1), (1,1,0), (1,1,1) \}$. Both $\calA$ and $\partial\calA$ are  indicated in Fig.~\ref{fig:ham_graph} and $|\partial\calA|=4$. For this dictator function, $\bI_1[\bone_{\calA}]=1$ and $\bI_2[\bone_{\calA}]=\bI_3[\bone_{\calA}]=0$ (as discussed after Definition~\ref{def:influence}). Note from Fig.~\ref{fig:ham_graph} that there are four dimension-$1$ edges and no  dimension-$2 $ and dimension-$3$ edges. Thus, the {\em fractions} of dimension-$1$, dimension-$2$ and dimension-$3$ edges that are boundary edges of~$\calA$ are $1$, $0$,  and $0$   respectively, corroborating Fact~\ref{item:fraction_dim}. Furthermore, $\bI[\bone_{\calA}]= \sum_{i=1}^3\bI_i[ \bone_\calA ]=1$ and the fraction of edges that belong to $\partial \calA$ is $1/3=4/12$, corroborating Fact~\ref{item:fraction}.
\end{example}

\subsection{Edge-Isoperimetric Inequalities}\label{sec:edge}
From Fact \ref{item:fraction}, we see that  the total influence of $f$
is related to the cardinality of the edge boundary  of its support set $\calA$.  A classical result due to \citet{harper1964optimal} quantifies this relation via the   
%the latter is the 
%following 
so-called {\em edge-isoperimetric inequality}.
%which is
%due to \citet{harper1964optimal}. 

%For brevity, we express the inequality 
%in terms of the total influence. 

\begin{theorem}[Edge-isoperimetric inequality]  \label{thm:EdgeIsoper}
For $f:\{0,1\}^{n}\to\{0,1\}$ with $a=\min\{\mathbb{E}[f],1-\mathbb{E}[f]\}$,
\begin{equation}
\mathbf{I}[f]\ge 2a\, \log \Big(\frac{1}{a}\Big).\label{eq:isoperi}
\end{equation}
\end{theorem}

This inequality can be seen as a Boolean function version of the \emph{log-Sobolev inequality}. The relationship between this edge-isoperimetric inequality and the real-valued function version of log-Sobolev inequalities will be discussed extensively in Section \ref{sec:log_sobolev}. 

This inequality in \eqref{eq:isoperi} is sharp for $a=2^{-k}$ with $1\le k\le n$, since
for this case, the indicator function of an $(n-k)$-subcube attains
the lower bound. 
%In fact, the precisely optimal edge-isoperimetric
%inequality exists which implies that 
If $a$, a dyadic rational, is the mean of $f$, $\mathbf{I}[f]$ is minimized when $f$ is the
indicator of a lexicographic set of size $2^{n}a$ (cf.\ Section~\ref{sec:subcubes}). The edge-isoperimetric
inequality will be used to resolve the extreme cases of the max $q$-stability problem
via the following two theorems that establish a connection between the $q$-stability and the total
influence. 

\begin{theorem} \label{thm:connect_qstab_TI} For $f:\{0,1\}^{n}\to\{0,1\}$, 
\begin{equation}
\mathbf{S}^{(2)}_{\rho}[f]=\sum_{\calS\subset[n]}\rho^{2|\calS|}\ \hat{f}_{\calS}^{2}=\sum_{k=0}^{n}\rho^{2k}\ \mathbf{W}_{k}[f].
\end{equation}
\end{theorem}

\begin{proof}
This theorem follows by \eqref{eq:inner_prod_chi} and the facts that 
\begin{equation}
\mathbf{S}^{(2)}_{\rho}[f]= \langle T_{\rho}f,T_{\rho}f \rangle \qquad  \mbox{and} \qquad(\widehat{T_{\rho}f})_{\calS}= \sum_{\calS\subset[n]} \rho^{|\calS|}\, \hat{f}_{\calS}, \label{eq:FourierTf}
\end{equation}
where $\{(\widehat{T_{\rho}f})_{\calS} \}_{\calS\subset [n]}$ are the Fourier coefficients of $T_\rho f$.
\end{proof}
This theorem implies that 
\begin{align}
 \frac{\mathrm{d}}{\mathrm{d}\rho}\mathbf{S}^{(2)}_{\rho}[f]\Big|_{\rho=1}  &=\frac{\mathbf{I}[f]}{2}, \\
 \frac{\mathrm{d}}{\mathrm{d}\rho}\mathbf{S}^{(2)}_{\rho}[f]\Big|_{\rho=0}  &=0, \qquad\mbox{and}\\
 \frac{\mathrm{d}^2}{\mathrm{d}\rho^2}\mathbf{S}^{(2)}_{\rho}[f]\Big|_{\rho=0} & =2 \,\mathbf{W}_{1}[f].
\end{align}
Theorem~\ref{thm:connect_qstab_TI} pertains to $q=2$. For general $q>1$, the derivatives of $\mathbf{S}^{(q)}_{\rho}[f]$ at $\rho =0$ and  $1$ are given in the following theorem which  is due to \citet{li2021boolean}.

\begin{theorem} For $f:\{0,1\}^{n}\to\{0,1\}$ with mean $a\in(0,1]$, 
\begin{align}
 \frac{\mathrm{d}}{\mathrm{d}\rho}\mathbf{S}^{(q)}_{\rho}[f]\Big|_{\rho=1} & =\frac{q}{4}\ \mathbf{I}[f], \label{eq:1nd_der} \\ 
 \frac{\mathrm{d}}{\mathrm{d}\rho}\mathbf{S}^{(q)}_{\rho}[f]\Big|_{\rho=0}  & =0,  \qquad\mbox{and}\\ 
 \frac{\mathrm{d}^2}{\mathrm{d}\rho^2}\mathbf{S}^{(q)}_{\rho}[f]\Big|_{\rho=0} & =q(q-1)\ a^{q-2}\ \mathbf{W}_{1}[f]. \label{eq:2nd_der}
\end{align}
\end{theorem}
\begin{proof}
By using the Fourier-analytic relations in  \eqref{eq:FourierTf},  we obtain
\begin{equation} 
T_{\rho}f(y^{n})=\sum_{\calS\subset[n]}\rho^{|\calS|}\,\hat{f}_{\calS}\,\chi_\calS(y^{n}).
\end{equation}
By the definition of the $q$-stability in \eqref{eq:q-stability} 
\begin{equation}
\mathbf{S}^{(q)}_{\rho}[f]=\mathbb{E}_{Y^{n}}\Bigg[\bigg(\sum_{\calS\subset[n]}\rho^{|\calS|}\hat{f}_{\calS}\chi_{\calS}(Y^{n})\bigg)^{q}\Bigg].
\end{equation}
Differentiating this with respect to $\rho$ yields 
\begin{align}
\frac{\rmd}{\rmd\rho}\mathbf{S}^{(q)}_{\rho}[f]  \! =\! q\ \mathbb{E}_{Y^{n}}\Bigg[\big(T_{\rho}f(Y^{n})\big)^{q-1}\! \sum_{\calS\subset[n]:|\calS|\ge 1}\! |\calS|\rho^{|\calS|-1}\hat{f}_{\calS}\chi_{\calS}(Y^{n})\Bigg].
\end{align}
Setting $\rho=1$, we obtain 
\begin{align}
\frac{\rmd}{\rmd\rho}\mathbf{S}^{(q)}_{\rho}[f]\Big|_{\rho=1} & =q\ \mathbb{E}_{Y^{n}}\Bigg[f(Y^{n})^{q-1}\sum_{\calS\subset[n]:|\calS|\ge 1}|\calS|\hat{f}_{\calS}\chi_{\calS}(Y^{n})\Bigg]\\
 & =q\ \mathbb{E}_{Y^{n}}\Bigg[f(Y^{n})\sum_{\calS\subset[n]:|\calS|\ge 1}|\calS|\hat{f}_{\calS}\chi_{\calS}(Y^{n})\Bigg]\label{eq:-8}\\
 & =q\ \sum_{\calS\subset[n]:|\calS|\ge 1}|\calS|\hat{f}_{\calS}^{2}\\
 & =\frac{q}{4}\ \mathbf{I}[f], \label{eqn:influ}
\end{align}
where \eqref{eq:-8} follows since $f$ only takes values in $\{0,1\}$,
and hence, $f^{q-1}=f$, and \eqref{eqn:influ} follows from \eqref{eq:InfluenceFourier}.  This proves~\eqref{eq:1nd_der}. The other equalities can be proved similarly. 
%The equations \eqref{eq:1nd_der} and  \eqref{eq:2nd_der} can be proven similarly. 
\end{proof}

\subsection{Max $q$-Stabilities in Extreme Cases of $\rho$}

Based on the concept of the total influence and the results stated in Sections~\ref{sec:influ} and~\ref{sec:edge}, we are now ready to analyze
the extreme cases of the max $q$-stability as $\rho\downarrow0$ and $\rho\uparrow1$. We first state a
lower bound on the derivative of the $q$-stability with respect to  $\rho$
evaluated at $\rho=1$. This result is due to \citet{mossel2005coin}
for integer $q$ and \citet{li2021boolean} for real
   $q>1$. 

\begin{theorem} \label{thm:StabDerivative} Let $q>1$. For a Boolean
function $f:\{0,1\}^{n}\to\{0,1\}$ with mean $a$,
\begin{equation}
 \frac{\partial}{\partial\rho}\mathbf{S}^{(q)}_{\rho}[f]\Big|_{\rho=1}\ge\frac{q}{2}\,a\,\log\Big(\frac{1}{a}\Big).
\end{equation}
This lower bound is attained if $a=2^{-k}$ for any $1\le k\le n$ and
$f$ is the indicator of an $(n-k)$-subcube. 
\end{theorem}

This theorem follows by the edge-isoperimetric inequality in~\eqref{eq:isoperi} and~\eqref{eq:1nd_der}. 
%, we have
%\begin{equation}
%\left.\frac{\partial}{\partial\rho}\mathbf{S}^{(q)}_{\rho}[f]\right|_{\rho=1}\ge\frac{1}{2}qa\log(1/a).
%\end{equation}
%This lower bound is attained if and only if $f$ is the indicator
%of a subcube.  
%\end{proof}

Note that if $\rho=1$, then for any $f:\{0,1\}^{n}\to\{0,1\}$ with
mean $a$, it holds that $\mathbf{S}^{(q)}_{\rho}[f]=a$ (cf.\ \eqref{eqn:simpleS}). Hence, from Theorem
\ref{thm:StabDerivative}, it is plausible, via ``continuity arguments'', that   if $a=2^{-k}$ for integer $k$ and $\rho$
is sufficiently close to $1$, then $\mathbf{S}^{(q)}_{\rho}[f]$ is maximized
by the indicator of an $(n-k)$-subcube. This  can be proven
rigorously using the fact that the number of Boolean functions for a given $n$ is   finite, some approximation arguments involving Taylor's theorem, and bounds on the derivative of $\bS_\rho^{(q)}$ evaluated at $\rho=1$ (Theorem \ref{thm:StabDerivative}).
This is stated formally in the following theorem which is  due
to  \citet{mossel2005coin} for integer $q$ and
 \citet{li2021boolean} for real   $q$. 

\begin{theorem} \label{thm:StabDerivative-1} Fix $n\ge1$, $q>1$,
and $a=2^{-k}$ with $1\le k\le n$. There exists an $\epsilon \in (0,1)$
such that for all $\rho\in[1-\epsilon,1]$, ${\Gamma}_{\rho}^{(q)}(a)$
is attained by the indicator of an $(n-k)$-subcube. \end{theorem}

\begin{proof}[Proof Sketch of Theorem~\ref{thm:StabDerivative-1}] Fix a  Boolean function $f$ and $\rho \in (0,1)$.   Using Taylor's theorem,
we can write 
\begin{equation}
\mathbf{S}_{\rho}^{(q)}[f]=\mathbf{S}_{1}^{(q)}[f]+(\rho-1) \frac{\partial}{\partial\rho}\mathbf{S}^{(q)}_{\rho}[f]\Big|_{\rho=1}+\phi_{f}(\tilde{\rho})(\rho-1)^{2},
\end{equation}
where $\phi_{f}:[0,1]\to\mathbb{R}$ is a bounded function induced
by $f$ and $\tilde{\rho} \in (\rho,1)$.  Since $n$ is fixed, the   number of Boolean functions
$f:\{0,1\}^{n}\to\{0,1\}$ is finite. % Let $\phi(\rho):=\max_{f }\phi_{f}(\rho)$ where the maximization runs over all Boolean functions with domain $\{0,1\}^n$. 
From this fact, we deduce that $\phi(\rho):=\max_{f:\{0,1\}^{n}\to\{0,1\}}\phi_{f}(\rho)$,
then $\phi$ is bounded, i.e., there is some constant $c_2$
such that $|\phi(\rho)|\le c_2$ for all $\rho \in [0,1]$. Moreover, if $f$
is not the indicator of an $(n-k)$-subcube, it holds that  (cf.\ Theorem~\ref{thm:StabDerivative})
\begin{equation}
\left.\frac{\partial}{\partial\rho}\mathbf{S}^{(q)}_{\rho}[f]\right|_{\rho=1}>\frac{q}{2}\,a\,\log\Big(\frac{1}{a}\Big).
\end{equation}
By again exploiting that fact that the number of Boolean functions is finite,
\begin{equation}
c_1:=\min_{f\in\mathcal{F}}\left.\frac{\partial}{\partial\rho}\mathbf{S}^{(q)}_{\rho}[f]\right|_{\rho=1}>\frac{q}{2}\,a\,\log\Big(\frac{1}{a}\Big),
\end{equation}
where $\mathcal{F}$ denotes the set of Boolean functions $f:\{0,1\}^{n}\to\{0,1\}$
that {\em cannot} be written as the indicator of an $(n-k)$-subcube. Therefore,
for any $f\in\mathcal{F}$,
\begin{equation}
\mathbf{S}_{\rho}^{(q)}[f]\le a+c_{1}(\rho-1)+ c_2(\rho-1)^{2}.\label{eq:-9}
\end{equation} 

By Taylor's theorem, one can lower bound the $q$-stability for the indicator
of an $(n-k)$-subcube $\bbC_{n-k}$ as  
\begin{equation}
\mathbf{S}_{\rho}^{(q)}[\bone_{\mathbb{C}_{n-k}}]\ge a+(\rho-1)\frac{q}{2}\,a\,\log\Big(\frac{1}{a}\Big)+c_3(\rho-1)^{2},\label{eq:-10}
\end{equation}
where $c_3$ is an absolute constant independent of $\rho$. Comparing~\eqref{eq:-9}
and \eqref{eq:-10}, we observe that  there exists a constant $\epsilon>0$
such that the right-hand side of \eqref{eq:-10} is larger than \eqref{eq:-9}
for all $\rho\in[1-\epsilon,1]$, concluding the proof sketch of Theorem~\ref{thm:StabDerivative-1}. \end{proof}
%\fi

Concerning the other extreme case, i.e., the limiting case as $\rho\downarrow0$,
following the proof ideas used in Theorems \ref{thm:StabDerivative}
and \ref{thm:StabDerivative-1}, one can also show the following result, which is
due to \citet{mossel2005coin} and  
\citet{li2021boolean}. % We omit the proof here. 

\begin{theorem} \label{thm:StabDerivative-1-1} Fix $n\ge1$, $q>1$,
and a dyadic rational $a\in(0,1)$. There exists an $\epsilon \in (0,1)$
such that for all $\rho\in[0,\epsilon]$, ${\Gamma}_{\rho}^{(q)}(a)$
is attained by  some Boolean function that maximizes the degree-$1$
Fourier weight~$\mathbf{W}_{1}$. In particular, for $a=1/2$, there
exists an $\epsilon>0$ such that for all $\rho\in[0,\epsilon]$,
${\Gamma}_{\rho}^{(q)}(1/2)$ is attained by  dictator functions.
\end{theorem}

Theorems \ref{thm:StabDerivative}, \ref{thm:StabDerivative-1},
and \ref{thm:StabDerivative-1-1} can be   extended to their symmetric
counterparts of the max $q$-stability. For $q=1$, they can also be extended to 
 the $\Phi$-versions of the max $q$-stabilities (cf.\ Definition~\ref{def:Phi_versions});  
see \citet{courtade2014boolean}, \citet{ordentlich2016improved}, and \citet{yang2019most}. 

\section{The Balanced Case }
\label{sec:balanced}

In this section, we consider   the balanced case, i.e., $a=1/2$,
and  discuss  recent progress on Conjectures~\ref{conj:AsymmetricStability}
and~\ref{conj:SymmetricStability}. We first focus on the case $q=1$,
i.e., the balanced version of the Courtade--Kumar conjecture, which can be 
%of the Courtade--Kumar conjecture~\cite{courtade2014boolean} can
 stated as follows.

\begin{conjecture} \label{conj:ck} For any $n\in\bbN$ and $\rho\in(0,1)$,
\begin{equation}
\max_{\textrm{Boolean }f:\mathbb{E}[f(X^{n})]=1/2}I(f(X^{n});Y^{n})\stackrel{?}{=}1-h\Big(\frac{1-\rho}{2}\Big).\label{eq:NICDCKconjecture}
\end{equation}
\end{conjecture} 

In the original version of Courtade--Kumar conjecture,
the Boolean function $f$ is not required to satisfy $\mathbb{E}[f(X^{n})]=1/2$. It 
 has been numerically verified to be true for all $n\le7$~\cite{courtade2014boolean}.
An old result by Witsenhausen and Wyner \cite{witsenhausen1975conditional}
(also see \citet{erkip1996efficiency}) yields the following bound.

\begin{proposition} It holds that 
\begin{equation}
\max_{\textrm{Boolean }f:\mathbb{E}[f(X^{n})]=1/2}I(f(X^{n});Y^{n})\le\rho^{2}.\label{eqn:ww}
\end{equation}
\end{proposition}

This proposition can be proved via the so-called Mrs.\ Gerber's lemma~\cite{wyner1973theorem}
or   the hypercontractivity inequality in \eqref{eq:NICDFHC-1}.
Here, we provide a short justification based on  the latter. By~\eqref{eq:NICDFHC-1}, we
obtain that for $q>1$ and any Boolean function $f$ with mean $a$,
\begin{align}
\mathbf{S}^{(q)}_{\rho}[f] & \le a^{\frac{q}{1+(q-1)\rho^{2}}}.\label{eq:NICDMDstability-2}
\end{align}
In other words,
\begin{equation}
{\Gamma}_{\rho}^{(q)}(a)\le a^{\frac{q}{1+(q-1)\rho^{2}}}\quad\mbox{and}\quad  
 \breve{\Gamma}_{\rho}^{(q)}(a)\le a^{\frac{q}{1+(q-1)\rho^{2}}}+\bara^{\frac{q}{1+(q-1)\rho^{2}}}.
\end{equation}
Substituting the latter into \eqref{eqn:qge1} and setting $a=1/2$ yields 
\begin{align}
\breve{\Pi}_{\rho}^{(q)}(1/2)\le\frac{2^{\frac{(1-q)(1-\rho^{2})}{1+(q-1)\rho^{2}}}-1}{(q-1)\ln2}.\label{eqn:qge1-1}
\end{align}
Letting $q\downarrow1$, we obtain $\breve{\Pi}_{\rho}^{(1)}(1/2)\le\rho^{2}-1$.
Substituting this into \eqref{eq:NICD-8} and noting that $h(1/2)=1$ yields \eqref{eqn:ww} as desired.

Considering small $\rho$, and using   Fourier
analysis and hypercontractivity, \citet{ordentlich2016improved} improved
the bound in \eqref{eqn:ww} to the following.

\begin{proposition} For $0\le\rho\le{1}/{\sqrt{3}}$, 
\begin{equation}
\max_{\textrm{Boolean }f:\mathbb{E}[f(X^{n})]=1/2}\!I(f(X^{n});Y^{n})\le\frac{\log e}{2}\rho^{2}+9\Big(1-\frac{\log e}{2}\Big)\rho^{4}.\label{eqn:ord}
\end{equation}
\end{proposition} The bounds in \eqref{eqn:ww} and \eqref{eqn:ord}
are illustrated in Fig.~\ref{fig:ck}. The bound in \eqref{eqn:ord}
is better than \eqref{eqn:ww} in the range $0<\rho<1/3$. Moreover,
the bound in \eqref{eqn:ord} is asymptotically tight as $\rho\downarrow0$,
i.e., the ratio of the bound in \eqref{eqn:ord} and the right-hand side of \eqref{eq:NICDCKconjecture}
converges to $1$ as $\rho\downarrow0$. This point can be seen from
the fact that by Taylor's theorem, as $\rho\downarrow0$, 
\begin{equation}
1-h\Big(\frac{1-\rho}{2}\Big)=\frac{\log e}{2}\rho^{2}+\frac{\log e}{12}\rho^{4}+O(\rho^{6}).
\end{equation}

%Thus far,   all partial resolutions  of the Courtade--Kumar conjecture  reviewed above  
%depend on the dimension  $n$. Proving the Courtade--Kumar conjecture for the correlation coefficient  $\rho$
% in a dimension-independent interval turns out to be very challenging. 

In 2016, \citet{samorodnitsky2016entropy} made a significant breakthrough
on   the Courtade--Kumar conjecture. Specifically, he proved the existence of a {\em dimension-independent} interval  for which Conjecture \ref{conj:ck} holds for all $\rho$ in the interval. 

\begin{figure}[!ht]
\centering \includegraphics[width=0.8\columnwidth]{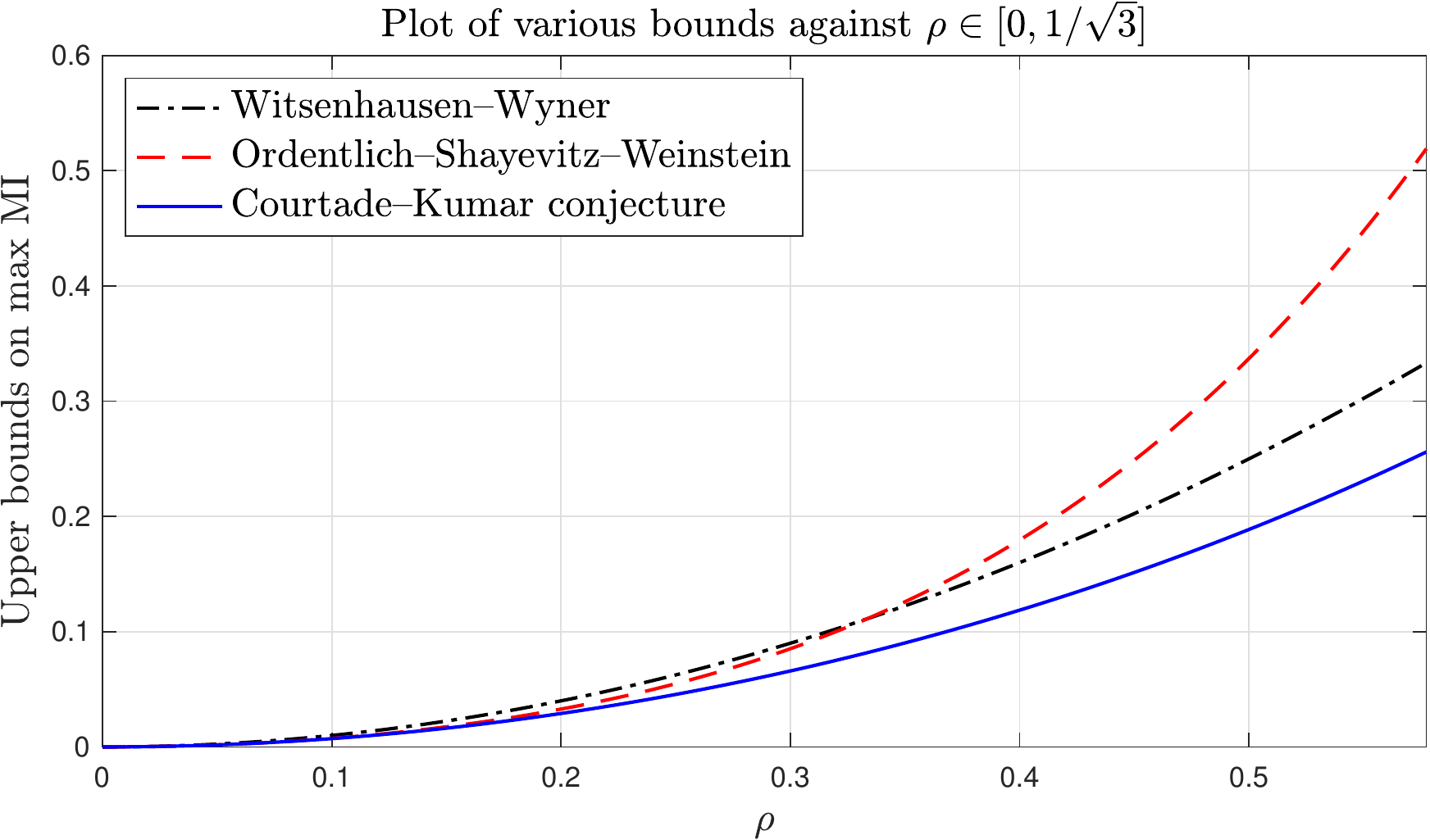}
\caption{Illustration bounds on $\max I(f(X^{n});Y^{n})$ by \citet{witsenhausen1975conditional}
in \eqref{eqn:ww}, \citet{ordentlich2016improved} in \eqref{eqn:ord},
and the Courtade--Kumar conjecture in \eqref{eq:NICDCKconjecture}}
\label{fig:ck}
\end{figure}

\begin{theorem}\label{thm:sam} There exists a constant $0<\rho_{0}<1$
(independent of $n$), such that \eqref{eq:NICDCKconjecture} holds
for any $n\in\bbN$ and any $\rho\in(0,\rho_{0}]$. \end{theorem} 

The proof by \citet{samorodnitsky2016entropy}  is based on  Fourier analysis,
random restrictions, techniques in \citet{ordentlich2016improved},
the Friedgut--Kalai--Naor (FKN) theorem \cite{kahn1988influence},
among others. Samorodnitsky's proof is highly technical,
so we do not present it here.   However, we should note that
in the proof of Theorem \ref{thm:sam}, $\rho_{0}$, which is not explicitly provided, is assumed to
be ``sufficiently small''. It is also worth noting that the conclusion that the value
$\rho_{0}$ is independent of $n$ (resulting in a {\em dimension-independent} interval $(0,\rho_0]$)  is the crux of this theorem.
Indeed, if we allow  $\rho_{0}$ to vary with $n$, then
the resulting theorem is merely an extension of Theorem \ref{thm:StabDerivative-1-1}
to the case $q=1$, which can be proved by combining the bound in~\eqref{eqn:ord}
by \citet{ordentlich2016improved} and the discreteness of the space
of Boolean functions; see \cite[Corollary 1]{ordentlich2016improved}.
This fact can also be deduced   using calculus \cite{yang2019most}.

Using Fourier analysis and optimization theory, the first author of
this monograph \cite{yu2021phi} provided an explicit threshold for Theorem \ref{thm:sam}. Specifically, he showed that \eqref{eq:NICDCKconjecture}
holds for any $n$ and any $\rho\in(0,\rho_{1}]$, where $\rho_{1}$
be the solution in $(0,1)$ to the equation 
\begin{equation}
(1+\rho^{2})\log\Big(\frac{1+\rho}{2}\Big)-(1-\rho)^{2}\log\Big(\frac{1-\rho}{2}\Big)=0.\label{eq:NICDpsi-1}
\end{equation}
The value of $\rho_{1}\approx0.461491$. %Besides,
%he also improved this threshold further by incorporating the FKN's
%theorem \cite{kahn1988influence} in his proof. 

In the Courtade--Kumar conjecture, if the Boolean function is set
to a dictator function $f(x^{n})=x_{1}$ (say), then the objective
function $I(f(X^{n});Y^{n})=I(X_{1};Y_{1})$. Motivated by this, in addition to 
the original Courtade--Kumar conjecture (in which $f$ is an arbitrary Boolean function and not required to
satisfy $\bbE[f(X^{n})]=1/2$), Courtade and Kumar also proposed a
weaker version of this conjecture. They conjectured that for any $n\in\bbN$
and $\rho\in(0,1)$, 
\begin{equation}
\max_{\textrm{Boolean }f,g}I(f(X^{n});g(Y^{n}))=1-h\Big(\frac{1-\rho}{2}\Big).\label{eq:NICDCKconjecture-1}
\end{equation}
This weaker version was proven by \citet{pichler2018dictator} by
using Fourier analysis and a novel partitioning technique.

\begin{theorem} The equality in~\eqref{eq:NICDCKconjecture-1} holds
for all $(n, \rho)\in\bbN\times (0,1)$. \end{theorem} 

Since the Li--M\'edard conjecture was only recently posed (at the
time of writing), there is less progress on it compared to the Courtade--Kumar
conjecture. Hence, we do not elaborate on it  apart from mentioning
some partial progress by \citet{yu2021phi} for a certain set of~$(q,\rho)$.

%As for the Li--M\'edard conjecture, observe that the Li--M\'edard conjecture
%is so new that there are very few related works in the literature
%that were conducted after this conjecture was posed. Moreover, related
%works conducted before the Li--M\'edard conjecture was posed are mainly
%the works related to the Courtade--Kumar conjecture, which have already
%been introduced above. Hence, we do not introduce related works on
%the Li--M\'edard conjecture. However, we should note that in \cite{yu2021phi},
%Yu proved the Li--M\'edard conjecture for a certain region of $(q,\rho)$. 

Finally, we summarize some recent progress on the Mossel--O'Donnell conjecture,
which states that dictator functions are optimal in attaining both
the asymmetric and symmetric max $q$-stabilities for $2<q\le9$
(and for any $n\in\bbN$ and any $\rho\in(0,1)$). As discussed  in Theorems
\ref{thm:StabDerivative-1} and~\ref{thm:StabDerivative-1-1}, the
limiting cases as $\rho\downarrow0$ and $\rho\uparrow1$ (with fixed $n$) were resolved in~\cite{mossel2005coin} for the symmetric case
and in \cite{mossel2005coin, li2021boolean} for the asymmetric case.
However, for other intermediate values of~$\rho$, there has been
fairly limited progress. % on 
% the Mossel--O'Donnell conjecture ever
% since it was posed. 
 For the symmetric case, the best known result is
Mossel and O'Donnell's result in Theorem \ref{thm:mos-stability}; this result 
resolved the eponymous conjecture for $q=3$ and for any $\rho \in (0,1)$.
Combining this with the result of \citet{barnes2020courtade} (in
Lemma~\ref{lem:BarnesOzgur}) yields the conclusion the Mossel--O'Donnell conjecture
holds for all $2<q\le3$. There is even less progress for the
asymmetric case in which the   best known result remains
that of Witsenhausen's result in Corollary~\ref{thm:wit-stability} for the
case $q=2$.
% (cf.\ Section~\ref{sec:clt_half}). 
 Recently, in \cite{yu2021phi},
the first author of this monograph made some progress on the Mossel--O'Donnell
conjecture. He showed that the symmetric version of the Mossel--O'Donnell
conjecture holds for $2<q\le5$, and the asymmetric version holds
for $2<q\le3$. These imply that $3\le q_{\max}<10$ and $5\le \breve{q}_{\max}<10$.
The proofs are based on Fourier analysis and optimization theory.

\section{Moderate and Large Deviations Regimes}
\label{sec:MD-and-LD}

In this section, we consider the max $q$-stabilities in the MD and LD regimes.
Recall  the definition of  the asymmetric max $q$-stability with $q\in[1,\infty)$  in \eqref{eqn:asymp_max_q}.
% By using the definition of
%the $L^{q}$ norm for random variables, %(recall $\|X\|_{q}:=(\bbE[|X|^{q}])^{1/q}$),
It can be rewritten as  
\begin{align}
{\Gamma}_{\rho}^{(q)}(a)%
 & =\bigg(\max_{\calA  \subset\{0,1\}^n:\pi_{X}^{n}(\calA)\le a}\Vert \pi_{X|Y}^{n}(\calA|Y^n)\Vert_{q}\bigg)^{q} , \label{eqn:asym_max_lq}
\end{align}
where the maximization is over all subsets of $\{0,1\}^n$. We now extend the asymmetric max $q$-stability  to the case of  $q\in(-\infty,1)\backslash\{0\}$. 
%We now define a counterpart of~\eqref{eqn:asym_max_lq} by replacing
%the maximization with a minimization, and at the same time, reversing
%the inequality in the constraint. 
%Define the {\em asymmetric
%minimal $q$-stability} (with $q\in(-\infty,1)\backslash\{0\}$)
For $q\in(-\infty,1)\backslash\{0\}$, define 
\begin{align}
{\Gamma}_{\rho}^{(q)}(a) & :=\bigg(\min_{\calA \subset\{0,1\}^n:\pi_{X}^{n}(\calA)\ge a}\Vert \pi_{X|Y}^{n}(\calA|Y^n)\Vert_{q}\bigg)^{q}.\label{eqn:asym_min_lq}
\end{align}
%We note that this definition is the same as  that in \eqref{eqn:asym_max_lq} except that the $\max$ in \eqref{eqn:asym_max_lq} is replaced by a $\min$. 
We note that even though a $\min$ is present in \eqref{eqn:asym_min_lq}, we still term this quantity as the asymmetric {\em max} $q$-stability. 

We are now interested in the MD and LD asymptotics of~\eqref{eqn:asym_max_lq} and~\eqref{eqn:asym_min_lq}. Similarly to the $2$-user NICD problem,
in the LD regime, the parameter $a$ is assumed to vanish exponentially fast
as $n\to\infty$, i.e., $a=2^{-n\alpha}$ for some fixed constant
$\alpha\in(0,1)$. In the MD regime, $a$ is assumed to vanish subexponentially
fast, i.e., $a=2^{-\theta_{n}\alpha}$ for an MD sequence $\{\theta_n\}_{n\in\bbN}$. 
%
%some positive sequence
%$\{\theta_{n}\}_{n\in\bbN}$ such that $\theta_{n}\to\infty$ and
%${\theta_{n}}/{n}\to0$, for some fixed constant $\alpha\in(0,\infty)$.

\begin{definition} \label{def:LDMD} We define the LD and MD exponents corresponding to the quantities in \eqref{eqn:asym_max_lq} and \eqref{eqn:asym_min_lq} as follows.
\begin{enumerate}
\item  For $n\ge1$, $\alpha\in[0,1]$, and $q\ge 1$, define the {\em LD exponent}
as 
\begin{align}
\hspace{-.3in}{\Upsilon}_{q,\mathrm{LD}}^{(n)}(\alpha) & :=-\frac{1}{n}\log\max_{\substack{\calA:\pi_{X}^{n}(\calA)\leq2^{-n\alpha}}
}\Vert \pi_{X|Y}^{n}(\calA|Y^n)\Vert_{q}.\label{eq:NICD-FLD}
\end{align}
For $q\in(-\infty,1)\backslash\{0\}$,  ${\Upsilon}_{q,\mathrm{LD}}^{(n)}(\alpha)$ is defined similarly but with the maximization in \eqref{eq:NICD-FLD} replaced by  a minimization, and 
the inequality reversed.  
%\,\, q\ge 1, \quad\mbox{and
%\\
%{\Upsilon}_{q,\mathrm{LD}}^{(n)}(\alpha) & :=-\frac{1}{n}\log\min_{\substack{\calA\subset\{0,1\}^{n}:\pi_{X}^{n}(\calA)\geq2^{-n\alpha}}
%}\Vert \pi_{X|Y}^{n}(\calA|Y^n)\Vert_{q}, \,\, q\in(-\infty,1)\backslash\{0\}.\label{eq:NICD-71-2-1}
\item  For $n\ge1$, $\alpha\in[0,\infty)$, $q\ge1$, and an MD sequence $\{\theta_{n}\}_{n\in\bbN}$, define
the {\em MD exponent} as 
\begin{align}
\hspace{-.3in}{\Upsilon}_{q,\mathrm{MD}}^{(n)}(\alpha) :=-\frac{1}{\theta_{n}}\log \max_{\substack{\calA :\pi_{X}^{n}(\calA)\leq2^{-\theta_{n}\alpha}}
}\Vert \pi_{X|Y}^{n}(\calA|Y^n)\Vert_{q} .\label{eq:NICD-RMD}
\end{align}
For $q\in(-\infty,1)\backslash\{0\}$,  ${\Upsilon}_{q,\mathrm{MD}}^{(n)}(\alpha)$ is defined similarly but with the maximization in \eqref{eq:NICD-RMD} replaced by  a minimization, and 
the inequality reversed.  
%\,\, q\ge 1,\;\;\mbox{and} \\
%{\Upsilon}_{q,\mathrm{MD}}^{(n)}(\alpha) & :=-\frac{1}{\theta_{n}}\log\min_{\substack{\calA\subset\{0,1\}^{n}:\pi_{X}^{n}(\calA)\geq2^{-\theta_{n}\alpha}}
%}\!\Vert \pi_{X|Y}^{n}(\calA|Y^n)\Vert_{q}, \,\, q\in(-\infty,1)\backslash\{0\}.\label{eq:NICD-RMD}
\item  Define ${\Upsilon}_{q,\mathrm{MD}}^{(\infty)}$ 
and ${\Upsilon}_{q,\mathrm{LD}}^{(\infty)}$ as the pointwise
limits of  \eqref{eq:NICD-FLD} and~\eqref{eq:NICD-RMD} as $n\to\infty$. 
\end{enumerate}
\end{definition}

Note that in the definitions in \eqref{eq:NICD-FLD}--\eqref{eq:NICD-RMD}, we remove the $q^{\mathrm{th}}$ power in \eqref{eqn:asym_max_lq}--\eqref{eqn:asym_min_lq}. This slight modification will result in a multiplicative factor of $q$ in the characterizations of these exponents.
% in the expressions of bounds on the LD and MD exponents. 
%which will affect a factor $q$ in the expression of the LD and MD
%exponents. 
We  deliberately choose such definitions since the bounds on the exponents  in Definition~\ref{def:LDMD}
 provided in the following  two theorems will be  consistent with
the bounds for the $2$-user NICD problem. We also remark that these quantities depend on $\rho$ but these dependencies are suppressed to avoid notational clutter in what follows.

%&  \includegraphics[width=0.45\columnwidth]{figs/LDq_plot} 
%, 
%and ${\Upsilon}_{q,\mathrm{LD}}$ appears to be convex for $q>1$, concave for $q<1$ but $q\neq 0$, and linear for $q=1$.    The convexity
%and concavity properties of~${\Upsilon}_{\mathrm{LD}}$  were proven in~\cite{yu2021strong}.

For $q\in(1-\rho^{-2},\infty)\backslash\{0\}$ and $\alpha>0$,  let
\begin{align}
{\Upsilon}_{q,\mathrm{MD}}(\alpha):= &  \frac{\alpha}{1+(q-1)\rho^{2}}.\label{eq:UpsilonMD}
\end{align}
By using the single-function versions of hypercontractivity inequalities
(Theorem~\ref{thm:hyper_single}), we can  obtain the following
result. \begin{theorem}[$q$-stability] \label{thm:strong-q-stability}
Let $n\ge1$ and $\ensuremath{\alpha>0}$. For $q\ge1$, 
\begin{align}
{\Upsilon}_{q,\mathrm{MD}}^{(n)}(\alpha) & \ge{\Upsilon}_{q,\mathrm{MD}}(\alpha),\label{eq:NICDMDstability}
\end{align}
and for $q\in(1-\rho^{-2},1)\backslash\{0\}$, 
\begin{align}
{\Upsilon}_{q,\mathrm{MD}}^{(n)}(\alpha) & \leq{\Upsilon}_{q,\mathrm{MD}}(\alpha).\label{eq:NICDMDstability2}
\end{align}
Moreover, these two bounds are asymptotically tight, i.e., for $q\in(1-\rho^{-2},\infty)\backslash\{0\}$, 
\begin{equation}
{\Upsilon}_{q,\mathrm{MD}}^{(\infty)}(\alpha)={\Upsilon}_{q,\mathrm{MD}}(\alpha). \label{eq:Exp_MD}
\end{equation}
Lastly, for $q\in(-\infty, 1-\rho^{-2}]$, 
$
{\Upsilon}_{q,\mathrm{MD}}^{(\infty)}(\alpha)= \infty.
$
The equalities   are achieved by sequences of Hamming balls or spherical shells. 
\end{theorem} 
\begin{figure}
\centering %
%\begin{tabular}{cc}
 \includegraphics[width=0.75\columnwidth]{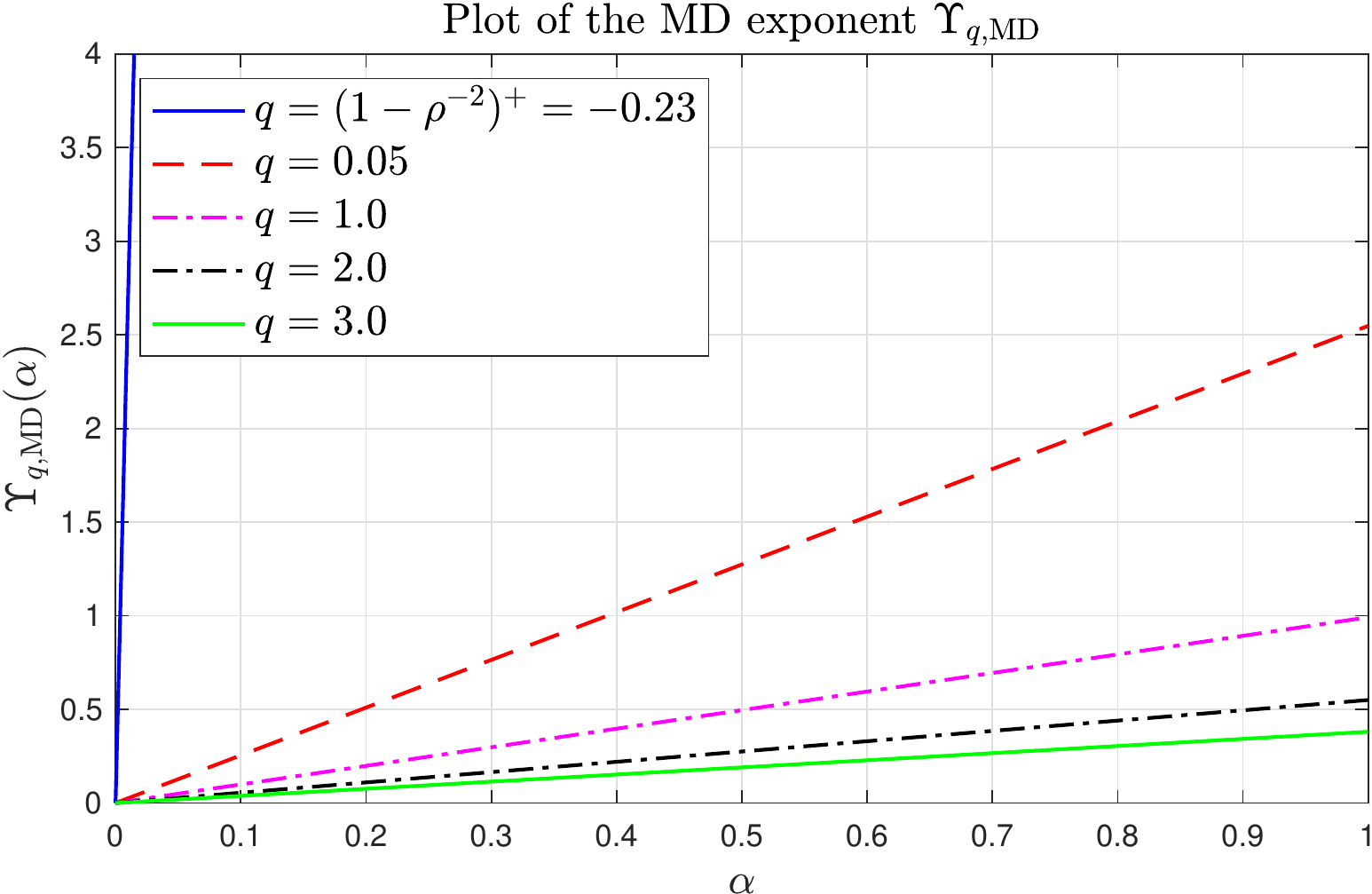}  
% \tabularnewline
% \multicolumn{3}{c}{{\footnotesize{}$\underline{\Upsilon}_{\mathrm{CL}}(\alpha,\beta)$}{\small{}
% $\qquad$ $\overset{\textrm{Zoom out}}{\Longrightarrow}$ $\qquad$
% $\qquad$}{\footnotesize{}$\underline{\Upsilon}_{\mathrm{MD}}(\alpha,\beta)$}{\small{}$\qquad$}{\footnotesize{}
% }{\small{} $\stackrel[\textrm{a neighborhood of the origin}]{\textrm{Zoom in to }}{\Longleftarrow}$
% $\qquad$}{\footnotesize{}$\underline{\Upsilon}_{\mathrm{LD}}(\alpha,\beta)$}}\tabularnewline
\caption{The MD  exponent of the $q$-stability  ${\Upsilon}_{q,\mathrm{MD}}$ for $\rho=0.9$. Observe that ${\Upsilon}_{q,\mathrm{MD}}$ is linear given each $q\neq 0$ and diverges as $q\downarrow 1-\rho^{-2}\approx -0.2346$.  }
%Theorem~\ref{thm:ops_resolve}
\label{fig:q_stability_MD} 
\end{figure}

The function ${\Upsilon}_{q,\mathrm{MD}}$, defined in \eqref{eq:UpsilonMD}, is plotted in Fig.~\ref{fig:q_stability_MD}. % For every fixed $\alpha$,  see \eqref{eq:UpsilonMD}.  
\begin{proof}[Proof of Theorem \ref{thm:strong-q-stability}]
This theorem  is a consequence
of the classic hypercontractivity inequalities in \eqref{eq:NICDFHC-1}
and \eqref{eq:NICDRHC-1}. Substituting $f\leftarrow\bone_{\calA}$
into \eqref{eq:NICDFHC-1} and \eqref{eq:NICDRHC-1}, we obtain for
$q\ge1$, 
\begin{equation}
\Vert \pi_{X|Y}^{n}(\calA|Y^n)\Vert_{q}\le \pi_{X}^{n}(\calA)^{\frac{1}{1+(q-1)\rho^{2}}},
\end{equation}
and for  $q\in(1-\rho^{-2},1)\backslash\{0\}$,
\begin{equation}
\Vert \pi_{X|Y}^{n}(\calA|Y^n)\Vert_{q}\ge \pi_{X}^{n}(\calA)^{\frac{1}{1+(q-1)\rho^{2}}}.
\end{equation}
These inequalities immediate imply \eqref{eq:NICDMDstability} and
\eqref{eq:NICDMDstability2}.

The asymptotic tightness of \eqref{eq:NICDMDstability} and \eqref{eq:NICDMDstability2}
can be verified by choosing the sets $\calA$ in the definition of the MD exponent 
%\eqref{eq:NICD-72-3-1-1}
%and \eqref{eq:NICD-RMD} 
to be sequences of Hamming balls or spherical
shells. The asymptotic tightness for  $q\in(-\infty, 1-\rho^{-2}]$ follows by the monotonicity of the $L^q$-norm in $q$, and taking limits as $q \downarrow 1-\rho^{-2}$ in \eqref{eq:Exp_MD}. We omit the details.  \end{proof}

We now turn our attention to the LD exponent. For $q\neq 0$, define 
\begin{equation}
\theta_{q}(Q_{X},Q_{Y}):=\rvD(Q_{X},Q_{Y}\|\pi_{XY})-\frac{D(Q_{Y}\|\pi_{Y})}{q'}.\label{eq:NICDtheta_q}
\end{equation}
where $q'$ is the H\"older conjugate of $q$
Define 
\begin{equation}
{\Upsilon}_{q,\mathrm{LD}}(\alpha):=\inf_{ Q_{X},Q_{Y}:
D(Q_{X}\|\pi_{X})\geq\alpha
}\theta_{q}(Q_{X},Q_{Y})\label{eq:NICDunderTheta_q}
\end{equation}
for $q\ge1$, and  
\begin{align}
 & {\Upsilon}_{q,\mathrm{LD}}(\alpha):=\left\{ \!\begin{array}{cc}
{\displaystyle \sup_{Q_{X}:D(Q_{X}\|\pi_{X})\le\alpha}\inf_{Q_{Y}}\theta_{q}(Q_{X},Q_{Y})} & 0<q<1\vspace{.03in}\\
{\displaystyle \sup_{Q_{X}:D(Q_{X}\|\pi_{X})\le\alpha}\sup_{Q_{Y}}\theta_{q}(Q_{X},Q_{Y})} & q<0
\end{array}\right.\label{eq:NICDoverTheta_q}
\end{align}
for $q\in(-\infty,1)\backslash\{0\}$. 
It can be   verified that  ${\Upsilon}_{q,\mathrm{LD}}(\alpha)\ge0$ for all $q \neq 0$. Asymptotically tight bounds  are provided in  the following
theorem, which  is known as the {\em strong $q$-stability theorem} and
was proved by the first author of this monograph~\cite{yu2021strong}.

\begin{theorem}[Strong $q$-stability] \label{thm:q-stability} For
any $n\ge1$ and $\alpha\in(0,1)$, it holds that for $q\ge1$, 
\begin{align}
{\Upsilon}_{q,\mathrm{LD}}^{(n)}(\alpha) & \ge\cvx[{\Upsilon}_{q,\mathrm{LD}}](\alpha),\label{eq:stab-FLD}
\end{align}
and for $q\in(-\infty,1)\backslash\{0\}$, 
\begin{align}
{\Upsilon}_{q,\mathrm{LD}}^{(n)}(\alpha) & \leq\cve[{\Upsilon}_{q,\mathrm{LD}}](\alpha).\label{eq:stab-RLD}
\end{align}
Moreover, these two bounds are asymptotically tight, i.e., 
\begin{equation}
{\Upsilon}_{q,\mathrm{LD}}^{(\infty)}(\alpha)=\cvx[{\Upsilon}_{q,\mathrm{LD}}](\alpha)\quad\mbox{and}\quad{\Upsilon}_{q,\mathrm{LD}}^{(\infty)}(\alpha)=\cve[{\Upsilon}_{q,\mathrm{LD}}](\alpha),\label{eqn:ld_balls_sphere}
\end{equation}
and these equalities are achieved by sequences
of Hamming balls or spheres. 
\end{theorem}

It has been shown in \cite{yu2021convexity} that for $q\ge1$, ${\Upsilon}_{q,\mathrm{LD}}$
is convex, and for $q\in(-\infty,1)\backslash\{0\}$, ${\Upsilon}_{q,\mathrm{LD}}$
is concave. Combining this result with the strong $q$-stability theorem
(Theorem~\ref{thm:q-stability}) tells us that the lower convex envelope
and upper concave envelope operations in \eqref{eqn:ld_balls_sphere}
can be removed and Hamming balls or spheres are optimal in the LD
regime. That is, for the DSBS and $\alpha\in(0,1)$, 
\begin{align}
{\Upsilon}_{q,\mathrm{LD}}^{(\infty)}(\alpha)={\Upsilon}_{q,\mathrm{LD}}(\alpha)\quad\mbox{and}\quad{\Upsilon}_{q,\mathrm{LD}}^{(\infty)}(\alpha)={\Upsilon}_{q,\mathrm{LD}}(\alpha). \label{eq:Upsilon_q_LD}
\end{align}
This is parallel to the discussion of the resolution of the OPS conjecture in Section~\ref{sec:ldr}; also see \eqref{eq:opstrue}.
The asymptotically tight bound ${\Upsilon}_{q,\mathrm{LD}}$, defined in \eqref{eq:NICDoverTheta_q}, is  plotted in   Fig.~\ref{fig:q_stability_LD} for various $q$'s.

\begin{figure}
\centering 
 \includegraphics[width=0.75\columnwidth]{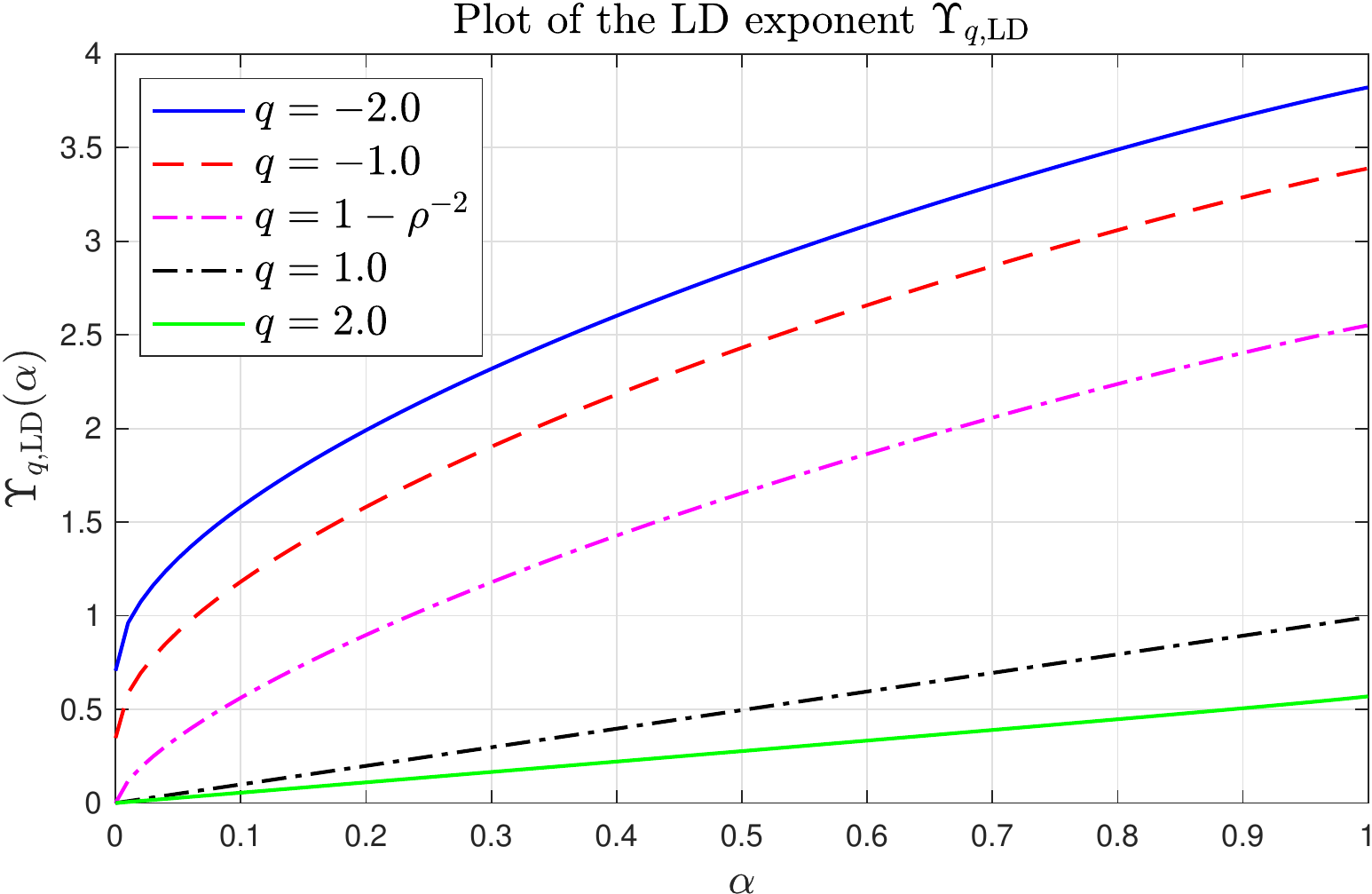}  
\caption{The LD  exponent of the $q$-stability  ${\Upsilon}_{q,\mathrm{LD}}$ for $\rho=0.9$. Observe that ${\Upsilon}_{q,\mathrm{LD}}$ appears to be (``only slightly'') convex for $q>1$, concave for $q \in (-\infty, 1)\setminus\{0\}$, and linear for $q=1$.  Also   ${\Upsilon}_{q,\mathrm{LD}}(0)$ vanishes when $q\le 1-\rho^{-2}\approx -0.2346$. }
\label{fig:q_stability_LD} 
\end{figure}
\section{Extensions to   Sources Beyond the DSBS}

\label{sec:stab_arb}

%Thus far, we have only considered the DSBS. 

Similarly to the discussion in Section~\ref{sec:nicd_arb}, the $q$-stability and strong $q$-stability theorems can
be extended to sources defined on arbitrary finite alphabets as well as jointly Gaussian sources. We discuss these extensions here. 

% which we discuss in Section~\ref{sec:finite_alpha}. Besides in Section~\ref{sec:gauss_stab}, we   also briefly discuss the Gaussian case, for which
%the max $q$-stability problem has been completely solved. 

\subsection{Finite Alphabets } \label{sec:finite_alpha}

Let $\pi_{XY}$ be a joint distribution defined on a finite alphabet.
We now consider its $q$-stability. For $\alpha\in[0,\alpha_{\max}(\pi_{X})]$ (defined in~\eqref{eqn:alpha_max}),
we reuse the definitions in \eqref{eq:NICD-FLD}--\eqref{eq:NICD-RMD}
for ${\Upsilon}_{q,\mathrm{MD}}^{(n)}$ 
and ${\Upsilon}_{q,\mathrm{LD}}^{(n)}$, but with the underlying
distribution set to be $\pi_{XY}$.   The strong $q$-stability theorem (Theorem \ref{thm:strong-q-stability})
can be extended to the following general version, which was first shown
in~\cite{yu2021strong}, as a consequence of the strong version of
hypercontractivity inequalities derived in~\cite{yu2021strong}. We   provide a simple proof of Theorem~\ref{thm:strongqstability}  in Section \ref{sec:NICD-Stability}. 

%Note that the original theorem in \cite{yu2021strong} is valid not only for   finite alphabets, but also for Polish spaces.
%The following theorem can be proven by following
%steps similar to those of the strong small-set theorem above. Hence,
%we omit the proof. 

\begin{theorem}[Strong $q$-stability: General version] \label{thm:strongqstability}
For any $n\ge1$ and $\alpha\in(0,\alpha_{\max}(\pi_{X})]$, \eqref{eq:stab-FLD}
holds for $q\ge1$, and \eqref{eq:stab-RLD} holds for $q\in(-\infty,1)\setminus\{0\}$.
Moreover, \eqref{eq:stab-FLD} and \eqref{eq:stab-RLD} are asymptotically
tight, i.e.,~\eqref{eqn:ld_balls_sphere} holds.  \end{theorem}

The $q$-stability theorem (Theorem \ref{thm:q-stability}) can be
also generalized to the finite alphabet case, but for general sources on finite alphabets, a limiting operation is needed. 

\begin{theorem}[$q$-Stability: General version] \label{thm:sse-1-1}
%Assume that ${\Upsilon}_{q,\mathrm{LD}}(0)=0$. 
For any $n\ge1$, $\alpha > 0$, and $q\ge 1$, 
\begin{align}
{\Upsilon}_{q,\mathrm{MD}}^{(n)}(\alpha) & \ge\lim_{\epsilon\downarrow0}\frac{1}{\epsilon}\cvx[{\Upsilon}_{q,\mathrm{LD}}](\epsilon\alpha).\label{eq:NICDMD-1}
\end{align}
If instead $q \in (-\infty, 1)\setminus\{0\}$ and ${\Upsilon}_{q,\mathrm{LD}}(0)=0$, then 
\begin{align}
{\Upsilon}_{q,\mathrm{MD}}^{(n)}(\alpha) & \leq\lim_{\epsilon\downarrow0}\frac{1}{\epsilon}\cve[{\Upsilon}_{q,\mathrm{LD}}](\epsilon\alpha) . \label{eq:NICDMD2-1}
\end{align}
Moreover, the inequalities in \eqref{eq:NICDMD-1}--\eqref{eq:NICDMD2-1}
are asymptotically tight. 
%Otherwise, if instead $q \in (-\infty, 1)\setminus\{0\}$ but ${\Upsilon}_{q,\mathrm{LD}}(0)>0$, then 
%\begin{align}
%{\Upsilon}_{q,\mathrm{MD}}^{(n)}(\alpha) & \leq \infty. \label{eq:NICDMD2-2}
%\end{align}
%Moreover, 
\end{theorem}

\subsection{Gaussian Sources}\label{sec:gauss_stab}
Finally, we turn our attention to 
 memoryless bivariate Gaussian sources with
correlation coefficient $\rho\in(0,1)$. For this class of sources, the max
$q$-stability problem was completely solved by \citet{borell1985geometric}
for all $a\in[0,1]$.  Let
$\pi_{XY}$ be a bivariate Gaussian distribution with zero mean
and covariance matrix $\bK$ given in \eqref{eqn:cov_mat_rho}.   Let $(X^{n},Y^{n})\sim \pi_{XY}^{n}$. For this distribution, a real number $q > 1$,
and $a\in[0,1]$, we define 
\begin{align}
\Gamma_{\rho}^{(q)}(a) & :=\sup\Vert \pi_{X|Y}^{n}(\calA|Y^n)\Vert_{q}^{q},
\end{align}
where the supremum runs over all measurable sets  $\calA\subset\bbR^{n}$
such that $\pi_{X}^{n}(\calA)=a$.  The following theorem is  due to \citet{borell1985geometric}.

\begin{theorem}[Borell's $q$-stability theorem]  \label{thm:borell}
 For any $n\ge1$, $q>1$, $0\le\rho<1$, and $a\in[0,1]$, one
has 
\begin{equation}
\Gamma_{\rho}^{(q)}(a)=\Lambda_{\rho}^{(q)}(a), \label{eq:borell}
\end{equation}
where $\Lambda_{\rho}^{(q)}$, the Gaussian $q$-stability function, is defined in \eqref{eq:GaussianStabFunc}. 
Moreover,  optimal  subsets $\calA$ (i.e., those  attaining $\Gamma_{\rho}^{(q)}$)
are parallel halfspaces. \end{theorem} % (almost surely).  \end{theorem} 

 The proof of this theorem can be found in \citet{borell1985geometric} and \citet{eldan2015two}.
Moreover, the proof of this theorem with $q$ being an integer can
also be found in \citet{isaksson2012maximally} and \citet{neeman2013isoperimetry}.
We remark that  Neeman's proof in \cite{neeman2013isoperimetry} is
 an extension of the one for   Borell's isoperimetric theorem
given in Section \ref{subsec:NICD-Gaussian} to the multi-user case.

We now consider the Gaussian version of  the Courtade--Kumar conjecture. 
Substituting~\eqref{eq:borell} into the $\Phi$-symmetric max $q$-stability in~\eqref{eq:-4} and taking   limits as $q \downarrow 1$, one can deduce that $\breve{\Pi}_{\rho}^{(1)}(a)$ is attained by halfspaces with $\pi_X^n$-probability $a$. This implies that 
\begin{align}
&\hspace{-.25in}\max_{\substack{f:\bbR^n\to \{0,1\} \textrm{ measurable}: \\ \mathbb{E}[f(X^{n})]=a}}-H\big(f(X^{n})\big|Y^{n}\big)\nn\\*
&\hspace{-.25in} =-H\big(\bone\{ X_1\!\le\!\Phi^{-1}(a)\} \big| Y_1\big) =- \bbE_{Y_1} \Bigg[ h \bigg(\Phi\Big(\frac{\Phi^{-1}(a)\!-\!\rho Y_1}{\sqrt{1-\rho^2}}\Big)\bigg)\Bigg]. \label{eq:GaussianCK}
\end{align}
That is, given $a\in [0,1]$, the indicator of any half-space  with $\pi_X^n$-probability $a$ (e.g., $(-\infty,\Phi^{-1}(a)]\times \bbR^{n-1}$)  maximizes the mutual information between $f(X^n)$ and $Y^n$  over all $\{0,1\}$-valued measurable functions $f$. This statement was also proved by \citet{kindler2015remarks} using a different method. If we do not fix $a$, then similarly to the original  Courtade--Kumar conjecture for the DSBS, it is  natural to conjecture that for this Gaussian version of  Courtade--Kumar conjecture, the mutual information is also  maximized at $a=1/2$ for every $\rho \in (0,1)$. This point can be confirmed numerically  as shown in Fig.~\ref{fig:GaussianCK} in which we plot the right-hand side of \eqref{eq:GaussianCK} plus $h(a)$ as a function of~$a \in [0,1/2]$ for different $\rho$'s.  Note that we only focus on the case $a\in [0,1/2]$ in Fig.~\ref{fig:GaussianCK}, since the function considered is symmetric with respect to $a=1/2$.  It is easily seen that the maxima of these curves occur  at $a=1/2$.  

\begin{figure}[!ht]
\centering \includegraphics[width=0.8\columnwidth]{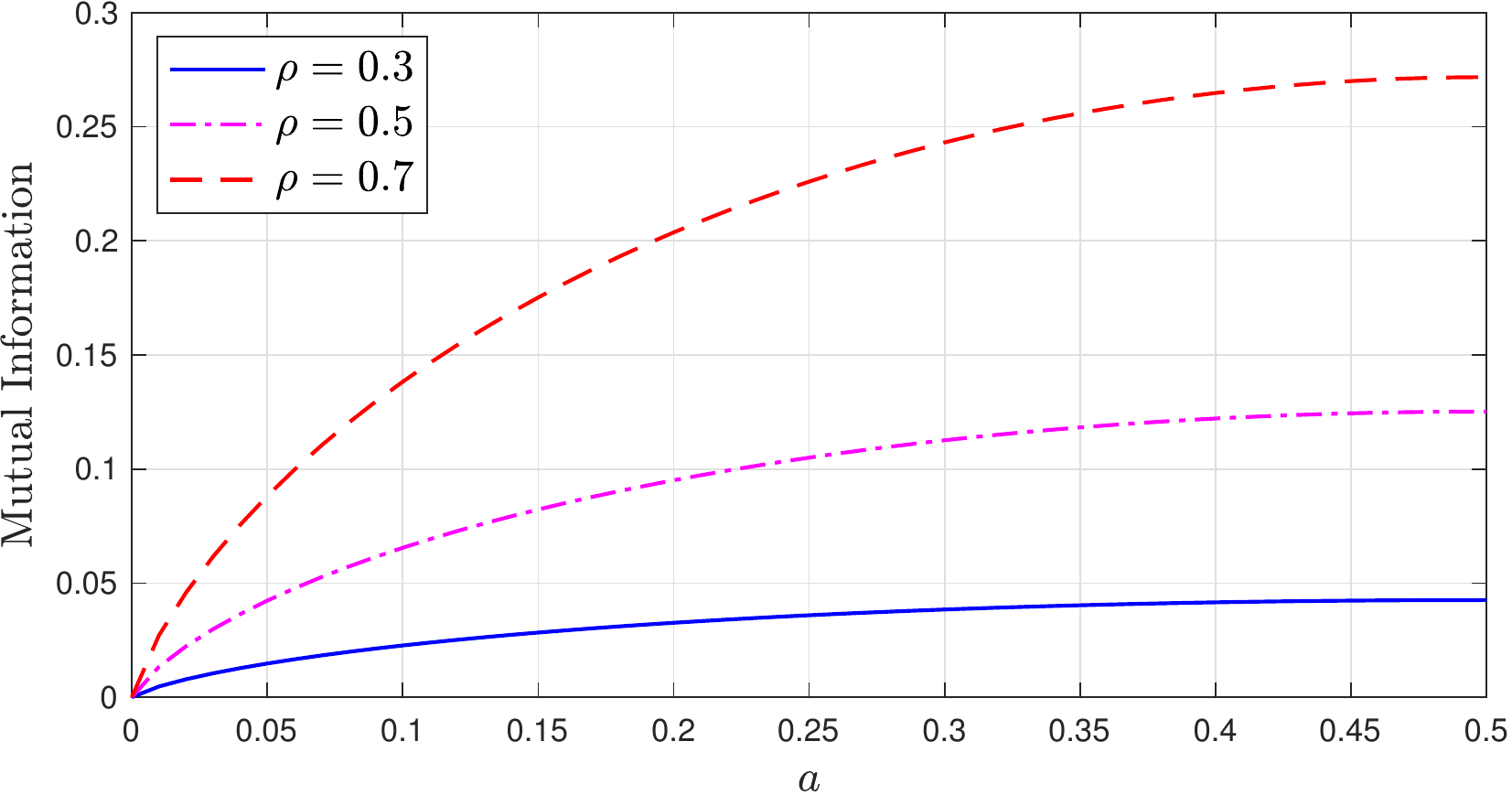} 
\caption{The mutual information, i.e., the right-hand side of~\eqref{eq:GaussianCK} plus $h(a)$.}
\label{fig:GaussianCK} 
\end{figure}

%\begin{equation}
%\frac{1}{2}\in\argmax_{0\le a\le 1/2} - \bbE_{Y_1} \Bigg[ h \bigg(\Phi\Big(\frac{\Phi^{-1}(a)-\rho Y_1}{\sqrt{1-\rho^2}}\Big)\bigg)\Bigg]+h(a).
%\end{equation}
%as a simple plot of the 
%This point is confirmed numerically in Fig. \ref{fig:GaussianCK}, in which the mutual information, i.e., 
%as  the right-hand-side right-hand-side \eqref{eq:GaussianCK} plus $h(a)$ is maximized at $a=1/2$, i.e., 

% \eqref{eq:GaussianCK} plus $h(a)$,  is plotted  as a function of $a$. 

%\red{Lei: It is natural to conjecture that for this Gaussian version of  Courtade--Kumar conjecture, the mutual information is maximized at $a=1/2$.
%Will try to show this later. }

\global\long\def\Ent{\mathrm{Ent}}
\global\long\def\sign{\mathrm{sign}}

\chapter{Functional Inequalities} \label{ch:funineq}

\looseness=+1
In this section, we consider functional extensions of the NICD and
the max $q$-stability problems as described in Sections~\ref{ch:NICD}
and \ref{ch:Stability} respectively. Recall that in the $2$-user NICD problem,
we optimize the probability of  agreement between two random bits that are generated
in a distributed manner via {\em Boolean functions} from a joint source
$(X^{n},Y^{n})$. In this section, we replace the Boolean functions $f,g\in\{0,1\}^n\to\{0,1\}$ 
with {\em arbitrary nonnegative functions}, and obtain corresponding
functional inequalities. Specifically, we will introduce the Brascamp--Lieb
inequalities, the hypercontractivity inequalities, and the log-Sobolev
inequalities, as well as their strengthened counterparts. We provide
information-theoretic characterizations of these inequalities, and
also use them to prove the strong SSE theorem and the strong
$q$-stability theorem  stated respectively in Sections \ref{sec:nicd_arb} and \ref{sec:stab_arb}. 
Analogously to the forward and reverse joint probabilities in the NICD problem (Definition~\ref{def:for_rev_jp}), the optimal
constants or exponents in these inequalities can be also regarded
as refinements of GKW's common information when the latter is equal
to zero, but with the ``information'' measured by the {\em entropy}
of a nonnegative function, rather than the Shannon entropy. 

This section concerning {\em functional inequalities} (or inequalities
involving functionals) starts by formally defining some convenient quantities, such as the minimum relative entropy region, 
in Section~\ref{sec:def_BL}. Using these new definitions, we provide alternative representations of the forward and reverse large deviations exponents in the NICD and $q$-stability problems.  These quantities are then used in Section~\ref{sec:classic_BL} 
to express the hypercontractivity regions (which generalize and strengthen the classic H\"older inequalities) and Brascamp--Lieb exponents
  in terms of single-letter, information-theoretic quantities. We then connect these exponents to the NICD and $q$-stability problems in Section~\ref{sec:NICD-Stability}, leading to a  short proof of the strong SSE theorem (Theorem~\ref{thm:strongsse-2}).  In Section~\ref{sec:log_sobolev}, we discuss the log-Sobolev inequalities, provide single-letter expressions for their optimal constants, and use the results as a bridge to connect the hypercontractivity inequalities to their strengthened counterparts, which are presented in Section~\ref{sec:stronger}. 
In Section~\ref{sec:stronger}, our discussion culminates with  expressions for the strong log-Sobolev
constant and a strengthened hypercontractivity inequality for the DSBS.   Throughout
this section, we focus on {\em information-theoretic} characterizations of optimal constants and exponents in various functional 
inequalities.

As there are several interconnected results in this section and Sections \ref{ch:NICD} and \ref{ch:Stability}, we illustrate their relationships by means of a graph in Fig.~\ref{fig:FIstructure}.

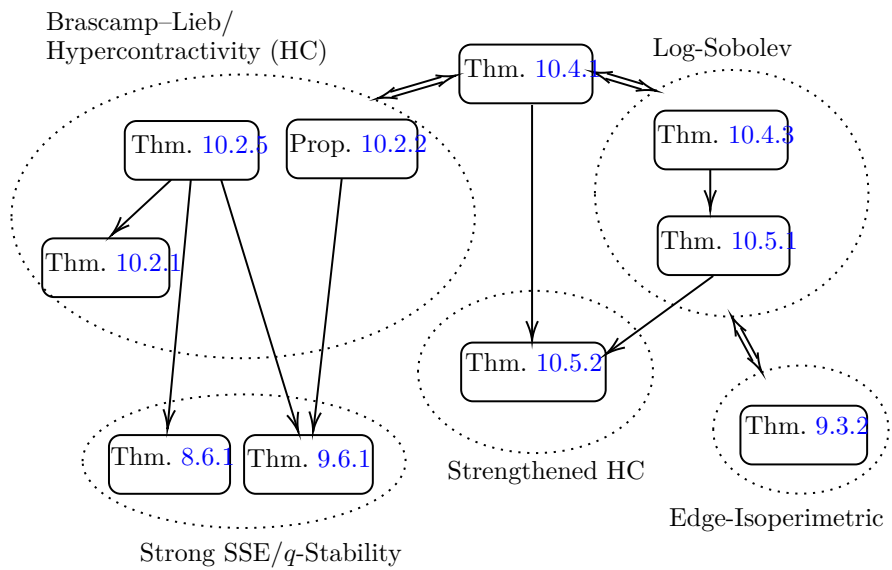
\begin{figure}[!ht]
\centering 

\tikzset{every picture/.style={line width=0.75pt}} %set default line width to 0.75pt        

\begin{tikzpicture}[x=0.75pt,y=0.75pt,yscale=-1,xscale=1,scale=0.9, every node/.style={scale=0.9}] 
%uncomment if require: \path (0,3485); %set diagram left start at 0, and has height of 3485

%Rounded Rect [id:dp7861770096268781] 
\draw   (870.22,3063.93) .. controls (870.22,3060.29) and (873.17,3057.35) .. (876.81,3057.35) -- (934.91,3057.35) .. controls (938.55,3057.35) and (941.5,3060.29) .. (941.5,3063.93) -- (941.5,3083.69) .. controls (941.5,3087.33) and (938.55,3090.28) .. (934.91,3090.28) -- (876.81,3090.28) .. controls (873.17,3090.28) and (870.22,3087.33) .. (870.22,3083.69) -- cycle ;
%Rounded Rect [id:dp13625233166965667] 
\draw   (916.5,2998.13) .. controls (916.5,2994.49) and (919.45,2991.54) .. (923.09,2991.54) -- (984.91,2991.54) .. controls (988.55,2991.54) and (991.5,2994.49) .. (991.5,2998.13) -- (991.5,3017.89) .. controls (991.5,3021.53) and (988.55,3024.47) .. (984.91,3024.47) -- (923.09,3024.47) .. controls (919.45,3024.47) and (916.5,3021.53) .. (916.5,3017.89) -- cycle ;
%Rounded Rect [id:dp19913790810333243] 
\draw   (1213.34,2992.31) .. controls (1213.34,2988.67) and (1216.29,2985.72) .. (1219.92,2985.72) -- (1281.41,2985.72) .. controls (1285.05,2985.72) and (1288,2988.67) .. (1288,2992.31) -- (1288,3012.06) .. controls (1288,3015.7) and (1285.05,3018.65) .. (1281.41,3018.65) -- (1219.92,3018.65) .. controls (1216.29,3018.65) and (1213.34,3015.7) .. (1213.34,3012.06) -- cycle ;
%Rounded Rect [id:dp6478215110481931] 
\draw   (907.64,3174.25) .. controls (907.64,3170.61) and (910.59,3167.67) .. (914.22,3167.67) -- (968.91,3167.67) .. controls (972.55,3167.67) and (975.5,3170.61) .. (975.5,3174.25) -- (975.5,3194.01) .. controls (975.5,3197.65) and (972.55,3200.6) .. (968.91,3200.6) -- (914.22,3200.6) .. controls (910.59,3200.6) and (907.64,3197.65) .. (907.64,3194.01) -- cycle ;
%Rounded Rect [id:dp03704101506669044] 
\draw   (1104,2955.26) .. controls (1104,2951.62) and (1106.95,2948.67) .. (1110.59,2948.67) -- (1172.01,2948.67) .. controls (1175.65,2948.67) and (1178.6,2951.62) .. (1178.6,2955.26) -- (1178.6,2975.02) .. controls (1178.6,2978.65) and (1175.65,2981.6) .. (1172.01,2981.6) -- (1110.59,2981.6) .. controls (1106.95,2981.6) and (1104,2978.65) .. (1104,2975.02) -- cycle ;
%Rounded Rect [id:dp7860381894879831] 
\draw   (983.36,3174.25) .. controls (983.36,3170.61) and (986.31,3167.67) .. (989.94,3167.67) -- (1048.91,3167.67) .. controls (1052.55,3167.67) and (1055.5,3170.61) .. (1055.5,3174.25) -- (1055.5,3194.01) .. controls (1055.5,3197.65) and (1052.55,3200.6) .. (1048.91,3200.6) -- (989.94,3200.6) .. controls (986.31,3200.6) and (983.36,3197.65) .. (983.36,3194.01) -- cycle ;
%Rounded Rect [id:dp09342489748195848] 
\draw   (1105,3122.38) .. controls (1105,3118.75) and (1107.95,3115.8) .. (1111.59,3115.8) -- (1179.41,3115.8) .. controls (1183.05,3115.8) and (1186,3118.75) .. (1186,3122.38) -- (1186,3142.14) .. controls (1186,3145.78) and (1183.05,3148.73) .. (1179.41,3148.73) -- (1111.59,3148.73) .. controls (1107.95,3148.73) and (1105,3145.78) .. (1105,3142.14) -- cycle ;
%Rounded Rect [id:dp16971921437952142] 
\draw   (1215.15,3051.58) .. controls (1215.15,3047.94) and (1218.1,3045) .. (1221.74,3045) -- (1282.41,3045) .. controls (1286.05,3045) and (1289,3047.94) .. (1289,3051.58) -- (1289,3071.34) .. controls (1289,3074.98) and (1286.05,3077.93) .. (1282.41,3077.93) -- (1221.74,3077.93) .. controls (1218.1,3077.93) and (1215.15,3074.98) .. (1215.15,3071.34) -- cycle ;
%Straight Lines [id:da32392072350363343] 
\draw    (942.5,3024.5) -- (911.59,3053.15) ;
\draw [shift={(909,3055.58)}, rotate = 314.33000000000004] [color={rgb, 255:red, 0; green, 0; blue, 0 }  ][line width=0.75]    (10.93,-3.29) .. controls (6.95,-1.4) and (3.31,-0.3) .. (0,0) .. controls (3.31,0.3) and (6.95,1.4) .. (10.93,3.29)   ;
%Straight Lines [id:da40996272745747] 
\draw    (953.5,3024.5) -- (940.57,3162.38) ;
\draw [shift={(940.39,3164.37)}, rotate = 275.39] [color={rgb, 255:red, 0; green, 0; blue, 0 }  ][line width=0.75]    (10.93,-3.29) .. controls (6.95,-1.4) and (3.31,-0.3) .. (0,0) .. controls (3.31,0.3) and (6.95,1.4) .. (10.93,3.29)   ;
%Straight Lines [id:da06486811752164012] 
\draw    (970.5,3024.5) -- (1013.22,3164.11) ;
\draw [shift={(1013.82,3166.02)}, rotate = 252.64] [color={rgb, 255:red, 0; green, 0; blue, 0 }  ][line width=0.75]    (10.93,-3.29) .. controls (6.95,-1.4) and (3.31,-0.3) .. (0,0) .. controls (3.31,0.3) and (6.95,1.4) .. (10.93,3.29)   ;
%Rounded Rect [id:dp583591666137599] 
\draw   (1007.27,2997.31) .. controls (1007.27,2993.67) and (1010.22,2990.72) .. (1013.85,2990.72) -- (1073.91,2990.72) .. controls (1077.55,2990.72) and (1080.5,2993.67) .. (1080.5,2997.31) -- (1080.5,3017.06) .. controls (1080.5,3020.7) and (1077.55,3023.65) .. (1073.91,3023.65) -- (1013.85,3023.65) .. controls (1010.22,3023.65) and (1007.27,3020.7) .. (1007.27,3017.06) -- cycle ;
%Straight Lines [id:da6055178959323861]  %modified by Lei
\draw    (1038.05,3024) -- (1022.11,3165.68) ;
\draw [shift={(1022,3165.67)}, rotate = 276.19] [color={rgb, 255:red, 0; green, 0; blue, 0 }  ][line width=0.75]    (10.93,-3.29) .. controls (6.95,-1.4) and (3.31,-0.3) .. (0,0) .. controls (3.31,0.3) and (6.95,1.4) .. (10.93,3.29)   ;
%Straight Lines [id:da3143511878728573] 
\draw    (1244.59,3018.66) -- (1244.59,3042.3) ;
\draw [shift={(1244.59,3041.3)}, rotate = 270] [color={rgb, 255:red, 0; green, 0; blue, 0 }  ][line width=0.75]    (10.93,-3.29) .. controls (6.95,-1.4) and (3.31,-0.3) .. (0,0) .. controls (3.31,0.3) and (6.95,1.4) .. (10.93,3.29)   ;
%Straight Lines [id:da7750524860583679] 
\draw    (1246.23,3078.45) -- (1187.62,3121.21) ;
\draw [shift={(1186,3122.38)}, rotate = 323.89] [color={rgb, 255:red, 0; green, 0; blue, 0 }  ][line width=0.75]    (10.93,-3.29) .. controls (6.95,-1.4) and (3.31,-0.3) .. (0,0) .. controls (3.31,0.3) and (6.95,1.4) .. (10.93,3.29)   ;
%Straight Lines [id:da2705174218798134] 
\draw    (1144.74,2982.54) -- (1144.74,3112.7) ;
\draw [shift={(1144.74,3114.7)}, rotate = 270] [color={rgb, 255:red, 0; green, 0; blue, 0 }  ][line width=0.75]    (10.93,-3.29) .. controls (6.95,-1.4) and (3.31,-0.3) .. (0,0) .. controls (3.31,0.3) and (6.95,1.4) .. (10.93,3.29)   ;
%Shape: Ellipse [id:dp26158356393113746] 
\draw  [dash pattern={on 0.84pt off 2.51pt}] (853.02,3045.02) .. controls (853.02,3001.08) and (911.15,2965.46) .. (982.86,2965.46) .. controls (1054.57,2965.46) and (1112.7,3001.08) .. (1112.7,3045.02) .. controls (1112.7,3088.96) and (1054.57,3124.58) .. (982.86,3124.58) .. controls (911.15,3124.58) and (853.02,3088.96) .. (853.02,3045.02) -- cycle ;
%Shape: Ellipse [id:dp08175578539721307] 
\draw  [dash pattern={on 0.84pt off 2.51pt}] (1179.95,3032.06) .. controls (1179.95,2993.86) and (1214.26,2962.9) .. (1256.6,2962.9) .. controls (1298.93,2962.9) and (1333.24,2993.86) .. (1333.24,3032.06) .. controls (1333.24,3070.25) and (1298.93,3101.21) .. (1256.6,3101.21) .. controls (1214.26,3101.21) and (1179.95,3070.25) .. (1179.95,3032.06) -- cycle ;
%Straight Lines [id:da7658600131708992] 
\draw [line width=0.75]    (1094.75,2969.69) -- (1066.11,2978.09)(1093.9,2966.81) -- (1065.27,2975.21) ;
\draw [shift={(1058.01,2978.9)}, rotate = 343.65999999999997] [color={rgb, 255:red, 0; green, 0; blue, 0 }  ][line width=0.75]    (10.93,-3.29) .. controls (6.95,-1.4) and (3.31,-0.3) .. (0,0) .. controls (3.31,0.3) and (6.95,1.4) .. (10.93,3.29)   ;
\draw [shift={(1102,2966)}, rotate = 163.66] [color={rgb, 255:red, 0; green, 0; blue, 0 }  ][line width=0.75]    (10.93,-3.29) .. controls (6.95,-1.4) and (3.31,-0.3) .. (0,0) .. controls (3.31,0.3) and (6.95,1.4) .. (10.93,3.29)   ;
%Straight Lines [id:da8456222775377178] 
\draw [line width=0.75]    (1206.13,2973.78) -- (1185.84,2967.73)(1206.99,2970.91) -- (1186.69,2964.85) ;
\draw [shift={(1178.6,2964)}, rotate = 376.62] [color={rgb, 255:red, 0; green, 0; blue, 0 }  ][line width=0.75]    (10.93,-3.29) .. controls (6.95,-1.4) and (3.31,-0.3) .. (0,0) .. controls (3.31,0.3) and (6.95,1.4) .. (10.93,3.29)   ;
\draw [shift={(1214.22,2974.63)}, rotate = 196.62] [color={rgb, 255:red, 0; green, 0; blue, 0 }  ][line width=0.75]    (10.93,-3.29) .. controls (6.95,-1.4) and (3.31,-0.3) .. (0,0) .. controls (3.31,0.3) and (6.95,1.4) .. (10.93,3.29)   ;
%Shape: Ellipse [id:dp14337321608101106] 
\draw  [dash pattern={on 0.84pt off 2.51pt}] (1080.82,3132.26) .. controls (1080.82,3107.13) and (1110.3,3086.76) .. (1146.66,3086.76) .. controls (1183.01,3086.76) and (1212.49,3107.13) .. (1212.49,3132.26) .. controls (1212.49,3157.39) and (1183.01,3177.77) .. (1146.66,3177.77) .. controls (1110.3,3177.77) and (1080.82,3157.39) .. (1080.82,3132.26) -- cycle ;
%Shape: Ellipse [id:dp36443053673987125] 
\draw  [dash pattern={on 0.84pt off 2.51pt}] (892.86,3182.31) .. controls (892.86,3161.62) and (933.37,3144.85) .. (983.36,3144.85) .. controls (1033.34,3144.85) and (1073.86,3161.62) .. (1073.86,3182.31) .. controls (1073.86,3203) and (1033.34,3219.77) .. (983.36,3219.77) .. controls (933.37,3219.77) and (892.86,3203) .. (892.86,3182.31) -- cycle ;
%Shape: Ellipse [id:dp03318768530876226] 
\draw  [dash pattern={on 0.84pt off 2.51pt}] (1246.26,3164.6) .. controls (1246.26,3144.77) and (1268.25,3128.7) .. (1295.38,3128.7) .. controls (1322.51,3128.7) and (1344.5,3144.77) .. (1344.5,3164.6) .. controls (1344.5,3184.43) and (1322.51,3200.5) .. (1295.38,3200.5) .. controls (1268.25,3200.5) and (1246.26,3184.43) .. (1246.26,3164.6) -- cycle ;
%Rounded Rect [id:dp48664569982774064] 
\draw   (1261.5,3157.79) .. controls (1261.5,3154.15) and (1264.45,3151.2) .. (1268.09,3151.2) -- (1325.91,3151.2) .. controls (1329.55,3151.2) and (1332.5,3154.15) .. (1332.5,3157.79) -- (1332.5,3177.55) .. controls (1332.5,3181.18) and (1329.55,3184.13) .. (1325.91,3184.13) -- (1268.09,3184.13) .. controls (1264.45,3184.13) and (1261.5,3181.18) .. (1261.5,3177.55) -- cycle ;
%Straight Lines [id:da179072642369962] 
\draw [line width=0.75]    (1268.38,3126.17) -- (1259.08,3108.97)(1271.02,3124.75) -- (1261.72,3107.54) ;
\draw [shift={(1256.6,3101.21)}, rotate = 421.62] [color={rgb, 255:red, 0; green, 0; blue, 0 }  ][line width=0.75]    (10.93,-3.29) .. controls (6.95,-1.4) and (3.31,-0.3) .. (0,0) .. controls (3.31,0.3) and (6.95,1.4) .. (10.93,3.29)   ;
\draw [shift={(1273.5,3132.5)}, rotate = 241.62] [color={rgb, 255:red, 0; green, 0; blue, 0 }  ][line width=0.75]    (10.93,-3.29) .. controls (6.95,-1.4) and (3.31,-0.3) .. (0,0) .. controls (3.31,0.3) and (6.95,1.4) .. (10.93,3.29)   ;

% Text Node
\draw (870.22,3063.93) node [anchor=north west][inner sep=0.75pt]   [align=left] {Thm.~\ref{thm:HC}};
% Text Node
\draw (918.5,2998.13) node [anchor=north west][inner sep=0.75pt]   [align=left] {Thm.~\ref{thm:ITcharacterization-1}};
% Text Node
\draw (1213.34,2992.31) node [anchor=north west][inner sep=0.75pt]   [align=left] {Thm.~\ref{thm:strongSLI}};
% Text Node
\draw (906.92,3173.04) node [anchor=north west][inner sep=0.75pt]   [align=left] {Thm.~\ref{thm:strongsse-2}};
% Text Node
\draw (1105,2955.26) node [anchor=north west][inner sep=0.75pt]   [align=left] {Thm.~\ref{thm:LSI-HC}};
% Text Node
\draw (983.36,3174.25) node [anchor=north west][inner sep=0.75pt]   [align=left] {Thm.~\ref{thm:strongqstability}};
% Text Node
\draw (1210.99,2943.35) node [anchor=north west][inner sep=0.75pt]   [align=left] {Log-Sobolev};
% Text Node
\draw (1105.76,3120.35) node [anchor=north west][inner sep=0.75pt]   [align=left] {Thm.~\ref{th:hcb} };
% Text Node
\draw (1215.15,3051.58) node [anchor=north west][inner sep=0.75pt]   [align=left] {Thm.~\ref{thm:lsi_cube}};
% Text Node
\draw (1007.27,2997.31) node [anchor=north west][inner sep=0.75pt]   [align=left] {Prop.~\ref{prop:equivalence}};
% Text Node
\draw (870.87,2928.17) node [anchor=north west][inner sep=0.75pt]   [align=left] {Brascamp--Lieb/};
% Text Node
\draw (870.87,2944.17) node [anchor=north west][inner sep=0.75pt]   [align=left] {Hypercontractivity (HC)};

% Text Node
\draw (1095.89,3180.57) node [anchor=north west][inner sep=0.75pt]   [align=left] {Strengthened HC};
% Text Node
\draw (923.51,3225.97) node [anchor=north west][inner sep=0.75pt]   [align=left] {Strong SSE/$q$-Stability};
% Text Node
\draw (1220.1,3207.36) node [anchor=north west][inner sep=0.75pt]   [align=left] {Edge-Isoperimetric};
% Text Node
\draw (1262.46,3155.75) node [anchor=north west][inner sep=0.75pt]   [align=left] {Thm.~\ref{thm:EdgeIsoper}};
\end{tikzpicture}
\caption{\label{fig:FIstructure}A graph of the main results in this and Sections \ref{ch:NICD} and \ref{ch:Stability}, where $\longrightarrow$ denotes an implication and $\Longleftrightarrow$ denotes a close relationship. }
\end{figure}

\section{Preliminary Definitions}\label{sec:def_BL}

Throughout this section, we assume that $\mathcal{X}$ and $\mathcal{Y}$ 
are finite  sets and $\pi_{XY}$ is a joint distribution on $\mathcal{X}\times\mathcal{Y}$. 

\begin{assumption}[Full support of marginals] \label{asm:full_supp}
The supports of $\pi_{X}$ and $\pi_{Y}$
are $\mathcal{X}$ and $\mathcal{Y}$ respectively. 
\end{assumption}

\begin{definition} Define the \emph{minimum relative entropy region}
with respect to a joint distribution $\pi_{XY}\in\calP(\calX\times\calY)$
as 
\begin{equation}
\mathcal{D}(\pi_{XY}):=\bigcup_{Q_X, Q_Y}\Big\{(D(Q_{X}\|\pi_{X}),D(Q_{Y}\|\pi_{Y}),\rvD(Q_{X},Q_{Y}\|\pi_{XY}))\Big\},
\end{equation}
where $\rvD(Q_{X},Q_{Y}\|\pi_{XY})$ is the minimal   relative entropy  with respect to $\pi_{XY}$ over all couplings
of $Q_{X}$ and $Q_{Y}$, defined in \eqref{eqn:min_rel_ent}. Due to Assumption~\ref{asm:full_supp},  any $Q_X$ and $Q_Y$ defined on $\calX$ and $\calY$ respectively are absolutely continuous with respect to $\pi_X$  and $\pi_Y$ respectively. 
\end{definition} The minimum relative entropy region is the subset of   $\bbR^3$ that is formed by the pair of relative entropies $(D(Q_{X}\|\pi_{X}),D(Q_{Y}\|\pi_{Y}))$
and the minimal relative entropy $\rvD(Q_{X},Q_{Y}\|\pi_{XY})$ as $Q_{X}$
and $Q_{Y}$ run over all distributions that are absolutely continuous with
respect to $\pi_{X}$ and $\pi_{Y}$  respectively. 

\begin{definition} \label{def:up_low} For $(s,t)\in[0,\alpha_{\max}(\pi_{X})]\times[0,\beta_{\max}(\pi_{Y})]$ (refer to~\eqref{eqn:alpha_max} for definitions), define the {\em upper} and {\em  lower envelopes} of the  minimal   relative entropy region $\mathcal{D}(\pi_{XY})$ respectively as
\begin{align}
\hspace{-.25in}\underline{\varphi}(s,t) %& :=\min_{Q_{XY}:D(Q_{X}\|\pi_{X})=s,D(Q_{Y}\|\pi_{Y})=t}D(Q_{XY}\|\pi_{XY})\\
 & :=\min_{Q_{X},Q_{Y}:D(Q_{X}\|\pi_{X})=s,D(Q_{Y}\|\pi_{Y})=t}\;\rvD(Q_{X},Q_{Y}\|\pi_{XY}),\label{eqn:def_varphi}
\end{align}
and 
\begin{align}
\hspace{-.25in}\overline{\varphi}(s,t) & :=\max_{Q_{X},Q_{Y}:D(Q_{X}\|\pi_{X})=s,D(Q_{Y}\|\pi_{Y})=t}\;\rvD(Q_{X},Q_{Y}\|\pi_{XY}).\label{eqn:def_psi}
\end{align}
\end{definition} Fix $(\alpha,\beta)\in[0,\alpha_{\max}(\pi_{X})]\times[0,\beta_{\max}(\pi_{Y})]$. Recall that the upper bound on the  forward LD exponent, previously defined in \eqref{eqn:LDsphere1},  is 
\begin{align}
\hspace{-.3in}\underline{\Upsilon}_{\mathrm{LD}}(\alpha,\beta)= \min_{ Q_X, Q_Y : D(Q_X\|\pi_X)\ge\alpha, D(Q_Y\|\pi_Y)\ge\beta} \rvD(Q_X,Q_Y\|\pi_{XY}), \label{eqn:LDsphere1_FI}
\end{align}
and  the lower bound on the reverse LD exponent, is
\begin{align}
\hspace{-.3in}\overline{\Upsilon}_{\mathrm{LD}}(\alpha,\beta) = \max_{ Q_X, Q_Y : D(Q_X\|\pi_X)\le\alpha, D(Q_Y\|\pi_Y)\le\beta} \rvD(Q_X,Q_Y\|\pi_{XY}).\label{eqn:LDsphere2_FI}
\end{align} 
  Based on the functions presented in Definition~\ref{def:up_low}, we may modify the definitions of $\underline{\Upsilon}_{\mathrm{LD}}$ and $\overline{\Upsilon}_{\mathrm{LD}}$ for the DSBS    to a source $(X,Y) \sim \pi_{XY}$ defined on a finite alphabet as follows
%the   lower convex envelope and the   upper concave
%envelope of $\mathcal{D}(\pi_{XY})$ are respectively 
\begin{align}
\underline{\Upsilon}(\alpha,\beta) & :=\min_{s\ge\alpha,t\ge\beta}\;\mathbb{L}[\underline{\varphi}](s,t)\quad\mbox{and}\label{eq:FIlce}\\
\overline{\Upsilon}(\alpha,\beta) & :=\max_{s\leq\alpha,t\leq\beta}\;\mathbb{U}[\overline{\varphi}](s,t). \label{eq:FIuce}
\end{align}
Note that $\underline{\Upsilon}_{\mathrm{LD}}$ and $\overline{\Upsilon}_{\mathrm{LD}}$ in \eqref{eqn:LDsphere1_FI} and \eqref{eqn:LDsphere2_FI}  may not be convex and concave respectively for arbitrary $\pi_{XY}$ (see the discussion following Theorem~\ref{thm:strongsse-2}). Hence, in the modified definitions in  \eqref{eq:FIlce} and \eqref{eq:FIuce}, we take the   lower convex envelope for $\underline{\varphi}$ and the   upper concave envelope for $\overline{\varphi}$.  
%for $(\alpha,\beta)\in[0,\alpha_{\max}(\pi_{X})]\times[0,\beta_{\max}(\pi_{Y})]$. 
%Note that $\underline{\Upsilon}$ and $\overline{\Upsilon}$ correspond  exactly to 
%  $\underline{\Upsilon}_{\mathrm{LD}}$ and $\overline{\Upsilon}_{\mathrm{LD}}$
%defined in \eqref{eqn:LDsphere1} and \eqref{eqn:LDsphere2} respectively. For brevity,   we omit the subscript
%$\mathrm{LD}$. 
With these operations, $\underline{\Upsilon}(\alpha,\beta)$
is convex and nondecreasing in $(\alpha,\beta)$, and $\overline{\Upsilon}(\alpha,\beta)$
is concave and nondecreasing in $(\alpha,\beta)$.\footnote{We say a function
of two variables is {\em nondecreasing} if it is nondecreasing
in one argument when the other is fixed.} Henceforth, we   omit the subscript $\mathrm{LD}$ in  \eqref{eq:FIlce} and \eqref{eq:FIuce}.

Before presenting the next definition, we recall the definition of $\theta_q$  in \eqref{eq:NICDtheta_q} as
\begin{equation}
\theta_{q}(Q_{X},Q_{Y}):=\rvD(Q_{X},Q_{Y}\|\pi_{XY})-\frac{D(Q_{Y}\|\pi_{Y})}{q'},
\end{equation}
 but   now, instead of being a DSBS, $\pi_{XY}$ is an arbitrary distribution defined on the finite set $\calX\times\calY$.
\begin{definition} \label{def:varphi_q}
For $q\ge1$ and for $s\in[0,\alpha_{\max}(\pi_{X})]$, define 
\begin{equation}
\varphi_{q}(s):=\min_{Q_{X},Q_{Y}:D(Q_{X}\|\pi_{X})=s}\;\theta_{q}(Q_{X},Q_{Y}),\label{eq:FIunderTheta_q-1}
\end{equation}
and for $q\in(-\infty,1)\backslash\{0\}$, define 
\begin{equation}
\varphi_{q}(s):=\left\{ \begin{array}{cc}
{\displaystyle \max_{Q_{X}:D(Q_{X}\|\pi_{X})=s}\;\min_{Q_{Y}}\;\theta_{q}(Q_{X},Q_{Y})} & 0<q<1\vspace{.03in}\\
{\displaystyle \max_{Q_{X}:D(Q_{X}\|\pi_{X})=s}\;\max_{Q_{Y}}\;\theta_{q}(Q_{X},Q_{Y})} & q<0
\end{array}\right. . \label{eq:FIpsi_q}
\end{equation}
\end{definition}

We denote $\Upsilon_{q}$
as the   lower convex envelope of $\varphi_{q}$ for $q\ge1$
and the   upper concave envelope of $\varphi_{q}$ for $q\in(-\infty,1)\backslash\{0\}$.
Specifically, for $\alpha\in[0,\alpha_{\max}(\pi_{X})]$,  
\begin{equation}
\Upsilon_{q}(\alpha):=\left\{ \begin{array}{cc}
{\displaystyle \min_{s\ge\alpha}\mathbb{L}[\varphi_{q}](s)} & q\ge1\vspace{.03in}\\
{\displaystyle \max_{s\le\alpha}\mathbb{U}[\varphi_{q}](s)} & q\in(-\infty,1)\backslash\{0\}
\end{array}\right..\label{eq:FIUpsilon_q}
\end{equation}
Observe that $\Upsilon_{q}$ is an alternative representation of  $\Upsilon_{q,\mathrm{LD}}$
defined in~\eqref{eq:NICDunderTheta_q} and~\eqref{eq:NICDoverTheta_q}. By definition, $\Upsilon_{q}(\alpha)$
is convex and nondecreasing in $\alpha$ for each $q\ge1$, and concave
and nondecreasing in $\alpha$ for each $q\in(-\infty,1)\backslash\{0\}$.

To avoid having to deal with the undefined arithmetic operation $\infty - \infty$, 
 we   adopt the following convention. 

 \begin{convention}
\label{conv:FI} When we write an optimization problem with distributions
as decision variables, we implicitly require that the distributions satisfy the condition that all the
integrals and relative entropies (appearing  in the constraints and the objective
function) to be {\em finite}. 
%If there is no such a distribution, 
Otherwise, the value
of the optimization problem is set to~$+\infty$ if it is an infimization,
and $-\infty$ if it is a supremization. 
\end{convention}

To keep notation uncluttered, we also adopt the following convention.

 \begin{convention}
\label{conv:FI2}
When we write an optimization over functions $f$ and~$g$, we implicitly require these functions to be {\em nonnegative}. 
\end{convention}

\section{Classic Hypercontractivity and Brascamp--Lieb  Inequalities}
\label{sec:classic_BL}

In this section, we introduce a class of functional inequalities,
known as {\em Brascamp--Lieb (BL) inequalities}. We also review
the well-known H\"older  and hypercontractivity inequalities
which are special cases of the BL inequalities. We introduced the
hypercontractivity inequalities in the context of of the DSBS
in Section~\ref{sec:clt_half}. In contrast, here we study these inequalities for {\em arbitrary}
sources defined on finite alphabets.

\subsection{H\"older and Hypercontractivity Inequalities}
\label{sec:holder_HC}

We review the well-known {\em forward} and {\em reverse H\"older
inequalities} here. Given a joint distribution $\pi_{XY}$ and an extended
real number $p\in\mathbb{R}\cup\{\pm\infty\}$, for any pair of nonnegative
functions $(f,g)$, the forward and reverse H\"older inequalities are respectively 
\begin{align}
\langle f,g\rangle & \le\Vert f\Vert_{p}\Vert g\Vert_{q}\quad\mbox{if}\quad p\ge1\quad\mbox{and}\label{eq:FIFHolder}\\
\langle f,g\rangle & \geq\Vert f\Vert_{p}\Vert g\Vert_{q}\quad\mbox{if}\quad p\le1,\label{eq:FIRHolder}
\end{align}
where 
%the inner product and the norms (or pseudo norms) are defined
%in~\eqref{eqn:inner_prod} and \eqref{eqn:Lp}, respectively, and
$q$ is the H\"older conjugate of $p$.
% Recall that  the H\"older conjugate
%$p'=\frac{p}{p-1}$ for $p\neq1$ for both inequalities, $p'=\infty$
%for $p=1$ in \eqref{eq:FIFHolder}, and $p'=-\infty$ for $p=1$
%in \eqref{eq:FIRHolder}. 
Since the (pseudo) $L^{q}$-norms $\Vert\cdot\Vert_{q}$
are nondecreasing in $q\in\bbR\cup\{\pm\infty\}$, the scalar $q$
in \eqref{eq:FIFHolder} can be replaced by  any $q\ge p'$. Similarly,
$q$ in~\eqref{eq:FIRHolder} can be replaced by   any $q\le p'$.

\enlargethispage{-2\baselineskip}
If $X=Y$ (i.e., $P_{Y|X }(\cdot|x)$ places all its mass at $x$ for every $x\in\calX$), then the forward and reverse H\"older inequalities are
sharp in the following sense. If $p>1$, then~\eqref{eq:FIFHolder} becomes an equality if and only if $|f|^p$ and $g^{p'}$ are  {\em linearly dependent}, i.e., there exist real numbers $a,b\ge 0$,
not both zero, such that $a|f|^{p}=b|g|^{p'}$ holds ($\pi_{X}$-almost
everywhere). If $p<1$, $\langle f,g\rangle<\infty$ and $\Vert g\Vert_{p'}>0$,
then~\eqref{eq:FIRHolder} is an equality if and only if the equality
$|f|^{p}=a|g|^{p'}$ holds ($\pi_{X}$-almost everywhere) for some $a\ge0$.
Moreover, for the case $X=Y$, the parameters $(p,q)$ in \eqref{eq:FIFHolder}
and \eqref{eq:FIRHolder} cannot be improved in the sense that given
$p\ge1$, for any $q<p'$, there exists a pair of $(f,g)$ that violates~\eqref{eq:FIFHolder};
similarly, given $p\le1$, for any $q>p'$, there exists a pair of
$(f,g)$ that violates~\eqref{eq:FIRHolder}.

However, the H\"older inequalities are  not sharp in general when
$X\neq Y$ (which is the case of interest to us). If $X\neq Y$, then
the parameters $(p,q)$ in the H\"older inequalities can be ``improved''.
Specifically, given a joint distribution $\pi_{XY}$ and $p\ge1$, we
are interested in how {\em small} $q\in\mathbb{R}\cup\{\pm\infty\}$
can be such that for {\em any} nonnegative functions $f:\mathcal{X}\to[0,\infty)$
and $g:\mathcal{Y}\to[0,\infty)$, it holds that 
\begin{align}
\langle f,g\rangle & \le\Vert f\Vert_{p}\Vert g\Vert_{q}.\label{eq:FIFHC}
\end{align}
By the forward H\"older inequality, the infimum of all such $q$'s is at most~$p'$, the H\"older conjugate of $p$. Similarly, given $p\le1$, we
are interested in how {\em large} $q\in\mathbb{R}\cup\{\pm\infty\}$
can be such that for any nonnegative functions $f$ and $g$, it holds that 
\begin{align}
\langle f,g\rangle & \ge\Vert f\Vert_{p}\Vert g\Vert_{q}.\label{eq:FIRHC}
\end{align}
For this case, the supremum of  all such $q$'s is at least~$p'$.
Inequalities~\eqref{eq:FIFHC} and \eqref{eq:FIRHC} for the case $X\ne Y$ are respectively
termed the {\em forward} and {\em reverse hypercontractivity
inequalities}, since the forward and reverse H\"older inequalities
in \eqref{eq:FIFHolder} and \eqref{eq:FIRHolder} respectively are
regarded as the (usual) {\em contractivity} inequalities, and inequalities
\eqref{eq:FIFHC} and \eqref{eq:FIRHC} with improved $(p,q)$ are
{\em strengthenings} of the forward and reverse H\"older inequalities.
%(This point can be better understood
%from the single-function forms of H\"older inequalities provided
%in Section \ref{subsec:Single-Function-Version}). 
%We now introduce the concepts of 
%forward  and reverse hypercontractivity regions  \cite{beigi2018phi,liu2018information}. 

Inequalities \eqref{eq:FIFHC} and \eqref{eq:FIRHC} motivate the following definitions. 
\begin{definition} \label{def:FRHypercontractivityR} The {\em
forward} and {\em reverse hypercontractivity regions} \cite{beigi2018phi,liu2018information}
are respectively defined as 
\begin{align}
\mathcal{R}_{\mathrm{FH}}(\pi_{XY}) & :=\big\{(p,q)\in[1,\infty)^{2}:\langle f,g\rangle\le\Vert f\Vert_{p}\Vert g\Vert_{q},~\forall \, f,g\ge0\big\}\label{eq:FIFHR}
\end{align}
and 
\begin{align}
\mathcal{R}_{\mathrm{RH}}(\pi_{XY}) := \big\{(p,q)\in (-\infty,1]^{2} : \langle f,g\rangle\ge\Vert f\Vert_{p}\Vert g\Vert_{q},~\forall\,  f,g\ge 0\big\}.\label{eq:FIRHR}
\end{align}
\end{definition} By definition, these two regions correspond to the
sets of parameters $(p,q)$ for which the forward or reverse hypercontractivity
inequalities in~\eqref{eq:FIFHC} and~\eqref{eq:FIRHC} hold. We
remark that the notion of {\em hypercontractivity ribbons} was
introduced in \citet[Eqn.~(6.117)]{anantharam2014hypercontractivity}
and \citet{kamath2015reverse}, prior to the hypercontractivity regions
 being introduced in
\citet{beigi2018phi} and \citet{liu2018information}. The hypercontractivity
ribbons correspond to the hypercontractivity regions apart from the exclusions of 
 the H\"older regions $\{(p,q)\in[1,\infty)^{2}:q\ge p'\}$
and $\{(p,q)\in(-\infty,1]^{2}:q\le p'\}$, and that the H\"older
conjugate of $q$ is taken.

We can write $\mathcal{R}_{\mathrm{RH}}(\pi_{XY})$ as the disjoint
union of four sets 
\begin{align}
\mathcal{R}_{\mathrm{RH}}^{++}(\pi_{XY}) & :=(0,1]^{2}\cap\mathcal{R}_{\mathrm{RH}}(\pi_{XY}),\label{eqn:subreg1}\\
\mathcal{R}_{\mathrm{RH}}^{+-}(\pi_{XY}) & :=\big((0,1]\times(-\infty,0)\big)\cap\mathcal{R}_{\mathrm{RH}}(\pi_{XY}),\label{eqn:subreg2}\\
\mathcal{R}_{\mathrm{RH}}^{-+}(\pi_{XY}) & :=\big((-\infty,0)\times(0,1]\big)\cap\mathcal{R}_{\mathrm{RH}}(\pi_{XY}),\quad\mbox{and}\label{eqn:subreg3}\\
\mathcal{R}_{\mathrm{RH}}^{--}(\pi_{XY}) & :=(-\infty,0]^{2}.\label{eqn:subreg4}
\end{align}
The forward hypercontractivity region and the first three subregions
of the reverse hypercontractivity region in \eqref{eqn:subreg1}, \eqref{eqn:subreg2}, and \eqref{eqn:subreg3}
admit the following information-theoretic characterizations; see \cite{ahlswede1976spreading,carlen2009subadditivity,kamath2015reverse,beigi2016equivalent,liu2018information,yu2021strong}.

\begin{theorem}[Information-theoretic characterizations of hypercontractivity regions] \label{thm:HC} The forward hypercontractivity region
$\mathcal{R}_{\mathrm{FH}}(\pi_{XY})$ can be expressed in terms of
the minimal  relative entropy as the set of $(p,q)\in[1,\infty)^{2}$
such that 
\begin{align}
\rvD(Q_{X},Q_{Y}\|\pi_{XY})\ge\frac{1}{p}\, D(Q_{X}\|\pi_{X})+\frac{1}{q}\, D(Q_{Y}\|\pi_{Y}).\label{eq:FIFHR-2}
\end{align}
%for all $Q_{X}\ll \pi_{X}$ and $Q_{Y}\ll \pi_{Y}$. %\begin{align}
%\mathcal{R}_{\mathrm{FH}}(\pi_{XY})   \! =\!\left\{  (p,q)\in [1,\infty)^2:   \parbox[c]{2.29in}{$\rvD(Q_{X},Q_{Y}\|\pi_{XY})\\ \textcolor{white}{.} \quad \ge  \frac{1}{p}D(Q_{X}\|\pi_{X})+\frac{1}{q}D(Q_{Y}\|\pi_{Y}) $   \vspace{0.03 in}\\ $\forall\, Q_X\ll \pi_X , Q_Y\ll \pi_Y$  }  \right\} . \label{eq:FIFHR-2}
%\end{align}
%\begin{align}
%\mathcal{R}_{\mathrm{FH}}(\pi_{XY}) & =\left\{ \begin{array}{l}
%(p,q)\in[1,\infty)^{2}:\rvD(Q_{X},Q_{Y}\|\pi_{XY})\\
%\qquad\ge\frac{1}{p}D(Q_{X}\|\pi_{X})+\frac{1}{q}D(Q_{Y}\|\pi_{Y}),\\
%\qquad\forall Q_{X}\ll \pi_{X},Q_{Y}\ll \pi_{Y}
%\end{array}\right\} ,\label{eq:FIFHR-2}
%\end{align}
In addition, $\mathcal{R}_{\mathrm{RH}}^{++}(\pi_{XY})$ is the set
of all $(p,q)\in(0,1]^{2}$ such that 
\begin{equation}
\rvD(Q_{X},Q_{Y}\|\pi_{XY})\le\frac{1}{p}\, D(Q_{X}\|\pi_{X})+\frac{1}{q}\, D(Q_{Y}\|\pi_{Y}). \label{eq:FIRHplusplus}
\end{equation}
%for all $Q_{X}\ll \pi_{X}$ and $Q_{Y}\ll \pi_{Y}$. 
Finally, $\mathcal{R}_{\mathrm{RH}}^{+-}(\pi_{XY})$
is the set of all $(p,q)\in(0,1]\times(-\infty,0)$ such that 
\begin{equation}
\min_{Q_{Y}}\Big\{\rvD(Q_{X},Q_{Y}\|\pi_{XY})-\frac{1}{q}\, D(Q_{Y}\|\pi_{Y})\Big\}\le\frac{1}{p}\, (Q_{X}\|\pi_{X}).\label{eq:FIRHplusminus}
\end{equation}
%for all $Q_{X}\ll \pi_{X}$. %\begin{align}
%\mathcal{R}_{\mathrm{RH}}^{++}(\pi_{XY})   \! =\!\left\{  (p,q)\in (0,1]^2:   \parbox[c]{2.29in}{$\rvD(Q_{X},Q_{Y}\|\pi_{XY})\\ \textcolor{white}{.} \quad \le  \frac{1}{p}D(Q_{X}\|\pi_{X})+\frac{1}{q}D(Q_{Y}\|\pi_{Y}) $   \vspace{0.03 in}\\ $\forall\, Q_X\ll \pi_X , Q_Y\ll \pi_Y$  }  \right\} . \label{eq:FIRHplusplus}
%\end{align}
%and 
%\begin{align}
%&\mathcal{R}_{\mathrm{RH}}^{+-}(\pi_{XY})  \nn\\*
%& \! =\! \left\{  (p,q) \in (0,1]\times (-\infty,0):   \parbox[c]{2.35in}{$ \inf_{Q_{Y|X}}D(Q_{XY}\| \pi_{XY})-\frac{1}{q}D(Q_{Y}\| \pi_{Y}) \\ \textcolor{white}{.} \quad  \le\frac{1}{p}D(Q_{X}\| \pi_{X}) $   \vspace{0.03 in}\\ $\forall\, Q_X \ll \pi_X  $}  \right\} .
%\end{align}
%%\begin{align}
%%\mathcal{R}_{\mathrm{RH}}^{+}(\pi_{XY}) & =\{ \begin{array}{l}
%%(p,q)\in(0,1]^{2}:\rvD(Q_{X},Q_{Y}\|\pi_{XY})\\
%%\qquad\le\frac{1}{p}D(Q_{X}\|\pi_{X})+\frac{1}{q}D(Q_{Y}\|\pi_{Y}),\\
%%\qquad\forall Q_{X}\ll \pi_{X},Q_{Y}\ll \pi_{Y}
%%\end{array}\} ,\\
%\mathcal{R}_{\mathrm{RH}}^{+-}(\pi_{XY}) & =\{ \begin{array}{l}
%(p,q)\in(0,1]\times(-\infty,0):\inf_{Q_{Y|X}}D(Q_{XY}\|\pi_{XY})\\
%\qquad-\frac{1}{q}D(Q_{Y}\|\pi_{Y})\le\frac{1}{p}D(Q_{X}\|\pi_{X}),\\
%\qquad\forall Q_{X}\ll \pi_{X}
%\end{array}\} .
%\end{align}
\end{theorem}
 
By symmetry, $\mathcal{R}_{\mathrm{RH}}^{-+}(\pi_{XY})$ can be characterized in an analogous manner to $\mathcal{R}_{\mathrm{RH}}^{+-}(\pi_{XY})$ in \eqref{eq:FIRHplusminus}.
The proof of Theorem~\ref{thm:HC} is provided in Section~\ref{sec:BL_ineq},
since it is a special case of the information-theoretic characterizations
of the BL inequalities, which we present therein. Theorem~\ref{thm:HC} can be specialized to Theorem~\ref{thm:hyper2} (the two function version of the hypercontractivity inequalities for the DSBS); see \cite{nair2016evaluating,nair2017reverse}.

Hypercontractivity inequalities were investigated 
in \cite{bonami1968ensembles,kiener1969uber,schreiber1969fermeture,bonami1970etude,gross1975logarithmic,ahlswede1976spreading, borell1982positivity,mossel2013reverse} among others.
Information-theoretic characterizations of the hypercontractivity
(and BL) inequalities can be traced back to the seminal work of \citet{ahlswede1976spreading}
in which, instead of the hypercontractivity regions, the hypercontractivity
{\em constants} (which are quantities induced by the hypercontractivity regions)
were characterized in terms of relative entropies. The information-theoretic
characterization of the forward hypercontractivity region is implied
by the information-theoretic characterization of the forward BL inequalities
on Euclidean spaces in \citet{carlen2009subadditivity}; this was
independently discovered later by \citet{nair2014equivalent} in the case of
finite alphabets.

An information-theoretic
characterization of $\mathcal{R}_{\mathrm{RH}}^{++}(\pi_{XY})$ for
finite alphabets was provided by \citet{kamath2015reverse}. Subsequently,
an information-theoretic characterization of the entire reverse hypercontractivity
region for finite alphabets was shown by \citet{beigi2016equivalent}.
Extensions of these characterizations to Polish spaces were studied
by \citet{liu2018information} using 
a minimax theorem known as the {\em Fenchel--Rockafellar
duality}.

%We mention in passing that, for the reverse part, Liu's
%proof requires $\mathcal{X}$ and $\mathcal{Y}$ to be locally compact
%and $\sigma$-compact Polish spaces. \citet{yu2021strong} provided
%a proof based on the theory of large deviations which relaxes this
%compactness requirement such that $\mathcal{X}$ and $\mathcal{Y}$
%are now allowed to be arbitrary Polish spaces.

As a consequence of Definitions~\ref{def:up_low}, \ref{def:varphi_q}, and Theorem \ref{thm:HC}, the regions $\mathcal{R}_{\mathrm{FH}}(\pi_{XY})$,
$\mathcal{R}_{\mathrm{RH}}^{++}(\pi_{XY})$, and $\mathcal{R}_{\mathrm{RH}}^{+-}(\pi_{XY})$
also admit the following  equivalent characterizations: % based on the quantities
%defined in Section~\ref{sec:def_BL}: 
\begin{align}
 \mathcal{R}_{\mathrm{FH}}(\pi_{XY}) &  = \Big\{(p,q)\in[1,\infty)^{2}:\underline{\varphi}(\alpha,\beta)\ge\frac{\alpha}{p}+\frac{\beta}{q},\, \forall\,\alpha,\beta\ge0\Big\}\\*
 &  = \Big\{(p,q)\in[1,\infty)^{2}:\underline{\Upsilon}(\alpha,\beta)\ge\frac{\alpha}{p}+\frac{\beta}{q},\, \forall\,\alpha,\beta\ge0\Big\},\label{eq:FI-69-1-1}\\
% &  = \Big\{(p,q)\in[1,\infty)^{2}:\lim_{t\downarrow0}\frac{\underline{\Upsilon}(t\alpha,t\beta)}{t}\ge\frac{\alpha}{p} + \frac{\beta}{q},\forall\alpha,\beta\ge0\Big\},\label{eq:FI-75-1-1}\\
 \mathcal{R}_{\mathrm{RH}}^{++}(\pi_{XY}) &  = \Big\{(p,q)\in(0,1]^{2}:\overline{\varphi}(\alpha,\beta)\leq\frac{\alpha}{p}+\frac{\beta}{q},\, \forall\,\alpha,\beta\ge0\Big\}\\*
 &  = \Big\{(p,q)\in( 0,1]^{2}:\overline{\Upsilon}(\alpha,\beta)\leq\frac{\alpha}{p}+\frac{\beta}{q},\, \forall\,\alpha,\beta\ge0\Big\}, \label{eq:FI-70-1-2}%\\
% &  = \Big\{(p,q)\in(0,1]^{2}:\lim_{t\downarrow0}\frac{\overline{\Upsilon}(t\alpha,t\beta)}{t}\leq\frac{\alpha}{p} + \frac{\beta}{q},\forall\alpha,\beta\ge0\Big\},\label{eq:FI-76-1-2}
\end{align}
and 
\begin{align}
\mathcal{R}_{\mathrm{RH}}^{+-}(\pi_{XY}) & =\Big\{(p,q)\in(0,1]\times(-\infty,0): \varphi_{q'}(\alpha)\leq\frac{\alpha}{p},\,\forall\, \alpha\ge0\Big\}\\
 & =\Big\{(p,q)\in(0,1]\times(-\infty,0): {\Upsilon}_{q'}(\alpha)\le\frac{\alpha}{p},\,\forall\, \alpha\ge0\Big\},\label{eq:FI-70-1-1-1-1}%\\
% & =\Big\{(p,q)\in(0,1]^{2}:\lim_{t\downarrow0}\frac{ {\Upsilon}_{q}(t\alpha)}{t}\leq\frac{\alpha}{p},\forall\alpha\ge0\Big\}.\label{eq:FI-76-1-1-1-1}
\end{align}
where $q'$ is the H\"older conjugate of $q$.
%\textcolor{red}{(Lei: the $q$ in the subscript of  $\overline{\Upsilon}_{q}$ or $\overline{\Upsilon}_{q'}$
%%needs to be determined.) }
%In the above characterizations, $\alpha$
%and $\beta$ are restricted to be nonnegative numbers.

%Furthermore, if $\phi$ is convex, then $\mathcal{R}_{\mathrm{FH}}(\pi_{XY})$ reduces  to $\mathcal{R}_{\mathrm{FH}}$ for the DSBS given in the previous chapter.     

\subsection{Brascamp--Lieb Inequalities}
\label{sec:BL_ineq}

The Brascamp--Lieb (BL) inequalities constitute a class of inequalities
that generalizes the families of H\"older and hypercontractivity inequalities.
The {\em forward} and {\em reverse BL inequalities} are defined
as follows. Given a distribution $\pi_{XY}$ and $p,q\in\mathbb{R}$,
for any pair of nonnegative functions $f:\mathcal{X}\to[0,\infty)$
and $g:\mathcal{Y}\to[0,\infty)$, 
\begin{align}
\langle f,g\rangle & \le\overline{C} \, \Vert f\Vert_{p}\Vert g\Vert_{q}\quad\mbox{and}\label{eq:FIFBL}\\
\langle f,g\rangle & \geq\underline{C} \, \Vert f\Vert_{p}\Vert g\Vert_{q},\label{eq:FIRBL}
\end{align}
where $\overline{C}=\overline{C}_{p,q}$ and $\underline{C}=\underline{C}_{p,q}$
depend only on $p$ and $q$ given the distribution $\pi_{XY}$. The
hypercontractivity inequalities in~\eqref{eq:FIFHC} and~\eqref{eq:FIRHC}
correspond to the BL inequalities with $\overline{C}=1$ in \eqref{eq:FIFBL}
and $\underline{C}=1$ in~\eqref{eq:FIRBL} respectively.

The forward version of  the BL inequalities in \eqref{eq:FIFBL} was originally studied in the
1970s by \citet{brascamp1976best}, who were motivated by problems
in particle physics. The reverse version in \eqref{eq:FIRBL} was initially studied by
\citet{barthe1998reverse}.  In fact, the inequalities in \eqref{eq:FIFBL}
and \eqref{eq:FIRBL} are special cases of the original forward and
reverse BL inequalities. We only discuss these
special cases.

\begin{definition} The (optimal) {\em forward} and {\em reverse
BL constants} are respectively defined as 
\begin{align}
\overline{C}_{p,q}^{*}(X;Y) & :=\sup_{f,g:\Vert f\Vert_{p}\Vert g\Vert_{q}>0}\frac{\langle f,g\rangle}{\Vert f\Vert_{p}\Vert g\Vert_{q}}\quad\mbox{and}\label{eq:FIFBL-1-1-2-2}\\
\underline{C}_{p,q}^{*}(X;Y) & :=\inf_{f,g:\Vert f\Vert_{p}\Vert g\Vert_{q}>0}\frac{\langle f,g\rangle}{\Vert f\Vert_{p}\Vert g\Vert_{q}}.\label{eq:FIRBL-1-1-2-2}
\end{align}
Additionally, define the {\em forward} and {\em reverse BL exponents} respectively
as 
\begin{align}
\underline{\Lambda}_{p,q}(X;Y) & :=-\log\overline{C}_{p,q}^{*}(X;Y)\quad\mbox{and}\label{eq:FI}\\
\overline{\Lambda}_{p,q}(X;Y) & :=-\log\underline{C}_{p,q}^{*}(X;Y).\label{eq:FI-1}
\end{align}
\end{definition}

It is well-known that the forward and reverse BL exponents possess
the important tensorization and the data processing properties.

\begin{lemma}[Tensorization] \label{lem:BLtensorization} Let $(X^{n},Y^{n})=\{(X_{1},Y_{1}),\ldots,(X_{n},Y_{n})\}$
be a collection of pairs of random variables that are mutually independent.
Then 
\begin{align}
\underline{\Lambda}_{p,q}(X^{n};Y^{n}) & =\sum_{i=1}^{n}\underline{\Lambda}_{p,q}(X_{i};Y_{i})\quad\mbox{and}\quad\label{eq:FI-2}\\
\overline{\Lambda}_{p,q}(X^{n};Y^{n}) & =\sum_{i=1}^{n}\overline{\Lambda}_{p,q}(X_{i};Y_{i}).\label{eq:FI-4}
\end{align}
\end{lemma}

\begin{proof} The proof here is due to \citet{beigi2016equivalent}
and is based on applying the one-dimensional BL inequality in \eqref{eq:FIFBL} to each
pair of random variables iteratively. To prove \eqref{eq:FI-2}, it
suffices to show that if for each $i\in[n]$, there exist a constant
$\overline{C}_{i}$ such that $\langle f_{i},g_{i}\rangle\le\overline{C}_{i}\Vert f_{i}\Vert_{p}\Vert g_{i}\Vert_{q}$
holds for all nonnegative $f_{i}$ and $g_{i}$ defined
on $\mathcal{X}$ and $\mathcal{Y}$, then $\langle f,g\rangle\le\overline{C}\Vert f\Vert_{p}\Vert g\Vert_{q}$
holds for all nonnegative $f$ and $g$   defined on $\mathcal{X}^{n}$
and $\mathcal{Y}^{n}$, where $\overline{C}=\prod_{i=1}^{n}\overline{C}_{i}$.
This point can be shown  as follows: 
\begin{align}
\langle f,g\rangle & =\mathbb{E}_{X^{n-1},Y^{n-1}}\big[\mathbb{E}_{X_{n},Y_{n}}[f(X^{n})g(Y^{n})\mid X^{n-1},Y^{n-1}]\big]\\
 & \le\overline{C}_{n}\mathbb{E}_{X^{n-1},Y^{n-1}}\big[\Vert f(X^{n-1},\cdot)\Vert_{p}\Vert g(Y^{n-1},\cdot)\Vert_{q}\big]\\
 & \le\overline{C}_{n}\overline{C}_{n-1}\mathbb{E}_{X^{n-2},Y^{n-2}}\big[\Vert f(X^{n-2},\cdot)\Vert_{p}\Vert g(Y^{n-2},\cdot)\Vert_{q}\big]\\
 & \quad\vdots\nn\\
 & \le\overline{C}\Vert f\Vert_{p}\Vert g\Vert_{q}.
\end{align}
Hence, we have \eqref{eq:FI-2}. The inequality in \eqref{eq:FI-4}
follows similarly. \end{proof}

\begin{lemma}[Data processing inequalities] \label{thm:(Data-processing-inequality).-2}
Assume random variables $U,X,Y,$ and $V$ form a Markov chain $U-X-Y-V$
in this order. Then for $p,q\ge1$, 
\begin{align}
\underline{\Lambda}_{p,q}(X;Y) & \leq\underline{\Lambda}_{p,q}(U;V),\label{eq:FI-19-2}
\end{align}
and for $p,q\le1$, 
\begin{equation}
\overline{\Lambda}_{p,q}(X;Y)\ge\overline{\Lambda}_{p,q}(U;V).\label{eq:FI-5}
\end{equation}
Moreover, if $U$ and $V$ are deterministic functions of $X$ and $Y$ respectively,
then the two inequalities hold for all $p,q\in\bbR$. \end{lemma}

\begin{proof} For any $f:\mathcal{U}\to[0,\infty)$ and $g:\mathcal{V}\to[0,\infty)$,
let $\hat{f}:x\in\calX\mapsto\mathbb{E}[f(U)\mid X=x]$ and $\hat{g}:y\in\calY\mapsto\mathbb{E}[g(V)\mid Y=y]$.
Then we have $\langle f,g\rangle=\langle\hat{f},\hat{g}\rangle$,
and by Jensen's inequality, $\Vert f\Vert_{p}\ge\Vert\hat{f}\Vert_{p}$
and $\Vert g\Vert_{q}\ge\Vert\hat{g}\Vert_{q}$ for $p,q\ge1$, and
the directions of these two inequalities are reversed for $p,q\le1$. These facts establish~\eqref{eq:FI-19-2}
and~\eqref{eq:FI-5}. %The statement for $U$ and $V$
%being respectively deterministic functions of $X$ and $Y$ follows by definition.
 \end{proof}

Similarly to the hypercontractivity regions (see Definition~\ref{def:FRHypercontractivityR}
and Lemma~\ref{thm:HC}), the BL exponents also admit rather natural
information-theoretic characterizations. Define the function 
\begin{align}
 & \phi(Q_{X},Q_{Y}) \!:=\!\inf_{R_{X},R_{Y}}\bigg\{\rvD(R_{X},R_{Y}\|\pi_{XY})\!+\!\frac{1}{p}\, D(R_{X}\|Q_{X})\!-\!\frac{1}{p}\, D(R_{X}\|\pi_{X})\nonumber \\*
 & \qquad\qquad\qquad\qquad\qquad\;+\frac{1}{q}\, D(R_{Y}\|Q_{Y})-\frac{1}{q}\, D(R_{Y}\|\pi_{Y})\bigg\},\label{eq:FIphi}
\end{align}
where according to Convention~\ref{conv:FI}, the infimization is
taken over all pairs of distributions $(R_{X},R_{Y}) \in\calP(\calX)\times\calP(\calY)$ such that all
the relative entropies   in the objective function are finite.
Then we have the following information-theoretic characterizations of the forward  and reverse BL exponents.

\begin{proposition} \label{prop:ITcharacterization} For $p,q\in\mathbb{R}\backslash\{0\}$,
if $(X,Y)\sim \pi_{XY}$, then 
\begin{align}
\underline{\Lambda}_{p,q}(X;Y) & =\inf_{Q_{X},Q_{Y}}\phi(Q_{X},Q_{Y})\quad\mbox{and}\label{eq:FILambdaUnderline}\\
\overline{\Lambda}_{p,q}(X;Y) & =\sup_{Q_{X},Q_{Y}}\phi(Q_{X},Q_{Y}).\label{eq:FILambdaOverline}
\end{align}
\end{proposition}

\begin{proof} The proof leverages the following ``duality'' lemma.
%[Proof of Proposition~\ref{prop:ITcharacterization}] 

\begin{lemma}[Duality of Relative Entropy] \label{lem:dual}Let $\{P_{i}\}_{i=1}^n$ be $n$ probability
mass functions on a finite set $\mathcal{X}$. Let $\{s_{i}\}_{i=1}^n\subset\bbR\setminus\{0\}$
be  nonzero real numbers such that $\sum_{i=1}^{n}s_{i}=1$. Let
$c:\mathcal{X}\to\mathbb{R}$ be a function. Define 
\begin{equation}
\beta:=\sum_{x\in\mathcal{X}}2^{-c(x)}\bigg(\prod_{i=1}^{n}P_{i}(x)^{s_{i}}\bigg).
\end{equation}
Then we have\footnote{We adopt the convention $\inf_{\emptyset}=\infty$, $0\cdot\infty=0$,
and $0^{s}=\infty$ for $s<0$.}
\begin{equation}
-\log\beta=\inf_{Q\ll P_i,\forall\, i\in [n]}\left\{ \sum_{i=1}^{n}s_{i}\, D(Q\|P_{i})+\mathbb{E}_{Q}[c(X)]\right\} .\label{eq:FIduality}
\end{equation}
%where the conventions that $\inf_{\emptyset}=\infty$, $0\cdot\infty=0$,
%and $0^{s}=\infty$ for $s<0$ are adopted, and the infimization
%in the RHS of (\ref{eq:FIduality}) is taken over all probability
%measures $Q$ on $\mathcal{X}$ such that $Q\ll P_i$ for all $1\le i\le n$.
%$\max_{1\le i\le n}D(Q\|P_{i})<\infty$.
Moreover, if $0<\beta<\infty$, the infimization in~\eqref{eq:FIduality}
is uniquely attained by the distribution 
\begin{equation}
Q^{*}(x)=\frac{2^{-c(x)}}{\beta}\Big(\prod_{i=1}^{n}P_{i}(x)^{s_{i}}\Big)\quad\mbox{for all}\;\, x\in\calX.
\end{equation}
\end{lemma}

This lemma was stated by \citet{shayevitz2011renyi}. It can be proved
by using Lagrange multipliers. The generalization of this lemma  to arbitrary measurable spaces can be proven by using the nonnegativity of the relative entropy; see \cite[Theorem~2.2.3]{liu2018information} or  \cite{yu2021strong}.
%  using
%which was proven by using the nonnegativity of the relative entropy. 

We may assume, by homogeneity, that $\Vert f\Vert_{p}=\Vert g\Vert_{q}=1$. 
Without loss of generality, we may  also assume, due to Assumption~\ref{asm:full_supp}, that $\supp(f)\subset\supp(\pi_{X})$
and $\supp(g)\subset\supp(\pi_{Y})$. Hence, we can write 
\begin{equation}
f(x)^{p}=\frac{Q_{X}(x)}{\pi_{X}(x)}\quad\mbox{and}\quad g(y)^{q}=\frac{Q_{Y}(y)}{\pi_{Y}(y)},\label{eq:FI-6}
\end{equation}
for some probability mass functions $Q_{X}$ and $Q_Y$. 
Moreover, since $f$ and $g$ are finite  on their  supports,
$Q_{X}$ and $ \pi_{X}$ are mutually absolutely continuous if $p<0$, and $Q_{Y}$ and $\pi_{Y}$ are mutually absolutely continuous if $q<0$.  From \eqref{eq:FI-6}, we see that
\begin{align}
\hspace{-.2in}  \langle f,g\rangle =\sum_{(x,y)\in\calX\times\calY} \pi_{XY}(x,y) \Big(\frac{Q_X(x) }{\pi_X(x)}\Big)^{1/p}  \Big(\frac{Q_Y(y)}{\pi_Y(y)} \Big)^{1/q}. \label{eqn:fg}
\end{align}
Now substituting~\eqref{eqn:fg} into the definitions of $\underline{\Lambda}_{p,q}$
and $\overline{\Lambda}_{p,q}$, and using   Lemma~\ref{lem:dual} (with the identifications $c \leftarrow 0$, $s_1\leftarrow 1$, $s_2 \leftarrow  1/p$, $s_3\leftarrow -1/p$, $s_4\leftarrow 1/q$, $s_5\leftarrow -1/q$, and $ P_1\leftarrow \pi_{XY}$, $P_2 \leftarrow  Q_X \pi_{Y|X}$, $P_3 \leftarrow  \pi_{XY}$,  $P_4 \leftarrow  Q_Y \pi_{X|Y}$, $P_5 \leftarrow  \pi_{XY}$), 
we obtain Proposition~\ref{prop:ITcharacterization}. 
\end{proof}

Define the following linear combination of relative entropies 
\begin{equation}
\theta(Q_{X},Q_{Y}):=\rvD(Q_{X},Q_{Y}\|\pi_{XY})-\frac{1}{p}\, D(Q_{X}\|\pi_{X})-\frac{1}{q}\, D(Q_{Y}\|\pi_{Y}).
\end{equation}
By using the tensorization property, the BL exponents also can be
written in the following alternative information-theoretic forms in
terms of variational characterizations of $\theta(Q_{X},Q_{Y})$.

\begin{theorem} \label{thm:ITcharacterization-1} For $p,q\in\mathbb{R}\backslash\{0\}$,
if $(X,Y)\sim \pi_{XY}$, then 
\begin{equation}
\underline{\Lambda}_{p,q}(X;Y)=\begin{cases}
{\displaystyle \inf_{Q_{X},Q_{Y}}\theta(Q_{X},Q_{Y})} & p,q>0\vspace{.03in}\\
{\displaystyle -\infty} & p<0\textrm{ or }q<0
\end{cases}\label{eq:FIInfChLambdaUnderline}
\end{equation}
and 
\begin{equation}
\overline{\Lambda}_{p,q}(X;Y)=\begin{cases}
{\displaystyle \sup_{Q_{X},Q_{Y}}\; \theta(Q_{X},Q_{Y})} & p,q>0\vspace{.03in}\\
{\displaystyle \sup_{Q_{X}}\; \inf_{Q_{Y}}\; \theta(Q_{X},Q_{Y})} & q<0<p\vspace{.03in}\\
{\displaystyle \sup_{Q_{Y}}\; \inf_{Q_{X}}\; \theta(Q_{X},Q_{Y})} & p<0<q\vspace{.03in}\\
{\displaystyle 0} & p,q\!<\!0
\end{cases}.\label{eq:FIInfChLambdaOverline}
\end{equation}
\end{theorem}

For Euclidean spaces, the forward part of this theorem, i.e., \eqref{eq:FIInfChLambdaUnderline},
was derived in \citet{carlen2009subadditivity}. The reverse part
of this theorem, i.e., \eqref{eq:FIInfChLambdaOverline}, for finite
alphabets was derived in \citet{beigi2016equivalent} for all $p,q\neq0$,
and also by \citet{liu2016brascamp} for $p,q>0$. 

%Similarly to the
%hypercontractivity regions in Theorem~\ref{thm:HC}, extensions of
%these characterizations to Polish spaces
%% and arbitrary spaces 
%were studied by \citet{liu2018information} and \citet{yu2021strong}.

%. Similarly to the hypercontractivity inequalities in Theorem~\ref{thm:HC},  Liu's proof is based on the Fenchel--Rockafellar duality. Moreover, for the
%reverse part, his proof requires $\mathcal{X},\mathcal{Y}$ to be
%locally compact and $\sigma$-compact Polish spaces. In \cite{yu2021strong},
%Yu provided a proof based on large deviation theory, which relaxes
%this requirement to that $\mathcal{X},\mathcal{Y}$ are arbitrary
%Polish spaces.

The characterizations in \eqref{eq:FIInfChLambdaUnderline} and \eqref{eq:FIInfChLambdaOverline}
are consistent with the ones for the hypercontractivity regions given
in Theorem \ref{thm:HC}. This can be seen observing that $\underline{\Lambda}_{p,q}\ge1$
if and only if $(p,q)\in\mathcal{R}_{\mathrm{FH}}(\pi_{XY})$, and $\underline{\Lambda}_{p,q}\le1$
if and only if $(p,q)\in\mathcal{R}_{\mathrm{RH}}(\pi_{XY})$. Hence,
Theorem \ref{thm:HC} is indeed a consequence of Theorem \ref{thm:ITcharacterization-1}.

\begin{proof}[Proof of Theorem~\ref{thm:ITcharacterization-1}]
The characterization in \eqref{eq:FIInfChLambdaUnderline} follows
directly from~\eqref{eq:FILambdaUnderline} by swapping the two infima.
We now prove the characterization in \eqref{eq:FIInfChLambdaOverline}.
We first consider the case of $p,q>0$. On one hand, by setting $(R_{X},R_{Y})$ in \eqref{eq:FIphi}
to be $(Q_{X},Q_{Y})$, we have that $\phi(Q_{X},Q_{Y})\leq\theta(Q_{X},Q_{Y}).$
Hence, 
\begin{equation}
\overline{\Lambda}_{p,q}(X;Y)\leq\sup_{Q_{X},Q_{Y}}\theta(Q_{X},Q_{Y}).\label{eqn:overLambda_theta}
\end{equation}
On the other hand, by the tensorization property stated in~\eqref{eq:FI-4} in Lemma~\ref{lem:BLtensorization},
for $(X^{n},Y^{n})\sim \pi_{XY}^{n}$, 
\begin{align}
\hspace{-.2in}\overline{\Lambda}_{p,q}(X;Y) & = \frac{1}{n}\overline{\Lambda}_{p,q}(X^{n},Y^{n})\\
 & =\sup_{f,g:
\Vert f\Vert_{p}\Vert g\Vert_{q}>0
}-\frac{1}{n}\log\frac{\langle f,g\rangle}{\Vert f\Vert_{p}\Vert g\Vert_{q}}\\
 & \ge\max_{\mathcal{A}_{n}\subset\mathcal{X}^{n},\mathcal{B}_{n}\subset\mathcal{Y}^{n}}-\frac{1}{n}\log\frac{\pi_{XY}^{n}(\mathcal{A}_{n}\times\mathcal{B}_{n})}{\pi_{X}^{n}(\mathcal{A}_{n})^{1/p}\pi_{Y}^{n}(\mathcal{B}_{n})^{1/q}},\label{eq:FI-8}
\end{align}
where in the last line, we restrict $f$ and $g$ to be the indicators of two non-empty sets
$\mathcal{A}_{n}\subset\mathcal{X}^{n}$ and $\mathcal{B}_{n}\subset\mathcal{Y}^{n}$,
respectively. 

To further lower bound~\eqref{eq:FI-8}, we take $(\mathcal{A}_{n},\mathcal{B}_{n})$
therein to be a pair of type classes $(\mathcal{T}_{T_{X}^{(n)}},\mathcal{T}_{T_{Y}^{(n)}})$ in which the sequence  of pairs of types
$\{(T_{X}^{(n)},T_{X}^{(n)})\}_{n\in\bbN}$ converges to some pair of distributions $(Q_X, Q_Y)$ as $n\to\infty$.
%be a sequence of pairs of types such that $(T_{X}^{(n)},T_{X}^{(n)})\to(Q_{X},Q_{Y})$
%as $n\to\infty$. 
Then, by Sanov's theorem (see Theorem~\ref{thm:sanov}),
\begin{equation}
\lim_{n\to\infty}-\frac{1}{n}\log\frac{\pi_{XY}^{n}\big(\mathcal{T}_{T_{X}^{(n)}}\times\mathcal{T}_{T_{Y}^{(n)}}\big)}{\pi_{X}^{n}\big(\mathcal{T}_{T_{X}^{(n)}}\big)^{1/p}\pi_{Y}^{n}\big(\mathcal{T}_{T_{Y}^{(n)}}\big)^{1/q}}=\theta(Q_{X},Q_{Y}). \label{eqn:sanov_in_prf}
\end{equation}
%%By setting $(\mathcal{A}_{n},\mathcal{B}_{n})$ to $(\mathcal{T}_{T_{X}^{(n)}},\mathcal{T}_{T_{Y}^{(n)}})$
%%and taking limits as $n\to\infty$, 
Hence,  we obtain $\overline{\Lambda}_{p,q}(X;Y)\ge\theta(Q_{X},Q_{Y}).$
Since $(Q_{X},Q_{Y})$ is arbitrary, we have $\overline{\Lambda}_{p,q}(X;Y)\ge\sup_{Q_{X},Q_{Y}}\theta(Q_{X},Q_{Y})$.
Therefore, \eqref{eq:FIInfChLambdaOverline}  holds.

We omit the proofs for other cases, since they are similar to the
above argument. \end{proof}

An interesting observation arising from this proof is the following.
For $p,q>0$, by combining~\eqref{eqn:overLambda_theta} and~\eqref{eq:FI-8}, we obtain
\begin{align}
 & \max_{\mathcal{A}_{n}\subset\calX^{n},\mathcal{B}_{n}\subset\calY^{n}}-\frac{1}{n}\log \pi_{XY}^{n}(\mathcal{A}_{n}\times\mathcal{B}_{n})+\frac{1}{np}\log \pi_{X}^{n}(\mathcal{A}_{n})+\frac{1}{nq}\log \pi_{Y}^{n}(\mathcal{B}_{n})\nonumber \\*
 & \hspace{4cm}\le\sup_{Q_{X},Q_{Y}}\theta(Q_{X},Q_{Y}).\label{eq:FIBLRSSE}
\end{align}
Furthermore, as shown in~\eqref{eqn:sanov_in_prf}, by appealing to Sanov's
theorem, this inequality is {\em asymptotically tight}
(which means that as $n\to\infty$, the limits of the left- and right-hand sides
are equal). Similarly, for $p,q>0$, one can observe that 
\begin{align}
 & \min_{\mathcal{A}_{n}\subset\calX^{n},\mathcal{B}_{n}\subset\calY^{n}}-\frac{1}{n}\log \pi_{XY}^{n}(\mathcal{A}_{n}\times\mathcal{B}_{n})+\frac{1}{np}\log \pi_{X}^{n}(\mathcal{A}_{n})+\frac{1}{nq}\log \pi_{Y}^{n}(\mathcal{B}_{n})\nonumber \\*
 & \hspace{4cm}\ge\inf_{Q_{X},Q_{Y}}\theta(Q_{X},Q_{Y}),\label{eq:FIBLFSSE}
\end{align}
and this inequality is also asymptotically tight by Sanov's theorem~\cite{Dembo}.

As a consequence of~\eqref{eq:FIBLRSSE} and~\eqref{eq:FIBLFSSE}, we find that
certain sequences of $\{0,1\}$-valued functions attain the BL exponents.
Hence, a BL inequality holds for all nonnegative functions if and
only if any of its  multi-dimensional extensions hold for any $\{0,1\}$-valued
functions.  In addition to the set of $\{0,1\}$-valued functions, one can
also use the following construction of functions to assert the (asymptotic) optimality
of a BL inequality. We can first identify a optimal pair $(f^{*},g^{*})$
for the one-dimensional case. By the tensorization property of the
BL exponents (Theorem~\ref{lem:BLtensorization}), the $n$-fold
product of $(f^{*},g^{*})$ also constitutes an optimal pair that allows
us to assert the optimality of a BL inequality. In contrast, the asymptotic
optimality of $\{0,1\}$-valued functions is advantageous in
our quest to prove the strong SSE theorem (Theorem~\ref{thm:strongsse-2}) as will
be done in Section~\ref{sec:NICD-Stability}. % More on the relation between the small-set
%expansion theorem and the information-theoretic characterization of
%BL inequalities will be discussed in Section \ref{sec:BooleanOptimality}.

\subsection{Single-Function Versions}

\label{subsec:Single-Function-Version} The BL inequalities discussed
in Section~\ref{sec:BL_ineq} involve {\em two} nonnegative functions.
In the literature, there exist {\em single-function} versions
of BL inequalities and they have been shown to be equivalent to their
two-function counterparts (as was discussed in the context of the DSBS in Section~\ref{sec:ab_half}).  We now introduce the single-function versions of BL inequalities. First recall from~\eqref{eqn:noise_op} that the {\em conditional
expectation operator} induced by $\pi_{X|Y}$ is the operator that maps  a  function $f:\mathcal{X}\to\mathbb{R}$ to  the  function 
\begin{equation}
y\in \calY \mapsto \pi_{X|Y=y}(f):=\bbE\big[f(X)\,\big|\, Y=y\big]=\sum_{x\in\mathcal{X}}\pi_{X|Y}(x|y)f(x).\label{eqn:cond_expectation_fi}
\end{equation}
 Then, given a joint distribution $\pi_{XY}$ and two real
numbers $p$ and $q$, for any nonnegative function $f:\mathcal{X}\to[0,\infty)$, the single-function versions
of the BL inequalities read
\begin{align}
\Vert \pi_{X|Y}(f)\Vert_{q} & \le\overline{C} \, \Vert f\Vert_{p}\quad\mbox{and}\label{eq:FIFBL-2}\\
\Vert \pi_{X|Y}(f)\Vert_{q} & \geq\underline{C} \, \Vert f\Vert_{p}\label{eq:FIRBL-3}
\end{align}
for some constants $\overline{C}$ and $\underline{C}$. 

 We remark that~\eqref{eq:FIFBL-2} and~\eqref{eq:FIRBL-3}
are in fact equivalent to the {\em strong
data processing inequalities} for the R\'enyi divergence  \cite{raginsky2016strong}.  The latter concerns the tradeoff between
 $D_{p}(Q_{X}\|\pi_{X})$ and $D_{q}(Q_{Y}\|\pi_{Y})$,
where $Q_Y$ represents the output distribution induced by the input distribution $Q_{X}$ and the stochastic kernel  $\pi_{Y|X}$, i.e., $Q_X\rightarrow \pi_{Y|X}\rightarrow Q_Y$. 
The equivalence follows since  we can
set $f={Q_{X}}/{\pi_{X}}$ and observe that
\begin{align}
\log\Vert f\Vert_{p} & =\frac{1}{p'}\, D_{p}(Q_{X}\|\pi_{X}) \label{eqn:strong_renyi_dpi}
\end{align}
and 
\begin{align}
\log\Vert \pi_{X|Y}(f)\Vert_{q} & =\frac{1}{q}\, \log\sum_{y\in \calY}\bigg(\sum_{x\in \calX}\frac{Q_{X}(x)}{\pi_{X}(x)}\pi_{X|Y}(x|y)\bigg)^{q}\pi_{Y}(y)\\
 & =\frac{1}{q}\, \log\sum_{y\in \calY}\Big(\frac{Q_{Y}(y)}{\pi_{Y}(y)}\Big)^{q}\pi_{Y}(y)\\*
 & =\frac{1}{q'}\,  D_{q}(Q_{Y}\|\pi_{Y}).
\end{align}
For more details, see the papers by \citet{raginsky2016strong} and \citet{yu2021strong}.

The promised equivalence between the single- and two-function versions of the BL inequalities is formalized in the following proposition. 

%Furthermore, as mentioned, the single-function versions of the BL inequalities in \eqref{eq:FIFBL-2}
%and~\eqref{eq:FIRBL-3} are also equivalent to the two-functions
%versions in \eqref{eq:FIFBL} and \eqref{eq:FIRBL}, as shown in the
%following proposition. 

%It is known that the inequality in~\eqref{eq:FIFBL-2}
%for $q\ge1$ and the inequality in~\eqref{eq:FIRBL-3} for $q\le1$
%are respectively equivalent to the two-function versions in~\eqref{eq:FIFBL}
%and \eqref{eq:FIRBL}.  This is stated formally as follows.

\begin{proposition} \label{prop:equivalence} Inequality~\eqref{eq:FIFBL-2}
for $q\ge1$ holds if and only if~\eqref{eq:FIFBL}
holds but with $q$ in the latter replaced by its H\"older conjugate
$q'=\frac{q}{q-1}$. Similarly, inequality~\eqref{eq:FIRBL-3}
for $q\le1$ holds if and only if~\eqref{eq:FIRBL}
holds but with $q$ in the latter replaced by its H\"older conjugate
$q'$. \end{proposition}

%By this equivalence, one can find the equivalence between the two-function
%version of hypercontractivity inequalities and the single-function
%version (i.e., \eqref{eq:FIFBL-2} and \eqref{eq:FIRBL-3} with $\overline{C}$
%and $\underline{C}$ set to $1$). % removing this because it's a repetition

%Moreover, it is worth nothing 
%
%This proposition verifies the promised equivalence between the two-function
%version and single-function version of hypercontractivity inequalities
%for the DSBS, which was given below Theorem~\ref{thm:hyper_single}. 

\begin{proof}By 
H\"older's inequality, for any $\hat{g}:\mathcal{Y}\to[0,\infty)$,
it holds that 
\begin{equation}
\Vert\hat{g}\Vert_{q}=\begin{cases}
{\displaystyle \sup_{g :\Vert g\Vert_{q'}>0}\frac{\langle\hat{g},g\rangle}{\Vert g\Vert_{q'}}} & q\ge1\vspace{.03in}\\
{\displaystyle \inf_{g :\Vert g\Vert_{q'}>0}\frac{\langle\hat{g},g\rangle}{\Vert g\Vert_{q'}}} & q\le1
\end{cases} ,
\end{equation}
where $1'=\infty$ and $1'=-\infty$ for the first and second clauses
respectively. Setting $\hat{g}$ to be $\pi_{X|Y}(f)$, we obtain the
following equivalences: For $q\ge1$, 
\begin{align}
 & \sup_{f :\Vert f\Vert_{p}>0}\frac{\Vert \pi_{X|Y}(f)\Vert_{q}}{\Vert f\Vert_{p}} =  \sup_{(f,g) :\Vert f\Vert_{p}>0 , \Vert g\Vert_{q'}>0}\;\frac{\langle f,g\rangle}{\Vert f\Vert_{p}\Vert g\Vert_{q'}},\label{eq:FI-38}
\end{align}
and  for $q\le1$, 
\begin{align}
 & \inf_{f :\Vert f\Vert_{p}>0}\frac{\Vert \pi_{X|Y}(f)\Vert_{q}}{\Vert f\Vert_{p}}=  \inf_{(f,g): \Vert f\Vert_{p}>0, \Vert g\Vert_{q'}>0}\;\frac{\langle f,g\rangle}{\Vert f\Vert_{p}\Vert g\Vert_{q'}}.\label{eq:FI-38-2}
\end{align}
%\begin{align}
% & \sup_{f:\mathcal{X}\to[0,\infty):\Vert f\Vert_{p}>0}\frac{\Vert \pi_{X|Y}(f)\Vert_{q}}{\Vert f\Vert_{p}}\nonumber \\
% & =\begin{cases}
%{\displaystyle \sup_{f:\mathcal{X}\to[0,\infty):\Vert f\Vert_{p}>0}\;\sup_{g:\mathcal{Y}\to[0,\infty):\Vert g\Vert_{q'}>0}\;\frac{\langle f,g\rangle}{\Vert f\Vert_{p}\Vert g\Vert_{q'}}} & q\ge1\vspace{.03in}\\
%{\displaystyle \sup_{f:\mathcal{X}\to[0,\infty):\Vert f\Vert_{p}>0}\;\inf_{g:\mathcal{Y}\to[0,\infty):\Vert g\Vert_{q'}>0}\;\frac{\langle f,g\rangle}{\Vert f\Vert_{p}\Vert g\Vert_{q'}}} & q\le1
%\end{cases}.\label{eq:FI-38}
%\end{align}
By the equivalence in \eqref{eq:FI-38}, for $q\ge1$, the single-function version of
BL inequality in~\eqref{eq:FIFBL-2} is equivalent to the two-function
version in~\eqref{eq:FIFBL} with $\overline{C}$ and $p$ unchanged
but with $q$ replaced by its H\"older conjugate~$q'$.
Similarly, for $q\le1$, by the equivalence in \eqref{eq:FI-38-2}, the single-function version of BL inequality
in~\eqref{eq:FIRBL-3} is equivalent to the two-function version
in~\eqref{eq:FIRBL} with $\underline{C}$ and $p$ unchanged but
with $q$ replaced by $q'$. \end{proof} 

\section{Connections to the NICD Problem and   $q$-Stability }
\label{sec:NICD-Stability} %  and the BL Exponents 

As observed in the proof of Theorem~\ref{thm:ITcharacterization-1},
certain sequences of $\{0,1\}$-valued functions attain the BL exponents.
%It means that a BL inequality holds for all nonnegative functions
%if and only if it holds for any Boolean functions (and for any dimensions).
We now provide a detailed discussion on this observation. We also
discuss the connections between the BL exponents and the NICD problem
(Section~\ref{ch:NICD}), as well as the   $q$-stability problem
(Section~\ref{ch:Stability}).

\enlargethispage{-2\baselineskip}
Recall the general version of the strong SSE theorem (Theorem~\ref{thm:strongsse-2})   and the  general version of the strong $q$-stability
theorem  (Theorem~\ref{thm:strongqstability}).
For the LD exponents $\underline{\Upsilon}_{\mathrm{LD}}^{(n)}$ and $\overline{\Upsilon}_{\mathrm{LD}}^{(n)}$
defined in \eqref{eq:FI-72} and \eqref{eq:FI-71}, the strong SSE
theorem states that for $\pi_{XY}$ defined on a finite alphabet, any $n\ge1$,
$\alpha\in(0,\alpha_{\max}(\pi_{X})]$, and $\beta\in(0,\beta_{\max}(\pi_{Y})]$,
it holds that 
\begin{align}
\underline{\Upsilon}_{\mathrm{LD}}^{(n)}(\alpha,\beta) & \ge\bbL[\underline{\Upsilon}_{\mathrm{LD}}](\alpha,\beta)\quad\mbox{and}\label{eq:NICDLD_FI}\\
\overline{\Upsilon}_{\mathrm{LD}}^{(n)}(\alpha,\beta) & \leq\bbU[\overline{\Upsilon}_{\mathrm{LD}}](\alpha,\beta).\label{eq:NICDLD2_FI}
\end{align}
%where $\underline{\Upsilon}_{\mathrm{LD}}$ and $\overline{\Upsilon}_{\mathrm{LD}}$
%coincide with $\underline{\Upsilon}$ and $\overline{\Upsilon}$ defined
%in \eqref{eq:FIlce} and \eqref{eq:FIuce} respectively. 
Moreover, the inequalities
in \eqref{eq:NICDLD_FI} and \eqref{eq:NICDLD2_FI} are   asymptotically
tight in the limit as $n\to\infty$.  We now provide a proof of
the strong SSE theorem by leveraging its connections to the information-theoretic
characterizations of BL exponents.

\begin{proof}[Proof of  Theorem~\ref{thm:strongsse-2}] Observe
that \eqref{eq:FIBLRSSE} and \eqref{eq:FIBLFSSE} for $p,q>0$ can
be rewritten as follows. For all $\mathcal{A}_{n}\subset\mathcal{X}^{n}$
and $\mathcal{B}_{n}\subset\mathcal{Y}^{n}$, 
\begin{align}
 & -\frac{1}{n}\log \pi_{XY}^{n}(\mathcal{A}_{n}\times\mathcal{B}_{n})+\frac{1}{np}\log \pi_{X}^{n}(\mathcal{A}_{n})+\frac{1}{nq}\log \pi_{Y}^{n}(\mathcal{B}_{n})\nonumber \\
 & \qquad\ge\inf_{s,t\ge0} \underline{\varphi} (s,t)-\frac{s}{p}-\frac{t}{q}\ge\inf_{s,t\ge0}\underline{\Upsilon}(s,t)-\frac{s}{p}-\frac{t}{q},\label{eq:FIBLFSSE-1}
\end{align}
where $\underline{\varphi}$ and $\underline{\Upsilon}$ are defined in \eqref{eqn:def_varphi}
and \eqref{eq:FIlce} respectively. Analogously, for all $\mathcal{A}_{n}\subset\mathcal{X}^{n}$
and $\mathcal{B}_{n}\subset\mathcal{Y}^{n}$, 
\begin{align}
 & -\frac{1}{n}\log \pi_{XY}^{n}(\mathcal{A}_{n}\times\mathcal{B}_{n})+\frac{1}{np}\log \pi_{X}^{n}(\mathcal{A}_{n})+\frac{1}{nq}\log \pi_{Y}^{n}(\mathcal{B}_{n})\nonumber \\
 & \qquad\le\sup_{s,t\ge0}\overline{\varphi}(s,t)-\frac{s}{p}-\frac{t}{q}\le\sup_{s,t\ge0}\overline{\Upsilon}(s,t)-\frac{s}{p}-\frac{t}{q},\label{eq:FIBLRSSE-1}
\end{align}
where $\overline{\varphi}$ and $\overline{\Upsilon}$ are defined in \eqref{eqn:def_psi}
and \eqref{eq:FIuce} respectively. For any $(\mathcal{A}_{n},\mathcal{B}_{n})$,
set $a:=-\frac{1}{n}\log \pi_{X}^{n}(\mathcal{A}_{n})$ and $b:=-\frac{1}{n}\log \pi_{Y}^{n}(\mathcal{B}_{n})$.
Let $(u,v)$ be a subgradient\footnote{Let $\calI\subset\bbR^d$ be convex. A vector $\bg\in\bbR^d$ is a {\em subgradient} of $f:\calI\to\bbR$ at $\bx\in\calI$ if for all $\bz\in\calI$, $f(\bz)\ge f(\bx)+\langle\bg,\bz-\bx\rangle$.}  of $\underline{\Upsilon}$ at $(a,b)$. Since $\underline{\Upsilon}$ is convex and {\em nondecreasing}, $u,v\ge0$.
%Such a  subgradient  $(u,v)$ exists, since  $\underline{\Upsilon}$ is nondecreasing.
%Recall that if let    $\calI$ be  a  convex subset  of  $\bbR^n$, then   a subgradient of a convex function $f:\calI\to \bbR$  at a point $x^n_0\in \calI$  is a real vector $c^n\in \bbR^n$ such that
%\begin{equation}
% f(x^n)-f(x^n_{0})\geq  \langle c^n,x^n-x^n_{0}\rangle \; \mbox{for all } x^n\in\calI,
%\end{equation} 
%where $\langle x^n,y^n\rangle:=\sum_{i=0}^n x_i y_i$ for $x^n,y^n\in \bbR^n$. 
%%Since $\underline{\Upsilon}$ is nondecreasing, we know $u,v\ge0$.
Hence, by definition of the subgradient, 
\begin{equation}
\inf_{s,t\ge0}\underline{\Upsilon}(s,t)-us-vt=\underline{\Upsilon}(a,b)-ua-vb. \label{eq:FI_subgrad_equality}
\end{equation}
Substituting $p=1/u$ and $q=1/v$ into \eqref{eq:FIBLFSSE-1} and utilizing  \eqref{eq:FI_subgrad_equality}, we
have 
\begin{equation}
-\frac{1}{n}\log \pi_{XY}^{n}(\mathcal{A}_{n}\times\mathcal{B}_{n})\ge\underline{\Upsilon}(a,b).
\end{equation}
Similarly, by using \eqref{eq:FIBLRSSE-1}, we have 
\begin{equation}
-\frac{1}{n}\log \pi_{XY}^{n}(\mathcal{A}_{n}\times\mathcal{B}_{n})\leq\overline{\Upsilon}(a,b).
\end{equation}
Hence, 
\begin{align}
\underline{\Upsilon}^{(n)}(\alpha,\beta) & \ge\min_{a\ge\alpha,b\ge\beta}\underline{\Upsilon}(a,b)=\underline{\Upsilon}(\alpha,\beta)\quad\mbox{and}\\
\overline{\Upsilon}^{(n)}(\alpha,\beta) & \le\max_{a\le\alpha,b\le\beta}\overline{\Upsilon}(a,b)=\overline{\Upsilon}(\alpha,\beta).
\end{align}
Finally, the asymptotic tightness of \eqref{eq:NICDLD} and \eqref{eq:NICDLD2}
can be verified by appealing to Sanov's theorem (Theorem~\ref{thm:sanov}). \end{proof}

The strong SSE theorem (Theorem \ref{thm:strongsse-2}) can be further strengthened if
the {\em exact} values of the marginal probabilities are given, instead
of  only bounds as in the definitions of $\underline{\Upsilon}^{(n)}$ and
$\overline{\Upsilon}^{(n)}$. It has been shown in \cite{yu2021strong}
that for all $\mathcal{A}_{n}\subset\mathcal{X}^{n}$ and $\mathcal{B}_{n}\subset\mathcal{Y}^{n}$,
\begin{equation}
\underline{\Upsilon}(a,b)\le-\frac{1}{n}\log \pi_{XY}^{n}(\mathcal{A}_{n}\times\mathcal{B}_{n})\leq\mathbb{U}[\overline{\varphi}](a,b), \label{eqn:strengthened_Upsilon}
\end{equation}
where $a=-\frac{1}{n}\log \pi_{X}^{n}(\mathcal{A}_{n})$,  $b=-\frac{1}{n}\log \pi_{Y}^{n}(\mathcal{B}_{n})$, and $\overline{\varphi}$ was defined in~\eqref{eqn:def_psi}.
Moreover,   the lower and upper bounds in~\eqref{eqn:strengthened_Upsilon} are asymptotically
tight as $n\to\infty$. 
%That is, for any $\alpha\in(0,\alpha_{\max}(\pi_{X})]$
%and $\beta\in(0,\beta_{\max}(\pi_{Y})]$, there exists a sequence of
%pairs of sets $\{(\calA_{n},\calB_{n})\}_{n\in\bbN}$ such that 
%\begin{equation}
%-\frac{1}{n}\log \pi_{X}^{n}({\calA}_{n})\downarrow\alpha,\qquad-\frac{1}{n}\log \pi_{Y}^{n}({\calB}_{n})\downarrow\beta,
%\end{equation}
%and 
%\begin{equation}
%\lim_{n\to\infty}-\frac{1}{n}\log \pi_{XY}^{n}({\calA}_{n}\times{\calB}_{n})=\underline{\Upsilon}(\alpha,\beta).
%\end{equation}
%There also exists another sequence $\{(\calA_{n},\calB_{n})\}_{n\in\bbN}$
%such that 
%\begin{equation}
%-\frac{1}{n}\log \pi_{X}^{n}({\calA}_{n})\uparrow\alpha,\qquad-\frac{1}{n}\log \pi_{Y}^{n}({\calB}_{n})\uparrow\beta
%\end{equation}
%and 
%\begin{equation}
%\lim_{n\to\infty}-\frac{1}{n}\log \pi_{XY}^{n}({\calA}_{n}\times{\calB}_{n})=\overline{\Upsilon}(\alpha,\beta).
%\end{equation}

A similar relation can be found between the single-function version
of the BL exponents and the notion of $q$-stability discussed in
Section~\ref{sec:form_q_stability}.  Following steps similar to
the proof for the strong SSE theorem, one can
prove the strong $q$-stability theorem (Theorem~\ref{thm:strongqstability}).
We omit the details here. Furthermore, similarly to the strong SSE theorem, the strong $q$-stability theorem
can be further strengthened if the marginal probabilities are specified.
For details, see \cite{yu2021strong}. %,
%just similar to the further strengthening of the strong small-set
%expansion theorem mentioned above.

\section{Logarithmic Sobolev Inequalities}
\label{sec:log_sobolev}

We discuss the  {\em logarithmic Sobolev} (or {\em log-Sobolev}) {\em inequalities} in 
this section. It will be seen (from Theorem~\ref{thm:LSI-HC}) that such inequalities
 turn out to be {\em equivalent}, in sense to be made precise,
to the hypercontractivity inequalities (cf.\ Section~\ref{sec:holder_HC}).
We will also focus on information-theoretic characterizations of certain
log-Sobolev inequalities. For more details on the classical
aspects of this rich topic, the reader is referred to \citet{RagSason}
and \citet{Ledoux_book}. The results in this section serve as
important elements of the proofs of the main results in Section~\ref{sec:stronger} 
in which the classic hypercontractivity inequalities are strengthened.
This section thus forms a bridge between the classic hypercontractivity
inequalities and their strengthened versions. 

%For this section only, all
%logarithms are to the natural base $\rme$ and $\ln$ denotes the
%natural logarithm. 

\subsection{Preliminaries on Dirichlet forms and Entropies}
Let $\mathcal{X}=\mathcal{Y}$. As assumed, $\mathcal{X}$ is a finite set.  %Note that any linear operator acting on real-valued functions defined on  a  finite set   can be represented by a matrix. Hence, when  there is no ambiguity, we do not distinguish between a linear operator acting on real-valued functions defined on $\mathcal{X}$ and a matrix of size $|\calX| \times |\calX|$. 
Let   $\bL$ be a $|\calX|\times|\calX|$
  matrix (a linear operator acting on real-valued functions defined on $\calX$) such that $L_{x,y}\ge0$ for $x\neq y$ and $\sum_{y\in\mathcal{X}}L_{x,y}=0$
for all $x$. Let $T_{t}:=\rme^{t\bL}$  ($t\ge0$) be a matrix  induced by $\bL$, where $\rme^{\bA}$
denotes the {\em matrix exponential} of~$\bA$. 
%\footnote{Here, the notation $\rme^{t\bL}$ denotes the {\em  matrix exponential}  which is defined as the convergent series $\sum_{j=0}^\infty {t^j \bL^j}/{j!}$.} 
% The operator $T_{t}$, defined by $T_{t}f(x)= \sum_{y\in\mathcal{X}} [T_{t}]__{x,y} f(y)$ for all $f: \calX \to \bbR$, is known as an {\em Markov operator}.
The operator $T_{t}$ is known as a  {\em Markov operator}, which is one that sends a real-valued function on  $\mathcal{X}$ to another real-valued function on  $\mathcal{X}$.
%The transpose $T^{\top}_{t}$  of $T_{t}$ corresponds to  the dual operator   $T^*_{t}$  of  $T_{t}$, which is   a (regular) conditional   distribution or \emph{Markov kernel} which sends a   measure on   $\calX$ to another measure on  $\calX$ \cite{liu2018information}. 
In addition, $\{T_{t}\}_{t\ge0}$ forms
a {\em Markov semigroup}, since it  satisfies the semigroup property, namely  that
$T_{t+s}=T_{t}T_{s}=T_{s}T_{t}$ for all $s,t\ge0$. 
%However, it is not a group since in general, there is 
%no inverse element for each element $T_{t}$  for $t>0$.  
For more details
on Markov operators and Markov semigroups, the reader is referred
to \citet{rudnicki2002markov,bakry2013analysis}.

Let $\pi$ be a stationary distribution corresponding to $\{T_{t}\}_{t\ge0}$,
i.e., $\pi=\pi T_{t}$ for all $t\ge0$ or, equivalently, $\pi\bL=\bzero$. We can regard $\pi$ and $T_t$ (for a fixed $t\ge0$) as corresponding to $\pi_Y$ and $\pi_{X|Y}$ respectively. 
%To connect to previous sections, especially
%Section \ref{subsec:Single-Function-Version}, one can consider the
%distribution $\pi$   as the  $\pi_{Y}$ in Section \ref{subsec:Single-Function-Version},
%and $T_{t}$ for a fixed $t\ge0$ as the conditional expectation operator induced by 
%%$\pi_{X|Y}$. 
%More
%precisely, 
As such, the $y^{\mathrm{th}}$ row of the matrix $T_{t}$ is $\pi_{X|Y}(\cdot|y)$.
As usual, denote  the inner product for two real-valued functions $f$ and $g$ defined on $\mathcal{X}$ as $\langle f,g\rangle_{\pi}:=\mathbb{E}_{\pi}[fg]=\sum_{x \in \calX}\pi(x)f(x)g(x)$.
\begin{definition}
The {\em Dirichlet form} of $\{T_{t}\}_{t\ge0}$ is 
\begin{equation}
\mathcal{E}(f,g):=-\sum_{(x,y )\in\calX^2}L_{x,y}f(y)g(x)\pi(x)=-\langle\bL f,g\rangle_{\pi},\label{eqn:dirichlet_form}
\end{equation}
where $(\bL f)(x):=\sum_{y\in\calX}L_{x,y}f(y)$. The {\em normalized
Dirichlet form} of $\{T_{t}\}_{t\ge0}$ is
\begin{equation}
\overline{\mathcal{E}}(f,g):=\frac{\mathcal{E}(f,g)}{\langle f,g\rangle_{\pi}}.
\end{equation}
\end{definition}
We now extend the definitions of the Dirichlet form and its normalized
version to the $n$-dimensional Cartesian product space $\mathcal{X}^{n}$. Let $T_{t}^{\otimes n}$
be the product semigroup on $\mathcal{X}^{n}$ induced by $T_{t}$.
Recall from Section~\ref{sec:influ}, that given a vector $x^{n}\in\calX^{n}$, let $x^{\setminus k}:=(x_{1},\ldots,x_{k-1},x_{k+1},\ldots,x_{n})\in\calX^{n-1}$
be the subvector of $x^{n}$ with the $k^{\mathrm{th}}$ coordinate
removed. For two real-valued functions $f$ and $g$ defined on $\mathcal{X}^{n}$,
let 
\begin{equation}
\psi(x^{\setminus k}):=\mathcal{E}\big(f(x^{\setminus k},\cdot),g(x^{\setminus k},\cdot)\big)
\end{equation}
be the action of the Dirichlet form $\mathcal{E}$ on the $k^{\mathrm{th}}$
coordinates of $f$ and~$g$ with other coordinates held fixed. Then,
the Dirichlet form of $f$ and~$g$ and its normalized version are
respectively given by 
\begin{align}
\mathcal{E}_{n}(f,g) & :=\sum_{k=1}^{n}\sum_{x^{\setminus k}\in\mathcal{X}^{n-1}}\psi(x^{\setminus k})\prod_{j \in [n]\setminus \{k\} }\pi(x_{j})\quad\mbox{and}\label{eq:FImaten}\\
\overline{\mathcal{E}}_{n}(f,g) & :=\frac{\mathcal{E}_{n}(f,g)}{\langle f,g\rangle_{\pi^{n}}}.\label{eqn:norm_df}
\end{align}

In addition to the Dirichlet form, the other quantity involved in log-Sobolev inequalities is the \emph{entropy} of a nonnegative function~$f$. 
\begin{definition}
For a nonnegative function $f$, the \emph{entropy} and the \emph{normalized
entropy }of $f$ are respectively defined as 
\begin{align}
\Ent(f) & :=\mathbb{E}_{\pi}[f\ln f]-\mathbb{E}_{\pi}[f]\ln\mathbb{E}_{\pi}[f]\quad\mbox{and}\quad \overline{\Ent}(f)  :=\frac{\Ent(f)}{\mathbb{E}_{\pi}[f]}.\label{eq:FInormalizedentropy}
\end{align}
\end{definition}
Note that these notions of entropy and normalized entropy are commonly encountered in functional analysis; see, for example, \citet{Ledoux_book}. They are related
to, but not the same as the  Shannon entropy in classical information
theory. Indeed, they bear more similarity  to the relative entropy,  in
the sense that if $f$ is the Radon--Nikodym derivative  ${\rmd Q}/{\rmd\pi}$ of a distribution $Q$ with respect to
$\pi$ (i.e., the function $x\in\calX\mapsto{Q(x)}/{\pi(x)}$ for the finite alphabet case),
then the entropy (and also the normalized entropy) of $f$ is equal
to the relative entropy of $Q$ from $\pi$, i.e., $D(Q\|\pi)$. 
By Jensen's inequality, both the entropy and the normalized entropy
are nonnegative. 

\subsection{Log-Sobolev Inequalities and Their Properties}

The {\em log-Sobolev} inequalities  quantify the
relation between the Dirichlet form of a Markov semigroup for an arbitrary
nonnegative function $f$ and a composite function $g = \varphi \circ f$ for some  given $\varphi:[0,\infty)\to [0,\infty)$,  and the entropy
of $f$. For $p \in\bbR\setminus \{0,1\}$,  let 
\begin{equation}
c_{p}:=\frac{p^{2}}{4(p-1)}.
\end{equation}
Following the definitions in \citet{mossel2013reverse}, we
define log-Sobolev inequalities as follows.

\begin{definition} For $p\in\mathbb{R}\backslash\{0,1\}$, the {\em
$p$-log-Sobolev inequality with constant $C$} is 
\begin{equation}
\Ent(f^{p})\le C\,  c_{p}\, \mathcal{E}(f,f^{p-1})\label{eq:FIpLSI}
\end{equation}
for   nonnegative $f$ if $p>1$ and for   positive~$f$ if $p<1$.For $p=1$, the {\em $1$-log-Sobolev inequality with constant $C$}
 for    positive~$f$ is 
\begin{equation}
\Ent(f)\le\frac{C}{4}\, \mathcal{E}(f,\ln f).\label{eq:FI1LSI}
\end{equation}
%for any nonnegative $f$.
For $p=0$, the {\em $0$-log-Sobolev inequality with constant $C$}
 for    positive~$f$ is
\begin{equation}
\mathrm{Var}(\ln f)\le-\frac{C}{2}\, \mathcal{E}\Big(f,\frac{1}{f}\Big).\label{eq:FI0LSI}
\end{equation}
\end{definition} 

%for any positive $f$. 
The cases corresponding to $p=0$ and $p=1$ of the $p$-log-Sobolev
inequality are the limiting cases of the $p$-log-Sobolev inequality
for $p\in\mathbb{R}\backslash\{0,1\}$ with the same constant $C$.

We now connect the log-Sobolev inequality and the classic hypercontractivity
inequalities in \eqref{eq:FIFBL-2} and~\eqref{eq:FIRBL-3} with
$\overline{C}$ and $\underline{C}$ set to~$1$. Indeed, we will see from Theorem~\ref{thm:LSI-HC} that the log-Sobolev
inequalities are {\em differential versions} of the hypercontractivity
inequalities evaluated at $t=0$.
%, and the hypercontractivity inequalities
%are  integral versions of log-Sobolev inequalities, 
%as formalized in
%the following theorem. 
This theorem is a classical result due to Gross
\cite{gross1975logarithmic}, and various proofs can be found in \cite{gross1975logarithmic, bakry1994hypercontractivite, ane2000inegalites, bakry2004functional,mossel2013reverse}.

Here we provide a short self-contained proof.

\begin{theorem}[Differential relationship between log-Sobolev and hypercontractivity inequalities] \label{thm:LSI-HC} Let $C$ be a positive constant.
Let $q:[0,\infty)\to\bbR$ be defined as 
\begin{align}
q(t)=1+(p-1)\, \rme^{4t/C}.\label{eqn:def_qt}
\end{align}
%Then, we have the following statements: 
\begin{enumerate}
\item[(a)] \label{itm:one} Fix $p>1$. If for any $r\in[p,\infty)$, the $r$-log Sobolev inequality is satisfied with constant $C$, then for any $t>0$,
\begin{align}
\Vert T_{t}f\Vert_{q(t)} & \le\Vert f\Vert_{p}\quad\mbox{for all}\;\, f\ge0,\label{eq:FIFHC-1}
\end{align}
where $(T_{t}f)(x)=\sum_{y}T_{t}(x,y)f(y)$.
\item[(b)] \label{itm:two} Fix $p<1$.  If for any $r\in(-\infty,p]$, the $r$-log-Sobolev
inequality is satisfied with constant $C$, then for any $t>0$, 
\begin{align}
\Vert T_{t}f\Vert_{q(t)} & \ge\Vert f\Vert_{p}\quad\mbox{for all} \;\,  f\ge0.\label{eq:FIRHC-1}
\end{align}
\item[(c)] \label{itm:three} Conversely, if \eqref{eq:FIFHC-1} holds for $p>1$
or \eqref{eq:FIRHC-1} holds for $p<1$, then the $p$-log-Sobolev
inequality is satisfied with constant $C$. 
\end{enumerate}
\end{theorem}

The inequalities in \eqref{eq:FIFHC-1} and \eqref{eq:FIRHC-1} are
respectively equivalent to the fact that\footnote{Here $q(t)'$ denotes the the H\"older conjugate of $q(t)$. In contrast,
we use $q'(t)$ to denote the derivative of $q$ evaluated at $t$. } $(p,q(t)')$ belongs to the forward and reverse hypercontractivity
regions of the joint distributions (Definition~\ref{def:FRHypercontractivityR}) induced by $(T_{t},\pi)$ for any $t>0$. In fact, the  relations between
the BL inequalities and generalized $p$-log-Sobolev inequalities
can also be established. For details, the reader is referred to \cite{gross1975logarithmic,bakry1994hypercontractivite,ane2000inegalites,bakry2004functional,mossel2013reverse}.
%  by  a similar method. 

\begin{proof}[Proof Sketch of Theorem~\ref{thm:LSI-HC}]  We first prove Statement (a) in which we
assume that $p>1$. Define the function $\zeta:[0,\infty)^{2}\to\bbR$
as %the following function on $\mathbb{R}_{\ge0}^{2}$ 
\begin{equation}
\zeta(t,s):=\ln\|T_{t}f\|_{\frac{1}{s}}\,.
\end{equation}
Then, one can check by direct differentiation that 
\begin{align}
\hspace{-.3in}\frac{\partial\zeta}{\partial s}  =-\overline{\Ent}\big((T_{t}f)^{\frac{1}{s}}\big)\quad\mbox{and}\quad\frac{\partial\zeta}{\partial t}=-\overline{\mathcal{E}}_{n}\big(T_{t}f,(T_{t}f)^{\frac{1}{s}-1}\big).\label{eq:FIhca_2-2b}
\end{align}
We define $\xi(t):=1/q(t)$, and hence, $\xi(0)=1/p$ (refer to \eqref{eqn:def_qt}).
Therefore, by \eqref{eq:FIhca_2-2b} and the chain rule, 
\begin{equation}
\frac{\mathrm{d}}{\mathrm{d}t}\zeta(t,\xi(t))=-\overline{\Ent}\Big((T_{t}f)^{\frac{1}{\xi(t)}}\Big)\xi'(t)-\overline{\mathcal{E}}_{n}\Big(T_{t}f,(T_{t}f)^{\frac{1}{\xi(t)}-1}\Big).\label{eq:FI-12}
\end{equation}
Observe that 
\begin{equation}
\xi'(t)=\frac{-4(p-1)\rme^{4t/C}}{C\, q(t)^{2}}.
\end{equation}
%and $\xi''(t)\ge0$ for $t>0$. 
It also holds that $\xi'(t)\ge -{1}/{(C\, c_{q(t)} )}$
for all $t\ge0$. Combining this with~\eqref{eq:FI-12} yields 
\begin{equation}
\frac{\mathrm{d}}{\mathrm{d}t}\zeta(t,\xi(t))\le\overline{\Ent}\big((T_{t}f)^{\frac{1}{\xi(t)}}\big)\frac{1}{C\, c_{q(t)}}-\overline{\mathcal{E}}_{n}\big(T_{t}f,(T_{t}f)^{\frac{1}{\xi(t)}-1}\big).\label{eq:FI-14}
\end{equation}
On the other hand, by assumption, for any $r\in[p,\infty)$, the $r$-log-Sobolev
inequality holds, i.e., 
\begin{equation}
\overline{\Ent}(g^{r})\le  C\,  c_r \,  \overline{\mathcal{E}}_{n}(g,g^{r-1})\quad\mbox{for all} \;\, g\ge0.\label{eq:FI-11-2}
\end{equation}
Substituting $g$ and $r$ for $T_{t}f$ and $q(t)=1/\xi(t)$ respectively
into~\eqref{eq:FI-11-2} and combining the resultant inequality with~\eqref{eq:FI-14}
yields
\begin{equation}
\frac{\mathrm{d}}{\mathrm{d}t}\zeta(t,\xi(t))\le0\quad\mbox{for all} \;\, t\ge0.\label{eq:FI-11-2-1}
\end{equation}
Finally, by integrating both sides of \eqref{eq:FI-11-2-1} from $0$
to $t$, we have 
\begin{equation}
\ln\|T_{t}f\|_{\frac{1}{\xi(t)}}-\ln\|f\|_{\frac{1}{\xi(0)}}\le0,
\end{equation}
which is precisely the hypercontractivity inequality in~\eqref{eq:FIFHC-1}.

Statement~(b) follows analogously but the directions of the inequalities above are reversed. Statement~(c) follows by first differentiating
the hypercontractivity inequalities and evaluating them at $t=0$.  Then, we
can recover the $p$-log-Sobolev inequalities from them.  \end{proof}

%Similarly, for Statement~\ref{itm:two}, one can establish the $p$-log-Sobolev
%inequality from the hypercontractivity inequality for the case $0<p<1$.
%Note that for this case, the directions of the inequalities in the
%argument above are reversed.

%Finally, for Statement \ref{itm:three}, by first differentiating
%the hypercontractivity inequalities and evaluating them at $t=0$,
%one can recover the hypercontractivity inequalities from the $p$-log-Sobolev
%inequalities. We omit the details. \end{proof}

It is also well-known (see, for example, \citet{Ledoux_book}) that
the $p$-log-Sobolev inequality satisfies the {\em tensorization
property}. %, as formalized in the following proposition. 

\begin{proposition} \label{prop:tensorization} If a certain $p$-log-Sobolev
inequality with constant~$C$ holds for $(T_{t},\pi)$, then the
$p$-log-Sobolev inequality with the {\em same} constant holds for
the product semigroup $(T_{t}^{\otimes n},\pi^{n})$. \end{proposition} 

\begin{proof} We provide an information-theoretic proof for
this proposition. We start by characterizing 
the optimal constant in the $p$-log-Sobolev inequality in terms of information-theoretic quantities. For the
product semigroup $(T_{t}^{\otimes n},\pi^{n})$, the optimal constant is
for $p\in\mathbb{R}\backslash\{0,1\}$ %is defined as 
\begin{equation}
%C_{p}^{(n)*}
C_{p,n}^*:=\sup_{f  \,:\, c_{p}\, \mathcal{E}_{n}(f,f^{p-1})>0}\;\frac{\Ent(f^{p})}{c_{p}\, \mathcal{E}_{n}(f,f^{p-1})}.\label{eq:FIpLSI-2}
\end{equation}
Since $(T_{t},\pi)$ is a special case of $(T_{t}^{\otimes n},\pi^{n})$
with $n$ set to $1$,    $C_{p,1}^*$ is the optimal
constant for the semigroup $(T_{t},\pi)$. 

For a given $Q_{X^{n}}$, define the $k^{\mathrm{th}}$ ``likelihood
ratio'' 
\begin{equation}
\ell_{k}(y,x^{\setminus k}):=\frac{Q_{X_{k}|X^{\setminus k}}(y|x^{\setminus k})}{\pi(y)}\quad\mbox{for all}\;\, (y,x^{\setminus k})\in\calX\times\calX^{n-1}.\label{eqn:zeta_k}
\end{equation}
If we write ${f^{p}}/{\mathbb{E}[f^{p}]}={Q_{X^{n}}}/{\pi^{n}}$
for a distribution $Q_{X^{n}}\ll\pi^{n}$, then 
\begin{align}
\Ent(f^{p}) & =D(Q_{X^{n}}\|\pi^{n})\quad\mbox{and}\label{eq:FI-3}\\
\mathcal{E}_{n}(f,f^{p-1}) & =\sum_{k=1}^{n}\mathbb{E}_{\pi^{n-1}}\big[\eta(X^{\setminus k})\big],\label{eq:FI-7}
\end{align}
where 
\begin{align}
\hspace{-0.15in}\eta(x^{\setminus k}):=-\frac{Q_{X^{\setminus k}}(x^{\setminus k})}{\pi^{n-1}(x^{\setminus k})}\sum_{x,y}\ L_{x,y}\big(\ell_{k}(y,x^{\setminus k})\big)^{1/p}\ \big(\ell_{k}(x,x^{\setminus k})\big)^{1/p'} \ \pi(x)\label{eqn:def_g} .
\end{align}
%and $p'$ is the H\"older conjugate of $p$. 
Uniting \eqref{eq:FI-3} and \eqref{eq:FI-7}, one can obtain the following
information-theoretic characterization of~$C_{p,n}^*$.

\begin{lemma} \label{lem:optconstIT} For $n\in\bbN$
and $p\in\mathbb{R}\backslash\{0,1\}$,  it holds that 
\begin{equation}
C_{p,n}^*=\sup_{Q_{X^{n}}\,:\,c_{p} \,  \sum_{k=1}^{n}\mathbb{E}_{\pi^{n-1}}[\eta(X^{\setminus k})]>0}\;\;\frac{D(Q_{X^{n}}\|\pi^{n})}{c_{p}\, \sum_{k=1}^{n}\mathbb{E}_{\pi^{n-1}}[\eta(X^{\setminus k})]}.\label{eq:FIoptimalconstantLSI}
\end{equation}
\end{lemma}

Continuing the proof of Proposition~\ref{prop:tensorization}, we notice that, on one hand, by the data processing inequality for the relative entropy,
we have 
\begin{align}
D(Q_{X^{n}}\|\pi^{n}) & =\sum_{k=1}^{n}D\big(Q_{X_{k}|X^{k-1}}\big\|\pi \big| Q_{X^{k-1}}\big)\\
 & \le\sum_{k=1}^{n}D\big(Q_{X_{k}|X^{\setminus k}}\big\|\pi \big|   Q_{X^{\setminus k}}\big)\\
 & =nD\big(Q_{X_{K}|X^{\setminus K}K}\big\|\pi  \big|  Q_{X^{\setminus K}|K}Q_{K}\big)\\
 & =nD\big(Q_{X|U}\big\|\pi \big|  Q_{U}\big),\label{eqn:div_sl}
\end{align}
where $K\sim Q_{K}:=\mathrm{Unif}[n]$ is independent of $X^{n}$
and $U:=(X^{\setminus K},K)$. On the other hand, using the definitions
of $\eta$ and $\ell_{k}$, consider, 
\begin{align}
 & \hspace{-0.1in}\sum_{k=1}^{n}\bbE_{\pi^{n-1}}[\eta(X^{\setminus k})]\nn\\*
 & \hspace{-0.1in}=-\sum_{k=1}^{n}\bbE_{Q_{X^{\setminus k}}}\bigg[\sum_{x,y}L_{x,y}\big(\ell_{k}(y,X^{\setminus k})\big)^{1/p}\big(\ell_{k}(x,X^{\setminus k})\big)^{1/p'}\pi(x)\bigg]\\
 & \hspace{-0.1in}=-n\bbE_{Q_{K}}\Bigg[\bbE_{Q_{X^{\setminus K}}}\bigg[\sum_{x,y}\! L_{x,y}\big(\ell_{K}(y,X^{\setminus K})\big)^{1/p}\big(\ell_{K}(x,X^{\setminus K})\big)^{1/p'}\pi(x)\bigg|K\bigg]\Bigg]\!\!\label{eqn:use_K}\\
 & \hspace{-0.1in}=n\sum_{u}Q_{U}(u)\mathcal{E}\left(\Big(\frac{Q_{X|U=u}}{\pi}\Big)^{1/p},\Big(\frac{Q_{X|U=u}}{\pi}\Big)^{1/p'}\right),\label{eqn:use_calE}
\end{align}
where the penultimate equality follows from the uniformity of $K$ and the final equality
follows from the definition of the Dirichlet form in~\eqref{eqn:dirichlet_form}
and that of $U=(X^{\setminus K},K)$. From  \eqref{eqn:div_sl}
and \eqref{eqn:use_calE}, we conclude that the objective function
in \eqref{eq:FIoptimalconstantLSI} satisfies 
\begin{align}
& \frac{D(Q_{X^{n}}\|\pi^{n})}{c_{p}\, \sum_{k=1}^{n}\mathbb{E}_{\pi^{n-1}}\big[\eta(X^{\setminus k})\big]} \nn \\
& \quad\le\frac{D(Q_{X|U}\|\pi|Q_{U})}{c_{p}\, \sum_{u}Q_{U}(u)\mathcal{E}\Big(\big(\frac{Q_{X|U=u}}{\pi}\big)^{1/p},\big(\frac{Q_{X|U=u}}{\pi}\big)^{1/p'}\Big)}\\
 &\quad \le\max_{u}\frac{D(Q_{X|U=u}\|\pi)}{c_{p}\, \mathcal{E}\Big(\big(\frac{Q_{X|U=u}}{\pi}\big)^{1/p},\big(\frac{Q_{X|U=u}}{\pi}\big)^{1/p'}\Big)} \label{eq:FI_max_n}\\
 &\quad \le C_{p,1}^*, \label{eq:FI_max_n_1}
\end{align}
where in \eqref{eq:FI_max_n}, the maximum is over all $u$ in the alphabet  of $U$ (i.e., $\calX^{n-1}\times [n]$) such that the denominator is positive, and  \eqref{eq:FI_max_n_1}  follows from Lemma~\ref{lem:optconstIT} (with
$n$ set to $1$). Therefore, $C_{p,n}^*\le C_{p,1}^*$ for
any $n$ and $p\in\mathbb{R}\backslash\{0,1\}$. 

On the other hand, setting $Q_{X^{n}}$ to be a product distribution
$Q_{X}^{n}$ in Lemma \ref{lem:optconstIT}, we   find that $C_{p,n}^*\ge C_{p,1}^*$ for all $n\in\bbN$.
Combining these two bounds yields $C_{p,n}^* = C_{p,1}^*$ 
for $p\in\mathbb{R}\backslash\{0,1\}$. By taking limits in $p$ (toward $0$ and $1$), one  
deduces that $C_{p,n}^*\ge C_{p,1}^*$  for $p\in\{0,1\}$. Hence,
the tensorizaton property holds. \end{proof}

The information-theoretic method employed in the proof of Proposition~\ref{prop:tensorization}
can be also used to study certain {\em nonlinear} versions of log-Sobolev
inequalities. 
These inequalities were proposed as a topic for research by \citet{kalai1995distance}  in 1995. However,  there was no progress for over twenty years  since the initial proposal of these inequalities until 
%after more than twenty years since then, there was no progress on this topic,
% until 
 recent works by \citet{samorodnitsky2008modified}, \citet{samorodnitsky2016entropy},
and \citet{polyanskiy2019improved}. 

%These inequalities were proposed by \citet{kalai1995distance}
%and studied  by \citet{samorodnitsky2008modified}, \citet{samorodnitsky2016entropy}
%and \citet{polyanskiy2019improved}. 

Note that \eqref{eq:FIpLSI} delineates a certain {\em linear} relationship between  the  entropy $\Ent(f^{p})$ and the Dirichlet form $\mathcal{E}(f,f^{p-1})$. 
%Note that the $p$-log-Sobolev
%inequality is homogeneous in~$f$, which means that  \eqref{eq:FIpLSI} holds for
%a nonnegative function $f$ if and only if it holds for $af$ for
%any scalar $a>0$. 
For $\alpha\ge0$ and $n\in\bbN$,  let 
\begin{equation}
%\mathcal{F}_{\alpha}:=\big\{ f:c_{p}\cdot\overline{\mathcal{E}}(f,f^{p-1})=\alpha\big\}.
\mathcal{F}_{\alpha}^{(n)}:=\big\{ f:c_{p}\ \overline{\mathcal{E}}_{n}(f,f^{p-1})=n\alpha\big\}.
\end{equation}
To study the \emph{nonlinear} tradeoff between the normalized Dirichlet form and the normalized entropy, we define the {\em  log-Sobolev function} as 
\begin{equation}
\Xi_{p}(\alpha):=\sup_{f\in\mathcal{F}_{\alpha}^{(1)}}\overline{\Ent}(f^{p}).\label{eq:FIpLSI-1-1}
\end{equation}
Extending the definition of $\Xi_{p}(\alpha)$ from $T_{t}$ to $T_{t}^{\otimes n}$,
we define 
\begin{equation}
\Xi_{p}^{(n)}(\alpha):=\sup_{f\in\mathcal{F}_{\alpha}^{(n)}}\frac{1}{n}\overline{\Ent}(f^{p}).\label{eqn:def_Xin}
\end{equation}
%where 
%\begin{equation}
%\mathcal{F}_{\alpha}^{(n)}:=\big\{ f:c_{p}\cdot\overline{\mathcal{E}}_{n}(f,f^{p-1})=n\alpha\big\}.
%\end{equation}
It would be useful to provide a  tight dimension-independent
bound for $\Xi_{p}^{(n)}(\alpha)$. As mentioned in Section~\ref{sec:balanced}, by {\em dimension-independent}, we
mean that the bound on $\Xi_{p}^{(n)}(\alpha)$ does not depend on
$n$; in information theory parlance, this is known as a {\em single-letter}
bound. A  tight dimension-independent
bound for $\Xi_{p}^{(n)}(\alpha)$ was shown by \citet{polyanskiy2019improved}
in the following theorem.

\begin{theorem} \label{thm:strongSLI} It holds that for $p\in\mathbb{R}\backslash\{0,1\}$,
\begin{equation}
\Xi_{p}^{(n)}(\alpha)\leq\mathbb{U}[\Xi_{p}](\alpha).\label{eq:FI-strongerlogSobolev}
\end{equation}
Moreover, this upper bound is asymptotically tight as $n\to\infty$,
which means that 
\begin{equation}
\lim_{n\to\infty}\Xi_{p}^{(n)}(\alpha)=\mathbb{U}[\Xi_{p}](\alpha).\label{eqn:asympt_tight_U}
\end{equation}
If additionally, $\Xi_{p}$ in \eqref{eq:FIpLSI-1-1} is concave, then the upper bound in~\eqref{eq:FI-strongerlogSobolev} is also
tight for all finite $n\ge1$. \end{theorem}

This theorem is a strengthening of the (linear) $p$-log-Sobolev inequality
in \eqref{eq:FIpLSI} with optimal constant $C_{p}^{*}:=C_{p,1}^*$ given
in \eqref{eq:FIoptimalconstantLSI}.  Here we set  $n$ in \eqref{eq:FIoptimalconstantLSI}   to $1$ since the tensorization property holds.  Since the function
$\mathbb{U}[\Xi_{p}]$ is nonlinear in general, the inequality in \eqref{eq:FI-strongerlogSobolev}
is known as the \emph{nonlinear $p$-log-Sobolev inequality}. 

To appreciate
the relation between the linear and nonlinear $p$-log-Sobolev inequalities,  
one can demonstrate  that the optimal constant $C_{p}^{*}$ in the
(linear) $p$-log-Sobolev inequality in \eqref{eq:FIpLSI} is  
 the right-derivative of $\mathbb{U}[\Xi_{p}](\alpha)$ at $\alpha=0$ if $\Xi_{p}(0)=0$.
If $\Xi_{p}(0)>0$, then  the  linear  $p$-log-Sobolev inequality
in \eqref{eq:FIpLSI} does not hold for any finite $C$. 

\begin{proof}[Proof of Theorem \ref{thm:strongSLI}]  We follow
the same steps as in  the proof of Lemma~\ref{prop:tensorization} up to~\eqref{eqn:use_calE}. Then, combining these steps with the definition of $\Xi_{p}^{(n)}(\alpha)$ in~\eqref{eqn:def_Xin},
%\eqref{eq:FI-3}, \eqref{eq:FI-7}, \eqref{eqn:div_sl}, and \eqref{eqn:use_calE},
we find that 
\begin{align}
\Xi_{p}^{(n)}(\alpha) & \le \sup_{Q_{XU}\,:\, c_{p}\, \bbE_{Q_{U}}\big[\mathcal{E}\big((\frac{Q_{X|U}}{\pi})^{1/p},(\frac{Q_{X|U}}{\pi})^{1/p'}\big)\big]=\alpha} D(Q_{X|U}\|\pi_{X}|Q_{U})\\
 & =\mathbb{U}[\Xi_{p}](\alpha).
\end{align}
The asymptotic tightness of~\eqref{eqn:asympt_tight_U} can be verified
by a   time-sharing argument (cf.\ the discussion after Theorem~\ref{thm:strongsse}). \end{proof}

\section{Strengthened Hypercontractivity Inequalities} \label{sec:stronger}

The tools we reviewed in the preceding sections serve as ingredients for the culmination of this section---namely, a strengthened version of the hypercontractivity inequality. For the sake of clarity, we focus on the DSBS. Before doing so, we provide explicit expressions for the linear and nonlinear $p$-log-Sobolev
inequalities particularized to the DSBS.

%In this section, we use the tools developed in the preceding sections to  introduce a  stronger version  of  
%hypercontractivity inequality for the DSBS, which is one of the
%main purposes of this chapter. Before doing this, we first need to
%provide explicit expressions for the linear and nonlinear $p$-log-Sobolev
%inequalities. Note that in this section, we mainly consider the DSBS. 

For the DSBS,  $\mathcal{X}=\{0,1\}$, $\pi=\pi_{Y}=\Bern(1/2)$
and $T_{t}$ is the Markov operator induced by  $\pi_{X|Y}$ and is   given by 
\begin{equation}
T_{t}f(y)=f(y)\frac{1+\rme^{-t}}{2}+f(1-y)\frac{1-\rme^{-t}}{2} \qquad y\in\{0,1\}.\label{eq:FIhypercube}
\end{equation}
Note that   the operator $T_{t}^{\otimes n}$ is the same  as  $T_{\rho}$
  in Section \ref{ch:NICD} with $\rho=\rme^{-t}$.   From \eqref{eq:FIhypercube},
we know that $L_{x,y}=\bone\{x\neq y\}-1/2$ which is obtained by differentiating
$T_{t}$ with respect to $t$ and evaluating the derivative at $t=0$. Moreover, the   Dirichlet form
for this case is
\begin{align}
\hspace{-.25in}{\mathcal{E}}_{n}(f,g) & =-\frac{1}{2}\langle\Delta f,g\rangle\quad\mbox{and}\label{eq:FIdiri_cube}\\
\hspace{-.25in}{\mathcal{E}}_{n}(f,f) & =\frac{2^{-n}}{4}\sum_{(x^{n},y^{n}):x^{n}\sim y^{n}}\big(f(x^{n})-f(y^{n})\big)^{2}=\frac{1}{4}\, \mathbf{I}[f]\,,\label{eq:Eff}
\end{align}
where $\Delta f(x^{n}):=\sum_{y^{n}:y^{n}\sim x^{n}}(f(y^{n})-f(x^{n}))$,
and $x^{n}\sim y^{n}$ means that $x^{n},y^{n}\in\{0,1\}^{n}$ differ
in exactly one coordinate. Recall that here $\mathbf{I}[f]$ denotes
the total influence of $f$; see Definition \ref{def:total_influence}.

% where 
%for arbitrary real-valued function $f$, $\mathbf{I}_{i}[f]$ is defined via 
%\eqref{eq:Infi-Di}. 

For the DSBS, the optimal constant $C$ in the (linear) $p$-log-Sobolev inequality is $2$; see \citet{gross1975logarithmic}.
The (linear) $p$-log-Sobolev inequality with optimal constant
$2$ can be derived   from the single-function version  of  the hypercontractivity
inequalities for the DSBS given in Theorem~\ref{thm:hyper_single}. This can be done  by differentiating
both sides of the hypercontractivity inequalities with respect to
$\rho$ and evaluating the derivative at  $\rho=1$. 

\subsection{Strong Log-Sobolev Inequalities}

\citet{polyanskiy2019improved} proved the 
dimension-independent \emph{nonlinear} $p$-log-Sobolev inequalities as stated in the next theorem. Before introducing these inequalities, we define  
$b_{p}:[0,\ln2]\to[0,\infty)$ to be  the convex increasing function given
by 
\begin{align}
b_{p}(t) & := \left\{ \!\begin{array}{cc}
\displaystyle\frac{\sign(p\! -\! 1)}{2}\left(1\!-\!y^{\frac{1}{p}}(1-y)^{1-\frac{1}{p}}\!-\! y^{1-\frac{1}{p}}(1-y)^{\frac{1}{p}}\right)  & p \ne 0,1 \vspace{.03in}\\
\displaystyle\Big(\frac{1}{2}-y\Big)\ln\frac{1-y}{y} & p=1
\end{array} \right. , %\label{eqn:def_bp}
\end{align}
where $y (t):=h^{-1}(\ln2-t)$ and $h^{-1}:[0,\ln2]\to[0,1/2]$ is the inverse of the binary entropy
function~$h$ {\em with base $\rme$} when its domain is restricted to $[0,1/2]$.

%When specializing the nonlinear $p$-log-Sobolev inequality in \eqref{eq:FI-strongerlogSobolev}
%to the DSBS case, \citet{polyanskiy2019improved} provided the following
%dimension-free \emph{nonlinear} $p$-log-Sobolev inequalities. The
%one  for   $p\in(-\infty,\infty)\setminus\{0,1\}$ is 
%given in Theorem \eqref{thm:lsi_cube-1}, and the one for  $p=1$ is given in Theorem \eqref{thm:lsi_cube}. 

\begin{theorem}[$p$-log-Sobolev Inequality for the DSBS]\label{thm:lsi_cube}
Let $p\in\bbR\setminus\{0,1\}$. For all $f:\{0,1\}^{n}\to[0,\infty)$
(and $f>0$ if $p<1$), 
it holds that 
\begin{equation}
\frac{1}{n}\sign(p-1)\ \overline{\mathcal{E}}_{n}(f,f^{p-1})\ge b_{p}\Big(\frac{1}{n}\overline{\Ent}(f^{p})\Big)\,,\label{eq:FIblsi_p}
\end{equation}
where the normalized Dirichlet form is given by~\eqref{eq:FIdiri_cube}.
%, and
%$b_{p}:[0,\ln2]\to[0,\infty)$ is a convex increasing function given
%by 
%\begin{align}
%b_{p}(t) & :=\frac{\sign(p-1)}{2}\left(1-y^{\frac{1}{p}}(1-y)^{1-\frac{1}{p}}-y^{1-\frac{1}{p}}(1-y)^{\frac{1}{p}}\right) \quad\mbox{where}\\
%y &:=h^{-1}(\ln2-t),
%\end{align}
%%with $y=h^{-1}(\ln2-t)$, 
%and 
% $h^{-1}:[0,\ln2]\to[0,1/2]$ is the inverse of the binary entropy
%function~$h$ with base $\rme$ when its domain is restricted to $[0,1/2]$.
%%$h^{-1}$ denoting the inverse of
%the restriction of the binary entropy function $h$ given in \eqref{eq:xxx} to the set $[0,1/2]$.
%\end{theorem}
%
%\begin{theorem}[$1$-LSI for the DSBS] \label{thm:lsi_cube} For
Let $p=1$. For all $f:\{0,1\}^{n}\to(0,\infty)$, it holds that 
\begin{equation}
\frac{1}{n}\overline{\mathcal{E}}_{n}(f,\ln f)\ge b_{1}\Big(\frac{1}{n}\overline{\Ent}(f)\Big)\,,\label{eq:FIblsi_1}
\end{equation}
\end{theorem}
%where the Dirichlet form is given by~\eqref{eq:FIdiri_cube}, and $b_{1}:[0,\ln2)\to[0,\infty)$
%is a convex increasing function given by 
%\begin{equation}
%b_{1}(t)=\,,\qquad y\in(0,1/2]\,,\label{eq:FIblsi_2}
%\end{equation}
%with $y=h^{-1}(\ln2-t)$. \end{theorem}

The inequality in \eqref{eq:FIblsi_1} is the limiting case of \eqref{eq:FIblsi_p}
as $p\to1$. 
%In other words, the inequality in \eqref{eq:FIblsi_1}
%can be proven from \eqref{eq:FIblsi_p} by dividing both sides of
%\eqref{eq:FIblsi_p} by $p-1$, and then taking limits as $p\to1$.
Moreover, these inequalities  
are sharp in the sense that given $p$, there exists a nonnegative
function $f$ such that the equality holds. 

% in \eqref{eq:FIblsi_p} or
%\eqref{eq:FIblsi_1}. 

\begin{proof}[Proof Sketch of Theorem \ref{thm:lsi_cube}] 
% Since Theorem
%\ref{thm:lsi_cube} is a consequence of Theorem \ref{thm:lsi_cube-1},
%we only need to prove the latter one. Indeed, 
Theorem \ref{thm:lsi_cube}
follows directly from Theorem \ref{thm:strongSLI} by observing that
$b_{p}$ is   the inverse   of $\Xi_{p}$ when $\pi_{XY}$ is particularized to the DSBS. Moreover, $\Xi_{p}$ is concave and increasing. The monotonicity
and concavity of $\Xi_{p}$ (or equivalently, the monotonicity and
convexity of $b_{p}$) can be shown by calculating the first and second
derivatives of $\Xi_{p}$ (or~$b_{p}$); see \citet{polyanskiy2019improved}
for more details. Furthermore, the asymptotic sharpness of \eqref{eq:FIblsi_p}
follows directly from the asymptotic sharpness of \eqref{eq:FI-strongerlogSobolev}.
\end{proof}

Based on Theorem  \ref{thm:lsi_cube}, we
are almost ready to introduce a strengthened version of the forward hypercontractivity
inequality  shown by \citet{polyanskiy2019improved}.  
Before doing so, we would like to discuss an intimate relationship
between the linear and nonlinear log-Sobolev inequalities and the
edge-isoperimetric inequality given in Theorem~\ref{thm:EdgeIsoper}, as promised
below  \eqref{eq:isoperi}. 
%in Section \ref{sec:influ}. 
We first consider the linear log-Sobolev
inequality. Consider the DSBS and the case $p=2$. The linear $2$-log-Sobolev
inequality with optimal constant $C=2$ reduces to
\begin{equation}
\Ent(f^{2})\le 2\,\mathcal{E}(f,f) \quad\mbox{for all}\;\, f\ge0.\label{eq:FIpLSI-1}
\end{equation}
%for any nonnegative $f$. 
By the tensorization property and utilizing
\eqref{eq:Eff}, 
\begin{equation}
\Ent(f^{2})\le 2\,\mathcal{E}_{n}(f,f)=\frac{1}{2}\ \mathbf{I}[f].\label{eq:FIpLSI-1-2}
\end{equation}
Setting $f$ to be a Boolean function with mean $a$, we obtain 
\begin{equation}
\mathbf{I}[f]\ge 2a\, \ln \Big(\frac{1}{a}\Big).\label{eq:FIedgeisoper}
\end{equation}
Note that in the sharp edge-isoperimetric inequality in~\eqref{eq:isoperi},
the logarithm used is $\log$ (to the base $2$), instead of $\ln$.
Hence, in \eqref{eq:FIedgeisoper}, there is a multiplicative factor $\log \rme$ off from the
sharp inequality in~\eqref{eq:isoperi}.

\begin{figure}
\centering \includegraphics[width=0.8\columnwidth]{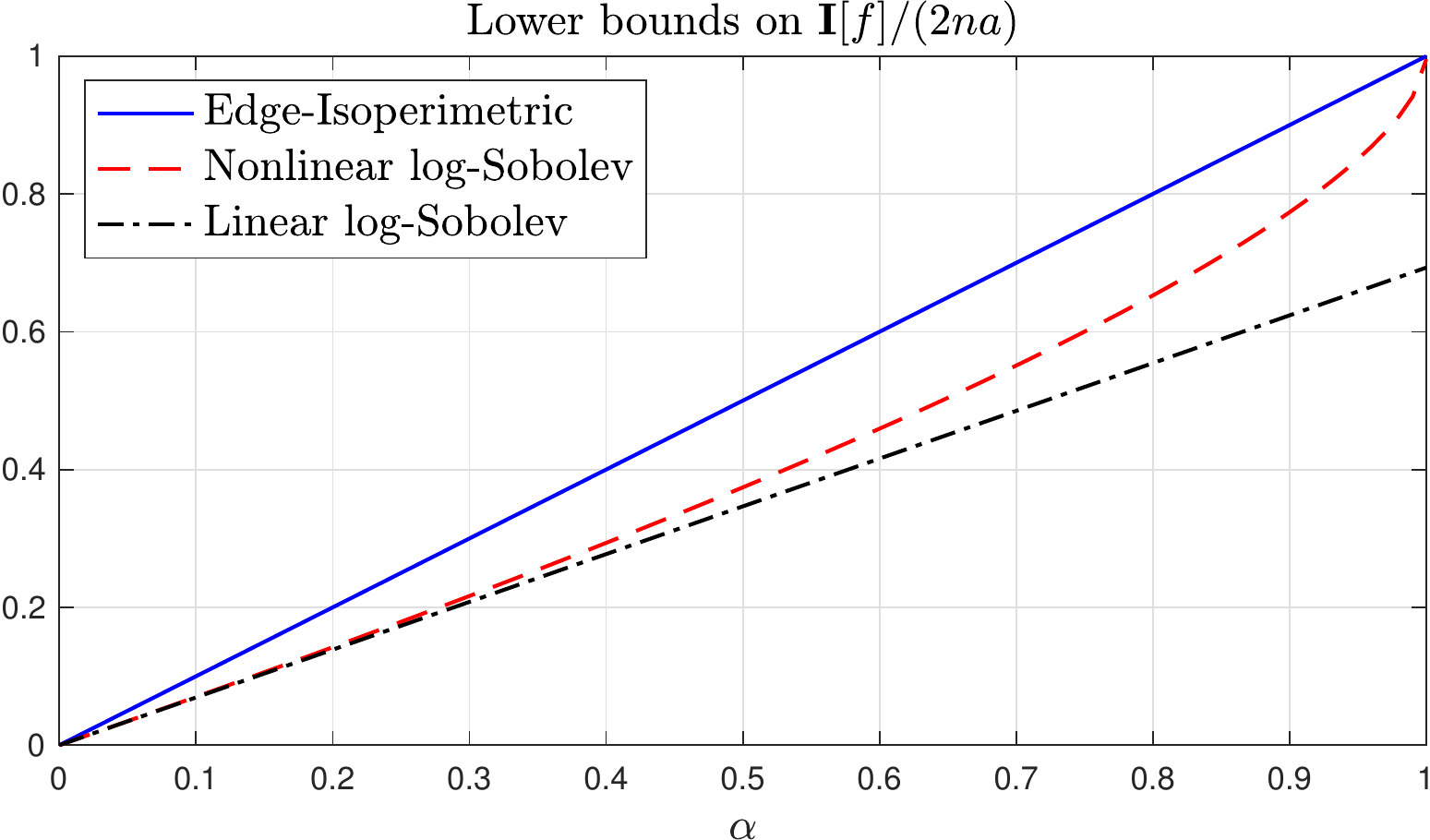} 
\caption{Comparison of the distinct parts  $  \alpha$ (edge-isoperimetric),  $\alpha\ln 2$ (linear log-Sobolev) and  $2b_2(\alpha)$ (nonlinear log-Sobolev)  when $a$ is set to  $2^{-\alpha}$.}
%Theorem~\ref{thm:ops_resolve}
\label{fig:edge_isoper_ineq} 
\end{figure}

We next consider the nonlinear log-Sobolev inequality for the DSBS and $p=2$.
%the DSBS and the case $p=2$. 
For this case, Theorem \ref{thm:lsi_cube}
reduces to the statement that  for all $f\ge0$, it holds that
\begin{equation}
\frac{1}{n}\overline{\mathcal{E}}_{n}(f,f)\ge b_{2}\Big(\frac{1}{n}\overline{\Ent}(f^{2})\Big)\,,\label{eq:FIblsi_p-1}
\end{equation}
where $b_{2}:[0,\ln2]\to[0,\infty)$ is the convex increasing function
given by 
\begin{align}
b_{2}(t) & =\frac{1-2\sqrt{y(1-y)}}{2}\,
\end{align}
with $y=h^{-1}(\ln2-t)$ (coinciding with the general definition of $b_p$). By~\eqref{eq:Eff} and setting $f$ to be
a Boolean function with mean $a$, we obtain 
\begin{equation}
\mathbf{I}[f]\ge4an\,b_{2}\Big(\frac{1}{n}\ln\Big( \frac{1}{a}\Big)\Big).\label{eq:FIedgeisoper2}
\end{equation}
This inequality is tighter than  \eqref{eq:FIedgeisoper},
but looser than \eqref{eq:isoperi}. This point can be observed
from the facts that $b_{2}$ is convex and increasing, $b_{2}'(0)=1/2$,
$b_{2}(0)=0, b_{2}(\ln2)=1/2$, and hence, $t/2\le b_{2}(t)\le\frac{t}{2\ln2}$
for all $t\in[0,\ln2]$. If we consider the case $a=2^{-\alpha}$, then the bounds in  \eqref{eq:isoperi},  \eqref{eq:FIedgeisoper}, and \eqref{eq:FIedgeisoper2} are  $2na \, \alpha$, $4na \, b_2 (\alpha)$, and $2na \, \alpha\ln2$ respectively. Omitting the common factor $2na$, we plot  $\alpha$, $2b_2(\alpha)$, and $\alpha\ln2$ in Fig.~\ref{fig:edge_isoper_ineq}, which depicts the relations among these three inequalities. %, showing that the edge-isoperimetric inequality yields the best lower bound. 

We   note that it makes eminent  sense that~\eqref{eq:FIedgeisoper} and~\eqref{eq:FIedgeisoper2}
are looser than~\eqref{eq:isoperi}. This is because
that  the former two inequalities are derived from the linear and nonlinear
log-Sobolev inequalities in~\eqref{eq:FIpLSI-1} and~\eqref{eq:FIblsi_p-1} which are valid not only for Boolean functions,
but for {\em any nonnegative functions}. 
%Moreover, the log-Sobolev
%inequalities in \eqref{eq:FIedgeisoper} and \eqref{eq:FIedgeisoper2} are sharp for nonnegative functions, i.e., there is a nonnegative function that makes the inequalities tight.
 In contrast, the
edge-isoperimetric inequality in~\eqref{eq:isoperi}, which is derived
by a combinatorial method, is specific to  and sharp for Boolean functions.  

%is only valid   and sharp for Boolean
%functions. 

\subsection{Strengthened Version of Hypercontractivity Inequalities}

We now introduce a strengthened version of forward hypercontractivity
inequality due to \citet{polyanskiy2019improved}.   We first introduce an additional definition.
For a nonnegative function $f:\calX^n\to [0,\infty)$, the \emph{$p$-entropy} of $f$ \cite{yu2021strong} is defined as 
\begin{equation}
\overline{\Ent}_{p}(f) := \frac{p}{p-1} \log \frac{\|f\|_p}{\|f\|_1}. \label{eq:FIp_ent}
\end{equation}
In fact, if $f$ is the Radon--Nikodym derivative of $Q$ with respect to~$\pi$, i.e.,  $f=\rmd Q/\rmd \pi$, % (or $f(x) = Q(x)/\pi(x)$ in the finite alphabet case), 
then $\overline{\Ent}_{p}(f)=D_p(Q\|\pi)$; also see~\eqref{eqn:strong_renyi_dpi}. %, where $D_p(Q\|\pi)$ denotes the R\'enyi divergence of order $p$. 
Basic properties of the $p$-entropy, such as its continuity and monotonicity, can be found in \citet{yu2021strong}. 
Let $g:[0,\ln2]\to[2,2/\ln2]$ be defined as 
\begin{equation}
g(t):=\frac{2-4\sqrt{y(1-y)}}{\ln2-h(y)},
\end{equation}
where $y=h^{-1}(\ln2-t)$.

\begin{theorem}\label{th:hcb} Fix two numbers $1<p<\infty$ and $0\le\alpha\le\ln2$.
Then the differential equation  in $u:[0,\infty)\to \bbR$  
\begin{equation}
\frac{\mathrm{d}}{\mathrm{d}t}u(t)=g\Big(\frac{\alpha}{p'}(1+\rme^{-u(t)})\Big)\label{eq:FIhcb_0}
\end{equation}
with initial solution $u(0)=\ln(p-1)$  
%with $u:[0,\infty)\to \bbR$ unknown and the function $g$ defined by $g(t)=$
%for $t\in[0,\ln2]$ with $y=h^{-1}(\ln2-t)$, 
has a unique solution
on $[0,\infty)$. Furthermore, for any $f:\{0,1\}^{n}\to[0,\infty)$
with $\frac{1}{n}\overline{\Ent}_{p}(f)\ge\alpha$, we have 
\begin{equation}
\|T_{t}^{\otimes n}f\|_{q(t)}\le\|f\|_{p}\quad\mbox{where}\quad q(t)=1+\rme^{u(t)}\,.\label{eq:FIhcb_1}
\end{equation}
\end{theorem}

The core idea of the proof of Theorem~\ref{th:hcb} is to integrate
both sides of the nonlinear $p$-log-Sobolev inequality in Theorem
\ref{thm:lsi_cube}. It is similar to the proof of Theorem
\ref{thm:LSI-HC}, and hence, omitted. Theorem~\ref{th:hcb} was used by \citet{ordentlich2020note} to prove the limiting cases as $\rho\downarrow0$
and $\rho\uparrow1$ of the NICD problem in the LD regime.

As remarked by \citet{polyanskiy2019improved}, the function $g$
is a smooth, convex, and strictly increasing bijection. Consequently,
the function $q(t)$ in \eqref{eq:FIhcb_1} is smooth and satisfies
$q(t)>1+(p-1)\rme^{2t}$ for all $t>0$. Note that the maximum (and hence best possible) parameter $q(t)$ for the classic forward hypercontractivity is equal to $1+(p-1)\rme^{2t}$; see   \eqref{eq:NICDFHC-1}.
Hence, the inequality in
\eqref{eq:FIhcb_1} strictly improves the classic forward hypercontractivity
inequality in \eqref{eq:NICDFHC-1}. Furthermore, $q(t)$ also satisfies
\begin{align}
q(t) & =p+q'(0)t+\frac{1}{2}q''(0)t^{2}+o(t^{2}),\qquad \mbox{as } t\to0\,,\label{eq:hcb_3}
\end{align} where
\begin{align}
q'(0) & =(p-1)g(\alpha)\label{eq:hcb_3a}, \quad\mbox{and}\\
q''(0) & =(p-1)\Big(g(\alpha)^{2}-g'(\alpha)g(\alpha)\frac{\alpha}{p}\Big).\label{eq:hcb_3c}
\end{align}

Since the strengthened version of the hypercontractivity inequality in \eqref{eq:FIhcb_1}
is obtained by integrating both sides of the \emph{sharp} nonlinear
$p$-log-Sobolev inequality in Theorem \ref{thm:lsi_cube}, one
can   observe that  \eqref{eq:FIhcb_1} is locally
sharp at $t=0$ in the following sense. For every $\hat{q}(t)$ such that
$\hat{q}(0)=p$ and $\hat{q}'(0)>q'(0)$ there exists a function $f$ with
$\frac{1}{n}\overline{\Ent}_{p}(f)\ge\alpha$ such that $\|T_{t}^{\otimes n}f\|_{\hat{q}(t)}>\|f\|_{p}$
holds for any sufficiently small $t$.  However, 
the inequality in \eqref{eq:FIhcb_1} does not appear to be  \emph{globally}
asymptotically sharp in the sense that there exists a function $\hat{q} : [0,\infty)\to\bbR$
such that $\hat{q}(t)>q(t)$ for all $t>0$, and $\|T_{t}^{\otimes n}f\|_{\hat{q}(t)}\le\|f\|_{p}$
  holds for all $f:\{0,1\}^{n}\to[0,\infty)$ with $\frac{1}{n}\overline{\Ent}_{p}(f)\ge\alpha$. Recently, a globally sharp inequality was derived by the first author of this monograph~\cite{yu2021strong}, who showed that given $q$, the minimum $p$  such that the inequality $\|T_{t}^{\otimes n}f\|_{q}\le\|f\|_{p}$
holds for any $f:\{0,1\}^{n}\to[0,\infty)$ with $\frac{1}{n}\overline{\Ent}_{p}(f)\ge\alpha$ is $\alpha/\varphi_{q}(\alpha)$ (where $\varphi_q$ is defined in Definition~\ref{def:varphi_q}), in which $\rho=\rme^{-t}$. Moreover, this sharp bound is asymptotically attained by indicators of Hamming spheres. 
We compare the sharp bound given by   \citet{yu2021strong}  and the bound in Theorem~\ref{th:hcb} in Fig.~\ref{fig:hypercont_pvsq}.
%, in which we use $q$ and $p$ 
%to denote $q(t)$  and $p_0$ respectively in Theorem \ref{thm:lsi_cube}. 
This figure  indicates that~\eqref{eq:FIhcb_1} is close to optimal as $p \downarrow 1$. 
%numerical results
%showed that the inequality in \eqref{eq:FIhcb_1} is not \emph{globally}
%asymptotically sharp in the sense that there exists a function $q : [0,\infty)\to\bbR$
%such that $q(t)>q(t)$ for all $t>0$, and $\|T_{t}^{\otimes n}f\|_{q(t)}\le\|f\|_{p}$
%  holds for all $f:\{0,1\}^{n}\to[0,\infty)$ with $\frac{1}{n}\overline{\Ent}_{p}(f)\ge\alpha$.
% \blue{Lei: I will add a figure later to show this point.} 

%Motivated by these considerations,

\begin{figure}[!ht]
\centering \includegraphics[width=0.8\columnwidth]{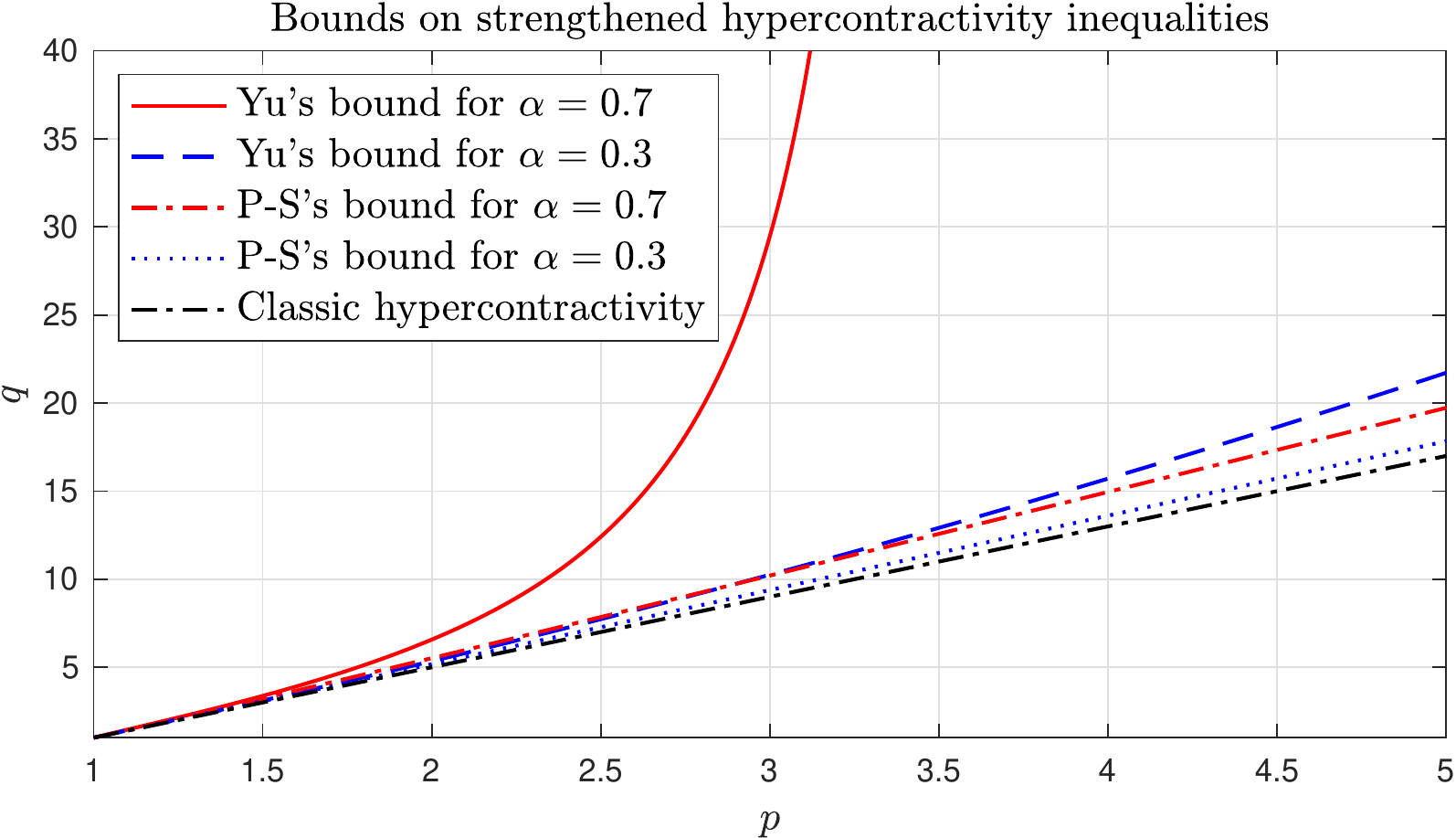} 
\caption{Comparison of the sharp bound derived by \citet{yu2021strong}  and the bound in Theorem~\ref{th:hcb} derived by \citet{polyanskiy2019improved} (P-S), where $\rho=0.5$.  Note that we use  $q$ and $p$ 
to denote $q(t)$  and $p$ respectively in Theorem \ref{thm:lsi_cube}. }
\label{fig:hypercont_pvsq} 
\end{figure}

Along the same lines, it is natural to investigate  sharper versions of BL inequalities. Indeed, Polyanskiy posed a conjecture concerning the asymptotically sharp  BL  
inequalities in 2016. This conjecture  is  stated in \citet{kirshner2021moment} and reproduced here.

\begin{conjecture} \label{conj:pol} Fix $\rho\in(0,1)$, $q>1$, and a scalar
$\alpha\in (0,\ln 2)$. Then, there exists 
\begin{equation}
p_0< 1+\rho^{2}(q-1)
\end{equation}
 such that
for any $p\ge p_0$ the maximum of $\frac{1}{n}\ln\frac{\|T_{\rho}f\|_{q}}{\|f\|_{p}}$
over all nonnegative functions $f$ with $\frac{1}{n}\overline{\Ent}_{p}(f)\ge\alpha$
is asymptotically attained by a sequence of functions that are indicators of Hamming
spheres with radii converging to some constant as $n\to\infty$. \end{conjecture}

This conjecture was confirmed in the affirmative   by 
\citet{kirshner2021moment} for the case $q=2$. As stated in~\cite{kirshner2021moment},
Conjecture~\ref{conj:pol} for all $q>1$ was proved by Polyanskiy in an
unpublished work~\cite{polyanskiy2019hypercontractivity2}. It was
also proven independently by the first author of this monograph in
\cite{yu2021strong,yu2021convexity}. In particular, it was shown
  that Conjecture~\ref{conj:pol}
holds even for $p=1$. Indeed, the asymptotically 
sharp bound on $q(t)$ such that $\|T_{t}^{\otimes n}f\|_{q(t)}\le\|f\|_{p}$
holds for any $f:\{0,1\}^{n}\to[0,\infty)$ with $\frac{1}{n}\overline{\Ent}_{p}(f)\ge\alpha$
can be characterized by the strong BL inequality  derived in the same references. 
%In fact, given $0\le\alpha\le\ln2$
%and the starting value $p>1$, the asymptotically sharp bound
%on $q(t)$ at time $t$ should satisfy that there exists $0\le\beta\le\ln2$
%such that $(1/p,1/q(t)')$ is the gradient of $\underline{\Upsilon}$
%at the point $(\alpha,\beta)$ for the DSBS induced
%by the pair $(T_{t},\pi)$, where $q(t)'$ denotes the H\"older conjugate
%of $q(t)$.  We omit details here. 
Readers may refer to \cite{yu2021strong,yu2021convexity} for the details. 

\chapter{Open Problems} \label{ch:summary}

We have taken a whirlwind tour of classic and contemporary notions related to the common information between two random variables. In this final section, we list some open problems that represent
fertile grounds for future research.

% We first introduce open problems
%on the extensions of Wyner's common information, and then introduce
%those on the extensions of G\'acs--K\"orner--Witsenhausen's common information.

\section{Open Problems Related to  Wyner's Common Information}

We now introduce two open problems related to extensions of Wyner's common information.  

\subsection{R\'enyi Common Information for all Orders}
As shown in Part \ref{part:two}, the (unnormalized and normalized)
R\'enyi common information forms a bridge between Wyner's common information
and the exact common information (see Fig.~\ref{fig:bridge}). The latter two quantities correspond to the R\'enyi common information of order $1$ (normalized)
and order $\infty$ (unnormalized) respectively. Hence, the R\'enyi common information of order
$\alpha \in [0,\infty]$ is a natural generalization of these quantities. However, the complete
characterization of R\'enyi common information of order $\alpha$ 
remains open for a large range of $\alpha$ and sources.  Here by ``complete characterization'', we   refer to 
  providing ``single-letter expressions''. In Theorem~\ref{thm:RCI},
we present upper and lower bounds on the R\'enyi common information for $\alpha\in [0,2]\cup\{\infty\}$.
For $\alpha\in(0,1]$, the unnormalized and normalized R\'enyi common
information of order $\alpha$ are both shown to be equal to Wyner's
common information, and hence, it has been completely characterized. However, for $\alpha\in(1,\infty]$,   the upper and
lower bounds given in Theorem~\ref{thm:RCI} only coincide for some
special cases, e.g., the case for the DSBS and $\alpha=\infty$, and
the case of sources  with Wyner-product distributions (cf.\ Definition~\ref{def:pseudo_pdt}). The complete
characterization of R\'enyi common information for all discrete and continuous
sources and for all orders $\alpha\in(1,\infty]$ is a major open problem on this topic. An interesting special case of this open problem is
Conjecture~\ref{conj:gauss}, which concerns the determination of
the exact common information (or the unnormalized R\'enyi common information
of order~$\infty$) for jointly Gaussian sources.

\subsection{Exact R\'enyi Common Information for all Orders}
 Another interesting observation from Part~\ref{part:two} is that the
\emph{exact}  R\'enyi common information of order $\alpha$ (originally defined in~\eqref{eqn:TEx_comm_ent_alpha})
\begin{equation}
T_\Ex^{(\alpha)}(\pi_{XY}) := \lim_{n\to\infty} \frac{ G_\alpha(\pi_{XY}^n)}{n} . \label{eqn:TEx_comm_ent_alpha_sum}
\end{equation}
 connects   the
exact common information and the nonnegative rank of a matrix; see Corollary \ref{cor:ECI_nnr} and Fig.~\ref{fig:bridge_exact}. Specifically, the exact common information
corresponds to the exact R\'enyi common information of order $1$. Given
a bivariate source $\pi_{XY}$, if we write its distribution as a matrix $\bM$, then
the asymptotic exponent of the nonnegative rank of $\bM^{\otimes n}$ is the exact R\'enyi common information of
order~$0$. Hence, the exact R\'enyi common information of order $\alpha$ in~\eqref{eqn:TEx_comm_ent_alpha_sum} simultaneously 
generalizes both the concepts of the exact common information and the  nonnegative
rank. This inspires us to define the {\em nonnegative $\alpha$-rank} 
\begin{equation}
\rank_+^{(\alpha)} (\bM) := 2^{G_\alpha(\pi_{XY}) } 
\end{equation}
in \eqref{eqn:rank_alpha0} in  Section~\ref{sec:nnar}. This notion  extends the concept
of common information beyond the realm  of  information theory. The complete characterization
of the  exact R\'enyi common information of order $\alpha$ remains 
open. In Corollary~\ref{cor:ECI_nnr}, we  provide a single-letter
expression for the exact R\'enyi common information only for the order $\infty$.
Since the exact common information for the DSBS has been completely
characterized,  the exact R\'enyi common information
of order $1$ for the DSBS is completely characterized as well. The complete
characterization of the exact R\'enyi common information of orders $\alpha\in [0,\infty) \setminus\{1\}$ for the DSBS  and
for $\alpha\in [0,\infty)$ for other sources  remains open. 

%\end{itemize}
%
\section[Open Problems Related to GKW's Common Information]{Open Problems Related to G\'acs--K\"orner--Witsenhausen's \\ Common Information} \sectionmark{Open Problems Related to GWK's Common Information}
%\sectionmark{Open Problems Related to   GKW's Common Information}

In this section, we introduce several interesting open problems on the extensions
of GKW's common information. These extensions mainly concern the   $q$-stability as discussed in Section~\ref{ch:Stability}.  
Recall from  Section~\ref{sec:form_q_stability} that the {\em noise stability} for a Boolean  function
$f:\mathcal{X}^{n}\to\{0,1\}$ with respect to $\rho$ is
 \begin{equation}\mathbf{S}_{\rho}[f]=\mathbb{E}[f(X^{n})f(Y^{n})],\end{equation} where 
$(X^{n},Y^{n})$  is  a source sequence generated by 
%a pair of  random vectors drawn from
a DSBS  with correlation coefficient $\rho\in[0,1]$. This concept can be   extended to  real-valued  functions $f:\mathcal{X}^{n}\to\bbR$. 
%Let $(X^{n},Y^{n})$ be a pair of  random vectors drawn from
%the DSBS  with correlation coefficient $\rho\in[-1,1]$.  For a function
%$f:\mathcal{X}^{n}\to\bbR$, the \emph{noise stability of $f$} is
%defined as $\mathbf{S}_{\rho}[f]=\mathbb{E}[f(X^{n})f(Y^{n})]$. 
%By definition, for a Boolean function $f:\{0,1\}^n\to \{0,1\}$, the noise stability $\mathbf{S}_{\rho}[f]=\Pr(f(X^n)=f(Y^n)=1)$. 
For the Gaussian source with correlation coefficient $\rho\in[0,1]$,
the noise stability of $f$ can be defined similarly. When there is
no ambiguity, for Gaussian sources, we also denote the noise stability of a real-valued function $f$ as $\mathbf{S}_{\rho}[f]$. For both
the DSBS and the Gaussian source, the noise stability and the $q$-stability
satisfy the relation
\begin{equation}
\mathbf{S}_{\rho^{2}}[f]=\mathbf{S}_{\rho}^{(2)}[f],
\end{equation}
for any $f$ and $\rho\in(0,1)$. The same equation for the DSBS and Boolean functions is given in~\eqref{eq:stability_stability}. 

We classify   open problems related to 
GKW's common information into three sets according to the underlying sources, namely, the DSBS, the   Gaussian source, and the so-called ball- and sphere-noise source. %We first introduce open problems for the DSBS. 

\subsection{The Doubly Symmetric Binary Source}
We introduce four open problems concerning the DSBS.
%Interesting open problems on the extensions of G\'acs--K\"orner--Witsenhausen's common information for the DSBS include:

\subsubsection{Determination of $q_{\min}, q_{\max}$, $\breve{q}_{\min}$ and $\breve{q}_{\max}$}
 One of main open problems on the $q$-stability is the determination of
 the thresholds $q_{\min}$ and $q_{\max}$ for the asymmetric max $q$-stability and  $\breve{q}_{\min}$ and $\breve{q}_{\max}$ for the symmetric max $q$-stability given in Lemma~\ref{lem:BarnesOzgur} due to \citet{barnes2020courtade}. 
Weaker versions of this open problem are stated in Conjectures~\ref{conj:AsymmetricStability} and~\ref{conj:SymmetricStability},
namely, the symmetric and asymmetric versions of the Mossel--O'Donnell,
the Courtade--Kumar, and the  Li--M\'edard conjectures. Although these conjectures
for certain ranges of $(q,\rho)$ have been resolved as discussed
in Section~\ref{sec:balanced}, other cases   remain wide open. These
conjectures are significant since they connect several different  fields including
discrete Fourier analysis, information theory,  discrete probability, etc. 
Among these conjectures,    the Courtade--Kumar conjecture is regarded as one of
the most fundamental conjectures at the interface of information theory
and the analysis of Boolean functions.

\subsubsection{Optimality of Majorities}
 The  general open problem as discussed above on the $q$-stability appears to be intractable with the current set of analytical tools. 
A possible strategy to make some progress is  to first find the structure of the optimal solutions attaining the max $q$-stability, and according to this observation,  to prove that the optimal solutions belong to a small class of functions. If functions in this class are well-behaved,  it is then relatively easy to deduce which Boolean functions \emph{in this class} maximize  the $q$-stability. It has been observed that  for odd  dimensions 
$n$ and mean $a=1/2$, the family of majority functions (defined in Section~\ref{sec:form_q_stability}), which is well-behaved,
may be a plausible candidate, since both   dictator functions and   indicators of Hamming balls are  majority functions. 
 Hence, it was conjectured by \citet{mossel2005coin} that 
$\mathrm{Maj}_{n}$ minimizes the symmetric $q$-stability  over
all \emph{anti-symmetric} Boolean functions. We state this formally as follows. 

 \begin{conjecture}[Optimality of majorities] \label{conj:opt_major}
Consider the DSBS  with correlation coefficient $\rho\in(0,1)$.
Fix $q>1$ and $n$ odd. Then, $\breve{\mathbf{S}}_{\rho}^{(q)}[f]$
is maximized among anti-symmetric Boolean functions $f$ by a majority
function $\mathrm{Maj}_{m}$ for  some odd number $m\in[n]$.
 \end{conjecture}
 
In the original conjecture \cite{mossel2005coin}, $q$ was restricted
to be a positive integer. Conjecture~\ref{conj:opt_major}  is weaker than what we hope to resolve the max $q$-stability problem, since only {\em anti-symmetric} Boolean functions are considered in this conjecture, instead of all balanced Boolean functions. Indeed, one can consider a more general
question whether $\breve{\mathbf{S}}_{\rho}^{(q)}[f]$ is maximized
by a majority function
$\mathrm{Maj}_{m}$   among all \emph{balanced} Boolean functions. If the answer is  affirmative, it would have significant implications  
in addressing the max $q$-stability problem in the sense that it allow us to
focus our attention  {\em only  on majority functions}.
  
\subsubsection{Stability of Majorities under Bounds on Coefficients}  
It is also interesting to investigate the noise stability for a specific class of  Boolean functions, e.g.,  the class of functions whose  influences or Fourier coefficients are constrained.  
It is well known that for the majority function $\mathrm{Maj}_{n}$,
all of its Fourier coefficients vanish as $n\to\infty$; see, e.g.,
\cite{ODonnell14analysisof}. On the other hand, the noise stability
$\mathbf{S}_{\rho}[\mathrm{Maj}_{n}]$ of $\mathrm{Maj}_{n}$ for
the DSBS with correlation coefficient $\rho\in(0,1)$ satisfies
\begin{equation}
\lim_{n\to\infty}\mathbf{S}_{\rho}[\mathrm{Maj}_{n}]=\frac{1}{4}+\frac{\arcsin \rho}{2\pi}. \label{eqn:maj_dsbs}
\end{equation}
This can be shown similarly to~\eqref{eq:NICDCL_ball} and~\eqref{eq:half_half} in which $a$ and $b$ are set to $1/2$ and one uses the central limit theorem to  approximate the DSBS by a jointly Gaussian source. 
%The achievability part follows by~\eqref{eq:NICDCL_ball} by setting $a=b=1/2$ and the converse part follows similarly to~\eqref{eq:half_half} for the Gaussian case in which the central limit theorem is invoked. 
%%See \eqref{eq:NICDCL_ball} for the achievability part and \eqref{eq:half_half} for converse part for the Gaussian case. The case for the DSBS behaves similar as $n\to\infty$. 
It has been 
conjectured by \citet{mossel2010noise} that for all balanced Boolean
functions $f$ with small Fourier coefficients, the noise stability
$\mathbf{S}_{\rho}[f]$ cannot exceed the right-hand side of~\eqref{eqn:maj_dsbs} ``by too much''. This is quantified in the following conjecture.

 \begin{conjecture}[Majority is most stable under bounds on the Fourier coefficients]
Consider the DSBS $\pi_{XY}$ with correlation coefficient $\rho\in(0,1)$
and let $(X^{n},Y^{n})\sim \pi_{XY}^{n}$. Let $f:\{0,1\}^{n}\to\{0,1\}$
be an arbitrary balanced function, i.e., $\mathbb{E}[f]=2^{-n}\sum_{x^n \in \{0,1\}^{n}}f(x^n) ={1}/{2}$.  Then,
\begin{equation}
\mathbf{S}_{\rho}[f]\le\frac{1}{4}+\frac{\arcsin \rho}{2\pi} +\eps_{\rho}\Big(\max_{S\subset[n]}|\hat{f}_{S}|\Big),
\end{equation}
where $\eps_{\rho}(\delta)\downarrow 0$ as $\delta \downarrow 0$ for each fixed
$\rho\in(0,1)$.  \end{conjecture}

A weaker version of this conjecture
in which $\max_{S\subset[n]}|\hat{f}_{S}|$ is replaced by the maximal influence
$\max_{i\in[n]}\mathbf{I}_{i}[f]$ was resolved in \cite{mossel2010noise}.

\subsubsection{Extracting a Constant or Sublinear Number of Bits}
In GKW's common information, the number of
bits that is required to be extracted from a  source $(X,Y)\sim \pi_{XY}$ is {\em linear} in the dimension or blocklength $n$.
  In contrast, in the Non-Interactive Correlation Distillation (NICD) or the max $q$-stability problem, only
a {\em single} or a {\em pair} of random bits is to be extracted. A natural generalization
of these two problems is to consider the ``intermediate regime'' in which the number of bits that one hopes to extract is more than one but sublinear in $n$.  For example, one may wish to extract  (a constant) $\ell\ge2$ bits by using
a function $f:\{0,1\}^{n}\to\{0,1\}^{\ell}$ applied to $X^n$ and $Y^n$ where $(X^n,Y^n)$ is a source sequence generated by a DSBS.  We require $f$ to be balanced (i.e., $f(X^{n})\sim\Unif\{0,1\}^{\ell}$),
and at the same time, we aim to maximize  the {\em agreement probability}
%$\Pr(f(X^{n})=f(Y^{n}))$. Here, the agreement probability can be rewritten
%as 
\begin{equation}
\Pr\big(f(X^{n})=f(Y^{n})\big)=\sum_{u^{\ell}\in\{0,1\}^{\ell}}\pi_{XY}^{n}\big( \mathcal{A}_{u^{\ell}}\times\mathcal{A}_{u^{\ell}}\big) , \label{eqn:agree_prob}
\end{equation}
where $\mathcal{A}_{u^{\ell}}=f^{-1}(u^{\ell})  = \{x^n\in\{0,1\}^n: f(x^n)=u^\ell\}$ for $u^{\ell}\in\{0,1\}^{\ell}$ and
$\pi_{XY}$ is the distribution of the DSBS with correlation coefficient
$\rho \in (0,1)$. In other words, we wish to find a partition   $\{\mathcal{A}_{u^{\ell}} \}_{u^{\ell}\in\{0,1\}^\ell}$
of $\{0,1\}^{n}$ such that each subset $\calA_{u^\ell}$ has the same cardinality and~\eqref{eqn:agree_prob} is maximized. 
% with equal cardinality for each $\mathcal{A}_{u^{l}}$
%such that $\sum_{i=1}^{2^{l}}\pi_{XY}^{n}(\mathcal{A}_{u^{l}}\times\mathcal{A}_{u^{l}})$
%is maximized. 
If we na\"ively output the first $\ell$ bits $x^{\ell}$ by
using the function $f(x^{n})=  x^{\ell}$---an indicator of an $(n-\ell)$-subcube (cf.\ Section~\ref{sec:subcubes})---then the induced agreement
probability is exactly $\bigl(\frac{1+\rho}{2}\bigr)^{\ell}$. Indeed, as a consequence
of our solution to (the forward part of) Mossel's mean-$1/4$ stability problem
given in Section~\ref{sec:quarter}, for $\ell=2$, the function $f$ that outputs the first two bits attains
the maximum of the agreement probability for this problem. In addition,
by using the hypercontractivity inequalities in Theorem~\ref{thm:hyper_single}, \citet{bogdanov2011extracting}
showed that  the maximal agreement probability 
\begin{equation}
\max_{f : \{0,1\}^n\to\{0,1\}^\ell }\Pr\big(f(X^{n})=f(Y^{n})\big)\le 2^{-\big(\frac{1-\rho}{1+\rho}\big) \ell}.
\end{equation}
This upper bound
is asymptotically tight as $\ell\to\infty$ in the sense that  there exists $f: \{0,1\}^n\to\{0,1\}^\ell$ such that 
\begin{equation}
\Pr\big(f(X^{n})=f(Y^{n})\big)\ge 0.003\, \sqrt{ \frac{2}{(1-\rho)\,\ell}}  \, 2^{-\big(\frac{1-\rho}{1+\rho}\big)\ell}.
\end{equation}
%the maximal
%agreement probability can be shown to be at least $0.003\cdot(\frac{1-\rho}{2}\cdot \ell)^{-1/2}2^{-\big(\frac{1-\rho}{1+\rho}\big)\ell}$. 
Thus, the exponents  of the lower and upper bounds coincide and are equal to $\frac{1-\rho}{1+\rho}$. 
%whose exponent $\frac{1-\rho}{1+\rho}$ coincides with the exponent
%of the upper bound $2^{-\bigl(\frac{1-\rho}{1+\rho}\bigr)\cdot l}$
%\cite{bogdanov2011extracting}. 
Determining the {\em exact} value of the
maximum agreement probability over all balanced $\{0,1\}^{\ell}$-valued
functions with fixed $\ell\ge3$ remains open.

\subsection{Gaussian Sources}
We next  introduce two open problems for bivariate  Gaussian sources. 

\subsubsection{Standard Simplex Conjecture}

We now consider a Gaussian version of the noise stability problem for   balanced $[m]$-valued functions. This problem is analogous  to its DSBS counterpart for $\{0,1\}^{\ell}$-valued functions. 
In the Gaussian version, we  extract a pair of random variables $U=f(X^{n})$ and $V=f(Y^{n})$
by using a deterministic (measurable) map  $f:\bbR^{n}\to[m]$ from
a pair of length-$n$ vectors $(X^{n},Y^{n})$ 
drawn from 
%a length-$n$ source sequence $(X^{n},Y^{n})$ of 
a bivariate Gaussian source with correlation coefficient $\rho\in(0,1)$.
We require $f$ to be balanced in the sense that $U$, or equivalently
$V$, is  uniformly distributed on $[m]$. We aim to maximize the {\em agreement probability} 
\begin{equation}
\Pr(U=V)=\sum_{i=1}^{m}\pi_{XY}^{n}\big(\mathcal{A}_{i}\times\mathcal{A}_{i}\big) \label{eqn:agree_gauss}
\end{equation}
with $\mathcal{A}_{i}=f^{-1}(i)$ for $i\in[m]$. For this $[m]$-valued
function version of Gaussian NICD problem, \citet{isaksson2012maximally} posed
the  \emph{standard simplex conjecture}. Before stating it, we have to introduce some terminology. 

A {\em flat} or {\em simplex partition} $\{\calA_i\}_{i=1}^m$ of $\bbR^n$ is one in which there exists vectors $\ba_0\in\bbR^n$ and $\{\ba_i\}_{i=1}^m\subset\bbR^n\setminus \{\bzero\}$ such that 
\begin{itemize}
\item for all $i , j \in [m]$ such that $i\ne j$, $\ba_i$ is not a positive multiple of~$\ba_j$;
\item for all $i \in [m]$, 
\begin{equation}
\mathcal{A}_{i} = \ba_0+\Big\{ \bx \in\bbR^{n}:\big\langle \ba_{i} ,\bx\big\rangle =\max_{j \in [m]}\big\langle \ba_{j} ,\bx\big\rangle\Big\} .
\end{equation}
\end{itemize}
 A \emph{standard
simplex partition} is a flat partition $\{\mathcal{A}_{i}\}_{i=1}^m$ where $\|\ba_i\|_2=1$ for all $i$ and $\langle \ba_i, \ba_j \rangle= -\frac{1}{m-1}$ for all $i\ne j$. 
%satisfying 
%\begin{equation}
%\mathcal{A}_{i}\supseteq\Big\{ x^{n}\in\bbR^{n}:\big\langle a_{i}^{n},x^{n}\big\rangle >\big\langle a_{j}^{n},x^{n}\big\rangle ,~\forall~j\neq i\Big\} 
%\end{equation}
%where $a_{1}^{n},\ldots,a_{m}^{n}\in\bbR^{n}$ are vertices of a \emph{regular
%simplex} centered at the origin. Such vectors $a_{1}^{n},\ldots,a_{m}^{n}$
%have the same length, say $L$ (i.e., $\|a_{i}^{n}\|_{2}=L$ for all $i\in[m]$)  
%and $a_{1}^{n},\ldots,a_{m}^{n}$ satisfy $\langle a_{i}^{n},a_{j}^{n}\rangle =-\frac{L^{2}}{m-1}$
%for all $i\neq j$. 

\begin{conjecture}[Standard simplex conjecture] \label{conj:std_simplex}
Consider the bivariate Gaussian source $(X,Y) \sim \pi_{XY}$ with correlation
coefficient $0<\rho<1$. Then, among all partitions $\{\mathcal{A}_{i}\}_{i=1}^m$
of $\bbR^{n}$ into $3\le m\le n+1$ measurable parts of equal $\pi_{XY}^n$-probability (i.e., $\pi_{XY}^n(\calA_i)=1/m$ for all $i\in [m]$), the agreement probability  in~\eqref{eqn:agree_gauss}  is \emph{maximized} by standard simplex partitions. Furthermore, for
$-1<\rho<0$, standard simplex partitions \emph{minimize} the agreement probability in~\eqref{eqn:agree_gauss}.
\end{conjecture} 

This conjecture was confirmed positively by \citet{heilman2014euclidean}
for the case $m=3$ and in the low correlation (i.e., small $\rho$) regime. Specifically, for $m=3$
and $n\ge2$, there exists a function $\rho_{0}(n)>0$ such
that the conjecture holds for $0<\rho<\rho_{0}(n)$. However, for
other cases, Conjecture~\ref{conj:std_simplex} remains open.

\subsubsection{Symmetric Gaussian Problem}
Recall that in the NICD  and  the max $q$-stability problem for the Gaussian case with mean $a=1/2$ (cf.\ Sections~\ref{subsec:NICD-Gaussian}
 and~\ref{sec:gauss_stab}), indicators of parallel halfspaces  attain the forward joint probability  and  the max $q$-stability.  Indicators of  halfspaces  are  {\em anti-symmetric}
(or {\em odd}) in the sense that $f(x^n) = 1-f(-x^n)$ for almost all $x^n \in \bbR^n$; see the analogous definition for functions defined on $\{0,1\}^n$ in Section~\ref{sec:form_q_stability}. It is interesting to ask which  {\em symmetric}
(or {\em even}) functions i.e., those that satisfy $f(x^n) = f(-x^n)$ for almost all $x^n \in \bbR^n$,  maximize the joint probability in the NICD problem
or the $q$-stability in the max $q$-stability problem. For Gaussian
sources, \citet{chakrabarti2012optimal} posed the following problem.

 \begin{problem}[Symmetric Gaussian problem] Fix $0<\rho,a,b<1$.
Let $\mathcal{A}\subset\bbR^n$ and $\mathcal{B}\subset\bbR^{n}$ have Gaussian measures
$a$ and $b$, respectively. Furthermore, suppose $\mathcal{A}$ is {\em centrally
symmetric}, i.e., $\mathcal{A}=-\mathcal{A}$. What is the maximum possible
value of $\Pr(X^{n}\in\mathcal{A},Y^{n}\in\mathcal{B})$, where $X^{n}$ and~$Y^{n}$
are $\rho$-correlated $n$-dimensional standard Gaussian vectors?
\end{problem} 
Even though the problem statement requires that $\calA$ is centrally symmetric, this problem is equivalent to requiring that {\em both} $\calA$ and~$\calB$ are centrally symmetric 
%$\mathcal{A}=-\mathcal{A}$ and $\mathcal{B}=-\mathcal{B}$
\cite{filmus2014real}. Indeed, given a set $\calA$, the optimal $\calB$ that maximizes $\pi_{XY}(\calA\times\calB)$ under the constraint $\pi_{Y}(\calB)=b$ is the set of $y$ such that $\frac{\rmd \pi_{Y|X}(y|\calA)}{\rmd \pi_Y(y)} \ge \lambda$ for some  $\lambda>0$.
%, where $y\mapsto {\rmd \pi_{Y|X}(y|\calA)}/{\rmd \pi_Y(y)}$  is the Radon--Nikodym derivative of $\pi_{Y|X}(\cdot |\calA)$ with respect to $\pi_Y$. 
Hence, if $\calA$ is centrally symmetric, so is the optimal $\calB$  since for this case,  
\begin{align}
\frac{\rmd \pi_{Y|X}(y|\calA)}{\rmd \pi_Y(y) }= \frac{\rmd \pi_{Y|X}(-y|\calA)}{\rmd \pi_Y(-y)} . 
\end{align}
It was conjectured in \citet{chakrabarti2012optimal} and \citet{o2012open} that $\Pr(X^{n}\in\mathcal{A}, Y^{n}\in\mathcal{B})$
is maximized by $(\bbB_{r},\bbB_{s})$ or by $(\bbB_{r}^{\rmc},\bbB_{s}^{\rmc})$
for some appropriate $r,s>0$, where $\bbB_{r}=\{x^{n}\in\bbR^{n} :\|x^{n}\|_{2}\le r\}$ denotes the ball centered at the origin with radius~$r$. This conjecture was, however, disproved by \citet{heilman2021low} in  dimensions two and higher.

\subsection{Ball- and Sphere-Noise Sources}
Up to this point, only {\em memoryless} sources  or, equivalently,  {\em product} distributions have been discussed. Extending the   NICD and $q$-stability problems to sources with memory  is a more challenging but fruitful endeavor, which may provide unique insights. For simplicity, here we consider {\em ball-} and {\em sphere-noise sources} since they behave similarly to   memoryless sources in some sense. Hence, results on  memoryless sources can be applied to ball- and  sphere-noise sources.

\subsubsection{NICD for Ball- and Sphere-Noise Sources} 

We first consider the ball-noise stability problem. Let  $\pi_{X^{n}Y^{n}}$
be  a joint distribution on $\{0,1\}^{n}\times\{0,1\}^{n}$ such that
$\pi_{X^{n}}=\mathrm{Unif}\{0,1\}^{n}$ and 
\begin{align}
Y^{n}=X^{n}\oplus Z^{n}=(X_{i}\oplus Z_{i})_{i\in [n]},    
\end{align}
where $Z^{n}\sim\mathrm{Unif}(\bbB_{r})$ is independent of $X^{n}$
and $\oplus$ denotes the modulo-$2$ sum. Here, $\bbB_{r}$ is the Hamming
ball centered at $0^{n}$ with radius $r$ (cf.\ Section~\ref{sec:ball}). The distribution $\pi_{X^nY^n}$ is no longer of a product form since the coordinates of $(X^n, Y^n)$  are correlated through the entries of $Z^n$.  For $(n,r,M)$
such that $1\le r\le n$ and $1\le M\le2^{n}$, define the \emph{forward
joint probability for $\pi_{X^{n}Y^{n}}$} (or \emph{maximal ball-noise
stability}) as 
\begin{align}
\Gamma_{\mathrm{Ball}}^{(n)}(M,r) & :=\max_{\substack{f:\{0,1\}^{n}\to\{0,1\}\\
\Pr ( f(X^{n})=1 )=a
}
}\Pr\bigl(f(X^{n})=f(Y^{n})=1\bigr),\label{eq:-102-1}
\end{align}
where $(X^{n},Y^{n})\sim \pi_{X^{n}Y^{n}}$ and $a=M/2^{n}$. For fixed
$a,\beta\in(0,1)$, define their upper and lower limits for {\em even} radii as $n\to\infty$
as 
\begin{align}
\overline{\Gamma}_{\mathrm{Ball,\,even}}^{(\infty)}(a,\beta) & :=\limsup_{n\to\infty}\;\Gamma_{\mathrm{Ball}}^{(n)} \bigg(\lfloor a2^{n}\rfloor ,2\Big\lfloor \frac{\beta n}{2}\Big\rfloor \bigg)\quad\mbox{and}\label{eq:-107-1-1}\\
\underline{\Gamma}_{\mathrm{Ball,\,even}}^{(\infty)}(a,\beta) & :=\liminf_{n\to\infty}\;\Gamma_{\mathrm{Ball}}^{(n)} \bigg(\lfloor a2^{n}\rfloor ,2\Big\lfloor \frac{\beta n}{2}\Big\rfloor \bigg).\label{eq:-107-1}
\end{align}
The limits for {\em odd} radii, denoted by $\overline{\Gamma}_{\mathrm{Ball,\,odd}}^{(\infty)}$
and $\underline{\Gamma}_{\mathrm{Ball,\,odd}}^{(\infty)}$, can be defined
analogously but with $2\lfloor \frac{\beta n}{2}\rfloor $
in~\eqref{eq:-107-1-1} and~\eqref{eq:-107-1} replaced by $2\lfloor \frac{\beta n}{2}\rfloor +1$.
In many information-theoretic problems (e.g., error exponents for channel coding), when the dimension $n$ is
sufficiently large, the uniform distribution on the Hamming ball $\bbB_{r}$
can be   thought of as the $n$-fold product of the Bernoulli
distribution $\Bern(r/n)$. This inspires the first author of this monograph to pose
the following conjecture in~\cite{yu2020edge}.  

\begin{conjecture}[NICD for ball-noise sources] \label{conj:ball_noise}
For $a,\beta\in(0,1/2)$, 
\begin{align}
\underline{\Gamma}_{\mathrm{Ball,\,even}}^{(\infty)}(a,\beta) & =\overline{\Gamma}_{\mathrm{Ball,\,even}}^{(\infty)}(a,\beta)=\overline{\Gamma}^{(\infty)}(a,a),\label{eq:even_radii}
\end{align}
where $\overline{\Gamma}^{(\infty)}$ is the asymptotic forward joint probability
  for the DSBS with correlation coefficient $\rho=1-2\beta$; see its definition in~\eqref{eqn:Gamma_infty}.
\end{conjecture}

Conjecture~\ref{conj:ball_noise} pertains only to even radii. For odd radii,  
%for $\widetilde{\Gamma}_{\mathrm{B,odd}}^{(\infty)}$
%and $\utilde{\Gamma}_{\mathrm{B,odd}}^{(\infty)}$ was proven by 
\citet{yu2020edge} 
 showed that for $a,\beta\in(0,1/2)$, 
\begin{align}
\underline{\Gamma}_{\mathrm{Ball,\,odd}}^{(\infty)}(a,\beta) & =\overline{\Gamma}_{\mathrm{Ball,\,odd}}^{(\infty)}(a,\beta)=\overline{\Gamma}^{(\infty)}(a,a).\label{eq:odd_radii}
\end{align}
The ball-noise stability problem can be interpreted as an isoperimetric
problem in the $r^{\mathrm{th}}$ power of the Hamming graph \cite{yu2020edge}.
The edge-isoperimetric inequality in Theorem \ref{thm:EdgeIsoper} is a special case of
this isoperimetric problem with $r=1$. 

In addition, similar questions
can be posed when we replace the ball-noise with the sphere-noise. That
is, we keep all things unchanged apart from the fact that $Z^{n}\sim\mathrm{Unif}(\bbB_{r})$
is replaced by $Z^{n}\sim\mathrm{Unif}(\bbS_{r})$, where~$\bbS_{r}$
is the Hamming sphere centered at $0^{n}$ (cf.\ Section~\ref{sec:spheres}).  For this case, the equalities for the odd case in~\eqref{eq:odd_radii}
still holds. However, \citet{yu2020edge} conjectured that the term
 $\overline{\Gamma}^{(\infty)}(a,a)$ in \eqref{eq:even_radii}
should be replaced by $\frac{1}{2}\overline{\Gamma}^{(\infty)}(2a,2a)$.
For more details, please refer to \cite{yu2020edge}.

\begin{acknowledgements}
%\addcontentsline{toc}{chapter}{Acknowledgements}  
%\chapter*{Acknowledgements}\label{ch:acks}

We sincerely thank the anonymous reviewers for their careful reading and their many insightful comments and suggestions. We would also like to thank the Editor-in-Chief Professor Alexander Barg, and Mr.\ Mike Casey from Now Publishers for their advice in preparing the monograph. We are extremely grateful to our colleagues Zhaoqiang Liu, Anshoo Tandon, Junwen Yang, Qiaosheng Zhang, and  especially Lin Zhou for their help in proofreading parts of the monograph. 

Lei Yu is supported by the National Natural Science Foundation of China (NSFC) grant 62101286 and the Fundamental Research Funds for the Central Universities of China (Nankai University).

Vincent Tan is supported by a Singapore National Research Foundation (NRF) Fellowship (A-0005077-00-00) and Singapore Ministry of Education AcRF Tier 1 grants (A-0009042-00-00, A-8000189-00-00, and A-8000196-00-00).  He would like to thank his wife Huili Guo and his four children Oliver Tan Ying Ren,  Giselle Tan Ying Ci, Hazel Tan Ying Shan, and Ashleigh Tan Ying Xi for their unwavering support and understanding during the writing of this monograph.
\end{acknowledgements}

\backmatter  % references, restarts sample
\sloppy
\printbibliography

\end{document}